\documentclass{book}
\usepackage[utf8]{inputenc}
\usepackage{amssymb}
\usepackage{amsmath}
\usepackage{amsfonts}
\usepackage{amsthm}
\usepackage[T1]{fontenc}
\usepackage[utf8]{inputenc}
\usepackage[main=english, polish]{babel}
\usepackage[utf8]{inputenc}
\usepackage{lscape}
\usepackage{mathtools}
\usepackage{stmaryrd}
\usepackage[hypcap=false]{caption}
\usepackage{multirow}
\usepackage{hhline}
\usepackage{mathrsfs}
\usepackage{tabularx}
\usepackage[table]{xcolor}
\usepackage{makeidx}
\usepackage[color,all]{xy}
\usepackage{xcolor,colortbl}
\usepackage[hmargin={2cm,2cm},vmargin={2cm,2cm}]{geometry}
\usepackage[pdftex]{graphicx}
\usepackage{comment}
\usepackage{tikz}
\usetikzlibrary{shapes.misc,shadows}

\usepackage{fancyhdr}

\usepackage[skip=5pt plus1pt, indent=0pt]{parskip}

\usepackage[colorlinks=true,linkcolor=black]{hyperref}

\hypersetup{urlcolor=blue, citecolor=black}

\usepackage{color}

\newcommand{\bea}{\begin{eqnarray}}
	\newcommand{\eea}{\end{eqnarray}}

\newcommand{\be}{\begin{equation}}
	\newcommand{\ee}{\end{equation}}

\theoremstyle{plain}
\newtheorem{theorem}{Theorem}
\newtheorem{corollary}[theorem]{Corollary}
\newtheorem{proposition}[theorem]{Proposition}
\newtheorem{lemma}[theorem]{Lemma}

\theoremstyle{definition}
\newtheorem{definition}[theorem]{Definition}

\newtheorem{example}{Example}

\newcommand{\rr}{\mathbb{R}}
\newcommand{\rz}{\mathbb{R} \backslash \{0\}}
\newcommand{\rpm}{\mathbb{R}_{\pm}}

\theoremstyle{definition}

\newcommand\xqed[1]{%
	\leavevmode\unskip\penalty9999 \hbox{}\nobreak\hfill
	\quad\hbox{#1}}
\newcommand\demo{\xqed{$\triangle$}}

\newcommand{\beas}{\begin{eqnarray*}}
	\newcommand{\eeas}{\end{eqnarray*}}

\numberwithin{equation}{section}
\numberwithin{theorem}{section}

\def\rc#1{\begin{color}{red}#1\end{color}}

\begin{document}

\setlength{\baselineskip}{14pt}

\begin{titlepage}
\begin{center}
{\huge \bf University of Warsaw}\\[0.5cm]
{\huge Faculty of Physics}\\[3cm]
{\LARGE Daniel Wysocki}\\[0.3cm]
{\large Student no.: 335602}\\[2cm]
{\huge \bf Geometric approaches to Lie bialgebras,\\\vskip .5cm their classification, and applications}\\[2cm]
{\Large Doctoral Dissertation}\\[0.1cm]
{\Large in}\\[0.1cm]
{\Large MATHEMATICS}\\[3cm]
\end{center}
\hspace*{6cm}{\large Advisor:}\\
\hspace*{6cm}{\large \bf dr. hab. Javier de Lucas Araujo, prof. UW}\\
\hspace*{6cm}{\large Department of Mathematical Methods in Physics}\\[3cm]

\begin{center}
{\large Warsaw, March 2023}
\end{center}
\end{titlepage}

\pagestyle{fancy}
\fancyhead{}
\fancyhead[LE,RO]{\thepage}
\fancyhead[LO,RE]{\nouppercase{\leftmark}}
\fancyfoot{}

\clearpage

{\it Oświadczenie kierującego pracą / Advisor's statement}
\vspace{0.5cm}

Oświadczam, że niniejsza praca została przygotowana pod moim kierunkiem i stwierdzam, że spełnia ona warunki do przedstawienia jej w postępowaniu o nadanie stopnia doktora w dziedzinie nauk ścisłych i przyrodniczych w dyscyplinie matematyka.

\vspace{0.5cm}

Hereby, I confirm that the presented thesis was prepared under my supervision and that it fulfills the requirements for the degree of PhD in Mathematics.

\vspace{2cm}

\hspace{1cm} {Data / Date} \hspace{5cm} {Podpis kierującego pracą / Advisor's signature}

\vspace{3cm}

{\it Oświadczenie autora pracy / Author's statement}
\vspace{0.5cm}

Świadom odpowiedzialności prawnej oświadczam, że niniejsza rozprawa doktorska 
została napisana przeze mnie samodzielnie i nie zawiera treści uzyskanych w sposób niezgodny z obowiązującymi przepisami.

Oświadczam również, że przedstawiona praca nie była wcześniej przedmiotem procedur związanych z uzyskaniem stopnia doktora w innej jednostce.

Oświadczam ponadto, że niniejsza wersja pracy jest identyczna z załączoną wersją elektroniczną. 

\vspace{0.5cm}

Hereby, I declare that the presented thesis was prepared by me and none of its contents was obtained by means that are against the law.

The thesis has never before been a subject of any procedure of obtaining an academic degree.

Moreover, I declare that the present version of the thesis is identical to the attached electronic version.

\vspace{2cm}

\hspace{1cm} {Data / Date} \hspace{5cm} {Podpis autora pracy / Author's signature}

\newpage

\newpage
\begin{center}
{\large \bf Abstract}
\end{center}
\vspace{0.5cm}

This thesis presents novel algebraic and geometric approaches to the classification problem of coboundary Lie bialgebras up to Lie algebra automorphisms. This entails the analysis of modified and classical Yang-Baxter equations and other related mathematical structures. More specifically, we develop Grassmann, graded algebra, and algebraic invariant techniques for the classification problem of coboundary Lie bialgebras and the so-called $r$-matrices.  The devised techniques are mainly focused on the study of $r$-matrices for three-dimensional and indecomposable four-dimensional Lie algebras. An  example of a decomposable four-dimensional Lie algebra, namely $\mathfrak{gl}_2$, is also considered. Other particular higher-dimensional Lie bialgebras, e.g. $\mathfrak{so}(2,2)$ and $\mathfrak{so}(3,2)$, are partially studied. Special relevance has the development of a new notion: the Darboux families, which provide a powerful method for the classification and determination of classes of $r$-matrices. In fact, the classification of $r$-matrices for four-dimensional indecomposable Lie bialgebras is a significative advance relative to previous results in the literature as the classification problem of Lie bialgebras has been predominantly treated algebraically in the literature. Meanwhile, we address the problem in a more geometric manner.

Moreover, the theory of Lie bialgebras and related notions are employed in the study of Lie systems, their generalisations, and Hamiltonian systems. More specifically, we study the use of $r$-matrices for the description of interesting Hamiltonian systems relative to symplectic and Poisson structures. Additionally, we extend the construction of integrable deformations of Lie--Hamilton systems on symplectic manifolds to Jacobi manifolds.

The outlook of this doctoral thesis suggests some further research directions, both in the abstract classification problem and the mathematical-physics applications. Appendices present a new practical method to obtain matrix representation for Lie algebras with nontrivial centre used throughout the work and the code written in Mathematica that has been used to facilitate some of the  computations for coboundary Lie bialgebras classifications.

\vspace{1cm}
\begin{center}
{\large \bf Keywords:}
\end{center}
\vspace{0.5cm}
\begin{center}
{Classical Yang--Baxter equation, Darboux family, deformation of Lie--Hamilton system, foliated Lie system,  invariant form,  Jacobi manifold, Lie bialgebra, modified Yang--Baxter equation, r-matrix}
\end{center}

\newpage

\begin{center}
{\large \bf Streszczenie}
\end{center}
\vspace{0.5cm}

Niniejsza rozprawa prezentuje nowe algebraiczne i geometryczne podejścia do problemu klasyfikacji kobrzegowych bialgebr Liego z dokładnością do automorfizmów algebr Liego. Analizowane są również zmodyfikowane i klasyczne równania Yanga--Baxtera oraz powiązane struktury matematyczne. W szczególności, rozwijamy techniki oparte na zastosowaniu algebr Grassmanna, algebr gradowanych oraz algebraicznych niezmienników w celu klasyfikacji kobrzegowych bialgebr Liego oraz tak zwanych r-macierzy. Skupiamy się przede wszystkim na badaniu r-macierzy dla trójwymiarowych oraz nierozkładalnych cztero\-wymiarowych algebr Liego. Rozważamy też czterowymiarowy przykład rozkładalnej algebry Liego, $\mathfrak{gl}_2$. Inne wyżejwymiarowe przypadki, m.in. $\mathfrak{so}(2,2)$ oraz $\mathfrak{so}(3,2)$, są częściowo analizowane. Szczególne znaczenie ma nowe pojęcie wprowadzone i rozwijane w pracy, rodziny Darboux, które dostarcza wyjątkowo użytecznej metody wyznaczania i klasyfikacji klas r-macierzy. Analiza r-macierzy przeprowadzona dla czterowymiarowych nierozkładalnych algebr Liego  jest jednym z niewielu geometrycznych ujęć problemu klasyfikacji bialgebr Liego i stanowi istotny postęp względem wcześniejszych prac, które oferowały głównie algebraiczne podejścia.

Ponadto, stosujemy pojęcia znane z teorii bialgebr Liego do badania układów Liego, ich uogólnień oraz układów hamiltonowskich. Przyglądamy się szczególnie zastosowaniom r-macierzy w opisie układów hamiltonowskich wględem struktur symplektycznych i poissonowskich. Dodatkowo, rozszerzamy konstrukcję całkowalnych deformacji układów Liego--Hamiltona na rozmaitościach symplektycznych na rozmaitości Jacobiego.

W Podsumowaniu pracy nakreślamy możliwe kierunki dalszych badań, zarówno w abstrakcyjnym problemie klasyfikacyjnym, jak i zastosowaniach w fizyce matematycznej. W Dodatkach opisujemy praktyczną metodę otrzymywania macierzowej reprezentacji algebr Liego posiadających nietrywialne centrum, używaną w całej pracy, oraz prezentujemy kod napisany w programie Mathematica, który był stosowany do wykonania rachunków niezbędnych podczas klasyfikowania kobrzegowych bialgebr Liego.

\vspace{1cm}
\begin{center}
{\large \bf Słowa kluczowe}
\end{center}
\vspace{0.5cm}
\begin{center}
{Bialgebra Liego, deformacja układu Liego--Hamiltona, foliacja układów Liego, forma niezmiennicza, klasyczne równanie Yanga--Baxtera, zmodyfikowane równanie Yanga--Baxtera,  r-macierz, rodzina Darboux, rozmaitość Jacobiego}
\end{center}
\vspace{1cm}
\begin{center}
{\large \bf Tytuł pracy w języku polskim}
\end{center}
\vspace{0.5cm}
\begin{center}
{Geometryczne podejścia do bialgebr Liego, ich klasyfikacji oraz zastosowań}
\end{center}

\newpage 

\begin{center}
{\large \bf Acknowledgements}
\end{center}
\vspace{0.5cm}

In the first place, I am grateful to my advisor, prof. Javier de Lucas, for accepting me as his student and giving me a unique opportunity to experience the academic world. Without his understanding and endless efforts, this thesis would not exist and my studies would never reach this point.

Secondly, I would like to thank my parents for their support during this challenging PhD journey of mine.

No words can sufficiently describe my gratitude to Beata Dul for standing by me for this whole time. And for reminding me that proper resting is as important as working.

Many thanks to Karol Kraszewski for opening his doors when I needed it the most. And the longest discussions I've ever had. I'm still standing after all this time thanks to you as well.

I am indebted to the amazing psychotherapists at CPP UW, Anna Mach and Renata Składanek, for their professional help during the darkest times. Also, a great shoutout to wonderful people I've met there.

Last but not least, I want to thank Clarisson Canlubo for many beautiful memories and countless life lessons. I cannot thank you well enough.

And to many other people whom I met along the way - thank you!

My PhD studies were partially funded by the PhD Scholarship Programme "Kartezjusz" for students at the University of Warsaw and Jagiellonian University.

\newpage

The content of this Doctoral Dissertation is based on three published research papers:
\begin{itemize}
\item J. de Lucas and D. Wysocki,
{\it A Grassmann and graded approach to coboundary Lie bialgebras, their classification, and Yang-Baxter equations}, 
J. Lie Theory {\bf 30} (2020), no. 4, 1161--1194.
\item J. de Lucas and D. Wysocki, 
{\it Darboux families and the classification of real four-dimensional indecomposable coboundary Lie bialgebras}, 
Symmetry {\bf 13} (2021), no. 3, 465.
\item J. F. Cari\~nena, J. de Lucas and D. Wysocki,
{\it Stratified Lie systems: theory and applications},
J. Phys. A {\bf 55} (2022), 385206, 31 pp.
\end{itemize}
and one unpublished research paper draft:
\begin{itemize}
\item J. de Lucas and D. Wysocki, 
{\it  Deformation of Lie-Hamilton systems via Jacobi structures}, in preparation (2023).
\end{itemize}

\newpage

\tableofcontents

\newpage

\chapter*{Introduction}
\chaptermark{Introduction}
\addcontentsline{toc}{chapter}{Introduction}

It has been a decade-long journey for the notion of Lie bialgebra to appear on the mathematical landscape and become a topic of separate studies. As for many modern concepts, its beginnings are tightly related with the intense development of physical theories in the previous century and can be traced back to the research on exactly solvable models in quantum mechanics by Yang \cite{Y67, Y68} and statistical mechanics by Baxter \cite{Bax72, Bax78}. In these independent work, the same type of equation, called later the quantum Yang--Baxter equation (QYBE), has been analysed. Further contributions in the study of QYBE has come from the researchers of the Leningrad Mathematical School. L. Faddeev and Sklyanin generalised the so-called inverse scattering method \cite{BIK93, Fa87}, an effective tool for analysing nonlinear partial differential equations, in order to address the integrability problems of Yang and Baxter \cite{Fa84, FS78, FT79, FST79, Sk82}. The relation between these methods and their underlying geometric structures have been discussed extensively by Drinfeld. These endeavours are considered the origin of the theory of Hopf algebras and quantum groups (see \cite{CP95, Dr87, Ma95}). It provided mathematical foundations for prior findings in the field. In particular, the description of r-matrices and the classical Yang--Baxter equation, which are crucial objects in the inverse scattering method, has been given a simple mathematical formulation. In this framework, an r-matrix $r$ is just an element of the tensor product $\mathfrak{g}^{\otimes 2}$ for a Lie algebra $(\mathfrak{g}, [\cdot, \cdot]_{\mathfrak{g}})$ satisfying the classical Yang--Baxter equation of the form
$$
[a_i, a_j] \otimes b_i \otimes b_j + a_i \otimes [b_i, a_j] \otimes b_j + a_i \otimes a_j \otimes [b_i, b_j] = 0
$$
where one uses the decomposition $r = \sum_i a_i \otimes b_i$, $a_i, b_i \in \mathfrak{g}$. Moreover, it could be understood as a special case of a new structure, a Lie bialgebra, that is a pair of Lie algebra structures on $\mathfrak{g}$ and $\mathfrak{g}^*$ with a suitable compatibility condition \cite{CP95}. 

Since the report \cite{Dr87} presented by Drinfeld in 1986 at the International Mathematical Symposium, the research on the topic has experienced noticeable growth. Particularly influential work has been done on the quantisation problem of Lie bialgebras \cite{EH10, EK96, Resh92}. Alongside these mathematical inquiries on quantisation, physical applications of Lie bialgebras have been developed, especially the study of quantum spacetimes \cite{BGGH17, GBGH19, BH96, BHM13, BHMN14, BHN15, BHN14, MS11, MS03}, field-theoretical models \cite{BKLSY16, CFG03} and deformations of classical systems \cite{BBHMR09, BCFHL18, BCFHL21, BO98}. A separate line of research concerned the study of solutions and geometry of CYBE \cite{B07, Dr83, GG97, Ma90, STS83, WX92} and QYBE \cite{CG89, EHo00, GG98}. Both notions also admitted interesting generalisations \cite{BDG17, BGP14, Do21, EL05, ESSo99}.

The systematic analysis of solutions to CYBE remains until now the most basic problem. In \cite{BD82, BF84, Dr83a}, Drinfeld classified the solutions to the classical Yang--Baxter equation in case of the complex semisimple Lie algebra and has made many relevant observations on the structure of the solutions to CYBE for the Lie algebra $\mathfrak{g}[u]$ associated with the complex semisimple $\mathfrak{g}$. It was the first result that addressed the classification problem for r-matrices. Since then, further research in this direction was relatively scarce, as the problem is commonly considered hard. However, several noticeable advancements have been made over last forty years on this issue.

By direct generalisation of the method presented in \cite{Dr83a}, Andruskiewitsch obtained the classification of r-matrices for real absolutely simple Lie algebras \cite{AJ03}. Another insightful approach is due to A. Stolin. In \cite{St91a, St91b}, he continued the initial study of Drinfeld on solutions over $\mathfrak{g}[u]$ and obtained the classification of so-called constant and rational solutions for $\mathfrak{sl}_n$, $n > 1$. He also proved that this problem is equivalent to the analysis of Frobenius subalgebras of $\mathfrak{g}$. In \cite{SP14}, he extended the previous analysis and by cohomological means, classified all bialgebra structures on $\mathfrak{sl}(n, \mathbb{K})$, where $\mathbb{K}$ is an arbitrary field of characteristic zero. Recently, S. Maximov et al. obtained the classification of Lie bialgebra structures over $\mathfrak{g}[u]$, where $\mathfrak{g}$ is any simple Lie algebra over a field of characteristic zero \cite{AMSZ22}.

Apart from already mentioned structural considerations on the classification problem, there has been several more direct algebraic approaches to the problem. Full classification of three-dimensional real Lie bialgebras has been obtained by Gomez \cite{Go00} and Farinati \cite{FJ15}, whereas the four-dimensional case has only been partially studied \cite{ARH17, BLT16}. A few  cases of particular structure have also been analysed \cite{FJ13, FJ22}. The problem has also been discussed in equivalent formulations (see \cite{De18} on the classification of Manin pairs or \cite{HS02} on the classification of Drinfeld doubles). However, Lie bialgebra structures on Lie algebras of physical importance has mostly been considered \cite{BCH00, BH96, BH97, BHMN14, BHMN17, BHP99, BLT16, HLS09, LT17, NT00, Op98, Op00, WWY15}. 

\section*{Outline of the thesis}
This thesis presents a few novel algebraic and geometric approaches to the classification problem of coboundary Lie bialgebras. In particular, we focus on the study of r-matrices for three-dimensional and indecomposable four-dimensional Lie algebras. An additional example of a four-dimensional decomposable Lie algebra $\mathfrak{gl}_2$ is also considered.

Chapter 1 is a survey of prerequisites that are necessary to understand the findings of this thesis. Basic notions and fundamental results in differential geometry, Hopf algebra theory, Lie algebras and Lie systems have been presented.

Chapter 2 is based on \cite{LW20} and it presents the algebraic approach to the classification problem. Its exposition notably clarifies and extends the content of the master thesis of the author \cite{W17}.
Let $\mathfrak{g}$ hereafter stand for a Lie algebra and let ${\rm End}(V)$ stand for the space of endomorphisms on a vector space $V$. Section \ref{Ch:alg_Sec:inv} discusses invariant maps on vector spaces. In particular, the construction of $k$-multilinear $\mathfrak{g}$-invariant maps on $\Lambda V$ and $\Lambda^mV$, for the exterior algebra $\Lambda V$ and the subspaces $\Lambda^mV$ of $m$-vectors, using $k$-multilinear $\mathfrak{g}$-invariant map on $V$ is detailed (Theorem \ref{extension}). Moreover, we show that the induced map on $\Lambda^mV$ is invariant relative to the Lie group ${\rm GL}(\Lambda^m\rho)$ of the $\mathfrak{g}$-module $\Lambda^mV$, namely the natural lifts to $\Lambda^mV$ of the automorphisms on $V$ induced by the exponentiation of the endomorphisms $\rho(\mathfrak{g})$. In Section \ref{Ch:alg_Sec:kil}, we apply these results to the maps induced on $\Lambda^m\mathfrak{g}$ by the Killing form, called {\it Killing-type forms}. We prove Killing-type forms are also invariant relative to the induced automorphism group GL$(\Lambda^m\widetilde{\rm ad})$. Some efficient methods to determine such forms are presented in Section \ref{Ch:alg_Sec:invcalc}. In the next part, we focus on the modified classical Yang-Baxter equation (mCYBE). Section \ref{Ch:alg_Sec:grad} introduces a gradation on the Grassmann algebra that simplifies the analysis of mCYBE solutions and provides an approximate geometric picture of the mCYBE. In Section \ref{Ch:alg_Sec:gInv}, several tools based on this gradation are presented, which allow us to easily obtain $\mathfrak{g}$-invariant elements in $\Lambda\mathfrak{g}$. We discuss how the study of mCYBE might be facilitated by such elements. In Section \ref{Ch:alg_Sec:3DClass}, previous techniques are employed to compute inequivalent $r$-matrices for all three-dimensional Lie algebras. As an illustration of further possibilities of our approach, $r$-matrices for the four-dimensional Lie algebra $\mathfrak{gl}_2$ are obtained in Section \ref{Ch:alg_Sec:4D}.

Chapter 3 is based on \cite{LW21} and it presents an approach from the geometric standpoint. The main tool is the so-called {\it  Darboux family}, an analogue of a Darboux polynomial, which is carefully explained in Section \ref{Ch:Darb_Sec:DarFam}. Its relation to mCYBE and other geometric properties are discussed in Section \ref{Ch:Darb_Sec:mcybe}. Finally, the usefulness of this notion in the classification problem is illustrated in Sections \ref{Ch:Darb_Sec:Cla} and \ref{Ch:Darb_Sec:4Ddec} by computing inequivalent $r$-matrices for all indecomposable four-dimensional Lie algebras and the decomposable Lie algebra $\mathfrak{gl}_2$.

Chapters 4 and 5 focus on using $r$-matrices to study differential equations and develop new applications. Based on \cite{CLW22}, Chapter 4 introduces the notion of stratified Lie systems and shows how Poisson structures induced by $r$-matrices might be applied to obtain and to study Hamiltonian systems of this type. Chapter 5 presents some unpublished results \cite{LW23}. We extend the construction of integrable deformations by Ballesteros et al. \cite{BCFHL18, BCFHL21} to Jacobi and contact settings. Moreover, the properties and applications of such generalisations are briefly discussed. 

The outlook chapter suggests some further research directions, both in the abstract classification problem and the applications in mathematical physics. 

Appendix A presents the code written in Mathematica that has been used in order to facilitate the computations necessary for the classification results presented in Chapter 2 and 3.

Appendix B discusses an effective method to obtain matrix representation for Lie algebras with nontrivial centre. Such representations are necessary to carry out the computations in Chapter 3.

\chapter{Fundamentals}
As already stressed in the Introduction, this thesis addresses the classification problem of Lie bialgebras in a geometric manner. In consequence, presented methods require basic familiarity with certain well-known results in differential geometry, Poisson geometry etc. This chapter is intended as a brief survey of such prerequisites. We introduce the necessary geometric tools and algebraic notions to be used in further chapters. 

Unless explicitly stated otherwise, several general assumptions hold throughout this work. The set of natural numbers $\mathbb{N}$ contains the zero element. All geometric objects are smooth and real. Manifolds are finite-dimensional and connected. Moreover, all geometric structures are globally defined. Of course, a more detailed treatment without previous assumptions is possible (especially it is believed that most results could be extended to the complex case without much difficulties). However, we decided to focus on presenting the main aspects of the theory in the simplest scenario and thus, we omit such minor technical considerations. 

Finally, for the tangent bundle $TM$ of a manifold $M$, the space of its sections will be denoted by $\mathfrak{X}(M)$. Each element of $\mathcal{X}(M)$ is called a {\it vector field} on $M$.
	
\section{Poisson geometry and Schouten-Nijenhuis bracket}

A {\it Poisson algebra} over $\mathbb{R}$ is a real vector space $V$  equipped with two bilinear maps $\bullet$ and $\{\cdot, \cdot\}$, such that $(V, \bullet)$ is a real algebra and $(V, \{\cdot, \cdot\})$ forms a Lie algebra. Moreover, the Lie bracket $\{\cdot, \cdot\}$ satisfies the Leibniz rule relative to $\bullet$, i.e. $\{a, b \bullet c\} = \{a, b\} \bullet c + b \bullet \{a, c\}$ for any $a,b,c \in V$.

The most relevant example of a Poisson algebra is the tensor algebra $T(\mathfrak{g})$ of a Lie algebra $(\mathfrak{g}, [\cdot, \cdot]_{\mathfrak{g}})$. Recall that $T(\mathfrak{g}) := \bigoplus_{p=0} T^p(\mathfrak{g})$, where $T^0(\mathfrak{g}) := \mathbb{R}$ and $T^k(\mathfrak{g}) := \mathfrak{g}^{\otimes k}$ for any $k \geq 0$. it is an associative algebra relative to the tensor product. The Lie bracket $[\cdot ,\cdot]_{T(\mathfrak{g})}$ on $T(\mathfrak{g})$ is defined as follows. First, $[m, g]_{T(\mathfrak{g})} := 0$ for $m \in T^0(\mathfrak{g})$ and $g \in T(\mathfrak{g})$. Next, we put $[p, q]_{T(\mathfrak{g})} := [p, q]_{\mathfrak{g}}$ for $p,q \in T^1(\mathfrak{g})$. Then recursively, we get the Lie bracket on the whole $T(\mathfrak{g})$ by imposing $[a \otimes b, c]_{T(\mathfrak{g})} = [a,c]_{T(\mathfrak{g})} + a \otimes [b,c]_{T(\mathfrak{g})}$ for $a,b,c \in T(\mathfrak{g})$.

The Poisson algebra structure is inherited by the universal enveloping algebra, that is the quotient $U(\mathfrak{g}) := T(\mathfrak{g}) / \mathcal{I}$, where $\mathcal{I}$ is the ideal generated by the elements $a \otimes b - b \otimes a - [a, b]_{\mathfrak{g}}$, where $a,b \in \mathfrak{g}$.

Another crucial example is the algebra $\mathcal{F}(M)$ of real smooth functions on a symplectic manifold $M$. In short, a symplectic manifold is a manifold $M$ equipped with a nondegenerate, closed 2-form $\omega \in \Omega^2(M)$. Given any function $f \in C^{\infty}(M)$, there exists an associated unique vector field $X_f \in \mathcal{X}(M)$, called {\it Hamiltonian vector field}, satisfying $\iota_{X_F}\omega = -{\rm d}f$ (see \cite{AM87} for details). Then for any $f,g \in C^{\infty}(M)$, the map $\{f, g\} := \omega(X_f, X_g)$ together with the pointwise multiplication of functions gives the sought Poisson algebra structure on $C^{\infty}(M)$.

The above case motivates a slightly general notion of a Poisson manifold. A {\it Poisson manifold} is a manifold $M$ such that its algebra $C^{\infty}(M)$ of smooth functions is equipped with a Poisson algebra structure. In other words, there exists an bilinear antisymmetric map $\{\cdot, \cdot\}: C^{\infty}(M) \times C^{\infty}(M) \to C^{\infty}(M)$ that satisfies Jacobi identity and Leibniz rule. This map is called a {\it Poisson bracket}. There is, however, an equivalent definition that we discuss next.

Since $\mathfrak{X}(M)$ is a vector space, we can construct the $k$-th tensor product of $\mathfrak{X}(M)$ with itself. Its anti-symmetric part is usually denoted by $\mathcal{V}^k M$ and the elements of $\mathcal{V}^k M$ are called $k$-{\it vector fields}. Let $\mathcal{V}^0 M := C^\infty(M)$ and define $\mathcal{V}M:=\bigoplus_{k\in \mathbb{N}}\mathcal{V}^k M$, the space of so-called {\it multivector fields}. 

From the fact that a Poisson bracket on $M$ is antisymmetric and satisfies the Leibniz rule, we conclude that $\{\cdot, \cdot\}$ gives rise to a 2-vector field $\pi \in \Lambda^2 M$. The remaining question is how to encode the Jacobi identity in terms of such a bivector field. This task can be accomplished using the notion introduced next.

\begin{definition}\label{SNbracket}
The {\it Schouten-Nijenhuis bracket} on $\mathcal{V}M$ is the unique bilinear map $[\cdot, \cdot]_{S}: \mathcal{V}M \times \mathcal{V}M \to \mathcal{V}M$ satisfying that: a) $[f,g]_S=0$ for every $f,g\in C^\infty(M)$; b) $Xf=[X,f]_S=-[f,X]_S$, for every $X\in \mathfrak{X}(M)$ and $f\in C^\infty(M)$; c) We define
\begin{equation*}
[X_1 \wedge \ldots \wedge X_k, Y_1 \wedge \ldots \wedge Y_l]_{S} := \sum_{i=1}^k \sum_{j =1}^l (-1)^{i+j} [X_i, Y_j] \wedge X_1 \wedge \ldots \wedge \widehat{X}_i \wedge \ldots X_k 
\wedge Y_1\wedge\ldots\wedge  \widehat{Y}_j \wedge \ldots \wedge Y_l,
\end{equation*}
where $X_1,\ldots X_k,Y_1,\ldots,Y_l$ are vector fields on $M$ with $k,l\in \mathbb{N}$, hatted vector fields are omitted in the exterior product and $[\cdot, \cdot]$ stands for the Lie bracket of vector fields.
\end{definition}

The Schouten-Nijenhuis bracket restricted to $\mathcal{V}^1 M = \mathfrak{X}(M)$ matches the Lie bracket of vector fields. In this sense, the Schouten-Nijenhuis bracket can be considered as a generalisation of the Lie derivative of vector fields to multivector fields \cite{Esposito,Ni55,Va94}. Since the Lie bracket of two vector fields is local, namely the value of $[X,Y]$ for two vector fields $X,Y$ on $M$ at $p\in M$ depends only on the value of $X$ and $Y$ on a local open neighbourhood of $p$. Due to this and Definition \ref{SNbracket}, the Schouten-Nijenhuis bracket is also local.

Finally, it is a matter of straightforward computation to verify that the Jacobi identity for the Poisson bracket is equivalent to the condition $[\pi, \pi]_{SN} = 0$ for the associated 2-vector field $\pi$.

\begin{proposition}\label{Pr:PropSchou}
Let $X \in \mathcal{V}^k M, Y \in \mathcal{V}^l M$, and $Z \in \mathcal{V}^m M$. The Schouten-Nijenhuis bracket satisfies the following properties:
\begin{enumerate}
\item $[X,Y]_S\in \mathcal{V}^{k+l-1}M$,
\item $[X, Y]_{S} = -(-1)^{(k-1)(l-1)} [Y, X]_{S}$,
\item $[X, Y \wedge Z]_{S} = [X, Y]_{S} \wedge Z + (-1)^{(k+1)l} Y \wedge [X, Z]_{S}$,
\item $[X, [Y, Z]_{S}]_{S} = [[X, Y]_{S}, Z]_{S} + (-1)^{(k-1)(l-1)} [Y,[X, Z]_{S}]_{S}$.
\end{enumerate}
\end{proposition}

The last property in Proposition \ref{Pr:PropSchou} can be rewritten as
\begin{equation*}
{\footnotesize (-1)^{(k-1)(m-1)} [X, [Y, Z]_{S}]_{S} + (-1)^{(l-1)(k-1)} [Y, [Z, X]_{S}]_{S} + (-1)^{(m-1)(l-1)}[Z, [X, Y]_{S}]_{S} = 0}
\end{equation*}
and, in this form, it is called a \textit{graded Jacobi identity} \cite{Ma97}. 
The space $\mathcal{V} M$ equipped with the Schouten-Nijenhuis bracket is an example of a {\it Gerstenhaber algebra}. 

In this thesis, we are mostly concerned with a special case of the Schouten-Nijenhuis bracket, defined on the exterior algebra $\Lambda \mathfrak{g}$ of a Lie algebra $\mathfrak{g}$. Let us see how its construction follows from the above definition. Let $G$ be a Lie group with an associated Lie algebra $\mathfrak{g}$. Since the Lie bracket of two left-invariant vector fields on $G$ is left-invariant, it stems from Definition \ref{SNbracket}that the Schouten-Nijenhuis bracket of two left-invariant multivector fields on $G$ is again a left-invariant multivector field on $G$. Moreover, left-invariant $0$-vector fields $\mathcal{V}^0(G)^L$ on a Lie group $G$ are given by the constant functions on $G$.  Indeed, the space of $0$-vector fields on $G$ consists of functions on $G$. Given $f \in \mathcal{V}^0(G)$, the left-invariance condition implies that $f(h) = f(gh)$ for any $g, h\in G$. Thus, $f$ is constant on $G$ and $(\mathcal{V}^0(G))^L$ is isomorphic, as a vector space, to $\mathbb{R}$. 
Therefore, the restriction of the Schouten-Nijenhuis bracket  on $\mathcal{V}G$  to the space $\mathcal{V}^LG$ of left-invariant multivector fields on $G$ gives rise to an algebra, which is again a Gerstenhaber algebra. Since $T_eG$ is isomorphic to the left-invariant vector fields on $G$, a $k$-antisymmetric tensor product $\Lambda^k\mathfrak{g}$ of elements of $\mathfrak{g}$ can be identified with $\mathcal{V}^k(G)^L$ for any $k \in \mathbb{N}$. In consequence, the whole $\mathcal{V}^LG$ can be written as $\bigoplus_{k \in \mathbb{N}} \Lambda^k \mathfrak{g}$. For convenience, we denote 
$\Lambda \mathfrak{g} := \bigoplus_{k \in \mathbb{N}} \Lambda^k \mathfrak{g}$. By identifying $\mathcal{V}^LG$ with $\Lambda \mathfrak{g}$, we obtain the Schouten bracket on $\Lambda\mathfrak{g}$. This bracket is known in the literature as the \textit{algebraic Schouten bracket} on the exterior algebra $\Lambda\mathfrak{g}$ \cite[pg. 172]{Va94}.

\begin{example} \label{Ex:sl2}
Consider the Lie algebra $\mathfrak{sl}_2 = \langle e_1,e_2,e_3 \rangle$ with commutation relations
\begin{equation*}
[e_1,e_2]=e_1,\qquad [e_1,e_3]=2e_2,\qquad [e_2,e_3]=e_3.
\end{equation*}
Then,
$
\Lambda^2\mathfrak{sl}_2=\langle e_1\wedge e_2,e_1,\wedge e_3,e_2\wedge e_3\rangle, \Lambda^3\mathfrak{sl}_2=\langle e_1\wedge e_2\wedge e_3\rangle.
$
Every Lie algebra can be understood as  the Lie algebra of left-invariant vector fields of a Lie group \cite{Ha15}. In particular, $\mathfrak{sl}_2$ is isomorphic to the Lie algebra of vector fields of the Lie group, $SL_2$, of unimodular real $2\times 2$ matrices. Then,  the elements of each $\Lambda^k\mathfrak{sl}_2$ can be identified with left-invariant $k$-vector fields on $SL_2$. In particular, $e_1,e_2,e_3$ can be understood as left-invariant vector fields on $SL_2$ and, in view of the formula for the Schouten bracket, it reduces to the Lie bracket of elements of $\mathfrak{sl}_2$, namely
$$
[e_1,e_2]_S=e_1,\qquad [e_1,e_3]_S=2e_2,\qquad 
[e_2,e_3]_S=e_3.
$$
Meanwhile, applying the property 3. of Proposition \ref{Pr:PropSchou}, we get
$$
[e_1,e_2\wedge e_3]_S=[e_1,e_2]_S\wedge e_3+e_2\wedge[e_1,e_3]_S=e_1\wedge e_3+e_2\wedge 2e_2=e_1\wedge e_3.
$$
Similarly, we obtain
$$
[e_1, e_1 \wedge e_2 \wedge e_3]_S = [e_1, e_1] \wedge e_2 \wedge e_3 + e_1 \wedge [e_1, e_2 \wedge e_3]_S = e_1 \wedge e_1 \wedge e_3 = 0
$$
One can proceed analogously with the remaining multivectors.
\end{example}

\section{Generalised distributions}\label{Sec:StSus}

A {\it generalised distribution} (also called a {\it Stefan--Sussmann distribution}) on $M$ is a correspondence $\mathcal{D}$ attaching to each $x\in M$ a subspace $\mathcal{D}_x\subset T_xM$. We say that  $\mathcal{D}$ is {\it differentiable} if for every point $x \in M$ there exist smooth vector fields $X_1, \ldots, X_p$ defined on some open neighbourhood $U$ of $x$ so that they take values in $\mathcal{D}$ at points of $U$, i.e. $X_1(x'),\ldots,X_p(x')\in \mathcal{D}_{x'}$ for every $x'\in U$, and they span the distribution at $p$, namely $\mathcal{D}_x = \langle X_1(x), \ldots, X_p(x)\rangle$. We stress that, from now on, we assume generalised distributions to be differentiables.

We call {\it rank} of $\mathcal{D}$ at $x$ the dimension of $\mathcal{D}_x$. As a function $\rho:x\in  M \to \dim \mathcal{D}_x\in \mathbb{N}$, the rank of a differentiable distribution is lower-semicontinuous (for a given topological space $X$, a function $f: X \to \mathbb{R}$ is called lower-semicontinuous at $x$ if for every $\epsilon > 0$ there exists a neighbourhood $U$ of $x$ such that any point $y \in U$ satisfies $f(y) > f(x) - \epsilon$.). Indeed, since $\mathcal{D}$ is differentiable, it is spanned by a finite family of vector fields being linearly independent at $x$ and taking values in $\mathcal{D}$. This ensures that, for every $x \in M$, there exists an open neighbourhood $U$ of $x$ such that $\rho(y)\geq \rho(x)$ for every $y \in U$. A generalised distribution need not have the same rank at every point of $M$. If $\mathcal{D}$ has the same rank at every point of $M$, then  $\mathcal{D}$ is said to be {\it regular} or $\mathcal{D}$ is simply  called a {\it distribution}. Otherwise, $\mathcal{D}$ is said to be {\it singular}. A generalised distribution $\mathcal{D}$ on $M$ is \emph{involutive} if every two vector fields taking values in $\mathcal{D}$ satisfy that their Lie bracket takes values in $\mathcal{D}$ as well. 

\begin{example}\label{Ex:GD1}
Let us consider the generalised distribution $\mathcal{D}$ on $\mathbb{R}^2$ of the form
\begin{equation*}
\mathcal{D}_{(x, y)} = 
\begin{cases}
\langle\frac{\partial}{\partial x}\rangle, & y \leq 0, \\
\langle\frac{\partial}{\partial x}, \exp\left(-\frac{1}{y^2}\right) \frac{\partial}{\partial y}\rangle, & y > 0,
\end{cases}
\end{equation*}
for every $(x,y)\in \mathbb{R}^2$. The rank of $\mathcal{D}$ is equal  to $2$ for points $(x,y)$ with $y > 0$, while  $\rho(x,y) = 1$ for  points with $y \leq  0$. Therefore, $\mathcal{D}$ is singular.
\end{example}

\begin{example}\label{Ex:GD2}
Let us analyse the generalised distribution $\mathcal{D}$ on $\mathbb{R}^2$ given by 
\begin{equation*}
\mathcal{D}_{(x, y)} = 
\begin{cases}
\langle \frac{\partial}{\partial x}\rangle, & y \leq 0, \\
\langle \frac{\partial}{\partial x}, \exp\left(-\frac{1}{x^2}\right) \frac{\partial}{\partial y}\rangle, & y > 0,
\end{cases}
\end{equation*}
By a similar discussion as the one in Example \ref{Ex:GD1}, we obtain  that $\mathcal{D}$ is singular.
\end{example}
A \emph{stratification} $\mathcal{F}$ on a manifold $M$ is a partition of $M$ into connected disjoint immersed submanifolds $\{\mathcal{F}_k\}_{k\in I}$, where $I$ is a certain set of indices, i.e. $M=\bigcup_{k\in I}\mathcal{F}_k$ and   the submanifolds $\{\mathcal{F}_k\}_{k\in I}$ satisfy $\mathcal{F}_{k}\cap \mathcal{F}_{k'}=\emptyset$ for $k\neq k'$ and $k,k'\in I$. The connected immersed submanifolds $\mathcal{F}_k$, with $k\in I$, are called the {\it strata} of the stratification. A stratification is \emph{regular} if its strata are immersed submanifolds of the same dimension, whereas it is \emph{singular} otherwise. Regular stratifications are called {\it foliations} and their strata are called {\it leaves}.  The tangent space to a stratum, $\mathcal{F}_k$, of a stratification passing through a point $x\in M$ is a subspace $\mathcal{D}_x\subset T_xM$. All the subspaces $\mathcal{D}_x\subset T_xM$ for every point $x\in M$ give rise to a {generalised distribution}  $\mathcal{D}:=\bigcup_{x\in M}\mathcal{D}_x$ on $M$. All the leaves of a foliation have the same dimension and, therefore, the generalised distribution formed by the tangent spaces at every point to its leaves is regular. Meanwhile, a singular stratification gives rise to a singular generalised distribution. 

In this work, we are specially interested in generalised distributions generated by finite-dimensional Lie algebras of vector fields, the so-called {\it Vessiot--Guldberg Lie algebras} \cite{LS20}.  More specifically, let $V$ be a Vessiot--Guldberg Lie algebra, the vector fields of $V$ span a generalised distribution $\mathcal{D}^V$ given by
 $$
	\mathcal{D}_x^V:=\{X_x:X\in V\}\subset T_xM,\qquad \forall x\in M.
	$$

Since the space of vector fields tangent to the strata of a stratification are closed under Lie brackets, the Lie bracket of vector fields on $M$ taking values in a distribution $\mathcal{D}$ can be restricted to each one of its strata. A generalised distribution $\mathcal{D}$ on $M$ is \emph{integrable} if there exists a stratification $\mathcal{F}$ on $M$ such that each stratum $\mathcal{F}_k$ thereof satisfies $T_x\mathcal{F}_k=\mathcal{D}_x$ for every $x\in \mathcal{F}_k$.

A relevant question is whether a generalised distribution on $M$ is integrable or not. For regular distributions, the Fr\"obenius theorem holds \cite{Fr77,La18}.

\begin{theorem} 
	If $\mathcal{D}$ is a regular generalised distribution on a manifold $M$, then $\mathcal{D}$ is integrable if and only if it is involutive.
\end{theorem}

In case of a generalised distribution, involutivity is a natural necessary condition for its integrability because the space of vector fields tangent to any stratum of a stratification is involutive. Nevertheless, if a generalised distribution is not regular, its involutiveness is not sufficient.

\begin{example}
We have already verified that the distribution in Example \ref{Ex:GD1} is singular. Let us see now whether it is integrable. To do so, it is enough to give a stratification whose tangent spaces to its strata. In the region $y \leq 0$, the strata are the straight lines given by $y = const.$ These are indeed connected disjoint immersed submanifolds. There is an additional stratum given by the whole half plane for which $y>0$. It is simple to see that the tangent space to these strata give the generalised distribution $\mathcal{D}$.

The distribution in Example \ref{Ex:GD2} is not integrable. As before, the lower half-plane (with $y \leq 0$) is stratified by the straight lines $y = const.$ Now, let us look for the stratification of the upper half-plane. For the quadrants $x>0, y > 0$ and $x < 0, y > 0$, we can take the quadrants themselves as the strata. The remaining $y$-axis, however, poses a problem. The distribution at these points is spanned by $\partial_x$. Since the tangent space to a point is zero, we cannot stratify it by separate points. If on the other hand, we consider an open interval, the tangent space would have the $\partial_y$-component. Therefore, any available option contradicts $T_x F_x = \mathcal{D}_x$.

\end{example}

Remarkably, integrability of generalised distributions can be asserted under additional conditions. Let us note two relevant results on this issue (for a detailed exposition of integrability theorems we refer to \cite{La18}). If a generalised distribution $\mathcal{D}$ on $M$ is {\it analytical}, i.e. for every $x\in M$ there exists a family of analytical vector fields taking values in $\mathcal{D}$ and spanning $\mathcal{D}_{x'}$ for every $x'$ in an open neighbourhood of $x$, one has the following proposition. 

 \begin{theorem}\emph{\textbf{(Nagano \cite{Na66})}}\label{Th:nagano}
	Let $M$ be a real analytic manifold, and let $V$ be a sub-Lie algebra of analytic vector fields on $M$. Then, the induced analytic distribution $\mathcal{D}^V$ is integrable.
\end{theorem}
In Example 2, the distribution is spanned by the vector fields
$$
X=\frac{\partial}{\partial x},\qquad Y=\exp(-1/y^2)\frac{\partial}{\partial y}
$$
which generate a two-dimensional abelian Lie algebra. However, since $Y$ is not an analytical vector field, Theorem \ref{Th:nagano} cannot ensure its integrability.

If a generalised distribution is not analytical, sufficient conditions for integrability are established by the following {\it Stefan-Sussmann theorem}. This result will be of great relevance in this thesis.

\begin{theorem}{\bf (Hermann \cite{He62})}\label{Th:StSus} 
Let $M$ be a smooth manifold. If $V$ is a  finite-dimensional Lie subalgebra of $\mathfrak{X}(M)$, then the distribution $\mathcal{D}^V$ is integrable.
\end{theorem}
 
As already noticed, vector fields $X, Y$ in Example 2 generate a two-dimensional Lie algebra. Thus in view of Theorem \ref{Th:StSus}, the induced distribution is integrable.

Once we assure the integrability of the generalised distribution, the next relevant problem concerns the determination of its strata. In case of a distribution $\mathcal{D}^V$ on a manifold $M$ spanned by the Vessiot--Guldberg Lie algebra $V$, it is known (see \cite{St80} for details) that the stratum $\mathcal{F}_k$ of $\mathcal{D}^V$ passing through $x \in M$ is given by the set of points
\begin{equation}\label{eq:dec}
\mathcal{F}_k = \{y \in M: \exists t_1, \ldots, t_s \in \mathbb{R}, \, y = \exp(t_1 X_1)\circ \exp(t_2 X_2)\circ \ldots \circ \exp(t_s X_s)x\},
\end{equation}
where $s$ is any natural number, $X_1,\ldots,X_s$ are any vector fields of  $V$, and each $\exp(X)$ stands for the local diffeomorphism on $M$ induced by the vector field $X$ on $M$, i.e. it maps $x$ onto the point $\gamma(1)$ of the integral curve $\gamma:]-\epsilon,\epsilon[\subset \mathbb{R}\mapsto M$ of $X$ passing through $x$.

\section{Lie algebras, $\mathfrak{g}$-modules, and Grassmann algebras}

For a vector space $V$, let $GL(V)$ and $\mathfrak{gl}(V)$ stand for the Lie group of automorphisms and the Lie algebra of endomorphisms on $V$, respectively. A {\it $\mathfrak{g}$-module} is a pair $(V,\rho)$, where  $\rho:v\in \mathfrak{g}\mapsto \rho_v\in \mathfrak{gl}(V)$ is a Lie algebra homomorphism. A $\mathfrak{g}$-module $(V,\rho)$ will be represented just by $V$, while $\rho_v(x)$, for any $v\in \mathfrak{g}$ and $x\in V$, will be written simply as $vx$  if $\rho$ is understood from context. If ${\rm Im}(\rho)$ leaves invariant a certain subspace $W \subset V$, then one can consider the restriction $\rho_W: \mathfrak{g} \to \mathfrak{gl}(W)$. A pair $(W, \rho_W)$ is called a {\it $\mathfrak{g}$-submodule}.
	
	\begin{example}\label{ex:Lie_alg_der}
	Let ${\rm ad}: v \in \mathfrak{g}\mapsto [v,\cdot]_\mathfrak{g} \in \mathfrak{gl}(\mathfrak{g})$ be the adjoint representation of $\mathfrak{g}$. Then, $(\mathfrak{g},{\rm ad})$ is a $\mathfrak{g}$-module \cite{FH91}. Since each $[v,\cdot]_\mathfrak{g}$, with $v\in\mathfrak{g}$, is a derivation of the Lie algebra $\mathfrak{g}$ \cite{FH91}, the map ${\rm ad}$ can be considered as a mapping ${\rm ad}:\mathfrak{g}\rightarrow \mathfrak{der}(\mathfrak{g})$, where $\mathfrak{der}(\mathfrak{g})$ is the Lie algebra of derivations on $\mathfrak{g}$.
	\end{example}
	
	\begin{example}\label{ex:Lie_alg_aut} The group ${\rm Aut}(\mathfrak{g})$ of Lie algebra automorphisms of $\mathfrak{g}$ admits a Lie group structure \cite{SW73} and its Lie algebra is denoted by $\mathfrak{aut}(\mathfrak{g})$. The tangent map at ${\rm id}_\mathfrak{g}\in {\rm Aut}(\mathfrak{g})$ to the injection $\iota:{\rm Aut}(\mathfrak{g})\hookrightarrow GL(\mathfrak{g})$ induces a Lie algebra morphism $\widehat {\rm ad}:\mathfrak{aut}(\mathfrak{g})\simeq {\rm T}_{{\rm id}_\mathfrak{g}}{\rm Aut}(\mathfrak{g})\rightarrow \mathfrak{gl}(\mathfrak{g})\simeq {\rm T}_{{\rm id}_\mathfrak{g}}GL(\mathfrak{g})$ and $(\mathfrak{g},\widehat{\rm ad})$ becomes an $\mathfrak{aut}(\mathfrak{g}$)-module. 
	\end{example}
	
Given a $\mathfrak{g}$-module $(V, \rho)$, a Grassmann algebra $\Lambda\mathfrak{g}$ and its specific subspaces can also be equipped with the $\mathfrak{g}$-module structure.
	
	\begin{proposition} \label{Prop:grass_gmod}
	If $(V, \rho)$ is a $\mathfrak{g}$-module, each  $(\Lambda^m V,\Lambda^m\rho)$, where  $\Lambda^m \rho: v \in \mathfrak{g} \mapsto \Lambda^m \rho_v \in \mathfrak{gl}(\Lambda^m V)$, and  $(\Lambda V,\Lambda \rho:v\in \mathfrak{g}\mapsto \Lambda\rho_v\in \mathfrak{gl}(\Lambda V))$ are $\mathfrak{g}$-modules.
	\end{proposition}	
	\begin{proof}
Let us firstly focus on the case $(\Lambda^m V,\Lambda^m\rho)$. The goal is to show that the map $\Lambda^m\rho: v \in \mathfrak{g} \mapsto \Lambda^m \rho_v \in \mathfrak{gl}(\Lambda^m V)$ is a Lie algebra homomorphism, i.e. $\Lambda^m\rho([v, w]) = [\Lambda^m\rho(v), \Lambda^m\rho(w)]$. 

Recall that every $T\in \mathfrak{gl}(V)$ gives rise to the mappings $\Lambda^mT \in \mathfrak{gl}(\Lambda^m V)$ such that $\Lambda^mT \equiv 0$ for $m\leq 0$, and for $m>0$ it is of the form
	$
	\Lambda^mT := \sum_{i=1}^{m} T_i,
	$ 
	where $T_i \in \mathfrak{gl}(\Lambda^m V)$ is given by $T_i(v_1, \otimes \ldots \otimes v_m):=v_1 \otimes \ldots \otimes T(v_i) \otimes \ldots \otimes v_m$ for any $v_1, \ldots, v_m \in V$ and $\Lambda^mT$ is assumed to be restricted to $\Lambda^mV$. Given two maps $S, T \in \mathfrak{gl}(V)$, we have $[S_i, T_j] = 0$ for $i \neq j$.

 In view of that, we get
 $$
 [\Lambda^m\rho(v), \Lambda^m\rho(w)] = \sum_{i,j =1}^{m} [(\rho_v)_i, (\rho_w)_j] = \sum_{i=1}^{m} [(\rho_v)_i, (\rho_w)_i] = \sum_{i=1}^{m} (\rho_{[v,w]})_i = \Lambda^m\rho([v, w])
 $$
 Denote $\Lambda T:=\bigoplus_{m\in \mathbb{Z}}\Lambda^mT\in\mathfrak{gl}(\Lambda V)$. Then, the second part of the proposition follows immediately from previous considerations.
 \end{proof}
 By the properties of the algebraic Schouten bracket, Proposition \ref{Prop:grass_gmod} implies that any $\mathfrak{g}$ gives rise to a $\mathfrak{g}$-module  $(\Lambda\mathfrak{g},{\rm ad})$, where ${\rm ad}:v\in \mathfrak{g}\mapsto [v,\cdot]_S \in \mathfrak{gl}(\Lambda \mathfrak{g})$ (cf. \cite{Va94}). 
 Similarly to Proposition \ref{Prop:grass_gmod}, one can show that the k-tensor subspaces $T^k(V)$ with the map $\otimes^k T := \sum_{p=1}^k T_p$ and the whole tensor algebra $T(V)$ with the map $\otimes T := \oplus_{m \in \mathbb{Z}} \otimes^k T$ are also $\mathfrak{g}$-modules. 

Let us now focus on certain properties of the Lie algebra representation $\rho$, which will be relevant for our classification methods later in this work. 
	
	\noindent
	\begin{minipage}{0.64\textwidth}
		\begin{lemma}\label{Lio} 
   \setlength{\baselineskip}{14pt}
  Let $(V,\rho)$ be a $\mathfrak{g}$-module and let $G$ be a connected Lie group with Lie algebra $\mathfrak{g}$. If $\Phi:G\rightarrow GL(V)$ is a Lie group homomorphism such that the diagram (\ref{diag}) is commutative,  where $\exp_G$ and $\exp$ are exponential maps on $\mathfrak{g}$ and $\mathfrak{gl}(V)$ respectively, then $\Phi(G)$ is an immersed Lie subgroup of $GL(V)$ generated by the elements $\exp(\rho(\mathfrak{g}))$.
		\end{lemma}
  \vskip0.5cm
	\end{minipage}
	\begin{minipage}{0.35\textwidth}
		\begin{center}
			\vskip-0.5cm
			\begin{equation}\label{diag}		
   \begin{gathered}
   \xymatrix{\mathfrak{g}\ar[rr]^{\rho\quad\quad}\ar[d]^{\exp_G}&&\mathfrak{gl}(V)\ar[d]^{\exp}\\G\ar[rr]^{\Phi}&&GL(V)}
			\end{gathered}
   \end{equation}
		\end{center}
	\end{minipage}

 This result stems from the fact that every element of a connected Lie group is a product of elements in the image of its exponential map (see \cite[p. 228]{Sc94}).
	
A Lie algebra homomorphism $\rho:\mathfrak{g}\rightarrow \mathfrak{gl}(V)$ gives rise to a Lie group homomorphism $\Phi:\widetilde{G}\rightarrow GL(V)$, where $\widetilde{G}$ is  connected and simply connected, so that the associated diagram of the form (\ref{diag}) is commutative \cite{DK00}. Thus, Lemma \ref{Lio} asserts $\Phi(\widetilde{G})$ is generated by the elements of $\exp(\rho(\mathfrak{g}))$. In other words, $\Phi(\widetilde{G})$ is the smallest group containing $\exp(\rho(\mathfrak{g}))$. 
Although $\Phi(\widetilde{G})$ may not be an embedded submanifold in $GL(V)$, it is always a Lie group \cite{L03}. Previous facts motivate the following definition.

	\begin{definition}
		The {\it Lie group} of a $\mathfrak{g}$-module $(V, \rho)$ is the immersed Lie subgroup $GL(\rho)$ of $GL(V)$ generated by $\exp (\rho(\mathfrak{g}))$. 
	\end{definition}

For further purposes, it is necessary to establish the form of Lie groups for $\mathfrak{g}$-modules discussed in Examples \ref{ex:Lie_alg_der} and \ref{ex:Lie_alg_aut}. The next two propositions show that in case of the $\mathfrak{g}$-module $(\mathfrak{g},{\rm ad})$ and the $\mathfrak{aut}(\mathfrak{g})$-module $(\mathfrak{g},\widehat{\rm ad})$, their Lie groups are  ${\rm Inn}(\mathfrak{g})$ and ${\rm Aut}_c(\mathfrak{g})$, respectively. 

	\begin{proposition}\label{Inn} 
	Let ${\rm Ad}:g\in G\mapsto {\rm Ad}_g\in GL(\mathfrak{g})$ be the adjoint action of a connected Lie group $G$ on its Lie algebra $\mathfrak{g}$. The Lie group of the $\mathfrak{g}$-module $(\mathfrak{g},{\rm ad})$ is  equal to ${\rm Ad}(G)$.
	\end{proposition}
	
\begin{proposition}\label{Out} 
 The Lie group of the $\mathfrak{aut}(\mathfrak{g})$-module $(\mathfrak{g},\widehat {\rm ad})$ is given by the connected component, ${\rm Aut}_c(\mathfrak{g})$, containing the neutral element of ${\rm Aut}(\mathfrak{g})$.
\end{proposition}

	\begin{proof} 
		The inclusion $\iota:{\rm Aut}_c(\mathfrak{g})\hookrightarrow GL(\mathfrak{g})$ has a tangent map $\widehat{{\rm ad}}:\mathfrak{aut}(\mathfrak{g})\rightarrow \mathfrak{gl}(\mathfrak{g})$ at ${\rm id}_\mathfrak{g}\in {\rm Aut}_c(\mathfrak{g})$, which 
		\vskip 0.1cm
		\noindent
		\begin{minipage}{0.6\textwidth}
  \setlength{\baselineskip}{14pt}
is a Lie algebra morphism. Thus, the right part of the diagram aside commutes. Let $\widetilde{{\rm Aut}}(\mathfrak{g})$ be the connected and simply connected Lie group associated with $\mathfrak{aut}(\mathfrak{g})$ and let $\widetilde{\iota}:\widetilde{{\rm Aut}}(\mathfrak{g})\rightarrow GL(\mathfrak{g})$ be the induced Lie group morphism. The properties of $\widetilde{\iota}$ for $\widetilde{\rm Aut}(\mathfrak{g})$ and $\widehat{\rm ad}$ imply the commutativity of
		\end{minipage}
		\begin{minipage}{0.4\textwidth}
			\vskip-0.4cm
   \centering
		\begin{equation*}
					\xymatrix{\mathfrak{aut}(\mathfrak{g})\ar[r]^(0.5){\widehat{\rm ad}}\ar[d]^{\exp_{\widetilde{{\rm Aut}}(\mathfrak{g})}}&\mathfrak{gl}(\mathfrak{g})\ar[d]^{\exp}&\mathfrak{aut}(\mathfrak{g})\ar[l]_{\widehat{\rm ad}}\ar[d]_{\exp_{{{\rm Aut}}_c(\mathfrak{g})}}\\\widetilde{{\rm Aut}}(\mathfrak{g})\ar[r]^{\widetilde{{\iota}}}&GL(\mathfrak{g})&{\rm Aut}_c(\mathfrak{g})\ar@{->}[l]_{\iota}}
				\end{equation*}
		\end{minipage}
  \noindent
 the left part of the diagram aside. This fact, together with Lemma \ref{Lio}, gives $GL(\widehat{\rm ad})=\widetilde{\iota}(\widetilde{{\rm Aut}}(\mathfrak{g}))={\rm Aut}_c(\mathfrak{g})$.
	\end{proof}
	
 We have shown in Proposition \ref{Prop:grass_gmod} that a $\mathfrak{g}$-module $(V,\rho)$ induces new ones $(\Lambda^m V,\Lambda^m\rho)$. Similarly close relation holds between their Lie groups $GL(\rho)$ and $GL(\Lambda^m\rho)$ as explained next.
	
	\begin{proposition}\label{Prop:invk}
	Let $(V,\rho)$ be a $\mathfrak{g}$-module. For any $m\in \mathbb{Z}$, the Lie group $GL(\Lambda^m\rho)$ of the $\mathfrak{g}$-module $(\Lambda^m V,\Lambda^m\rho)$ is given as $GL(\Lambda^m\rho)\!=\!\left\{T^{\otimes m} \in GL(\Lambda^m V) : T\in GL(\rho)\right\}.$
		
	\end{proposition}
	\begin{proof}
 Assume first $m>0$. By the discussion in the proof of Proposition \ref{Prop:grass_gmod}, it follows that $[T_i,T_j]=0$ for $i,j=1,\ldots,m$, and $T\in \mathfrak{gl}(\mathfrak{g})$. Thus,
		$$
		\exp(\Lambda^m \rho_v) \!=\! \exp \left(\sum_{i=1}^{m} (\rho_v)_i \right)=\exp(\rho_v)^{\otimes m}
		$$
		for all $v\in \mathfrak{g}.$ Hence, $GL(\Lambda^m \rho)$ is generated by the composition of  operators $T^{\otimes m}$, where $T$ is a composition of operators $\exp(\rho_v)$ with $v\in \mathfrak{g}$. As the maps $\exp(\rho_v)$, for every $v\in \mathfrak{g}$, generate $GL(\rho)$, then $T$ is any element of $GL(\rho)$, which finishes the proof for $m>0$. The case $m\leq 0$ is immediate. 
	\end{proof}

For convenience, let us introduce the following notation. If $T \in GL(V)$, we define $\Lambda^m T:=T^{\otimes m}$ for $m\geq 1$ and $\Lambda^m T$ is the identity on $\Lambda^mV$ for $m\leq 0$. Then, we conclude by Proposition \ref{Prop:invk} that a group action $\Phi:g\in G\mapsto \Phi_g\in GL(V)$ leads to a new one $\Lambda^m\Phi: g\in G\mapsto \Lambda^m\Phi_g\in GL(\Lambda^mV)$ for $m\in \mathbb{Z}$. 
 
	\noindent
 \begin{minipage}{0.6\textwidth}
 \setlength{\baselineskip}{14pt}
Let $\mathfrak{inn}(\mathfrak{g})$ be the Lie algebra of ${\rm Inn}(\mathfrak{g})$. Note that $\mathfrak{inn}(\mathfrak{g})$ and ${\rm Inn}(\mathfrak{g})$ can be naturally embedded in $\mathfrak{gl}(\mathfrak{g})$ and $GL(\mathfrak{g})$, respectively. Let us denote both embeddings by $\iota_1$ and $\iota_2$. Previous comments and the fact that ${\rm Ad}(G)={\rm Inn}(\mathfrak{g})$ allow us to extend the diagram of Proposition \ref{Inn} as shown aside. 
  \end{minipage}	
 \begin{minipage}{0.4\textwidth}
 \vskip-0.5cm
	\begin{center}
			\begin{equation*}
		\xymatrix@R=5.3mm@C=1.3cm{\mathfrak{g}\ar[r]^(0.5){\Lambda^m{\rm ad}}\ar[d]^{\exp_{{G}}}&\mathfrak{gl}(\Lambda^m\mathfrak{g})\ar[d]^{\exp}&\mathfrak{inn}(\mathfrak{g})\ar[l]_{\Lambda^m\iota_1}\ar[d]_{\exp_{{{\rm Inn}}(\mathfrak{g})}}\\{G}\ar[r]^{{\Lambda^m{\rm Ad}}}&GL(\Lambda^m\mathfrak{g})&{\rm Inn}(\mathfrak{g})\ar@{->}[l]_{{\Lambda^m \iota_2}}}
			\end{equation*}
		\end{center}
			\end{minipage}

In this thesis, ${\rm Inn}(\mathfrak{g})$-actions on $\Lambda^2 \mathfrak{g}$ and $\Lambda^3 \mathfrak{g}$ are of particular interest, since further studies of Lie bialgebra equivalence rely on the analysis of their orbits. Especially, it is necessary to know their dimension. The following Proposition presents a simple way to compute it.

	\begin{proposition}\label{proporb}
Let $G$ be a connected Lie group associated with a Lie algebra $\mathfrak{g}$ and let $\mathcal{O}^{(m)}_w$ denote the orbit of ${\rm Inn}(\mathfrak{g})$-action on $\Lambda^m\mathfrak{g}$ through $w\in \Lambda^m\mathfrak{g}$. Then,
$
\dim \mathcal{O}^{(m)}_w = \dim  {\rm Im}\,\Theta^m_w,
		$
where $\Theta^{(m)}_w:v\in \mathfrak{inn}(\mathfrak{g})\mapsto [v, w]  \in \Lambda^m\mathfrak{g}$.
	\end{proposition}
	\begin{proof}
		The orbit $\mathcal{O}^{(m)}_w$ of $w\in\Lambda^m \mathfrak{g}$ relative to ${\rm Inn}(\mathfrak{g})$ is given by the set $\{\Lambda^m {\rm Ad}_g w: g\in G\}$. Define $g_t := \exp(tv)$, $g_1 :=g := \exp(v)$ for $v\in\mathfrak{g}$. 
		Then, $\dim \mathcal{O}^{(m)}_w = \dim {G}\cdot w=\dim (\mathfrak{g})-\dim (\mathfrak{g}_w)$, where the Lie algebra $\mathfrak{g}_w$ of $G_w$, the isotropy group of $w\in \Lambda^m\mathfrak{g}$, is given by those $v\in\mathfrak{g}$ such that 
		$\frac{d}{dt}\big|_{t=0}\Lambda^m{\rm Ad}_{g_t}(w)=[v,w] =0$. This amounts to $v \in \ker\,\Theta^{m}_w$.
		Hence, $\dim \mathcal{O}^{(m)}_w = \dim \mathfrak{g}-\dim \mathfrak{g}_w=\dim {\rm Im}\, \Theta_{w}^{m}$.
	\end{proof}

A similar result, which is a rather straightforward generalisation of Proposition \ref{proporb}, can be proved for ${\rm Aut}(\mathfrak{g})$-actions and their orbits.

\begin{proposition}\label{prop:derbiv}
	The dimension of the orbit $\mathscr{O}_w$ of the action of ${\rm Aut}(\mathfrak{g})$ on $\Lambda^m\mathfrak{g}$ through $w\in \Lambda^m\mathfrak{g}$ is 
	$
	\dim  {\rm Im}\,\Upsilon^m_w,
	$
	where $\Upsilon^m_w:d\in \mathfrak{der}(\mathfrak{g})\mapsto (\Lambda^md)(w) \in \Lambda^m\mathfrak{g}$.
\end{proposition}
\begin{proof}
	The orbit of $w\in\Lambda^m \mathfrak{g}$ relative to ${\rm Aut}(\mathfrak{g})$ is given by the points $ (\Lambda^mT)( w)$ for every $T\in {\rm Aut}(\mathfrak{g})$. Define $\exp(td ) =:T_t$, with $t\in \mathbb{R}$, for $d\in\mathfrak{der}(\mathfrak{g})$. Then, the tangent space at $w$ of $\mathscr{O}_w$  is spanned by the tangent vectors
	\begin{equation}\label{Eq:Red}
	\frac{d}{dt}\bigg|_{t=0}[\Lambda^m\exp(td)](w)=(\Lambda^md)(w)\in T_w\Lambda^m\mathfrak{g}\simeq \Lambda^m\mathfrak{g}.
	\end{equation}
	Then, $\dim \mathscr{O}_w=\dim \mathfrak{der}(\mathfrak{g})-\dim \mathfrak{g}_w$, where $G_w$ is the isotropy group of $w\in \Lambda^m\mathfrak{g}$ relative to the action of ${\rm Aut}(\mathfrak{g})$ and $\mathfrak{g}_w$ is the Lie algebra of $G_w$. Moreover, $\mathfrak{g}_w$ is given by those $d\in\mathfrak{der}(\mathfrak{g})$ such that  $(\Lambda^md)(w)=0$.
	This amounts to $d\in \ker\,\Upsilon^{m}_w$.
	Hence, $\dim \mathscr{O}_w = \dim \mathfrak{der}(\mathfrak{g})-\dim \mathfrak{g}_w=\dim {\rm Im}\, \Upsilon_{w}^{m}$.
\end{proof}
\begin{corollary}\label{Re:DerAlg}
The fundamental vector fields of a Lie group action $\varphi:G\times N\rightarrow N$ are defined as $X_v(x) :=\frac{d}{dt}\big|_{t=0}\varphi(\exp(-tv),x)$ for every $v\in T_eG$ and $x\in N$. As a consequence of Proposition \ref{prop:derbiv}, the fundamental vector fields of the natural Lie group action of ${\rm Aut}(\mathfrak{g})$ on $\Lambda^2 \mathfrak{g}$ are spanned by $X^2_v(w) = \sum_{i=1}^s[\Upsilon^2_w]_i \partial_{x_i}$ for arbitrary $w \in \Lambda^2\mathfrak{g}$ and $v \in \mathfrak{der}(\mathfrak{g})$, given any linear coordinate system $\{x_1, \ldots, x_s\}$ on $\Lambda^2 \mathfrak{g}$. Moreover, the coordinates of $X^2_v$ are the coordinates of the extensions to $\Lambda^2\mathfrak{g}$ of the derivations of $\mathfrak{g}$, i.e. linear maps $d \in \mathfrak{gl}(\mathfrak{g})$ satisfying $d([e_1,e_2])=[d(e_1),e_2]+[e_1,d(e_2)]$ for all $e_1,e_2 \in \mathfrak{g}$.
\end{corollary}

Finally, let us shortly mention about a special kind of maps associated with $\mathfrak{g}$-modules. 
	\begin{definition}
	Let us consider a $\mathfrak{g}$-module $(V, \rho)$. A {\it trace form} associated with $\mathfrak{g}$ is the bilinear symmetric mapping ${\rm tr}_\rho$ on $V$ given by ${\rm tr}_\rho(v_1,v_2):={\rm tr}(\rho(v_1)\circ \rho(v_2))$ for any $v_1, v_2 \in V$.
	\end{definition}
	
	\begin{lemma} \label{Lem:trace_form}
 Given a $\mathfrak{g}$-module $(V,\rho)$, its associated trace form is ${\rm ad}$-invariant, i.e. 
	$$
	{\rm tr}_\rho([v_1,v_2],v_3)+	{\rm tr}_\rho(v_2,[v_1,v_3])=0,\qquad \forall v,v',v''\in \mathfrak{g}.
	$$
	Moreover, $\ker {\rm tr}_\rho$ is an ideal of $\mathfrak{g}$. 
	\end{lemma}
	\begin{proof} 
Since $\rho: \mathfrak{g} \to \mathfrak{gl}(V)$ is an algebra homomoprhism and by the cyclicity of the trace, we get
\begin{equation*}
\begin{split}
&{\rm tr}_\rho([v_1,v_2],v_3)+	{\rm tr}_\rho(v_2,[v_1,v_3]) = {\rm tr}(\rho([v_1, v_2])\circ \rho(v_3)) + {\rm tr}(\rho(v_2)\circ \rho([v_1, v_3])) \\
&= {\rm tr}((\rho(v_1) \circ \rho(v_2) - \rho(v_2) \circ \rho(v_1))\circ \rho(v_3)) + {\rm tr}(\rho(v_2)\circ (\rho(v_1) \circ \rho(v_3) - \rho(v_3) \circ \rho(v_1))) \\
&= {\rm tr}(\rho(v_1) \circ \rho(v_2) \circ \rho(v_3)) - {\rm tr}(\rho(v_2) \circ \rho(v_1) \circ \rho(v_3)) + {\rm tr}(\rho(v_2)\circ \rho(v_1) \circ \rho(v_3)) - {\rm tr}(\rho(v_2)\circ \rho(v_3) \circ \rho(v_1)) \\
&= 0
\end{split}
\end{equation*}
 
Recall that the kernel of a bilinear form $b$ on V is a subspace $\ker b \subset V$ such that $b(v,w) = 0$ for any $v \in V$ and $w \in \ker b$ \cite{L02}. Proving the second part of our lemma amounts to showing that if $v\in \ker {\rm tr}_\rho$ and $w\in \mathfrak{g}$, then $[w,v]\in \ker {\rm tr}_{\rho}$. From the definition of the trace form, for any $g \in \mathfrak{g}$ we obtain
 \begin{equation*}
\begin{split}
 &{\rm tr}_{\rho}([w, v], g) = {\rm tr}(\rho([w, v]), \rho(g)) = {\rm tr}(\rho(w) \circ \rho(v) \circ \rho(g) - \rho(v) \circ \rho(w) \circ \rho(g)) \\
 &= {\rm tr}(\rho(w) \circ \rho(v) \circ \rho(g)) - {\rm tr}(\rho(v) \circ \rho(w) \circ \rho(g)) = {\rm tr}(\rho(v) \circ \rho(g) \circ \rho(w)) - {\rm tr}(\rho(v) \circ \rho(w) \circ \rho(g)) \\
  &= {\rm tr}(\rho(v) \circ (\rho(g) \circ \rho(w) - \rho(w) \circ \rho(w))) = {\rm tr}(\rho(v) \circ \rho([g, w])) = {\rm tr}_{\rho}(v, [g, w]) = 0
 \end{split}
\end{equation*}
Thus, $[w,v] \in \ker {\rm tr}_{\rho}$.
\end{proof}
 
\section{The Chevalley-Eilenberg cohomology of Lie algebras}

Following the standard exposition in \cite{Jacobson}, we present the construction of Lie algebra cohomology. For a complex Lie algebra $\mathfrak{g}$, let $(M, \rho)$ be a $\mathfrak{g}$-module. For any $n \in \mathbb{N}$, denote the space of $n$-cochains as $C^n(\mathfrak{g},M) := {\rm Hom}_{\mathbb{C}}(\Lambda^n \mathfrak{g}, M)$, consisting of skew-symmetric n-multilinear functions on $\mathfrak{g}^{\times n}$ with values in $M$. We understand $C^0(\mathfrak{g}, M)$ as the set of constant functions from $\mathfrak{g}$ to a fixed element of $M$. For $n \geq 1$, the coboundary map ${\rm d_n}: C^n(\mathfrak{g},M) \to C^{n+1}(\mathfrak{g},M)$ is defined as
\begin{equation}\label{cob_map}
\begin{split}
{\rm d_n}f(x_1, \ldots, x_n, x_{n+1}) := &\sum_{p=1}^{n+1} (-1)^{p} \rho_{x_p} f(x_1, \ldots, \hat{x_i}, \ldots, x_{n+1}) \\
&+ \sum_{1 \leq i < j \leq n+1} (-1)^{i+j} f([x_i, x_j], x_1, \ldots, \hat{x_i}, \hat{x_j}, \ldots, x_{n+1})
\end{split}
\end{equation}
where the hatted elements are omitted and by the symbol $\rho_x$ we denote the action of $x \in \mathfrak{g}$ on $M$. In case of $n=0$, we set ${\rm d_0}f(x_1) := \rho_{x_1} x_f$ for the constant map $f: g \in \mathfrak{g} \mapsto x_f \in M$. By a tedious computation, one shows that $d_{n+1} \circ d_n = 0$. This property implies that $B_n \subseteq Z_{n+1}$, where $Z_n := \ker({\rm d_{n}})$ is the set of {\it n-cocycles} and $B_n := {\rm Im }({\rm d_{n-1}})$ is the set of {\it n-coboundaries}. By default, we put $B_0 = \varnothing$. In consequence, we can define the cohomology space in the usual way, as the quotient space $H_n(\mathfrak{g}, M) := Z_{n+1} / B_n$.

It is remarkable that for a semi-simple complex Lie algebra $\mathfrak{g}$, its cohomology spaces $H_1(\mathfrak{g}, M)$ and $H_2(\mathfrak{g}, M)$ are trivial. These two facts are commonly known as, respectively, the first and the second {\it Whitehead lemma}.

Let us briefly focus on the first lemma. By definition, the space of 1-cocycles consists 
of linear functions $\Psi: \mathfrak{g} \to M$ such that ${\rm d}_1 \Psi (x,y) = 0$, or equivalently $\rho_x \cdot \Psi(y) - \rho_y \cdot \Psi(x) - \Psi([x,y]) = 0$. On the other hand, the space of 1-coboundaries is formed by the functions $f: \mathfrak{g} \ni v \mapsto \rho_v m \in M$ for a fixed $m \in M$.

Thus, vanishing of the cohomology space $H_1(\mathfrak{g}, M)$ is equivalent to the following statement: any linear function $\Psi: \mathfrak{g} \to M$ such that $\rho_x \cdot \Psi(y) - \rho_y \cdot \Psi(x) - \Psi([x,y]) = 0$ can be written as $\Psi(g) = \rho_g m$ for a certain $m \in M$. This result follows immediately from the next theorem.

\begin{theorem}\label{Th:whitehead} 
Let $\mathfrak{g}$ be a semi-simple Lie algebra over a field $k$ of characteristic zero and let $(V, \rho)$ be a $\mathfrak{g}$-module. If $\phi:\mathfrak{g}\rightarrow V$ is a linear map such that 
\begin{equation}\label{wh_cond}
\phi([v_1,v_2])=\rho_{v_1} \phi (v_2)- \rho_{v_2} \phi(v_1),
\end{equation}
then there exists $x\in V$ such that $\phi(v)= \rho_v x$ for every $v\in \mathfrak{g}$.
\end{theorem}
\begin{proof}
Before we move to the main part of the proof, let us make an important observation. For two ideals $\mathfrak{h}_1, \mathfrak{h}_2$ of $\mathfrak{g}$, assume there exists a nondegenerate bilinear form $(\cdot, \cdot): \mathfrak{h}_1 \times \mathfrak{h}_2 \to k$ such that
\begin{equation}\label{bil_form}
([h_1, g], h_2) + (h_1, [h_2, g]) = 0, \qquad \forall h_1 \in \mathfrak{h}_1, \quad \forall h_2 \in \mathfrak{h}_2, \quad \forall g \in \mathfrak{g}.
\end{equation}
Let $\{u_1, \ldots, u_m\}$ be a basis of $\mathfrak{h}_1$ and denote by $\{u^1, \ldots, u^m\}$ the dual basis for $\mathfrak{h}_1$. Moreover, introduce coefficients $\alpha_{ij}, \beta_{ij}$ such that
$$
[u_i, g] = \sum_{i,j} \alpha_{ij} u_j, \qquad [u^i, g] = \sum_{i,j} \beta_{ij} u^j
$$
Then, (\ref{bil_form}) implies that $\alpha_{ij} = -\beta_{ji}$.

Let $\rho$ and $\rho^{*}$ denote a representation of $\mathfrak{g}$ and $\mathfrak{g}^*$ on $V$. Define an element $\Gamma \in {\rm End}(M)$ as
\begin{equation}\label{casimir}
\Gamma := \sum_{k=1}^m \rho_{u_k} \rho_{u^k}
\end{equation}
where for convenience, we omit the composition symbol for elements of ${\rm End}(M)$. Notice that for any $g \in \mathfrak{g}$ we have
\begin{equation*}
\begin{split}
[\Gamma, \rho_{g}] &= \sum_{k=1}^m \rho_{u_k} \rho_{u^k} \rho_{g} - \sum_{k=1}^m \rho_{g} \rho_{u_k} \rho_{u^k} 
= \sum_{k=1}^m \rho_{u_k} [\rho_{u^k}, \rho_{g}] - \sum_{k=1}^m \rho_{u_k} \rho_{g} \rho_{u^k} + \sum_{k=1}^m [\rho_{u_k}, \rho_{g}] \rho_{u^k} - \sum_{k=1}^m \rho_{u_k} \rho_{g} \rho_{u^k} \\
&= \sum_{k=1}^m \rho_{u_k} [\rho_{u^k}, \rho_{g}] + \sum_{k=1}^m [\rho_{u_k}, \rho_{g}] \rho_{u^k} 
= \sum_{k=1}^m \sum_{l=1}^m \beta_{kl} \rho_{u_k} \rho_{u^l} + \sum_{k=1}^m \sum_{l=1}^m \alpha_{kl}\rho_{u_l} \rho_{u^k} \\
&= \sum_{k=1}^m \sum_{l=1}^m \beta_{kl} \rho_{u_k} \rho_{u^l} - \sum_{k=1}^m \sum_{l=1}^m \beta_{lk}\rho_{u_l} \rho_{u^k} 
= \sum_{k=1}^m \sum_{l=1}^m \beta_{kl} \rho_{u_k} \rho_{u^l} - \sum_{k=1}^m \sum_{l=1}^m \beta_{kl}\rho_{u_k} \rho_{u^l} 
= 0.
\end{split}
\end{equation*}
This shows that the element $\Gamma$, called the {\it Casimir element}, commutes with a representation of any element of $\mathfrak{g}$.

Now, let $\mathfrak{g}$ denote a Lie algebra satisfying the assumptions of this Theorem. Denote by $\mathfrak{R}$ the kernel of a representation $\rho$ of $\mathfrak{g}$. Since $\mathfrak{g}$ is semisimple, it admits a Killing form $\kappa$. Thus, $\mathfrak{g} \simeq \mathfrak{R} \oplus \mathfrak{R}^{\perp}$, where $\mathfrak{R}^{\perp} := \{v \in \mathfrak{g}: \kappa(v,w) = 0 \, \forall w \in \mathfrak{R}\}$. The restriction of $\rho$ to $\mathfrak{R}^{\perp}$ is injective and, in consequence, $\mathfrak{R}^{\perp}$ is semisimple. Therefore, the trace form on $\mathcal{L}$ is non-degenerate \cite{Bo05} and by Lemma \ref{Lem:trace_form}, it satisfies (\ref{bil_form}). Notice additionally that both $\mathfrak{R}$ and $\mathfrak{R}^{\perp}$ are ideals of $\mathfrak{g}$. Then by the remark at the beginning of the proof, we can construct the Casimir element $\Gamma$ for $\mathfrak{R}^{\perp}$ as given by the formula (\ref{casimir}). Moreover, ${\rm tr}(\Gamma) = \sum_{k=1}^m (u_k, u^k) = m = {\rm dim}(\mathcal{L})$.

Recall the Fitting's lemma \cite{Jacobson}: any linear endomorphism $A$ of a finite-dimensional vector space $V$ yields a decomposition $V \simeq V_0 \oplus V_1$ into $A$-invariant subspaces, such that $A$ restricted to $V_0$ is nilpotent, whereas its restriction to $V_1$ gives an isomorphism. Using that fact, we obtain the decomposition $V \simeq V_0 \oplus V_1$ relative to $\Gamma$. Since $\Gamma \rho_{g} = \rho_{g} \Gamma$ for any $g \in \mathfrak{g}$, the subspaces $V_0, V_1$ are $\mathfrak{g}$-submodules. Indeed, if $v \in V_0$, then $\Gamma \rho_{g} v = \rho_{g} \Gamma v = 0$, which means that $\rho_{g} v \in V_0$. If $v \in V_1$, then $\Gamma \rho_{g} v = \rho_{g} \Gamma v = \rho_{g} v$, whence $\rho_{g} v \in V_1$.

Similarly, the map $\phi$ admits a decomposition $\phi(v) = \phi_0(v) + \phi_1(v)$ for any $v \in \mathfrak{g}$, where $\phi_0(v) \in V_0$ and $\phi_1(v) \in V_1$. Since $\phi$ satisfies (\ref{wh_cond}), we get
$$
\phi_0([v,w]) + \phi_1([v,w]) = (v \phi_0(w) - w \phi_0(v)) + (v \phi_1(w) - w \phi_1(v)) \in V_0 \oplus V_1
$$
In consequence, both maps $\phi_0, \phi_1$ satisfy the condition (\ref{wh_cond}) as well. 

If $V_0, V_1$ are both nonzero, then ${\rm dim}(V_i) < {\rm dim}(V)$, $i = 0,1$. Since for both $V_0, V_1$ there exist maps $\phi_0, \phi_1$ satisfying the assumptions of this Theorem, we can repeat the whole argument for these smaller submodules. As $V$ is finite-dimensional, this procedure stops finally, yielding a decomposition $V \simeq \bigoplus_k Z^{(k)}$ into one-dimensional submodules $Z^{(k)}$ and for every such $Z^{(k)}$, we get $\phi_k(v) = \rho_{v} d_k$ for a certain element $d_k \in Z^{(k)}$. Therefore, $\phi(v) = \sum_k \phi_k(v) = \rho_v \left(\sum_k d_k\right)$ and $d:= \sum_k d_k \in V$ such that $\phi(v) = \rho_{v} d$.

If $V_1 = \{0\}$, then $\Gamma$ is a nilpotent linear transformation on $V$ by the Fitting's lemma. As it is nilpotent, its trace must vanish and thus, ${\rm dim}(\mathfrak{R}^{\perp}) = 0$. That means the whole $\mathfrak{g}$ is the kernel of $\rho$ and the condition (\ref{wh_cond}) reads $\phi([v_1, v_2]) = 0$ for any $v_1, v_2, \in \mathfrak{g}$. Since $\mathfrak{g}$ is semi-simple, $[\mathfrak{g}, \mathfrak{g}] = \mathfrak{g}$ and we conclude that $\phi(v) = 0$ for any $v \in \mathfrak{g}$. Thus, $x = 0$ satisfies that $\phi(v) = \rho_v x$. 

Finally, let $V_0 = \{0\}$. As before, denote a basis of $\mathcal{L}$ and its dual by $\{u_1, \ldots, u_m\}$ and $\{u^1, \ldots, u^m\}$, respectively. Define $V \ni y := \sum_{k=1}^m \rho^{*}_{u^k} \phi(u_k)$. Then for $a \in \mathfrak{g}$,
\begin{equation*}
\begin{split}
\rho_a y &= \sum_{k=1}^m \rho_a\rho^{*}_{u^k} \phi(u_k) 
= \sum_{k=1}^m \rho^{*}_{u^k}\rho_a \phi(u_k) + \sum_{k=1}^m [\rho_a, \rho^{*}_{u^k}] \phi(u_k) 
= \sum_{k=1}^m \rho^{*}_{u^k}\rho_a \phi(u_k) - \sum_{k=1}^m [\rho^{*}_{u^k}, \rho_a] \phi(u_k) \\
&= \sum_{k=1}^m \rho^{*}_{u^k}\rho_a \phi(u_k) - \sum_{k=1}^m \beta_{kl} \rho_{u^l} \phi(u_k) 
= \sum_{k=1}^m \rho^{*}_{u^k}\rho_a \phi(u_k) + \sum_{k=1}^m \sum_{l=1}^m \alpha_{lk} \rho_{u^l} \phi(u_k) \\
&= \sum_{k=1}^m \rho^{*}_{u^k}\rho_a \phi(u_k) + \sum_{k=1}^m \sum_{l=1}^m \rho_{u^l} \phi(\alpha_{lk} u_k) 
= \sum_{k=1}^m \rho^{*}_{u^k}\rho_a \phi(u_k) + \sum_{l=1}^m  \rho_{u^l} \phi(\sum_{k=1}^m \alpha_{lk} u_k) \\
&= \sum_{k=1}^m \rho^{*}_{u^k}\rho_a \phi(u_k) + \sum_{l=1}^m  \rho_{u^l} \phi([u_l, a]) 
= \sum_{k=1}^m \rho^{*}_{u^k}\rho_a \phi(u_k) + \sum_{l=1}^m  \rho_{u^l} \rho_{u_l} \phi(a) - \sum_{l=1}^m  \rho_{u^l} \rho_{a} \phi(u_l) \\
&= \sum_{l=1}^m  \rho_{u^l} \rho_{u_l} \phi(a) 
= \Gamma \phi(a)
\end{split}    
\end{equation*}
By the Fitting's lemma, $\Gamma$ is an isomorphism on $V_1$ and in this case, on the whole $V$. Thus, an element $x:= \Gamma^{-1} y$ satisfies the condition $\phi(v) = \rho_v x$.
\end{proof}

\section{Lie bialgebras}

This section describes the most important facts on Lie bialgebras (see \cite{CP95, Dr83, Dr87,KS04,Le90} for details). 

\begin{definition}\label{Def:bialg}
A \textit{Lie bialgebra} is a pair $(\mathfrak{g}, \delta)$, where $\mathfrak{g}$ is a Lie algebra with a Lie bracket $[\cdot, \cdot]$ and  $\delta: \mathfrak{g} \to T^2 \mathfrak{g}$ is a linear map satisfying that:
\begin{enumerate}
\item the (restricted) dual map $\delta^*: T^2 \mathfrak{g}^* \to \mathfrak{g}^*$ is a Lie bracket on $\mathfrak{g}^*$,
\item the map $\delta$ satisfies the condition
\begin{equation}\label{cocycle_cond}
\delta([x,y]) = \otimes^2 {\rm ad}_x \delta(y) - \otimes^2 {\rm ad}_y \delta(x)
\qquad \forall x,y \in \mathfrak{g}.
\end{equation}
\end{enumerate}
\end{definition}

 \begin{definition}\label{bialg_hom}
\textit{A homomorphism of Lie bialgebras} $(\mathfrak{g}, \delta_{\mathfrak{g}})$ and $(\mathfrak{h}, \delta_{\mathfrak{h}})$ is a homomorphism $\phi: \mathfrak{g} \to \mathfrak{h}$ of Lie algebras such that
\begin{equation*}
(\phi \otimes \phi) \circ \delta_{\mathfrak{g}} = \delta_{\mathfrak{h}} \circ \phi.
\end{equation*}
\end{definition}

The map $\delta$ is usually called a {\it cocommutator} to emphasize the fact that its transpose is indeed a Lie bracket (also called commutator at times) on $\mathfrak{g}^*$. 

The first condition in Definition \ref{Def:bialg} implies that ${\rm Im} \delta \subseteq \Lambda^2 \mathfrak{g} \subset T^2 \mathfrak{g}$ as $\delta^*$ is required to be a Lie bracket on $\mathfrak{g}^*$. In consequence, (\ref{cocycle_cond}) can be written using Schouten bracket as
\begin{equation}\label{cocycle_cond_anti}
\delta([x,y]) = [x, \delta(y)]_{S} - [y, \delta(x)]_{S}, \qquad \forall x,y \in \mathfrak{g}.
\end{equation}

Note that $\delta^*$ is well-defined regardless of the Lie algebra dimension. Indeed, for every finite-dimensional Lie algebra $\mathfrak{g}$, one has that $(T^2\mathfrak{g})^*\simeq T^2\mathfrak{g}^*$, whereas in the infinite-dimensional case, we have $\mathfrak{g}^* \otimes \mathfrak{g}^* \subset (\mathfrak{g} \otimes \mathfrak{g})^*$. The latter situation requires only the restriction of the dual map. 

The second condition in Definition \ref{Def:bialg} has a concise cohomological interpretation. For $\delta: \mathfrak{g} \to T^2 \mathfrak{g}$, one can understand (\ref{cocycle_cond}) as ${\rm d}_1 \delta = 0$, where ${\rm d}_1$ is the coboundary map given by (\ref{cob_map}) for $M = T^2 \mathfrak{g}$. In other words, a cocommutator $\delta$ is a 1-cocycle in the Chevalley-Eilenberg cohomology of $\mathfrak{g}$ with values in $T^2 \mathfrak{g}$.

Naturally, 1-coboundaries represent a particular example of such 1-cocycles. Thus, 1-coboundaries give rise to the following special case of the Lie bialgebra structure.

\begin{definition}\label{coboundary}
A \textit{coboundary Lie bialgebra} is a Lie bialgebra $(\mathfrak{g}, \delta_r)$ whose cocommutator $\delta_r$ takes the form $\delta_r(x) = T^2 {\rm ad}_x r$ for certain $r \in T^2 \mathfrak{g}$ and all $x \in \mathfrak{g}$. The element $r$ is called a \textit{quasitriangular $r$-matrix} of $\delta_r$. If $r \in \Lambda^2 \mathfrak{g}$, then it is called a {\it triangular $r$-matrix}.
\end{definition}

Once again, the condition $\delta_r(x) = T^2 {\rm ad}_x r$ imposed on the cocommutator $\delta_r$ follows immediately from the definition of the coboundary map (\ref{cob_map}). In the particular case of $r \in \Lambda^2 \mathfrak{g}$, it can be written concisely as $\delta_r(x) = [x, r]_S$.

It is obvious that for any Lie algebra $\mathfrak{g}$, there is a trivial Lie bialgebra structure $\delta(v) = 0 \in \Lambda^2 \mathfrak{g}$ for any $v \in \mathfrak{g}$. The natural question is whether there are any nontrivial bialgebra structures as well. For a large class of finite-dimensional Lie algebras, the positive answer was given. For abelian Lie algebras, any Lie bracket on $\mathfrak{g}^*$ induces a Lie bialgebra structure on $\mathfrak{g}$. In case of nonabelian Lie algebras, de Smedt \cite{Sm94} asserts the existence of a triangular structure over fields of characteristic zero. Lie algebras over fields of arbitrary characteristic were considered in \cite{Fel99, Mich94}. However, very little is known for general infinite-dimensional Lie algebras - only specific examples were studied in detail (e.g. \cite{HLS09, Lei94, Mich94, WWY15}).

\begin{example} 
Consider the Lie algebra $\mathfrak{sl}_2$ with the basis and commutation relations as in Example \ref{Ex:sl2}. Let us define a mapping $\delta:v\in \mathfrak{sl}_2\mapsto [v,e_2\wedge e_3]_S\in \Lambda^2\mathfrak{sl}_2$. Then by the third property in Proposition \ref{Pr:PropSchou}, one has that
$$
\delta([v_1,v_2])=[[v_1,v_2],e_2\wedge e_3]_S=-[e_2\wedge e_3,[v_1,v_2]_S]_S=-[[e_2\wedge e_3 ,v_1]_S,v_2]_S+[[e_2\wedge e_3,v_2]_S,v_1]_S.
$$
Thus,
$$
\delta([v_1,v_2])=-[v_2,\delta(v_1)]_S+[v_1,\delta(v_2)]_S.
$$
In particular, one has
$$
\delta(e_1)=[e_1,e_2\wedge e_3]=[e_1,e_2]\wedge e_3 +e_2\wedge [e_1, e_3]=e_1\wedge e_3,\quad \delta(e_2)=e_2\wedge e_3,\quad \delta(e_3)=0.
$$
Let us work out $\delta^*:\Lambda^2 \mathfrak{sl}_2^*\rightarrow\mathfrak{sl}_2^*$. Denote by $\{e^1,e^2,e^3\}$ the basis dual to $\{e_1,e_2,e_3\}$. Then, $\{e^1\wedge e^2,e^2\wedge e^3,e^1\wedge e^3\}$ is the dual basis to $\{e_1\wedge e_2,e_2\wedge e_3,e_1\wedge e_3\}$. For each element of the dual basis, let us compute its image by $\delta^*$. For example,
\begin{equation*}
\begin{split}
&[\delta^*(e^1\wedge e^2)](e_1) = \langle e^1\wedge e^2, \delta(e_1) \rangle = \langle e^1\wedge e^2, e_1 \wedge e_3 \rangle = 0 \\
&[\delta^*(e^1\wedge e^2)](e_2) = \langle e^1\wedge e^2, \delta(e_2) \rangle = \langle e^1\wedge e^2, e_2 \wedge e_3 \rangle = 0 \\
&[\delta^*(e^1\wedge e^2)](e_3) = \langle e^1\wedge e^2, \delta(e_3) \rangle = \langle e^1\wedge e^2, 0 \rangle = 0 
\end{split}
\end{equation*}
Thus, $\delta^*(e^1\wedge e^2) = 0$. Similarly, we get
$$
\delta^*(e^1\wedge e^3)=e^1,\quad \delta^*(e^2\wedge e^3)==e^2.
$$
For convenience, let us write $[\theta^1,\theta^2]_{\mathfrak{sl}_2^*}:=\delta^*(\theta^1\wedge \theta^2)$. Then,
$$
[e^1,e^2]_{\mathfrak{sl}^*}=0,\qquad [e^1,e^3]_{\mathfrak{sl}^*}=e^1,\qquad [e^2,e^3]_{\mathfrak{sl}^*}=e^2.
$$
In order to prove the Jacobi identity, it is enough to check that
$$
[e^1,[e^2,e^3]_{\mathfrak{sl}^*}]_{\mathfrak{sl}^*}+[e^2,[e^3,e^1]_{\mathfrak{sl}^*}]_{\mathfrak{sl}^*}+[e^3,[e^1,e^2]_{\mathfrak{sl}^*}]_{\mathfrak{sl}^*}=0,
$$
We obtain
$$
[e^1,[e^2,e^3]_{\mathfrak{sl}^*}]_{\mathfrak{sl}^*}=[e^1,e^2]_{\mathfrak{sl}^*}=0,
$$
and
$$
[e^2,[e^3,e^1]_{\mathfrak{sl}^*}]_{\mathfrak{sl}^*}=[e^2,-e^1]_{\mathfrak{sl}^*}=0,\qquad
[e^3,[e^1,e^2]_{\mathfrak{sl}^*}]_{\mathfrak{sl}^*}=[e^3,0]_{\mathfrak{sl}^*}=0.
$$
Therefore, $\delta^*$ defines a Lie bracket on $\mathfrak{sl}_2^*$ and thus, $\delta$ gives a coboundary cocommutator relative to the triangular r-matrix $e_2 \wedge e_3$.
\end{example}

Interestingly, semi-simple Lie algebras admit only coboundary Lie bialgebra structures. Let $\mathfrak{g}$ be a semi-simple Lie algebra and let $M=T^2\mathfrak{g}$. Then, Theorem \ref{Th:whitehead} implies that for any cocommutator $\delta: \mathfrak{g} \to M$, there exists $r\in M $ such that $\delta(v)=\otimes^2 {\rm ad}_v r$ for any $v\in \mathfrak{g}$.

Let us now address an important question what conditions an element $r \in T^2 \mathfrak{g}$ should satisfy so that the map $\delta_r: v \in \mathfrak{g} \mapsto \otimes^2 {\rm ad}_v r$ is a cocommutator. First, write $r = s + a$, where $s \in S^2 \mathfrak{g}$ and $a \in \Lambda^2 \mathfrak{g}$. Since the dual of $\delta_r = {\rm d}_0 r$ must yield a Lie bracket, which is antisymmetric by definition, it follows that ${\rm d}_0 s = 0$ or equivalently, $\otimes^2 {\rm ad}_v s = 0$ for any $v \in \mathfrak{g}$.

\begin{definition}
An element of $w\in T^k\mathfrak{g}$ is called $\mathfrak{g}$-invariant if $\otimes^2 {\rm ad}_v w = 0$ for every $v\in \mathfrak{g}$. We write $(T^k\mathfrak{g})^\mathfrak{g}$ for all $\mathfrak{g}$-invariant elements in $T^k\mathfrak{g}$. Similarly, an element of $w\in\Lambda^k\mathfrak{g}$ is called $\mathfrak{g}$-invariant if $[v,w]_S=0$ for every $v\in \mathfrak{g}$. We write $(\Lambda^k\mathfrak{g})^\mathfrak{g}$ for all $\mathfrak{g}$-invariant elements in $\Lambda^k\mathfrak{g}$.
\end{definition}

In other words, for the element $r \in T^2 \mathfrak{g}$ to be a quasitriangular r-matrix, its symmetric part has to be $\mathfrak{g}$-invariant. 

Additionally, in order for $({\rm d}_0 r)^*$ to define a Lie bracket, the dual $({\rm d}_0 a)^*$ must satisfy the Jacobi identity. In turn, this puts an additional condition on the element $r \in T^2 \mathfrak{g}$, which we discuss next.

First, let us introduce a comfortable notation. For $r = \sum_i a_i \otimes b_i \in T^2 \mathfrak{g}$, we define
$$
[[r, r]] := [r_{12}, r_{13}] + [r_{12}, r_{23}] + [r_{13}, r_{23}]
$$
where 
\begin{equation*}
[r_{12}, r_{13}] := [a_i, a_j] \otimes b_i \otimes b_j, \quad
[r_{12}, r_{23}] := a_i \otimes [b_i, a_j] \otimes b_j, \quad
[r_{13}, r_{23}] := a_i \otimes a_j \otimes [b_i, b_j].
\end{equation*}
Obviously, $[[r,r]]$ gives an element of $T^3 \mathfrak{g}$. Let us note that for $r \in \Lambda^2 \mathfrak{g}$, the element $[[r, r]]$ reduces to the Schouten bracket $[r, r]_S$ of $r$ with itself.

\begin{definition}\label{def:alt}
The map $\textrm{Alt}: \mathfrak{g}^{\otimes (n+1)} \to \mathfrak{g}^{\otimes (n+1)}$ is defined by
\begin{equation*}
\textrm{Alt}(a_1 \otimes \ldots \otimes a_n) := \sum_{k=0}^{n} a_{\sigma_k(1)} \otimes \ldots a_{\sigma_k(n)}, \quad \forall a_1, \ldots, a_n \in \mathfrak{g},
\end{equation*}
where $\sigma_k$ is a cyclic permutation of $\{1, \ldots, n\}$ such that $\sigma_k(i) := (i+k)\, (\textrm{mod}\, n)$ for $i = 1, 2, \ldots, n$.
\end{definition}

\begin{lemma}\label{lemma}
Let the map $\delta_r: \mathfrak{g} \to \Lambda^2 \mathfrak{g}$ be defined by $\delta_r(x) := [x, r]_{S}$ for $r \in T^2 \mathfrak{g}$. Then
\begin{equation*}
J_{\delta_r} (x) + \otimes^3 {\rm ad}_x [[r, r]] = 0,
\end{equation*}
where $J_{\delta_r}(x) := {\rm Alt}[((\delta_r \otimes {\rm id}) \circ \delta_r)(x)]$.
\end{lemma}

The proof is a direct computation (see \cite{CP95} for details).

\begin{theorem}\label{Thm:YBE}
The map $\delta_r: \mathfrak{g} \ni v \mapsto [v, r]_S \in \Lambda^2\mathfrak{g}$ defined by $r \in \Lambda^2 \mathfrak{g}$ is a cocommutator if and only if 
\begin{equation}\label{rl3inv}
[[r, r]] \in (T^3 \mathfrak{g})^{\mathfrak{g}}.
\end{equation}
\end{theorem}
\begin{proof}

Let us take any linear map $\delta: \mathfrak{g} \to \Lambda^2 \mathfrak{g}$ and define a bilinear map on $\mathfrak{g}^*$ by $[\xi, \eta]_{\mathfrak{g}^*} := \delta^*(\xi \wedge \eta)$. Since $[[\xi, \eta]_{\mathfrak{g}^*}, \zeta]_{\mathfrak{g}^*} = \delta^*(\delta^*(\xi\otimes \eta)\otimes \zeta) = \delta^*(\delta^* \otimes \textrm{id}) (\xi \otimes \eta \otimes \zeta)$, we see that $[\cdot, \cdot]_{\mathfrak{g}^*}$ satisfies Jacobi identity if and only if $J_{\delta}$ vanishes. Thus, by virtue of Lemma \ref{lemma}, an element $[[r, r]]$ has to be $\mathfrak{g}$-invariant.
\end{proof}

The condition (\ref{rl3inv}) is known as the \textit{modified classical Yang-Baxter equation} (mCYBE). Its special case $[[r, r]] = 0$ is called the \textit{classical Yang-Baxter equation} (CYBE). Following the previously introduced nomenclature, the solutions to the mCBYE are called {\it quasitriangular $r$-matrices} (or simply, {\it r-matrices}) while solutions to the CYBE are called {\it triangular} $r$-matrices. 

This thesis is solely devoted to study the triangular r-matrices. We will hereafter refer to such elements shortly as r-matrices, since it will not lead to any confusion later.

\begin{example}
Let us illustrate previous notions by analysing a two-dimensional non-abelian real Lie algebra $\mathfrak{g} = \langle x, y\rangle$ with $[x, y] = x$. Since $\Lambda^2\mathfrak{g}=\langle x\wedge y\rangle$, every cocommutator is of the form $\delta(x) = a x \wedge y$ and $\delta(y) = b x \wedge y$ for certain $a,b \in \mathbb{R}$. The cocycle condition (\ref{cocycle_cond})
is satisfied since
\begin{equation*}
\delta([x,y]) = [x, \delta(y)]_{S} - [y, \delta(x)]_{S}=[x,bx\wedge y]-[y,a x\wedge y]]=ax\wedge y=\delta(x).
\end{equation*}
Since $\mathfrak{g}^*$ is a two-dimensional vector space, $\delta^T$ is an antisymmetric map that becomes automatically a Lie bracket on $\mathfrak{g}^*$. Indeed, by multilinearity of the Jacobi identity, its value on any elements $\theta_1^*,\theta_2^*,\theta_3^*$ reduces to obtaining its value on elements of a dual basis $\{x^*,y^*\}$ to $\{x,y\}$. Among four cases, we get two expressions that require further calculations, namely
$$
[x^*,[x^*,y^*]]+[x^*,[y^*,x^*]]+
[y^*,[x^*,x^*]]=[x^*,[x^*,y^*]]+[x^*,[y^*,x^*]]=[x^*,[x^*,y^*]]-[x^*,[x^*,y^*]]=0
$$
and
$$
[y^*,[y^*,x^*]]+[y^*,[x^*,y^*]]+
[x^*,[y^*,y^*]]=0
$$
Thus, the Jacobi identity holds and it follows that $\delta$ gives a cocommutator with no restrictions on the coefficients $a, b$.

Let us now assume that $(\mathfrak{g}, \delta)$ is a coboundary Lie bialgebra. Since $\dim \mathfrak{g}=2$, any $r$-matrix $r \in \Lambda^2 \mathfrak{g}$ satisfies (\ref{rl3inv}) and thus, defines a cocommutator. Denote $r = c x \wedge y$ for any $c \in \mathbb{R}$. From Definition \ref{coboundary}, it follows that 
\begin{equation*}
\delta(x) = [x, r]_{S}=0=ax\wedge y \quad \iff \quad a = 0
\end{equation*}
and
\begin{equation*}
\delta(y) = [y, r]_{S}=-cx\wedge y=bc\wedge y \quad \iff \quad b = -c.
\end{equation*}
\end{example}
It is important to establish when two $r$-matrices induce the same cocommutator. This is answered by the following proposition.

\begin{proposition}\label{prop:requiv} 
Two $r$-matrices $r_1,r_2\in \Lambda^2\mathfrak{g}$ satisfy that $\delta_{r_1}=\delta_{r_2}$ if and only if  $r_1-r_2\in(\Lambda^2\mathfrak{g})^{\mathfrak{g}}$.
\end{proposition}
\begin{proof} 
Let $\delta_{r_1}$ and $\delta_{r_2}$ be the cocommutators respectively induced by elements $r_1,r_2 \in \lambda^2 \mathfrak{g}$. Since $\delta_{r_1}=\delta_{r_2}$, one obtains that $[r_1,v]_S=[r_2,v]_S$ for every $v\in \mathfrak{g}$. Hence, $[r_1 - r_2, v]_S = 0$ for any $v\in \mathfrak{g}$, which implies $r_1-r_2\in (\Lambda^2\mathfrak{g})^{\mathfrak{g}}$. The converse is immediate.
	\end{proof}

In other words, two $r$-matrices give rise to the same cocommutator if and only if their difference is a $\mathfrak{g}$-invariant bivector. Thus in order to determine coboundary Lie bialgebras, one should focus not on $r$-matrices themselves, but on their equivalence classes in $\Lambda^2\mathfrak{g}/(\Lambda^2\mathfrak{g})^\mathfrak{g}$.

It follows from Definition \ref{Def:bialg} that given a Lie bialgebra $(\mathfrak{g}, \delta_{\mathfrak{g}})$, one can construct a Lie bialgebra structure on $\mathfrak{g}^*$ and vice versa. This observation motivates another characterisation of the Lie bialgebra notion, which makes such a symmetry more apparent.

\begin{definition}
A Manin triple consists of a triple of Lie algebras $(\mathfrak{p}, \mathfrak{p}_{1}, \mathfrak{p}_{2})$ and a nondegenerate symmetric bilinear form $(\cdot, \cdot)$ on $\mathfrak{p}$ such that
\begin{itemize}
\item $\mathfrak{p}_1$ and $\mathfrak{p}_2$ are Lie subalgebras of $\mathfrak{p}$, 
\item there is a decomposition $\mathfrak{p} = \mathfrak{p}_1 \oplus \mathfrak{p}_2$ as vector spaces,
\item $\mathfrak{p}_1$ and $\mathfrak{p}_2$ are isotropic relative to the bilinear form $(\cdot, \cdot)$, that is $(v, w) = 0$ for all $w \in \mathfrak{p}$ and any $v \in \mathfrak{p}_1$ (and analogously, for $v \in \mathfrak{p}_2$).
\end{itemize}
\end{definition}

For finite-dimensional Lie algebras, the notions of Manin triples and Lie bialgebras are equivalent, as proved by the next proposition.

\begin{proposition}
Let $\mathfrak{g}$ be a finite-dimensional Lie algebra. There is a one-to-one correspondence between Lie bialgebra structures on $\mathfrak{g}$ and Manin triples $(\mathfrak{p}, \mathfrak{p}_{1}, \mathfrak{p}_{2})$ with $\mathfrak{p}_1 = \mathfrak{g}$.
\end{proposition}

It is known that Lie algebras are infinitesimal counterparts of Lie groups \cite{Ha15}. Similarly for Lie bialgebras, a cocommutator induces an additional structure at the Lie group level.

\begin{definition}
A Poisson--Lie group is a Lie group $G$ equipped with a Poisson bracket $\{\cdot, \cdot\}$ such that the multiplication in $G$ gives a Poisson map, that is
$$
\{f, g\}(p \cdot q) = \{f \circ L_p, g \circ L_p\}(q) + \{f \circ R_q, g \circ R_q\}(p)
$$
for any $f,g \in C^{\infty}(G)$ and $p, q \in G$. The maps $L_g, R_g$ correspond to the left and right translation maps on $G$, respectively.
\end{definition}
In terms of the Poisson bivector $w \in \Lambda^2 TG$, the above condition reads
\begin{equation}\label{eq:poi_bv}
w_{p \cdot q} = (T_q L_p \otimes T_q L_p)(w_q) + (T_p R_q \otimes T_p R_q)(w_p), \qquad \forall p, q \in G.
\end{equation}
Let us briefly investigate the correspondence between Lie bialgebras and Poisson--Lie groups. Introduce the map $w^R: G \to \mathfrak{g} \otimes \mathfrak{g}$ given by $w^R(g) := (T_g R_{g^{-1}} \otimes T_g R_{g^{-1}})(w_g)$, the right translation of the Poisson bivector $w$ to the identity element $e$ of $G$. Its derivative at $e$ yields a map $\delta: \mathfrak{g} \to \mathfrak{g} \otimes \mathfrak{g}$. Using $w^R$, one can rewrite condition (\ref{eq:poi_bv}) as
$$
w^R(p \cdot q) = ({\rm Ad}_p \otimes {\rm Ad}_p)(w^R(q)) + w^R(p)
$$
Differentiating the above equality at $e$ gives
$$
\delta([v, w]) = \otimes^2 {\rm ad}_v \delta(w) - \otimes^2 {\rm ad}_w \delta(v)
$$
This shows that the Lie bialgebra structure on a Lie algebra $\mathfrak{g}$ can be derived from the Poisson--Lie groups structure on the corresponding Lie group $G$. Conversely, any Lie bialgebra  on $\mathfrak{g}$ gives rise to the Poisson--Lie group structure on the connected and simply-connected Lie group $G$ associated with $\mathfrak{g}$. We shall not discuss the converse result (see \cite{CP95} and references therein for details).

\section{Lie systems}\label{Sec:LS_intro}

Among first-order ordinary differential equations, Lie systems constitute a particular class, as their general solutions can be expressed in a special way. In this section, we summarise necessary results on such systems. If not otherwise stated, every system of differential equations considered in this Section is non-autonomous and moreover, $N$ is an $n$-dimensional manifold.
	
First, let us recall the correspondence between first-order ODEs and $t$-dependent vector fields. Define $\pi_2:(t,x)\in \mathbb{R}\times N\mapsto x\in N$ and let  $\tau_N:TN\rightarrow N$ be the tangent bundle projection.
	A $t$-dependent vector field on $N$ is a mapping $X:(t,x)\in \mathbb{R}\times N\mapsto X(t,x)\in TN$ such that $\tau_N\circ X=\pi_2$.  An {\it integral curve} of $X$ is a mapping $\gamma:\mathbb{R}\rightarrow N$ such that
	\begin{equation}\label{Asso}
		\frac{d\gamma}{dt}(t)=X(t,\gamma(t)),\qquad \forall t\in \mathbb{R}.
	\end{equation}
	Given a $t$-dependent vector field $X$, one can construct \cite{AM87,Dissertationes} a vector field $\tilde{X}$ on $\mathbb{R}\times N$, called an {\it autonomisation of $X$}, of the form $\widetilde{X}=\partial_t+X$ with an integral curve $\widetilde{\gamma}:t\in \mathbb{R}\mapsto (t,\gamma(t))\in \mathbb{R}\times N$. Conversely, if $\widetilde{\gamma}:\mathbb{R}\rightarrow \mathbb{R}\times N$ is an integral curve of $\widetilde{X}$ and a section of the bundle $\pi_1:(t,x)\in \mathbb{R}\times N\mapsto t\in \mathbb{R}$, then $\pi_2\circ \widetilde{\gamma}$ is a solution to (\ref{Asso}). In consequence, we can identify a system  (\ref{Asso}) and its associated $t$-dependent vector field $X$. 
		
	A {\it superposition rule} \cite{CGM07,Dissertationes,PW} for a system $X$ on a manifold $N$ is a map $\Psi:N^m\times N\rightarrow N$ satisfying that the general solution, $x(t)$, to $X$ can be written as
	$$
	x(t)=\Psi(x_{(1)}(t),\ldots ,x_{(m)}(t),k),
	$$
	for a generic family of particular solutions $x_{(1)}(t),\ldots,x_{(m)}(t)$  of $X$ and a parameter $k\in N$ to be related to the initial condition of $x(t)$. We call {\it Lie system} a system of first-order ODEs admitting a superposition rule.

Given a system of first-order ODEs, one may ask whether it forms a Lie system. The following theorem provides a suitable geometric condition.
 
	\begin{theorem}{\bf (The Lie--Scheffers theorem \cite{CGM00,CGM07,LS,PW})} A system $X$ on $N$ admits a
		superposition rule if and only if 
		$
		X={\displaystyle \sum_{\alpha=1}^r}b_\alpha(t)X_\alpha
		$
		for a certain family $X_1,\ldots,X_r$ of vector fields on $N$ spanning an $r$-dimensional Lie algebra of vector fields, a so-called {\it Vessiot-Guldberg Lie algebra} of $X$, and
		a family $b_1(t),\ldots,b_r(t)$  of $t$-dependent functions.
	\end{theorem}
	
	As one of the simplest non-trivial examples of Lie systems, let us discuss the Riccati equation \cite{LS}, that is 
 $$
 \frac{{\rm d}x}{{\rm d}t} = a_1(t) + a_2(t) x + a_3(t) x^2
 $$ 
 for certain $t$-dependent functions $a_1(t),a_2(t),$ and $a_3(t)$. The associated $t$-dependent
	vector field on $\mathbb{R}$ is of the form 
	$$
	X^{Ric}=(a_1(t)+a_2(t)x+a_3(t)x^2)\frac{\partial}{\partial x}.
	$$
	 Then, $
	X^{Ric}=\displaystyle{\sum_{\alpha=1}^3} a_\alpha (t)X_\alpha,
	$
	where $X_1, X_2, X_3$ are vector fields on $\mathbb{R}$ satisfying the commutation relations
	$$
	[X_1, X_2] = X_1,\qquad [X_1,X_3]=2X_2,\qquad [X_2,X_3]=X_3,
	$$
	and thus, spanning a Lie algebra of vector fields isomorphic to $\mathfrak{sl}_2$ \cite{LS,PW}. According  to the Lie--Scheffers theorem, Riccati equations must admit a superposition rule. Indeed, it is known \cite{CGM00,In44,LS} that the general solution to a Riccati equation, $x(t)$, can be written in terms of a function $\Psi:\mathbb{R}^3\times \mathbb{R}\rightarrow \mathbb{R}$ in the form
	$$
	x(t)=\Psi(x_{(1)}(t),x_{(2)}(t),x_{(3)}(t),k),\qquad k\in \mathbb{R},
	$$
	where $x_{(1)}(t),x_{(2)}(t),x_{(3)}(t)$ are three different particular solutions to $X$ and 
	$$
	\Psi(u_1,u_2,u_3;k)=\frac{u_1(u_3-u_2)-ku_2(u_3-u_1)}{(u_3-u_2)-k(u_3-u_1)} ,
	$$
	where the limit $k\to \infty$ should be admitted to find the solution $x_{(2)}$. 
	
	Another relevant example of Lie system (see \cite{CGM00,Dissertationes,LS20}) is given by the system on an $r$-dimensional Lie group $G$ of the form
	\begin{equation}\label{Aut}
		X^G(t,g)=\sum_{\alpha=1}^rb_\alpha(t)X^R_\alpha(g),\qquad g\in G,
	\end{equation}
	where $X^R_1,\ldots,X_r^R$ stand for a basis of right-invariant vector fields
	on $G$ and $b_1(t),\ldots,b_r(t)$ are arbitrary $t$-dependent functions. Indeed,  if $R_g$ is the right-translation map $R_g:h\in G\mapsto hg\in G$ and $\{e_\alpha\mid \alpha=1,\ldots,r\}$ 
	is a basis of $T_eG$, then the right-invariant vector fields $X^R_1,\ldots,X^R_r$, defined by  $X_\alpha ^R(g)=R_{g*e}e_\alpha $ span an $r$-dimensional Lie algebra of vector fields on $G$. Consequently,  the $t$-dependent vector field (\ref{Aut}) defines  a Lie system. The  Lie--Scheffers theorem states that  the differential equation on $G$ determining its integral curves,  
	\begin{equation}\label{Aut2}
 \frac{dg }{dt}=X^G(t,g),
 \end{equation}  
 admits a superposition rule.
	A simple application of the right-translation $R_{g(t)^{-1}*g(t)}$ to both sides leads to an  equivalent equation for $g(t)$, 
	\begin{equation}\label{Aut3}
 R_{g^{-1}*g}\frac{dg}{dt} =  \sum_{\alpha=1}^rb_\alpha(t)e_\alpha\in T_eG.
	\end{equation}
	The right-invariance of the $t$-dependent vector field (\ref{Aut}) relative to the action of $G$ on  itself from the right shows the right-invariance of equation (\ref{Aut2}), or its equivalent (\ref{Aut3}), i.e.     that for every $h\in G$ a particular solution $g_0(t)$ to (\ref{Aut2}) gives rise to new particular solutions $R_h g_0(t)$ of (\ref{Aut2}).  As initial conditions at $t=0$ determine univocally the particular solutions, $g(t)$, of  (\ref{Aut2}), the general solution to  (\ref{Aut2}) can be
	brought into the form
	$$
	g(t)=R_hg_0(t),
	$$
	where $g_0(t)$ is a particular solution to (\ref{Aut2}) and $h\in G$. 
	Then, $X^{Ric}$ admits a superposition rule  involving one particular solution given by  $\Psi:(g,h)\in G\times G\mapsto R_hg\in G$.
	
	Lie systems of the form (\ref{Aut}) are called {\it automorphic Lie
		systems} \cite{Dissertationes}. Their special role on the theory of Lie systems is explained
	by the following theorem, which states that the general solution to every Lie system
	can be obtained from the knowledge of a particular solution of a related automorphic Lie system \cite{CGM00,LS20,Ve93,Ve99}.
	
	\begin{theorem}
		Let $X$ be a Lie system on  $N$ of the form $X={\displaystyle\sum_{\alpha=1}^r}b_\alpha(t)X_\alpha$ for certain $t$-dependent functions $b_1(t),\ldots,b_r(t)$ and an $r$-dimensional Vessiot--Guldberg Lie algebra $V=\langle X_1,\ldots,X_r\rangle$. Let
		$G$ be the unique connected and simply connected Lie group with Lie algebra 
		isomorphic to $V$. Let $\varphi:G\times N\rightarrow N$ be the local Lie group action
		whose fundamental vector fields are spanned by $X_1,\ldots,X_r$. Then, the general  form of the integral curves,   $x(t)$, of $X$ can be written as $x(t)=\varphi(g(t),x_0)$, where $x_0\in N$ and $g(t)$ is  the particular solution of the differential equation (\ref{Aut2}) associated with the automorphic Lie system on $G$ of the form
		$
		X^G(t,g)=-{\displaystyle \sum_{\alpha=1}^r}b_\alpha (t)X_\alpha^R(g).
		$
	\end{theorem}

\section{Jacobi structures}\label{Sec:JS}
	
	This section briefly introduces the basic theory on Jacobi structures, independently introduced by Kirillov and Lichnerowicz(see  \cite{HLS15, Ki76,Li78,Va94} for further details). If not otherwise stated, it is assumed that all results and structures are real and globally defined. Moreover, manifolds are assumed to be connected. 
	
	\begin{definition}
 A {\it Jacobi manifold} is a triple $(M,\Lambda,R)$, where $M$ is a manifold, $\Lambda$ is a bivector field on $M$, and $R$ is a vector field, the so-called {\it Reeb vector field}, satisfying
		$$
		[\Lambda,\Lambda]_{SN}=-2R\wedge \Lambda,\qquad [R,\Lambda]_{SN}=0,
		$$
		where $[\cdot,\cdot]_{SN}$ is the \textit{Schouten-Nijenhuis bracket} on $M$ (cf.  \cite{Va94}).
	\end{definition}
	To understand the relation of our results with the formalism in other works, e.g. \cite{LMP97}, it is convenient to clarify that it is commonly assumed that $[\Lambda,\Lambda]'=2R\wedge \Lambda$ and $[\Lambda,R]'=0$ by using the original Schouten bracket $[\cdot,\cdot]'$ given in \cite{Sc40,Ni55,Va94}. 
 
	\begin{example} Since a Poisson manifold is a pair $(M,\Lambda)$ where $\Lambda$ is a bivector field on $M$ such that $[\Lambda,\Lambda]_{SN} = 0$ \cite{Va94}, every Poisson manifold can be understood as a Jacobi manifold of the form $(M,\Lambda,R=0)$.
	\end{example}
	
	\begin{example}
		The {\it continuous Heisenberg group} \cite{We00} is the matrix Lie group
		\begin{equation}
		\mathbb{H}=\left\{\left(\begin{array}{ccc}1&x&z\\0&1&y\\0&0&1\end{array}
		\right)\bigg|\,x,y,z\in\mathbb{R}\right\}.
		\end{equation}
		Then, $\{x,y,z\}$ is a natural global coordinate system on $\mathbb{H}$. Consider the bivector field on $\mathbb{H}$ given by
		\begin{equation}
		\Lambda_\mathbb{H}= -y\frac{\partial}{\partial y}\wedge\frac{\partial}{\partial z}+\frac{\partial}{\partial x}\wedge\frac{\partial}{\partial y}
		\end{equation}
		and the vector field $R_\mathbb{H} := \partial/\partial z$. Thus,
		$$
		[\Lambda_\mathbb{H},\Lambda_\mathbb{H}]_{SN}=-2\frac{\partial}{\partial x}\wedge\frac{\partial}{\partial y}\wedge\frac{\partial}{\partial z}=- 2R_\mathbb{H}\wedge \Lambda_\mathbb{H},\qquad [R_\mathbb{H},\Lambda_\mathbb{H}]_{SN}=0.
		$$
		Hence, $(\mathbb{H},\Lambda_\mathbb{H},R_\mathbb{H})$ is a Jacobi manifold.
	\end{example}
	
	\begin{example}\label{SL2}	
		The Lie group $SL_2$ is a  matrix Lie group of the form 
		\begin{equation}\label{Hei}
		SL_2:=\left\{\left(\begin{array}{cc}a&b\\c&d\end{array}
		\right)\bigg|ad-cb=1,a,b,c,d\in\mathbb{R}\right\}.
		\end{equation}
		The elements of $SL_2$ on the open neighbourhood 
		$$U=\left\{\left(\begin{array}{cc}a&b\\c&d\end{array}
		\right)\in SL_2\bigg|a\neq 0\right\}
		$$ of the neutral element of $ SL_2$ can be parametrised via $\{a,b,c\}$, which becomes a  coordinate system on $U$. A short calculation shows that a basis of left-invariant vector fields on $U$ read
		$$
		X^L_1 := a\partial_a - b \partial_b + c \partial_c, \quad X^L_2 := a \partial_b, \quad X^L_3 := b \partial_a + \left(\frac{1+bc}{a}\right) \partial_c.
		$$
		These vector fields satisfy the standard commutation relations for $SL_2$, namely
		\begin{equation}\label{ConRel}
		[X^L_1,X^L_2]=2X^L_2,\qquad [X^L_1,X^L_3]=-2X^L_3,\qquad [X^L_2,X^L_3]=X^L_1.
		\end{equation}
		Then,
		$$
		\Lambda_{\rm SL} := X^L_3 \wedge X^L_2 = ab \partial_a \wedge \partial_b - (1+bc) \partial_b \wedge \partial_c, \quad R_{\rm SL} := X^L_1 = a \partial_a-b\partial_b+c\partial_c,
		$$
		give rise to  a Jacobi manifold $(SL_2,\Lambda_{\rm SL_2},R_{\rm SL_2})$. 
	\end{example}
	
Every bivector field $\Lambda$ on a manifold $N$ leads to a mapping $\Lambda^\sharp:T^*M\rightarrow TM$ of the form $\Lambda^\sharp(\theta_p)=\Lambda_p(\theta_p,\cdot)\in T_p^*M$, where $\Lambda_p$  is the value of $\Lambda$ at an arbitrary $p\in M$, for every $\theta_p\in T^*_pM$.
	
	\begin{definition} \label{jacobi_vect_field}
		Let $(M,\Lambda,R)$ be a Jacobi manifold. A vector field on $M$ is called {\it Hamiltonian} if there exists a function $f\in C^\infty(M)$ so that$$
		X=[\Lambda,f]_{SN}+fR=\Lambda^\sharp({\rm d}f)+fR,
		$$
		where $f$ is called a {\it Hamiltonian function} of $f$. If $f$ is additionally a first-integral of $R$, then $f$ is called a {\it good Hamiltonian function} and $X$ is called a {\it good Hamiltonian vector field} \cite{HLS15}.
	\end{definition}
	
	We hereafter write ${\rm Ham}(M,\Lambda,R)$ for the space of Hamiltonian vector fields relative to the Jacobi manifold $(M,\Lambda,R)$. Note that the Reeb vector field of $(M,\Lambda,R)$ is the Hamiltonian vector field relative to the constant function $f=1$.

 The notion of a Hamiltonian vector field relative to a Poisson manifold $(M,\Lambda)$ is retrieved from the above definition as, in such a case, the Poisson manifold is a Jacobi manifold $(M,\Lambda,R=0)$ and then the Hamiltonian vector field related to a function $f\in C^\infty(M)$ is  of the form $X_f=\Lambda^\sharp({\rm d}f)$. 

A Hamiltonian vector field relative to a Jacobi manifold $(M,\Lambda,R)$ may have different Hamiltonian functions. 
 Finally, it is worth noting that a vector field $X$ on $M$ is Hamiltonian relative to a Jacobian manifold $(M,\Lambda,R)$ if and only if $\mathcal{L}_X\Lambda=-(Rf)\Lambda$ for some function $f\in C^\infty(M)$.	
	
	\begin{example} Given the Jacobi manifold $(\mathbb{H},\Lambda_\mathbb{H},R_\mathbb{H})$ and the vector field $X_1^L= \partial/\partial x$, one has that
		$$
		X_1^L=[\Lambda_\mathbb{H},-y]_{SN}-yR_\mathbb{H}=\Lambda^\sharp_\mathbb{H}(-{\rm d}y)-yR_\mathbb{H}.
		$$
		Hence, $X_1^L$ is a Hamiltonian vector field with Hamiltonian function $h_1^{L}=-y$ relative to $(\mathbb{H},\Lambda_\mathbb{H},R_\mathbb{H})$.
	\end{example}
	
	\begin{example} 
		Consider again the Jacobi manifold $(SL_2,\Lambda_{\rm SL},R_{\rm SL})$ on $SL_2$ given in Example \ref{SL2}. Then, the function on $SL_2$ given by $f(a,b,c)=a$ on $U$ and vanishing off $U$ has a Hamiltonian vector field 
		$$
		X_f = [\Lambda_{\rm SL},a]_{SN}+aR_{\rm SL}=a^2\partial_a +ac \partial_c.
		$$ 
	\end{example}
	
	It stems from Definition \ref{jacobi_vect_field} that every  $f \in C^{\infty}(M)$ admits a unique Hamiltonian vector field relative to $(M,\Lambda,R)$. Meanwhile, a Hamiltonian vector field may have several Hamiltonian functions. A condition ensuring the uniqueness of the associated Hamiltonian function is given in the following proposition.
	
	\begin{proposition}
		Let $(M, \Lambda,R)$ be a Jacobi manifold such that the distribution $\mathcal{D}^R$ spanned by $R$ satisfies $\mathcal{D}^R\cap {\rm Im}\, \Lambda^\sharp=0$ for ${\rm Im} \,\Lambda^\sharp := \{\Lambda^\sharp({\rm d}f): f \in C^{\infty}(M)\}$. Then, every Hamiltonian vector field admits a unique Hamiltonian function.
	\end{proposition}
	\begin{proof}
		If $f_1$ and $f_2$ are two different Hamiltonian functions of a vector field $X$ on $M$, then
		$$
		X =\Lambda^\sharp(\textrm{d}f_1) + f_1 R = \Lambda^\sharp(\textrm{d}f_2) + f_2 R\Longrightarrow
		0 = \Lambda^\sharp(\textrm{d}(f_1 - f_2)) + (f_1 - f_2)R.
		$$
		By assumption, $R \notin \textrm{Im}\,\Lambda^\sharp$. Thus, the above equality implies that $(f_1 - f_2)R = 0$. Since $R$ does not take values in ${\rm Im}\,\Lambda^\sharp$, it does not vanish and $f_1 = f_2$.
	\end{proof}

 \begin{definition} Every Jacobi manifold $(M,\Lambda,R)$ gives rise to a bracket on $\{\cdot,\cdot\}_{\Lambda,R}:C^\infty(M)\times C^\infty(M)\rightarrow C^\infty(M)$ of the form
 $$
 \{f,g\}_{\Lambda,R}:=\Lambda^\sharp(df,dg)+f Rg-gRf,\qquad \forall f,g\in C^\infty(M).
 $$
 \end{definition}
	The following proposition, whose proof can be found in \cite{Va02} or derived from the results in \cite{Va94}, will be of interest.
	
	\begin{proposition}\label{Prop:Use} Let $(M,\Lambda,R)$ be a Jacobi manifold and let $f\in C^\infty(M)$. Then,
		\begin{itemize}
			\item $X_{\{f,g\}}=[X_f,X_g]$,
			\item $[X_f,R]=-(Rf)R-\Lambda^\sharp({\rm d}(Rf))$,
			\item $\mathcal{L}_{X_f}\Lambda=-(Rf)\Lambda$.
		\end{itemize}
	\end{proposition}
	
	It stems from Proposition \ref{Prop:Use} that ${\rm Ham}(M,\Lambda,R)$ is a Lie algebra with respect to the standard Lie bracket of vector fields \cite{Va94}. Additionally, $C^\infty(M)$ can be endowed with a Lie algebra structure as shown below.
	
	\begin{proposition}\label{Prop:Inv}
		Let $(M, \Lambda,R)$ be a Jacobi manifold. The map $\{\cdot, \cdot\}_{\Lambda,R}$ is a Lie bracket on $C^{\infty}(M)$. This Lie bracket becomes a Poisson bracket if and only if $R=0$. Moreover, the morphism $\phi_{\Lambda,R}:f\in C^\infty(M)\mapsto X_f\in {\rm Ham}(\Lambda,R)$ is a Lie algebra morphism.
	\end{proposition}
	
	While Proposition \ref{Prop:Inv} ensures that every Jacobi manifold gives rise to a Lie bracket that acts as a derivation of up to order one in each entry, the following proposition, whose proof can be found in \cite{Va94}, ensures the converse of this result.  
	\begin{proposition} A Jacobi manifold on $M$ amounts to a Lie bracket $\{\cdot,\cdot\}:C^\infty(M)\otimes C^\infty(M)\rightarrow C^\infty(M)$ which is a derivation of up to order one in each entry.
	\end{proposition}
	
	\begin{definition} Let $(M,\Lambda,R)$ be a Jacobi manifold. The {\it characteristic distribution}, $\mathcal{D}^{\Lambda,R}$, of $(M,\Lambda,R)$  is the distribution on $M$ spanned by the elements of ${\rm Ham}(M,\Lambda,R)$. 
	\end{definition}
	
	It stems form Proposition \ref{Prop:Use} that the characteristic distribution $\mathcal{D}^{\Lambda,R}$ of a Jacobi manifold $(M,\Lambda,R)$ is invariant under the flows of Hamiltonian vector fields. 
  Therefore, one obtains the following result (see \cite[Theorem 2.9'']{Va94} for details).
	
	\begin{theorem} The characteristic distribution of a Jacobi manifold $(M,\Lambda,R)$ is integrable.
	\end{theorem}

	The integral leaves of the characteristic distribution of a Jacobi manifold can be divided into two types: even- and odd-dimensional. Since ${\rm Im}\,\Lambda^\sharp$ is an even-dimensional distribution,  even-dimensional characteristic distributions are those in which $R\in {\rm Im}\,\Lambda^\sharp$. Meanwhile, odd-dimensional ones are those one where $R\notin {\rm Im}\,\Lambda^\sharp$. A characteristic distribution $\mathcal{D}^{\Lambda,R}$ on $M$ is said to be {\it transitive} when $\mathcal{D}^{\Lambda,R}=TM$. In the case of a Poisson manifold $(M,\Lambda)$, the characteristic distribution in the Jacobi sense induced by $(M,\Lambda,R=0)$ retrieves the characteristic distribution of the Poisson manifold, which is given by ${\rm Im }\, \Lambda^\sharp$.

	\subsection{Jacobi manifolds and contact manifolds}\label{La:JacCon}
	
	This section  recalls that every transitive Jacobi manifold on an odd-dimensional manifold amounts to a so-called {\it co-orientable contact manifold}. In particular, this applies to the odd-dimensional leaves of the characteristic distribution of a Jacobi manifold. Let us recall the definition of a co-orientable contact manifold. For simplicity, and according to the standard convention in the literature, co-orientable contact manifolds will be simply called contact manifolds \cite{BW58}. 
	
	\begin{definition} A {\it contact manifold} is a pair $(M,\theta)$, where $M$ is a $(2m+1)$-dimensional manifold and $\theta$ is a differential one-form on $M$ satisfying that $\theta\wedge ({\rm d}\theta)^m\neq 0$. We call {\it Reeb vector field} of $(M,\theta)$ the unique vector field $R$ on $M$ such that $\iota_R\theta=1$ and $\iota_R{\rm d}\theta=0$.
	\end{definition}

 A Hamiltonian vector field on a contact manifold $(M,\theta)$ is a vector field on $M$ such that there exists a function $f\in C^\infty(M)$ so that $f=-\iota_X\theta$ and 
 $$
 df=\iota_Xd\theta+(Rf)\theta.
 $$
 It can proved that a vector field $X$ on $M$ is Hamiltonian relative to the contact form $\theta$ if and only if $\mathcal{L}_X\theta=-h\theta$ for some $h\in C^\infty(M)$. In such a case, $h=Rf$. 
 
	Let us show that every odd-dimensional leaf $\mathcal{F}$ of $\mathcal{D}^{\Lambda,R}$, the characteristic distribution of a Jacobi manifold $(M, \Lambda, R)$, amounts to a contact manifold. In this case, $R_x\notin ({\rm Im}\,\Lambda^\sharp)_x$ at every point $x \in \mathcal{F}$ and thus, $\mathcal{D}^{\Lambda,R}_x=({\rm Im}\, \Lambda^\sharp)_x\oplus \langle R_x\rangle$. Consequently,  there exists a unique differential one-form $\theta$ on $\mathcal{F}$, such that 
	$$
	\iota_R\theta=-1,\qquad \iota_{\Lambda^\sharp (\xi)}\theta=0,\qquad \forall \xi\in \Omega^1(\mathcal{F}).
	$$
It stems from Proposition \ref{Prop:Use} that $\iota_R{\rm d}\theta=0$. 
	Finally, one has to show that $\theta \wedge ({\rm d}\theta)^m \neq 0$. To prove it, one can use that  
	\begin{equation}\label{Eq:Exp}
	\iota_{\Lambda^\sharp(\eta)}{\rm d}\theta=-(\iota_R\eta)\theta+\eta,\qquad \forall \eta \in \Omega^1(\mathcal{F}).
	\end{equation}
	The above formula yields  that $\theta$ does not belong to the image of the mapping ${\rm d}\theta^{\flat}: v_p \in TM \mapsto {\rm d}\theta_p^{\flat}(v_p, \cdot) \in T^*M$. In fact, if $\theta=\iota_{\Upsilon}{\rm d}\theta$, for $\Upsilon\in \Omega^1(\mathcal{F})$, then (\ref{Eq:Exp}) shows that every $\eta\in \Omega^1(\mathcal{F})$ belongs to the image of ${\rm d}\theta^{\flat}$. This is impossible as the image of every two-form has even rank while  $\mathcal{F}$ is odd-dimensional. Moreover, (\ref{Eq:Exp}) shows that ${\rm d}\theta^{\flat}$ has  rank $2m$. Therefore, $\theta\wedge ({\rm d}\theta)^m\neq 0$.
	
	Conversely,  a contact manifold gives rise to a Jacobi manifold as follows. The Reeb vector field of the Jacobi manifold is assumed to be minus the Reeb vector field, $R$, of the contact one. Next, $\Lambda$ is defined to be  the mapping $\Lambda^\sharp:T^*M\rightarrow TM$ determined by the conditions
	$$
	\iota_{\Lambda^\sharp(\eta)}\theta=0,\qquad \iota_{\Lambda^\sharp(\eta)}{\rm d}\theta=\eta-(\iota_R\eta)\theta.
	$$
	In fact, since ${\rm d}\theta^{\flat}:TM\rightarrow T^*M$ is a $2m$-dimensional rank map, the second condition determines the value of $\Lambda^\sharp (\eta)$ up to an element of $\ker {\rm d}\theta$. Since $\iota_{\Lambda^\sharp(\eta)}\theta=0$, this fixes exactly the value of $\Lambda^\sharp$ and, therefore, $\Lambda$. From these conditions, it follows that $[\Lambda,\Lambda]_{SN}=-R\wedge \Lambda$ and $[R,\Lambda]_{SN}=0$. A simple way to prove the latter facts is to write $\theta={\rm d}z-\sum_{i=1}^np_i {\rm d}x^i$ using local  canonical coordinates $\{z,x^i,p_i\}$ via the Darboux theorem for contact forms. Then, $\Lambda=-\sum_{i=1}^n(\partial_{x_i}+p_i\partial_{z})\wedge \partial_{p_i}$ and $R=-\partial_z$. It is immediate that $[R,\Lambda]_{SN}=0$, whilst
	$$
	[\Lambda,\Lambda]_{SN}=-2\partial_{z}\wedge \sum_{i=1}^n\left[(\partial_{x_i}+p_i\partial_{z})\wedge \partial_{p_i}\right]=-2R\wedge \Lambda.
	$$
	
	\subsection{Jacobi manifolds and locally conformal symplectic manifolds}
	This section shows that every transitive Jacobi manifold $(M,\Lambda, R)$ on an even-dimensional $M$ amounts to a  so-called locally conformal symplectic manifold (see \cite{Va84,Va13} for details).  In particular, this applies to the even-dimensional leaves of the characteristic distribution of a Jacobi manifold $(M,\Lambda,R)$. There exists a plethora of works in the literature dealing with locally conformally  symplectic manifolds, in particular, they appear in the study of Hamilton-Jacobi theories for mechanical systems \cite{Ba02,ELSZ21}. 
	
	\begin{definition}A {\it locally conformal symplectic (l.c.s) form}  on $M$ is a non-degenerate differential two-form $\Theta$ on $M$ such that 
		$$
		{\rm d}\Theta=\eta\wedge \Theta
		$$
		for a closed-one form $\eta$ on $M$, the so-called {\it Lee form} \cite{Ba02}. The triple $(M,\Theta,\eta)$ is called a {\it l.c.s. manifold}. 
	\end{definition}
	
	It is worth noting that $\eta$ is unique. A l.c.s. structure $(M,\Theta,\eta)$ allows us to define a new cohomology by means of the operator ${\rm d}_\eta:\Upsilon\in \Omega(M)\mapsto {\rm d}\Upsilon-\eta\wedge \Upsilon\in \Omega(M)$, i.e. ${\rm d}_\eta^2=0$. 
	
	\begin{definition} Given a l.c.s. manifold $(M,\Theta,\eta)$, its Reeb vector field is the unique vector field $R$ on $M$ such that
	$$
	\iota_R\Theta=\eta.
	$$
	\end{definition}
	
	\begin{definition}A {\it Hamiltonian vector field} relative to a l.c.s. manifold $(M,\Theta,\eta)$ is a vector field $X$ on $M$ satisfying that there exists $f\in C^\infty(M)$ such that ${\rm d}_\eta f=\iota_{X}\Theta$.
	\end{definition}
	
	Since $\Theta$ is non-degenerate, every $f\in C^\infty(M)$ gives rise to an associated Hamiltonian vector field $X_f$ on $M$ relative to $(M,\Theta,\eta)$. Moreover,
	one can endow $C^\infty(M)$ with a Poisson structure $\{\cdot,\cdot\}_{\Theta}:C^\infty(N)\times C^\infty(N)\rightarrow C^\infty(N)$ by defining
	$$
	\{f,g\}_{\Theta}:=\Theta(X_g,X_f)=\iota_{X_f}{\rm d}_\eta g,\qquad \forall f,g\in C^\infty(M).
	$$
	It is worth noting that Hamiltonian vector fields do not need to be Lie symmetries of $\Theta$. Instead,
	$$
	\mathcal{L}_{X_f}\Theta=(\iota_{X_f} \eta)\Theta,\qquad \forall f\in C^\infty(M).
	$$
	If $\iota_{X_f}\eta=0$, we say that $X_f$ is horizontal. Moreover, $\iota_{X_f}\eta = 0$ if and only if $Rf = 0$. Indeed, since $\iota_R\eta = 0$, it follows that
	$$
	Rf = \iota_R {\rm d}f = \iota_R {\rm d}_{\eta} f = \iota_R \iota_{X_f} \Theta = - \iota_{X_f} \iota_R \Theta = - \iota_{X_f} \eta. 
	$$
	
	 Assume that the bivector field $\Lambda$ of a Jacobi manifold $(M,\Lambda,R)$ such that $\Lambda^\sharp: T^*M \to TM$ is invertible and it gives rise to a non-degenerate two-form $\Theta$ on $M$ such that $\Theta^\flat = (\Lambda^\sharp)^{-1}$. The following proposition shows how $\Theta$ gives rise to a locally conformal symplectic manifold on $M$ (see \cite{Va02}).
	
	\begin{proposition}
 Let $\Theta^\flat=\Lambda^{\sharp -1}$, where $\Lambda$ is  the bivector field of a Jacobi manifold $(M,\Lambda,R)$ and $\iota_{-R}\Theta=\eta$. Then, 
		$$
		{{\rm d}\Theta}=[{(\Lambda^\sharp)^{-1}\otimes (\Lambda^\sharp)^{-1}]([\Lambda,\Lambda]_{SN})}=\eta\wedge \Theta,\qquad \mathcal{L}_R\Theta=0.
		$$
	\end{proposition}
	
	Conversely, a l.c.s. manifold $(M,\Theta,\eta)$ gives rise to a transitive Jacobi manifold on $M$. In fact, the Reeb vector field $R$ of the Jacobi manifold is the only vector field on $M$ satisfying that 
	$
	\iota_R\Theta=\eta$, while $\Lambda$ is the only bivector field on $M$ whose induced $\Lambda^\sharp:T^*M\rightarrow TM$ satisfies that
	$$
	\iota_{\Lambda^\sharp (\beta)}\Theta=\beta,\qquad \forall \beta \in \Omega^1(M).
	$$
	
	\begin{definition} Let $(M,\Theta,\eta)$ be a locally conformal symplectic manifold.  A {\it good Hamiltonian} vector field relative to $(M,\Theta,\eta)$ is a vector field on $M$ such that
		$$
		\iota_{X}\Theta={\rm d}_\eta f,\qquad \mathcal{L}_X\Theta=0.
		$$
	\end{definition}

\chapter{Algebraic methods in the classification of three-dimensional Lie bialgebras}\label{Ch:alg_met}

In this Chapter, based on \cite{LW20}, we present a novel algebraic approach to the classification of coboundary Lie biaglebras. It relies on the so-called $\mathfrak{g}$-invariant maps, which allow to identify equivalent r-matrices. This notion and associated results are discussed at the beginning of the chapter. A separate part of this chapter is devoted to the analysis of a modified Yang--Baxter equation (mCYBE, for short). We show that it is sufficient to study this equation in the reduced Grassmann algebra $\Lambda_R \mathfrak{g} := \bigoplus_{n \in \mathbb{N}} \Lambda^n \mathfrak{g} / (\Lambda^n \mathfrak{g})^{\mathfrak{g}}$. We also demonstrate how other structures can be reduced to this space. Since $\mathfrak{g}$-invariant elements play a notable role in this procedure, we discuss several methods to obtain them. Among others, we focus on a technique which employs a Lie algebra gradation. Finally, all these tools are tested in the classification of equivalent r-matrices for three-dimensional Lie algebras. Our results are summarised in Table. We retrieve the classification presented in \cite{FJ15,Go00}.

	\section{The \texorpdfstring{$\mathfrak{g}$}{}-invariant maps on Grassmann algebras \texorpdfstring{$\Lambda\mathfrak{g}$}{}} \label{Ch:alg_Sec:inv}

Many notions introduced for Lie algebras can be extended to Grassmann algebras and even for general $\mathfrak{g}$-modules. Let us focus particularly on the ad-invariance.
	
	\begin{definition} \label{ginv}
		A $k$-linear map $b:V^{\otimes k}\rightarrow \mathbb{R}$ is \textit{$GL(\rho)$-invariant} relative to a $\mathfrak{g}$-module $(V,\rho)$ if $T^* b = b$ for every $T \in GL(\rho)$, i.e. $b(Tx_1, \ldots, Tx_k) = b(x_1, \ldots, x_k)$ for all $x_1, \ldots, x_k \in V$. Moreover, $b$ is {\it $\mathfrak{g}$-invariant} relative to  $(V, \rho)$ if 
		\begin{equation}\label{ginva}
			b(\rho_v(x_1), \ldots, x_k) + \ldots+b(x_1,\ldots,\rho_v(x_k))= 0,\qquad \forall v\in \mathfrak{g},\quad \forall x_1,\ldots,x_k\in V.
		\end{equation}
	\end{definition}

As explained next, $\mathfrak{g}$-invariance can be interpreted as an extension of ad-invariance to arbitrary $\mathfrak{g}$-modules. 

	\begin{example} 
Let us consider the bilinear form $\kappa_{\rho}$ on $\mathfrak{g}$ induced by a $\mathfrak{g}$-module $(V,\rho)$ of the form $\kappa_{\rho}(v_1, v_2) := \textrm{tr}(\rho_{v_1} \circ \rho_{v_2})$ for every $v_1,v_2 \in \mathfrak{g}$. It satisfies that $\kappa_{\rho}({\rm ad}_{v}(v_1), v_2) + \kappa_{\rho}(v_1, {\rm ad}_{v}( v_2)) = 0$ for all $v, v_1, v_2 \in \mathfrak{g}$ (cf. \cite{SW73}). Note that the above notion reduces to the usual Killing form for the $\mathfrak{g}$-module $(\mathfrak{g},{\rm ad})$. Thus, the Killing form $\kappa_\mathfrak{g}$ is $\mathfrak{g}$-invariant relative to $(\mathfrak{g}, \textrm{ad})$. 
	\end{example}
	
Given a $\mathfrak{g}$-module $(V,\rho)$, both introduced notions of invariance for $k$-linear maps, the invariance relative to $GL(\rho)$ and $\mathfrak{g}$-invariance, are related, as shown in Proposition \ref{prop:glrho_ginv}. The proof, based on the use of the exponential map, is quite immediate.
	
	\begin{proposition}\label{prop:glrho_ginv}
		A $k$-linear map $b:V^{\otimes k} \to \mathbb{R}$ is $GL(\rho)$-invariant relative to a $\mathfrak{g}$-module $(V,\rho)$ if and only if $b$ is $\mathfrak{g}$-invariant relative to $(V,\rho)$.
	\end{proposition}

Given $\mathfrak{g}$-invariant map on a vector space $V$, one can obtain an associated $\mathfrak{g}$-invariant map on a Grassmann algebra $\Lambda V$. 

Before discussing this fact, let us introduce useful notation. If $\{v_1,\ldots,v_r\}$ is a basis of $V$, we define $v_{J}:=v_{J(1)}\wedge \ldots \wedge v_{J(m)}$, where $J:=(J(1),\ldots,J(m))$ with $J(1), \ldots, J(m) \in \{1,\ldots, r\}$ represents a multi-index of length $|J|=m$. Moreover, we denote the the permutation group of $m$ elements by $S_m$ and ${\rm sg}(\sigma)$ stands for the sign of a permutation $\sigma \in S_m$.  
	
	\begin{theorem}\label{extension}
		Every $\mathfrak{g}$-invariant $k$-linear map $b:V^{\otimes k}\rightarrow \mathbb{R}$ relative to a $\mathfrak{g}$-module $V$ induces a $\mathfrak{g}$-invariant $k$-linear map, $b_{\Lambda V}: (\Lambda V)^{\otimes k} \to \mathbb{R}$, relative to the induced $\mathfrak{g}$-module on $\Lambda V$ such that:	
  \begin{itemize}
\item[1)] the spaces $\Lambda^mV$, with $m\in\mathbb{Z}$, are orthogonal between themselves relative to $b_{\Lambda V}$; 
\item[2)] $b_{\Lambda V}(1,\ldots,1)=1$; 
 \item[3)]  the restriction, $b_{\Lambda^mV}$, of $b_{\Lambda V}$ to $\Lambda^m V$, with $m\in\mathbb{N}$, satisfies
			
				\begin{equation}\label{Lambdag}
					b_{\Lambda^{m} V}(v_{J_1},\ldots ,v_{J_k}):=\!\!\!\!\!\!\!\sum_{\sigma_1,\ldots,\sigma_{k} \in S_m}\!\!\!\!\!\!{\rm sg }(\sigma_1\ldots\sigma_k)\frac{1}{m!}\prod_{r=1}^m b\left( v_{J_1(\sigma_1^{-1}(r))},\ldots ,v_{J_k(\sigma_{k}^{-1}(r))}\right).
			\end{equation}
   \end{itemize}
	\end{theorem}
	\begin{proof} 
 Recall that $\Lambda V$ is spanned by $1\in \Lambda^0V$ and the elements $v_J=v_{J(1)}\wedge\ldots\wedge v_{J(m)}$. Since $b_{\Lambda V}$ is $k$-linear, conditions 1, 2, and 3 establish $b_{\Lambda V}$ completely. 
 
 Let us notice that the restricted form $b_{\Lambda^m V}$ given in Condition 3 is defined on the representative elements $v_{J_s} \in \Lambda^m V$, with $s\in \{1,\ldots,k\}$. However, it is not obvious whether the value of $b_{\Lambda^m V}$ depend on the order of elements within each $v_{J_s}$. Let us show that it is indeed well defined, namely it is independent of the representative for each $v_{J_s}$, with $s\in \{1,\ldots,k\}$. Define $\sigma v_{J}:=v_{J(\sigma^{-1}(1))}\wedge\ldots\wedge v_{J(\sigma^{-1}(m))}$ and $\tilde{\tilde{\sigma}}_j :=\tilde\sigma_j\cdot \sigma_j $. Then,
		\vskip -0.4cm
		\begin{multline*}
		b_{\Lambda^m V}(\tilde{\sigma}_1v_{J_1},\ldots, \tilde{\sigma}_kv_{J_k})=\!\!\!\!\!\!\!\!\sum_{\sigma_1,\ldots,\sigma_k\in S_m}\!\!\!\!\!\!{\rm sg }(\sigma_1\ldots\sigma_k)\frac{1}{m!}\prod_{r=1}^mb\left(v_{J_1( \sigma_1^{-1}\tilde{\sigma}_1^{-1} (r))},\ldots,v_{J_k( \sigma_k^{-1}\tilde{\sigma}_k^{-1} (r))}\right)\\
		=\!\!\!\!\!\sum_{\tilde{\tilde{\sigma}}_1,\ldots,\tilde{\tilde{\sigma}}_k\in S_m}\!\!\!\!\!{\rm sg }(\tilde{\tilde{\sigma}}_1 \ldots\tilde{\tilde{\sigma}}_k){\rm sg }(\tilde\sigma_1\ldots\tilde\sigma_k)\frac{1}{m!}\prod_{r=1}^mb\left(v_{J_1 (\tilde{\tilde{\sigma}}_1^{-1}(r))},\ldots,v_{J_k (\tilde{\tilde{\sigma}}_k^{-1}(r))}\right).
		\end{multline*}
		Hence, $b_{\Lambda^m V}(\tilde{\sigma}_1v_{J_1},\ldots, \tilde{\sigma}_kv_{J_k})\!\!=\!\!{\rm sg }(\tilde\sigma_1\ldots\tilde \sigma_k)b_{\Lambda^m V}(v_{J_1},\ldots,v_{J_k})$  and $b_{\Lambda V}$ is well-defined.
		
		Let us prove that $b_{\Lambda V}$ is $\mathfrak{g}$-invariant relative to the $\mathfrak{g}$-module  on $\Lambda V$ induced by the $\mathfrak{g}$-module on $V$. By Proposition \ref{prop:glrho_ginv},  the $\mathfrak{g}$-invariance of $b_{\Lambda V}$ follows from its $GL(\Lambda\rho)$-invariance. Consequently, the restrictions $b_{\Lambda^mV}$ for $m\in \overline{\mathbb{N}}=\mathbb{N}\cup \{0\}$ are $GL(\Lambda^m\rho)$-invariant. Set $e^{\rho_v}:=\exp(\rho_v)$ for every $v\in \mathfrak{g}$ and $m\in \mathbb{N}$, Since $b$ is $GL(\rho)$-invariant, we get
$$		
\begin{aligned}
&b_{\Lambda^mV}(\Lambda^m e^{\rho_v}(v_{J_1}),\ldots,\Lambda^m e^{\rho_v}(v_{J_k}))=\sum_{\sigma_1,\ldots,\sigma_{k} \in S_m} \!\!\!\!\!\!\!\!{\rm sg }(\sigma_1\ldots\sigma_k)\frac{1}{m!}\prod_{r=1}^m b\left( e^{\rho_v}v_{J_1(\sigma_1^{-1}(r))},\ldots , e^{\rho_v}v_{J_k(\sigma_{k}^{-1}(r))}\right)\\
 &=\!\!\!\!\!\!\sum_{\sigma_1,\ldots,\sigma_{k} \in S_m} \!\!\!\!\!\!\!\!{\rm sg }(\sigma_1 \ldots\sigma_k)\frac{1}{m!}\prod_{r=1}^m b\left(v_{J_1(\sigma_1^{-1}(r))},\ldots , v_{J_k(\sigma_{k}^{-1}(r))}\right)= b_{\Lambda^mV}(v_{J_1},\ldots, v_{J_k}).
		\end{aligned}
		$$
		As the invariance of $b_{\Lambda^0V}$ is obvious, $b_{\Lambda V}$ is $GL(\Lambda\rho)$-invariant and Proposition \ref{prop:glrho_ginv} ensures that is $\mathfrak{g}$-invariant.
	\end{proof}
		
	Next corollary gives an immediate consequence of Proposition \ref{prop:glrho_ginv} and Theorem \ref{extension}. 
	
	\begin{corollary}\label{ExtInvRep} If $b$ is a $\mathfrak{g}$-invariant $k$-linear map on $V$, then  $b_{\Lambda^mV}$ is $GL(\Lambda^m \rho)$-invariant, i.e. $b_{\Lambda^mV}(\Lambda^m T\cdot , \ldots, \Lambda^m T\cdot )=b_{\Lambda^mV}( \cdot ,\ldots, \cdot )$ for every  $T\in GL(\rho)$.
	\end{corollary}

 Not all extensions of $\mathfrak{g}$-invariant $k$-linear symmetric maps described in Theorem \ref{extension} are nontrivial. Under certain conditions, discussed in the following proposition, we obtain zero maps.
 
	\begin{proposition} If $b$ is a $\mathfrak{g}$-invariant $k$-linear map on $V$, then $b_{\Lambda^mV}=0$ for $m >1$ and odd $k>1$. 
	\end{proposition}
	\begin{proof} 	
 Firstly, let us introduce the equivalence relation in $S_m^k := {S_m \times \stackrel{k\,\,{\rm times}}{\ldots} \times S_m}$. For any $k$-tuple $(\sigma_1, \ldots, \sigma_k), (\tilde\sigma_1, \ldots, \tilde\sigma_k) \in S_m^k$, we set
		$
		(\sigma_1, \ldots, \sigma_k) \equiv (\tilde\sigma_1, \ldots, \tilde\sigma_k)$ if and only if there exists $ \sigma \in S_m$ such that $(\tilde\sigma_1, \ldots, \tilde\sigma_k) =  (\sigma \sigma_1, \ldots, \sigma \sigma_k)$. It is immediate to verify that such a relation on $S_m^k$ is an equivalence relation. We denote the space of equivalence classes in $S_m^k$ by $\mathcal{R}$.
  
  Let $[(\sigma_1,\ldots,\sigma_k)]$ be the equivalence class of $(\sigma_1,\ldots,\sigma_k)\in S_m^k$. For a generic element ${\bf w}$ of $S^k_m$, the map $b_{\Lambda^mV}$ defined by (\ref{Lambdag}) satisfies
		
		\begin{equation*}
			b_{\Lambda^{m}V}(v_{J_1},\ldots ,v_{J_k}):=\sum_{\substack{[{\bf w}]\in \mathcal{R}\\(\sigma_1,\ldots,\sigma_{k}) \in [{\bf w}]}}\!\!\!\!\!\!\!\!\!\!{\rm sg }(\sigma_1\ldots\sigma_k)\frac{1}{m!}\prod_{r=1}^m b\left( v_{J_1(\sigma_1^{-1}(r))},\ldots ,v_{J_k(\sigma_{k}^{-1}(r))}\right).
		\end{equation*}
Since every equivalence class  reads $[(\sigma_1,\ldots,\sigma_k)]=\{(\sigma\sigma_1,\ldots,\sigma\sigma_k):\sigma \in S_m\}$, we obtain

		\begin{equation*}
			b_{\Lambda^{m} V}(v_{J_1},\ldots ,v_{J_k})=\!\!\!\!\!\!\sum_{\substack{[(\sigma_1 ,\ldots,\sigma_k)]\in \mathcal{R}\\\sigma\in S_m}}\!\!\!\!\!\!{\rm sg }(\sigma\sigma_1 \ldots\sigma\sigma_k)\frac{1}{m!}\prod_{r=1}^m b\left( v_{J_1(\sigma_1^{-1}\sigma^{-1}(r))},\ldots ,v_{J_k(\sigma_{k}^{-1}\sigma^{-1}(r))}\right).
		\end{equation*}
		
		Let us show that the above sum vanishes for every equivalence class of $\mathcal{R}$.
		First, notice that
		$$
		\prod_{r=1}^mb\left(v_{J_1({\sigma}_1^{-1} \sigma^{-1} (r))},\ldots,v_{J_k({\sigma}_k^{-1} \sigma^{-1} (r))}\right)=\prod_{r=1}^mb\left(v_{J_1({\sigma}_1^{-1} (r))},\ldots,v_{J_k({\sigma}_k^{-1} (r))}\right).
		$$
  since $\sigma^{-1} \in S_m$ is a bijection and thus, the index $\sigma^{-1}(r)$ in the right-hand product runs over a set $\{1, \ldots, m\}$ as well.
  
	Define ${\rm sg}({\bf w}):={\rm sg}(\widehat{\sigma}_1\ldots\widehat{\sigma}_k)$ and $\sigma {\bf w}:=(\sigma\widehat{\sigma}_1,\ldots,\sigma\widehat{\sigma}_k)$ for ${\bf w}=(\widehat{\sigma}_1,\ldots,\widehat{\sigma}_k)$ and $\sigma\in S_m$. As $k$ is odd, one has ${\rm sg}(\sigma {\bf w})={\rm sg}(\sigma)^{k}{\rm sg}({\bf w})={\rm sg}(\sigma){\rm sg}(\bf w)$. 
	
		Since $|S_m| = m!$, every equivalence class of $\mathcal{R}$ has $m!$ elements having the same absolute value. Half of them are odd for $m>1$ and the elements of the other half are even. Hence, 
		$$
		b_{\Lambda^{m} V}(v_{J_1},\ldots ,v_{J_k})=\sum_{\sigma\in S_m}\!\!{\rm sg }(\sigma {\bf w})\frac{1}{m!}\prod_{r=1}^m b\left(v_{J_1(\sigma_1^{-1}(r))},\ldots , v_{J_k(\sigma_{k}^{-1}(r))}\right)=0.
		$$
  This yields $b_{\Lambda^mV}=0$.
	\end{proof}

	\begin{example} \label{Ex:su2}
		Consider the Lie algebra $\mathfrak{su}_2$ and its Killing form $\kappa_{\mathfrak{su}_2}$, which is a $\mathfrak{su}_2$-invariant, bilinear, symmetric map on $\mathfrak{su}_2$. Take the basis $\{e_1, e_2, e_3\}$ of $\mathfrak{su}_2$ given in Table \ref{tabela3w}. Theorem \ref{extension}  extends $\kappa_{\mathfrak{su}_2}$ to the  forms $\kappa_{\Lambda^2 \mathfrak{su}_2}$, $\kappa_{\Lambda^3 \mathfrak{su}_2}$ on $\Lambda^2\mathfrak{su}_2$ and $\Lambda^3\mathfrak{su}_2$, respectively. In the bases $\{e_{12}, e_{13}, e_{23}\}$, $\{e_{123}\}$ for the spaces $\Lambda^2\mathfrak{su}_2, \Lambda^3 \mathfrak{su}_2$ (see Table \ref{tabela3w}), we obtain 
		$
		[\kappa_{\mathfrak{su}_2}]=-2{\rm Id}_{3\times 3}$, 
		$[\kappa_{\Lambda^2\mathfrak{su}_2}]=4{\rm Id}_{3\times 3}$, 
		$[\kappa_{\Lambda^3\mathfrak{su}_2}]=\left(-8\right).$
	\end{example}
	
	We observe in Example \ref{Ex:su2} that $\kappa_{\mathfrak{su}_2}$ and its extensions to $\Lambda^2\mathfrak{su}_2$ and $\Lambda^3\mathfrak{su}_2$ are simultaneously diagonal and non-degenerate. This fact is explained by the corollary below.
	
	\begin{corollary}\label{Cor:Diatwoform} 
 If $b$ is a symmetric $\mathfrak{g}$-invariant $k$-linear mapping on an $r$-dimensional $\mathfrak{g}$-module $V$, then $b_{\Lambda V}$ is symmetric. If $b$ is bilinear, then: a) it diagonalises in a basis $\{e_1,\ldots,e_r\}$ and  $b_{\Lambda^mV}$ diagonalises in the basis $\{e_J\}_{|J|=m}$; b) $b$ is non-degenerate if and only if $b_{\Lambda V}$ is non-degenerate as well. 
	\end{corollary}
	\begin{proof} 
 If $b$ is a symmetric $\mathfrak{g}$-invariant $k$-linear mapping on $V$, then the symmetry of $b_{\Lambda^m V}$ on decomposable elements of $\Lambda^m V$, $m\in \mathbb{N}$, is guaranteed by condition 3 of Theorem \ref{extension}, whereas condition 2 ensures the same for $m=0$.  Since $b_{\Lambda V}$ is multilinear, it is symmetric on  $\Lambda V$.
		
		If $b$ is bilinear and symmetric, it can always be put into diagonal form in a certain basis $\{e_1,\ldots, e_r\}$ for $V$. This gives rise to the bases $\{e_J\}_{m\geq |J|\geq 0}$ of $\Lambda^m V$, for $m \in \{1, \ldots r\}$, and $\{e_J\}_{r\geq |J|\geq 0}$ of $\Lambda V$. The expression (\ref{Lambdag}) for $b_{\Lambda^m V}$ implies that $b_{\Lambda^m V}$ is also diagonal. In consequence, $b_{\Lambda V}$ is diagonal as well. The element $b_{\Lambda V}(e_J,e_J)$ lying in its diagonal, for $|J|\geq 1$, reads
		$
		\prod_{j=1}^{|J|} b(e_{J(j)},e_{J(j)})$. Thus, we conclude that $b$ is non-degenerate if and only if the induced symmetric $b_{\Lambda^mV}$ on each $\Lambda^mV$ is non-degenerate as well. 
	\end{proof}
	\begin{example}\label{ExExsl2}Consider the Lie algebra $\mathfrak{sl}_2$ and a basis $\{e_1,e_2,e_3\}$ satisfying the commutation relations in Table \ref{tabela3w}. In the induced bases $\{e_{12},e_{13},e_{23}\}$ and $\{e_{123}\}$ in $\Lambda^2\mathfrak{sl}_2$ and $\Lambda^3\mathfrak{sl}_{2}$, respectively, one has

{\small \begin{equation}\label{sl2A}
			[\kappa_{\mathfrak{sl}_{2}}] \!=\! \left(\begin{array}{ccc}
				2&0&0\\
				0&0&2\\
				0&2&0\\
			\end{array}\right),\qquad
			[\kappa_{\Lambda^2 \mathfrak{sl}_{2}}]\! =\!\left(\begin{array}{ccc}
				0&4&0\\
				4&0&0\\
				0&0&-4\\
			\end{array}\right),\qquad\, [\kappa_{\Lambda^3\mathfrak{sl}_{2}}]\!=\!(-8).
		\end{equation}}

		Since $\mathfrak{sl}_{2}$ is simple, the Cartan criterion states that $\kappa_{\mathfrak{sl}_2}$ is non-degenerate \cite{FH91}. Then, Corollary \ref{Cor:Diatwoform} ensures that $\kappa_{\Lambda^2 \mathfrak{sl}_{2}}$ and $\kappa_{\Lambda^3 \mathfrak{sl}_{2}}$ must be non-degenerate. Their expressions showed in (\ref{sl2A}) confirm that observation.
	\end{example}
	
	\section{Killing-type forms}\label{Ch:alg_Sec:kil}
	The Killing form \cite{FH91} gives the most crucial example of a symmetric bilinear form, playing an indisputable role in the structure theory of Lie algebras. In this section we focus on certain multilinear symmetric maps on the spaces $\Lambda^m\mathfrak{g}$ induced by Killing forms and study their invariance properties. These maps will be of interest in the description of coboundary cocommutators in further sections.

As shown by Theorem \ref{extension}, any $\mathfrak{g}$-invariant k-linear symmetric map on $\mathfrak{g}$ gives rise to the $\mathfrak{g}$-invariant k-linear symmetric map on the Grassmann algebra $\Lambda \mathfrak{g}$. For brevity, extended maps obtained via the method described in Theorem \ref{extension} from the Killing form will be called {\it Killing-type forms}. Such forms share similar properties as the Killing form. The most important one is the invariance relative to the automorphism group ${\rm Aut}(\mathfrak{g})$, as illustrated by the following proposition. 
 
	\begin{proposition}\label{ExtKil} 
 The Killing-type form $\kappa_{\Lambda \mathfrak{g}}$ is invariant relative to the action of ${\rm Aut}(\mathfrak{g})$ on $\Lambda \mathfrak{g}$. In particular, $\kappa_{\Lambda\mathfrak{g}}$ is $\mathfrak{aut}(\mathfrak{g})$-invariant relative to the $\mathfrak{aut}(\mathfrak{g})$-module $(\Lambda\mathfrak{g},\Lambda\widehat{\rm ad})$.
	\end{proposition}
	\begin{proof}
 Denote by $v_{J_1},v_{J_2}$ the decomposable elements of $\Lambda^m \mathfrak{g}$, for $m\geq 1$ and set $T\in {\rm Aut}(\mathfrak{g})$. Since the Killing form $\kappa_\mathfrak{g}$ is invariant relative to the action of ${\rm Aut}(\mathfrak{g})$ \cite{Ha15}, we obtain
		{
			\begin{multline}\label{relatt}
				\kappa_{\Lambda \mathfrak{g}}(\Lambda^m Tv_{J_1},\Lambda^m Tv_{J_2}):=\!\!\!\!\!\!\sum_{\sigma_1, \sigma_2\in S_m}\!\!\!\!{\rm sg }(\sigma_1\sigma_2)\frac1{m!}\prod_{r=1}^m \kappa_\mathfrak{g}\left( Tv_{J_1(\sigma_1^{-1}(r))},Tv_{J_2(\sigma_{2}^{-1}(r))}\right)\\
				=\!\!\!\!\!\!\sum_{\sigma_1, \sigma_2\in S_m}\!\!\!\!{\rm sg }(\sigma_1\sigma_2)\frac 1{m!}\prod_{r=1}^m \kappa_\mathfrak{g}\left( v_{J_1(\sigma_1^{-1}(r))},v_{J_2(\sigma_{2}^{-1}(r))}\right)=\kappa_{\Lambda \mathfrak{g}}(v_{J_1},v_{J_2}).
		\end{multline}}
  Let us note that the expression (\ref{relatt}) is satisfied for decomposable elements of $\Lambda^m\mathfrak{g}$, which span $\Lambda^m\mathfrak{g}$. Thus, the bilinearity of $\kappa_{\Lambda \mathfrak{g}}$ implies that for any $m \in \{1 ,\ldots, r \}$, the induced form $\kappa_{\Lambda^m \mathfrak{g}}$ is invariant relative to the action of ${\rm Aut}(\mathfrak{g})$ on $\Lambda^m \mathfrak{g}$. Moreover, since the spaces $\Lambda^m\mathfrak{g}$ for different $m$ are orthogonal relative to $\kappa_{\Lambda\mathfrak{g}}$ by condition 1 of Theorem \ref{extension}, the whole $\kappa_{\Lambda\mathfrak{g}}$ is invariant relative to the action of ${\rm Aut}(\mathfrak{g})$. 
  
  Let us prove the second part of this proposition. In view of Proposition \ref{prop:glrho_ginv}, the map $\kappa_{\Lambda\mathfrak{g}}$ is $\mathfrak{aut}(\mathfrak{g})$-invariant if and only if $\kappa_{\Lambda\mathfrak{g}}$ is $GL(\Lambda\widehat{{\rm ad}})$-invariant. This in turn means that $\kappa_{\Lambda\mathfrak{g}}(\Lambda T\cdot ,\Lambda T\cdot )=\kappa_{\Lambda\mathfrak{g}}(\cdot ,\cdot)$ for every $T\in {\rm Aut}_c(\mathfrak{g})$, where ${\rm Aut}_c(\mathfrak{g})$ stands for the connected part of ${\rm Aut}(\mathfrak{g})$ containing its neutral element. This equality is an immediate consequence of (\ref{relatt}).
	\end{proof}
	
	Since $\kappa_{\Lambda\mathfrak{g}}$ is invariant under the maps $\Lambda T$, with $T\in {\rm Aut}(\mathfrak{g})$, it is therefore invariant under $\Lambda T$ with $T\in {\rm Inn}(\mathfrak{g})$. In view of Proposition \ref{prop:glrho_ginv}, the $\kappa_{\Lambda\mathfrak{g}}$ is also $\mathfrak{g}$-invariant relative to $(\Lambda\mathfrak{g},{\rm ad})$.
	
Inspired by the structure of the Killing form, we define its $k$-linear version and show in the proposition below that it enjoys similar invariance properties. 
 
	\begin{proposition}\label{Previo}
		The map $b(v_1, \ldots, v_k) := \sum_{\sigma \in S_k} {\rm Tr}\left({\rm ad}_{v_{\sigma (1)}} \circ \ldots \circ {\rm ad}_{v_{\sigma (k)}}\right)$, where $v_1, \ldots, v_k \in \mathfrak{g}$, is invariant under the action of ${\rm Aut}(\mathfrak{g})$ and $\mathfrak{aut}(\mathfrak{g})$-invariant relative to $(\mathfrak{g},\widehat{\rm ad})$. The $k$-linear anti-symmetric map  given by
		$
		b(v_1, \ldots, v_k) := \sum_{\sigma \in S_k} {\rm sg}(\sigma){\rm Tr}({\rm ad}_{v_{\sigma (1)}} \circ \ldots \circ {\rm ad}_{v_{\sigma (k)}})$, for all $v_1, \ldots, v_k \in \mathfrak{g},
		$ 
		is $\mathfrak{aut}(\mathfrak{g})$-invariant with respect to the $\mathfrak{aut}(\mathfrak{g})$-module $(\mathfrak{g}, \widehat{{\rm ad}})$.
	\end{proposition}
\begin{proof}
First, notice that for any $v \in \mathfrak{g}$ and $A \in {\rm Aut}(\mathfrak{g})$, we have ${\rm ad}_{Av} = A \circ {\rm ad}_v \circ A^{-1}$. Indeed, given $z \in \mathfrak{g}$, we obtain $[Av, z] = [Av, AA^{-1}z] = A[v, A^{-1}z]$, which yields the desired result. Thus,
\begin{equation*}
\begin{split}
b(Av_1, \ldots, Av_k) &= \sum_{\sigma \in S_k} {\rm sg}(\sigma){\rm Tr}({\rm ad}_{Av_{\sigma (1)}} \circ \ldots \circ {\rm ad}_{Av_{\sigma (k)}}) \\
&= \sum_{\sigma \in S_k} {\rm sg}(\sigma){\rm Tr}(A \circ {\rm ad}_{v_{\sigma (1)}} \circ A^{-1} \circ \ldots \circ A \circ {\rm ad}_{v_{\sigma (k)}} \circ A^{-1}) \\
&= \sum_{\sigma \in S_k} {\rm sg}(\sigma){\rm Tr}(A \circ {\rm ad}_{v_{\sigma (1)}} \circ \ldots \circ {\rm ad}_{v_{\sigma (k)}} \circ A^{-1}) \\
&= \sum_{\sigma \in S_k} {\rm sg}(\sigma){\rm Tr}(A^{-1} \circ A \circ {\rm ad}_{v_{\sigma (1)}} \circ \ldots \circ {\rm ad}_{v_{\sigma (k)}}) \\
&= \sum_{\sigma \in S_k} {\rm sg}(\sigma){\rm Tr}({\rm ad}_{v_{\sigma (1)}} \circ \ldots \circ {\rm ad}_{v_{\sigma (k)}}) = b(v_1, \ldots, v_k)
\end{split}
\end{equation*}
By virtue of Proposition \ref{prop:glrho_ginv}, the map $b$ is also $\mathfrak{aut}(\mathfrak{g})$-invariant. One proves the second statement in a similar manner.
\end{proof}

 Another way of obtaining $\mathfrak{g}$-invariant $k$-linear symmetric maps on $\mathfrak{g}$ comes from a polynomial Casimir element $C$ of order $k$, i.e. a symmetric element of $\mathfrak{g}^{\otimes k}$ satisfying that $\otimes^k {\rm ad}_v C = 0$ for every $v\in \mathfrak{g}$. Recall that $\otimes^k {\rm ad}_v: \mathfrak{g}^{\otimes k} \to \mathfrak{g}^{\otimes k}$ given by 
 $$
 \otimes^k {\rm ad}_v (w_1 \otimes \ldots \otimes w_k) := \sum_{j=1}^k w_1 \otimes \ldots \otimes {\rm ad}_v w_j \otimes \ldots \otimes w_k
 $$
for any $w_1, \ldots, w_k \in \mathfrak{g}$.

	\begin{theorem}\label{PolPol}
		Every polynomial Casimir element $C$ of order $k$ on a Lie algebra $\mathfrak{g}$ induces a $\mathfrak{g}$-invariant $k$-linear symmetric map on $\mathfrak{g}$ given by $b(v_1,\ldots,v_k):=C(\widetilde{\kappa}_\mathfrak{g}(v_1),\ldots,\widetilde{\kappa}_\mathfrak{g}(v_k))$ for every $v_1, \ldots, v_k \in \mathfrak{g}$, where $\widetilde{\kappa}_\mathfrak{g}: v\in \mathfrak{g} \mapsto \kappa_{\mathfrak{g}}(v,\cdot)\in \mathfrak{g}^*$ is a map induced by a Killing form $\kappa_{\mathfrak{g}}$ on $\mathfrak{g}$.
	\end{theorem}
	\begin{proof} 
 Since $\kappa_\mathfrak{g}$ is $\mathfrak{g}$-invariant, we have 
		$$
		[{\rm ad}_v^*\circ\widetilde{\kappa}_\mathfrak{g}(v_1)](v_2)\!=\!\widetilde{\kappa}_\mathfrak{g}(v_1)({\rm ad}_vv_2)\!=\!\kappa_{\mathfrak{g}}(v_1,{\rm ad}_vv_2)\!=\!-\kappa_{\mathfrak{g}}({\rm ad}_vv_1,v_2)\!=\!-[\widetilde{\kappa}_\mathfrak{g}\circ {\rm ad}_v(v_1)](v_2),$$
		for all $v,v_1,v_2\in\mathfrak{g}$.	Hence,
		$\widetilde\kappa_\mathfrak{g}\circ {\rm ad}_v=-{\rm ad}_v^*\circ \widetilde\kappa_\mathfrak{g}$ for every $v\in \mathfrak{g}$. As $C$ is a Casimir element, $\otimes^k {\rm ad}_v C=0$. Thus for all $ v_1,\ldots,v_k,v\in\mathfrak{g}$, we get
		\begin{equation*}
  \begin{split}
&\sum_{j=1}^{k} b(v_1, \ldots,\textrm{ad}_v v_j,\ldots, v_k) = \sum_{j=1}^{k} C(\widetilde \kappa_\mathfrak{g}(v_1), \ldots, \widetilde \kappa_\mathfrak{g}(\textrm{ad}_v v_j),  \ldots, \widetilde \kappa_\mathfrak{g}(v_k)) \\
  &= \sum_{j=1}^{k} \!C(\widetilde\kappa_\mathfrak{g}(v_1), \ldots, {\rm ad}^*_v\widetilde\kappa_\mathfrak{g}(v_j), \ldots, \widetilde\kappa_\mathfrak{g}(v_k)) 
= \sum_{j=1}^{k} [(1^{\otimes j-1} \otimes {\rm ad}_v \otimes 1^{\otimes k-j}) C] (\widetilde\kappa_\mathfrak{g}(v_1), \ldots, \widetilde\kappa_\mathfrak{g}(v_k)) \\
   &=(\otimes^k {\rm ad}_v C)(\widetilde\kappa_\mathfrak{g}(v_1), \ldots, \widetilde\kappa_\mathfrak{g}(v_k)) \!=\! 0.
   \end{split}
		\end{equation*}
  Hence, the map $b$ is $\mathfrak{g}$-invariant.
	\end{proof}
	
	If $\mathfrak{g}$ is semi-simple, then the converse of Theorem \ref{PolPol} holds and any $\mathfrak{g}$-invariant $k$-linear symmetric form on $\mathfrak{g}$ amounts to a Casimir element. Indeed, it is known cite{FH91} that the map $\widetilde \kappa_{\mathfrak{g}}$ is a vector space isomorphism for a semi-simple Lie algebra $\mathfrak{g}$ over a field of characteristic zero. Thus, any map $b: \mathfrak{g}^{\otimes k} \to \mathbb{R}$ can be understood as a mapping $C: (\mathfrak{g}^*)^{\otimes k} \to \mathbb{R}$. Then, the computation presented in the proof of Theorem \ref{PolPol} ensures that $C$ is a Casimir element. 
 
 If $\mathfrak{g}$ is not semi-simple,  $\widetilde\kappa_\mathfrak{g}$ is not invertible and the set of $\mathfrak{g}$-invariant multilinear symmetric maps is usually bigger than the set of Casimir elements. In consequence, such maps allow a more detailed study of r-matrix equivalence in the future.
	
\section{Obtaining \texorpdfstring{$\mathfrak{g}$}{}-invariant bilinear maps}\label{Invgbil}\label{Ch:alg_Sec:invcalc}
	
Although the calculation of $\mathfrak{g}$-invariant maps seems to be an easy linear algebra exercise, it often happens to be computationally time-consuming, especially when $\dim\mathfrak{g}$ is large. In this section, we present several methods to simplify this task.

	\begin{proposition}\label{prop:sym_form}
		Let $b$ be a $\mathfrak{g}$-invariant bilinear map on a $\mathfrak{g}$-module $V$. Then, $
		b({\rm Im}\, \rho_v, {\rm ker}\, \rho_v) = 0$ for every $v\in \mathfrak{g}$ and $x\in V$.  Moreover, $b({\rm ad}_v(w),w)=0$ for every $v, w\in \mathfrak{g}$. If $b$ is additionally symmetric, then 
		$b(\rho_v(x),x)=0$ for any $v\in \mathfrak{g}$ and $x\in V$.
	\end{proposition}
	\begin{proof}
	Every $x_1\in {\rm Im}\rho_{v}$ can be written as $x_1 :=\rho_{v}(x_3)$ for some $x_3\in V$. Assume that $x_2 \in {\rm ker}\, \rho_{v}$. As $b$ is $\mathfrak{g}$-invariant, $b(x_1, x_2) = b(\rho_{v}(x_3), x_2) = -b(x_3, \rho_{v}(x_2)) = 0$ and $b({\rm Im}\,\rho_v,\ker \rho_v)=0$. 
 For any $v_1,v_2 \in \mathfrak{g}$, we get $b(\textrm{ad}_{v_1}(v_2), v_2) = -b(\textrm{ad}_{v_2}(v_1), v_2) = b(v_1, \textrm{ad}_{v_2}(v_2)) = 0$.
Finally by the symmetricity of $b$, we obtain $b(\rho_v(x), x) = -b(x, \rho_v(x)) = -b(\rho_v(x),x)$ for all $v \in \mathfrak{g}$ and $x\in V$. Therefore, $b(\rho_v(x), x) = 0$ for every $x\in V$ and $v\in \mathfrak{g}$. 
	\end{proof}
	
	\begin{proposition}\label{prop:form_cond}
		Let $b:V^{\otimes 2} \to \mathbb{R}$ be a $\mathfrak{g}$-invariant bilinear map relative to the $\mathfrak{g}$-module $(V, \rho)$. If $W$ is a two-dimensional linear subspace of $V$ satisfying that $\rho_v W\subset W$ for any $v \in \mathfrak{g}$, then for any linearly independent $f_s, f_t \in W$ one has
		\begin{equation}\label{exp}
			{\rm Tr}(\rho_v|_W)b(f_s, f_t)f_t\wedge f_s =b(f_s, f_s)(\rho_v f_t)\wedge f_t + b(f_t, f_t) f_s\wedge(\rho_v f_s).
		\end{equation}
	\end{proposition}
	\begin{proof} 
 Since $\rho_v W\subset W$, there exists constants $\alpha_1,\alpha_2,\beta_1,\beta_2\in \mathbb{R}$ such that
		$\rho_v f_s=\alpha_1f_s+\alpha_2f_t$ and $\rho_v f_t=\beta_1f_s+\beta_2f_t$.
		The $\mathfrak{g}$-invariance of $b$ ensures that
		\begin{equation*}
			\alpha_1 b(f_s, f_t) + \alpha_2 b(f_t, f_t) = b(vf_s, f_t)
			= -b(f_s, vf_t) = -\beta_1 b(f_s, f_s) - \beta_2 b(f_s, f_t).
		\end{equation*}
  The above expression gives (\ref{exp}) after rearranging.
	\end{proof}
	\begin{example}\label{ex:heisenberg} 
		Consider the 3D Heisenberg Lie algebra $\mathfrak{h}$. Take a basis $\{e_1,e_2,e_3\}$ of $\mathfrak{h}$ as in Table \ref{tabela3w}. Since $\mathfrak{h}$ is nilpotent, its Killing form $\kappa_\mathfrak{h}$ vanishes \cite[pg. 480]{Ha15}. In the bases $\{e_{12}, e_{13}, e_{23}\}$ and $\{e_{123}\}$ of $\Lambda^2 \mathfrak{h}$ and $\Lambda^3 \mathfrak{h}$, respectively, a $\mathfrak{h}$-invariant form $b$ and its extensions to $\Lambda^2 \mathfrak{h}$ and $\Lambda^3 \mathfrak{h}$ given by Propositions \ref{extension} and \ref{prop:sym_form} read
		$$
		[b]:=\left(\begin{array}{ccc}
		\alpha_1&\alpha_2&0\\
		\alpha_3&\alpha_4&0\\
		0&0&0
		\end{array}\right), \quad
		[b_{\Lambda^2\mathfrak{h}}]:=\left(\begin{array}{ccc}
		\alpha_1 \alpha_4 - \alpha_2 \alpha_3&0&0\\
		0&0&0\\
		0&0&0\\
		\end{array}\right),\quad 
		[b_{\Lambda^3\mathfrak{g}}]:=\left(0\right), 
		$$
		for $\alpha_1,\ldots,\alpha_4\in \mathbb{R}.$ Using again Proposition \ref{prop:sym_form}, we can compute the general $\mathfrak{h}$-invariant bilinear map $\widetilde{b}$ on $\Lambda^2\mathfrak{h}$, which reads
		$$
		[\widetilde{b}]:=\left(\begin{array}{ccc}
		\beta_3&\beta_2&\beta_1\\
		\beta_4&0&0\\
		\beta_5&0&0\\
		\end{array}\right).
		$$
		\vskip -0.2cm
		If $\widetilde{b}$ is symmetric, Proposition \ref{prop:sym_form} yields $\beta_1 = \beta_5 = 0$ and $\beta_2 = \beta_4 = 0$. If $\widetilde{b}$ is anti-symmetric, then $\beta_5 = -\beta_1$, $\beta_4 = -\beta_2$ and $\beta_3 = 0$. All given $\widetilde{b}$ are $\mathfrak{h}$-invariant.
	\end{example}
 
	\begin{example}\label{Ex:r31}
		Let us consider the Lie algebra $\mathfrak{r}_{3,1} := \langle e_1, e_2, e_3 \rangle$ with commutation relations given in Table \ref{tabela3w}. Using Proposition \ref{prop:sym_form}, we obtain the following form of a bilinear $\mathfrak{r}_{3,1}$-invariant map $\omega_{\Lambda^2\mathfrak{r}_{3,1}}$ on $\Lambda^2\mathfrak{r}_{3,1}$ in the basis $\{e_{12}, e_{13}, e_{23}\}$ of $\Lambda^2 \mathfrak{r}_{3,1}$:
		$$
		[\omega_{\Lambda^2\mathfrak{r}_{3,1}}]:=\left(\begin{array}{ccc}
		0&a&0\\
		b&0&0\\
		0&0&0\\
		\end{array}\right)
		$$
		for some $a,b \in \mathbb{R}$. With the help of Proposition \ref{prop:form_cond}, let us find the admissible values of coefficients $a, b$. Assuming $v:=e_1$ and $W := {\rm span}(e_{12}, e_{13})$, we obtain $\rho_v e_{12} = e_{12}$, $\rho_v e_{13} = e_{13}$. Thus, ${\rm Tr} \rho_v \vert_{W} = 2$. It is then easy to verify that $b(e_{12}, e_{13}) = 0$ and $b(e_{13}, e_{12}) = 0$. Therefore, $a=b=0$ and there are no $\mathfrak{r}_{3,1}$-invariant forms on $\Lambda^2\mathfrak{r}_{3,1}$.
	\end{example}
	
	\section{\texorpdfstring{$G$}{}-Gradations, Grassmann algebras, and CYBEs}\label{Ch:alg_Sec:grad}
	
In this section, we discuss another method for an efficient study of mCYBE solutions. It is based on a particular type of gradation admitted by a Lie algebra $\mathfrak{g}$. We show that it induces a decomposition of $\Lambda \mathfrak{g}$ which is compatible, in a way to be stated, with the algebraic Schouten bracket.
	
	\begin{definition} \label{Def:Ggrad}
 We say that $\mathfrak{g}$ admits a $G$-{\it gradation} if
		$
		\mathfrak{g}= \bigoplus_{\alpha \in G\subset \mathbb{R}^n} \mathfrak{g}^{(\alpha)},$
		where $(G\subset \mathbb{R}^n,\star)$ is a commutative group, the $\mathfrak{g}^{(\alpha)}$ are subspaces  of $\mathfrak{g}$, and $[\mathfrak{g}^{(\alpha)},\mathfrak{g}^{(\beta)}]\subset \mathfrak{g}^{(\alpha\star\beta)}$ for all $\alpha, \beta \in G$. The spaces $\mathfrak{g}^{(\alpha)}$, for $\alpha\in G$, are called the {\it homogeneous spaces} of the gradation and $\alpha$ is the {\it degree} of $\mathfrak{g}^{(\alpha)}$.
	\end{definition}

A $G$-gradation is a type of graded Lie algebra \cite{Va94}. Obviously, every $\mathfrak{g}$ admits a trivial $\mathbb{Z}$-gradation with $\mathfrak{g}^{(0)}=\mathfrak{g}$ and  $\mathfrak{g}^{(\alpha)}=\langle 0\rangle$ with $\alpha\neq 0$. If $\mathfrak{g}$ admits a $G$-gradation and $G$ is known from context or of minor importance, we will simply say that $\mathfrak{g}$ admits a gradation. It is worth noticing that several results presented in this section can be extended for $G$ being a semigroup. However, such a generalisation will not be necessary for our purposes.
	
	\begin{example} \label{Gradsu2} 
 It stems from  the commutation relations in Table \ref{tabela3w} that $\mathfrak{su}_2$ is isomorphic to the Lie algebra on $\mathbb{R}^3$ with the Lie bracket given by the vector product in $\mathbb{R}^3$, denoted by $\times$. Consider a basis $\{e_a, e_b, e_c\}$ of $\mathbb{R}^3$, where $e_a\in \mathbb{R}^3\backslash\{0\}$, the vector $e_b$ is perpendicular to $e_a$, and $e_c :=e_a\times e_b$. Then, $\mathfrak{su}_2$ admits a $\mathbb{Z}_2$-gradation $\mathfrak{g}^{(0)}=\langle e_a\rangle$ and $\mathfrak{g}^{(1)}=\langle e_b, e_c\rangle$. Since $e_a$ is not unique, $\mathfrak{su}_2$ admits several $\mathbb{Z}_2$-gradations.
	\end{example}

\begin{example}
Consider the Lie algebra $\mathfrak{sl}_2 = {\rm span}(e_1, e_2, e_3)$ with commutation relations
$$
[e_1, e_2] = e_2, \quad [e_1, e_3] = -e_3, \quad [e_2, e_3] = e_1
$$
One can immediately verify that $\mathfrak{sl}_2$ admits a $\mathbb{Z}$-gradation with $\mathfrak{g}_{1} = \langle e_2 \rangle$, $\mathfrak{g}_{-1} = \langle e_3 \rangle$, $\mathfrak{g}_{0} = \langle e_1 \rangle$ and $\mathfrak{g}_{\alpha} = \{0 \}$ for any $\alpha \in \mathbb{Z} \backslash \{-1, 0, 1\}$.
\end{example}

\begin{example}\label{Ex:rootdec}
Any complex semisimple Lie algebra $\mathfrak{g}$ over an algebraically closed field of characteristic zero admits the so-called root space decomposition \cite{Jacobson}, namely
$$
\mathfrak{g} = \mathfrak{h} \oplus \bigoplus_{\alpha \in \Phi} \mathfrak{g}_{\alpha}
$$
for a Cartan subalgebra $\mathfrak{h}$ and a root system $\Phi$. Among other properties, the subspaces $\mathfrak{g}_{\alpha}$ satisfy $[\mathfrak{h}, \mathfrak{g}_{\alpha}] \subseteq \mathfrak{g}_{\alpha}$ and $[\mathfrak{g}_{\alpha}, \mathfrak{g}_{\beta}] \subseteq \mathfrak{g}_{\alpha + \beta}$, whenever $\alpha + \beta \in \Phi$. It can be shown that the root system gives rise to the lattice $\Pi$ in $\mathbb{R}^n$ (see \cite{FH91} for details). In consequence, we can extend the root system to the whole lattice $\Pi$ and construct a generalised root space decomposition by associating previously introduces subspaces $\mathfrak{g}_{\alpha}$ for $\alpha \in \Phi$ and $\mathfrak{g}_{\beta} = \{0\}$ for $\beta \in \Pi \backslash \Phi$. Thus, a root space decomposition gives rise to a $\mathbb{Z}^k$-gradation, where $k = {\rm dim}(\mathfrak{h})$.
\end{example}
 
Further examples of gradations are presented for all three-dimensional Lie algebras in Table \ref{tabela3w}. Additionally, Figures \ref{so22_diags} and \ref{so32_diags} give $\mathbb{Z}^2$-gradations for the special pseudo-orthogonal Lie algebras $\mathfrak{so}(3,2)$ and $\mathfrak{so}(2,2)$ \cite{CO13,HWZ92}. 

Example \ref{Ex:rootdec} details a type of gradation that stems from a specific structure of semisimple Lie algebras given by the root decomposition. Although more general Lie algebras do not admit such a decomposition, it is sometimes possible to introduce a $G$-gradation whose properties are close to standard root decompositions. This kind of gradation will be useful for deriving $\mathfrak{g}$-invariant subspaces in $\Lambda\mathfrak{g}$.
	
	\begin{definition}\label{Def:rootgrad}
 A $\mathbb{Z}^k$-gradation on a Lie algebra $\mathfrak{g}$ is called a {\it root $\mathbb{Z}^k$-gradation} if there exist an abelian Lie subalgebra $\mathfrak{g}^{(0)} \subseteq \mathfrak{g}$ and an injective group homomorphism $\Xi:\alpha\in \mathbb{Z}^k \mapsto \widetilde{\alpha}\in \mathfrak{g}^{(0)*}$ such that $\dim \mathfrak{g}^{(0)}=k$ and $[e,e^{(\alpha)}]=\widetilde{\alpha}(e)e^{(\alpha)}$ for every $e\in \mathfrak{g}^{(0)}$ and $e^{(\alpha)}\in \mathfrak{g}^{(\alpha)}$.
	\end{definition}
	
	Since $\mathfrak{g}^{(0)}$ is  abelian, every root decomposition gives rise to a root $\mathbb{Z}^k$-gradation. 
 
 \begin{example}\label{Ex:so22}
 Consider the basis $\{e_{-}, e_0, e_{+}, f_{-}, f_0, f_{+}\}$ of $\mathfrak{so}(2,2)\simeq \mathfrak{sl}_2\oplus\mathfrak{sl}_2$, where $\{e_{-}, e_0, e_{+}\}$, $\{f_{-}, f_0, f_{+}\}$ are bases of each copy of $\mathfrak{sl}_2$ within $\mathfrak{so}(2,2)$. The first diagram of Figure \ref{so22_diags} shows a root decomposition for $\mathfrak{so}(2,2)$ leading to a root $\mathbb{Z}^2$-gradation given by $\mathfrak{so}(2,2)^{(0)} = {\rm span}(e_0,f_0)$ and $\Xi:(i,j)\in \mathbb{Z}^2\mapsto ie^0+jh^0\in (\mathfrak{so}(2,2)^{(0)})^*$, where $\{e^0,f^0\}$ denote the elements dual to $\{e_0,f_0\}$, accordingly.

 		\begin{figure}[h!]
			\begin{center}
				\begin{minipage}{\textwidth}
					\begin{center}
						\begin{tikzpicture}[scale=0.7]
						{\tiny 
							\draw [help lines,dotted] (-2,-2) grid (2,2);
							\draw [-] (-2,0)--(2,0) node[right] {$e^0$};
							\draw [-] (0,-2)--(0,2) node[above] {$f^0$};
							\color{blue}
							\filldraw (1,0) circle (2pt) node[above] {$e_{+}$};
							\filldraw (0,1) circle (2pt) node[left] {$f_{+}$};
							\filldraw (-1,0) circle (2pt) node[below] {$e_{-}$};
							\filldraw (0,-1) circle (2pt) node[right] {$f_{-}$};
							\filldraw (0,0) circle (2pt) node[below right] {$e_0$} node[above right] {$f_0$};\color{black}}
						\end{tikzpicture}$\quad\quad$
						\begin{tikzpicture}[scale=0.7]
						{\tiny
							\draw [help lines,dotted] (-2,-2) grid (2,2);
							\draw [-] (-2,0)--(2,0) node[right] {$e^0$};
							\draw [-] (0,-2)--(0,2) node[above] {$f^0$};
							\color{blue}
							\filldraw (1,0) circle (2pt) node[below] {\color{black}$1$\color{blue}};
							\filldraw (1,1) circle (2pt);
							\filldraw (0,1) circle (2pt) node[above right] {\color{black}$1$\color{blue}};
							\filldraw (-1,0) circle (2pt) node[below] {\color{black}$-1$\color{blue}};
							\filldraw (-1,-1) circle (2pt);
							\filldraw (0,-1) circle (2pt) node[below] {\color{black}$-1$\color{blue}};
							\filldraw (1,-1) circle (2pt);
							\filldraw (-1,1) circle (2pt);
							\filldraw (0,0) circle (2pt) node[below right] {\color{black}$0$\color{blue}};
							\draw[red, dashed] (-1,-1)--(-1,1);
							\draw[red, dashed] (-1,1)--(1,1);
							\draw[red, dashed] (1,1)--(1,-1);
							\draw[red, dashed] (1,-1)--(-1,-1);
							\color{black}}
						\end{tikzpicture}$\quad\quad$
						\begin{tikzpicture}[scale=0.7]
						{\tiny 
							\draw [help lines,dotted] (-2,-2) grid (2,2);
							\draw [-] (-2,0)--(2,0) node[right] {$e^0$};
							\draw [-] (0,-2)--(0,2) node[above] {$f^0$};
							\color{blue}
							\filldraw (1,0) circle (2pt) node[below] {\color{black}$1$\color{blue}};
							\filldraw (1,1) circle (2pt);
							\filldraw (0,1) circle (2pt) node[above right] {\color{black}$1$\color{blue}};
							\filldraw (-1,0) circle (2pt) node[below] {\color{black}$-1$\color{blue}};
							\filldraw (-1,-1) circle (2pt);
							\filldraw (0,-1) circle (2pt) node[below] {\color{black}$-1$\color{blue}};
							\filldraw (1,-1) circle (2pt);
							\filldraw (-1,1) circle (2pt);
							\filldraw (0,0) circle (2pt) node[below right] {\color{black}$0$\color{blue}};
							
							\color{black}}
						\end{tikzpicture}
					\end{center}
					\caption{Root decomposition of $\mathfrak{so}(2,2)$ (left), and the induced decompositions on $\Lambda^2 \mathfrak{so}(2,2)$ (center) and $\Lambda^3 \mathfrak{so}(2,2)$ (right) given by Theorem \ref{thm:root}. Non-zero homogeneous spaces are indicated by bold points. A basis for each homogeneous subspace of $\mathfrak{so}(2,2)$ is detailed in the first diagram. Limit homogeneous spaces are represented by points over a dashed line.}\label{so22_diags}
				\end{minipage}
			\end{center}
		\end{figure}
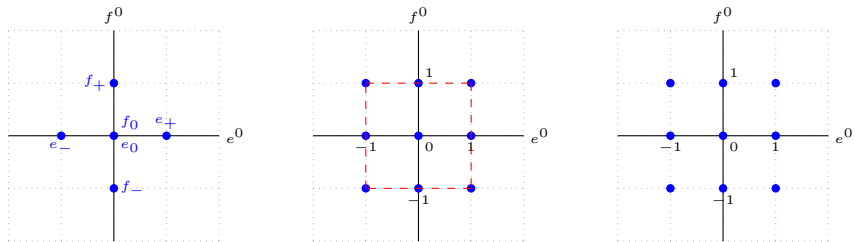
 \end{example}

As stated in Definition \ref{Def:Ggrad}, homogeneous spaces $\mathfrak{g}_{\alpha}$, $\alpha \in G$, of a $G$-gradation on a Lie algebra $\mathfrak{g}$ satisfy $[\mathfrak{g}_{\alpha}, \mathfrak{g}_{\beta}] \subseteq \mathfrak{g}_{\alpha \star \beta}$. Thus, we say that the Lie bracket on $\mathfrak{g}$ is compatible with the decomposition into homogeneous spaces. A $G$-gradation on $\mathfrak{g}$ induces a decomposition of the Grassmann algebra $\Lambda \mathfrak{g}$ and its associated subspaces also enjoy the similar compatibility property relative to the Schouten bracket. The details of this construction are presented in the following theorem.
 	
	\begin{theorem}\label{thm:root}
Let $\mathfrak{g}$ admit a $G$-gradation  $
		\mathfrak{g}=\bigoplus_{\alpha \in G\subset \mathbb{R}^n} \mathfrak{g}^{(\alpha)}$. Then, each subspace $\Lambda^m\mathfrak{g}$ of a Grassmann algebra $\Lambda \mathfrak{g} := \bigoplus _{m \in \mathbb{N}} \Lambda^m \mathfrak{g}$ admits a decomposition $\Lambda^m\mathfrak{g}=\bigoplus_{\alpha\in G}(\Lambda^m\mathfrak{g})^{(\alpha)}$ into so-called {\it homogeneous spaces} of the form
	\begin{equation}\label{dec}
(\Lambda^m\mathfrak{g})^{(\alpha)}:=\!\!\!\!\!\bigoplus_{\stackrel{\{\alpha_1,\ldots,\alpha_m\}\subset G}{\alpha_1\star\ldots\star\alpha_m=\alpha}}\!\!\!\!\!\mathfrak{g}^{(\alpha_1)}\wedge\ldots\wedge \mathfrak{g}^{(\alpha_m)}, \, m>0,\qquad (\Lambda^0\mathfrak{g})^{(0)}=\mathbb{R}.
		\end{equation}
 Moreover, the homogeneous spaces are compatible with the algebraic Schouten bracket on $\Lambda\mathfrak{g}$, namely $[(\Lambda^p\mathfrak{g})^{(\alpha)},(\Lambda^q\mathfrak{g})^{(\beta)}]_{S} \subset \Lambda^{p+q-1}\mathfrak{g}^{(\alpha\star\beta)}$, for $p,q\in \mathbb{Z}$, $\alpha,\beta\in G.$
	\end{theorem}
	\begin{proof} 
First, let us show that the homogeneous subspaces are well defined. Given any $m \geq 1$ and $\alpha \in G$, let us consider the exterior product  $e^{(\alpha_1)}\wedge\ldots\wedge e^{(\alpha_m)}$ of elements $e^{(\alpha_i)}\in \mathfrak{g}^{(\alpha_i)}$ for every $i=1,\ldots,m$, such that $\alpha=\alpha_1\star\ldots\star\alpha_m$. Since $G$ is commutative, such elements belong to $(\Lambda^m\mathfrak{g})^{(\alpha)}$ regardless of the order of $e^{(\alpha_i)}$ in the exterior product. 

Let us assume $w \in (\Lambda^m \mathfrak{g})^{(\beta_1)}$ and $w \in (\Lambda^m \mathfrak{g})^{(\beta_2)}$ for $\beta_1 \neq \beta_2$. For simplicity, we can take $w$ to be indecomposable. Recall that any indecomposable $w \in (\Lambda^m \mathfrak{g})^{(\alpha)}$ is of the form $e^{\alpha_1} \wedge \ldots \wedge e^{\alpha_m}$, where $\alpha_1 \star \ldots \star \alpha_m = \alpha$. Since $w \in (\Lambda^m \mathfrak{g})^{(\alpha)} \cap (\Lambda^m \mathfrak{g})^{(\beta)}$, we get that $\alpha_1 \star \ldots \alpha_m = \alpha$ and simultaneously $\alpha_1 \star \ldots \alpha_m = \beta$. This contradicts the assumption $\alpha \neq \beta$. 

Moreover, any element of $\Lambda^m \mathfrak{g}$ is a finite sum of indecomposable elements of $\Lambda^m \mathfrak{g}$, each of which were belong to a given $(\Lambda^m \mathfrak{g})^{(\alpha)}$. Hence, $\Lambda^m \mathfrak{g} \subseteq \bigoplus_{\alpha\in G}(\Lambda^m\mathfrak{g})^{(\alpha)}$. Conversely, any element of $(\Lambda^m \mathfrak{g})^{(\alpha)}$ for arbitrary $\alpha \in G$ belong to $\Lambda^m \mathfrak{g}$ and thus, any finite sum of such elements belong to $\Lambda^m \mathfrak{g}$ as well. Thus, $\bigoplus_{\alpha\in G}(\Lambda^m\mathfrak{g})^{(\alpha)} \subseteq \Lambda^m \mathfrak{g}$ and in consequence, $\Lambda^m \mathfrak{g} = \bigoplus_{\alpha\in G}(\Lambda^m\mathfrak{g})^{(\alpha)}$. 

The above argument justifies that there exists a decomposition of $\Lambda^m \mathfrak{g}$ into homogeneous spaces that is indeed a direct sum of $(\Lambda^m \mathfrak{g})^{\alpha}$.
		
Finally, let us take the indecomposable elements $v \in (\Lambda^p\mathfrak{g})^{(\alpha)}$ and $w \in (\Lambda^q\mathfrak{g})^{(\beta)}$ and consider their Schouten bracket. Since both elements can be written as $v = e^{\alpha_1} \wedge \ldots \wedge e^{\alpha_p}$ for $\alpha_1 \star \ldots \star \alpha_p = \alpha$ and $w = e^{\beta_1} \wedge \ldots \wedge e^{\beta_q}$ for $\beta_1 \star \ldots \beta_q = \beta$, the Schouten bracket $[v, w]_S$ reads
$$
[v, w]_S = \sum_{i=1}^p \sum_{j =1}^q (-1)^{i+j} [e^{\alpha_i}, e^{\beta_j}] \wedge e^{\alpha_1} \wedge \ldots \wedge \widehat{e^{\alpha_i}} \wedge \ldots e^{\alpha_p} \wedge e^{\beta_1} \wedge \ldots \wedge \widehat{e^{\beta_j}} \wedge \ldots \wedge e^{\beta_q}
$$
Obviously, each term of the above sum belongs to $\Lambda^{p+q-1}$. Let us analyse their degrees. Since $G$ is commutative and $[e^{\alpha_i}, e^{\beta_j}] \subseteq \mathfrak{g}^{(\alpha_i \star \beta_j)}$, it follows that the degree of each summand equals 
\begin{equation*}
\begin{split}
&(\alpha_i \star \beta_j) \star \alpha_1 \star \ldots \star \alpha_{i-1} \star \alpha_{i+1} \star \ldots \star \alpha_p \star \beta_1 \star \ldots \star \beta_{j-1} \star \beta_{j+1} \star \ldots \star \beta_q \\
&= (\alpha_1 \star \ldots \star \alpha_p) \star (\beta_1 \star \ldots \star \beta_q) = \alpha \star \beta.
\end{split}
\end{equation*}
Therefore, the Schouten bracket $[(\Lambda^p\mathfrak{g})^{(\alpha)}, (\Lambda^q\mathfrak{g})^{(\beta)}]_S \subseteq (\Lambda^{p+q}\mathfrak{g})^{(\alpha \star \beta)}$.
	\end{proof}

	\begin{example}\label{ex:decomp_so22} 
Consider a root $\mathbb{Z}^2$-gradation on $\mathfrak{so}(2,2)$, described in Example \ref{Ex:so22} and depicted in the first diagram of Figure \ref{so22_diags}. By Theorem \ref{thm:root}, it gives rise to induced decompositions on $\Lambda^2\mathfrak{so}(2,2)$ and $\Lambda^3(\mathfrak{so}(2,2))$. Their non-zero homogeneous spaces are indicated by bold points in the central and right diagrams of Figure \ref{so22_diags}. For instance, the homogeneous space $(\Lambda^2\mathfrak{so}(2,2))^{(1,1)} = \langle e_{+} \wedge f_{+}\rangle$ corresponds to the upper-right point in the second diagram of Figure \ref{so22_diags}.
	\end{example}

 \begin{example}
	Let $\{j_{\pm},j_3,k_{\pm},k_3,s_{\pm},r_{\pm}\}$ be a basis of $\mathfrak{so}(3,2)$ satisfying the non-vanishing commutation relations \cite{CO13}
\begin{align*}
			&[j_{\pm},k_{\pm}]=\pm r_{\pm},&[j_{\mp},k_{\pm}]&=\pm s_{\pm},
			&[j_{\mp},r_{\pm}]&=\pm 2k_{\pm},&[j_3,r_{\pm}]&=\pm r_{\pm},\\
			&[j_{\pm},s_{\pm}]=\pm 2k_{\pm},
			&[j_3,s_{\pm}]&=\mp s_{\pm},
			&[k_{\mp},r_{\pm}]&=\pm 2j_{\pm},& [k_3,r_{\pm}]&=\pm r_{\pm},\\
			&[s_+,s_-]=-4(k_3-j_3),
			&[k_{\mp},s_{\pm}]&=\pm 2j_{\mp},&[k_3,s_{\pm}]&=\pm s_{\pm},&[r_+,r_-]&=-4(k_3+j_3),\\
			&&[k_-,k_+]&=2k_3,&[j_-,j_+]&=-2j_3.
	\end{align*}
The diagrams for the homogeneous spaces of $\Lambda \mathfrak{so}(3,2)$ are presented in Figure \ref{so32_diags}.

		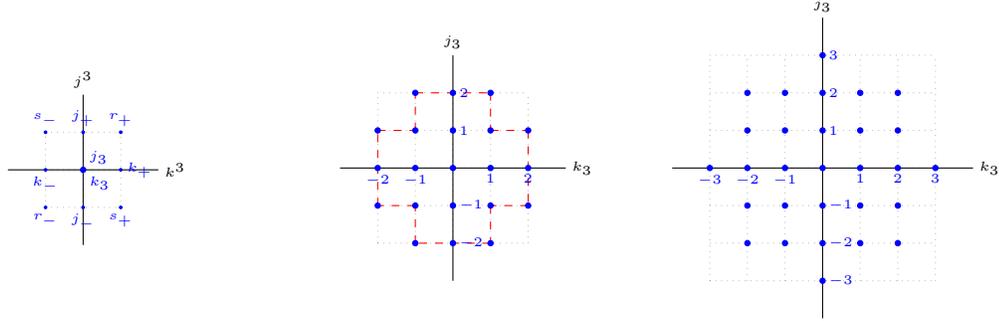
\begin{figure}[h!]
		\centering
		\begin{minipage}{0.3\textwidth}
			\begin{center}
				\begin{tikzpicture}[scale=0.5]
				{\tiny 
					\draw [help lines,dotted] (-1,-1) grid (1,1);
					\draw [-] (-2,0)--(2,0) node[right] {$k^3$};
					\draw [-] (0,-2)--(0,2) node[above] {$j^3$};
					\color{blue}
					\filldraw (1,0) circle (1pt) node[right] {$k_{+}$};
					\filldraw (1,1) circle (1pt) node[above] {$r_{+}$};
					\filldraw (0,1) circle (1pt) node[above] {$j_{+}$};
					\filldraw (-1,0) circle (1pt) node[below] {$k_{-}$};
					\filldraw (-1,-1) circle (1pt) node[below] {$r_{-}$};
					\filldraw (0,-1) circle (1pt) node[below] {$j_{-}$};
					\filldraw (-1,1) circle (1pt) node[above] {$s_{-}$};
					\filldraw (1,-1) circle (1pt) node[below] {$s_{+}$};
					\filldraw (0,0) circle (2pt) node[below right] {$k_3$} node[above right] {$j_3$};\color{black}}
				\end{tikzpicture}
			\end{center}
		\end{minipage}
		\begin{minipage}{0.3\textwidth}
			\begin{center}
				\begin{tikzpicture}[scale=0.5]
				{\tiny
					\draw [help lines,dotted] (-2,-2) grid (2,2);
					\draw [-] (-3,0)--(3,0) node[right] {$k_3$};
					\draw [-] (0,-3)--(0,3) node[above] {$j_3$};
					\color{blue}
					\draw[red, dashed] (-1,2)--(1,2);
					\draw[red, dashed] (1,2)--(1,1);
					\draw[red, dashed] (1,1)--(2,1);
					\draw[red, dashed] (2,1)--(2,-1);
					\draw[red, dashed] (2,-1)--(1,-1);
					\draw[red, dashed] (1,-1)--(1,-2);
					\draw[red, dashed] (1,-2)--(-1,-2);
					\draw[red, dashed] (-1,-2)--(-1,-1);
					\draw[red, dashed] (-1,-1)--(-2,-1);
					\draw[red, dashed] (-2,-1)--(-2,1);
					\draw[red, dashed] (-2,1)--(-1,1);
					\draw[red, dashed] (-1,1)--(-1,2);
					\filldraw (1,0) circle (2pt) node[below] {$1$};
					\filldraw (2,0) circle (2pt) node[below] {$2$};
					\filldraw (2,1) circle (2pt);
					\filldraw (1,1) circle (2pt);
					\filldraw (1,2) circle (2pt);
					\filldraw (0,1) circle (2pt) node[right] {$1$};
					\filldraw (0,2) circle (2pt) node[right] {$2$};
					\filldraw (-1,1) circle (2pt);
					\filldraw (-1,2) circle (2pt);
					\filldraw (-1,0) circle (2pt) node[below] {$-1$};
					\filldraw (-2,0) circle (2pt) node[below] {$-2$};
					\filldraw (-1,-2) circle (2pt);
					\filldraw (-1,-1) circle (2pt);
					\filldraw (-2,-1) circle (2pt);
					\filldraw (0,-1) circle (2pt) node[right] {$-1$};
					\filldraw (0,-2) circle (2pt) node[right] {$-2$};
					\filldraw (-2,1) circle (2pt);
					\filldraw (1,0-2) circle (2pt);
					\filldraw (2,-1) circle (2pt);
					\filldraw (1,-1) circle (2pt);
					\filldraw (0,0) circle (2pt);\color{black}}
				\end{tikzpicture}
			\end{center}
		\end{minipage}
		\begin{minipage}{0.3\textwidth}
			\begin{center}
				\begin{tikzpicture}[scale=0.5]
				{\tiny 
					\draw [help lines,dotted] (-3,-3) grid (3,3);
					\draw [-] (-4,0)--(4,0) node[right] {$k_3$};
					\draw [-] (0,-4)--(0,4) node[above] {$j_3$};
					\color{blue}
					\filldraw (1,0) circle (2pt) node[below] {$1$};
					\filldraw (2,0) circle (2pt) node[below] {$2$};
					\filldraw (3,0) circle (2pt) node[below] {$3$};
					\filldraw (1,1) circle (2pt);
					\filldraw (2,2) circle (2pt);
					\filldraw (2,1) circle (2pt);
					\filldraw (1,2) circle (2pt);
					\filldraw (0,1) circle (2pt) node[right] {$1$};
					\filldraw (0,2) circle (2pt) node[right] {$2$};
					\filldraw (0,3) circle (2pt) node[right] {$3$};
					\filldraw (-1,1) circle (2pt);
					\filldraw (-1,2) circle (2pt);
					\filldraw (-2,1) circle (2pt);
					\filldraw (-2,2) circle (2pt);
					\filldraw (-1,0) circle (2pt) node[below] {$-1$};
					\filldraw (-2,0) circle (2pt) node[below] {$-2$};
					\filldraw (-3,0) circle (2pt) node[below] {$-3$};
					\filldraw (-1,-1) circle (2pt);
					\filldraw (-2,-1) circle (2pt);
					\filldraw (-1,-2) circle (2pt);
					\filldraw (-2,-2) circle (2pt);
					\filldraw (0,-1) circle (2pt) node[right] {$-1$};
					\filldraw (0,-2) circle (2pt) node[right] {$-2$};
					\filldraw (0,-3) circle (2pt) node[right] {$-3$};
					\filldraw (1,-1) circle (2pt);
					\filldraw (2,-1) circle (2pt);
					\filldraw (1,-2) circle (2pt);
					\filldraw (2,-2) circle (2pt);
					\filldraw (0,0) circle (2pt);\color{black}}
				\end{tikzpicture}
			\end{center}
		\end{minipage}
		\caption{Graded and induced decompositions of $\mathfrak{so}(3,2)$ (left), $\Lambda^2\mathfrak{so}(3,2)$ (center), $\Lambda^3\mathfrak{so}(3,2)$ (right) and limit homogeneous spaces (subspaces over the dashed line).}\label{so32_diags}
	\end{figure}
	
 \end{example}

Apart from a Grassmann algebra, a $G$-gradation on a Lie algebra $\mathfrak{g}$ also induces a decomposition on the space $(\Lambda\mathfrak{g})^{\mathfrak{g}}$ of $\mathfrak{g}$-invariant elements and on the quotient space $\Lambda_R \mathfrak{g} := \bigoplus_{m \in \mathbb{N}} (\Lambda^m \mathfrak{g}) \backslash (\Lambda^m \mathfrak{g})^{\mathfrak{g}}$. The latter space will be studied in further sections, as it plays an important role in the analysis of CYBE solutions.
 
	\begin{proposition}\label{LemgDe} 
 Let $\mathfrak{g}$ admit a $G$-gradation $\mathfrak{g} = \bigoplus_{\alpha \in g} \mathfrak{g}^{(\alpha)}$. Then, the space $(\Lambda^m\mathfrak{g})^\mathfrak{g}$ of $\mathfrak{g}$-invariant elements and $\Lambda_R^m \mathfrak{g} := (\Lambda^m \mathfrak{g}) / (\Lambda^m \mathfrak{g})^{\mathfrak{g}}$ admit decompositions of the form, accordingly, 
 \begin{equation}\label{inv_decomp}
 (\Lambda^m\mathfrak{g})^\mathfrak{g}=\bigoplus_{\alpha\in G}(\Lambda^m\mathfrak{g})^{(\alpha)}\cap (\Lambda^m\mathfrak{g})^\mathfrak{g}
 \end{equation}
 and
 \begin{equation}\label{red_decomp}
\Lambda^{m}_R\mathfrak{g}=\bigoplus_{\alpha\in G} (\Lambda^m\mathfrak{g})^{(\alpha)} / ((\Lambda^m\mathfrak{g})^\mathfrak{g}\cap (\Lambda^m\mathfrak{g})^{(\alpha)}).
\end{equation}
	\end{proposition}
	\begin{proof}  
 By Theorem \ref{thm:root}, every $w\in \Lambda^m\mathfrak{g}$ can be written in a unique way as $w=\sum_{\alpha\in G}w^{(\alpha)}$ for some $w^{(\alpha)}\in (\Lambda^m\mathfrak{g})^{(\alpha)}$. In particular, let us take $w \in (\Lambda^m\mathfrak{g})^{\mathfrak{g}}$ As it is $\mathfrak{g}$-invariant, 
		$0=\sum_{\alpha\in G}[v^{(\beta)},w^{(\alpha)}]_S$ for every $v^{(\beta)}\in\mathfrak{g}^{(\beta)}$.
Consider the elements $\beta, \alpha_1, \alpha_2 \in G$ with $\alpha_1 \neq \alpha_2$. Since $G$ is a group, the condition $\alpha_1 \neq \alpha_2$ implies $\beta + \alpha_1 \neq \beta + \alpha_2$. In consequence, the $m$-vectors $[v^{(\beta)},w^{(\alpha)}]$ belong to different homogeneous spaces of $\Lambda^m\mathfrak{g}$. Hence, they vanish separately and $w^{(\alpha)}\in ((\Lambda^m\mathfrak{g})^{(\alpha)})^{\mathfrak{g}}$. Therefore, $(\Lambda^m\mathfrak{g})^\mathfrak{g}\subset \bigoplus_{\alpha\in G}(\Lambda^m\mathfrak{g})^{(\alpha)}\cap (\Lambda^m\mathfrak{g})^{\mathfrak{g}}$. 
		Since the inclusion $\bigoplus_{\alpha\in G}(\Lambda^m\mathfrak{g})^{(\alpha)}\cap (\Lambda^m\mathfrak{g})^{\mathfrak{g}} \subset (\Lambda^m\mathfrak{g})^\mathfrak{g}$ is immediate, we get the decomposition $(\Lambda^m\mathfrak{g})^\mathfrak{g} = \bigoplus_{\alpha\in G}(\Lambda^m\mathfrak{g})^{(\alpha)}\cap (\Lambda^m\mathfrak{g})^{\mathfrak{g}}$.

Let $p, q \in \Lambda^m \mathfrak{g}$ such that $p - q \in (\Lambda^m \mathfrak{g})^{\mathfrak{g}}$. Since $\Lambda^m \mathfrak{g} = \bigoplus_{\alpha \in G} (\Lambda^m \mathfrak{g})^{(\alpha)}$ by Theorem \ref{thm:root}, we can write $p = \sum_{\alpha \in G} p_{(\alpha)}$ and $q = \sum_{\alpha \in G} q_{(\alpha)}$ for $p_{\alpha}, q_{\alpha} \in (\Lambda^m \mathfrak{g})^{(\alpha)}$. Using the decomposition (\ref{inv_decomp}) of $(\Lambda^m \mathfrak{g})^{\mathfrak{g}}$, we get
$$
p - q = \sum_{\alpha \in G} p_{\alpha} - q_{\alpha} \in (\Lambda^m \mathfrak{g})^{\mathfrak{g}} = \bigoplus_{\alpha\in G}(\Lambda^m\mathfrak{g})^{(\alpha)}\cap (\Lambda^m\mathfrak{g})^{\mathfrak{g}}
$$
As $p_{\alpha} - q_{\alpha} \in (\Lambda^m \mathfrak{g})^{(\alpha)}$ for each $\alpha \in G$, it follows that $p_{\alpha} - q_{\alpha} \in (\Lambda^m\mathfrak{g})^{(\alpha)}\cap (\Lambda^m\mathfrak{g})^{\mathfrak{g}}$. Thus, any equivalence class $[p] \in \Lambda^m_R \mathfrak{g}$ can be written as $[p] = \sum_{\alpha \in G} [p_{\alpha}]$, where $[p_{\alpha}] \in (\Lambda^m\mathfrak{g})^{(\alpha)} / ((\Lambda^m\mathfrak{g})^\mathfrak{g}\cap (\Lambda^m\mathfrak{g})^{(\alpha)})$ and in consequence, we obtain (\ref{red_decomp}).
	\end{proof}
	
	The decomposition on a Grassmann algebra $\Lambda \mathfrak{g}$, detailed in Theorem \ref{thm:root}, allow us to obtain certain solutions to CYBE in the following manner. Given a $G$-gradation of $\mathfrak{g}$, consider an indecomposable element $r \in (\Lambda^2 \mathfrak{g})^{(\alpha)}$ for a certain $\alpha \in G$. Since in that case $[r, r]_S \in (\Lambda^3 \mathfrak{g})^{(\alpha \star \alpha)}$, it follows that $r$ gives an $r$-matrix if $(\Lambda^3 \mathfrak{g})^{(\alpha \star \alpha)} = 0$. The following definition distinguishes the homogeneous spaces of $\Lambda^2 \mathfrak{g}$ that satisfy this property.
	
	\begin{definition}\label{limitspace}
		A {\it limit homogeneous space} of a graded decomposition of $\Lambda^2\mathfrak{g}$ is a homogeneous subspace $(\Lambda^2 \mathfrak{g})^{(\alpha)} \subset \Lambda^2\mathfrak{g}$ such that $(\Lambda^3\mathfrak{g})^{(\alpha\star \alpha)}=0$.
	\end{definition}

 Thus, elements of limit homogeneous spaces are solutions of the CYBE.

\begin{example}
Consider the diagram of $\Lambda^2\mathfrak{so}(2,2)$ in Figure \ref{so22_diags} showing its limit homogeneous subspaces. By direct computation, one can verify that all the elements of these subspaces are solutions of the CYBE. For instance, consider the middle-right point corresponding to the homogeneous space $(\Lambda^2 \mathfrak{so}(2,2))^{(0,1)} = \langle e_0 \wedge e_{+}, f_0 \wedge e_{+}\rangle$. Indeed, 
$$
[e_0 \wedge e_{+}, e_0 \wedge e_{+}]_S = [e_0, e_0] \wedge e_{+} \wedge e_{+} - [e_0, e_{+}] \wedge e_{+} \wedge e_0 - [e_{+}, e_0] \wedge e_0 \wedge e_{+} + [e_{+}, e_{+}] \wedge e_0 \wedge e_0 = 0
$$
and similarly for the second element of that space.
\end{example}
 
Let us consider $r_1\in (\Lambda^2\mathfrak{g})^{(\alpha)}$ and $r_2\in (\Lambda^2\mathfrak{g})^{(\beta)}$. If $(\Lambda^3\mathfrak{g})^{(\alpha \star \beta)}=\{0\}$, then any linear combination of elements $r_1,r_2$ is an $r$-matrix. One can verify this condition easily by examining the decomposition diagram of $\Lambda^2 \mathfrak{g}$.

\begin{example}
Let us consider the homogeneous spaces of $\Lambda^2 \mathfrak{so}(2,2)$ induced by the root $\mathbb{Z}^2$-gradation on $\mathfrak{so}(2,2)$ that we mentioned in previous examples, namely $(\Lambda^2 \mathfrak{so}(2,2))^{(0,1)} = \langle e_0 \wedge e_{+}, f_0 \wedge e_{+}\rangle$ and $(\Lambda^2\mathfrak{so}(2,2))^{(1,1)} = \langle e_{+} \wedge f_{+}\rangle$. Since $(1,1) + (0,1) = (1,2)$ and the element $(1,2) \in \mathbb{Z}^2$ corresponds to an empty homogeneous space of $\Lambda^2 \mathfrak{so}(2,2)$, any element of $(\Lambda^2 \mathfrak{so}(2,2))^{(0,1)} \oplus (\Lambda^2 \mathfrak{so}(2,2))^{(1,1)}$ yields an $r$-matrix.
\end{example}
		
	A short calculation shows that $f_0\wedge e_0$ is an $r$-matrix for $\mathfrak{so}(2,2)$. This is a particular case of the following elementary result.
	
\begin{proposition} 
 If $\mathfrak{g}$ admits a root $\mathbb{Z}^k$-gradation, then $\Lambda^2(\mathfrak{g}^{(0)})$ yields a CYBE solution. 
\end{proposition}
\begin{proof}
Let us write $\mathfrak{g}^{(0)} := \langle v_1, \ldots, v_k\rangle$ for any linearly independent elements $v_1, \ldots, v_k \in \mathfrak{g}^{(0)}$. As $\mathfrak{g}^{(0)}$ is abelian, $[w, w'] = 0$ for any two elements $w, w' \in \mathfrak{g}^{(0)}$. Thus, $[w, v_i  \wedge v_j]_S = [w, v_i] \wedge v_j + v_i \wedge [w, v_j] = 0$ for any $i,j \in \{1, \ldots, k\}$ and $w \in \mathfrak{g}^{(0)}$. Consequently, elements of $\Lambda^2 \mathfrak{g}^{(0)}$ are CYBE solutions.
\end{proof}

	\section{Geometry of \texorpdfstring{$\mathfrak{g}$}{}-invariant elements}\label{Ch:alg_Sec:gInv}
	
As indicated previously, the determination of $\mathfrak{g}$-invariant elements $(\Lambda\mathfrak{g})^{\mathfrak{g}}$ for a Lie algebra $\mathfrak{g}$ is necessary for further analysis of CYBE solutions. In this section, we focus our attention on the spaces $(\Lambda^m \mathfrak{g})^\mathfrak{g}$ of $\mathfrak{g}$-invariant $m$-vectors. It is mainly devoted to study in more detail the structure of ($\Lambda \mathfrak{g})^{\mathfrak{g}}$ for general Lie algebras. However, we present some useful observations in two special cases: Lie algebras with a root gradation and nilpotent Lie algebras. 
 Although these results do not characterise $(\Lambda^m\mathfrak{g})^{\mathfrak{g}}$ completely, they help obtain many of its elements and might simplify noticeably the determination the whole $(\Lambda^m\mathfrak{g})^{\mathfrak{g}}$.

Let us begin with a simple fact concerning the structure of $(\Lambda \mathfrak{g})^{\mathfrak{g}}$.
	
	\begin{proposition}\label{prop:lgg}
		The space $(\Lambda \mathfrak{g})^\mathfrak{g}$ is an $\mathbb{R}$-algebra relative to the exterior product and each space $(\Lambda^m \mathfrak{g})^\mathfrak{g}$ is $\mathfrak{aut}(\mathfrak{g})$-invariant.
	\end{proposition}
	\begin{proof}  
		It is immediate that $(\Lambda \mathfrak{g})^{\mathfrak{g}}$ is a subalgebra of $\Lambda \mathfrak{g}$. By Proposition \ref{prop:glrho_ginv}, the second part of the proposition amounts to the fact that $(\Lambda^m\mathfrak{g})^\mathfrak{g}$ is invariant relative to $\Lambda^mT$ for every  $T\in {\rm Aut}_c(\mathfrak{g})$. Take $w\in (\Lambda^m\mathfrak{g})^{\mathfrak{g}}$ and $T\in {\rm Aut}_c(\mathfrak{g})$. Then for every $v \in \mathfrak{g}$,
		$$
		[v,\Lambda^mTw]_S =[TT^{-1}v,\Lambda^mTw]_S =\Lambda^mT[T^{-1}v,w]_S =0
		$$
		and in consequence, $\Lambda^m Tw \in (\Lambda^m\mathfrak{g})^{\mathfrak{g}}$. Hence, $(\Lambda^m\mathfrak{g})^{\mathfrak{g}}$ is ${\rm Aut}_c(\mathfrak{g})$-invariant.
	\end{proof}
	
The next notion gives further possibilities to study $\mathfrak{g}$-invariant elements.
	
	\begin{definition}
	A {\it traceless ideal} of $\mathfrak{g}$ is a nonzero ideal $\mathfrak{h}\subset \mathfrak{g}$ such that the restriction ${\rm ad}_v|_{\mathfrak{h}}$ to $\mathfrak{h}$ of each ${\rm ad}_v$, for $v\in\mathfrak{g}$, is traceless. If the elements ${\rm ad}(v)$ are traceless for all $v \in \mathfrak{g}$, then $\mathfrak{g}$ is called {\it unimodular}.
	\end{definition}
	
The Lie algebras of abelian, compact, semi-simple, or nilpotent groups are unimodular \cite{ACK99,BR86,Ha15}. In the next proposition, we observe that traceless ideals give rise to some $\mathfrak{g}$-invariant elements.
	
	\begin{proposition}\label{prop:inv_trlsub}
		Every traceless ideal $\mathfrak{h}\subset\mathfrak{g}$ is such that $\Lambda^{\dim\mathfrak{h}}\mathfrak{h}\subset(\Lambda^{\dim\mathfrak{h}}\mathfrak{g})^{\mathfrak{g}}$.
	\end{proposition}
	\begin{proof} 
 The ideal $\mathfrak{h}$ gives a one-dimensional space $\Lambda^{\dim \mathfrak{h}}\mathfrak{h}$. Let $\mathfrak{h} = \langle v_1, \ldots, v_m\rangle$ for some linearly independent $v_1, \ldots v_m \in \mathfrak{g}$. Then, $r\in \Lambda^{\dim \mathfrak{h}}\mathfrak{h}$ is of the form $r = v_1 \wedge \ldots \wedge v_m$. As $\mathfrak{h}$ is an ideal, each $[v, v_j] \in \mathfrak{h}$ for any $v \in \mathfrak{g}$. Let us write $[v, v_j] = \sum_{k = 1}^m a_{jk} v_k$ for certain coefficients $a_{jk} \in \mathbb{R}$. It is immediate to notice that $v_1 \wedge \ldots [v, v_j] \wedge \ldots v_m = a_{jj} v_1 \wedge \ldots \wedge v_m$ and hence, 
 $$
 [v,r]_S = \sum_{j=1}^m v_1 \wedge \ldots \wedge [v, v_j] \wedge \ldots \wedge v_m = {\rm Tr} ({\rm ad}_v|_{\mathfrak{h}})r = 0
 $$
 for every $v\in \mathfrak{g}$. Thus, $r\in (\Lambda^{\dim \mathfrak{h}}\mathfrak{g})^\mathfrak{g}$. 
	\end{proof}
	
	By Proposition \ref{prop:inv_trlsub}, any traceless ideal in $\mathfrak{g}$ induces a particular $\mathfrak{g}$-invariant element. The following theorem gives more details on the relation between traceless ideals and elements of $(\Lambda \mathfrak{g})^\mathfrak{g}$.
 
	\begin{theorem}\label{Theo:InvLk} 
The space spanned by a decomposable $ \Omega\in (\Lambda\mathfrak{g})^\mathfrak{g}$ is in one-to-one correspondence with a traceless ideal $\mathfrak{h}\subset \mathfrak{g}$ such that $\langle\Omega\rangle=\Lambda^{\dim\mathfrak{h}}\mathfrak{h}$.
	\end{theorem} 
	\begin{proof}
For any value of $m \in \mathbb{N}$, a decomposable element $\Omega \in (\Lambda^m\mathfrak{g})^{\mathfrak{g}}$ takes the form $\Omega = v_1 \wedge \ldots \wedge v_m$ for some linearly independent $v_1,\ldots,v_m\in \mathfrak{g}$. This defines a subspace $\mathfrak{h} := \langle v_1, \ldots, v_m\rangle\subset\mathfrak{g}$ that is independent of the chosen $v_1,\ldots,v_m$, and $\langle \Omega\rangle= \Lambda^{\dim\mathfrak{h}}\mathfrak{h}$. Obviously, all elements of $\langle \Omega\rangle$ give rise to the same $\mathfrak{h}$. 

Let us prove that $\mathfrak{h}$ is a traceless ideal. If $\widehat \Omega: \beta \in \mathfrak{g^*} \mapsto \iota_{\beta} \Omega \in \Lambda^{m-1} \mathfrak{g}$, then $\ker \widehat\Omega=\mathfrak{h}^\circ$. For $v\in \mathfrak{g}$ and $\theta\in \mathfrak{h}^\circ$, we obtain
		\begin{multline*}
			\iota_{{\rm ad}_v^*\theta}\Omega
			= \sum_{j=1}^{m} (-1)^j v_1 \wedge \ldots \wedge {\rm ad}_v^*\theta (v_j) \wedge \ldots \wedge v_m = \iota_\theta \left(\sum_{j=1}^m (-1)^j v_1 \wedge \ldots \wedge [v, v_j] \wedge \ldots \wedge v_m \right)=\iota_\theta[v,\Omega]_S =0
		\end{multline*}
		Thus, ${\rm ad}^*_v\theta\in \mathfrak{h}^\circ$ for every $\theta\in \mathfrak{h}^\circ$. Consequently, ${\rm ad}_v\mathfrak{h}\subset\mathfrak{h}$ for every $v\in \mathfrak{g}$, which means that $\mathfrak{h}$ is an ideal of $\mathfrak{g}$. Since $[v,\Omega]_S =({\rm Tr}\, {\rm ad}_v|_\mathfrak{h})\Omega=0$, one gets that ${\rm ad}_v|_{\mathfrak{h}}$ is traceless.
		
		Conversely, let $\mathfrak{h}$ be a non-zero traceless ideal. Then, $ \Lambda^{\dim \mathfrak{h}}\mathfrak{h}$ is one-dimensional and it admits a basis, $\Omega$, given by the exterior product of the elements of a basis of $\mathfrak{h}$.  Since every ${\rm ad}_v$, with $v\in\mathfrak{g}$, acts on $\mathfrak{h}$ tracelessly by assumption, $[v, \Omega]  = (\textrm{Tr ad}_v|_\mathfrak{h})\Omega = 0$ and thus,  $\Omega \in (\Lambda^{\dim \mathfrak{h}} \mathfrak{g})^{\mathfrak{g}}$. Obviously, all elements in  $\langle \Omega\rangle$ belong to $(\Lambda^{\dim\mathfrak{h}}\mathfrak{g})^\mathfrak{g}$. \vskip-0.4cm
	\end{proof}
		
If a Lie algebra $\mathfrak{g}$ admits a root gradation, one can restrict the search of $\mathfrak{g}$-invariant elements to the specific homogeneous space of the Grassmann algebra decomposition.

	\begin{proposition}\label{gin0}
		If $\mathfrak{g}$ admits a root gradation, then  $(\Lambda^m\mathfrak{g})^\mathfrak{g}\subset (\Lambda^m\mathfrak{g})^{(0)}$.
	\end{proposition}
	\begin{proof} 
 Let $w\in (\Lambda^m\mathfrak{g})^{\mathfrak{g}}$. By Proposition \ref{LemgDe}, any $w$ can be written as $w=\sum_{\alpha \in G}w^{(\alpha)}$, where each $w^{(\alpha)}\in (\Lambda^m\mathfrak{g})^{(\alpha)} \cap (\Lambda^m \mathfrak{g})^{\mathfrak{g}}$ for $\alpha\in G$. For every $e\in \mathfrak{g}^{(0)}$, one has that $[e,w^{(\alpha)}]=\widetilde{\alpha}(e)w^{(\alpha)}=0$ and therefore, $\pmb{\alpha}=0$. Since the mapping $\Xi: \alpha \in G \mapsto \widetilde{\alpha} \in (\mathfrak{g}^{(0)})^\star$ of the root gradation is injective, $\alpha=0$. Thus, $w\subset (\Lambda^m\mathfrak{g})^{(0)}$.
	\end{proof}

 In Section \ref{Invgbil}, we described several methods that simplify the determination of $\mathfrak{g}$-invariant forms. When $\mathfrak{g}$ admits a root gradation, we obtain new results, detailed below, that help in further analysis of these forms. 
	
	\begin{theorem}\label{PerpL} 
Let a Lie algebra $\mathfrak{g}$ admit a root $G$-gradation. If $b: \mathfrak{g}^{\otimes 2} \to \mathbb{R}$ is a $\mathfrak{g}$-invariant bilinear symmetric map on $\mathfrak{g}$, then its extension $b_{\Lambda^m\mathfrak{g}}: (\Lambda^m\mathfrak{g})^{\otimes 2} \to \mathbb{R}$ satisfies $b_{\Lambda^m\mathfrak{g}}(v_{(\alpha)}, v_{(\beta)})=0$ for every $v_{(\alpha)} \in \Lambda^m\mathfrak{g}^{(\alpha)}$ and $v_{(\beta)} \in \Lambda^m\mathfrak{g}^{(\beta)}$ with $\alpha+\beta\neq 0 \in G$. 
	\end{theorem}
	\begin{proof} 
 Since $b_{\Lambda^m \mathfrak{g}}$ is $\mathfrak{g}$-invariant by Theorem \ref{extension}, we have
		$
		b_{\Lambda^m\mathfrak{g}}([h,v_{(\alpha)}]_{S}, v_{(\beta)})=-b_{\Lambda^m\mathfrak{g}}(v_{(\alpha)},[h,v_{(\beta)}]_{S})$, for every $h\in \mathfrak{g}^{(0)}$,  $v_{(\alpha)} \in \Lambda^p\mathfrak{g}^{(\alpha)}$, and $ v_{(\beta)} \in \Lambda^p\mathfrak{g}^{(\beta)}.
		$
		Hence,
		$
		(\widetilde{\alpha}+\widetilde{\beta})(h)b_{\Lambda^m\mathfrak{g}}(v_{(\alpha)},v_{(\beta)})=0.
		$
		Since $\alpha+\beta\neq 0$ by assumption, the injectivity of the map $\Xi$ defining the root gradation implies that $(\widetilde{\alpha}+\widetilde{\beta})(h)\neq 0$ for some $h\in \mathfrak{g}^{(0)}$. In consequence, $b_{\Lambda^m\mathfrak{g}}(v_{(\alpha)},v_{(\beta)})=0$.
	\end{proof}
	
	\begin{example} 
 Let us illustrate Theorem \ref{PerpL} for $\Lambda^2\mathfrak{so}(2,2)$. Using the basis $\{e_-,e_0,e_+,f_-,f_0,f_+\}$ of $\mathfrak{so}(2,2)$ introduced in Example \ref{Ex:so22}, we obtain that the only non-zero values for the Killing form $\kappa_{\mathfrak{so}(2,2)}$ in the chosen basis read
		$$
		\kappa_{\mathfrak{so}(2,2)}(e_0, e_0) = \kappa_{\mathfrak{so}(2,2)}(f_0, f_0) = 2,\quad  \kappa_{\mathfrak{so}(2,2)}(e_{-}, e_{+}) = \kappa_{\mathfrak{so}(2,2)}(f_{-}, f_{+}) = 2.
		$$
Consider the homogeneous spaces $(\Lambda^2\mathfrak{so}(2,2))^{(1,1)} = \langle e_{+} \wedge f_{+}\rangle$, $(\Lambda^2 \mathfrak{so}(2,2))^{(0,1)} = \langle e_0 \wedge e_{+}, f_0 \wedge e_{+}\rangle$ and $(\Lambda^2\mathfrak{so}(2,2))^{(-1,-1)} = \langle e_{-} \wedge f_{-}\rangle$. By direct computation one can verify that, for instance, 
\begin{equation*}
\begin{split}
\kappa_{\Lambda^2 \mathfrak{so}(2,2)}(e_{+} \wedge f_{+}, e_0 \wedge e_{+}) &= \kappa_{\mathfrak{so}(2,2)}(e_{+}, e_0) \kappa_{\mathfrak{so}(2,2)}(f_{+}, e_{+}) - \kappa_{\mathfrak{so}(2,2)}(e_{+}, e_{+}) \kappa_{\mathfrak{so}(2,2)}(f_{+}, e_0) \\
&- \kappa_{\mathfrak{so}(2,2)}(f_{+}, e_0) \kappa_{\mathfrak{so}(2,2)}(e_{+}, e_{+}) + \kappa_{\mathfrak{so}(2,2)}(f_{+}, e_{+}) \kappa_{\mathfrak{so}(2,2)}(e_{+}, e_0) = 0
\end{split}
\end{equation*}
and
\begin{equation*}
\begin{split}
\kappa_{\Lambda^2 \mathfrak{so}(2,2)}(e_{+} \wedge f_{+}, e_{-} \wedge f_{-}) &= \frac{1}{2} [\kappa_{\mathfrak{so}(2,2)}(e_{+}, e_{-}) \kappa_{\mathfrak{so}(2,2)}(f_{+}, f_{-}) - \kappa_{\mathfrak{so}(2,2)}(e_{+}, f_{-}) \kappa_{\mathfrak{so}(2,2)}(f_{+}, e_{-}) \\
&- \kappa_{\mathfrak{so}(2,2)}(f_{+}, e_{-}) \kappa_{\mathfrak{so}(2,2)}(e_{+}, f_{-}) + \kappa_{\mathfrak{so}(2,2)}(f_{+}, f_{-}) \kappa_{\mathfrak{so}(2,2)}(e_{+}, e_{-})] = 4.
\end{split}
\end{equation*}
	\end{example}
	
	\begin{corollary}\label{cor:nondeg_subsp} 
Let $\mathfrak{g}$ admit a root gradation. If $b$ is a non-degenerate bilinear symmetric form on $\mathfrak{g}$, the restrictions of $b_{\Lambda^m\mathfrak{g}}$ to $(\Lambda^m\mathfrak{g})^{(0)}$ and $\Lambda^m \mathfrak{g}^{(\alpha)} \oplus \Lambda^m \mathfrak{g}^{(-\alpha)}$ are non-degenerate. 
	\end{corollary}
	\begin{proof} 
 We prove both results by contradiction. If the restriction of $b_{\Lambda^m\mathfrak{g}}$ to $(\Lambda^m\mathfrak{g})^{(0)}$ is degenerate, there exists a $w\in (\Lambda^m\mathfrak{g})^{(0)}$  perpendicular, relative to $b_{\Lambda^m\mathfrak{g}}$, to every element of $(\Lambda^m\mathfrak{g})^{(0)}$. Additionally, Theorem \ref{PerpL} yields that $w$ is perpendicular to $(\Lambda^m\mathfrak{g})^{(\alpha)}$ for $\alpha \neq 0 \in G$. Thus, $w$ is perpendicular to the whole $\Lambda^m\mathfrak{g}$ and in consequence, $b_{\Lambda^m\mathfrak{g}}$ is degenerate. However, Theorem \ref{Cor:Diatwoform} asserts that the induced form $b_{\Lambda^m\mathfrak{g}}$ is non-degenerate, hence the contradiction.  Therefore, $b_{\Lambda^m\mathfrak{g}}$ is non-degenerate on $(\Lambda^m\mathfrak{g})^{(0)}$.
		
		Similarly, if $w \in \Lambda^m\mathfrak{g}^{(\alpha)}\oplus \Lambda^m\mathfrak{g}^{(-\alpha)}$ is orthogonal to $\Lambda^m\mathfrak{g}^{(\alpha)}\oplus \Lambda^m\mathfrak{g}^{(-\alpha)}$, then it stems from Theorem \ref{PerpL} that $w$ is orthogonal to $\Lambda^m\mathfrak{g}$, which again contradicts the non-degeneracy of the induced $b_{\Lambda^m\mathfrak{g}}$. 
	\end{proof}
	
	\begin{example} 
Let $\{e_{-} \wedge e_0 \wedge e_{+},\, e_{-} \wedge e_{+} \wedge f_0,\, e_0 \wedge f_{-} \wedge f_{+},\, f_{-} \wedge f_0 \wedge f_{+}\}$ denote the basis of $(\Lambda^3\mathfrak{so}(2,2))^{(0)}$. A short calculation shows that the matrix form of $\kappa_{\Lambda^3\mathfrak{so}(2,2)}$ restricted to $(\Lambda^3\mathfrak{so}(2,2))^{(0)}$ in the given basis reads
  $$
[b_{\Lambda^3\mathfrak{so}(2,2)}]:=\left(\begin{array}{cccc}
-8&0&0&0\\
0&-8&0&0\\
0&0&-8&0\\
0&0&0&-8
\end{array}\right)
		$$
Thus, it is non-degenerate as claimed in Corollary \ref{cor:nondeg_subsp}.
	\end{example}
	
	Finally, let us discuss invariant elements for nilpotent Lie algebras. Recall that every nilpotent Lie algebra $\mathfrak{g}$ possesses a flag of ideals, called the \textit{lower central series} of $\mathfrak{g}$, defined recurrently as $\mathfrak{g}_{s)} :=[\mathfrak{g},\mathfrak{g}_{s-1)}]$ for $s\in \mathbb{N}$ with $\mathfrak{g}_{0)}:=\mathfrak{g}$. Then, $\mathfrak{g}\supset \mathfrak{g}_{1)} \supset \ldots\supset \mathfrak{g}_{p-1)}\supset\mathfrak{g}_{p)}=\{0\}$ for a certain $p$ (see \cite{FH91} for details). Let us use this fact to study $(\Lambda \mathfrak{g})^{\mathfrak{g}}$. Similarly as in Theorem \ref{Theo:InvLk}, certain decomposable elements of $(\Lambda^m\mathfrak{g})^{\mathfrak{g}}$ can be characterised by particular Lie subalgebras of $\mathfrak{g}$, as shown in the following proposition.
 
	\begin{proposition} 
Let $\mathfrak{g}$ be a nilpotent Lie algebra. Then, the space spanned by any decomposable element $\Omega \in (\Lambda\mathfrak{g})^\mathfrak{g}$ is in one-to-one correspondece with a non-zero ideal $\mathfrak{h}$ of $\mathfrak{g}$ such that $\Lambda^{{\rm dim }\mathfrak{h}}\mathfrak{h} = \langle \Omega\rangle$.
	\end{proposition}
 \begin{proof}
Since nilpotent Lie algebras are unimodular \cite{Ha15}, this result follows from Theorem \ref{Theo:InvLk}.
 \end{proof}

As shown by the next proposition, the knowledge of the lower central series of $\mathfrak{g}$ allows to obtain some invariant elements quite immediately. Recall that $\mathfrak{z}(\mathfrak{g})$ denotes the center of Lie algebra $\mathfrak{g}$, i.e. a Lie subalgebra $\mathfrak{z}(\mathfrak{g}) \subset \mathfrak{g}$ such that $[v, w] = 0$ for any $v \in \mathfrak{g}$ and $w \in \mathfrak{z}(\mathfrak{g})$.

	\begin{proposition}\label{lambda2nil} 
 Let $\mathfrak{g}$ be a nilpotent Lie algebra with $\mathfrak{g}_{p)}=0$ and $\mathfrak{g}_{p-1)}\neq 0$. Then:
 \begin{itemize}
 \item If $\dim\mathfrak{z}(\mathfrak{g})=1$, then $\mathfrak{z}(\mathfrak{g})\wedge \mathfrak{g}_{p-2)} \subset (\Lambda^2 \mathfrak{g})^{\mathfrak{g}}$
\item If $\dim\mathfrak{z}(\mathfrak{g})=2$, then $\Lambda^2\mathfrak{z}(\mathfrak{g})\wedge \mathfrak{g}_{p-2)} \subset (\Lambda^3 \mathfrak{g})^{\mathfrak{g}}$; 
\item If $\dim\mathfrak{z}(\mathfrak{g})=1$ and $\dim\mathfrak{g}_{p-2)}>1$, then $\mathfrak{z}(\mathfrak{g})\wedge\Lambda^2\mathfrak{g}_{p-2)} \subset (\Lambda^3 \mathfrak{g})^{\mathfrak{g}}$;
\item If $\dim\mathfrak{z}(\mathfrak{g})=1$ and $\dim\mathfrak{g}_{p-2)}=1$, then $\mathfrak{z}(\mathfrak{g})\wedge\mathfrak{g}_{p-2)}\wedge \mathfrak{g}_{p-3)} \subset (\Lambda^3 \mathfrak{g})^{\mathfrak{g}}$.
\end{itemize}
	\end{proposition}
 \begin{proof}
 Consider a decomposable element $r = a \wedge b \in \mathfrak{z}(\mathfrak{g})\wedge \mathfrak{g}_{p-2)}$ and assume $\dim\mathfrak{z}(\mathfrak{g})=1$. Then for any $v \in \mathfrak{g}$, $[v, a \wedge b]_S = [v, a] \wedge b + a \wedge [v, b]$. Since $a \in \mathfrak{z}(\mathfrak{g})$, it follows that $[v, a] = 0$. Moreover, $[v, b] \in \mathfrak{z}(\mathfrak{g})$ and as $\dim\mathfrak{z}(\mathfrak{g})=1$, we obtain $a \wedge [v, b] = 0$. Thus, $[v, r]_S =0$ and $r \in (\Lambda^2 \mathfrak{g})^{\mathfrak{g}}$.

Let $\dim\mathfrak{z}(\mathfrak{g})=2$ and consider an element $r = a \wedge b \wedge c \in \Lambda^2\mathfrak{z}(\mathfrak{g})\wedge \mathfrak{g}_{p-2)}$, where $a, b \in \mathfrak{z}(\mathfrak{g})$ are linearly independent and $c \in \mathfrak{g}_{p-2)}$. Then for any $v \in \mathfrak{g}$, we get $[v, r]_S = [v, a] \wedge b \wedge c + a \wedge [v, b] \wedge c + a \wedge b \wedge [v, c]$. Since $a \in \mathfrak{z}(\mathfrak{g})$, we have $[v, a] = 0$ and $[v, b] = 0$. By the properties of the lower central series, $[v, c] \in \mathfrak{z}(\mathfrak{g})$. As $\dim\mathfrak{z}(\mathfrak{g})=2$, it follows that $[v, c] = \alpha_1 a + \alpha_2 b$ for certain $\alpha_1, \alpha_2 \in \mathbb{R}$. Thus, $a \wedge b \wedge [v, c] = 0$ and $r \in (\Lambda^3 \mathfrak{g})^{\mathfrak{g}}$.

Let $\dim\mathfrak{z}(\mathfrak{g})=1$ and consider an element $r = a \wedge b \wedge c \in \Lambda^2\mathfrak{g}_{p-2)}\wedge \mathfrak{z}(\mathfrak{g})$, where $a, b \in \mathfrak{g}_{p-2)}$ are linearly independent and $c \in \mathfrak{z}(\mathfrak{g})$. Then for any $v \in \mathfrak{g}$, we get $[v, r]_S = [v, a] \wedge b \wedge c + a \wedge [v, b] \wedge c + a \wedge b \wedge [v, c]$. Since $[v, a] \in \mathfrak{z}(\mathfrak{g})$ and $[v, b] \in \mathfrak{z}(\mathfrak{g})$, we get $[v, a] \wedge b \wedge c = 0$ and $a \wedge [v, b] \wedge c = 0$. Obviously, $[v, c] = 0$ as $c \in \mathfrak{z}(\mathfrak{g})$. Thus, $[v, r]_S = 0$ and $r \in (\Lambda^3 \mathfrak{g})^{\mathfrak{g}}$.

Let $\dim\mathfrak{z}(\mathfrak{g})=1$ and $\dim \mathfrak{g}_{p-2)} = 1$. Consider an element $r = a \wedge b \wedge c \in \mathfrak{z}(\mathfrak{g}) \wedge \mathfrak{g}_{p-2)}\wedge \mathfrak{g}_{p-3)}$. Then for any $v \in \mathfrak{g}$, we get $[v, r]_S = [v, a] \wedge b \wedge c + a \wedge [v, b] \wedge c + a \wedge b \wedge [v, c]$. Obviously, $[v, a] = 0$ as $a \in \mathfrak{z}(\mathfrak{g})$. Since $[v, b] \in \mathfrak{z}(\mathfrak{g})$ and $\dim\mathfrak{z}(\mathfrak{g})=1$, it follows that $a \wedge [v, b] \wedge c = 0$. Similarly, since $[v, c] \in \mathfrak{g}_{p-2)}$ and $\dim \mathfrak{g}_{p-2)} = 1$, we get $a \wedge b \wedge [v, c] = 0$. Thus, $[v, r]_S = 0$ and $r \in (\Lambda^3 \mathfrak{g})^{\mathfrak{g}}$.
 \end{proof}
	
	\section{Reduced mCYBE}
 
For low-dimensional Lie algebras $\mathfrak{g}$, the analysis of mCYBE can be performed by hand, as the dimension of spaces $\Lambda^2 \mathfrak{g}$ and $(\Lambda^3 \mathfrak{g})^{\mathfrak{g}}$ is small. However,	if $\dim \mathfrak{g}\geq 4$, mCYBE is frequently very complicated to solve. In this section, we show that it is sometimes possible to simplify the study of mCYBE by considering its equivalent version on the quotient space $\Lambda^m_R\mathfrak{g}:=\Lambda^m\mathfrak{g}/(\Lambda^m\mathfrak{g})^\mathfrak{g}$.
	
	\begin{definition}
		The elements of $\Lambda^m_R\mathfrak{g}$, for $m \in \mathbb{Z}$, are called {\it reduced $m$-vectors}. An equivalence class in $\Lambda^m_R\mathfrak{g}$ associated with an element $w \in \Lambda^m_R\mathfrak{g}$ will be denoted by $[w]$. If $r$ is an $r$-matrix, then the equivalence class of $r$ in $\Lambda^2_R\mathfrak{g}$ is called a {\it reduced $r$-matrix}.
	\end{definition}

 Many structures discussed in previous sections can be defined on the reduced space $\Lambda^m_R\mathfrak{g}$. In particular, the Schouten braket will be of special interest to us.
	
	\begin{proposition}\label{prop:reduced}
		Let $\pi_p:w_p\in \Lambda^p\mathfrak{g}\mapsto [w_p]\in\Lambda^p_R\mathfrak{g}$, with $p \in \mathbb{Z}$. The algebraic Schouten bracket induces a new bracket, called the reduced Schouten bracket, on $
		\Lambda_R\mathfrak{g}:=\bigoplus_{p\in \mathbb{Z}}\Lambda^p_R\mathfrak{g},$ 
		of the form
		\begin{equation}\label{RSN}
			{[[w_p],[w_q]]_R}:=[[w_p,w_q]_S],\qquad \forall w_p\in \Lambda^p\mathfrak{g},\quad \forall w_q\in\Lambda^q\mathfrak{g},
		\end{equation}
such that $(\Lambda_R \mathfrak{g}, [\cdot, \cdot]_R)$ admits a Gerstenhaber algebra structure. In particular, the decomposition $\Lambda_R\mathfrak{g}=\bigoplus_{p\in \mathbb{Z}}\Lambda^p_R\mathfrak{g}$ is compatible with the reduced Schouten bracket, that is $[\Lambda^p_R\mathfrak{g},\Lambda^q_R\mathfrak{g}]_R\subset \Lambda^{p+q-1}_R\mathfrak{g}$. Moreover, the projection $\pi := \bigoplus_{p\in \mathbb{Z}}\pi_p$ satisfies that $[\pi(a), \pi(b)]_R = \pi([a,b] )$ for arbitrary $a,b \in \Lambda\mathfrak{g}$. Additionally, $r\in \Lambda^2\mathfrak{g}$ is an $r$-matrix if and only if $[\pi(r),\pi(r)]_{R}=0$.
	\end{proposition}
	\begin{proof} 
		Let us show that (\ref{RSN}) is well defined. If $[w_p] = [\bar{w}_p]$ and $[w_q] = [\bar{w}_q]$ for $w_p, \bar{w}_p\in \Lambda^p\mathfrak{g}$ and $w_q, \bar{w}_q\in \Lambda^q\mathfrak{g}$, then $w_p - \bar{w}_p,w_q-\bar w_q \in (\Lambda\mathfrak{g})^{\mathfrak{g}}$. Since $[(\Lambda\mathfrak{g})^\mathfrak{g},\Lambda \mathfrak{g}]=0$, we obtain
		$
		[[w_p],[w_q]]_R :=[[w_p,w_q]]=[[w_p-\bar{w}_p+\bar{w}_p,w_q-\bar w_q+\bar w_q]]=[[\bar w_p],[\bar w_q]]_R.
		$
		
Let us verify that the decomposition on $\Lambda_R\mathfrak{g}$ is compatible with the reduced Schouten  bracket. Take any $[v] \in \Lambda^p_R\mathfrak{g}$ and $[w] \in \Lambda^q_R\mathfrak{g}$. Since $[[v], [w]]_R = [[v, w]_S]$ and $[v, w]_S \in \Lambda^{p+q-1} \mathfrak{g}$, we have $[[v], [w]]_R \in \Lambda_R^{p+q-1} \mathfrak{g}$.
In a similar manner, one can show that for all $u \in \Lambda^p \mathfrak{g}, v \in \Lambda^q \mathfrak{g}, w \in \Lambda^s \mathfrak{g}$, we have
$$
[[u], [v]]_R= -(-1)^{(p-1)(q-1)} [[v], [u]]_R
$$
and
$$
(-1)^{(p-1)(s-1)} [[u], [[v], [w]]_{R}] _R+ (-1)^{(q-1)(p-1)} [[v], [[w], [u]]_{R}]_R + (-1)^{(s-1)(q-1)} [[w], [[u], [v]]_{R}]_R = 0.
$$
In other words, $(\Lambda_R \mathfrak{g}, [\cdot, \cdot]_R)$ admits a Gerstenhaber algebra structure.		

For any $a, b \in \Lambda \mathfrak{g}$, the definition of $\pi$ and the reduced bracket give immediately $[\pi(a), \pi(b)]_R = [[a], [b]]_R = [[a, b]_S] = \pi([a, b]_S)$.
  
If $r$ is an $r$-matrix, then $[r, r]_S \in (\Lambda^3 \mathfrak{g})^{\mathfrak{g}}$. Thus, $[\pi(r),\pi(r)]_R=\pi([r,r] )\in \pi((\Lambda^3\mathfrak{g})^{\mathfrak{g}})=0$. Conversely, since $[0] = [\pi(r),\pi(r)]_R=\pi([r,r])$, it follows that $[r, r]_S \in (\Lambda^3 \mathfrak{g})^{\mathfrak{g}}$.
	\end{proof}
	
Next lemma proves that the reduced space $\Lambda_R^m \mathfrak{g}$ is a $\mathfrak{g}$-module. In consequence of that fact, it is possible to define the Chevalley-Eilenberg cohomology of $\mathfrak{g}$ with values in $\Lambda_R^m \mathfrak{g}$. Therefore, the cohomological constructions of Lie bialgebra and coboundary Lie bilagebra hold for the reduced space $\Lambda_R^m \mathfrak{g}$ as well.
	
	\begin{lemma}\label{reduced_action}
		The pair  $(\Lambda^m_R\mathfrak{g},\, \sigma: v\in \mathfrak{g} \mapsto \sigma_v:=[[v], \cdot]_R \in \mathfrak{gl}(\Lambda^m_R \mathfrak{g}))$ is a $\mathfrak{g}$-module and $\Psi:T\in {\rm Aut}(\mathfrak{g})\mapsto [\Lambda^mT]\in { GL}(\Lambda^m_R\mathfrak{g})$, with $[\Lambda^mT]([w]):=[\Lambda^mT(w)]$ for every $w\in \Lambda^m\mathfrak{g}$, is a Lie group action.
	\end{lemma}
	\begin{proof}
		Let us show that $(\Lambda^m_R\mathfrak{g},\sigma)$ is a $\mathfrak{g}$-module. Since the reduced bracket is well defined, $\sigma$ is well defined. As $(\Lambda^m_R\mathfrak{g}, [\cdot, \cdot]_R)$ admits a Gerstenhaber algebra structure, the reduced Schouten bracket satisfies the graded Jacobi identity, namely
  $$
 (-1)^{(p-1)(s-1)} [[u], [[v], [w]]_{R}] _R+ (-1)^{(q-1)(p-1)} [[v], [[w], [u]]_{R}]_R + (-1)^{(s-1)(q-1)} [[w], [[u], [v]]_{R}]_R = 0.
  $$
for any $u \in \Lambda^p \mathfrak{g}, v \in \Lambda^q \mathfrak{g}, w \in \Lambda^s \mathfrak{g}$. If $p = q = 1$, the above identity can be written as 
  $$
  [[[u],[v]]_R, [w]]_R = [[u],[[v],[w]]_R]_R-[[v],[[u],[w]]_R]_R
  $$
Using the map $\sigma$, it reads
		$
		\sigma_{w_1}\sigma_{w_1'}([w_s])-\sigma_{w_1'}\sigma_{w_1}([w_s])=\sigma_{[w_1,w_1']}([w_s]).
		$
  Thus, $\sigma$ is a Lie algebra homomorphism.
		
		Let us show that $T\in{\rm Aut}(\mathfrak{g})$ acts on $\Lambda^m_R\mathfrak{g}$ via $\Psi$. First, let us verify that $\Psi(T)$ is well defined. This amounts to proving that $[\Lambda^mT]([w])=[\Lambda^mT]([w'])$ for all $w,w'\in [w]$. We have $\Lambda^mT(w)=\Lambda^mT(w-w'+w')=\Lambda^mT(w-w')+\Lambda^mT(w')$. Since $w - w' \in (\Lambda^2 \mathfrak{g})^{\mathfrak{g}}$ and $\Lambda^mT(\Lambda^m\mathfrak{g})^\mathfrak{g}\subset (\Lambda^m\mathfrak{g})^\mathfrak{g}$ due to Proposition \ref{prop:lgg}, we get $[\Lambda^mT(w)]=[\Lambda^mT(w')]$. Thus, $[\Lambda^mT]([w])=[\Lambda^mT(w)]=[\Lambda^mT(w')]=[\Lambda^mT]([w'])$.

  Obviously, the identity ${\rm id} \in {\rm Aut}(\mathfrak{g})$ gives $\Lambda^m {\rm id} = {\rm id}_{\Lambda^m \mathfrak{g}}$. Thus, $\Psi({\rm id}) = {\rm id}_{\Lambda_R^m \mathfrak{g}}$. Next, let $T_1, T_2 \in {\rm Aut}(\mathfrak{g})$. We want to show that $\Psi(T_1 \circ T_2) = \Psi(T_1) \circ \Psi(T_2)$. By the definition of $\Lambda^m T_i$ (see proposition \ref{Prop:grass_gmod}), it is immediate to verify that $\Lambda^m (T_1 \circ T_2) = \Lambda^m T_1 \circ \Lambda^m T_2$. Then,
  \begin{equation*}
  \begin{split}
  \Psi(T_1 \circ T_2)([w]) &= [\Lambda^m (T_1 \circ T_2)]([w]) = [\Lambda^m (T_1 \circ T_2)(w)] = [\Lambda^m T_1(\Lambda^m T_2 (w))] = [\Lambda^m T_1]([\Lambda^m T_2 (w)]) \\
  &= ([\Lambda^m T_1] \circ [\Lambda^m T_2])([w])
  \end{split}
  \end{equation*}
  This hshows that $\Psi$ is indeed a Lie group action.
	\end{proof}
	
Under certain conditions, one can obtain $\mathfrak{g}$-invariant maps on $\Lambda^m_R\mathfrak{g}$ out of $\mathfrak{g}$-invariant maps on $\Lambda^m\mathfrak{g}$, as the following proposition shows.
	
	\begin{proposition}\label{RedgInv} If $b:(\Lambda^m\mathfrak{g})^k\rightarrow \mathbb{R}$ is a $\mathfrak{g}$-invariant $k$-linear map and its kernel contains $(\Lambda^m\mathfrak{g})^\mathfrak{g}$, then there is a $\mathfrak{g}$-invariant $k$-linear map $b_R$ on $\Lambda_R^m\mathfrak{g}$ given by
	$b_R([w_1],\ldots,[w_k]):=b(w_1,\ldots,w_k)$, for all $ w_1,\ldots,w_k\in \Lambda^m\mathfrak{g}$.
	\end{proposition}
	\begin{proof} 
 First, let us verify that the map $b_R$ is well defined. If $w'_i\in [w_i]$ for $i=1 ,\ldots,k$, then $w_i-w'_i\in (\Lambda^m\mathfrak{g})^{\mathfrak{g}}$.  Since the kernel of $b$ contains $(\Lambda^m\mathfrak{g})^{\mathfrak{g}}$, one has 
\begin{equation*}
\begin{split}
b_R([w_1],\ldots,[w_k]) &= b(w_1,\ldots,w_k) = b((w_1-w_1')+w_1', \ldots, (w_k-w_k')+w_k') = b(w_1',\ldots,w'_k) \\
&= b_R([w_1'],\ldots,[w_k'])
\end{split}
\end{equation*}
		and  $b_R$ does not depend on the representative of each particular equivalence class of $\Lambda^m_R\mathfrak{g}$. 
		The $\mathfrak{g}$-invariance of $b_R$ stems immediately from the $\mathfrak{g}$-invariance of $b$. 
	\end{proof}

	\section{On the automorphisms of a Lie algebra}\label{R:3DSec:Aut}
 
Further studies of equivalent $r$-matrices and Lie bialgebra structures require the knowledge of the group ${\rm Aut}(\mathfrak{g})$ of automorphisms of a Lie algebra $\mathfrak{g}$. In this section, we briefly investigate $\textrm{Aut}(\mathfrak{g})$ and comment on relevant features of ${\rm Aut}(\mathfrak{g})$-invariant metrics on $\Lambda\mathfrak{g}$. These remarks will be useful in the remaining part of this chapter, which focuses on the classification of $r$-matrices for three-dimensional Lie algebras.

Let us begin our analysis with the study of the Lie algebra $\mathfrak{aut}(\mathfrak{g})$ of the group ${\rm Aut}(\mathfrak{g})$. For that purpose, recall a relevant property of Lie algebra automorphisms \cite{SW73}.
		
	\begin{proposition}\label{aut} 
	If $\mathfrak{der}(\mathfrak{g})$ is the Lie algebra of derivations of $\mathfrak{g}$, then $\mathfrak{der}(\mathfrak{g}) \simeq \mathfrak{aut}(\mathfrak{g})$.
	\end{proposition}
	
 Derivations of $\mathfrak{g}$ are elements $D\in \mathfrak{gl}(\mathfrak{g})$ such that $D([v_1,v_2])\!=\![D(v_1),v_2]+[v_1,D(v_2)]$ for any $v_1,v_2\in \mathfrak{g}$ (see \cite{FH91}). This condition can be solved via computer programs even for relatively high-dimensional Lie algebras and thus, one can easily determine $\mathfrak{aut}(\mathfrak{g})$. In consequence, it is also possible to determine ${\rm Aut}_c (\mathfrak{g})$, the connected part of the neutral element of ${\rm Aut}(\mathfrak{g})$, by using the exponential map. 

  However, the determination of the other connected parts of ${\rm Aut}(\mathfrak{g})$ might not be straightforward, especially when $\mathfrak{g}$ is not complex and semi-simple  \cite{Ha15,Jacobson}. Let us consider the real Lie algebra $\mathfrak{sl}_2$ to illustrate the possible difficulties. The Killing form $\kappa_{\mathfrak{sl}_2}$, given by (\ref{sl2A}) in the basis $\{e_1,e_2,e_3\}$ indicated in Table \ref{tabela3w}, is indefinite with signature $(2,1)$. The $\kappa_{\mathfrak{sl}_2}$ induces a quadratic function  $f(xe_1+ye_2+ze_3)=2x^2+4yz$ on $\mathfrak{sl}_2$. Its level set $S_k$ consists of points $(x, y, z)$, where $2x^2+4yz=k$. If $k<0$, then $S_k$ is a two-sheeted hyperboloid contained in the region of $\mathfrak{sl}_2$ with $z>0$ or in the region of $\mathfrak{sl}_2$ with $z<0$.
	The elements of ${\rm Aut}(\mathfrak{sl}_2)$ are isometries of $\kappa_{\mathfrak{sl}_2}$ \cite{FH91}. Thus, the elements of ${\rm Aut}_c(\mathfrak{sl}_2)$ leave invariant each component of each two-sheeted hyperboloid. However, notice that $T \in {\rm Aut}(\mathfrak{sl}_2)$ given by $T(e_1)=-e_1$, $T(e_2)=-e_3$, and $T(e_3)=-e_2$ does not preserve the sign of  $z$. Consequently, it swaps connected parts of each two-sheeted hyperbolid and $T\notin  {\rm Aut}_c(\mathfrak{sl}_2)$. Hence, ${\rm Aut}(\mathfrak{sl}_2)$ is not connected. 

 The above example shows that the analysis of ${\rm Aut}(\mathfrak{g})$ is a complex task. In consequence, its effectiveness in the classification problem of Lie bialgebras is relatively restricted (e.g. in \cite{FJ15}, the real $\mathfrak{sl}_2$ case is studied only up to inner automorphisms). On the other hand, the ${\rm Aut}(\mathfrak{g})$-invariant bilinear symmetric maps are easier to obtain. For instance, Proposition \ref{prop:glrho_ginv} assures such maps must be $\mathfrak{der}(\mathfrak{g})$-invariant. This in turn is equivalent to the ${\rm Aut}_c (\mathfrak{g})$-invariance of these maps. In practical applications, the Killing forms and Killing-type forms are mostly employed. These forms are known to be invariant relative to the whole group ${\rm Aut}(\mathfrak{g})$. Consequently, the induced forms on $\Lambda \mathfrak{g}$ inherit the invariance property. This fact follows from the next result, which is a straightforward generalisation of Proposition \ref{ExtKil}.
	
	\begin{theorem}\label{ExtInvariance} 
 If $b$ is a $k$-linear map on $\mathfrak{g}$ invariant under ${\rm Aut}(\mathfrak{g})$, $b_{\Lambda^m\mathfrak{g}}$ is invariant under  ${\rm Aut}(\mathfrak{g})$.  
	\end{theorem}
	
	The ${\rm Aut}(\mathfrak{g})$-invariance of the induced $b_{\Lambda^m \mathfrak{g}}$ plays a crucial role in the study of equivalent $r$-matrices. Consider the level sets, $S_k := \{v \in \Lambda^m \mathfrak{g}: b_{\Lambda^m\mathfrak{g}}(v,\ldots, v) = k \in \mathbb{R} \}$. Since $b_{\Lambda^m \mathfrak{g}}$ is ${\rm Aut}(\mathfrak{g})$-invariant, $b_{\Lambda^m\mathfrak{g}}(\Lambda^m T v,\ldots, \Lambda^m T v) = b_{\Lambda^m\mathfrak{g}}(v,\ldots, v)$ for any $T \in {\rm Aut}(\mathfrak{g})$ and $v \in S_k$. Thus, a given orbit of ${\rm Aut}(\mathfrak{g})$ on $\Lambda^m\mathfrak{g}$ is contained in a single level set $S_k$. This observation greatly simplifies our analysis of equivalent elements of $\lambda^m \mathfrak{g}$. Although these orbits do not need to be connected, it is ensured that only the orbits within the same level set $S_k$ could contain equivalent elements. In consequence, we can restrict our analysis of ${\rm Aut}(\mathfrak{g})$-action to each $S_k$. 
 
 Given a level set $S_k$, we first investigate the action of the group ${\rm Inn}(\mathfrak{g})$ of inner automorphisms of $\mathfrak{g}$. Inner automorphisms consist a subgroup of ${\rm Aut}(\mathfrak{g})$ that is relatively easy to compute. Moreover, it informs about the connected components of ${\rm Aut}(\mathfrak{g})$. Then, it remains to find such automorphism $T \in {\rm Aut}(\mathfrak{g})$ that identifies different orbits of ${\rm Inn}(\mathfrak{g})$ within the same $S_k$. Since each $S_k$ usually contains a small number of such orbits, this task can be achieved by hand. 
	
In order to perform the last step of the procedure described above, let us now provide a few hints to characterise automorphisms of Lie algebras. 
If $\mathfrak{g}$ is a complex simple or semi-simple Lie algebra, then the quotient group ${\rm Aut}(\mathfrak{g}) / {\rm Inn}(\mathfrak{g})$ is isomorphic to the automorphisms of Dynkin diagrams of $\mathfrak{g}$ \cite{Ha15,Jacobson}. Meanwhile, automorphisms of general Lie algebras cannot be determined so easily. Let us end this section with few remarks on automorphisms of solvable and nilpotent Lie algebras.
	
Recall that the \textit{derived series} of a Lie algebra $\mathfrak{g}$ are the sequence of ideals defined recurrently by
	$
	\mathfrak{g}^{p)}:=[\mathfrak{g}^{p-1)},\mathfrak{g}^{p-1)}]$ for all $p\in \mathbb{N}$  and $ \mathfrak{g}^{0)}=\mathfrak{g}$ (see  \cite{Ha15}). It follows by induction that
$T\mathfrak{g}^{p)}=\mathfrak{g}^{p)}$  for every $p\in \mathbb{N} \cup \{0\}$ and $T\in {\rm Aut}(\mathfrak{g})$. A similar result applies to  lower central series.  If $\mathfrak{g}$ is solvable, the elements of ${\rm Aut}(\mathfrak{g})$ leave invariant the elementary sequence
	$
	\mathfrak{s}_{pq}:=\mathfrak{g}_{p)}\wedge \mathfrak{g}_{q)}$, $ p\leq q,p,q\in \mathbb{N}\cup \{0\}.
	$
	If $\mathfrak{s}_{pq}\neq 0$, then $\mathfrak{s}_{pq} \supset \mathfrak{s}_{lm}$ if and only if $p\leq l$ and $q\leq m$. Similar results apply to 
	$
	\mathfrak{s}^{pq}:=\mathfrak{g}^{p)}\wedge \mathfrak{g}^{q)},  p\leq q,
	$ 
	for $p,q\in \mathbb{N}\cup \{0\}$. 
	
	\section{Study of real three-dimensional coboundary Lie bialgebras}\label{Ch:alg_Sec:3DClass}
	
	In this section, we use previously devised methods to analyse and to classify, up to Lie algebra automorphisms, coboundary Lie bialgebra structures on real three-dimensional Lie algebras. In each case, the following general scheme is used. First, we employ gradations to obtain $\mathfrak{g}$-invariant elements. If possible, we perform further analysis on the reduced space $\Lambda^2_R \mathfrak{g}$. After computing the $\mathfrak{g}$-invariant symmetric forms $b_{\Lambda^2 \mathfrak{g}}$, the orbits of ${\rm Inn}(\mathfrak{g})$-action on level sets $S_k$ of $b_{\Lambda^2 \mathfrak{g}}$ are obtained. Finally, we find automorphisms that connect separate connected orbits within each $S_k$. In the last step, we avoid using the full automorphism group in the classification of Lie bialgebras  (cf. \cite{FJ15}). Our results are gathered in Table \ref{tabela3w}.

For each Lie algebra $\mathfrak{g}$, its basis will be hereafter denoted by $\{e_1,e_2,e_3\}$ with commutation relations given in Table \ref{tabela3w}. We also choose the induced bases $\{e_{12},e_{13},e_{23}\}$ and $\{e_{123}\}$ in $\Lambda^2\mathfrak{sl}_2$ and $\Lambda^3\mathfrak{sl}_{2}$, respectively.
	
\subsection{General properties}

In order to facilitate further analysis of coboundary Lie bialgebra structures on three-dimensional Lie algebras, let us prove a few additional results on the characterisation of $(\Lambda^m\mathfrak{g})^\mathfrak{g}$ and ${\rm Aut}(\mathfrak{g})$. For the remaining part of this section, we assume that $\dim \mathfrak{g}=3$. 
	
	\begin{proposition}\label{Prop:Aut3} 
		Let $\mathfrak{g}$ be such that $\kappa_{\mathfrak{g}}\neq 0$ and $\mathfrak{g}_{1)}\subset \ker \kappa_{\mathfrak{g}}$ is a two-dimensional abelian Lie subalgebra. If $v\notin \mathfrak{g}_{1)}$, then every $T \in {\rm Aut}(\mathfrak{g})$ leaves invariant the set of eigenvectors of ${\rm ad}_v|_{\mathfrak{g}_{1)}}$.
	\end{proposition}
	\begin{proof} 
		Let us prove that if $T\in {\rm Aut}(\mathfrak{g})$, then $Tv\in v+\mathfrak{g}_{1)}$ or $Tv\in -v+\mathfrak{g}_{1)}$ for every $v\notin \mathfrak{g}_{1)}$. Since $T\in {\rm Aut}(\mathfrak{g})$, one has that $\kappa_{\mathfrak{g}}(Tv,Tv)=\kappa_{\mathfrak{g}}(v,v)$. Since $v$ and $\mathfrak{g}_{1)}$ generate $\mathfrak{g}$ and $T$ is  injective, $Tv=\lambda v+h$ for an $h\in \mathfrak{g}_{1)}$ and $\lambda\in \mathbb{R}\backslash\{0\}$. As  $\mathfrak{g}_{1)}\subset \ker\kappa_\mathfrak{g}$, then $\kappa_\mathfrak{g}(Tv,Tv)=\lambda^2\kappa_\mathfrak{g}(v,v)$ for every $v\in \mathfrak{g}$. Since $\kappa_\mathfrak{g}\neq 0$ , one has that $\lambda\in \{\pm 1\}$. Hence, $Tv\in v+\mathfrak{g}_{1)}$ or $Tv\in -v+\mathfrak{g}_{1)}$. Since $\dim\mathfrak{g}_{1)}=2$ and $\mathfrak{g}_{1)}$ is an abelian ideal of $\mathfrak{g}$ invariant under ${\rm Aut}(\mathfrak{g})$, one obtains that 
		$
		{\rm ad}_{v}|_{\mathfrak{g}_{1)}}=\pm{\rm ad}_{Tv}|_{\mathfrak{g}_{1)}}=\pm T|_{\mathfrak{g}_{1)}}\circ{\rm ad}_{v}|_{\mathfrak{g}_{1)}}\circ T^{-1}|_{\mathfrak{g}_{1)}}.
		$
		Consequently, if $e$ is an eigenvector of ${\rm ad}_v|_{\mathfrak{g}_{1)}}$, then $Te$ is a new eigenvector of ${\rm ad}_v|_{\mathfrak{g}_{1)}}$.  
	\end{proof}

	Proposition \ref{Prop:Aut3} can be helpful to study  $T|_{\mathfrak{g}^{(1)}}$ in various scenarios. For instance, if ${\rm ad}_v|_{\mathfrak{g}^{1)}}$ has two eigenvectors $e_1,e_2$ with different eigenvalues $\lambda_1,\lambda_2$ satisfying that $\lambda_1 +\lambda_2 = 0$ and $Tv \in -v + \mathfrak{g}_{1)}$, then $T|_{\mathfrak{g}^{1)}}$ is an anti-diagonal matrix in the basis $\{e_1,e_2\}$. If $\lambda_1+\lambda_2 \neq 0$ and $\lambda_1\neq\lambda_2$, then $Tv \in v + \mathfrak{g}_{1)}$ and $T|_{\mathfrak{g}^{1)}}$ is diagonal. Several variations of this reasoning can be applied, e.g. when ${\rm ad}_v|_{\mathfrak{g}^{1)}}$ is triangular.

	\begin{proposition}\label{Prop:Aut2} 
		Let $\Omega\in (\Lambda^3\mathfrak{g}^*)\backslash \{0\}$ and assume that $\Upsilon:\Lambda^2\mathfrak{g}\ni r\mapsto \Omega ([r,r] )\in \mathbb{R}$ is a semi-definite function (i.e. a non-negative or non-positive function different than zero). Then, every automorphism of $\mathfrak{g}$ has positive determinant.
	\end{proposition}

\begin{proof} 
 Since $\dim\mathfrak{g}=3$, the $\Omega$ is a basis of $\Lambda^3\mathfrak{g}^*$ and there exists a dual basis $\theta\in \Lambda^3\mathfrak{g}$. As $\Omega([r,r] )=\Upsilon(r)$, then $[r,r]= \Upsilon(r)\theta$. Since $\Upsilon$ is not identically zero, there exists an $r\in \Lambda^2\mathfrak{g}$ such that	$[r,r]=\Upsilon(r)\theta\neq 0$. If $T\in {\rm Aut}(\mathfrak{g})$, then
$\Upsilon(r) \det (T)\,\theta=\det (T)\, [r, r]= \Lambda^3T[r,r]$. By the properties of the Schouten bracket,  $\Lambda^3T[r,r]=[\Lambda^2Tr,\Lambda^2Tr]=\Upsilon(\Lambda^2Tr)\theta$. Hence, $\Upsilon(r)\det (T)=\Upsilon(\Lambda^2Tr).
		$
		As $\Upsilon$ is semi-definite and $\Upsilon(r)\neq 0$, then $\det (T)=\Upsilon(\Lambda^2Tr)/\Upsilon(r)>0$.
	\end{proof}
	
Finally, the derivation of $(\Lambda^3\mathfrak{g})^\mathfrak{g}$ is simplified by the following proposition. 

	\begin{proposition}\label{Prop:Sym}
		Each $G$-gradation on $\mathfrak{g}$ has a unique homogeneous subspace  $(\Lambda^3\mathfrak{g})^{(\alpha)}\neq 0$ and $[\mathfrak{g}^{(\beta)},(\Lambda^3\mathfrak{g})^{(\alpha)}]=0$ for $\beta\neq 0$. If $\mathfrak{g}$ has a root gradation, then $\Lambda^3\mathfrak{g}=(\Lambda^3\mathfrak{g})^\mathfrak{g}$ if and only if $\alpha=0$.
	\end{proposition}
 \begin{proof}
Since $\dim (\mathfrak{g}) = 3$, the subspace $\Lambda^3 \mathfrak{g}$ is one-dimensional and thus, it yields a unique nonzero homogeneous subspace of a certain degree, say $\alpha$. By Theorem \ref{thm:root}, $[\mathfrak{g}^{(\beta)}, (\Lambda^3 \mathfrak{g})^{(\alpha)}]_S \subset (\Lambda^3 \mathfrak{g})^{(\alpha \star \beta)}$. For $\beta \neq 0$, we get $\alpha \star \beta \neq \alpha$. As the only nonzero homogeneous subspace of $\Lambda^3 \mathfrak{g}$ is of degree $\alpha$, it follows that $[\mathfrak{g}^{(\beta)}, (\Lambda^3 \mathfrak{g})^{(\alpha)}]_S = 0$. Let $\mathfrak{g}$ admit a root gradation. If $\Lambda^3\mathfrak{g}=(\Lambda^3\mathfrak{g})^\mathfrak{g}$, then for any $v \in \mathfrak{g}^{(0)}$ and $w \in \Lambda^3 \mathfrak{g}$, we get $[v, w]_S = [\Xi(\beta_1)(v) + \Xi(\beta_2)(v) + \Xi(\beta_3)(v)] w = \Xi(\beta_1 + \beta_2 + \beta_3)(v) = 0$, where $\beta_i$ denote the (possibly different) degrees of basis elements of $\mathfrak{g}$ such that $\beta_1 + \beta_2 + \beta_3 = \alpha$ and $\Xi$ is the group homomorphism introduced in Definition \ref{Def:rootgrad}. Then, $\Xi(\alpha) = 0$. Since $\Xi$ is an injective map, we obtain $\alpha = 0$. Conversely, let us take $\alpha = 0$. Then, $[v, w]_S = [\Xi(\beta_1)(v) + \Xi(\beta_2)(v) + \Xi(\beta_3)(v)] w = \Xi(\beta_1 + \beta_2 + \beta_3)(v) = \Xi(0)(v)$ for any $v \in \mathfrak{g}^{(0)}$. Since $\Xi$ is a homomorphism, $\Xi(0) = 0$ and $[v, w]_S = 0$ for all $v \in \mathfrak{g}^{(0)}$.
 \end{proof}

\subsection{Lie algebra \texorpdfstring{$\mathfrak{sl}_2$}{}}
	
The Lie algebra $\mathfrak{sl}_2$ admits a root decomposition which gives rise to a root $\mathbb{Z}$-gradation such that the whole $\Lambda^3\mathfrak{sl}_2$ is a homogeneous space of zero degree. Thus, Proposition \ref{Prop:Sym} yields that $\Lambda^3\mathfrak{sl}_2=(\Lambda^3\mathfrak{sl}_2)^{\mathfrak{sl}_2}$ and every element of $\Lambda^2\mathfrak{sl}_2$ satisfies the mCYBE. The root gradation of $\mathfrak{sl}_2$ also implies, by virtue of Proposition \ref{gin0}, that  $(\Lambda^2\mathfrak{sl}_2)^{\mathfrak{sl}_2}\subset (\Lambda^2\mathfrak{sl}_2)^{(0)}$. It is then immediate to verify that $(\Lambda^2\mathfrak{sl}_2)^{\mathfrak{sl}_2}=0$. Thus, every $r\in \Lambda^2\mathfrak{sl}_2$ induces a different cocommutator $\delta_r(\cdot):=[\cdot, r]_{S}$.

 Let us analyse the level sets $S_k$ of $f_{\Lambda^2\mathfrak{sl}_2}:r\in \Lambda^2\mathfrak{sl}_2\mapsto \kappa_{\Lambda^2\mathfrak{sl}_{2}}(r,r)\in \mathbb{R}$. 
	If $r=x e_{12}+ y e_{13} + z e_{23}$, then  $f_{\Lambda^2\mathfrak{sl}_2}(r):=8xy - 4z^2$ and $f_{\Lambda^2\mathfrak{sl}_2}$ admits three types of $S_k$ according to the sign of $k$:
 \begin{itemize}
\item if $k<0$, then $S_k$ is a one-sheeted hyperboloid; 
\item the level set $S_0$ consists of two cones (without their common vertex), one opposite to the other, and the origin of $\Lambda^2\mathfrak{sl}_2$; 
\item if $k>0$, then $S_k$ is a two-sheeted hyperboloid with two parts contained within the region $x>0,y>0$ and $x<0,y<0$, respectively.
 \end{itemize}
 The representative level sets of each type are presented in Figure \ref{sl2_aut}.
	
	As explained in the previous section, each $S_k$ is the union of different orbits $\mathcal{O}_w$, $w \in \Lambda^2\mathfrak{sl}_2$, of the ${\rm Inn}(\mathfrak{sl}_2)$-action on $\Lambda^2\mathfrak{sl}_2$. Using Proposition \ref{proporb}, one gets that the dimension of $\mathcal{O}_w$ is $\dim \Theta_w^2=2$ for $w \in \Lambda^2\mathfrak{sl}_2\backslash\{0\}$ and $\dim \Theta_0^2=0$. Since ${\rm Inn}(\mathfrak{sl}_{2})$ is connected, the $\mathcal{O}_w$ are two- or zero-dimensional connected immersed submanifolds. Then:
 \begin{itemize}
 \item for $k<0$, each $S_k$ is a single orbit; 
 \item each $S_k$ with $k>0$ consists of two connected orbits; 
 \item the $S_0$ has three orbits given by two cones for points with $z>0$ or $z<0$, and the origin $(0,0,0)$. 
 \end{itemize}
 Consequently, there are six families of inequivalent classes of $r$-matrices on $\mathfrak{sl}_{2}$ relative to the action of ${\rm Inn}(\mathfrak{sl}_{2})$ (cf. \cite{FJ15}). The representatives of each class are $r_0 =0, r=ae_{23}$, with $a>0$ (one-sheeted hyperboloids), $r = a (e_{12}+e_{13})$, with $a\neq 0$, (two-sheeted hyperboloids),  and $r = \pm e_{12}$ (cones). 
	
Finally, let us search for the automorphisms that would identify the disconnected orbits within each $S_k$. The $T\in {\rm Aut}(\mathfrak{sl}_2)$ such that $T(e_1) := e_1, T(e_2) := -e_2, T(e_3) := -e_3$ can be extended to $\Lambda^2T$ giving rise to a map satisfying that $\Lambda^2T(e_{12})=- e_{12},\Lambda^2T(e_{13})=- e_{13}, \Lambda^2T(e_{23})=e_{23},
	$
This extended map indetifies two connected parts of the two-sheeted hyperboloids in $S_k$ for each fixed $k> 0$. It also maps the two cones contained in $S_0$. This leads to four separate classes of inequivalent $r$-matrices: $r = 0$, $r = |k| e_{23}$ (for $S_k$ with $k < 0$), $r = k (e_{12} + e_{13})$ (for $S_k$ with $k > 0$) and $r = e_{12}$ (for $S_0$).
	
Our result agrees with \cite{Go00}, but they do not match \cite[p. 56]{FJ15}), since the authors performed the classification only up to ${\rm Inn}(\mathfrak{sl}_2)$. However, as shown in the previous section, ${\rm Inn}(\mathfrak{sl}_2) \neq {\rm Aut}(\mathfrak{sl}_2)$.

\subsection{Lie algebra \texorpdfstring{$\mathfrak{su}_{2}$}{}}
	
	The $\mathbb{Z}_2$-gradations of $\mathfrak{su}_2$ and the associated decomposition of $\Lambda^3\mathfrak{su}_2$, both presented in Table \ref{tabela3w}, allow to show that the space $\Lambda^3\mathfrak{su}_2$ is $\mathfrak{su}_2$-invariant. Thus, every element of $\Lambda^2 \mathfrak{su}_2$ is a solution to the mCYBE.	
	
Although the $\mathbb{Z}_2$-gradation of $\mathfrak{su}_2$ implies the whole space $\Lambda^3 \mathfrak{su}_2$ is a single homogeneous space of degree zero, Proposition \ref{Prop:Sym} cannot be applied to analyse $(\Lambda^2\mathfrak{su}_2)^{\mathfrak{su}_2}$, as the considered gradation is not a root gradation. Since $(\Lambda^2\mathfrak{su}_2)^{\mathfrak{su}_2}=\{0\}$, every $r\in \Lambda^2\mathfrak{su}_2$ induces a different cocomutator and the classification of coboundary cocomutators of $\mathfrak{su}_2$  up to ${\rm Aut}(\mathfrak{su}_2)$ amounts to classifying $r$-matrices.
	
	Let us consider the Killing and Killing-type forms on $\mathfrak{su}_2, \Lambda^2 \mathfrak{su}_2$, and $\Lambda^3 \mathfrak{su}_2$. In the bases given in Table \ref{tabela3w}, these forms read $[\kappa_{\mathfrak{su}_{2}}] = -2\ \mathbb{I}_{3 \times 3},$ $[\kappa_{\Lambda^2 \mathfrak{su}_{2}}] = 4\ \mathbb{I}_{3 \times 3}$,  $[\kappa_{\Lambda^3\mathfrak{su}_{2}}]=-8\ \mathbb{I}_{1\times 1}$, where $\mathbb{I}_{k \times k}$ denotes the identity matrix of size $k$.
Then, the level sets $S_k$ are given $\kappa_{\Lambda^2\mathfrak{su}_{2}}(r, r)=4(x^2 + y^2 + z^2) = k$. Since ${\rm Inn}(\mathfrak{su}_{2})$ is connected, Proposition \ref{proporb} implies that the dimension of the orbits $\mathcal{O}_w$ of the ${\rm Inn}(\mathfrak{su}_{2})$-action on $\Lambda^2\mathfrak{su}_{2}$ satisfies ${\rm dim}\,\mathcal{O}_w = 2$ for $w\in \Lambda^2\mathfrak{su}_2\backslash\{0\}$ and ${\rm dim}\,\mathcal{O}_0 = 0$.  

As the orbits of the ${\rm Inn}(\mathfrak{su}_2)$-action on $\Lambda^2\mathfrak{su}_2$ are connected immersed submanifolds contained in the level sets $S_k$ and must be open relative to the topology of each $S_k$ (with $k\geq 0$), which are connected, each orbit of ${\rm Inn}(\mathfrak{su}_2)$ must be the whole $S_k$ for each $k\geq 0$. 

	Hence, non-equivalent $r \in \Lambda^2 \mathfrak{su}_2$, with respect to the action of ${\rm Inn}(\mathfrak{su}_2)$, are given by elements $r_a=a e_{12}$, with $a\geq 0$. Since the orbits of the action of ${\rm Aut}(\mathfrak{su}_2)$ on $\Lambda^2\mathfrak{su}_2$ are given by the sum of orbits of ${\rm Inn}(\mathfrak{su}_2)$ and they are contained in the surfaces $S_k$, the orbits of the action of ${\rm Aut}(\mathfrak{su}_2)$ in $\Lambda^2\mathfrak{su}_2$ are the spheres $S_k$ with $k>0$ and the point $k=0$ (see Figure \ref{su2_aut}). Therefore, we conclude there are two classes of inequivalent $r$-matrices: $r = 0$ and $r_a = a e_{12}$ for $a > 0$. This agrees with the results given in \cite{FJ15,Go00}.

	\noindent
		\begin{minipage}{0.45\textwidth}
		\begin{center}
			\includegraphics[scale=0.33]{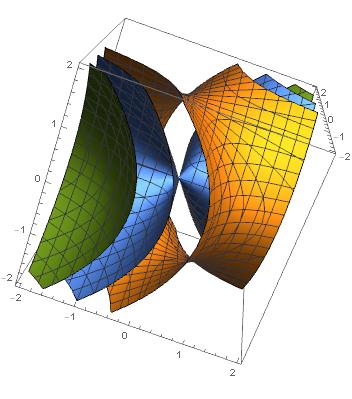}
			\captionof{figure}{Orbits of ${\rm Aut}(\mathfrak{sl}_2)$ acting on $\Lambda^2 \mathfrak{sl}_2$.}\label{sl2_aut}
		\end{center}
	\end{minipage}$\quad$
		\begin{minipage}{0.45\textwidth}
		\begin{center}
			\includegraphics[scale=0.33]{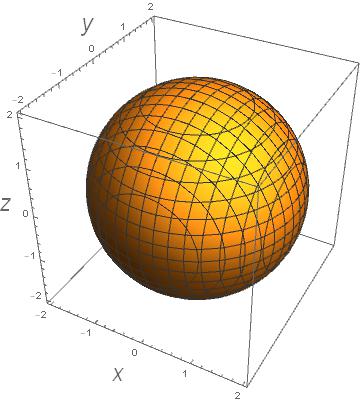}
			\captionof{figure}{An orbit of ${\rm Aut}(\mathfrak{su}_2)$ acting on $\Lambda^2 \mathfrak{su}_2$.}\label{su2_aut}
		\end{center}
	\end{minipage}

\subsection{Lie algebra \texorpdfstring{$\mathfrak{h}$}{}}	
	
	Let us analyse the three-dimensional {\it Heisenberg algebra} $\mathfrak{h}$ \cite{FJ15} described in Table \ref{tabela3w}. This is the only three-dimensional nilpotent Lie algebra \cite{SW14}.
	
With the help of the $\mathbb{Z}$-gradation of $\lambda^3 \mathfrak{h}$ presented in Table \ref{tabela3w}, one can verify that $(\Lambda^3\mathfrak{h})^{\mathfrak{h}}=\Lambda^3\mathfrak{h}$ and thus, every element of $\Lambda^2\mathfrak{h}$ is a solution to the mCYBE. Similarly, using $\mathbb{Z}$-gradation on $\Lambda^2 \mathfrak{h}$, one shows that $(\Lambda^2\mathfrak{h})^\mathfrak{h}=\langle e_{13},e_{12}\rangle$. Therefore, we restrict further analysis to the reduced space $\Lambda^2_R \mathfrak{h}$. 
	
As a consequence of Proposition \ref{prop:requiv}, every class of $\Lambda^2_R\mathfrak{h}$ gives rise to a unique Lie bialgebra. Thus, in order to classify coboundary Lie bialgebras, it is sufficient to study reduced $r$-matrices in $\Lambda^2_R\mathfrak{h}$. In view of Proposition \ref{proporb}, the orbits $\mathcal{O}_w$ of ${\rm Inn}$-action on $\Lambda^2_R\mathfrak{h}$ satisfy ${\rm dim} \mathcal{O}_w = 1$ for $w \in \Lambda^2_R \mathfrak{h} \backslash \{[0]\}$ and $\mathcal{O}_w = 0$ for $w = [0]$. Thus, the reduced $r$-matrices inequivalent relative to ${\rm Inn}(\mathfrak{h})$ read $r = [0]$ and $r_{\pm} = [\pm e_{12}]$. 

Define the maps $T_\alpha\in {\rm Aut}(\mathfrak{g})$, with $\alpha \in \mathbb{R}\backslash\{0\}$, given by
		$
		T_\alpha(e_1):=\alpha e_1, T(e_2):=e_2, T(e_3):= \alpha e_3.
		$ 	Therefore, $\Lambda^2 T_\alpha (e_{12}) = \alpha e_{12}$ for any $\alpha\neq 0$. By Lemma \ref{reduced_action}, the action $\Lambda^2 T$ induces an action $[\Lambda^2 T]$ on $\Lambda^2_R\mathfrak{h}$. Since for $\alpha = (-1)$, $[\Lambda^2 T]$ identifies $r_{\pm}$, we conclude that there are two orbits of ${\rm Aut}(\mathfrak{h})$-action on $\Lambda^2_R \mathfrak{h}$ given by $r_0 = [0]$ and $r = [e_{12}]$. Figure \ref{h_aut} depicts the equivalence classes of $\Lambda^2_R\mathfrak{h}$ in $\Lambda^2\mathfrak{h}$.
		
		\noindent
		\begin{minipage}{0.50\textwidth}
		\begin{center}
			\includegraphics[scale=0.35]{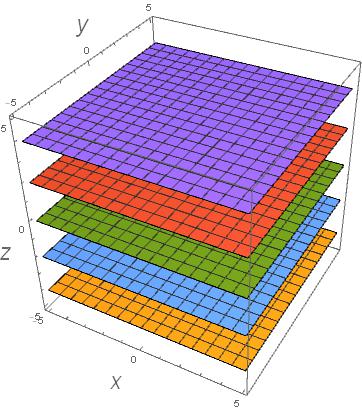}
			\captionof{figure}{Orbits of ${\rm Inn}(\mathfrak{h})$ acting on $\Lambda^2 \mathfrak{h}$.}\label{h_inn}
		\end{center}
	\end{minipage}$\quad$
	\begin{minipage}{0.50\textwidth}
		\begin{center}
			\includegraphics[scale=0.35]{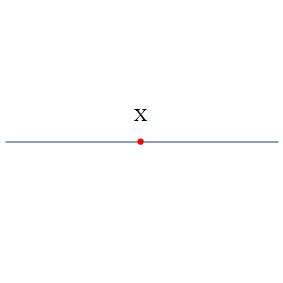}
			\captionof{figure}{Orbits of ${\rm Aut}(\mathfrak{h})$ on $\Lambda^2_R \mathfrak{h}$.}\label{h_aut}
		\end{center}
	\end{minipage}	

\subsection{Lie algebra \texorpdfstring{$\mathfrak{r}'_{3,0}$}{}}
	
In view of gradations presented in Table \ref{tabela3w}, it is immediate to verify that $(\Lambda^3\mathfrak{r}'_{3,0})^{\mathfrak{r}'_{3,0}}=\Lambda^3\mathfrak{r}'_{3,0}$. Thus, all the elements of $\Lambda^2\mathfrak{r}'_{3,0}$ are solutions to mCYBE, and the study of coboundary cocommutators reduces to analysing the equivalent reduced $r$-matrices in $\Lambda^2_R\mathfrak{r}'_{3,0}$.
	
	Recall that $(\Lambda^2\mathfrak{r}'_{3,0})^{\mathfrak{r}'_{3,0}}$ is the sum of homogeneous $\mathfrak{r}'_{3,0}$-invariant elements in $\Lambda^2\mathfrak{r}'_{3,0}$. It easily follows by using the gradations in $\mathfrak{r}'_{3,0}$ that $\langle e_{23}\rangle\subset (\Lambda^2\mathfrak{r}'_{3,0})^{(2)}$ is $\mathfrak{r}'_{3,0}$-invariant. To obtain the $\mathfrak{r}'_{3,0}$-invariant elements within $\Lambda^2(\mathfrak{r}'_{3,0})^{(1)}$, we consider an arbitrary element $e_1\wedge \lambda(e_2,e_3)\in\Lambda^2(\mathfrak{r}'_{3,0})^{(1)}$, where $\lambda(e_2,e_3)$ stands for a linear combination of $e_2$ and $e_3$. Then,
	$
	[e_2,e_1\wedge \lambda(e_2,e_3)] =-e_3\wedge \lambda(e_2,e_3)=0, [e_3,e_1\wedge \lambda(e_2,e_3)] =e_2\wedge \lambda(e_2,e_3)=0.
	$
	Hence, $\lambda(e_2,e_3)=0$ and $(\Lambda^2\mathfrak{r}'_{3,0})^{\mathfrak{r}'_{3,0}}=\langle e_{23}\rangle$. In consequence, we will continue our analysis on the reduced space $\Lambda^2_R \mathfrak{r}'_{3,0}$.
	
 Let us discuss the existence of $\mathfrak{r}'_{3,0}$-invariant metrics on $\Lambda^2_R \mathfrak{r}'_{3,0}$. Consider the basis $\{[e_{12}], [e_{13}]\}$ of $\Lambda^2_R \mathfrak{r}'_{3,0}$. If $\Lambda^2_R{\rm ad}: v\in \mathfrak{r}'_{3,0} \mapsto [[v],\cdot]_R\in \mathfrak{gl}(\Lambda^2_R \mathfrak{r}'_{3,0})$, where $[\cdot,\cdot]_R$ is the bracket on $\Lambda^2_R \mathfrak{r}'_{3,0}$ induced by the algebraic bracket on $\Lambda^2 \mathfrak{r}'_{3,0}$, then
$$
\begin{aligned}
	&{\rm Im}\,\Lambda^2_R{{\rm ad}}_{e_1}=\langle [e_{13}],[e_{12}]\rangle,\,\,&\ker\,\Lambda^2_R{{\rm ad}}_{e_1}&=\langle [0]\rangle,\,\, &{\rm Im}\,\Lambda^2_R{{\rm ad}}_{e_2}=\langle [0]\rangle,\,\,\\ 
	&\ker\,\Lambda^2_R{{\rm ad}}_{e_2}=\langle [e_{13}]\rangle,\,\,
	&{\rm Im}\,\Lambda^2_R{{\rm ad}}_{e_3}&=\langle [e_{23}]\rangle =\langle [0]\rangle,\,\,&\ker \,\Lambda^2_R{{\rm ad}}_{e_3}=\langle [e_{12}]\rangle,
	\end{aligned}
	$$
	\begin{equation*}
		\begin{gathered}
			b^R_{\Lambda^2\mathfrak{r}'_{3,0}}([e_{12}], [e_{12}]) = b^R_{\Lambda^2\mathfrak{r}'_{3,0}}([[e_1], [e_{13}]]_R, [e_{12}]) = -b^R_{\Lambda^2\mathfrak{r}'_{3,0}}([e_{13}], [[e_1], [e_{12}]]_R)=
			b^R_{\Lambda^2\mathfrak{r}'_{3,0}}([e_{13}], [e_{13}]),\\
			b^R_{\Lambda^2\mathfrak{r}'_{3,0}}([e_{13}],[e_{12}]) = -b^R_{\Lambda^2\mathfrak{r}'_{3,0}}([[e_1], [e_{12}]]_R, [e_{12}]) = b^R_{\Lambda^2\mathfrak{r}'_{3,0}}([e_{12}],[[e_1],[e_{12}]]_R)=-b^R_{\Lambda^2\mathfrak{r}'_{3,0}}([e_{12}],[e_{13}]).
		\end{gathered}
	\end{equation*}
	
Then, Propositions \ref{prop:sym_form} and \ref{prop:form_cond} yield  that	
	$[b_{\Lambda^2 \mathfrak{r}'_{3,0}}^R]=a_1{\rm Id}$ for $a_1\in \mathbb{R}$, 
	which is an $\mathfrak{r}'_{3,0}$-invariant metric. For simplicity, we hereafter assume that $a_1=1$.

 Let us study the equivalence of reduced $r$-matrices up to {\rm Inn}$(\mathfrak{r}_{3,0}')$. if we write an element $[r] \in \Lambda^2_R\mathfrak{r}_{3,0}'$ in the given basis as $[r] = x [e_{12}] + y [e_{13}]$ for some $x, y \in \mathbb{R}$, then it follows that $b^R_{\Lambda^2 \mathfrak{r}'_{3,0}} ([r], [r]) = x^2 + y^2$. Thus, the level sets $S_k$ of $b^R_{\Lambda^2 \mathfrak{r}'_{3,0}}$ are circles of radius $k$ centered around the origin and the point $(0,0)$. Since by Proposition \ref{proporb}, orbits $\mathcal{O}_w$ of the ${\rm Inn}(\mathfrak{r}'_{3,0})$-action satisfy $\dim \mathcal{O}_w = 1$ for $[r] \in \Lambda^2_R \mathfrak{r}'_{3,0}$ such that its coordinates $x^2 +y^2 \neq 0$ and $\dim \mathcal{O}_w = 0$ otherwise. Thus, each level set $S_k$ gives a single orbit in $\Lambda_R^2\mathfrak{r}_{3,0}'$ relative to the action of ${\rm Inn}(\mathfrak{r}'_{3,0})$. Therefore, there exist two classes of inequivalent $r$-matrices relative to ${\rm Inn}(\mathfrak{r}'_{3,0})$: $r = 0$ and $r_{\mu} = [\mu e_{12}]$, where $\mu \in \mathbb{R}^+$.
	
Consider the $T_\alpha\in {\rm Aut}(\mathfrak{r}_{3,0}')$, with $\alpha \in \mathbb{R}\backslash\{0\}$, satisfying  $T_\alpha(e_1):=e_1, T_\alpha(e_2):=\alpha e_2, T_\alpha(e_3):=\alpha e_3.$ Then, each $\Lambda^2T_\alpha\in GL(\Lambda^2\mathfrak{r}_{3,0}')$ satisfies that $\Lambda^2T_\alpha(e_{12})=\alpha e_{12}$. Thus, the functions $\Lambda^2_RT_\alpha$ identifies $r_{\mu}$ for all $\mu \in \mathbb{R}_{+}$. Hence, there are two classes of inequivalent reduced $r$-matrices up to the action of ${\rm Aut}(\mathfrak{r}_{3,0}')$, represented by $r = 0$ and $r = e_{12}$ (see Figure \ref{r'30_aut}). This matches the results in \cite{FJ15}. 

	\noindent
	\begin{minipage}{0.45\textwidth}
		\begin{center}
			\includegraphics[scale=0.35]{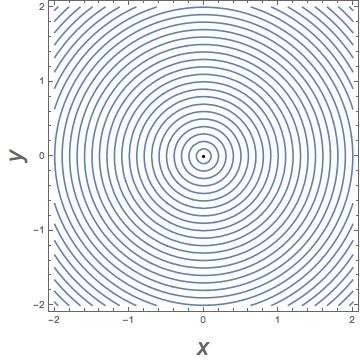}
			\captionof{figure}{Some orbits of ${\rm Inn}(\mathfrak{r}'_{3,0})$ acting on $\Lambda^2_R \mathfrak{r}'_{3,0}$.}\label{r'30_inn}
		\end{center}
	\end{minipage}$\quad$
		\begin{minipage}{0.45\textwidth}
		\begin{center}
			\includegraphics[scale=0.35]{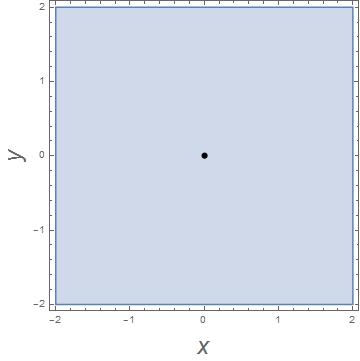}
			\captionof{figure}{Orbits of ${\rm Aut}(\mathfrak{r}'_{3,0})$ acting on $\Lambda^2_R \mathfrak{r}'_{3,0}$.}\label{r'30_aut}
		\end{center}
		\end{minipage}

\subsection{Lie algebra \texorpdfstring{$\mathfrak{r}_{3,-1}$}{}}
	
	Since $\mathfrak{r}_{3,-1}$ admits a root gradation (see Table \ref{tabela3w}) and $\Lambda^3\mathfrak{r}_{3,-1}$ has degree 0, Proposition \ref{Prop:Sym} shows that $\Lambda^3\mathfrak{r}_{3,-1}=(\Lambda^3\mathfrak{r}_{3,-1})^{\mathfrak{r}_{3,-1}}$ and every $r\in \Lambda^2\mathfrak{r}_{3,-1}$ is an $r$-matrix, while Proposition \ref{gin0} yields that $(\Lambda^2\mathfrak{r}_{3,-1})^{\mathfrak{r}_{3,-1}}\subset (\Lambda^2\mathfrak{r}_{3,-1})^{(0)}$. It is then immediate that $(\Lambda^2\mathfrak{r}_{3,-1})^{\mathfrak{r}_{3,-1}}=\langle e_{23}\rangle$ and $\Lambda^2_R\mathfrak{r}_{3,-1}=\langle [e_{13}],[e_{12}]\rangle$.

	Let us classify cocommutators on $\mathfrak{r}_{3,-1}$ via  $\mathfrak{r}_{3,-1}$-invariant metrics, $b^R_{\Lambda^2\mathfrak{r}_{3,-1}}$, on $\Lambda^2_R\mathfrak{r}_{3,-1}$. 
	Define $\Lambda^2_R{\rm ad}:v\in \mathfrak{r}_{3,-1}\mapsto [[v],\cdot]_R\in \mathfrak{gl}(\Lambda^2_R\mathfrak{r}_{3,-1})$. In the basis $\{[e_{12}], [e_{13}]\}$ of $\Lambda^2_R\mathfrak{r}_{3,-1}$, one gets
	$
	b^R_{\Lambda^2\mathfrak{r}_{3,-1}}([[e_1],[e_{12}]]_R,[e_{12}])=b^R_{\Lambda^2\mathfrak{r}_{3,-1}}([e_{12}],[e_{12}]),\,\,$ and $ b^R_{\Lambda^2\mathfrak{r}_{3,-1}}([[e_1],[e_{13}]]_R,[e_{13}])=-b^R_{\Lambda^2\mathfrak{r}_{3,-1}}([e_{13}],[e_{13}]).
	$ Then, a short calculation shows that
	the $\mathfrak{r}_{3,-1}$-invariant metrics on $\Lambda^2_R\mathfrak{r}_{3,-1}$ are
	$$
	[b^R_{\Lambda^2\mathfrak{r}_{3,-1}}] \!=\! \left(\begin{array}{cc}
	0&\beta\\
	\beta&0
	\end{array}\right), \,\, \beta\in \mathbb{R}.
	$$
	
		Let $\{[e_{12}],[e_{13}]\}$ be a basis of $\Lambda^2_R\mathfrak{r}_{3,-1}$. Then, any $r_R \in \Lambda^2_R\mathfrak{r}_{3,-1}$ can be written as $r_R = x [e_{12}] + y[ e_{13}]$ for certain $x,y \in \mathbb{R}$ and $b^R_{\Lambda^2\mathfrak{r}_{3,-1}}(r_R, r_R) = 2xy$. Thus, the level sets of $b^R_{\Lambda^2\mathfrak{r}_{3,-1}}$ are given by the equation $xy = k$. By Proposition \ref{proporb}, the orbits $\mathcal{O}_w$ of ${\rm Inn}(\mathfrak{r}_{3,-1})$-action satisfy $\dim \mathcal{O}_w = 1$ for $w \in \Lambda^2_R \mathfrak{r}_{3,-1}$ such that its coordinates $x^2 + y^2 \neq 0$ and $\dim \mathcal{O}_w = 0$ for $w = [0]$. Chosen orbits are depicted in  Figure \ref{14}.
		
Thus, the representatives of inequivalent reduced $r$-matrices up to the action of ${\rm Inn}(\mathfrak{r}_{3,-1})$ read
		$r_0 = [e_{23}], r_2^{(\pm)} = \pm  [e_{12}], r_3^{(\pm)} = \pm  [e_{13}]$, and $r^{(\pm, \pm)} =a (\pm[e_{12}] \pm [e_{13}])$ for $a \in \mathbb{R}_{+}$. 
   
  Let us search for automorphisms of $\mathfrak{r}_{3,-1}$ that might identify some of the $r$-matrices listed above. Notice that $\mathfrak{r}_{3,-1}$ satisfies the conditions given in Proposition \ref{Aut3}.  Hence, automorphisms of $\mathfrak{r}_{3,-1}$ must match one of the following automorphisms
			
	\begin{equation*}
		\begin{gathered}
			T_{\alpha,\beta}(e_1):=e_1+v,\qquad T_{\alpha,\beta}(e_2):=\alpha e_2,\qquad T_{\alpha,\beta}(e_3):=\beta e_3,\qquad \forall \alpha,\beta\in \mathbb{R}\backslash\{0\},\\
			T_{\alpha,\beta}'(e_1):=-e_1+v, \qquad T_{\alpha,\beta}'(e_2):=\alpha e_3, \qquad T_{\alpha,\beta}'(e_3):=\beta e_2,\qquad \forall \alpha,\beta\in \mathbb{R}\backslash\{0\},
		\end{gathered}
	\end{equation*}
\noindent		for certain $v\in \langle e_2,e_3\rangle$. Then,
	\begin{equation*}
		\begin{gathered}
			\Lambda_R^2T_{\alpha,\beta}([e_{12}])=\alpha [e_{12}],\qquad \Lambda^2_RT_{\alpha,\beta}([e_{13}])=\beta [e_{13}],\qquad \forall \alpha,\beta\in \mathbb{R}\backslash\{0\},\\
			\Lambda_R^2T'_{\alpha,\beta}([e_{12}])=-\alpha [e_{13}],\qquad \Lambda_R^2T'_{\alpha,\beta}([e_{13}])=-\beta [e_{12}],\qquad \forall \alpha,\beta\in \mathbb{R}\backslash\{0\}.
		\end{gathered}
	\end{equation*}

The map $\Lambda_R^2 T_{\alpha,\beta}$ identifies all four elements $r^{(\pm, \pm)}$. For $\alpha = (-1)$ and $\beta = 1$, it identifies both $r_2^{(\pm)}$. Similar identification happens for $r_3^{(\pm)}$ when $\alpha = 1, \beta = (-1)$. Finally, $r_2^{(\pm)}$ and $r_3^{(\pm)}$ are identified by $\Lambda_R^2 T'_{\alpha,\beta}$ for $\alpha = \beta = 1$. Thus, there are only three reduced $r$-matrices on $\Lambda^2 \mathfrak{r}_{3,-1}$ given by $r_0 = 0$, $r_1 = [e_{12}]$ and $r_2 = [e_{12}] + [e_{13}]$.
 
	\noindent
	\begin{minipage}{0.45\textwidth}
		\begin{center}
			\includegraphics[scale=0.35]{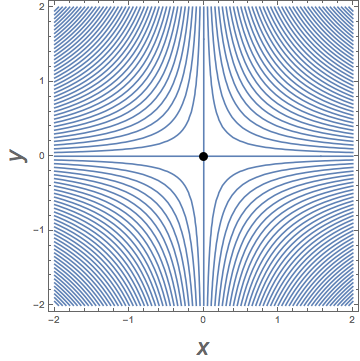}
			\captionof{figure}{Some orbits of ${\rm Inn}(\mathfrak{r}_{3,-1})$ acting on $\Lambda^2_R \mathfrak{r}_{3,-1}$.}\label{14}
		\end{center}
		
	\end{minipage}$\quad$
	\begin{minipage}{0.45\textwidth}
		\begin{center}
			\includegraphics[scale=0.35]{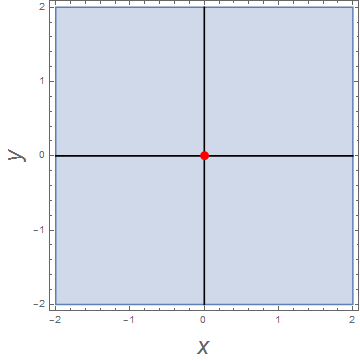}
			\captionof{figure}{Orbits of ${\rm Aut}(\mathfrak{r}_{3,-1})$ acting on $\Lambda^2_R \mathfrak{r}_{3,-1}$}\label{r3-1_aut}
		\end{center}
		
	\end{minipage}

	\subsection{Lie algebra \texorpdfstring{$\mathfrak{r}_{3,1}$}{}}
	
	Since $\mathfrak{r}_{3,1}$ admits a root gradation and the unique homogeneous space in $\Lambda^3\mathfrak{r}_{3,1}$ is not related to the zero element of the group (see Table \ref{tabela3w}), Proposition \ref{Prop:Sym} shows that $(\Lambda^3\mathfrak{r}_{3,1})^{\mathfrak{r}_{3,1}}=\{0\}$. Moreover, the root gradation of $\mathfrak{r}_{3,1}$ yields $(\Lambda^2\mathfrak{r}_{3,1})^{\mathfrak{r}_{3,1}}\subset \Lambda^2(\mathfrak{r}_{3,1})^{(0)}$. It is then immediate to verify that $(\Lambda^2\mathfrak{r}_{3,1})^{\mathfrak{r}_{3,1}}=0$.  Therefore, one cannot use the reduced space to simplify the analysis and additionally, one has to solve the mCYBE directly.
	
	In the coordinates $\{x,y,z\}$  corresponding to the basis $\{e_{12},e_{13},e_{23}\}$ of $\Lambda^2\mathfrak{r}_{3,1}$, any $r\in \Lambda^2\mathfrak{r}_{3,1}$ can be written as $r = x e_{12} + y e_{13} + z e_{23}$. Simple computation shows that $[r,r] =0$. Hence, every element of $\Lambda^2\mathfrak{r}_{3,1}$ is an $r$-matrix. 

 As we have shown in Example \ref{Ex:r31}, the Lie algebra $\mathfrak{r}_{3,1}$ admits no $\mathfrak{r}_{3,1}$-invariant forms. Therefore, previously employed methods does not allow to analyse this case. Nevertheless, it is possible to overcome that difficulty by studying directly the action of ${\rm Inn}(\mathfrak{g})$, as it is relatively easy to determine this group.
	
	The fundamental vector fields of the action of ${\rm Inn}(\mathfrak{r}_{3,1})$ on $\Lambda^2\mathfrak{r}_{3,1}$ are spanned by
	$
	X_1:=x\partial_ x+y\partial_ y+2z\partial_z,\,\, X_2:=-y\partial_z,\,\, X_3:=x\partial_z$. Proposition \ref{proporb} implies that the distribution $\mathcal{D}$ spanned by these vector fields is singular, as its dimension varies depending on the value of $(x,y,z) \in \mathbb{R}^3$ (e.g. for $x=y=0$ and $z \neq 0$, it equals one, whereas apart from the $z$-axis it equals two). By Theorem \ref{Th:StSus}, this singular distribution is integrable. Let us determine its integral manifolds. Notice that in cylindrical coordinates $(r, \phi, z)$, the vector field $X_1$ reads $X_1 = r \partial_r + 2z \partial_z$. This observation helps to conclude that off the line $x = y = 0$, the integral manifolds of $\mathcal{D}$ are two-dimensional semi-planes $\mathcal{I}_{\phi} := \{(r \sin (\phi), r \cos (\phi), z) \in \mathbb{R}^3: r \in \mathbb{R}_{+}, z \in \mathbb{R}\}$ for $\phi \in [0, 2\pi[$, presented in Figure \ref{r31_inn}.

When restricted to the line $x=y=0$, the distribution reads $X_1 = 2z \partial_z, X_2 =0, X_3 = 0$. Then, one immediately obtains three orbits of the ${\rm Inn}(\mathfrak{r}_{3,1})$-action: $\{(x,y,z) \in \mathbb{R}^3: x= 
 0, y = 0, z >0 \}$, $\{(x,y,z) \in \mathbb{R}^3: x= 
 0, y = 0, z <0 \}$ and the origin $\{(0,0,0)\}$ (see Figure \ref{r31_inn}).

 Therefore, we obtain three classes of inequivalent $r$-matrices relative to ${\rm Inn}(\mathfrak{r}_{3,1})$, namely $r = 0$, $r^{(\pm)} = \pm e_{23}$ and $r_{\phi} = \cos (\phi) e_{12} + \sin (\phi) e_{13}$ with $\phi \in [0, 2\pi[$.
		
Let us study now the equivalence of $r$-matrices up to the action of ${\rm Aut}(\mathfrak{r}_{3,1})$. Obviously, the $r=0$ is an orbit of the action of ${\rm Aut}(\mathfrak{r}_{3,1})$ on $\Lambda^2\mathfrak{r}_{3,1}$. As follows from the final observations in Section \ref{R:3DSec:Aut}, elements of ${\rm Aut}(\mathfrak{r}_{3,1})$ leave the first derived ideal $[\mathfrak{r}_{3,1},\mathfrak{r}_{3,1}]=\langle e_2,e_3\rangle$ invariant. Then, the induced action of ${\rm Aut}(\mathfrak{r}_{3,1})$ on $\Lambda^2\mathfrak{r}_{3,1}$ must leave the subspace $\langle e_{23}\rangle$ invariant and every point within it must be contained in an orbit within $\langle e_{23}\rangle$. 
		
Finally, consider the Lie algebra automorphisms $T_{\alpha,\beta,\gamma,\delta}$ given by 
$$
	T_{\alpha,\beta,\gamma,\delta}(e_1):=e_1,\quad T_{\alpha,\beta,\gamma,\delta}(e_2):= \alpha e_2+\beta e_3,\quad
	T_{\alpha,\beta,\gamma,\delta}(e_3)=\gamma e_2+\delta e_3,\quad \alpha\delta-\beta\gamma\neq 0,
	$$
 which gives rise to the induced map $\Lambda^2T_{\alpha,\beta,\gamma,\delta}$ of the form
 $$
\Lambda^2T_{\alpha,\beta,\gamma,\delta}(e_12) = \alpha e_{12} + \beta e_{13}, \quad \Lambda^2T_{\alpha,\beta,\gamma,\delta}(e_{13}) = \gamma e_{12} + \delta e_{13}, \quad \Lambda^2T_{\alpha,\beta,\gamma,\delta}(e_{23}) = (\alpha \delta - \beta \gamma) e_{23}
 $$
Take $T_{\alpha,\beta,\gamma,\delta}$ with $\alpha = \cos(\psi)$, $\delta = \cos(\psi)$, $\gamma = \sin (\psi)$ and $\beta = -\sin (\psi)$ for any $\psi \in [0, 2\pi[$. Then,
\begin{equation*}
 \begin{split}
T_{\alpha,\beta,\gamma,\delta}(r_{\phi}) &= \cos(\phi) [\cos(\psi) e_{12} - \sin(\psi) e_{13}] + \sin(\phi) [\sin(\phi) e_{12} + \cos(\phi) e_{13}] \\
&= [\cos(\phi \cos(\psi) + \sin(\phi) \sin(\psi))] e_{12} + [\sin(\phi) \cos(\psi) - \cos(\phi) \sin(\psi)] e_{13} \\
&= \cos(\phi - \psi) e_{12} + \sin(\phi - \psi) e_{13} = r_{\phi - \psi}.
\end{split}
\end{equation*}
Since $\psi$ can be chosen arbitrarily, it follows that such $T_{\alpha,\beta,\gamma,\delta}$ identifies all semi-planes $\mathcal{I}_{\phi}$. Moreover, $T_{\alpha,\beta,\gamma,\delta}$ also connects the parts $z>0$ and $z<0$ of the line $x=y=0$. 

Hence, there exist three non-equivalent $r$-matrices for $\mathfrak{r}_{3,1}$, namely $r = 0$, $r_1=e_{13}$ and $r_2=e_{23}$. 
	
\noindent
\begin{minipage}{0.45\textwidth}
		\begin{center}
			\includegraphics[scale=0.35]{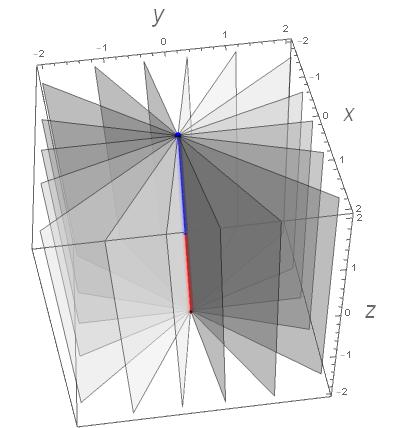}
			\captionof{figure}{Representative orbits of ${\rm Inn}(\mathfrak{r}_{3,1})$ acting on $\Lambda^2 \mathfrak{r}_{3,1}$.}\label{r31_inn}
		\end{center}
	\end{minipage}$\quad$
		\begin{minipage}{0.45\textwidth}
		\begin{center}
			\includegraphics[scale=0.35]{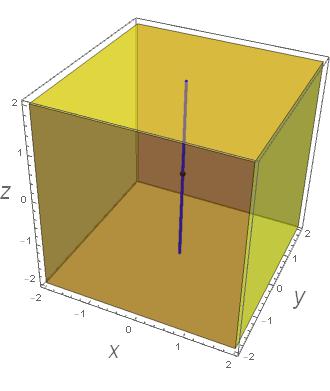}
			\captionof{figure}{Representative orbits of the action of ${\rm Aut}(\mathfrak{r}_{3,1})$ on $\Lambda^2 \mathfrak{r}_{3,1}$.}\label{r31_aut}
		\end{center}
	\end{minipage}$\quad$	
		
	\subsection{Lie algebra \texorpdfstring{$\mathfrak{r}_3$}{}}	
		
Since $\mathbb{Z}$-gradation of $\mathfrak{r}_3$ shown in Table \ref{tabela3w} is not a root gradation and one cannot use Proposition \ref{Prop:Sym} to analyse $(\Lambda^3 \mathfrak{r}_3)^{\mathfrak{r}_3}$, it is necessary to directly verify that $(\Lambda^3 \mathfrak{r}_3)^{\mathfrak{r}_3} = \{0\}$. Thus, the determination of $r$-matrices requires solving the mCYBE. 
	
		Let us determine $(\Lambda^2\mathfrak{r}_3)^{\mathfrak{r}_3}$ to 
	know whether different $r$-matrices induce different coboundary cocommutators. Since $\mathfrak{r}_3$ admits a $\mathbb{Z}$-gradation, $(\Lambda^2\mathfrak{r}_3)^{\mathfrak{r}_3}$ is the sum of the $\mathfrak{r}_3$-invariant elements on each homogeneous subspace of  $\Lambda^2\mathfrak{r}_3$ by Proposition \ref{LemgDe}. Using this gradation of $\mathfrak{r}_3$, one also sees that $[v^{(\alpha)},w^{(\beta)}] =0$, for $v^{(\alpha)}\subset \mathfrak{r}_3^{(\alpha)}, w^{(\beta)}\in (\Lambda^2\mathfrak{r}_3)^{(\beta)}$ when $\alpha+\beta\neq 2$. Inspecting remaining Schouten brackets, one obtains $(\Lambda^2\mathfrak{r}_3)^{\mathfrak{r}_3}=\{0\}$ and every $r$-matrix induces a different coboundary cocommutator. 
	
		Let $\{x,y,z\}$ be the coordinates on $\Lambda^2\mathfrak{r}_3$ induced by the basis $\{e_{12},e_{13},e_{23}\}$. The mCYBE, where $r=xe_{12}+ye_{13}+ze_{23}$, reads $[r,r] =-2z^2e_{123}$. Hence, $YB=\langle e_{12},e_{13}\rangle$ is the space of solutions to the mCYBE, presented in Figure \ref{r3_inn}. 
		
		A long but simple calculation shows that $\Lambda^2\mathfrak{r}_3$ does not admit a non-zero $\mathfrak{r}_3$-invariant metrics. As in the previous case, we need to carry on the analysis of the orbits of ${\rm Aut}(\mathfrak{r}_3)$-action directly. The fundamental vector fields of the action of ${\rm Inn}(\mathfrak{r}_3)$ on $\Lambda^2\mathfrak{r}_3$ are given by
	$
	X_1:=z\partial_x,\,\, X_2:=(-y+z)\partial_x,\,\, X_3:= 2x\partial_x + (y+z)\partial_y + z\partial_z.
	$  
	Since automorphisms preserve the set of mCYBE solutions, the above vector fields are tangent to $YB$ and their restrictions to $YB$ read $X_1|_{YB}=0,\,\, X_2|_{YB}=-y\partial_ x,\,\, X_3|_{YB}=2x\partial_ x+y\partial_y$. By Theorem \ref{Th:StSus}, the distribution spanned by $X_1, X_2, X_3$ is integrable. For $y \neq 0$, the integral manifolds are semi-planes $\mathcal{P}_{+} := \{(x,y,z) \in \mathbb{R}^3: x \in \mathbb{R}, y >0, z =0 \}$ and $\mathcal{P}_{-} := \{(x,y,z) \in \mathbb{R}^3: x \in \mathbb{R}, y <0, z =0 \}$. For $y = 0$, we get $\mathcal{L}_{+} := \{(x,y,z) \in \mathbb{R}^3: x> 0, y = 0, z = 0\}$ and $\mathcal{L}_{-} := \{(x,y,z) \in \mathbb{R}^3: x< 0, y = 0, z = 0\}$. Finally for $z=y=x=0$, we obtain the single-point orbit $\{(0,0,0)\}$. All mentioned orbits of ${\rm Inn}(\mathfrak{r}_3)$ are depicted in Figure \ref{r3_inn}. Representative $r$-matrices for these orbits read $r = 0$, $r_1^{\pm} = \pm e_{12}$ and $r_2^{\pm} = \pm e_{13}$.
 
 Let us classify coboundary cocommutators up to the action of elements of ${\rm Aut}(\mathfrak{r}_{3})$ on $YB$. By the remarks in Section \ref{R:3DSec:Aut}, it follows that the derived ideal $[\mathfrak{r}_3,\mathfrak{r}_3]=\langle e_1,e_2\rangle $ is invariant under the action of ${\rm Aut}(\mathfrak{r}_3)$. Thus, $\langle e_{12}\rangle$ is also invariant under the action of ${\rm Aut}(\mathfrak{r}_3)$ on $\Lambda^2\mathfrak{r}_3$. In consequence, the orbits $\mathcal{L}_{\pm}$ will not be identified with any of the semi-planes $\mathcal{P}_{\pm}$.
 
 Consider $T \in {\rm Aut}(\mathfrak{r}_3)$ given by $T_{\alpha}(e_1) = \alpha e_1$, $T_{\alpha} (e_2) = \alpha e_2$ and $T_{\alpha}(e_3) = e_3$. The induced action $\Lambda^2 T$ on $\Lambda^2 \mathfrak{r}_3$ reads $\Lambda^2 T (e_{12}) = \alpha^2 e_{12}$, $\Lambda^2 T (e_{13}) = \alpha e_{13}$ and $\Lambda^2 T (e_{23}) = \alpha e_{23}$. Consequently, $\Lambda^2 T$ identifies both $r_2^{\pm}$ for $\alpha = (-1)$.

 Finally, one needs to verify whether there exists an automorphism $T$ of $\mathfrak{r}_3$ that identifies both $r_1^{\pm}$. As none of our methods can address this issue, we will try to construct this map explicitly. Denote $T(e_i) = \sum_{j=1}^3 a_{ij} e_j$ for $i \in \{1,2,3\}$ and assume $\Lambda^2 T(e_{12}) = -e_{12}$. This implies the following conditions must hold:
 $$
 a_{11} a_{22} - a_{21} a_{12} = (-1), \quad a_{11} a_{23} = a_{21} a_{13}, \quad a_{12} a_{23} = a_{22} a_{13}
 $$
 If both $a_{13}, a_{23} \neq 0$, then it follows that $a_{11} a_{22} = a_{21} a_{12}$, contradicting the first equation. If $a_{13} = 0$ and $a_{23} \neq 0$, then $a_{11} = a_{12} = 0$, which again contradicts the first equation. Similarly, one gets a contradiction for $a_{23} = 0$ and $a_{13} \neq 0$. Thus, $a_{13} = a_{23} = 0$.

 Since the commutator $[e_1, e_3] = -e_1$ yields $[T(e_1), T(e_3)] = -T(e_1)$, we get
 $$
a_{11} a_{33} + a_{12} a_{33} = a_{11}, \quad a_{12} a_{33} = a_{12}
 $$
 If $a_{12} \neq 0$, then $a_33 = 1$ and the first equation gives $a_{12} = 0$, contrary to the assumption. Thus, $a_{12} = 0$ and we are left with $a_{11} a_{33} = a_{11}$. Since $a_{11} a_{22} - a_{21} a_{12} = (-1) = a_{11} a_{22}$, we conclude that $a_{11} \neq 0$ and $a_{33} = 1$.

 By similar reasoning, the commutator $[e_3, e_2] = e_1 + e_2$ yields the condition $a_{11} = a_{22}$. But then, $a_{11} a_{22} = (-1) = a_{11}^2$. As we work over the reals, this equation has no solution. In consequence, there exists no automorphism of $\mathfrak{r}_3$ such that $r_1^{\pm}$ can be identified.
 	
Hence, we have the $r$-matrices inequivalent relative to ${\rm Aut}(\mathfrak{r}_3)$ read
	$
	r_0 = 0, r_{\pm}=\pm e_{12}, r= e_{13},
	$ 
	as depicted in Figure \ref{r3_aut}. 

	\noindent\begin{minipage}{0.45\textwidth}
		\begin{center}
			\includegraphics[scale=0.4]{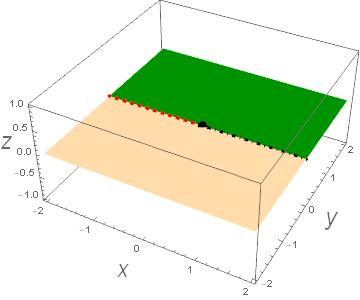}
			\captionof{figure}{\small Orbits of the action of ${\rm Inn}(\mathfrak{r}_3)$ on $YB\subset \Lambda^2 \mathfrak{r}_3$.}\label{r3_inn}
		\end{center}
	\end{minipage}
	$\quad$
		\begin{minipage}{0.45\textwidth}
		\begin{center}
			\includegraphics[scale=0.4]{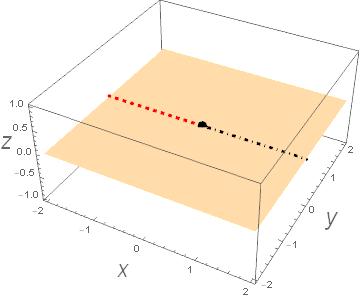}
			\captionof{figure}{\small Orbits of ${\rm Aut}(\mathfrak{r}_3)$ acting on $YB\subset \Lambda^2 \mathfrak{r}_3$.}\label{r3_aut}
		\end{center}
	\end{minipage}
	
	\subsection{Lie algebra \texorpdfstring{$\mathfrak{r}_{3, \lambda}$}{} (\texorpdfstring{$\lambda \in (-1,1)$}{})}
	
	In view of Proposition \ref{Prop:Sym} and the fact that  $\mathfrak{r}_{3,\lambda}$ admits a root gradation and the unique non-zero homogeneous space in $\Lambda^3\mathfrak{r}_{3,\lambda}$ has degree three, one has that  $(\Lambda^3\mathfrak{r}_3)^{\mathfrak{r}_{3,\lambda}}=\{0\}$. The root gradation of $\mathfrak{r}_{3,\lambda}$ also implies that $(\Lambda^2\mathfrak{r}_{3,\lambda})^{\mathfrak{r}_{3,\lambda}}\subset (\Lambda^2\mathfrak{r}_{3,\lambda})^{(0)}=\{0\}$.
 
 If we write an element $r \in \Lambda^2 \mathfrak{r}_{3, \lambda}$ as $r=xe_{12}+ye_{13}+ze_{23}$, then the mCYBE reads $[r,r] =2(\lambda-1)yz \,e_{123}$. Hence, the space of $r$-matrices, denoted by $YB$, consists of tow intersecting planes, one with  $y=0$ and the other with $z=0$. 
	
Similarly to the previous case, $\mathfrak{r}_{3, \lambda}$ does not admit $\mathfrak{r}_{3, \lambda}$-invariant forms. Thus, we will directly inspect the orbits of ${\rm Aut}(\mathfrak{r}_{3, \lambda})$-action on $YB \subset \Lambda^2 \mathfrak{r}_{3, \lambda}$. Let us first classify $r$-matrices up to the action of ${\rm Inn}(\mathfrak{r}_{3,\lambda})$. The fundamental vector fields of the action of  ${\rm Inn}(\mathfrak{r}_{3,\lambda})$ on $\Lambda^2\mathfrak{r}_{3,\lambda}$ are given by
		$
		Z_1 := z{\partial_x}, \,\, Z_2:=-\lambda y{\partial_x},\,\,
		Z_3 := (1+\lambda)x{\partial_x} +y{\partial_ y}+\lambda z{\partial_z}.
		$
We will analyse the distribution $\mathcal{D}$ spanned by $Z_1, Z_2, Z_3$ for three separate subsets of $YB$, namely $YB_1 = \{(x,y,z) \in \mathbb{R}^3: x \in \mathbb{R}, y = 0, z \neq 0\}$ , $YB_2 = \{(x,y,z) \in \mathbb{R}^3: x \in \mathbb{R}, y \neq 0, z = 0\}$ and $YB_3 = \{(x,y,z) \in \mathbb{R}^3: x \in \mathbb{R}, y = 0, z = 0\}$.

Assume $\lambda \neq 0$. On $YB_3$, one immediately notices that $\dim \mathcal{D}|_{YB_3} = 1$ for $x \neq 0$ and $\dim \mathcal{D}|_{YB_3} = 0$ otherwise. Hence, $YB_3$ is divided into three orbits of ${\rm Inn}(\mathfrak{g})$, namely $\{(x,0,0) \in \mathbb{R}^3: x >0\}$, $\{(x,0,0) \in \mathbb{R}^3: x <0\}$ and $\{(0,0,0)\}$. On $YB_1$, the distribution $\mathcal{D}$ is spanned by $Z_1 = z{\partial_x}, Z_3 = (1+\lambda)x{\partial_x} +\lambda z{\partial_z}$. By Proposition \ref{proporb}, one gets that the orbits of ${\rm Inn}(\mathfrak{r}_{3,\lambda})$-action are two-dimensional for $z \neq 0$. Thus, we obtain two separate orbits $\{(x, 0, z) \in \mathbb{R}^3: x \in \mathbb{R}, z > 0\}$ and $\{(x, 0, z) \in \mathbb{R}^3: x \in \mathbb{R}, z < 0\}$. On $YB_2$, the distribution $\mathcal{D}$ is spanned by $Z_2=-\lambda y{\partial_x}, Z_3 := (1+\lambda)x{\partial_x} +y{\partial_ y}$. By Proposition \ref{proporb}, the orbits of ${\rm Inn}(\mathfrak{r}_{3,\lambda})$-action are two-dimensional. Thus, we also obtain two orbits $\{(x, y, 0) \in \mathbb{R}^3: x \in \mathbb{R}, y > 0\}$ and $\{(x, y, 0) \in \mathbb{R}^3: x \in \mathbb{R}, y < 0\}$. 

Therefore, there are seven classes of $r$-matrices inequivalent relative to ${\rm Inn}(\mathfrak{r}_{3,\lambda})$, that is $r = 0$, $r_1^{\pm} = \pm e_{23}$, $r_2^{\pm} = \pm e_{13}$, and $r_3^{\pm} = \pm e_{12}$. 

Let us now classify these $r$-matrices up to the action of ${\rm Aut}(\mathfrak{r}_{3,\lambda})$. Since $[\mathfrak{r}_{3,\lambda},\mathfrak{r}_{3,\lambda}]=\langle e_1,e_2\rangle$ is invariant under ${\rm Aut}(\mathfrak{r}_{3,\lambda})$, the space $\langle e_{12}\rangle $ is also invariant relative to the action of ${\rm Aut}(\mathfrak{r}_{3,\lambda})$. Consider the automorphisms of the form
	$
	T_{\alpha,\beta}(e_1)= \beta e_1,T_{\alpha,\beta}(e_2)= \alpha e_2, T_{\alpha,\beta}(e_3)=e_3,$ for all  $\alpha\in \mathbb{R}\backslash\{0\}$.
The induced map $\Lambda^2 T_{\alpha,\beta}$ reads $\Lambda^2 T_{\alpha, \beta}(e_{12}) = \alpha \beta e_{12}$, $\Lambda^2 T_{\alpha, \beta}(e_{13}) = \beta e_{13}$, $\Lambda^2 T_{\alpha, \beta}(e_{23}) = \alpha e_{23}$. Thus, we conclude that $T_{1, -1}$ identifies both $r_3^{\pm}$. Similarly, $T_{1, -1}$ identifies $r_2^{\pm}$ and $T_{-1,1}$ connects $r_1^{\pm}$. Hence, we obtain four classes of $r$-matrices inequivalent relative to ${\rm Aut}(\mathfrak{r}_{3,\lambda})$, namely $r = 0$, $r_1 = e_{23}$, $r_2 = e_{13}$ and $r_3 = e_{12}$. This is depicted in Figure \ref{r3l_aut}.
		
	Let us now tackle the case $\lambda=0$, namely $\mathfrak{r}_{3,0}$. The analysis of solutions to the mCYBE goes as in the previous case. The fundamental vector fields of the action of  ${\rm Inn}(\mathfrak{r}_{3,0})$ read
	$
	Z_1 := z{\partial_x},  Z_2:=0, Z_3 := x{\partial_x} +y{\partial_y}.
	$
	On $YB_3$, the distribution spanned by $Z_1,Z_2,Z_3$ has rank one for $x\neq 0$ and zero for $x=0$. Therefore, we obtain three orbits: $\{(x, 0, 0) \in \mathbb{R}^3: x > 0\}$, $\{(x, 0, 0) \in \mathbb{R}^3: x < 0\}$, and $\{(0, 0, 0)\}$. In consequence, we get three classes of $r$-matrices: $r = 0$ and $r_1^{\pm} = \pm e_{12}$. 
	
	Restricting the fundamental fields to $YB_1$, we get
	$
	Z_1\vert_{YB_1}=z\partial _x, Z_2\vert_{YB_1}=0, Z_3\vert_{YB_1}=x{\partial _x},
	$
	which span $\langle \partial_x \rangle$. Thus, the orbits of the action of ${\rm Inn}(\mathfrak{r}_{3,0})$ on this space are the lines $\mathcal{L}_{z_0} = \{(x,0,z_0) \in \mathbb{R}^3: x \in \mathbb{R}\}$ with a constant  $z_0\neq 0$. It gives the family of $r$-matrices of the form $r_{a} = a e_{23}$ with $a \in \mathbb{R} \backslash \{0\}$.
 
 Finally, the restriction of $Z_1,Z_2,Z_3$ to $YB_2$ gives
	$
	Z_2\vert_{YB_2} = y\partial_x, Z_3\vert_{YB_2}=x{\partial _x}+y{\partial_y}.
	$ Since $y \neq 0$, Proposition \ref{proporb} yields that the orbits of ${\rm Inn}(\mathfrak{r}_{3,0})$-action are two-dimensional. Thus, there are two such orbits, i.e. $\{(x,y,0) \in \mathbb{R}^3: x \in \mathbb{R}, y > 0 \}$ and $\{(x,y,0) \in \mathbb{R}^3: x in \mathbb{R}, y < 0 \}$. The representative $r$-matrices of these orbits are $r_2^{\pm} = \pm e_{13}$.
	
Consider the maps $T_{\alpha,\beta,\gamma}\in {\rm Aut}(\mathfrak{r}_{3,0})$ of the form $T_{\alpha,\beta,\gamma}(e_1) := \alpha e_1+\gamma e_2, T_{\alpha,\beta,\gamma}(e_2) := \beta e_2,T_{\alpha,\beta,\gamma}(e_3) := e_3, $ with $\alpha,\beta\in \mathbb{R}\backslash \{0\}$ and $\gamma\in \mathbb{R}$. The induced map $\Lambda^2 T_{\alpha,\beta,\gamma}$ read $\Lambda^2 T_{\alpha,\beta,\gamma}(e_{12}) = \alpha \beta e_{12}$, $\Lambda^2 T_{\alpha,\beta,\gamma}(e_{13}) = \alpha e_{13}$ and $\Lambda^2 T_{\alpha,\beta,\gamma}(e_{23}) = \beta e_{23}$. It is easy to verify that $T_{-1,1,1}$ identifies both elements in $r_1^{\pm}$ and $r_2^{\pm}$. Moreover, the maps $T_{1,\beta,1}$ for $\beta \in \mathbb{R} \backslash \{0\}$ identify all the elements in the family $r_a$. Hence, we get the following classes of inequivalent $r$-matrices: $r = 0$, $r_1 = e_{12}$, $r_2 = e_{13}$ and $r_z = e_{23}$, as shown in Figure \ref{r3l_aut}.
		
		\noindent\begin{minipage}{0.45\textwidth}
		\begin{center}
			\includegraphics[scale=0.35]{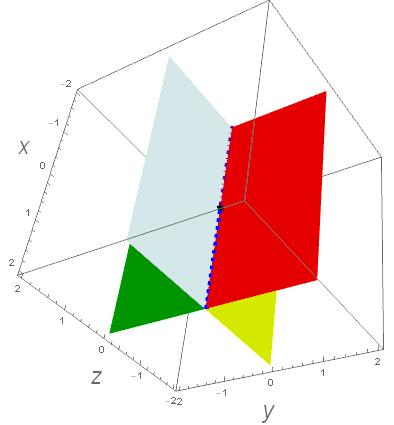}
			\captionof{figure}{Orbits of the action ${\rm Inn}(\mathfrak{r}_{3,\lambda})$ on $YB\subset \Lambda^2 \mathfrak{r}_{3, \lambda}$.}\label{r3l_inn}
		\end{center}
	\end{minipage}$\quad$
		\begin{minipage}{0.45\textwidth}
		\begin{center}
			\includegraphics[scale=0.35]{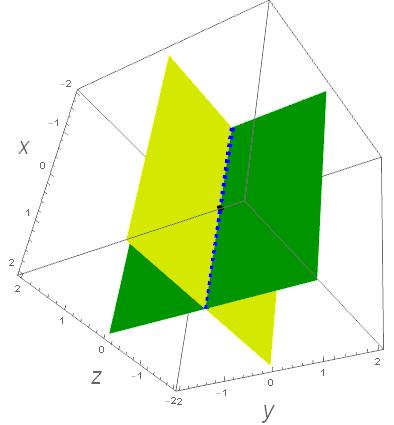}
			\captionof{figure}{\small Orbits of the action of ${\rm Aut}(\mathfrak{r}_{3,\lambda})$ on $YB\subset \Lambda^2 \mathfrak{r}_{3,\lambda}$.}\label{r3l_aut}
		\end{center}
	\end{minipage}

	\subsection{Lie algebra \texorpdfstring{$\mathfrak{r}'_{3, \lambda} (\lambda > 0)$}{}}
	
	It stems from Table \ref{tabela3w} that the unique non-zero homogeneous space in $\Lambda^3{\mathfrak{r}'_{3, \lambda}}$ is not invariant relative to the action of the basis element $e_3$ and hence $(\Lambda^3{\mathfrak{r}'_{3, \lambda}})^{\mathfrak{r}'_{3, \lambda}}=0$.   Due to Proposition \ref{LemgDe}, the space $(\Lambda^2{\mathfrak{r}'_{3, \lambda}})^{\mathfrak{r}'_{3, \lambda}}$ can be easily determined by inspecting elements within each homogeneous subspace in $\Lambda^2{\mathfrak{r}'_{3, \lambda}}$. Using gradation in Table \ref{tabela3w}, one verifies that $(\Lambda^2\mathfrak{r}'_{3,\lambda})^{\mathfrak{r}'_{3,\lambda}}=0$, since $\lambda\neq 0$.

The corresponding mCYBE read $[r,r] =-2(y^2+z^2)e_{123}$, where $r \in \Lambda^2 \mathfrak{r}'_{3, \lambda}$ is written in the basis $\{e_{12}, e_{13}, e_{23}\}$ of $\Lambda^2 \mathfrak{r}'_{3, \lambda}$ as $r = x e_{12} + y e_{13} + z e_{23}$. Thus, the space of mCYBE solutions is $YB:=\{(x,y,z)\in \Lambda^2{\mathfrak{r}'_{3, \lambda}}: y=z=0\}$. 
 
	\noindent
		\begin{minipage}{0.4\textwidth}
		\begin{center}
			\includegraphics[scale=0.4]{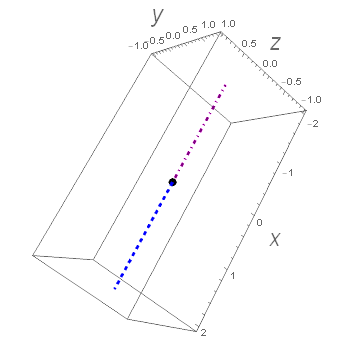}
			\captionof{figure}{\small Orbits of the action of ${\rm Aut}(\mathfrak{r}'_{3,\lambda})$ on $YB\subset \Lambda^2 \mathfrak{r}'_{3,\lambda}$.}\label{r3'l_aut}
		\end{center}
	\end{minipage}$\quad$
	\begin{minipage}{0.55\textwidth}
 \setlength{\baselineskip}{14pt}
 By direct calculation one obtains that there are no $\mathfrak{r}'_{3,\lambda}$-invariant forms on $\mathfrak{r}'_{3,\lambda}$. Then, we resort to the direct analysis of the orbits of ${\rm Aut}(\mathfrak{r}'_{3,\lambda})$. The fundamental vector fields of the ${\rm Inn}(\mathfrak{r}'_{3,\lambda})$-action on $\Lambda^2 \mathfrak{r}'_{3, \lambda}$ are spanned by $Z_1 = (y + \lambda z)\partial_x$, $Z_2 = (z - \lambda y)\partial_x$ and $Z_3 = 2\lambda x \partial_x + (\lambda y + z)\partial_y + (\lambda z - y)\partial_z$. On $YB$, these vector fields reduce to a single nonzero $Z_3\vert_{YB} = 2\lambda x \partial_x$. Thus, we immediately conclude that there are three orbits of ${\rm Inn}(\mathfrak{r}'_{3,\lambda})$-action on $YB$, namely $\{(x, 0, 0) \in \mathbb{R}^3: x > 0\}$, $\{(x, 0, 0) \in \mathbb{R}^3: x < 0\}$ and $\{(0,0,0)\}$. Thus, we obtain two classes of $r$-matrices inequivalent relative to ${\rm Inn}(\mathfrak{r}'_{3,\lambda})$, namely $r = 0$ and $r_{\pm} = \pm e_{12}$. Finally, let us attempt to construct the automorphism $T \in {\rm Aut}(\mathfrak{r}'_{3,\lambda})$ that would connect $r_{\pm}$. Let us assume the induced map $\Lambda^2 T$ satisfies $\Lambda^2 T(e_{12}) = -e_{12}$.
 	\end{minipage}
 
Denote $T(e_i) = \sum_{j=1}^3 a_{ij} e_j$ for $i \in \{1,2,3\}$. Then, the condition $\Lambda^2 T(e_{12}) = -e_{12}$ yields $a_{13} = a_{23} = 0$. One can show by a long computation that $[T(e_1), T(e_3)] = T([e_1, e_3])$ and $[T(e_3), T(e_2)] = T([e_3, e_2])$ result in the following set of equations:
\begin{equation*}
\begin{split}
&(a_{33} - 1)(a_{11} + a_{22}) = 0, \qquad (a_{33} - 1)(a_{21} - a_{12}) = 0, \\
&(1 + \lambda^2) a_{11} a_{33} = \lambda^2 a_{11} - \lambda a_{12} - \lambda a_{21} + a_{22}, \\
&(1 + \lambda^2) a_{12} a_{33} = \lambda a_{11} - a_{12} - \lambda^2 a_{21} - \lambda a_{22}.
\end{split}
\end{equation*}
If $a_{33} \neq 1$, then $a_{11} = -a_{22}$, $a_{12} = a_{21}$ and  we get
\begin{equation*}
(1 + \lambda^2) a_{11} a_{33} = a_{11} (\lambda^1 - 1) - 2\lambda a_{12}, \qquad
(1 + \lambda^2) a_{12} a_{33} = 2\lambda a_{11} + a_{12} (\lambda^2 - 1).
\end{equation*}
It follows that $a_{12} [(1 + \lambda^2) a_{33} - (\lambda^2 - 1)]^2 = -4\lambda^2 a_{12}$. Since we work over reals, $[(1 + \lambda^2) a_{33} - (\lambda^2 - 1)]^2 = -4\lambda^2$ has no solutions for $a_{33}$ and $a_{12} = 0$. In consequence, the third equation gives $a_{11} = 0$, contradicting that $a_{11} a_{22} - a_{12} a_{21} = (-1)$.

Similarly, assume $a_{33} = 1$. Then, we have
\begin{equation*}
(1 + \lambda^2) a_{11} = a_{11} (\lambda^1 - 1) - 2\lambda a_{12}, \qquad
(1 + \lambda^2) a_{12} = 2\lambda a_{11} + a_{12} (\lambda^2 - 1).
\end{equation*}
These equations can be rewritten as $a_{11} = -\lambda a_{12}$ and $a_{12} = \lambda a_{11}$. Then, $a_{11} = -\lambda^2 a_{11}$ and $a_{11} = 0$ as $\lambda > 0$. Consequently, $a_{12} = 0$ and this again contradicts the condition $a_{11} a_{22} - a_{12} a_{21} = (-1)$.

Therefore, there are no automorphisms of $\mathfrak{r}'_{3,\lambda}$ such that $r_{\pm}$ can be identified. Hence, we get the following classes of $r$-matrices inequivalent relative to ${\rm Aut}(\mathfrak{r}'_{3,\lambda})$: $r = 0$ and $r_{\pm} = \pm e_{12}$, illustrated in Figure \ref{r3'l_aut}.

		\begin{landscape}
		\begin{table}
		{\small
			\centering
			\begin{tabular}{|c|c|c|c|c|c|c|c|c|c|}
				\hline
				& $[e_1,e_2]$ &$[e_1,e_3]$&$[e_3,e_2]$ & $\mathcal{S}_k$& $G$ & $\mathfrak{g}$ &$\Lambda^2 \mathfrak{g}$&$\Lambda^3 \mathfrak{g}$ &Root\\
				\hline 
				$\mathfrak{sl}_2$& $e_2$ &$-e_3,$&$ -e_1$& \parbox[c]{4.6cm}{\vspace{3pt}$2xy-z^2 = k$, $k \in \mathbb{R}_+$\\$2xy-z^2=0, z\neq 0$, $k=0$\\$2xy-z^2 = k$, $k \in \mathbb{R}_-$\vspace{3pt}}&
				$\mathbb{Z}$&
				\parbox[c]{2.9cm}{
					\centering
					\begin{tikzpicture}[scale = 0.5]
					{
						\filldraw (0,0) circle (1pt) node[above]{$e_1$} node[below]{$(0)$};
						\draw [-] (0,0)--(2,0) node[above]{$e_2$};
						\filldraw (2,0) circle (1pt) node[below]{$(1)$};
						\draw [-] (0,0)--(-2,0) node[above]{$e_3$};
						\filldraw (-2,0) circle (1pt) node[below]{$(-1)$};
					}
					\end{tikzpicture}
				}
				&
				\parbox[c]{2.9cm}{
					\centering
					\begin{tikzpicture}[scale = 0.5]
					{
						\filldraw (0,0) circle (1pt) node[above]{$e_{23}$} node[below]{$(0)$};
						\draw [-] (0,0)--(2,0) node[above]{$\mathbf{e_{12}}$};
						\filldraw (2,0) circle (1pt) node[below]{$(1)$};
						\draw [-] (0,0)--(-2,0) node[above]{$\mathbf{e_{13}}$};
						\filldraw (-2,0) circle (1pt) node[below]{$(-1)$};
					}
					\end{tikzpicture}
				}
				&
				\parbox[c]{1.5cm}{
					\centering
					\begin{tikzpicture}[scale = 0.5]
					{
						\filldraw (0,0) circle (1pt) node[above]{$e_{123}$} node[below]{$(0)$};
						\draw [-] (0,0)--(1,0);
						\draw [-] (0,0)--(-1,0);
					}
					\end{tikzpicture}
				}&Yes
				\\\hline
				$\mathfrak{su}_2$ &
				$e_3$&$-e_2$&$-e_1$& \parbox[c]{4.6cm}{\vspace{3pt}$x^2 + y^2 + z^2 = k$, $k \in \mathbb{R}_{+}$\vspace{3pt}}&
				$\mathbb{Z}_2$&
				\parbox[c]{2.9cm}{
					\centering
					\begin{tikzpicture}[scale = 0.5]
					{
						\filldraw (0,0) circle (1pt) node[above]{$e_a$} node[below]{$(0)$};
						\draw [-] (0,0)--(3,0) node[above]{$e_b, e_c$};
						\filldraw (3,0) circle (1pt) node[below]{$(1)$};
					}
					\end{tikzpicture}
				}
				&
				\parbox[c]{2.9cm}{
					\centering
					\begin{tikzpicture}[scale = 0.5]
					{
						\filldraw (0,0) circle (1pt) node[above]{$e_{bc}$} node[below]{$(0)$};
						\draw [-] (0,0)--(3,0) node[above]{${ e_{ba}, e_{ac}}$};
						\filldraw (3,0) circle (1pt) node[below]{$(1)$};
					}
					\end{tikzpicture}
				}
				&
				\parbox[c]{1.5cm}{
					\centering
					\begin{tikzpicture}[scale = 0.5]
					{
						\filldraw (0,0) circle (1pt) node[above]{$e_{abc}$} node[below]{$(0)$};
						\draw [-] (0,0)--(1,0);
						\draw [-] (0,0)--(-1,0);
					}
					\end{tikzpicture}
				}&No
				\\ \hline
				$\mathfrak{h}$&
				$e_3$& $0$ &$0$& \parbox[c]{4.6cm}{\vspace{3pt}$z\neq 0$, $k = 1$\vspace{3pt}} &$\mathbb{Z}$&
				\parbox[c]{2.9cm}{
					\centering
					\begin{tikzpicture}[scale = 0.5]
					{
						\filldraw (0,0) circle (1pt) node[above]{$e_{2}$} node[below]{$(2)$};
						\draw [-] (0,0)--(2,0) node[above]{$e_{3}$};
						\filldraw (2,0) circle (1pt) node[below]{$(3)$};
						\draw [-] (0,0)--(-2,0) node[above]{$e_{1}$};
						\filldraw (-2,0) circle (1pt) node[below]{$(1)$};
					}
					\end{tikzpicture}
				}
				&
				\parbox[c]{2.9cm}{
					\centering
					\begin{tikzpicture}[scale = 0.5]
					{
						\filldraw (0,0) circle (1pt) node[above]{$\mathbf{e_{13}}$} node[below]{$(4)$};
						\draw [-] (0,0)--(2,0) node[above]{$\mathbf{e_{23}}$};
						\filldraw (2,0) circle (1pt) node[below]{$(5)$};
						\draw [-] (0,0)--(-2,0) node[above]{$e_{12}$};
						\filldraw (-2,0) circle (1pt) node[below]{$(3)$};
					}
					\end{tikzpicture}
				}
				&
				\parbox[c]{1.5cm}{
					\centering
					\begin{tikzpicture}[scale = 0.5]
					{
						\filldraw (0,0) circle (1pt) node[above]{$e_{123}$} node[below]{$(6)$};
						\draw [-] (0,0)--(1,0);
						\draw [-] (0,0)--(-1,0);
					}
					\end{tikzpicture}
				}&No
				\\ \hline
				$\mathfrak{r}'_{3,0}$ &$
				-e_3$&$e_2$&$0$& \parbox[c]{4.6cm}{\vspace{3pt}$x^2 + y^2 > 0$, $k = 1$\vspace{3pt}}&
				$\mathbb{Z}$&
				\parbox[c]{2.9cm}{
					\centering
					\begin{tikzpicture}[scale = 0.5]
					{
						\filldraw (0,0) circle (1pt) node[above]{$e_{1}$} node[below]{$(0)$};
						\draw [-] (0,0)--(3,0);
						\filldraw (3,0) circle (1pt) node[above]{$e_{2}, e_3$} node[below]{$(1)$};
					}
					\end{tikzpicture}
				}
				&
				\parbox[c]{2.9cm}{
					\centering
					\begin{tikzpicture}[scale = 0.5]
					{
						\filldraw (0,0) circle (1pt) node[above]{$e_{12}, e_{13}$} node[below]{$(1)$};
						\draw [-] (0,0)--(3,0);
						\filldraw (3,0) circle (1pt) node[above]{$\mathbf{e_{23}}$} node[below]{$(2)$};
					}
					\end{tikzpicture}
				}
				&
				\parbox[c]{1.5cm}{
					\centering
					\begin{tikzpicture}[scale = 0.5]
					{
						\filldraw (0,0) circle (1pt) node[above]{$e_{123}$} node[below]{$(2)$};
						\draw [-] (0,0)--(1,0);
						\draw [-] (0,0)--(-1,0);
					}
					\end{tikzpicture}
				}&No
				\\\hline $\mathfrak{r}_{3,-1}$&$e_2$ &$-e_3$ &$0$& \parbox[c]{4.6cm}{\vspace{3pt}$xy=0$, $x^2+y^2\neq 0$, $k = 1$, \\ $xy \neq 0 \neq 0$, $k = 2$\vspace{3pt}}&
				$\mathbb{Z}$&
				\parbox[c]{2.9cm}{
					\centering
					\begin{tikzpicture}[scale = 0.5]
					{
						\draw [-] (0,0)--(-2,0);
						\filldraw (-2,0) circle (1pt) node[above]{$e_{3}$} node[below]{$(-1)$};
						\filldraw (0,0) circle (1pt) node[above]{$e_{1}$} node[below]{$(0)$};
						\draw [-] (0,0)--(2,0);
						\filldraw (2,0) circle (1pt) node[above]{$e_{2}$} node[below]{$(1)$};
					}
					\end{tikzpicture}
				}
				&
				\parbox[c]{2.9cm}{
					\centering
					\begin{tikzpicture}[scale = 0.5]
					{
						\draw [-] (0,0)--(-2,0);
						\filldraw (-2,0) circle (1pt) node[above]{$\mathbf{e_{13}}$} node[below]{$(-1)$};
						\filldraw (0,0) circle (1pt) node[above]{$e_{23}$} node[below]{$(0)$};
						\draw [-] (0,0)--(2,0);
						\filldraw (2,0) circle (1pt) node[above]{$\mathbf{e_{12}}$} node[below]{$(1)$};
					}
					\end{tikzpicture}
				}
				&
				\parbox[c]{1.5cm}{
					\centering
					\begin{tikzpicture}[scale = 0.5]
					{
						\filldraw (0,0) circle (1pt) node[above]{$e_{123}$} node[below]{$(0)$};
						\draw [-] (0,0)--(1,0);
						\draw [-] (0,0)--(-1,0);
					}
					\end{tikzpicture}
				}&Yes
				\\\hline
				\multirow{2}{*}{\vspace{-0.7cm}$\mathfrak{r}_{3,1}$}&\multirow{2}{*}{\vspace{-0.7cm}$e_2$} &\multirow{2}{*}{\vspace{-0.7cm}$e_3$} &\multirow{2}{*}{\vspace{-0.7cm}$0$}&\multirow{2}{*}{\vspace{-0.7cm}\parbox[c]{4.6cm}{\vspace{3pt}$x^2+y^2\neq 0, z \in \mathbb{R}$, $k = 1$,\\ $x = y = 0, z \neq 0$, $k = 2$\vspace{3pt}}}&
				$\mathbb{Z}$&
				\parbox[c]{2.9cm}{
					\centering
					\begin{tikzpicture}[scale = 0.5]
					{
						\filldraw (0,0) circle (1pt) node[above]{$e_{1}$} node[below]{$(0)$};
						\draw [-] (0,0)--(3,0);
						\filldraw (3,0) circle (1pt) node[above]{$e_{2}, e_{3}$} node[below]{$(1)$};
					}
					\end{tikzpicture}
				}
				&
				\parbox[c]{2.9cm}{
					\centering
					\begin{tikzpicture}[scale = 0.5]
					{
						\filldraw (0,0) circle (1pt) node[above]{$e_{12}, e_{13}$} node[below]{$(1)$};
						\draw [-] (0,0)--(3,0);
						\filldraw (3,0) circle (1pt) node[above]{$\mathbf{e_{23}}$} node[below]{$(2)$};
					}
					\end{tikzpicture}
				}
				&
				\parbox[c]{1.5cm}{
					\centering
					\begin{tikzpicture}[scale = 0.5]
					{
						\filldraw (0,0) circle (1pt) node[above]{$e_{123}$} node[below]{$(2)$};
						\draw [-] (0,0)--(1,0);
						\draw [-] (0,0)--(-1,0);
					}
					\end{tikzpicture}
				}&Yes
				\\\cline{6-10}
				& & &&&
				$\mathbb{Z}^2$&
				\parbox[c]{2.9cm}{
					\centering
					\begin{tikzpicture}[scale = 0.5]
					{
						\filldraw (0,0) circle (1pt) node[left]{$e_{1}$} node[below]{$(0,0)$};
						\draw [-] (0,0)--(3,0);
						\draw [-] (0,0)--(0,1);
						\filldraw (3,0) circle (1pt) node[right]{$e_{2}$} node[below]{$(1,0)$};
						\filldraw (0,1) circle (1pt) node[right]{$e_{3}$} node[left]{$(0,1)$};
					}
					\end{tikzpicture}
				}
				&
				\parbox[c]{2.9cm}{
					\centering
					\begin{tikzpicture}[scale = 0.5]
					{
						\filldraw (3,0) circle (1pt) node[right]{${\bf e_{12}}$} node[left]{$(1,0)$};
						\filldraw (0,1) circle (1pt) node[above]{${\bf e_{13}}$} node[below]{$(0,1)$};
						\draw [-] (0,1)--(3,1);
						\draw [-] (3,0)--(3,1);
						\filldraw (3,1) circle (1pt) node[above]{${\bf e_{23}}$} node[right]{$(1,1)$};
					}
					\end{tikzpicture}
				}
				&
				\parbox[c]{1.5cm}{
					\centering
					\begin{tikzpicture}[scale = 0.5]
					{
						\filldraw (0,0) circle (1pt) node[above]{$e_{123}$} node[below]{$(1,1)$};
						\draw [-] (0,0)--(1,0);
						\draw [-] (0,0)--(-1,0);
					}
					\end{tikzpicture}
				}&No
				\\\hline
				$\mathfrak{r}_{3}$&
				$0$ &$-e_1$ &$e_1+e_2$&\parbox[c]{4.6cm}{\vspace{3pt}$x>0, y=0, z = 0$ $k = 1$,\\ $x < 0, y = 0, z = 0$, $k = 2$, \\ $ y \neq  0, z = 0$, $k = 3$\vspace{3pt}}&
				$\mathbb{Z}$&
				\parbox[c]{2.9cm}{
					\centering
					\begin{tikzpicture}[scale = 0.5]
					{
						\filldraw (0,0) circle (1pt) node[above]{$e_{3}$} node[below]{$(0)$};
						\draw [-] (0,0)--(3,0);
						\filldraw (3,0) circle (1pt) node[above]{$e_{1}, e_{2}$} node[below]{$(1)$};
					}
					\end{tikzpicture}
				}
				&
				\parbox[c]{2.9cm}{
					\centering
					\begin{tikzpicture}[scale = 0.5]
					{
						\filldraw (0,0) circle (1pt) node[above]{$e_{13}, e_{23}$} node[below]{$(1)$};
						\draw [-] (0,0)--(3,0);
						\filldraw (3,0) circle (1pt) node[above]{$\mathbf{e_{12}}$} node[below]{$(2)$};
					}
					\end{tikzpicture}
				}
				&
				\parbox[c]{1.5cm}{
					\centering
					\begin{tikzpicture}[scale = 0.5]
					{
						\filldraw (0,0) circle (1pt) node[above]{$e_{123}$} node[below]{$(2)$};
						\draw [-] (0,0)--(1,0);
						\draw [-] (0,0)--(-1,0);
					}
					\end{tikzpicture}
				}&No
				\\\hline
				$\mathfrak{r}_{3,\lambda \in (-1,1)}$ &
				$0$ &$-e_1$ &$\lambda e_2$ &\parbox[c]{4.6cm}{\vspace{3pt}$y =0, z \neq 0, x \in \mathbb{R}$, $k = 1$, \\ $z = 0, y \neq 0, x \in \mathbb{R}$, $k = 2$, \\ $y = 0, z = 0, x\in \mathbb{R}$, $k = 3$\vspace{3pt}}&
				$\mathbb{R}$&
				\parbox[c]{2.9cm}{
					\centering
					\begin{tikzpicture}[scale = 0.5]
					{
						\draw [-] (0,0)--(-2,0);
						\filldraw (-2,0) circle (1pt) node[above]{$e_{3}$} node[below]{$(0)$};
						\filldraw (0,0) circle (1pt) node[above]{$e_{1}$} node[below]{$(1)$};
						\draw [-] (0,0)--(2,0);
						\filldraw (2,0) circle (1pt) node[above]{$e_{2}$} node[below]{$(\lambda)$};
					}
					\end{tikzpicture}
				}
				&
				\parbox[c]{2.9cm}{
					\centering
					\begin{tikzpicture}[scale = 0.5]
					{
						\draw [-] (0,0)--(-2,0);
						\filldraw (-2,0) circle (1pt) node[above]{${\bf e_{13}}$} node[below]{$(1)$};
						\filldraw (0,0) circle (1pt) node[above]{${\bf e_{23}}$} node[below]{$(\lambda)$};
						\draw [-] (0,0)--(2,0);
						\filldraw (2,0) circle (1pt) node[above]{${\bf e_{12}}$} node[below]{$(1+\lambda)$};
					}
					\end{tikzpicture}
				}
				&
				\parbox[c]{1.5cm}{
					\centering
					\begin{tikzpicture}[scale = 0.5]
					{
						\draw [-] (0,0)--(-1,0);
						\filldraw (0,0) circle (1pt) node[above]{$e_{123}$} node[below]{$(1+\lambda)$};
						\draw [-] (0,0)--(1,0);
					}
					\end{tikzpicture}
				}&Yes
				\\\hline
				$\mathfrak{r}'_{3,\lambda > 0}$&
				$0$ &$e_2-\lambda e_1$ &$\lambda e_2+e_1$&\parbox[c]{4.6cm}{\vspace{3pt}$x > 0, y = 0, z = 0$, $k = 1$, \\ $x < 0, y = 0, z = 0$, $k = 2$\vspace{3pt}}&
				$\mathbb{Z}$&
				\parbox[c]{2.9cm}{
					\centering
					\begin{tikzpicture}[scale = 0.5]
					{
						\filldraw (0,0) circle (1pt) node[above]{$e_{3}$} node[below]{$(0)$};
						\draw [-] (0,0)--(3,0);
						\filldraw (3,0) circle (1pt) node[above]{$e_{1}, e_{2}$} node[below]{$(1)$};
					}
					\end{tikzpicture}
				}
				&
				\parbox[c]{2.9cm}{
					\centering
					\begin{tikzpicture}[scale = 0.5]
					{
						\filldraw (0,0) circle (1pt) node[above]{$e_{13}, e_{23}$} node[below]{$(1)$};
						\draw [-] (0,0)--(3,0);
						\filldraw (3,0) circle (1pt) node[above]{$\mathbf{e_{12}}$} node[below]{$(2)$};
					}
					\end{tikzpicture}
				}
				&
				\parbox[c]{1.5cm}{
					\centering
					\begin{tikzpicture}[scale = 0.5]
					{
						\filldraw (0,0) circle (1pt) node[above]{$ e_{123}$} node[below]{$(2)$};
						\draw [-] (0,0)--(1,0);
						\draw [-] (0,0)--(-1,0);
					}
					\end{tikzpicture}
				}&No
				\\\hline
			\end{tabular}
			\caption{Commutation relations, non-zero orbits of equivalent $r$-matrices, and $G$-gradations of $\mathfrak{g}$ with $\dim\mathfrak{g}=3$ along with decompositions on their Grassmann algebras. The letters $a,b,c$ stand for arbitrary different values within $\{1,2,3\}$. As throughout the work, $e_{i_1\ldots i_r}:=e_{i_1}\wedge\ldots\wedge e_{i_r}$, with $i_1,\ldots,i_r\in \overline{1,r}$, and $\{x,y,z\}$ is the dual basis to $\{e_{12},e_{13},e_{23}\}$. Solutions of CYBEs obtained via gradations are written in bold. The spaces $\mathcal{S}_k$ are the equivalence classes of reduced $r$-matrices.
			}\label{tabela3w}
			}
	\end{table}
	\end{landscape}
 
	\section{A four-dimensional case}\label{Ch:alg_Sec:4D}
In the previous section, we used the methods introduced at the beginning of this chapter to classify coboundary Lie bialgebra structure on three-dimensional Lie algebras. As it is the simplest nontrivial case that has been thoroughly studied \cite{FJ15, Go00}, it mostly served as a benchmark for the introduced formalism. In order to illustrate that our techniques facilitate the analysis of higher-dimensional examples, let us classify coboundary Lie bialgebra structures on $\mathfrak{gl}_2$. 
	
 First, note that $\mathfrak{gl}_2=\mathfrak{sl}_2\oplus \mathbb{R}$. Denote a basis of $\mathfrak{gl}_2$ by $\{e_1,\ldots,e_4\}$, where $e_1,e_2,e_3$ satisfy the conditions for $\mathfrak{sl}_2$ in Table \ref{tabela3w} and $\langle e_4\rangle $ is the centre of $\mathfrak{gl}_2$. One has the following gradations for $\mathfrak{gl}_2$ and the associated Grassmann spaces $\Lambda^2 \mathfrak{gl}_2$ and $\Lambda^3 \mathfrak{gl}_2$:

	\begin{center}
	\begin{tikzpicture}[scale = 0.5]
		{
			\filldraw (0,0) circle (1pt) node[above left]{$e_1$} node[below]{$(0,0)$};
			\draw [-] (0,0)--(3,0)
			node[above]{$e_2$};
			\filldraw (3,0) circle (1pt) node[below]{$(1,0)$};
			\draw [-] (0,0)--(-3,0) node[above]{$e_3$};
			\filldraw (-3,0) circle (1pt) node[below]{$(-1,0)$};
			\filldraw (0,2) circle (1pt) node[right]{$e_4$};
			\draw [-] (0,0)--(0,2) node[left]{$(0,1)$};
		}
		\end{tikzpicture}$\qquad$			\begin{tikzpicture}[scale = 0.5]
		{
			\filldraw (0,0) circle (1pt) node[below]{$e_{23}$};
			\draw [-] (0,0)--(3,0) node[below]{$\mathbf{e_{12}}$};
			\filldraw (3,0) circle (1pt);
			\draw [-] (0,0)--(-3,0) node[below]{$\mathbf{e_{13}}$};
			\filldraw (-3,0) circle (1pt);
			\filldraw (0,2) circle (1pt);
			\draw [-] (0,0)--(0,2) node[above]{$\mathbf{e_{14}}$};
			\filldraw (3,2) circle (1pt);
			\draw [-] (0,0)--(3,2) node[above]{$\mathbf{e_{24}}$};
			\filldraw (-3,2) circle (1pt);
			\draw [-] (0,0)--(-3,2) node[above]{$\mathbf{e_{34}}$};
		}
		\end{tikzpicture}	
		\begin{tikzpicture}[scale = 0.5]
		{
			\filldraw (0,0) circle (1pt) node[below]{$e_{123}$};
			\draw [-] (0,0)--(3,2) node[above]{$e_{124}$};
			\filldraw (3,2) circle (1pt);
	    	\draw [-] (0,0)--(-3,2) node[above]{$e_{134}$};
			\filldraw (-3,2) circle (1pt);
			\draw [-] (0,0)--(0,2) node[above]{$e_{234}$};
			\filldraw (0,2) circle (1pt);
		}
		\end{tikzpicture}
		\end{center}
						
It is straightforward to verify that $(\Lambda^2\mathfrak{gl}_2)^{\mathfrak{gl}_2}=\{0\}$ and  $(\Lambda^3\mathfrak{gl}_2)^{\mathfrak{gl}_2}=\{e_{123}\}$. Let $\{e_{12},e_{13},e_{14},e_{23},e_{24},e_{34}\}$ be the basis of $\Lambda^2\mathfrak{gl}_2$, and write its dual as $\{\lambda_{12},\lambda_{13},\lambda_{14},\lambda_{23},\lambda_{24},\lambda_{34}\}$. With the help of computer software (see Appendix \ref{App:code} for an exemplary code used for examples in this thesis), one finds that in the given basis of $\Lambda^2 \mathfrak{gl}_2$, symmetric $\mathfrak{gl}_2$-invariant forms $b_{\Lambda^2 \mathfrak{gl}_2}$ on $\Lambda^2\mathfrak{gl}_2$ read
$$
[b_{\Lambda^2 \mathfrak{gl}_2}] = \left(
\begin{array}{cccccc}
0 & a & 0 & 0 & 0 & -b \\
a & 0 & 0 & 0 & b & 0 \\
0 & 0 & c & -b & 0 & 0 \\
0 & 0 & -b & -a & 0 & 0 \\
0 & b & 0 & 0 & 0 & c \\
-b & 0 & 0 & 0 & c & 0
\end{array}
\right), \qquad a,b,c \in \mathbb{R}.
$$
In the new basis $\{\tilde{e}_1, \tilde{e}_2, \tilde{e}_3, \tilde{e}_4, \tilde{e}_5 ,\tilde{e}_6\}$ of $\Lambda^2 \mathfrak{gl}_2$, given by 
$$
\tilde{e}_1=\frac{e_{24}+e_{34}}{\sqrt{2}},\,\, \tilde{e}_2=e_{14}, \,\,  \tilde{e}_3=\frac{e_{34}-e_{24}}{\sqrt{2}}, \,\, \tilde{e}_4=\frac{e_{12}-e_{13}}{\sqrt{2}},\,\, \tilde{e}_5=e_{23},\,\,
\tilde{e}_6=\frac{e_{12}+e_{13}}{\sqrt{2}},
$$
we obtain the following form of $b_{\Lambda^2 \mathfrak{gl}_2}$:
$$
[b_{\Lambda^2 \mathfrak{gl}_2}] = \left(
\begin{array}{cccccc}
c & 0 & 0 & -b & 0 & 0 \\
0 & c & 0 & 0 & -b & 0 \\
0 & 0 & -c & 0 & 0 & -b \\
-b & 0 & 0 & -a & 0 & 0 \\
0 & -b & 0 & 0 & -a & 0 \\
0 & 0 & -b & 0 & 0 & a
\end{array}
\right), \qquad a,b,c \in \mathbb{R}.
$$

Let us define
	\begin{equation}\label{ChVa}
\begin{gathered}
x=\frac{\lambda_{24}+\lambda_{34}}{\sqrt{2}},\,\, y=\lambda_{14}, \,\,  z=\frac{\lambda_{34}-\lambda_{24}}{\sqrt{2}}, \,\, \dot x=\frac{\lambda_{12}-\lambda_{13}}{\sqrt{2}},\,\, \dot y=\lambda_{23},\,\,
\dot z=\frac{\lambda_{12}+\lambda_{13}}{\sqrt{2}},
\end{gathered}
\end{equation}
which can be understood as a global adapted coordinate system to $T\mathbb{R}^3$. 
Then, we conclude that the invariant metrics on $\Lambda^2\mathfrak{gl}_2$ are linear combinations of
	$$
 f_1=z^2-x^2-y^2, \qquad f_2=\dot z^2-\dot y^2-\dot x^2, \qquad f_3=x\dot x+y\dot y+z\dot z.
 $$
 
In the next step, let us analyse the orbits of the action of ${\rm Inn}(\mathfrak{gl}_2)$ on $\Lambda^2\mathfrak{gl}_2$. The fundamental vector fields of this action are spanned by
	\begin{equation}\label{gl2_alg_fvf}
\begin{gathered}
X_1 = \dot x \partial_{\dot z} + \dot z \partial_{\dot x} - x \partial_{z} - z \partial_{x}, \quad 
X_2 = - \dot y \partial_{\dot x} -  (\dot x - \dot z) \partial_{\dot y} -  \dot y \partial_{\dot z} -  y \partial_{x} +  (x + z) \partial_{y} +  y \partial_{z}, \\
X_3 = - \dot y \partial_{\dot x} -  (\dot x + \dot z) \partial_{\dot y} +  \dot y \partial_{\dot z} +  y \partial_{x} -  (x - z) \partial_{y} +  y \partial_{z}.
\end{gathered}
\end{equation}

For $p=(x,y,z,\dot x,\dot y,\dot z) \in T \mathbb{R}$, an easy but tedious calculation gives that $\mathcal{D}_p=\langle (X_1)_p,(X_2)_p,(X_3)_p\rangle$ has dimension three if and only if the subspace of $\mathbb{R}^3$ spanned by the vectors $(-x,-y,z)$ and $(\dot x,\dot y,\dot z)$ is two-dimensional. 
Moreover, $\dim\mathcal{D}_0=0$ and $\dim \mathcal{D}_p=2$ elsewhere. The connected components of the orbits of the action of ${\rm Inn}(\mathfrak{gl}_2)$ on $\Lambda^2\mathfrak{gl}_2$ are, analogously as in the previous section, the maximal connected submanifolds of each strata $S_{(k_1,k_2,k_3)}=f_1^{-1}(k_1)\cap f_2^{-1}(k_2)\cap f_3^{-1}(k_3)$. 
	 
Note that due to the isomorphism $(x,y,z)\in \mathbb{R}^3\mapsto (x,y,z)\in T_\xi \mathbb{R}^3$ for any $\xi \in \mathbb{R}^3$, each $T_\xi\mathbb{R}^3$ can be endowed with an inner product induced by the standard inner product in $\mathbb{R}^3$.
Consider $f_1$ as a function on $\mathbb{R}^3$. Then, every $S_{(k_1,k_2,k_3)}$ is given by $(\xi,\dot \xi)\in T\mathbb{R}^3$, where $\xi=(x,y,z),\dot \xi=(\dot x,\dot y,\dot z)$, such that $\xi\in f_1^{-1}(k_1)$, while $\dot \xi$ belongs to the intersection of $S_{\xi, k_2} := f_2^{-1}(k_2)\cap T_\xi\mathbb{R}^3$ and the plane $\pi_{\xi,k_3}=\{\dot \xi\in T_\xi\mathbb{R}^3:\xi\cdot \dot \xi=k_3\}$, whose normal vector is $(x,y,z)\in T_\xi\mathbb{R}^3$. 

Let $\mathcal{O}_{k}:=f_1^{-1}(k)\subset \mathbb{R}^3$ for any $k \in \mathbb{R}$ and denote by $0_n$ the zero vector in $\mathbb{R}^n$, $n > 1$. For $\xi=(x,y,z) \in \mathbb{R}^3\backslash \{0_3\}$, the tangent space to $\mathcal{O}_{f_1(\xi)}$ at $\xi$ is given by $T_\xi \mathcal{O}_{f_1(\xi)}= \ker ({\rm d}f_1)_{\xi}$ (see \cite{L03}). Explicitly, $T_\xi \mathcal{O}_{f_1(\xi)}=\{\zeta\in T_\xi\mathbb{R}^3:  \zeta \cdot(-x,-y,z)=0\}$. Hence, $\dim \mathcal{D}_{(\xi,\dot \xi)}<3$ if and only if $\dot \xi\in (T_\xi \mathcal{O}_{f_1(\xi)})^\perp$. If $\xi=0_3$, $\dim \mathcal{D}_{(\xi,\dot \xi)}<3$ for each $(0,\dot \xi)\in T \mathbb{R}^3$. Since $\dim \mathcal{D}_{(0_3,0_3)}=0$, we get that $\dim \mathcal{D}_p = 2$ for $p \in \mathcal{A}_2$ and $\dim \mathcal{D}_p = 3$ for $p \in \mathcal{A}_3$, where we introduce
\begin{equation*}
\begin{split}
&\mathcal{A}_2 := (\{0_3\} \times T_0\mathbb{R}^3\backslash\{0 \})\,\,\bigcup \,\,\{(\xi,\dot \xi) \in T\mathbb{R}^3: \xi \in \mathbb{R}^3\backslash \{0\},\dot \xi\in (T_\xi \mathcal{O}_{f_1(\xi)})^\perp\}, \\
&\mathcal{A}_3 := \{(\xi,\dot \xi) \in T\mathbb{R}^3: \xi \in \mathbb{R}^3\backslash \{0_3\},\dot \xi\notin (T_\xi \mathcal{O}_{f_1(\xi)})^\perp\}
\end{split}
\end{equation*}

Let us finally study the orbits of ${\rm Inn}(\mathfrak{gl}_2)$ acting on $\Lambda^2\mathfrak{gl}_2 \simeq T\mathbb{R}^3$. Since any automorphism of $\mathfrak{gl}_2$ respects the decomposition $\mathfrak{gl}_2\simeq \langle e_1,e_2,e_3\rangle\oplus \langle e_4\rangle$, it leaves invariant the subspaces $\langle\lambda_{12}, \lambda_{13}, \lambda_{23} \rangle$ and $\langle \lambda_{14}, \lambda_{24}, \lambda_{34}\rangle$ of $\Lambda^2 \mathfrak{gl}_2^\star$. In consequence, $g(\xi,\dot \xi)=(g\xi,g\dot \xi)$, where $g\in {\rm Aut}(\mathfrak{gl}_2)$ acts linearly on each component. Hence, $gT_0\mathbb{R}^3=T_0\mathbb{R}^3$. Since $f_3(0_3,\dot\xi)=0$, it follows that the orbits of ${\rm Inn}(\mathfrak{gl}_2)$ in $T\mathbb{R}^3$ are $(0_3, 0_3)$ and $\{\{0_3\} \times f_2^{-1}(k)\backslash\{0_3\}\}_{k\in \mathbb{R}}$. Representative elements of each connected components for the latter family of orbits read
$$
r_{k>0,\pm} = (0_3,0,0,\pm \sqrt{k}), \quad r_{k < 0} = (0_3,\sqrt{|k|},0,0), \quad r_{0, \pm} = (0_3,0, \pm 1, 1).
$$

Let us take $(\xi,\dot \xi)\in \mathcal{A}_2$ and $\xi \in \mathbb{R}^3\backslash \{0\}$. In this case, ${\rm Inn}(\mathfrak{gl}_2)\xi=\mathcal{O}_{f_1(\xi)}\backslash\{0\}$. If $\dot \xi = 0$, it immediately follows that the orbits of ${\rm Inn}(\mathfrak{gl}_2)$ in $T\mathbb{R}^3$ are given by $ \times f_1^{-1}(k)\backslash\{0_3\}\}_{k\in \mathbb{R}} \times \{\{0_3\}$. Thus, representative elements of each connected components of orbits of ${\rm Inn}(\mathfrak{gl}_2)$ read
$$
r_{k>0,\pm} = (0,0,\pm \sqrt{k},0_3), \quad r_{k < 0} = (\sqrt{|k|},0,0,0_3), \quad r_{0, \pm} = (0, \pm 1, -1,0_3)
$$
Let $\dot \xi \in \mathbb{R}^3 \backslash \{0\}$. it can be verified by direct inspection of the graphs that for $k_1\neq 0$, the plane $\pi_{g\xi,k_3}$ cuts $(T_{g\xi} \mathcal{O}_{f_1(g\xi)})^\perp$ at one point determining $g\dot \xi$. Similarly, on sees that for $k_1=0$, we get that $(\mathcal{O}_{f_1(g\xi)})^\perp$ is parallel to $\pi_{\xi,k_3}$ and in particular, $\mathcal{O}_{f_1(g\xi)})^\perp \subset \pi_{\xi,k_3}$ for $k_3 = 0$.

Take a point $v_{\xi} := (\xi,\dot \xi)\in \mathcal{A}_3$. The action of ${\rm Aut}(SL_2)$ leaves $S_{k_1,k_2,k_3}$ invariant. Hence, the orbit of $v_\xi$ is included in $S_{(k_1,k_2,k_3)}$ for $k_i=f_i(v_\xi)$. Hence, $\xi\in \mathcal{O}_{k_1}$, $f_2(\dot \xi)=k_2$, and $f_3(\xi,\dot \xi)=k_3$. Then, the $g\xi$ generate $\mathcal{O}_{k_1}\backslash\{0\}$. Let $G_\xi$ be the isotropy group of $\xi$. Since $\dim \mathcal{O}_{k_1} = 2$ at $\xi$ and $\dim \mathcal{D}_{v_{\xi}} = 3$, then $G_\xi\dot\xi$ must be a one-dimensional submanifold of $T_\xi \mathbb{R}^3$ within the intersection of $S_{\xi,k_2} := f_2^{-1}(k_2)\cap T_\xi \mathbb{R}^3$ with $\pi_{\xi,k_3}$.  

For convenient qualitative analysis of the intersections mentioned previously, we introduce appropriate coordinate systems on each level set $f_1 = k_1$ and $f_2 = k_2$, depending on the value of $k_1, k_2 \in \mathbb{R}$. If $k_i < 0$, $i \in \{1,2\}$, we set $(\theta_i, \phi_i)$, $\theta_i \in \mathbb{R}$, $\phi_i \in [0, 2\pi[$ by
\begin{align*}
&x = \sqrt{|k_1|} \cosh{\theta_1} \sin{\phi_1}, & 
&y = \sqrt{|k_1|} \cosh{\theta_1} \cos{\phi_1}, & 
&z = \sqrt{|k_1|} \sinh{\theta_1} \\
&\dot x = \sqrt{|k_2|} \cosh{\theta_2} \sin{\phi_2}, & 
&\dot y = \sqrt{|k_2|} \cosh{\theta_2} \cos{\phi_2}, & 
&\dot z = \sqrt{|k_2|} \sinh{\theta_2}
\end{align*}
If $k_i > 0$, we define $(\theta_i, \delta_i)$, $\theta_i, \delta_i \in \mathbb{R}$, by
\begin{align*}
&x = \sqrt{k_1} \sinh{\theta_1} & 
&y = ak_1 \cosh{\theta_1} \sinh{\delta_1} & 
&z = ak_1 \cosh{\theta_1} \cosh{\delta_1} \\
&\dot x = \sqrt{k_2} \sinh{\theta_2} & 
&\dot y = bk_2 \cosh{\theta_2} \sinh{\delta_2} & 
&\dot z = bk_2 \cosh{\theta_2} \cosh{\delta_2}
\end{align*}
where $a,b \in \{\pm 1\}$. If $k_i = 0$, then we introduce $(z_i, \phi_i)$, $z_i \in \mathbb{R}$, $\phi_i \in [0, 2\pi[$, by
\begin{align*}
&x = z_1 \sin{\phi_1}, & &y= z_1 \cos{\phi_1}, & &z = z_1 \\
&x = z_2 \sin{\phi_2}, & &y= z_2 \cos{\phi_2}, & &z = z_2
\end{align*}

It is immediate to conclude that in general, there are nine choices of pairs $(k_1, k_2)$ which possibly yield qualitatively different intersections of $S_{\xi, k_2}$ with $\pi_{k_3}$. Let us first discuss the solutions of the modified Yang--Baxter equation, which happen to reduce the analysis to only few specific cases.

The mCYBE reads
$$
x \dot y - y \dot x = 0, \quad y \dot z + z \dot y = 0, \quad x \dot z + z \dot x = 0.
$$
By tedious computations, we get that the only nontrivial solutions can be obtained for pairs $(k_1 < 0, k_2, k_2 < 0)$, $(k_1 > 0, k_2 > 0)$ and $(k_1 = 0, k_2 = 0)$. These solutions expressed in the initial coordinate system $\{x, y, z, \dot x, \dot y, \dot z\}$ read:
\begin{itemize}
\item for $k_1 < 0$ and $k_2 < 0$, we get
\begin{itemize}
\item[1)] $\theta_2 = -\theta_1$ and $\phi_2 = \phi_1$, which yields
\begin{align*}
&x = \sqrt{|k_1|} \cosh{\theta_1} \sin{\phi_1}, & 
&y = \sqrt{|k_1|} \cosh{\theta_1} \cos{\phi_1}, & 
&z = \sqrt{|k_1|} \sinh{\theta_1} \\
&\dot x = \sqrt{|k_2|} \cosh{\theta_1} \sin{\phi_1}, & 
&\dot y = \sqrt{|k_2|} \cosh{\theta_1} \cos{\phi_1}, & 
&\dot z = -\sqrt{|k_2|} \sinh{\theta_1}
\end{align*}
\item[2)] $\theta_2 = \theta_1$ and $\phi_2 = \phi_1 - \pi$, which yields
\begin{align*}
&x = \sqrt{|k_1|} \cosh{\theta_1} \sin{\phi_1}, & 
&y = \sqrt{|k_1|} \cosh{\theta_1} \cos{\phi_1}, & 
&z = \sqrt{|k_1|} \sinh{\theta_1} \\
&\dot x = -\sqrt{|k_2|} \cosh{\theta_1} \sin{\phi_1}, & 
&\dot y = -\sqrt{|k_2|} \cosh{\theta_1} \cos{\phi_1}, & 
&\dot z = \sqrt{|k_2|} \sinh{\theta_1}
\end{align*}
\end{itemize}
\item for $k_1 > 0$ and $k_2 > 0$, we get
\begin{itemize}
\item[3)] $\theta_2 = -\theta_1$ and $\delta_2 = -\delta_1$ (for $a = +1$, $b = +1$), which yields
\begin{align*}
&x = \sqrt{k_1} \sinh{\theta_1} & 
&y = \sqrt{k_1} \cosh{\theta_1} \sinh{\delta_1} & 
&z = \sqrt{k_1} \cosh{\theta_1} \cosh{\delta_1} \\
&\dot x = -\sqrt{k_2} \sinh{\theta_1} & 
&\dot y = -\sqrt{k_2} \cosh{\theta_1} \sinh{\delta_1} & 
&\dot z = \sqrt{k_2} \cosh{\theta_1} \cosh{\delta_1}
\end{align*}
\item[4)] $\theta_2 = \theta_1$ and $\delta_2 = -\delta_1$ (for $a = +1$, $b = -1$), which yields
\begin{align*}
&x = \sqrt{k_1} \sinh{\theta_1} & 
&y = \sqrt{k_1} \cosh{\theta_1} \sinh{\delta_1} & 
&z = \sqrt{k_1} \cosh{\theta_1} \cosh{\delta_1} \\
&\dot x = \sqrt{k_2} \sinh{\theta_1} & 
&\dot y = \sqrt{k_2} \cosh{\theta_1} \sinh{\delta_1} & 
&\dot z = -\sqrt{k_2} \cosh{\theta_1} \cosh{\delta_1}
\end{align*}
\item[5)] $\theta_2 = \theta_1$ and $\delta_2 = -\delta_1$ (for $a = -1$, $b = +1$), which yields
\begin{align*}
&x = \sqrt{k_1} \sinh{\theta_1} & 
&y = -\sqrt{k_1} \cosh{\theta_1} \sinh{\delta_1} & 
&z = -\sqrt{k_1} \cosh{\theta_1} \cosh{\delta_1} \\
&\dot x = \sqrt{k_2} \sinh{\theta_1} & 
&\dot y = -\sqrt{k_2} \cosh{\theta_1} \sinh{\delta_1} & 
&\dot z = \sqrt{k_2} \cosh{\theta_1} \cosh{\delta_1}
\end{align*}
\item[6)] $\theta_2 = -\theta_1$ and $\delta_2 = -\delta_1$ (for $a = -1$, $b = -1$), which yields
\begin{align*}
&x = \sqrt{k_1} \sinh{\theta_1} & 
&y = -\sqrt{k_1} \cosh{\theta_1} \sinh{\delta_1} & 
&z = -\sqrt{k_1} \cosh{\theta_1} \cosh{\delta_1} \\
&\dot x = -\sqrt{k_2} \sinh{\theta_1} & 
&\dot y = \sqrt{k_2} \cosh{\theta_1} \sinh{\delta_1} & 
&\dot z = -\sqrt{k_2} \cosh{\theta_1} \cosh{\delta_1}
\end{align*}
\end{itemize}
\item for $k_1 = 0$ and $k_2 = 0$, we get
\begin{itemize}
\item[7)] $\phi_2 = \phi_1 - \pi$, which yields
\begin{align*}
&x = z_1 \sin{\phi_1}, & &y= z_1 \cos{\phi_1}, & &z = z_1 \\
&\dot x = -z_2 \sin{\phi_1}, & &\dot y= -z_2 \cos{\phi_1}, & &\dot z = z_2
\end{align*}
\end{itemize}
\end{itemize}

On subsets given by all seven cases listed above, we compute the rank of the distribution $\mathcal{D}_{\mathfrak{gl}_2}$ spanned by vector fields $X_1, X_2, X_3$. In each case, we get ${\rm rk} \mathcal{D}_{\mathfrak{gl}_2} = 2$. Thus, we conclude that each subset given by solutions 1)-6) is a separate connected component and thus, a single orbit of the action of ${\rm Inn}(\mathfrak{gl}_2)$. For these orbits, let us conveniently choose the following representative elements:
\begin{itemize}
\item[1)] $r_1 = \sqrt{|k_1|} e_{12} + \sqrt{|k_2|} e_{23}$ (for $\theta_1 = 0$, $\phi_1 = \frac{\pi}{2}$ and $k_1, k_2 < 0$)
\item[2)] $r_2 = \sqrt{|k_1|} e_{12} - \sqrt{|k_2|} e_{23}$ (for $\theta_1 = 0$, $\phi_1 = \frac{3\pi}{2}$ and $k_1, k_2 < 0$)
\item[3)] $r_3 = \sqrt{\frac{k_1}{2}} e_{12} + \sqrt{\frac{k_1}{2}} e_{13} - \sqrt{\frac{k_2}{2}} e_{24} + \sqrt{\frac{k_2}{2}} e_{34}$ (for $\theta_1 = 0$, $\delta_1 = 0$ and $k_1, k_2 > 0$)
\item[4)] $r_4 = -\sqrt{\frac{k_1}{2}} e_{12} - \sqrt{\frac{k_1}{2}} e_{13} - \sqrt{\frac{k_2}{2}} e_{24} + \sqrt{\frac{k_2}{2}} e_{34}$ (for $\theta_1 = 0$, $\delta_1 = 0$ and $k_1, k_2 > 0$)
\item[5)] $r_5 = \sqrt{\frac{k_1}{2}} e_{12} + \sqrt{\frac{k_1}{2}} e_{13} + \sqrt{\frac{k_2}{2}} e_{24} - \sqrt{\frac{k_2}{2}} e_{34}$ (for $\theta_1 = 0$, $\delta_1 = 0$ and $k_1, k_2 > 0$)
\item[6)] $r_6 = -\sqrt{\frac{k_1}{2}} e_{12} - \sqrt{\frac{k_1}{2}} e_{13} + \sqrt{\frac{k_2}{2}} e_{24} - \sqrt{\frac{k_2}{2}} e_{34}$ (for $\theta_1 = 0$, $\delta_1 = 0$ and $k_1, k_2 > 0$)
\end{itemize}

Notice that the solution 7) depends on $(z_1, z_2, \phi_1)$. Since ${\rm rk} \mathcal{D}_{\mathfrak{gl}_2} = 2$ on this subset, one needs to analyse the vector fields (\ref{gl2_alg_fvf}) in order to find the orbits of ${\rm Inn}(\mathfrak{gl}_2)$-action. In coordinates $(z_1, z_2, \phi_1)$, these vector fields read
\begin{equation*}
\begin{split}
&X_1 = -2 \cos{\phi_1} \partial_{\phi_1} - 2z_1 \sin{\phi_1} \partial_{z_1} - 2z_2 \sin{\phi_1} \partial_{z_2} \\
&X_2 = -\sqrt{2} (1 + \sin{\phi_1}) \partial_{\phi_1} + \sqrt{2} z_1 \cos{\phi_1} \partial_{z_1} + \sqrt{2} z_2 \cos{\phi_1} \partial_{z_2} \\ 
&X_3 = -\sqrt{2} (-1 + \sin{\phi_1}) \partial_{\phi_1} + \sqrt{2} z_1 \cos{\phi_1} \partial_{z_1} + \sqrt{2} z_2 \cos{\phi_1} \partial_{z_2}
\end{split}
\end{equation*}

Denote $X_{+} := X_1 + X_2$ and $X_{-} := X_1 - X_2$. In new coordinate system $(\phi_1, z_{+}, z_{-})$, where $z_{+} := z_1 + z_2$ and $z_{-} := z_1 - z_2$, we obtain
\begin{equation*}
\begin{split}
&X_1 = -2 \cos{\phi_1} \partial_{\phi_1} - 2z_{+} \sin{\phi_1} \partial_{z_{+}} \\
&X_{+} = -2\sqrt{2} \partial_{\phi_1} + 2\sqrt{2} \cos{\phi_1} z_{+} \partial_{z_{+}} \\
&X_{-} = 2\sqrt{2} \partial_{\phi_1}.
\end{split}
\end{equation*}
Then, it is immediate to verify that these vector fields do not change the value of $z_{-}$. Thus, $z_1 - z_2 = const$ for any orbit of ${\rm Inn}(\mathfrak{gl}_2)$-action. In consequence, each connected component of the family, parametrised by $\beta \in \mathbb{R}$, within the subset given by solution 7) of the form
\begin{align*}
&x = z_1 \sin{\phi_1}, & &y= z_1 \cos{\phi_1}, & &z = z_1 \\
&\dot x = (\beta - z_1) \sin{\phi_1}, & &\dot y= (\beta - z_1) \cos{\phi_1}, & &\dot z = z_1 - \beta
\end{align*}
gives rise to the orbit of ${\rm Inn}(\mathfrak{gl}_2)$-action on $\Lambda^2 \mathfrak{gl}_2$. Since both $z_1, z_2$ cannot be zero, we get that $z_1 \in \mathbb{R} \backslash \{0,\beta\}$. Thus, there are three connected components for $\beta \neq 0$, given by the subsets $z > \max (0, \beta)$, $\min (0, \beta) < z < \max (0, \beta)$ and $z < \min (0, b\eta)$. Their representative elements read (for $\phi_1 = \frac{3\pi}{2}$)
$$
r_{1, \beta \neq 0} = \beta\sqrt{2} e_{12} - 2\beta \sqrt{2} e_{24}, \quad
r_{2}, \beta \neq 0 = -\frac{1}{\sqrt{2}} \beta e_{12} - \frac{1}{\sqrt{2}} \beta e_{24}, \quad
r_{3, \beta \neq 0} = -2\beta \sqrt{2} e_{12} + \beta \sqrt{2} e_{24}.
$$
Similarly, two components for $\beta = 0$ are given by $z > 0$ and $z < 0$ and the representative elements of each subset read
$$
r_{1, \beta = 0} = -\sqrt{2} e_{12} + \sqrt{2} e_{24}, \quad r_{2, \beta \neq 0} = \sqrt{2} e_{12} - \sqrt{2} e_{24}
$$

It is immediate to verify that the automorphism $T \in {\rm Aut}(\mathfrak{gl}_2)$ given by 
$$
T(e_1) := e_1, \quad T(e_2) := \frac{1}{\alpha} e_2, \quad T(e_3) := \alpha e_3, \quad T(e_4) := \beta e_4
$$ 
with $\alpha, \beta \in \rz$ allows us to identify all connected components within a given value of $\beta$ and moreover, all elements in the family. Thus, we conclude that there is one orbit associated with the solution 7) and its representative element is of the form $r_7 = e_{12} + e_{24}$.

Similarly, the automorphism $T$ allows us to identify certain classes among solutions 1) - 6). In consequence, we obtain $r_{I, \alpha} = \alpha e_{13} + e_{24}$ and $r_{II, \alpha} = \alpha e_{12} + \alpha e_{13} + e_{24} - e_{34}$ for $\alpha > 0$. Upon a similar identification, the orbits listed previously that belong to $\mathcal{A}_2$ give
\begin{align*}
&r_{0, k > 0, \pm} = \pm \sqrt{\frac{k}{2}} e_{12} \pm \sqrt{\frac{k}{2}} e_{13}, &
&r_{0, k < 0} = \sqrt{|k|} e_{23}, &
&r_{0, k = 0, \pm} = e_{12} + e_{13} \pm \sqrt{2} e_{23} \\
&r_{k > 0, 0, \pm} = \pm e_{24} \mp e_{34}, &
&r_{k < 0, 0} = e_{14}, &
&r_{k = 0, \pm} = e_{24} - e_{34} \pm \sqrt{2} e_{14}
\end{align*}
Thus, we obtain the following classes of $r$-matrices relative to ${\rm Aut}(\mathfrak{gl}_2)$:
\begin{align*}
&r_{1, \alpha > 0, \pm} = \pm \alpha e_{12} \pm \alpha e_{13}, &
&r_{2, \alpha > 0} = \alpha e_{23}, &
&r_{3, \pm} = e_{12} + e_{13} \pm \sqrt{2} e_{23} \\
&r_{4, \pm} = \pm e_{24} \mp e_{34}, &
&r_{5} = e_{14}, &
&r_{6, \pm} = e_{24} - e_{34} \pm \sqrt{2} e_{14}, \\
&r_{7, \alpha > 0} = \alpha e_{13} + e_{24}, &
&r_{8, \alpha > 0} = \alpha e_{12} + \alpha e_{13} + e_{24} - e_{34}, &
&r_{9} = e_{12} + e_{24}.
\end{align*}

\chapter{Darboux families and the classification of Lie bialgebras}\label{Ch:Darboux}

 This chapter introduces a new concept, the so-called {\it Darboux family}, which is employed to determine, to analyse geometrically, and to classify up to Lie algebra automorphisms, in a relatively easy manner, coboundary Lie bialgebras on real four-dimensional indecomposable Lie algebras.  The Darboux family notion can be considered as a generalisation of the Darboux polynomial for a vector field. The classification of $r$-matrices and solutions to classical Yang-Baxter equations for real four-dimensional indecomposable Lie algebras is also given in detail. Our methods can further be applied to general, even higher-dimensional, Lie algebras. As a byproduct, a method to obtain matrix representations of certain Lie algebras with a non-trivial center is developed.

\section{On the notion of Darboux family}\label{Ch:Darb_Sec:DarFam}

On a finite-dimensional vector space $E$, a polynomial function is a polynomial expression on a set of linear coordinates on $E$. A {\it Darboux polynomial}, $P$, for a polynomial vector field $X$ on $E$ is a polynomial function on $E$ so that $XP=fP$ for a certain polynomial $f$ on $E$, the so-called {\it cofactor} of $P$ relative to $X$ \cite{Gu11,LV13}. It is worth noting that, geometrically, one cannot define intrinsically what a polynomial on a general manifold is, as the explicit form of a function on a manifold  depends on the chosen coordinate system. However, one can introduce the following generalisation of a Darboux polynomial on manifolds.

\begin{definition}\label{Def:Darb_fam}
A {\it Darboux family} for a Vessiot--Guldberg Lie algebra $V$ on a manifold $M$ is an $s$-dimensional linear space $\mathcal{A}$ spanned by certain functions $f_1,\ldots,f_s \in C^{\infty}(M)$ such that  one can write
$$
Xf_j=\sum_{i=1}^sh_{j X}^i f_i, \quad \forall j=1,\ldots,s
$$
for every vector field $X\in V$ and certain smooth functions $h_{jX}^i \in C^{\infty}(M)$. We call each $h_{jX}^i$ a {\it co-factor} of $f_j$ relative to $X$ and the basis $\{f_1,\ldots,f_s\}$. The subset $\ell_{\mathcal{A}}:= \{p \in M: f(p)=0, \forall f \in \mathcal{A}\}$ is called the {\it locus} of the Darboux family $\mathcal{A}$.
\end{definition}

If all the functions $h_{jX}^i$ are equal to zero for every $X \in V$ and $i = 1, \ldots, s$, then $f_j $ becomes a constant of motion for the vector fields of $V$ on $M$. If the $h_{jX}^i$ are constants for every $X \in V$ and $i,j = 1, \ldots, s$, then we say that the Darboux family is {\it linear}. In this case, $V$ gives rise to a Lie algebra representation $\rho:X\in V\mapsto D_X\in {\rm End}(\mathcal{A})$, where $D_X$ stands for the action of the vector field $X$ on the space of functions $\mathcal{A}$. 

The vector fields of $V$ span a generalised distribution $\mathcal{D}^V$, which is integrable by Theorem \ref{Th:StSus}. This leads to a stratification of $M$ by some disjoint immersed submanifolds which may have different dimensions. 

\begin{lemma}\label{Lem:locus_orb}
For a Darboux family $\mathcal{A}$, its locus $\ell_{\mathcal{A}}$ is the sum (as subsets of $M$) of a collection of strata of the generalised distribution $\mathcal{D}^V$.
\end{lemma}
\begin{proof}
Let us first show that the integral curves of vector fields in $V$ passing through the locus $\ell_\mathcal{A}$ of the Darboux family $\mathcal{A}$ remain within $\mathcal{A}$. 
Indeed, take a point in the locus $\ell_{\mathcal{A}}$ and consider a basis $Y_1,\ldots,Y_q$ of the Vessiot--Guldberg Lie algebra $V$ associated with $\mathcal{A}$. Then, each function $f_i \in \mathcal{A}$ satisfy
$$
\frac{df_i(\Psi(t))}{dt}=(Xf_i)\Psi(t)=\sum_{i=1}^sh_{iX}^j(\Psi(t))f_j(\Psi(t)),\qquad i=1,\ldots,s.
$$
where $\Psi(t)$ is an integral curve of a vector field $X \in V$ such that $\Psi(0)\in \mathcal{S}_\mathcal{A}$.
Hence, the values of the $f_i(\Psi(t))$ can be understood as the solutions $u_i(t)$ to the linear system of first-order ordinary differential equations in normal form with $t$-dependent coefficients 
\begin{equation}\label{darb_sys}
\frac{du_i}{dt}=\sum_{j=1}^sh^j_{iX}(\Psi(t))u_j,\qquad i=1,\ldots,s,
\end{equation}
with initial condition $u_1(0)=\ldots=u_s(0)=0$. By the existence and uniqueness theorem, the solution to (\ref{darb_sys}) is $u_1(t)=\ldots=u_s(t)=0$. Therefore, the functions $f_1,\ldots,f_s$ vanish on the integral curves of $X$. 

From this fact and the decomposition (\ref{eq:dec}), it follows that the functions $f_1,\ldots,f_s$ vanish on the strata of the distribution $\mathcal{D}^V$ containing a point within $\ell_\mathcal{A}$. In other words, the strata of $\mathcal{D}^V$ containing some point of $\ell_{\mathcal{A}}$ are fully contained in $\ell_{\mathcal{A}}$. In consequence, the locus $\ell_{\mathcal{A}}$ must be the sum (as subsets of $M$) of a collection of strata of the generalised distribution $\mathcal{D}^V$.
\end{proof}

In view of Lemma \ref{Lem:locus_orb}, one can separate the analysis of the strata in $M$ of the generalised distribution $\mathcal{D}^V$ to two regions of $M$: the locus $\ell_{\mathcal{A}}$ and its completion $M \backslash \ell_{\mathcal{A}}$. This approach will be particularly useful to obtain the strata of $\mathcal{D}^V$ at points where the generalised distribution is not regular and it may not exist a constant of motion common to all the vector fields in $V$ that can be used to obtain the strata of $\mathcal{D}^V$. 

As we show next, Darboux families are useful in the study of solutions to  mCYBEs and CYBEs. Let us start with a general result.

 \begin{theorem}\label{Th:ConsDar32} 
 Let $\mathcal{A}^{(3)}_{\mathfrak{g}}$ be a Darboux family for a Vessiot--Guldberg Lie algebra $V_{\mathfrak{g}}^{(3)}$ on $\Lambda^3\mathfrak{g}$. Then, the space of functions
 \begin{equation}\label{Eq:Ind}
 \mathcal{A}: = \left\{g \in C^{\infty}(\Lambda^2 \mathfrak{g}): \, g = f \circ [\cdot, \cdot], \, f\in \mathcal{A}_{\mathfrak{g}}^{(3)} \right\}
 \end{equation}
 is a Darboux family for a Vessiot--Guldberg Lie algebra $V^{(2)}_{\mathfrak{g}}$ on $\Lambda^2\mathfrak{g}$. If $\mathcal{A}_{\mathfrak{g}}^{(3)}$ is a linear Darboux family, then $\mathcal{A}$ is also a linear Darboux family.
 \end{theorem}
 \begin{proof}
Take $X^{(2)} \in V^{(2)}_{\mathfrak{g}}$. Since $X^{(2)}$ is a fundamental vector field of the action of ${\rm Aut}(\mathfrak{g})$ on $\Lambda^2 \mathfrak{g}$, its flow is given by the one-parametric group of diffeomorphisms $\Lambda^2 T_t$ induced by a certain one-parametric group $\{T_t\}_{t \in \mathbb{R}}$ of Lie algebra automorphisms of $\mathfrak{g}$. Every $g \in \mathcal{A}$ is of the form $g(r) = f([r, r])$ for some $f \in \mathcal{A}_{\mathfrak{g}}^{(3)}$ and every $r\in \Lambda^2\mathfrak{g}$. Then,
 \begin{equation*}
 \begin{split}
 (X^{(2)}g)(r) &= \frac{d}{dt}\bigg\vert_{t = 0} g(\Lambda^2 T_t (r)) = \frac{d}{dt}\bigg\vert_{t = 0} f([\Lambda^2 T_t (r), \Lambda^2 T_t (r)]) \\
 &= \frac{d}{dt}\bigg\vert_{t = 0} f(\Lambda^3 T_t [r, r]) = (X^{(3)} f)([r, r]) = \sum_{i = 1}^{r} h^i([r,r])f_i([r, r]),
 \end{split}
 \end{equation*}
 where $f_1, \ldots, f_s$ form a basis of the Darboux family $\mathcal{A}_{\mathfrak{g}}^{(3)}$, the functions $h^1,\ldots,h^s$ are the cofactors of $f$ relative to $X^{(3)}$ and the basis $\{f_1,\ldots,f_s\}$, and $X^{(3)}$ is the fundamental vector field of the action of ${\rm Aut}(\mathfrak{g})$ on $\Lambda^3\mathfrak{g}$ induced by the one-parameter group $\{T_t\}_{t\in \mathbb{R}}$ of Lie algebra automorphisms of $\mathfrak{g}$. For any $i \in \{1, \ldots, s \}$, the function $g_i \in C^{\infty}(\Lambda^2 \mathfrak{g})$ of the form $g_i(r) := f_i([r, r])$ belongs to $\mathcal{A}$ by definition. Moreover, each $\tilde{h}_j(r) := h^j([r,r])$ for $j \in \{1, \ldots, s\}$ gives a function on $\Lambda^2\mathfrak{g}$. Thus, it follows that $\mathcal{A}$ forms a Darboux family for $V_{\mathfrak{g}}$ on $\Lambda^2 \mathfrak{g}$.
 
 If $\mathcal{A}^{(3)}_{\mathfrak{g}}$ is a linear Darboux family, then the cofactors $h^1,\ldots,h^s$ are constants. In consequence, the functions $h^i([\cdot, \cdot])$ for $i \in \{1, \ldots, s\}$ are also constants. Hence, $\mathcal{A}$ is a linear Darboux family for $V^{(2)}_{\mathfrak{g}}$ on $\Lambda^2\mathfrak{g}$.
 \end{proof}
 
 The following proposition shows that the solutions to the mCYBE on $\mathfrak{g}$ correspond to the locus of a Darboux family relative to $V_{\mathfrak{g}}$. A similar result holds for the CYBE.

 \begin{proposition}  
 Let $((\Lambda^3\mathfrak{g})^\mathfrak{g})^\circ$ denote the annihilator of $(\Lambda^3\mathfrak{g})^\mathfrak{g}$, i.e. the subspace of elements of $(\Lambda^3\mathfrak{g})^*$ vanishing on $(\Lambda^3\mathfrak{g})^\mathfrak{g}$.
The functions
\begin{equation}\label{Eq:ParFam}
f_\upsilon:r\in \Lambda^2\mathfrak{g}\mapsto \upsilon([r,r])\in \mathbb{R},\qquad \upsilon\in ((\Lambda^3\mathfrak{g})^\mathfrak{g})^\circ,
\end{equation}
span a linear Darboux family, $\mathcal{A}$, for the Lie algebra $V_\mathfrak{g}$ of fundamental vector fields of the action of ${\rm Aut}(\mathfrak{g})$ on $\Lambda^2\mathfrak{g}$. The locus of $\mathcal{A}$ is the set of solutions to the mCYBE for $\mathfrak{g}$. Moreover, the components of the CYBE in a coordinate system given by a basis of $\Lambda^2\mathfrak{g}^*$ span a linear Darboux family for $V_{\mathfrak{g}}$ on $\Lambda^2\mathfrak{g}$.  
\end{proposition}
\begin{proof} 
Let us prove that the annihilator of $(\Lambda^3\mathfrak{g})^{\mathfrak{g}}$ is a Darboux family for $V_{\mathfrak{g}}^{(3)}$ on $\Lambda^3\mathfrak{g}$. In fact, if $X^{(3)} \in V_{\mathfrak{g}}^{(3)}$, then its flow is given by a one-parameter group of diffeomorphisms of the form $\Lambda^3 T_t$ for certain Lie algebra automorphisms $T_t\in {\rm Aut}(\mathfrak{g})$ with $t\in \mathbb{R}$. Consider a function $f \in \Lambda^3 \mathfrak{g}^*$ such that $f(w) = 0$ for every $w \in (\Lambda^3\mathfrak{g})^{\mathfrak{g}}$, i.e. $f\in ((\Lambda^3\mathfrak{g})^\mathfrak{g})^\circ$. Then,
 $$
 (X^{(3)}f)(w) = \frac{d}{dt}\bigg\vert_{t = 0} f((\Lambda^3 T_t)( w)) = \frac{d}{dt}\bigg\vert_{t = 0} f\left(w_t \right) =
 0,
 $$
where we have used  that $(\Lambda^3\mathfrak{g})^{\mathfrak{g}}$ is closed under the extension to $\Lambda^3\mathfrak{g}$ of Lie algebra automorphisms of $\mathfrak{g}$ and thus, $(\Lambda^3T_t)(w) = :w_t\in (\Lambda^3\mathfrak{g})^{\mathfrak{g}}$ for any $w \in (\Lambda^3\mathfrak{g})^{\mathfrak{g}}$. Hence, $X^{(3)}f$ vanishes on $(\Lambda^3\mathfrak{g})^{\mathfrak{g}}$. Since $X^{(3)}$ is a linear vector field on $\Lambda^3\mathfrak{g}$ and $f$ is a linear function on $\Lambda^3\mathfrak{g}$, one obtains that $X^{(3)}f$ is a linear function on $\Lambda^3\mathfrak{g}$ vanishing on $(\Lambda^3\mathfrak{g})^{\mathfrak{g}}$. Therefore, $X^{(3)}f$ must be a linear combination of elements of a basis of $((\Lambda^3\mathfrak{g})^{\mathfrak{g}})^\circ$ and in consequence, $((\Lambda^3\mathfrak{g})^{\mathfrak{g}})^\circ$ becomes a linear Darboux family of $V_{\mathfrak{g}}^{(3)}$ on $\Lambda^3\mathfrak{g}$. By Theorem \ref{Th:ConsDar32}, the space of functions (\ref{Eq:ParFam}) 
becomes a linear Darboux family for $V_{\mathfrak{g}}$ on $\Lambda^2\mathfrak{g}$. 

The locus of (\ref{Eq:ParFam}) consists of elements $r \in \Lambda^2 \mathfrak{g}$ such that $f_v(r) = v([r,r]) = 0$ for all $v \in ((\Lambda^3\mathfrak{g})^{\mathfrak{g}})^{\circ}$. Hence, $[r,r] \in (\Lambda^3\mathfrak{g})^{\mathfrak{g}}$, which means that $r$ is a solution of the mCYBE for $\mathfrak{g}$. Conversely, any solution  $r \in \Lambda^2 \mathfrak{g}$ to mCYBE satisfies $[r,r] \in (\Lambda^3\mathfrak{g})^{\mathfrak{g}}$. Then, it follows that $v([r,r]) = 0$ for any $v \in ((\Lambda^3\mathfrak{g})^{\mathfrak{g}})^{\circ}$. Thus, the locus of (\ref{Eq:ParFam}) is the set of solutions to mCYBE. 

Since every $X^{(3)}\in V^{(3)}_{\mathfrak{g}}$ is linear vector field in a linear coordinate system on $\Lambda^3 \mathfrak{g}$, the space $\Lambda^3\mathfrak{g}^*$ is also a Darboux family for $V_{\mathfrak{g}}^{(3)}$ on $\Lambda^3\mathfrak{g}$. By Theorem \ref{Th:ConsDar32}, one has that the family of functions on $\Lambda^2\mathfrak{g}$ induced by (\ref{Eq:Ind}) is a linear Darboux family relative to $V_{\mathfrak{g}}$ on $\Lambda^2\mathfrak{g}$. The Darboux family is indeed spanned by the components of the CYBE of $\mathfrak{g}$, i.e. $[r,r]=0$, and its locus is its set of solutions.
\end{proof}

Finally, let us provide several results on the determination of Darboux families. 

\begin{proposition}\label{Prop:darb_sum} 
If $\mathcal{A}_1$ and $\mathcal{A}_2$ are Darboux families for the same Vessiot--Guldberg Lie algebra $V$ on a manifold $M$, then the sum $\mathcal{A}_1+\mathcal{A}_2$ is a Darboux family for $V$ on $M$ as well.
\end{proposition} 
\begin{proof}
Let $\mathcal{A}_1 := \langle f^{(1)}_1, \ldots, f^{(1)}_{s_1}\rangle$ and $\mathcal{A}_2 := \langle f^{(2)}_1, \ldots, f^{(2)}_{s_2}\rangle$. Then, the sum is spanned by $\mathcal{A}_1+\mathcal{A}_2 = \langle f^{(1)}_1, \ldots, f^{(1)}_{s_1}, f^{(2)}_1, \ldots, f^{(2)}_{s_2} \rangle$. Obviously, any element $g$ of $\mathcal{A}_1+\mathcal{A}_2$ reads 
$$
g = \sum_{i = 1}^{s_1} \alpha_i f^{(1)}_i + \sum_{j=1}^{s_2} \beta_j f^{(2)}_j
$$
for certain $\alpha_i, \beta_j \in \mathbb{R}$. Given $X \in V$, one gets
\begin{equation*}
\begin{split}
X(g) &= \sum_{i = 1}^{s_1} \alpha_i X(f^{(1)}_i) + \sum_{j=1}^{s_2} \beta_j X(f^{(2)}_2) = \sum_{i = 1}^{s_1} \alpha_i \sum_{p = 1}^{s_1} h^{(1,i)}_p f^{(1)}_p + \sum_{j=1}^{s_2} \beta_j \sum_{q = 1}^{s_2} h^{(2,j)}_p f^{(2)}_p \\
&= \sum_{p = 1}^{s_1} \left[\sum_{i = 1}^{s_1} \alpha_i h^{(1,i)}_p \right] f^{(1)}_p + \sum_{q=1}^{s_2} \left[\sum_{j = 1}^{s_2} \beta_j h^{(2,j)}_p \right] f^{(2)}_p 
\end{split}
\end{equation*}
Thus, $\mathcal{A}_1+\mathcal{A}_2$ forms a Darboux family for $V$.
\end{proof}

Obviously, the locus of $\mathcal{A}_1+\mathcal{A}_2$ is contained in the locus of $\mathcal{A}_1$ and $\mathcal{A}_2$. This fact will be used to distinguish between different strata of the generalised distribution spanned by the vector fields of $V_{\mathfrak{g}}$ on $\mathcal{Y}_{\mathfrak{g}}$. 

In next sections, an important role is played by the linear Darboux families for a certain $V_{\mathfrak{g}}$ given by a one-dimensional $\mathcal{A}\subset (\Lambda^2\mathfrak{g})^*$. We will call such a linear one-dimensional Darboux family a {\it brick}. From the practical standpoint, bricks are quite easy to analyse due to their dimension and moreover, by Proposition \ref{Prop:darb_sum}, one can efficiently construct more complex and useful Darboux families with such one-dimensional objects. Indeed, for every linear Darboux family for $V_{\mathfrak{g}}$, bricks  are easy to obtain since, in view of (\ref{Eq:Red}), they are given by the intersection of an eigenvector space for each endomorphism of the form $\Lambda^2d\in \mathfrak{gl}(\Lambda^2\mathfrak{g})$ with  $d\in \mathfrak{der}(\mathfrak{g})$.

Let us give another interesting proposition, which is an immediate consequence of the fact that the vector fields of $V_{\mathfrak{g}}$ are linear relative to a linear coordinate set on $\Lambda^2\mathfrak{g}$.

\begin{proposition} The space $\Lambda^2\mathfrak{g}^*$ is a linear Darboux family of functions on  $\Lambda^2\mathfrak{g}$ relative to the Lie algebra $V_{\mathfrak{g}}$.
\end{proposition}

Since $\Lambda^2 \mathfrak{g}^*$ is a linear Darboux family of $V_{\mathfrak{g}}$, there exists a linear representation of $\mathfrak{aut}(\mathfrak{g})$ on $
\Lambda^2\mathfrak{g}^*$. Its irreducible representations also give rise to Darboux families, which can be summed (as linear spaces) to the Darboux family, or potentially elements thereof, associated with the mCYBE to determine new Darboux families. The locus of such sums will be interested to as when they contain elements of  $\mathcal{Y}_{\mathfrak{g}}$. This will be employed to obtain the strata of $V_{\mathfrak{g}}$ within $\mathcal{Y}_{\mathfrak{g}}$ and, therefore, inequivalent $r$-matrices relative to the action of ${\rm Aut}_c(\mathfrak{g})$ on $\Lambda^2\mathfrak{g}$. This in turn will be used to obtain families of inequivalent $r$-matrices and coboundary cocommutators for real four-dimensional indecomposable Lie algebras \cite{SW14} in Section \ref{Ch:Darb_Sec:Cla}.

\section{Geometry of mCYBE solutions and Darboux families}\label{Ch:Darb_Sec:mcybe}

As reminded throughout this thesis, the classification of equivalent $r$-matrices for the Lie algebra $\mathfrak{g}$ amounts to the study of the orbits of the action of ${\rm Aut}(\mathfrak{g})$ on the space $\mathcal{Y}_{\mathfrak{g}} \subset \Lambda^2 \mathfrak{g}$ of solutions to mCYBE. Since this task is not elementary and it cannot be accomplished efficiently by straightforward brute-force approach, we developed in Chapter \ref{Ch:alg_met} several techniques based on $\mathfrak{g}$-invariant forms to facilitate the work on this problem.

In this section, we attempt the analysis of the orbits of the ${\rm Aut}(\mathfrak{g})$-action on $\mathfrak{Y}_{\mathfrak{g}}$ within the geometric framework. We detail some observations regarding the geometry of $\mathcal{Y}_{\mathfrak{g}}$ and most importantly, we discuss the role of Darboux families in the study of such orbits.

As previously, we hereafter write $V_{\mathfrak{g}}$ for the Lie algebra of fundamental vector fields of the action of ${\rm Aut}(\mathfrak{g})$ on $\Lambda^2\mathfrak{g}$ and $\mathscr{E}_\mathfrak{g}$ stands for the distribution on $\Lambda^2\mathfrak{g}$ spanned  by the vector fields of $V_{\mathfrak{g}}$. Recall that $\Lambda^2_R\mathfrak{g}:=\Lambda^2\mathfrak{g}/(\Lambda^2\mathfrak{g})^\mathfrak{g}$ and $\pi_\mathfrak{g}:w\in \Lambda^2\mathfrak{g}\mapsto [w]\in \Lambda^2_R\mathfrak{g}$ is the canonical projection onto the quotient space  $\Lambda_R^2\mathfrak{g}$. Finally, $\mathcal{Y}_\mathfrak{g}$ denotes the space of $r$-matrices of $\mathfrak{g}$, i.e. solutions to mCYBE.

By Lemma \ref{reduced_action}, the action of ${\rm Aut}(\mathfrak{g})$ on $\Lambda^2\mathfrak{g}$  induces an action of ${\rm Aut}(\mathfrak{g})$ on $\Lambda^2_R\mathfrak{g}$ of the form 
$$
\Psi:T\in {\rm Aut}(\mathfrak{g})\mapsto [\Lambda^2T]\in { GL}(\Lambda^2_R\mathfrak{g}), \qquad [\Lambda^2T]([w]):=[\Lambda^2T(w)],\qquad \forall w\in \Lambda^2\mathfrak{g}.
$$
In consequence,
\begin{equation}\label{Eq:Equiv}
\pi_{\mathfrak{g}}((\Lambda^2T)(w))=[\Psi(T)](\pi_{\mathfrak{g}}(w)),\qquad \forall w\in \Lambda^2\mathfrak{g},\qquad \forall T\in {\rm Aut}(\mathfrak{g}).
\end{equation}

Let $N_1, N_2$ be manifolds and denote the action of a Lie group on $N_1$ and $N_2$ by $\Phi_1:G\times N_1\rightarrow N_1$ and $\Phi_2:G\times N_2\rightarrow N_2$, respectively. W say that $\Phi_1$ and $\Phi_2$ are {\it equivariant} relative to  $\phi:N_1\rightarrow N_2$ if $\phi(\Phi_1(g,x))=\Phi_2(g,\phi(x))$ for every $g\in G$ and $x\in N_1$. It can be proved that if $\Phi_1$ and $\Phi_2$ are equivariant, the Lie algebra of fundamental vector fields of $\Phi_1$ projects via $\phi_*$ onto the Lie algebra of fundamental vector fields of $\Phi_2$ and moreover, the orbits of $\Phi_1$ project via $\phi$ onto the orbits of $\Phi_2$  (see \cite{AM87} for details). 

Expression (\ref{Eq:Equiv}) shows that the actions of ${\rm Aut}(\mathfrak{g})$ on $\Lambda^2\mathfrak{g}$ and $\Lambda^2_R\mathfrak{g}$  are equivariant relative to $\pi_\mathfrak{g}:\Lambda^2\mathfrak{g}\rightarrow \Lambda^2_R\mathfrak{g}$. Then, the vector fields of  $V_\mathfrak{g}$ project via $\pi_{\mathfrak{g}*}$ onto the Lie algebra of fundamental vector fields of the action of ${\rm Aut}(\mathfrak{g})$ on $\Lambda^2_R\mathfrak{g}$. We write $\mathscr{E}^R_{\mathfrak{g}}$ for the generalised distribution spanned by the fundamental vector fields of the action of ${\rm Aut}(\mathfrak{g})$ on $\Lambda^2_R\mathfrak{g}$. This means that the strata of the generalised distribution $\mathscr{E}_{\mathfrak{g}}$ project onto the strata of the generalised distribution $\mathscr{E}^R_\mathfrak{g}$. 

As the elements of ${\rm Aut}(\mathfrak{g})$ preserve the space $\mathcal{Y}_{\mathfrak{g}}$ due to the properties of the Schouten bracket, the orbits of ${\rm Aut}(\mathfrak{g})$ containing a solution to the mCYBE consist of solutions to the mCYBE. Similarly, the orbits of the action of ${\rm Aut}(\mathfrak{g})$ on $\Lambda^2\mathfrak{g}$ containing a solution to the CYBE consist of solutions to the CYBE. Thus, two $r$-matrices are equivalent up to Lie algebra automorphisms of $\mathfrak{g}$ if and only if they belong to the same orbit of ${\rm Aut}(\mathfrak{g})$ within $\mathcal{Y}_\mathfrak{g}$. 

Recall that  ${\rm Aut}(\mathfrak{g})/{\rm Aut}_c(\mathfrak{g})$, where ${\rm Aut}_c(\mathfrak{g})$ is the connected component of the identity of the Lie group ${\rm Aut}(\mathfrak{g})$, is discrete and therefore countable \cite{Bo05}. In other words, the Lie group ${\rm Aut}(\mathfrak{g})$ is a numerable sum (as subsets) of disjoint and connected subsets of ${\rm Aut}(\mathfrak{g})$  diffeomorphic to ${\rm Aut}_c(\mathfrak{g})$. The strata of $\mathscr{E}_\mathfrak{g}$ coincide with the orbits of ${\rm Aut}_c(\mathfrak{g})$. Therefore, the orbits of ${\rm Aut}(\mathfrak{g})$ are given by a numerable sum of  strata of $\mathscr{E}_\mathfrak{g}$. Moreover, each particular orbit of ${\rm Aut}(\mathfrak{g})$ is an immersed submanifold in $\Lambda^2\mathfrak{g}$ of a fixed dimension. Hence, each one of the orbits of ${\rm Aut}_c(\mathfrak{g})$, whose sum gives rise to an orbit of ${\rm Aut}(\mathfrak{g})$, must have the same dimension. Similarly, the orbits of ${\rm Aut}(\mathfrak{g})$ on $\Lambda^2_R\mathfrak{g}$ are given by a numerable sum (as subsets) of strata of $\mathscr{E}^R_{\mathfrak{g}}$ of the same dimension, which are orbits of the action of ${\rm Aut}_c(\mathfrak{g})$ on $\Lambda^2_R\mathfrak{g}$. 

By Proposition \ref{prop:requiv}, two coboundary cocommutators $\delta_i:v\in \mathfrak{g}\mapsto [r_i,v]\in \mathfrak{g}\wedge\mathfrak{g}$, with $i=1,2$ and $r_i\in \mathcal{Y}_\mathfrak{g}$, are equivalent under a Lie algebra automorphism of $\mathfrak{g}$ if and only if there exists $T\in {\rm Aut}(\mathfrak{g})$ such that $(\Lambda^2T)(r_1)$ and $r_2$ belong to the same equivalence class in $\Lambda^2_R\mathfrak{g}$. Consequently, two $r$-matrices are equivalent relative to the action of ${\rm Aut}(\mathfrak{g})$ on $\Lambda^2\mathfrak{g}$ if and only if they belong to the same family of strata of the distribution $\mathscr{E}_\mathfrak{g}$ giving rise to an orbit of ${\rm Aut}(\mathfrak{g})$ in $\mathcal{Y}_\mathfrak{g}$. 

The following proposition summarises previous considerations.

\begin{proposition}\label{Pro:ClassificationRmatrix} 
The strata of $\mathscr{E}_\mathfrak{g}$ coincide with the orbits of the action of ${\rm Aut}_c(\mathfrak{g})$ on $\Lambda^2 \mathfrak{g}$. Each set of $r$-matrices for $\mathfrak{g}$ that are equivalent relative to ${\rm Aut}_c(\mathfrak{g})$ corresponds to a single stratum of the generalised distribution $\mathscr{E}_\mathfrak{g}$ within $\mathcal{Y}_\mathfrak{g}$. Moreover, there exists a one-to-one correspondence between the families of equivalent (relative to ${\rm Aut}(\mathfrak{g})$) coboundary cocommutators of $\mathfrak{g}$ and the orbits of the action of ${\rm Aut}(\mathfrak{g})$ on $\pi_{\mathfrak{g}}(\mathcal{Y}_{\mathfrak{g}})\subset \Lambda^2_R\mathfrak{g}$. Every such an orbit in $\pi_{\mathfrak{g}}(\mathcal{Y}_{\mathfrak{g}})$ is the sum (as subsets) of a numerable collection of strata of the same dimension of the generalised distribution $\mathscr{E}_\mathfrak{g}^R$.
\end{proposition}

In practice, we shall obtain each orbit of ${\rm Aut}(\mathfrak{g})$ on $\mathcal{Y}_\mathfrak{g}$ as a numerable family of strata of $\mathscr{E} _\mathfrak{g}$ in $\mathcal{Y}_\mathfrak{g}$. Such orbits represent the families of $r$-matrices that are equivalent up to an element of ${\rm Aut}(\mathfrak{g})$. Then, we will derive all strata in $\mathcal{Y}_\mathfrak{g}$ that map onto the same space in $\Lambda_R^2\mathfrak{g}$. This last task will give us, along with the orbits of ${\rm Aut}(\mathfrak{g})$ on $\mathcal{Y}_\mathfrak{g}$, the equivalent classes of  coboundary cocommutators on $\mathfrak{g}$ up to ${\rm Aut}(\mathfrak{g})$. 

Finally, let us show how Darboux families can be employed to obtain the orbits of ${\rm Aut}(\mathfrak{g})$ on $\Lambda^2\mathfrak{g}$ and to identify the strata of $\mathscr{E}_\mathfrak{g}$.

\begin{proposition}\label{Pro:UseProp} 
The locus $\ell_{\mathfrak{g}}$ of a Darboux family $\mathcal{A}_\mathfrak{g}$ relative to the Lie algebra $V_\mathfrak{g}$ on $\Lambda^2\mathfrak{g}$ is a sum (as subsets) of the orbits of the action of ${\rm Aut}_c(\mathfrak{g})$ on $\Lambda^2\mathfrak{g}$ containing a point in $\ell_{\mathfrak{g}}$. If $\ell_{\mathfrak{g}}$ is a connected submanifold in $\Lambda^2\mathfrak{g}$ of dimension given by the rank of the generalised distribution $\mathscr{E}_{\mathfrak{g}}$, then $\ell_{\mathfrak{g}}$ is a separate stratum of the generalised distribution $\mathscr{E}_{\mathfrak{g}}$. 
\end{proposition}
\begin{proof} 
Let $E$ denote a single stratum of a generalised distribution $\mathcal{E}_{\mathfrak{g}}$ passing through a point $p \in \ell_{\mathcal{A}}$. Lemma \ref{Lem:locus_orb} implies that all functions of $\mathcal{A}^E_{\mathfrak{g}}$ are equal to zero on $E$. Thus, $E$ is contained in $\ell_{\mathfrak{g}}$. Since the strata of $\mathscr{E}_\mathfrak{g}$ are the orbits of ${\rm Aut}_c(\mathfrak{g})$ in $\Lambda^2\mathfrak{g}$ by Proposition \ref{Pro:ClassificationRmatrix}, it follows that $\ell_{\mathfrak{g}}$ is a sum of strata of $\mathscr{E}_\mathfrak{g}$.

If the rank of $\mathscr{E}_\mathfrak{g}$ on $\ell_\mathfrak{g}$ is equal to $\dim \ell_{\mathfrak{g}}$, where $\ell_{\mathfrak{g}}$ is assumed to be a connected submanifold, then $\ell_{\mathfrak{g}}$ is locally generated around any point $p\in \ell_{\mathfrak{g}}$ by the action of ${\rm Aut}_c(\mathfrak{g})$ on that point. Since $\ell_1$ is connected, it is wholly generated by the action of ${\rm Aut}_c(\mathfrak{g})$ and it becomes a strata of $\mathscr{E}_{\mathfrak{g}}$.
\end{proof}

Note that we are interested only in those loci of Darboux families of $V_{\mathfrak{g}}$ contained in the space $\mathcal{Y}_\mathfrak{g}$ of solutions to the mCYBE of $\mathfrak{g}$.

\section{Study of coboundary Lie bialgebras on four-dimensional indecomposable Lie algebras}\label{Ch:Darb_Sec:Cla}
\sectionmark{Study of coboundary Lie bialgebras on 4D indecomposable Lie alg.}

A Lie algebra $\mathfrak{g}$ is {\it decomposable} when it can be written as the direct sum of two of its proper ideals. Otherwise, we say that $\mathfrak{g}$ is {\it indecomposable}. Winternitz and \v{S}nobl classified in \cite[Part 4]{SW14} all indecomposable Lie algebras up to dimension six (see \cite[p. 217]{SW14} for comments on previous classifications) . The aim of this section is to present the classification, up to Lie algebra automorphisms, of coboundary Lie bialgebras on real four-dimensional indecomposable Lie algebras (according to the Winternitz--\v{S}nobl (W\v{S}) classification) via Darboux families. Moreover, the equivalence of $r$-matrices and solutions to CYBEs for these Lie algebras are also analysed.

We hereafter assume $\mathfrak{g}$ to be a real four-dimensional and indecomposable Lie algebra. For the sake of completeness, the structure constants for every $\mathfrak{g}$ in a basis $\{e_1,e_2,e_3,e_4\}$ are given in Table \ref{Tab:StruCons}. We hereafter endow $\Lambda^2\mathfrak{g}$  with a basis $\{e_{12}, e_{13}, e_{14}, e_{23}, e_{24}, e_{34} \}$ with $e_{ij} := e_1 \wedge e_j$, while its dual basis is denoted by $\{x_1,x_2,x_3,x_4,x_5,x_6\}$. Additionally, $\Lambda^3\mathfrak{g}$ is endowed with a basis $\{e_{123},e_{124},e_{134},e_{234}\}$, where $e_{ijk} := e_i \wedge e_j \wedge e_k$. 
The classes of Lie algebras in Table \ref{Tab:StruCons} may contain several non-isomorphic Lie algebras depending on several parameters, e.g. $\alpha,\beta$. To simplify the notation, we will not detail specific values of the parameters when we refer to properties of a whole class of Lie algebras or when the particular case we are discussing is known from the context.

	\begin{table}[h]
\centering		
\begin{tabular}{|c||c|c|c|c|c|c|c|}
	\hline
	Lie algebra & $[e_1,e_2]$ & $[e_1,e_3]$ & $[e_1,e_4]$ & $[e_2,e_3]$ & $[e_2,e_4]$ & $[e_3,e_4]$ & Parameters \\
	\hhline{|=#=|=|=|=|=|=|=|}
	$\mathfrak{s}_{1}$ & 0 & 0 & 0 & 0 & $-e_1$ & $-e_3$ & \\
	\hline
	$\mathfrak{s}_{2}$ & 0 & 0 & $-e_1$ & 0 & $-e_1 -e_2$ & $-e_2 -e_3$ & \\
	\hline
	\multirow{2}{*}{$\mathfrak{s}_{3}$} & \multirow{2}{*}{0} & \multirow{2}{*}{0} & \multirow{2}{*}{$-e_1$} & \multirow{2}{*}{0} & \multirow{2}{*}{$-\alpha e_2$} & \multirow{2}{*}{$-\beta e_3$} & $0 < |\beta| \leq |\alpha| \leq 1$, \\
 & & & & & & & $(\alpha, \beta) \neq (-1,-1)$  \\
	\hline
	$\mathfrak{s}_{4}$ & 0 & 0 & $-e_1$ & 0 & $-e_1 -e_2$ & $-\alpha e_3$ & $\alpha \in \rz$ \\
	\hline
	$\mathfrak{s}_{5}$ & 0 & 0 & $-\alpha e_1$ & 0 & $-\beta e_2 + e_3$ & $-e_2 -\beta e_3$ & $\alpha > 0, \beta \in \rr$ \\
	\hline
	$\mathfrak{s}_{6}$ & 0 & 0 & 0 & $e_1$ & $-e_2$ & $e_3$ & \\
	\hline
	$\mathfrak{s}_{7}$ & 0 & 0 & 0 & $e_1$ & $e_3$ & $-e_2$ & \\
	\hline
	$\mathfrak{s}_{8}$ & 0 & 0 & $-(1 + \alpha) e_1$ & $e_1$ & $-e_2$ & $-\alpha e_3$ & $\alpha \in ]-1,1] \backslash \{0\}$ \\
	\hline
	$\mathfrak{s}_{9}$ & 0 & 0 & $-2\alpha e_1$ & $e_1$ & $-\alpha e_2 + e_3$ & $-e_2 -\alpha e_3$ & $\alpha > 0$ \\
	\hline
	$\mathfrak{s}_{10}$ & 0 & 0 & $-2e_1$ & $e_1$ & $-e_2$ & $-e_2 -e_3$ & \\
	\hline
	$\mathfrak{s}_{11}$ & 0 & 0 & $-e_1$ & $e_1$ & $-e_2$ & 0 & \\
	\hline
	$\mathfrak{s}_{12}$ & 0 & $-e_1$ & $e_2$ & $-e_2$ & $-e_1$ & 0 & \\
	\hline
	$\mathfrak{n}_{1}$ & 0 & 0 & 0 & 0 & $e_1$ & $e_2$ & \\
	\hline
\end{tabular}
\caption{Structure constants for real four-dimensional indecomposable Lie algebras according to the W\v{S} classification. Such Lie algebras are denoted by $\mathfrak{s}_{4,k},\mathfrak{n}_{4,1}$ for $1\leq k\leq 12$ in the W\v{S} classification. Nevertheless, the subindex concerning the dimension of the Lie algebra will be removed in our work to simplify the notation. The letters $\mathfrak{n}$ and $\mathfrak{s}$ are used to denote nilpotent and solvable (but not nilpotent) Lie algebras, respectively. To avoid duplicities of Lie algebras, the parameters $\alpha,\beta$ of $\mathfrak{s}_3$ must satisfy additional restrictions that were not detailed explicitly in \cite{SW14} (see \cite[p. 228]{SW14}). Such restrictions will be determined in Subsection \ref{Subsec:3}.}  \label{Tab:StruCons}
	\end{table}

Finally, it is worth mentioning some physical and mathematical applications of the Lie bialgebras given in our list and their relations to some previous works. For instance, the nilpotent Lie algebra $\mathfrak{n}_1$ corresponds to the so-called Galilean Lie algebra, which is deformed  in the study of XYZ models with the quantum Galilei group appearing in  \cite{BCGST92}. Lie bialgebras defined on Lie algebras of the type $\mathfrak{n}_1$ also occur in \cite{Ne16,Op98}. The coboundary Lie bialgebras on the Lie algebra $\mathfrak{s}_6$, which corresponds to the so-called {\it harmonic oscillator algebra}, are obtained  in \cite{BH96} without studying the equivalence up to Lie algebra automorphisms. In that work, the quantisations and $R$-matrices of such a sort of Lie bialgebras are also studied. In particular, equation (3.7) in \cite{BH96} matches the mCYBE obtained in Section \ref{Se:s6}. The work \cite{BH97} proves that all Lie bialgebras on the harmonic oscillator algebra are coboundary ones. This shows that our classification of Lie bialgebras on $\mathfrak{s}_6$ finishes the classification of all Lie bialgebras on the harmonic oscillator algebra. On the other hand, the Lie algebra $\mathfrak{s}_8$ with $\alpha=-1/2$ is isomorphic to the Lie algebra $A\bar{G}_2(1)$ in \cite{Ne16}. The Lie algebra $\mathfrak{s}_7$ provides a central extension of Caley-Klein Lie algebras analysed in \cite[eq. (3.13)]{BHOS93}. Other cases of real four-dimensional indecomposable Lie bialgebras can be found in \cite{BCH00,Op00}, where, for instance, two-dimensional central extended Galilei algebras are studied. 

\subsection{General classification scheme}
Recall that by Proposition \ref{Pro:UseProp}, the locus of a Darboux family for a Lie algebra $V_{\mathfrak{s}_1}$ is the sum (as subsets) of strata of the generalised distribution $\mathscr{E}_{\mathfrak{s}_1}$ spanned by the vector fields of $V_{\mathfrak{s}_1}$. Moreover, these loci that form a connected submanifold of $\Lambda^2 \mathfrak{s}_1$, whose dimension is equal to the rank of $\mathscr{E}_{\mathfrak{s}_1}$, are the orbits of ${\rm Aut}_c(\mathfrak{s}_1)$ on $\mathcal{Y}_{\mathfrak{s}_1}$, the space of $r$-matrices for $\mathfrak{s}_1$.

For a given Darboux family $\mathcal{A}_1$ of $V_\mathfrak{g}$ on $\Lambda^2\mathfrak{g}$, let us consider its locus $\ell_1$. As we are interested in determining the strata of $\mathscr{E}_{\mathfrak{g}}$ within $\mathcal{Y}_{\mathfrak{g}}$, we look for an $\ell_1$ containing some solution to the mCYBE, namely $\ell_1\cap \mathcal{Y}_\mathfrak{g}\neq \emptyset$. Recall that  the strata of $\mathscr{E}_\mathfrak{g}$ are completely contained within the locus or completely contained off the locus of every Darboux family for $V_{\mathfrak{g}}$. This divides the strata of $\mathscr{E}_\mathfrak{g}$  on $\mathcal{Y}_\mathfrak{g}$ into those within or outside $\ell_1$. If $\ell_1\cap \mathcal{Y}_\mathfrak{g}$ is a submanifold of dimension given by the rank of $\mathscr{E}_{\mathfrak{g}}$ on it, its connected parts are the strata of $\mathscr{E}_{\mathfrak{g}}$ in $\ell_1\cap \mathcal{Y}_{\mathfrak{g}}$. Otherwise, to obtain the strata within $\ell_1\cap \mathcal{Y}_\mathfrak{g}$, we look for a new Darboux family $\mathcal{A}_2$ for $V_\mathfrak{g}$ containing the previous Darboux family $\mathcal{A}_1$ so that its new locus, $\ell_2$,  will satisfy $\ell_2\cap \mathcal{Y}_\mathfrak{g}\neq\emptyset$. This process is repeated until we obtain a Darboux family $\mathcal{A}_k$, whose locus $\ell_k$ is a submanifold of dimension equal to the rank of $\mathscr{E}_\mathfrak{g}$ on $\ell_k$. Then, the connected components of the locus $\ell_k$ are the orbits of ${\rm Aut}_c(\mathfrak{g})$ in $\ell_k \cap \mathcal{Y}_\mathfrak{g}$. 

In the next step, we determine all the strata of $\mathscr{E}_\mathfrak{g}$ on $\mathcal{Y}_\mathfrak{g}\cap (\ell_k)^c$, where $(\ell_k)^c$ stands for the complementary subset of $\ell_k$ in $\Lambda^2\mathfrak{g}$, which is always open. For that purpose, one considers a subspace of the final Darboux family $\mathcal{A}_k$ and the previous procedure is applied again iteratively to the obtained Darboux family on $\Lambda^2\mathfrak{g}\backslash\ell_k$. This leads to another family of strata of $\mathscr{E}_{\mathfrak{g}}$ and the procedure can be repeated to search for remaining strata of $\mathscr{E}_{\mathfrak{g}}$ in $\mathcal{Y}_{\mathfrak{g}}$. In the end, the process described above allows us to obtain the orbits of ${\rm Aut}_c(\mathfrak{g})$ on $\mathcal{Y}_\mathfrak{g}$ for all four-dimensional indecomposable Lie algebras $\mathfrak{g}$.

Our procedure can be represented in a diagram, which is hereafter called a {\it Darboux tree}. Every Darboux family of the above-mentioned method is described by a collection of boxes going from an edge of the diagram on the left to one of the edges on the right. The squared boxes of the type $f=0$, for a certain function $f$, give the generating functions of the Darboux family while the squared boxes of the type $f\neq 0$,  give the conditions restricting the manifold in $\Lambda^2\mathfrak{g}$ where the Darboux family and $V_{\mathfrak{g}}$ are restricted to. The connected parts of the loci of all the Darboux families represented in a Darboux tree give rise to a decomposition of $\mathcal{Y}_{\mathfrak{g}}$ into orbits of ${\rm Aut}_c(\mathfrak{g})$. We use sometimes oval boxes with additional information that helps undertand the ongoing calculations. In many cases considered in this section, bricks were employed to generate Darboux families. 

\subsection{Lie algebra $\mathfrak{s}_{1}$} \label{Sec:s1}

Using the commutation relations for $\mathfrak{s}_1$ given in Table \ref{Tab:StruCons}, one verifies directly that $(\Lambda^2\mathfrak{s}_1)^{\mathfrak{s}_1}=\langle e_{12}\rangle$ and $(\Lambda^3\mathfrak{s}_1)^{\mathfrak{s}_1}=0$.

By Remark \ref{Re:DerAlg}, the Lie algebra $V_{\mathfrak{s}_1}$ of fundamental vector fields of the natural Lie group action of ${\rm Aut}(\mathfrak{s}_{1})$ on $\Lambda^2\mathfrak{s}_{1}$ is spanned by the basis (over the reals)
\begin{equation}\label{fundv_s1}
\begin{gathered}
X_1=2x_1\partial_{x_1}+x_2\partial_{x_2}+x_3\partial_{x_3}+x_4\partial_{x_4}+x_5\partial_{x_5},\qquad X_2=x_4\partial_{x_2}+x_5\partial_{x_3},\qquad X_3= -x_5\partial_{x_1}-x_6\partial_{x_2},\\
X_4=x_3\partial_{x_1}-x_6\partial_{x_4},\qquad X_5=x_2\partial_{x_2}+x_4\partial_{x_4}+x_6\partial_{x_6},\qquad X_6=x_3\partial_{x_2}+x_5\partial_{x_4}.
\end{gathered}
\end{equation}
As $X_1,\ldots,X_6$ close on a finite-dimensional Lie algebra,  Theorem \ref{Th:StSus} shows that they span an integrable generalised distribution $\mathscr{E}_{\mathfrak{s}_1}$ on $\Lambda^2\mathfrak{s}_{1}$.

We define $M(p)$ to be a matrix whose entry $(i,j)$ is the $j$-coefficient of $X_i$  at $p\in \Lambda^2\mathfrak{s}_1$ in the basis $\partial_{x_1}|_p,\ldots,\partial_{x_6}|_p$, namely
$$
M(p):=\left(
\begin{array}{cccccc}
2x_1 & x_2 & x_3 & x_4 & x_5 & 0 \\
0 & x_4 & x_5 & 0 & 0 & 0 \\
-x_5 & -x_6 & 0 & 0 & 0 & 0 \\
x_3 & 0 & 0 & -x_6 & 0 & 0 \\
0 & x_2 & 0 & x_4 & 0 & x_6 \\
0 & x_3 & 0 & x_5 & 0 & 0 
\end{array}
\right),\qquad p:=(x_1,\ldots,x_6)\in \Lambda^2\mathfrak{s}_{1}.
$$
The rank of  $\mathscr{E}_{\mathfrak{s}_1}$ at $p\in \Lambda^2\mathfrak{s}_{1}$ is equal to the rank of $M(p)$.

For an element $r \in \Lambda^2 \mathfrak{s}_1$, we get
$$
[r,r] = 2(-x_2 x_5 + x_3 x_4 - x_4 x_5)e_{123}  -2x_5^2e_{124} + 2(x_3 -x_5)x_6e_{134} + 2x_5 x_6e_{234}.
$$
This expression can easily be derived by using the properties of the algebraic Schouten bracket and the structure constants for $\mathfrak{s}_1$. 
Since $(\Lambda^3 \mathfrak{s}_1)^{\mathfrak{s}_1} = 0$, the mCYBE and the CYBE for $\mathfrak{s}_1$ are equal and read
\begin{equation}\label{s1_cybe}
x_3 x_4 = 0, \quad x_3 x_6 = 0, \quad x_5=0.
\end{equation}

Let us start the construction of the Darboux tree for $\mathfrak{s}_1$. First, we will search for the bricks. It is immediate that the one-dimensional subspace $\mathcal{A}_1 := \langle x_5 \rangle$ gives a Darboux family. Its locus $\ell_{\mathcal{A}_1}$ belongs to $\mathcal{Y}_{\mathfrak{s}_1}$, whereas the complement $\Lambda^2 \mathfrak{s}_1 \backslash \ell_{\mathcal{A}_1}$ does not. Similarly, one verifies that $\mathcal{A}_2 := \langle x_6 \rangle$ is also a Darboux family belonging to $\mathcal{Y}_{\mathfrak{s}_1}$. By Proposition \ref{Prop:darb_sum}, the sum $\mathcal{A}_{12} := \mathcal{A}_1 + \mathcal{A}_2$ is a Darboux family, which is associated with the first two boxes in the first row of the Darboux tree.

We continue our analysis within the locus $\ell_{\mathcal{A}_{12}}$ of the Darboux family $\mathcal{A}-{12} = \langle x_5, x_6 \rangle$. Consider the subspace $\mathcal{B} = \langle x_5 x_6, x_3 \rangle$. Obviously, it gives a Darboux family and its locus $\ell_\mathcal{B}$ lies in $\mathcal{Y}_{\mathfrak{s}_1}$. For the points in $\ell_\mathcal{B}$, the matrix $M(p)$ reads
$$
M(p):=\left(
\begin{array}{cccccc}
2x_1 & x_2 & 0 & x_4 & 0 & 0 \\
0 & x_4 & 0 & 0 & 0 & 0 \\
0 & 0 & 0 & 0 & 0 & 0 \\
0 & 0 & 0 & 0 & 0 & 0 \\
0 & x_2 & 0 & x_4 & 0 & 0 \\
0 & 0 & 0 & 0 & 0 & 0 
\end{array}
\right),\qquad p:=(x_1,\ldots,x_6)\in \ell_{\mathcal{B}}.
$$
It is easy to verify that the rank of $M(p)$ for points $p \in \ell_{\mathcal{B}}$ is not constant. Thus, $\ell_{\mathcal{B}}$ still consists of orbits of ${\rm Aut}_c(\mathfrak{s}_1)$-action of different dimension and we need to search for more complex Darboux families containing $\mathcal{B}$ in order to fully determine them.

Let $\mathcal{B}_1 := \langle x_5, x_6, x_3, x_4\rangle \subset \mathcal{B}$. Then, $\mathcal{B}_1$ is a Darboux family and its locus $\ell_{\mathcal{B}_1} \subset \mathcal{Y}_{\mathfrak{s}_1}$. The inspection of the matrix $M(p)$ for $p \in \ell_{\mathcal{B}_1}$ shows that its rank is again not constant. We reach the same conclusion for $\mathcal{B}_2 := \langle x_5, x_6, x_3, x_4, x_2 \rangle \subset \mathcal{B}_1$. 

Finally, take $\mathcal{B}_3 := \langle x_5, x_6, x_3, x_4, x_2, x_1 \rangle \subset \mathcal{B}_2$. Its locus $\ell_{\mathcal{B}_3} \in \mathcal{Y}_{\mathfrak{s}_1}$. One sees immediately that $\ell_{\mathcal{B}_3}$ is a submanifold of dimension zero and moreover, ${\rm rk} M(p) = 0$ for $p \in \ell_{\mathcal{B}_3}$. Thus, $\ell_{\mathcal{B}_3}$ gives a separate orbit of the action of ${\rm Aut}_c (\mathfrak{s}_1)$ on $\Lambda^2 \mathfrak{s}_1$, which is described by the conditions in the first row of the Darboux tree. 

Let us consider the complement $\ell_{\mathcal{B}_2} \backslash \ell_{\mathcal{B}_3}$, given by the conditions $x_5 = x_6 = x_3 = x_4 = x_2 = 0$ and $x_1 \neq 0$. Obviously, it consists of two connected submanifolds of dimension one. Additionally, one checks directly that ${\rm rk} M(p) = 1$ for points in $\ell_{\mathcal{B}_2} \backslash \ell_{\mathcal{B}_3}$. Then, it follows that both connected components of $\ell_{\mathcal{B}_2} \backslash \ell_{\mathcal{B}_3}$ give rise to new orbits of ${\rm Aut}_c (\mathfrak{s}_1)$, denoted in the Darboux tree by the index I. 

In the next step, we consider the complement $\ell_{\mathcal{B}_1} \backslash \ell_{\mathcal{B}_2}$, given by the conditions $x_5 = x_6 = x_3 = x_4 = 0$ and $x_2 \neq 0$. By direct computation, one verifies that the rank of $M(p)$ for $p \in \ell_{\mathcal{B}_1} \backslash \ell_{\mathcal{B}_2}$ is not constant. This means we need to find a new Darboux family on $\ell_{\mathcal{B}_1} \backslash \ell_{\mathcal{B}_2}$. This is satisfied by the subspace $\langle x_1 \rangle$. One sees that the locus $x_1 = 0$ consists of two connected parts in $\ell_{\mathcal{B}_1} \backslash \ell_{\mathcal{B}_2}$ of dimension one. Meanwhile, ${\rm rk} M(p) = 1$ for $p \in \ell_{\mathcal{B}_1} \backslash \ell_{\mathcal{B}_2}$. Similarly, the complement $x_1 \neq 0$ in $\ell_{\mathcal{B}_1} \backslash \ell_{\mathcal{B}_2}$ has four connected components of dimension two, which equals the rank of $M(p)$ on this subspace. Both cases yield separate orbits of ${\rm Aut}_c (\mathfrak{s}_1)$ (denoted by II and III in the Darboux tree).

Similar analysis is carried out for $\ell_{\mathcal{B}} \backslash \ell_{\mathcal{B}_1}$. As the matrix $M(p)$ is not constant on $\ell_{\mathcal{B}} \backslash \ell_{\mathcal{B}_1}$, we introduce the new Darboux family $\langle x_1 \rangle$ on $\ell_{\mathcal{B}} \backslash \ell_{\mathcal{B}_1}$. One computes as in previous cases the dimension of its locus on $\ell_{\mathcal{B}} \backslash \ell_{\mathcal{B}_1}$ and verifies that the rank of $M(p)$ on this subset is the same. Analogously, one obtains that ${\rm rk} M(p) = 3$ for the complement $x_1 \neq 0$ in $\ell_{\mathcal{B}} \backslash \ell_{\mathcal{B}_1}$ and it equals the dimension of the four connected components of this subset, which give rise to the orbits of ${\rm Aut}_c (\mathfrak{s}_1)$ denoted by IV and V.

Consider the complement $\ell_{\mathcal{A}_{12}} \backslash \ell_{\mathcal{B}}$. On this subset, the subspace $\langle x_4 \rangle$ defines a Darboux family such that its locus belongs to $\mathcal{Y}_{\mathfrak{s}_1}$. Meanwhile, its complement in $\ell_{\mathcal{A}_{12}} \backslash \ell_{\mathcal{B}}$ does not, which is noted in the oval box. Moreover, ${\rm rk} M(p) = 3$ for the points in the locus and it is equal to its dimension. In consequence, only the locus gives rise to orbits of ${\rm Aut}_c (\mathfrak{s}_1)$, denoted by VI.

Finally, we are left with the last region of $\Lambda^2 \mathfrak{s}_1$ that has not been considered, namely the complement $\ell_{\mathcal{A}_1} \backslash \mathcal{A}_{12}$. We consider the Darboux family $\langle x_3 \rangle$ on this subset. It is straightforward to see that only its locus belongs to $\mathcal{Y}_{\mathfrak{s}_1}$. Further analysis for the new Darboux family $\langle x_3, x_1\rangle$ follows as in the previous cases and leads to the last two strata containing orbits of dimension four and three (denoted by VII and VIII, respectively).

The Darboux tree constructed by the procedure discussed earlier is given below.

\begin{center}
{\small
\begin{tikzpicture}[
roundnode/.style={rounded rectangle, draw=green!40, fill=green!3, very thick, minimum size=2mm},
squarednode/.style={rectangle, draw=red!30, fill=red!2, thick, minimum size=4mm}
]
\node[squarednode] (brick) at (0,0) {$x_5=0$};

\node[squarednode] (u)  at (2,0) {$x_6=0$};
\node[squarednode] (d)  at (2,-7) {$x_6 \neq 0$};

\node[squarednode] (uu)  at (4,0) {$x_3=0$};
\node[squarednode] (ud)  at (4,-5) {$x_3 \neq 0$};
\node[squarednode] (du)  at (4,-7) {$x_3=0$};
\node[roundnode] (dd)  at (4,-8) {$\stackrel{\tiny{\rm No \, solutions}}{x_3 \neq 0}$};

\node[squarednode] (uuu)  at (6,0) {$x_4=0$};
\node[squarednode] (uud)  at (6,-4) {$x_4 \neq 0$};
\node[roundnode] (udu)  at (6,-5) {$\stackrel{\tiny{\rm No \, solutions}}{x_4 \neq 0}$};
\node[squarednode] (udd)  at (6,-6) {$x_4=0$};
\node[squarednode] (duu)  at (6,-7) {$x_1=0$};
\node[squarednode] (dud)  at (6,-8) {$x_1\neq 0$};

\node[squarednode] (uuuu)  at (8,0) {$x_2=0$};
\node[squarednode] (uuud)  at (8,-2) {$x_2 \neq 0$};
\node[squarednode] (uudu)  at (8,-4) {$x_1=0$};
\node[squarednode] (uudd)  at (8,-5) {$x_1 \neq 0$};

\node[squarednode] (uuuuu)  at (10,0) {$x_1 = 0$};
\node[squarednode] (uuuud)  at (10,-1) {$x_1 \neq 0$};
\node[squarednode] (uuudu)  at (10,-2) {$x_1 = 0$};
\node[squarednode] (uuudd)  at (10,-3) {$x_1 \neq 0$};

\node[squarednode] (0) at (12,0) {0};
\node[squarednode] (I) at (12,-1) {I};
\node[squarednode] (II) at (12,-2) {II};
\node[squarednode] (III) at (12,-3) {III};
\node[squarednode] (IV) at (12,-4) {IV};
\node[squarednode] (V) at (12,-5) {V};
\node[squarednode] (VI) at (12,-6) {VI};
\node[squarednode] (VII) at (12,-7) {VII};
\node[squarednode] (VIII) at (12,-8) {VIII};

\draw[->] (brick.east) -- (u.west);
\draw[->] (brick.east) -- (d.west);

\draw[->] (u.east) -- (uu.west);
\draw[->] (u.east) -- (ud.west);
\draw[->] (d.east) -- (du.west);
\draw[->] (d.east) -- (dd.west);

\draw[->] (uu.east) -- (uuu.west);
\draw[->] (uu.east) -- (uud.west);
\draw[->] (ud.east) -- (udu.west);
\draw[->] (ud.east) -- (udd.west);
\draw[->] (du.east) -- (duu.west);
\draw[->] (du.east) -- (dud.west);

\draw[->] (uuu.east) -- (uuuu.west);
\draw[->] (uuu.east) -- (uuud.west);
\draw[->] (uud.east) -- (uudu.west);
\draw[->] (uud.east) -- (uudd.west);

\draw[->] (uuuu.east) -- (uuuuu.west);
\draw[->] (uuuu.east) -- (uuuud.west);
\draw[->] (uuud.east) -- (uuudu.west);
\draw[->] (uuud.east) -- (uuudd.west);

\end{tikzpicture}
}
\end{center}

The connected components of the loci of the Darboux families presented in the Darboux tree above give the orbits of ${\rm Aut}_c(\mathfrak{s}_1)$ on $\mathcal{Y}_{\mathfrak{s}_1}$. Results are given in Table \ref{Tab:g_orb_1}. Explicitly, the representative $r$-matrices for each orbit of ${\rm Aut}_c(\mathfrak{s}_1)$ read
\begin{align*}
&r_0 = 0, & &r^{\pm}_{I} = \pm e_{12}, & &r^{\pm}_{II} = \pm e_{13}, & &r^{\pm, \pm}_{III} = \pm e_{12} \pm e_{13}, & &r^{\pm}_{IV} = \pm e_{23}, \\ &r^{\pm, \pm}_{V} = \pm e_{12} \pm e_{23}, & &r^{\pm}_{VI} = \pm e_{14}, & &r^{\pm}_{VII} = \pm e_{34}, & &r^{\pm, \pm}_{VIII} = \pm e_{12} \pm e_{34}
\end{align*}
Let us study the equivalence of $r$-matrices up to action of the whole ${\rm Aut}(\mathfrak{s}_1)$. For that purpose, we use that the group of automorphisms of $\mathfrak{s}_1$ reads
$$
{\rm Aut}(\mathfrak{s}_1) = \left\{
\left( 
\begin{array}{cccc}
T_2^2 & T^2_1 & 0 & T^4_1 \\
0 & T^2_2 & 0 & T^4_2 \\
0 & 0 & T_3^3 & T^4_3 \\
0 & 0 & 0 & 1
\end{array}
\right): T_2^2,  T_3^3 \in \rz, \, T^2_1, T^4_1, T^4_2, T^4_3 \in \mathbb{R}
\right\}.
$$
This Lie group is easy to be derived from the structure constants of $\mathfrak{s}_1$ with the help of symbolic computation software (see Apendix \ref{App:code}). In reality, it is enough for our purposes to consider an element of each connected component of ${\rm Aut}(\mathfrak{s}_1)$. In our case, it is immediate that ${\rm Aut}(\mathfrak{s}_1)$ has four such components. One element of each connected component of ${\rm Aut}(\mathfrak{s}_1)$ and its extension  to $\Lambda^2\mathfrak{s}_1$ are given by 
\begin{equation*}
{\small T_{\lambda_1,\lambda_2}:=\left(
\begin{array}{cccc}
\lambda_1 & 0 & 0 & 0 \\
0 & \lambda_1 & 0 & 0 \\
0 & 0 & \lambda_2 & 0 \\
0 & 0 & 0 & 1
\end{array}
\right), \qquad
\Lambda^2T_{\lambda_1,\lambda_2}:=\left(
\begin{array}{cccccc}
1 & 0 & 0 & 0 &0&0 \\
0 & \lambda_1 \lambda_2 & 0 & 0 &0&0 \\
0 & 0 & \lambda_1 & 0 &0&0 \\
0 & 0 & 0 & \lambda_1 \lambda_2&0&0 \\
0 & 0 & 0 & 0&\lambda_1&0 \\
0 & 0 & 0 & 0&0&\lambda_2 \\
\end{array}
\right),\qquad \lambda_1, \lambda_2 \in \{\pm 1\}}.
\end{equation*}
The orbits of ${\rm Aut}(\mathfrak{s}_1)$ in $\mathcal{Y}_{\mathfrak{s}_1}$ are given by the action of $\Lambda^2T_{\lambda_1,\lambda_2}$ on the orbits of ${\rm Aut}_c(\mathfrak{s}_1)$ on $\mathcal{Y}_{\mathfrak{s}_1}$.
By using all mappings $\Lambda^2T_{\lambda_1,\lambda_2}$, we can verify whether some of the orbits of ${\rm Aut}_c(\mathfrak{s}_1)$ on $\mathcal{Y}_{\mathfrak{s}_1}$ can be  connected among themselves by a Lie algebra automorphism of $\mathfrak{s}_1$. Our results are summarised in Table \ref{Tab:g_orb_1}.

Recall that $(\Lambda^2\mathfrak{s}_1)^{\mathfrak{s}_1}=\langle e_{12}\rangle$. Consequently, all orbits of ${\rm Aut}(\mathfrak{s}_1)$ on $\mathcal{Y}_{\mathfrak{s}_1}$ mapping onto the same space in $\Lambda^2_R\mathfrak{s}_1$ via $\pi_{\mathfrak{s}_1}$ give equivalent coboundary coproducts up to the action of elements of ${\rm Aut}(\mathfrak{s}_1)$. In particular, we get five classes of inequivalent cocommutators induced by the following representative $r$-matrices:
$$
r^{\pm}_{I} = \pm e_{12}, \quad r_{II} = e_{13}, \quad r^{+}_{IV} = e_{23}, \quad r^{+}_{VI} = e_{14}, \quad r^{+}_{VII} = e_{34}
$$
Remarkably, all given $r$-matrices are solutions to the CYBE.

\subsection{Lie algebra $\mathfrak{s}_{2}$}
Let us apply the formalism given in the previous section to  Lie algebra $\mathfrak{s}_2$. Using the commutators for the basis elements of $\mathfrak{s}_2$ given in Table \ref{Tab:StruCons}, one verifies directly that $(\Lambda^2\mathfrak{s}_2)^{\mathfrak{s}_2}=0$ and $(\Lambda^3\mathfrak{s}_2)^{\mathfrak{s}_2}=0$.

By Remark \ref{Re:DerAlg}, it follows that $$\mathfrak{der}(\mathfrak{s}_2)=
\left\{\left(
\begin{array}{cccc}
\mu_{11} & \mu_{12} & \mu_{13} & \mu_{14} \\
0 & \mu_{11} & \mu_{12} & \mu_{24} \\
0 & 0 & \mu_{11} & \mu_{34} \\
0 & 0 & 0 & 0
\end{array}
\right): \mu_{11}, \mu_{12}, \mu_{13}, \mu_{14}, \mu_{24}, \mu_{34} \in \mathbb{R}\right\}.
$$
By lifting these derivations to $\Lambda^2\mathfrak{s}_2$, one obtains the following basis of $V_{\mathfrak{s}_2}$, the space of fundamental vector fields of the action of ${\rm Aut}(\mathfrak{s}_2)$ on $\Lambda^2\mathfrak{s}_2$:
\begin{equation*}
\begin{gathered} 
X_1 = 2x_1 \partial_{x_1} + 2x_2 \partial_{x_2} + x_3 \partial_{x_3} + 2x_4 \partial_{x_4} + x_5 \partial_{x_5} +  x_6 \partial_{x_6}, \quad
X_2 = x_2 \partial_{x_1} + x_4 \partial_{x_2} + x_5 \partial_{x_3} + x_6 \partial_{x_5}, \\
X_3 = -x_4 \partial_{x_1} + x_6 \partial_{x_3}, \quad
X_4 = -x_5 \partial_{x_1} - x_6 \partial_{x_2}, \quad
X_5 = x_3 \partial_{x_1} - x_6 \partial_{x_4}, \quad
X_6 = x_3 \partial_{x_2} + x_5 \partial_{x_4}.
\end{gathered}
\end{equation*}

Meanwhile, for $r\in \Lambda^2\mathfrak{s}_2$, one has
$$
[r,r]=2(2x_1x_6 - 2x_2x_5 + x_2x_6 + 2x_3x_4 - x_4x_5)e_{123} + 2(x_3x_6 - x_5^2)e_{124} - 2x_5x_6e_{134} -2x_6^2e_{234}.
$$

Since $(\Lambda^3 \mathfrak{s}_2)^{\mathfrak{s}_2}=0$, the  mCYBE and the CYBE are the same and read
$$
x_3x_4 = 0, \quad x_5 = 0, \quad x_6 = 0.
$$

With the previous information, we are  ready to obtain the classification of orbits of the action of Aut$_c(\mathfrak{s}_2)$ on the subset of $r$-matrices $\mathcal{Y}_{\mathfrak{s}_2}\subset \Lambda^2\mathfrak{s}_2$ by using Darboux families. We shall start by using bricks. The only brick of $\mathfrak{s}_2$ is $x_6$. The procedure is accomplished as in the previous section and it is summarised in the following Darboux tree. The resulting orbits of ${\rm Aut}_c(\mathfrak{s}_2)$ are also detailed in Table \ref{Tab:g_orb_1}.

\begin{center}
 {\small
\begin{tikzpicture}[
roundnode/.style={rounded rectangle, draw=green!40, fill=green!3, very thick, minimum size=2mm},
squarednode/.style={rectangle, draw=red!30, fill=red!2, thick, minimum size=4mm}
]
\node[squarednode] (brick) at (0,0) {$x_6=0$};

\node[squarednode] (u)  at (2,0) {$x_5=0$};
\node[roundnode] (d)  at (2,-5) {$\stackrel{\tiny{\rm No\,\, solutions}}{x_5 \neq 0}$};

\node[squarednode] (uu)  at (4,0) {$x_4=0$};
\node[squarednode] (ud)  at (4,-4) {$x_4 \neq 0$};

\node[squarednode] (uuu)  at (6,0) {$x_3=0$};
\node[squarednode] (uud)  at (6,-3) {$x_3 \neq 0$};
\node[squarednode] (udu)  at (6,-4) {$x_3=0$};
\node[roundnode] (udd)  at (6,-5) {$\stackrel{\tiny{\rm No\,\, solutions}}{x_3 \neq 0}$};

\node[squarednode] (uuuu)  at (8,0) {$x_2=0$};
\node[squarednode] (uuud)  at (8,-2) {$x_2 \neq 0$};

\node[squarednode] (uuuuu)  at (10,0) {$x_1 = 0$};
\node[squarednode] (uuuud)  at (10,-1) {$x_1 \neq 0$};

\node[squarednode] (0) at (12,0) {0};
\node[squarednode] (I) at (12,-1) {I};
\node[squarednode] (II) at (12,-2) {II};
\node[squarednode] (III) at (12,-3) {III};
\node[squarednode] (IV) at (12,-4) {IV};

\draw[->] (brick.east) -- (u.west);
\draw[->] (brick.east) -- (d.west);

\draw[->] (u.east) -- (uu.west);
\draw[->] (u.east) -- (ud.west);

\draw[->] (uu.east) -- (uuu.west);
\draw[->] (uu.east) -- (uud.west);
\draw[->] (ud.east) -- (udu.west);
\draw[->] (ud.east) -- (udd.west);

\draw[->] (uuu.east) -- (uuuu.west);
\draw[->] (uuu.east) -- (uuud.west);

\draw[->] (uuuu.east) -- (uuuuu.west);
\draw[->] (uuuu.east) -- (uuuud.west);

\end{tikzpicture}}
\end{center}

As  commented in Subsection \ref{Sec:s1}, the extension to $\Lambda^2\mathfrak{s}_2$ of a single element of each connected component of Aut$(\mathfrak{s}_2)$ is enough to obtain the orbits of ${\rm Aut}(\mathfrak{s}_2)$ on $\mathcal{Y}_{\mathfrak{s}_2}$.
The automorphisms of the Lie algebra $\mathfrak{s}_{2}$, obtained with the help of symbolic computation software (see Appendix \ref{App:code}), read
$$
{\rm Aut}(\mathfrak{s}_2)=\left\{\left(
\begin{array}{cccc}
T_2^2 & T^3_2 & T^3_1 & T^4_1 \\
0 & T^2_2 & T^3_2 & T^4_2 \\
0 & 0 & T^2_2 & T^4_3 \\
0 & 0 & 0 & 1
\end{array}
\right): T^2_2 \in \rz, \, T_{1}^3,T_2^3,T_1^4,T_2^4,T_3^4 \in \mathbb{R}\right\}.
$$
Then, one element of ${\rm Aut}(\mathfrak{s}_2)$ for each of its connected components and their extensions to $\Lambda^2\mathfrak{s}_2$ read
\begin{equation*}
T_\lambda:=\left(
\begin{array}{cccc}
\lambda & 0 & 0 & 0 \\
0 & \lambda & 0 & 0 \\
0 & 0 & \lambda & 0 \\
0 & 0 & 0 & 1
\end{array}
\right), \qquad \Lambda^2T_\lambda={\small
\left(
\begin{array}{cccccc}
1 & 0 & 0 & 0 &0&0 \\
0 & 1 & 0 & 0 &0&0 \\
0 & 0 & \lambda & 0 &0&0 \\
0 & 0 & 0 & 1&0&0 \\
0 & 0 & 0 & 0&\lambda&0 \\
0 & 0 & 0 & 0&0&\lambda \\
\end{array}
\right)},\qquad \lambda \in \{\pm 1\}.
\end{equation*}
The connected parts of the submanifolds $0$, I$_\pm$, II$_\pm$, III, IV$_\pm$ are the orbits of Aut$_c(\mathfrak{s}_2)$ on $\mathcal{Y}_{\mathfrak{s}_2}$. 
It is simple to see which of them are further connected through an element of  Aut$(\mathfrak{s}_2)$. In result,
the orbits of Aut$(\mathfrak{s}_2)$ on $\mathcal{Y}_{\mathfrak{s}_2}$ are given by the eight submanifolds $0,{\rm I}_-,{\rm I}_+,{\rm II}_-,{\rm II}_+, {\rm III}, {\rm IV}_-, {\rm IV}_+$ given in Table \ref{Tab:g_orb_1}.  Since $(\Lambda^2\mathfrak{s}_2)^{\mathfrak{s}_2}=0$, each such a submanifold gives a family of equivalent coboundary coproducts that is non-equivalent to $r$-matrices within remaining submanifolds. This gives all possible classes of non-equivalent coboundary coproducts. As previously, all $r$-matrices are solutions to the CYBE in $\mathfrak{s}_2$.

\subsection{Lie algebra $\mathfrak{s}_{3}$}\label{Subsec:3}

The structure constants of the Lie algebras of type $\mathfrak{s}_3$ are given in Table \ref{Tab:StruCons}. It was noted in \cite[p. 228]{SW14}  that additional restrictions must be imposed on the parameters $\alpha,\beta$ given in Table \ref{Tab:StruCons} to avoid the repetition of isomorphic Lie algebras within the Lie algebra class $\mathfrak{s}_3$. Such restrictions were not explicitly detailed in \cite{SW14}, but it is immediate that isomorphic cases within Lie algebras of the class $\mathfrak{s}_3$ can be classified by the adjoint action of $e_4$ on $\langle e_1,e_2,e_3\rangle$. In particular, if $|\alpha|=|\beta|$, then the Lie algebras with parameters $(\alpha,\beta)$ and $(\beta,\alpha)$ are isomorphic relative to the Lie algebra isomorphism that interchanges $e_2$ with $e_3$ and leaves $e_1$ invariant. 
Due to the previous isomorphism, we restrict ourselves to the case $\alpha\geq \beta$ when $|\alpha|=|\beta|$.

Using the structure constants given in Table \ref{Tab:StruCons}, one gets that $(\Lambda^2 \mathfrak{s}_3)^{\mathfrak{s}_3} \subset \{e_{12}, e_{13}, e_{23}\}$. More specifically, $e_{12} \in (\Lambda^2 \mathfrak{s}_3)^{\mathfrak{s}_3}$ for $\alpha = -1$, while $e_{13} \in (\Lambda^2 \mathfrak{s}_3)^{\mathfrak{s}_3}$ for $\beta = -1$. Finally, $e_{23} \in (\Lambda^2 \mathfrak{s}_3)^{\mathfrak{s}_3}$ for $\alpha + \beta = 0$. Moreover, $(\Lambda^3\mathfrak{s}_3)^{\mathfrak{s}_3}=\langle e_{123}\rangle$ if $\alpha+\beta=-1$. Otherwise, $(\Lambda^3\mathfrak{s}_3)^{\mathfrak{s}_3}=0$.

Since the Lie algebra class $\mathfrak{s}_3$ contains so many subcases that relevant properties may change from one to another subcase with given parameters $(\alpha,\beta)$, e.g. the structure of the Lie algebra automorphism group, we develop here a modification of our classification scheme.

First, we apply the Darboux family method to the fundamental vector fields of the action of the Lie group ${\rm Aut}_{\rm all}(\mathfrak{s}_3)$ of common Lie algebra automorphisms for all parameters $\alpha,\beta$. In the next step, we analyse its relation to the Lie algebra automorphisms for each case, $\mathfrak{s}_{3}^{\alpha,\beta}$, to obtain our final classification. We omit writing the parameters $\alpha,\beta$ in $\mathfrak{s}_{3}^{\alpha,\beta}$ whenever their values are irrelevant for the argument or if this information was communicated differently.

By Remark \ref{Re:DerAlg}, one obtains that $\mathfrak{der}_{\rm all}(\mathfrak{s}_3)$, the Lie algebra of derivations of $\mathfrak{s}_3$  common for all the values of $\alpha$ and $\beta$, read:
$$
\mathfrak{der}_{\rm all}(\mathfrak{s}_3)=\left\{\left(
\begin{array}{cccc}
\mu_{11} & 0 & 0 & \mu_{14} \\
0 & \mu_{22} & 0 & \mu_{24} \\
0 & 0 & \mu_{33} & \mu_{34} \\
0 & 0 & 0 & 0
\end{array}
\right): \mu_{11}, \mu_{22}, \mu_{33}, \mu_{14}, \mu_{24}, \mu_{34} \in \mathbb{R}\right\},
$$
By lifting the elements of a basis of $\mathfrak{der}_{\rm all}(\mathfrak{s}_3)$ to $\Lambda^2\mathfrak{s}_3$, we obtain a basis of the Lie algebra $V^{\rm all}_{\mathfrak{s}_3}$ of fundamental vector fields of the action of the Lie algebra automorphism group ${\rm Aut}_{\rm all}(\mathfrak{s}_3)$ common to all $\alpha,\beta$ acting on $\Lambda^2\mathfrak{s}_3$, namely 
\begin{align*}
&X_1 = x_1 \partial_{x_1} + x_2 \partial_{x_2} + x_3 \partial_{x_3}, &
&X_2 = - x_5 \partial_{x_1} - x_6 \partial_{x_2}, &
&X_3 = x_1 \partial_{x_1} + x_4 \partial_{x_4} + x_5 \partial_{x_5}, \\
&X_4 = x_3 \partial_{x_1} - x_6 \partial_{x_4}, &
&X_5 = x_2 \partial_{x_2} + x_4 \partial_{x_4} + x_6 \partial_{x_6}, &
&X_6 = x_3 \partial_{x_2} + x_5 \partial_{x_4}.
\end{align*}
Interestingly, for $\alpha=\beta$ or if one on the coefficients $\alpha,\beta$ is equal to one, the Lie algebra of derivations of the particular $\mathfrak{s}^{\alpha,\beta}_3$ is larger, since there exists Lie algebra automorphisms of $\mathfrak{s}^{\alpha,\beta}_3$ leaving invariant the eigenspaces of ${\rm ad}_{e_4}$ acting on $\langle e_1,e_2,e_3\rangle$. These cases are discussed separately at the end of this subsection. 

Using $\mathfrak{der}_{\rm all}(\mathfrak{s}_3)$, we shall derive the orbits of the connected part of the identity element of the group ${\rm Aut}_{\rm all}(\mathfrak{s}_3)$ on each $\mathcal{Y}_{\mathfrak{s}^{\alpha,\beta}_3}$ via Darboux families. The final classification of inequivalent $r$-matrices up to the action of ${\rm Aut}(\mathfrak{s}^{\alpha,\beta}_3)$ for each pair $(\alpha,\beta)$ will be obtained by using the action of elements of ${\rm Aut}(\mathfrak{s}^{\alpha,\beta}_3)$ not contained in ${\rm Aut}_{\rm all,c}(\mathfrak{s}_3)$ for each particular pair of parameters $(\alpha,\beta)$. 

For an element $r \in \Lambda^2 \mathfrak{s}_3$, we get
$$
[r,r] = 2[(1 + \alpha)x_1 x_6 - (1 + \beta)x_2 x_5 + (\alpha + \beta)x_3 x_4]e_{123} + 2(\alpha - 1)x_3 x_5e_{124} + 2(\beta - 1)x_3 x_6e_{134} + 2(\beta - \alpha)x_5 x_6e_{234}.
$$

For $1 + \alpha + \beta \neq 0$, we have $(\Lambda^3 \mathfrak{s}^{\alpha,\beta}_3)^{\mathfrak{s}^{\alpha,\beta}_3} = 0$. Then, the mCYBE and the CYBE are the same in this case and they read
$$
(1 + \alpha)x_1 x_6 - (1 + \beta)x_2 x_5 + (\alpha + \beta)x_3 x_4 = 0, \quad (\alpha - 1)x_3 x_5 = 0, \quad (\beta - 1)x_3 x_6 = 0, \quad (\beta - \alpha)x_5 x_6 = 0.
$$

If $1 + \alpha + \beta = 0$, then $(\Lambda^3 \mathfrak{s}^{\alpha,-1-\alpha}_3)^{\mathfrak{s}^{\alpha,-1-\alpha}_3} = \langle e_{123}\rangle$ and the mCYBE reads
$$
(\alpha - 1)x_3 x_5 = 0, \quad (\beta - 1)x_3 x_6 = 0, \quad (\beta - \alpha)x_5 x_6 = 0.
$$

Since $x_3, x_5, x_6$ are   bricks of $\mathfrak{s}^{\alpha,\beta}_3$ for every pair $(\alpha,\beta)$, our Darboux tree for the Darboux families of  $V_{\mathfrak{s}_3}^{\rm all}$ starts with cases $x_i = 0$ and $x_i \neq 0$, $i \in \{3,5,6\}$. The full Darboux tree is presented below.

\begin{center}
{\small
\begin{tikzpicture}[
roundnode/.style={rounded rectangle, draw=green!40, fill=green!3, very thick, minimum size=2mm},
squarednode/.style={rectangle, draw=red!30, fill=red!2, thick, minimum size=4mm}
]
\node[squarednode] (brick) at (0,0) {$x_6 = 0$};

\node[squarednode] (u)  at (2,0) {$x_5=0$};
\node[squarednode] (d)  at (2,-10) {$x_5 \neq 0$};

\node[squarednode] (uu)  at (4,0) {$x_3=0$};
\node[squarednode] (ud)  at (4,-8) {$x_3 \neq 0$};
\node[squarednode] (du)  at (4,-10) {$x_3=0$};
\node[squarednode] (dd)  at (4,-12) {$\stackrel{\alpha = 1}{x_3 \neq 0}$};

\node[squarednode] (uuu)  at (6,0) {$x_4=0$};
\node[squarednode] (uud)  at (6,-4) {$x_4 \neq 0$};
\node[squarednode] (udu)  at (6,-8) {$x_4=0$};
\node[squarednode] (udd)  at (6,-9) {$\stackrel{\alpha + \beta \in \{0, -1\}}{x_4 \neq 0}$};
\node[squarednode] (duu)  at (6,-10) {$x_2=0$};
\node[squarednode] (dud)  at (6,-11) {$\stackrel{\beta = -1 \, {\rm or} \, \alpha + \beta = -1}{x_2 \neq 0}$};
\node[squarednode] (ddu)  at (7,-12) {$x_3 x_4 - x_2 x_5=0$};
\node[squarednode] (ddd)  at (7,-13) {$\stackrel{\beta  = -1}{x_3 x_4 - x_2 x_5 \neq 0}$};

\node[squarednode] (uuuu)  at (8,0) {$x_2=0$};
\node[squarednode] (uuud)  at (8,-2) {$x_2 \neq 0$};
\node[squarednode] (uudu)  at (8,-4) {$x_2=0$};
\node[squarednode] (uudd)  at (8,-6) {$x_2 \neq 0$};

\node[squarednode] (uuuuu)  at (10,0) {$x_1=0$};
\node[squarednode] (uuuud)  at (10,-1) {$x_1 \neq 0$};
\node[squarednode] (uuudu)  at (10,-2) {$x_1=0$};
\node[squarednode] (uuudd)  at (10,-3) {$x_1 \neq 0$};
\node[squarednode] (uuduu)  at (10,-4) {$x_1=0$};
\node[squarednode] (uudud)  at (10,-5) {$x_1 \neq 0$};
\node[squarednode] (uuddu)  at (10,-6) {$x_1=0$};
\node[squarednode] (uuddd)  at (10,-7) {$x_1 \neq 0$};

\node[squarednode] (0) at (12,0) {0};
\node[squarednode] (I) at (12,-1) {I};
\node[squarednode] (II) at (12,-2) {II};
\node[squarednode] (III) at (12,-3) {III};
\node[squarednode] (IV) at (12,-4) {IV};
\node[squarednode] (V) at (12,-5) {V};
\node[squarednode] (VI) at (12,-6) {VI};
\node[squarednode] (VII) at (12,-7) {VII};
\node[squarednode] (VIII) at (12,-8) {VIII};
\node[squarednode] (IX) at (12,-9) {${\rm IX}_{\alpha + \beta \in \{0, -1\}}$};
\node[squarednode] (X) at (12,-10) {X};
\node[squarednode] (XI) at (12,-11) {${\rm XI}^{\beta = -1}_{\alpha + \beta = -1}$};
\node[squarednode] (XII) at (12,-12) {${\rm XII}_{\alpha = 1}$};
\node[squarednode] (XIII) at (12,-13) {${\rm XIII}_{\alpha = -\beta = 1}$};

\draw[->] (brick.east) -- (u.west);
\draw[->] (brick.east) -- (d.west);

\draw[->] (u.east) -- (uu.west);
\draw[->] (u.east) -- (ud.west);
\draw[->] (d.east) -- (du.west);
\draw[->] (d.east) -- (dd.west);

\draw[->] (uu.east) -- (uuu.west);
\draw[->] (uu.east) -- (uud.west);
\draw[->] (ud.east) -- (udu.west);
\draw[->] (ud.east) -- (udd.west);
\draw[->] (du.east) -- (duu.west);
\draw[->] (du.east) -- (dud.west);
\draw[->] (dd.east) -- (ddu.west);
\draw[->] (dd.east) -- (ddd.west);

\draw[->] (uuu.east) -- (uuuu.west);
\draw[->] (uuu.east) -- (uuud.west);
\draw[->] (uud.east) -- (uudu.west);
\draw[->] (uud.east) -- (uudd.west);

\draw[->] (uuuu.east) -- (uuuuu.west);
\draw[->] (uuuu.east) -- (uuuud.west);
\draw[->] (uuud.east) -- (uuudu.west);
\draw[->] (uuud.east) -- (uuudd.west);
\draw[->] (uudu.east) -- (uuduu.west);
\draw[->] (uudu.east) -- (uudud.west);
\draw[->] (uudd.east) -- (uuddu.west);
\draw[->] (uudd.east) -- (uuddd.west);

\end{tikzpicture}
}
\end{center}

\begin{center}
{\small
\begin{tikzpicture}[
roundnode/.style={rounded rectangle, draw=green!40, fill=green!3, very thick, minimum size=2mm},
squarednode/.style={rectangle, draw=red!30, fill=red!2, thick, minimum size=4mm}
]
\node[squarednode] (brick) at (0,0) {$x_6 \neq 0$};

\node[squarednode] (u)  at (2,0) {$x_5=0$};
\node[squarednode] (d)  at (2,-4) {$\stackrel{\alpha = \beta}{x_5 \neq 0}$};

\node[squarednode] (uu)  at (4,0) {$x_3=0$};
\node[squarednode] (ud)  at (4,-2) {$\stackrel{\beta = 1}{x_3 \neq 0}$};
\node[squarednode] (du)  at (4,-4) {$x_3=0$};
\node[squarednode] (dd)  at (4,-6) {$\stackrel{\alpha = 1}{x_3 \neq 0}$};

\node[squarednode] (uuu)  at (6,0) {$x_1=0$};
\node[squarednode] (uud)  at (6,-1) {$\stackrel{\alpha = -1 \, {\rm or} \, \alpha + \beta = -1}{x_1 \neq 0}$};
\node[squarednode] (udu)  at (7,-2) {$x_3 x_4 + x_1 x_6=0$};
\node[roundnode] (udd)  at (7,-3) {$\stackrel{{\rm Isomorphic\, to\, XIII}_{\alpha=-\beta=1}}{\alpha = -1,x_3 x_4 + x_1 x_6 \neq 0}$};
\node[squarednode] (duu)  at (7,-4) {$x_1 x_6 - x_2 x_5=0$};
\node[squarednode] (dud)  at (7,-5) {$\stackrel{\alpha=\beta=-1/2}{x_1 x_6 - x_2 x_5 \neq 0}$};
\node[squarednode] (ddu)  at (8,-6) {$x_1 x_6 - x_2 x_5 + x_3 x_4=0$};
\node[roundnode] (ddd)  at (8,-7) {$\stackrel{{\rm No \, solutions}}{x_1 x_6 - x_2 x_5 + x_3 x_4 \neq 0}$};

\node[squarednode] (0) at (12,0) {XIV};
\node[squarednode] (I) at (12,-1) {${\rm XV}^{\alpha = -1}_{\alpha + \beta = -1}$};
\node[squarednode] (II) at (12,-2) {${\rm XVI}_{\beta = 1}$};
\node[squarednode] (IV) at (12,-4) {${\rm XVII}_{\alpha = \beta}$};
\node[squarednode] (V) at (12,-5) {${\rm XVIII}_{\alpha =\beta=-1/2}$};
\node[squarednode] (VI) at (12,-6) {${\rm XIX}_{\alpha = \beta=1}$};

\draw[->] (brick.east) -- (u.west);
\draw[->] (brick.east) -- (d.west);

\draw[->] (u.east) -- (uu.west);
\draw[->] (u.east) -- (ud.west);
\draw[->] (d.east) -- (du.west);
\draw[->] (d.east) -- (dd.west);

\draw[->] (uu.east) -- (uuu.west);
\draw[->] (uu.east) -- (uud.west);
\draw[->] (ud.east) -- (udu.west);
\draw[->] (ud.east) -- (udd.west);
\draw[->] (du.east) -- (duu.west);
\draw[->] (du.east) -- (dud.west);
\draw[->] (dd.east) -- (ddu.west);
\draw[->] (dd.east) -- (ddd.west);

\end{tikzpicture}
}
\end{center}

We now study four subcases: a) $
\alpha=\beta=1$, b) $\alpha=1\neq \beta$, c) $\alpha=\beta\neq 1$, d) remaining non-isomorphic cases.

{\bf Case d)}: In this case, ${\rm ad}_{e_4}$ acts on $\langle e_1,e_2,e_3\rangle$ having three different eigenvalues and this leads to the fact that the only derivations are those common to all $\alpha,\beta$, i.e. $\mathfrak{der}(\mathfrak{s}^{\alpha,\beta}_3)=\mathfrak{der}_{\rm all}(\mathfrak{s}_3)$. The connected parts of the loci of the Darboux families of the above Darboux tree can be found in Table \ref{Tab:g_orb_2}. Such connected parts are the orbits of ${\rm Aut}_{c}(\mathfrak{s}^{\alpha,\beta}_3)$ in $\mathcal{Y}_{\mathfrak{s}^{\alpha,\beta}_3}$ for $\alpha \neq 1$ and $\beta \neq 1$. Let us obtain the classification of $r$-matrices up to elements of ${\rm Aut}(\mathfrak{s}^{\alpha,\beta}_3)$. The Lie algebra automorphism group reads
$$
{\rm Aut}(\mathfrak{s}^{\alpha,\beta}_3) = \left\{ \left(
\begin{array}{cccc}
T^1_1 & 0 & 0 & T^4_1 \\
0 & T^2_2 & 0 & T^4_2 \\
0 & 0 & T^3_3 & T^4_3 \\
0 & 0 & 0 & 1
\end{array}
\right):  T^1_1, T^2_2, T^3_3 \in \rz, T^4_1, T^4_2, T^4_3 \in \rr \right\}.
$$
As in previous sections, we only need for our purposes one element of each connected part of ${\rm Aut} (\mathfrak{s}^{\alpha,\beta}_3)$. An element of each connected component and its lift to $\Lambda^2\mathfrak{s} _3$ are given by
\begin{equation*}
T_{\lambda_1,\lambda_2,\lambda_3}:=\left(
\begin{array}{cccc}
\lambda_1 & 0 & 0 & 0 \\
0 & \lambda_2 & 0 & 0 \\
0 & 0 & \lambda_3 & 0 \\
0 & 0 & 0 & 1
\end{array}
\right), \quad \Lambda^2T_{\lambda_1,\lambda_2,\lambda_3}:={\small
\left(
\begin{array}{cccccc}
\lambda_1 \lambda_2 & 0 & 0 & 0 &0&0 \\
0 & \lambda_1 \lambda_3 & 0 & 0 &0&0 \\
0 & 0 & \lambda_1 & 0 &0&0 \\
0 & 0 & 0 & \lambda_2 \lambda_3&0&0 \\
0 & 0 & 0 & 0&\lambda_2&0 \\
0 & 0 & 0 & 0&0&\lambda_3 \\
\end{array}
\right)}
\end{equation*}

for $\lambda_1, \lambda_2, \lambda_3 \in \{\pm 1\}$. By using the maps $\Lambda^2 T_{\lambda_1, \lambda_2, \lambda_3}$, we can verify whether some of the strata of $\mathscr{E}_{\mathfrak{s} _3}$ in $\mathcal{Y}_{\mathfrak{s}_3^{\alpha,\beta}}$ are still connected by a Lie algebra automorphism of $\mathfrak{s}^{\alpha,\beta}_3$. The results are summarised in Table \ref{Tab:g_orb_1}. 

Equivalence classes of coboundary cocommutators for each $\alpha,\beta$ are obtained by identifying the orbits in Table \ref{Tab:g_orb_1} whose elements are the same up to an element of $(\Lambda^2\mathfrak{s}^{\alpha,\beta}_3)^{\mathfrak{s}^{\alpha,\beta}_3}$.
	In particular, we have the following subcases:
\begin{itemize}
    \item Case $\alpha=-1, \beta \neq 1$. Hence, $(\Lambda^2\mathfrak{s}^{-1,\beta}_3)^{\mathfrak{s}^{-1,\beta}_3}=\langle e_{12}\rangle$. By analysing Table \ref{Tab:g_orb_1}, we obtain the equivalence classes of coboundary cocommutators induced by the following representative $r$-matrices:
$$
r_{I} = e_{12}, \quad r_{II} = e_{13}, \quad r_{IV} = e_{23}, \quad r_{VI} = e_{13} + e_{23}, \quad r_{VIII} = e_{14}, \quad r_{X} = e_{24}, \quad r_{XIV} = e_{34} 
$$
\item Case $\alpha\neq -1$ and $\beta\neq -1$. Since $(\Lambda^2\mathfrak{s}^{\alpha,\beta}_3)^{\mathfrak{s}^{\alpha,\beta}_3}=0$, each orbits of ${\rm Aut}(\mathfrak{s}^{\alpha,\beta}_3)$ within $\mathcal{Y}_{\mathfrak{s}^{\alpha,\beta}_3}$ correspond to the separate class of equivalent coboundary cocommutators. 
\end{itemize}
{\bf Case c):} For $\alpha=\beta\neq 1$, we have
$$
{\rm Aut}(\mathfrak{s}^{\alpha,\alpha}_3) = \left\{ \left(
\begin{array}{cccc}
T^1_1 & 0 & 0 & T^4_1 \\
0 & T^2_2 & T_2^3 & T^4_2 \\
0 & T_3^2 & T^3_3 & T^4_3 \\
0 & 0 & 0 & 1
\end{array}
\right):  T^1_1, T^2_2T^3_3-T^2_3T^3_2 \in \rz, T_1^1,T^2_2,T^3_3,T_2^3,T_3^2,T^4_1, T^4_2, T^4_3 \in \rr \right\}.
$$
To obtain the orbits of ${\rm Aut}(\mathfrak{s}^{\alpha,\alpha}_3)$ on $\mathcal{Y}_{\mathfrak{s}^{\alpha,\alpha}_3}$ from the orbits of ${\rm Aut}_{{\rm all},c}(\mathfrak{s}_3)$, it is necessary to write ${\rm Aut}(\mathfrak{s}^{\alpha,\alpha}_3)$ as a composition of ${\rm Aut}_{{\rm all},c}(\mathfrak{s}_3)$ with certain Lie algebra automorphisms of $\mathfrak{s}^{\alpha,\alpha}_3$ so that their composition generates ${\rm Aut}(\mathfrak{s}^{\alpha,\alpha}_3)$. This can be done by using the Lie algebra automorphisms of $\mathfrak{s}_3^{\alpha,\alpha}$ of the form $
T_A:={\rm Id}\otimes A{\otimes} {\rm Id}
$, for $A\in GL(2,\mathbb{R})$. Then, $\Lambda^2T_A=A\otimes {\rm Id}\otimes (\det A) {\rm Id}\otimes A$. By taking the action of these $\Lambda^2T_A$  on the strata of the distribution spanned by $V^{\rm all}_{\mathfrak{s}_3}$ in $\mathcal{Y}_{\mathfrak{s}^{\alpha,\alpha}_3}$, we obtain the orbits of ${\rm Aut}(\mathfrak{s}^{\alpha,\alpha}_3)$ on $\mathcal{Y}_{\mathfrak{s}^{\alpha,\alpha}_3}$. Our results are summarised in Table \ref{Tab:g_orb_2}. 

Since $(\Lambda^2\mathfrak{s}^{\alpha,\alpha}_3)^{\mathfrak{s}^{\alpha,\alpha}_3}=0$, each class of equivalent $r$-matrices detailed in Table \ref{Tab:g_orb_2} gives rise to a separate equivalence class of coboundary cocommutators for any value of $\alpha=\beta$.

{\bf Case b):} For $\alpha=1\neq \beta$,
$$
{\rm Aut}(\mathfrak{s}^{1,\beta}_3) = \left\{ \left(
\begin{array}{cccc}
T^1_1 & T_1^2 & 0 & T^4_1 \\
T_2^1 & T^2_2 & 0 & T^4_2 \\
0 & 0 & T^3_3 & T^4_3 \\
0 & 0 & 0 & 1
\end{array}
\right):  T^3_3, T^1_1T^2_2-T^1_2T^2_1 \in \rz, T_1^1,T_2^2,T_1^2,T_2^1,T^4_1, T^4_2, T^4_3 \in \rr \right\}.
$$
To derive the orbits of ${\rm Aut}(\mathfrak{s}^{1,\beta}_3)$ on $\mathcal{Y}_{\mathfrak{s}^{1,\beta}_3}$ from the orbits of ${\rm Aut}_{{\rm all},c}(\mathfrak{s}_3)$, we again write ${\rm Aut}(\mathfrak{s}^{1,\beta}_3)$ as a composition of ${\rm Aut}_{{\rm all},c}(\mathfrak{s}_3)$ with certain Lie algebra automorphisms of $\mathfrak{s}^{1,\beta}_3$ so that their composition generates ${\rm Aut}(\mathfrak{s}^{1,\beta}_3)$. This can be done by employing the Lie algebra autmorphisms of $\mathfrak{s}^{1,\beta}_{3}$ given by $
T_A:=A{\otimes} {\rm Id}
\otimes {\rm Id}$ for $A\in GL(2,\mathbb{R})$. Then $\Lambda^2T_A=(\det A){\rm Id}\otimes (\tau_{43}\circ A\otimes A\circ \tau_{43})\otimes {\rm Id}$, where $\tau_{43}$ is the permutation of coordinates three and four in $\Lambda^2\mathfrak{s}^{1,\beta}_3$. Considering the action of these $\Lambda^2T_A$  on the strata of the distribution spanned by $V^{\rm all}_{\mathfrak{s}_3}$ in $\mathcal{Y}_{\mathfrak{s}^{1,\beta}_3}$, we obtain the orbits of ${\rm Aut}(\mathfrak{s}^{1,\beta}_3)$ on $\mathcal{Y}_{\mathfrak{s}^{1,\beta}_3}$. Our results are summarised in Table \ref{Tab:g_orb_2}. 

In order classify coboundary cocommutators, let us identify the equivalence classes of $r$-matrices in Table \ref{Tab:g_orb_1} whose elements are the same up to an element of $(\Lambda^2\mathfrak{s}^{1,\beta}_3)^{\mathfrak{s}^{1,\beta}_3}$. If $\beta\neq -1$, these cocommutators correspond to the orbits of ${\rm Aut}(\mathfrak{s}^{1,\beta}_3)$ on $\mathcal{Y}_{\mathfrak{s}^{1,\beta}_3}$. Otherwise, $(\Lambda^2\mathfrak{s}^{1,-1}_3)^{\mathfrak{s}^{1,-1}_3}=\langle e_{13},e_{23}\rangle$ and the cocommutators are induced by the following representative $r$-matrices:
 $$
r_{I} = 0, \quad r_{IV} = e_{14}, \quad r_{VII} = e_{34}
 $$

{\bf Case a):} In this case, due to the larger family of symmetries for $\alpha=\beta=1$, the Lie algebra $V_{\mathfrak{s}_3^{1,1}}$ is spanned by
\begin{align*}
&X_1 = x_1 \partial_{x_1} + x_2 \partial_{x_2} + x_3 \partial_{x_3}, &
&X_2 = - x_5 \partial_{x_1} - x_6 \partial_{x_2}, &
&X_3 = x_1 \partial_{x_1} + x_4 \partial_{x_4} + x_5 \partial_{x_5}, \\
&X_4 = x_3 \partial_{x_1} - x_6 \partial_{x_4}, &
&X_5 = x_2 \partial_{x_2} + x_4 \partial_{x_4} + x_6 \partial_{x_6}, &
&X_6 = x_3 \partial_{x_2} + x_5 \partial_{x_4},\\
&X_7=x_4\partial_{x_2}+x_5\partial_{x_3}, & &X_8=x_4\partial_{x_1}+x_6\partial_{x_3}, & &X_9=x_2\partial_{x_4}+x_3\partial_{x_5},\\
&X_{10}=x_2\partial_{x_1}+x_6\partial_{x_5},&
&X_{11}=-x_1\partial_{x_4}+x_3\partial_{x_6}, & &X_{12}=x_1\partial_{x_2}+x_5\partial_{x_6}.
\end{align*}

One obtains the following simple Darboux tree:
\begin{center}
{\small
\begin{tikzpicture}[
roundnode/.style={rounded rectangle, draw=green!40, fill=green!3, very thick, minimum size=2mm},
squarednode/.style={rectangle, draw=red!30, fill=red!2, thick, minimum size=4mm}
]
\node[squarednode] (brick) at (0,0) {$x_3x_4-x_2x_5+x_1x_6=0$};
\node[roundnode] (nosol) at (0,-3) {$\stackrel{{\rm No\, solution}}{x_3x_4-x_2x_5+x_1x_6\neq 0}$};
\node[squarednode] (u)  at (4,0) {$x_3=0,x_5=0,x_6=0$};
\node[squarednode] (d)  at (4,-2) {$x_3^2+x_5^2+x_6^2\neq 0$};

\node[squarednode] (uu)  at (8,0) {$x_2=0,x_1=0,x_4=0$};
\node[squarednode] (ud)  at (8,-1) {$x_2^2+x_1^2+x_4^2\neq 0$};

\node[squarednode] (0) at (12,0) {0};
\node[squarednode] (I) at (12,-1) {I};
\node[squarednode] (II) at (12,-2) {II};

\draw[->] (brick.east) -- (u.west);
\draw[->] (brick.east) -- (d.west);

\draw[->] (u.east) -- (uu.west);
\draw[->] (u.east) -- (ud.west);
\end{tikzpicture}
}
\end{center}
It is immediate to verify that the above orbits give rise to two inequivalent classes of $r$-matrices, as shown in Table \ref{Tab:g_orb_2}. Since in this case $(\Lambda^2\mathfrak{s}^{1,1}_3)^{\mathfrak{s}^{1,1}_3}=0$, each class of equivalent $r$-matrices amounts to a class of equivalent coboundary cocommutators.

\subsection{Lie algebra $\mathfrak{s}_{4}$}

Using Table \ref{Tab:g_orb_1}, one directly verifies that $(\Lambda^3\mathfrak{s}_4)^{\mathfrak{s}_4}=\langle e_{123}\rangle$ for $2+\alpha=0$ and $(\Lambda^3\mathfrak{s}_4)^{\mathfrak{s}_4}=\{0\}$ for $2+\alpha \neq 0$. Meanwhile, $(\Lambda^2\mathfrak{s}_4)^{\mathfrak{s}_4}=0$ for $\alpha\neq -1$ and $(\Lambda^2\mathfrak{s}_4)^{\mathfrak{s}_4}=\langle e_{13}\rangle$ for $\alpha=-1$.

As the Lie algebras of the class $\mathfrak{s}_4$ depend on a parameter $\alpha\in \mathbb{R}\backslash\{0\}$, the space of derivations depend on $\alpha$. It is indeed the same for all values of $\alpha\in \mathbb{R}\backslash\{0,1\}$, and it becomes larger for $\alpha=1$. In consequence, we shall proceed as in Subsection \ref{Subsec:3}. 
By Remark \ref{Re:DerAlg}, the space of common derivations of $\mathfrak{s}_4$ for all possible values of $\alpha$ read
$$
\mathfrak{der}_{\rm all}(\mathfrak{s}_4):=\left\{\left(
\begin{array}{cccc}
\mu_{11} & \mu_{12} & 0 & \mu_{14} \\
0 & \mu_{11} & 0 & \mu_{24} \\
0 & 0 & \mu_{33} & \mu_{34} \\
0 & 0 & 0 & 0
\end{array}
\right):\mu_{11}, \mu_{12}, \mu_{33}, \mu_{14}, \mu_{24}, \mu_{34} \in \mathbb{R}\right\}.
$$
By extending the previous derivations to $\Lambda^2\mathfrak{s}_4$, we obtain a basis of $V^{\rm all}_{\mathfrak{s}_4}$ of the form
\begin{align*}
&X_1 = 2x_1 \partial_{x_1} + x_2 \partial_{x_2} + x_3 \partial_{x_3} + x_4 \partial_{x_4} + x_5 \partial_{x_5}, &
&X_2 = x_4 \partial_{x_2} + x_5 \partial_{x_3}, &
&X_3 = -x_5 \partial_{x_1} - x_6 \partial_{x_2}, \\
&X_4 = x_2 \partial_{x_2} + x_4 \partial_{x_4} + x_6 \partial_{x_6}, &
&X_5 = x_3 \partial_{x_1} - x_6 \partial_{x_4}, &
&X_6 = x_3 \partial_{x_2} + x_5 \partial_{x_4}.
\end{align*}
We recall that these vector fields span the Lie algebra of fundamental vector fields of the action on $\Lambda^2\mathfrak{s}_4$ of the Lie algebra automorphisms ${\rm Aut}_{\rm all}(\mathfrak{s}_4)$ that are common for all values of $\alpha\in \mathbb{R}\backslash\{0\}$.

For an element $r \in \Lambda^2 \mathfrak{s}_4$, we get
$$
[r,r] = 2[2x_1 x_6 - (1 + \alpha)x_2 x_5 + (1 + \alpha)x_3 x_4 - x_4 x_5]e_{123} - 2x_5^2e_{124} + 2[-(1 - \alpha)x_3 x_6 - x_5 x_6]e_{134} + 2(\alpha - 1)x_5 x_6e_{234}.
$$
If $\alpha \neq -2$, we have $(\Lambda^3 \mathfrak{s}_4)^{\mathfrak{s}_4} = 0$. Thus, the mCYBE and the CYBE are equal in this case and they read
$$
2x_1 x_6 + (1 + \alpha) x_3 x_4 = 0, \quad (1 - \alpha)x_3 x_6 = 0, \quad x_5 = 0.
$$
For the case $\alpha + 2 = 0$, we have $(\Lambda^3 \mathfrak{s}_4)^{\mathfrak{s}_4} = \langle e_{123} \rangle$. Thus, the mCYBE reads
$$
\quad (1 - \alpha)x_3 x_6 = 0, \quad x_5 = 0.
$$

Since $x_5, x_6$ are the bricks for $\mathfrak{s}_4$, our Darboux tree  for the Darboux families starts with the cases $x_i  = 0$ and $x_i \neq 0$, $i \in \{5,6\}$. The full Darboux tree is presented below.

\begin{center}
{\small
\begin{tikzpicture}[
roundnode/.style={rounded rectangle, draw=green!40, fill=green!3, very thick, minimum size=2mm},
squarednode/.style={rectangle, draw=red!30, fill=red!2, thick, minimum size=4mm}
]
\node[squarednode] (brick) at (0,0) {$x_6=0$};

\node[squarednode] (u)  at (2,0) {$x_5=0$};
\node[roundnode] (d)  at (2,-7) {$\stackrel{{\rm No \, solutions}}{x_5 \neq 0}$};

\node[squarednode] (uu)  at (4,0) {$x_4=0$};
\node[squarednode] (ud)  at (4,-5) {$x_4 \neq 0$};

\node[squarednode] (uuu)  at (6,0) {$x_3=0$};
\node[squarednode] (uud)  at (6,-4) {$x_3 \neq 0$};
\node[squarednode] (udu)  at (6,-5) {$x_3=0$};
\node[squarednode] (udd)  at (6,-7) {$\stackrel{\alpha = -2 \lor \alpha = -1}{x_3 \neq 0}$};

\node[squarednode] (uuuu)  at (8,0) {$x_2=0$};
\node[squarednode] (uuud)  at (8,-2) {$x_2 \neq 0$};
\node[squarednode] (uduu)  at (8,-5) {$x_1=0$};
\node[squarednode] (udud)  at (8,-6) {$x_1 \neq 0$};

\node[squarednode] (uuuuu)  at (10,0) {$x_1=0$};
\node[squarednode] (uuuud)  at (10,-1) {$x_1 \neq 0$};
\node[squarednode] (uuudu)  at (10,-2) {$x_1=0$};
\node[squarednode] (uuudd)  at (10,-3) {$x_1 \neq 0$};

\node[squarednode] (0) at (12,0) {0};
\node[squarednode] (I) at (12,-1) {I};
\node[squarednode] (II) at (12,-2) {II};
\node[squarednode] (III) at (12,-3) {III};
\node[squarednode] (IV) at (12,-4) {IV};
\node[squarednode] (V) at (12,-5) {V};
\node[squarednode] (VI) at (12,-6) {VI};
\node[squarednode] (VII) at (12,-7) {${\rm VII}_{\alpha \in \{-2,-1\}}$};

\draw[->] (brick.east) -- (u.west);
\draw[->] (brick.east) -- (d.west);

\draw[->] (u.east) -- (uu.west);
\draw[->] (u.east) -- (ud.west);

\draw[->] (uu.east) -- (uuu.west);
\draw[->] (uu.east) -- (uud.west);
\draw[->] (ud.east) -- (udu.west);
\draw[->] (ud.east) -- (udd.west);

\draw[->] (uuu.east) -- (uuuu.west);
\draw[->] (uuu.east) -- (uuud.west);
\draw[->] (udu.east) -- (uduu.west);
\draw[->] (udu.east) -- (udud.west);

\draw[->] (uuuu.east) -- (uuuuu.west);
\draw[->] (uuuu.east) -- (uuuud.west);
\draw[->] (uuud.east) -- (uuudu.west);
\draw[->] (uuud.east) -- (uuudd.west);
\end{tikzpicture}

\vspace{1cm}

\begin{tikzpicture}[
roundnode/.style={rounded rectangle, draw=green!40, fill=green!3, very thick, minimum size=2mm},
squarednode/.style={rectangle, draw=red!30, fill=red!2, thick, minimum size=4mm}
]
\node[squarednode] (brick) at (0,0) {$x_6\neq 0$};

\node[squarednode] (u)  at (2,0) {$x_5=0$};
\node[roundnode] (d)  at (2,-2) {$\stackrel{{\rm No \, solutions}}{x_5 \neq 0}$};

\node[squarednode] (uu)  at (4,0) {$x_3=0$};
\node[squarednode] (ud)  at (4,-2) {$\stackrel{\alpha = 1}{x_3 \neq 0}$};

\node[squarednode] (uuu)  at (6,0) {$x_1=0$};
\node[squarednode] (uud)  at (6,-1) {$\stackrel{\alpha = -2}{x_1 \neq 0}$};
\node[squarednode] (udu)  at (7,-2) {$x_3 x_4 + x_1 x_6 = 0$};
\node[roundnode] (udd)  at (7,-3) {$\stackrel{{\rm No \, solutions}}{x_3 x_4 + x_1 x_6 \neq 0}$};

\node[squarednode] (VIII) at (10,0) {VIII};
\node[squarednode] (IX) at (10,-1) {${\rm IX}_{\alpha = -2}$};
\node[squarednode] (X) at (10,-2) {${\rm X}_{\alpha = 1}$};

\draw[->] (brick.east) -- (u.west);
\draw[->] (brick.east) -- (d.west);

\draw[->] (u.east) -- (uu.west);
\draw[->] (u.east) -- (ud.west);

\draw[->] (uu.east) -- (uuu.west);
\draw[->] (uu.east) -- (uud.west);
\draw[->] (ud.east) -- (udu.west);
\draw[->] (ud.east) -- (udd.west);

\end{tikzpicture}
}
\end{center}
The connected parts of the subspaces denoted in the above diagram are the orbits of ${\rm Aut}_{\rm all,c}(\mathfrak{s}_4)$ on $\mathcal{Y}_{\mathfrak{s}_4}$. 

The Lie group ${\rm Aut}_{\rm all}(\mathfrak{s}_4)$ of the common Lie algebra automorphisms of $\mathfrak{s}_4$ for all $\alpha \in \mathbb{R} \backslash \{0\}$ is given by
$$
{\rm Aut}_{\rm all}(\mathfrak{s}_4) = \left\{
\left(
\begin{array}{cccc}
T^2_2 & T^2_1 & 0 & T^4_1 \\
0 & T^2_2 & 0 & T^4_2 \\
0 & 0 & T^3_3 & T^4_3 \\
0 & 0 & 0 & 1
\end{array}
\right): T^2_2, T^3_3 \in \rz, \, T^2_1, T_1^4, T_2^4, T_3^4 \in \rr
\right\},
$$
In order to obtain the orbits of ${\rm Aut}_{\rm all}(\mathfrak{s}_4)$ on $\mathcal{Y}_{\mathfrak{s}_4}$, let us search for elements of ${\rm  Aut}_{\rm all}(\mathfrak{s}_4)$ that identify some of the connected components of the orbits of ${\rm Aut}_{\rm all, c} (\mathfrak{s}_4)$ on $\mathcal{Y}_{\mathfrak{s}_4}$. As previously, we consider one element of each connected component of ${\rm Aut}_{\rm all}(\mathfrak{s}_4)$ and their lifts to $\Lambda^2\mathfrak{s}_4$, that is 
\begin{equation*}
T_{\lambda_1,\lambda_2}:=\left(
\begin{array}{cccc}
\lambda_1 & 0 & 0 & 0 \\
0 & \lambda_1 & 0 & 0 \\
0 & 0 & \lambda_2 & 0 \\
0 & 0 & 0 & 1
\end{array}
\right), \quad \Lambda^2T_{\lambda_1,\lambda_2}={\small
\left(
\begin{array}{cccccc}
1 & 0 & 0 & 0 &0&0 \\
0 & \lambda_1 \lambda_2 & 0 & 0 &0&0 \\
0 & 0 & \lambda_1 & 0 &0&0 \\
0 & 0 & 0 & \lambda_1 \lambda_2&0&0 \\
0 & 0 & 0 & 0&\lambda_1&0 \\
0 & 0 & 0 & 0&0&\lambda_2 \\
\end{array}
\right)}, \quad \lambda_1, \lambda_2 \in \{\pm 1\}.
\end{equation*}
For $\alpha\in \mathbb{R}\backslash\{1,0\}$, one has that ${\rm Aut}(\mathfrak{s}^\alpha_4)={\rm Aut}_{\rm all}(\mathfrak{s}_4)$, where $\mathfrak{s}^\alpha_4$ stands for the Lie algebra $\mathfrak{s}_4$ for a fixed value of $\alpha$. Then, the classes of equivalent $r$-matrices (up to Lie algebra automorphisms of $\mathfrak{s}_4^\alpha$) on $\mathfrak{s}^\alpha_4$ can  easily be obtained and they are summarised in Table \ref{Tab:g_orb_1}. 

For $\alpha=-1$, we get $(\Lambda^2\mathfrak{s}_4)^{\mathfrak{s}_4}=\langle e_{13}\rangle$ and thus, the classes of equivalent coboundary cocommutators in this case are induced by the following representative $r$-matrices:
$$
r_{0} = 0, \quad r^{}_{I} = \pm e_{12}, \quad r_{IV} = e_{14}, \quad r_{V} = e_{23}, \quad r^{}_{VI} = \pm e_{12} + e_{23}, \quad r_{VII} = e_{14} + e_{23}, \quad e_{VIII} = e_{34}
$$
For those $\mathfrak{s}_{4}^\alpha$ with $\alpha \in \mathbb{R} \backslash \{-1,0,1\}$, one has $(\Lambda^2\mathfrak{s}_4^{\alpha})^{\mathfrak{s}_4^{\alpha}}=0$ and the classes of equivalent coboundary coproducts are given by each orbit of ${\rm Aut}(\mathfrak{s}^\alpha_4)$.

Note that for values $|\alpha|\neq 1$, not all the classes listed in Table \ref{Tab:g_orb_2} may be simultaneously available as they arise for particular values of $\alpha$. 

Let us consider now the case of the Lie algebra $\mathfrak{s}^1_4$, i.e. the Lie algebra of the class $\mathfrak{s}_4$ for $\alpha=1$. In this case, the group of Lie algebra automorphisms is larger than for remaining admissible values of $\alpha$. In particular,  
$$
{\rm Aut}(\mathfrak{s}^1_4) = \left\{
\left(
\begin{array}{cccc}
T^2_2 & T^2_1 & T^3_1 & T^4_1 \\
0 & T^2_2 & 0 & T^4_2 \\
0 & T_3^2 & T^3_3 & T^4_3 \\
0 & 0 & 0 & 1
\end{array}
\right): T^2_2, T^3_3 \in \rz, \, T_3^2,T^2_1,T^3_1, T_1^4, T_2^4, T_3^4 \in \rr
\right\}
$$
Moreover, one obtains ${\rm Aut}(\mathfrak{s}^1_4)$ by composing ${\rm Aut}_{\rm all}(\mathfrak{s}_4)$ with the Lie algebra automorphisms of $\mathfrak{s}_4^1$ of the form
$$
T(e_1)=e_1,\,T(e_2)=e_2+\lambda e_3,\,T(e_3)=e_3+\mu e_1,\,T(e_4)=e_4,\qquad \forall \lambda,\mu\in \mathbb{R}.
$$
Hence, we get the orbits of the action ${\rm Aut}(\mathfrak{s}_4)$ on $\mathcal{Y}_{\mathfrak{s}_4}$ by acting with $\Lambda^2T_A$ on the orbits of ${\rm Aut}_{\rm all}(\mathfrak{s}_4)$. Note that
\begin{align*}
&\Lambda^2T(e_{12})=e_{12}+\lambda e_{13}, & &\Lambda^2T(e_{13})=e_{13}, & &\Lambda^2T(e_{14})=e_{14}, \\
&\Lambda^2T(e_{23})=e_{23}-e_{12}\mu-\lambda\mu e_{13}, & &\Lambda^2T(e_{24})=e_{24}+\lambda e_{34}, & &\Lambda^2T_{34}=e_{34}+\mu e_{14},
\end{align*}
for every $\lambda,\mu\in \mathbb{R}$.
This will give the final orbits detailed in Table \ref{Tab:g_orb_2}. Since $(\Lambda^2\mathfrak{s}^1_4)^{\mathfrak{s}^1_4}=0$,  the classes of equivalent coboundary cocommutators are induced by all classes of equivalent $r$-matrices.

\subsection{Lie algebra $\mathfrak{s}_{5}$}\label{Sec:s5}

Structure constants for Lie algebra $\mathfrak{s}_5$ are given in Table \ref{Tab:StruCons}. Using this information, one obtains $(\Lambda^2 \mathfrak{s}_5)^{\mathfrak{s}_5} = \langle e_{23}\rangle$ for $\beta = 0$ and $(\Lambda^2 \mathfrak{s}_5)^{\mathfrak{s}_5} = \{0\}$ for $\beta \neq 0$. Moreover, $(\Lambda^3\mathfrak{s}_5)^{\mathfrak{s}_5}=0$ for $\alpha+2\beta\neq 0$ and $(\Lambda^3\mathfrak{s}_5)^{\mathfrak{s}_5}=\langle e_{123}\rangle$ for $\alpha+2\beta = 0$.

By Remark \ref{Re:DerAlg}, we obtain that the derivations of $\mathfrak{s}_5$ read
$$
\mathfrak{der}(\mathfrak{s}_5):=\left\{\left(
\begin{array}{cccc}
 \mu_{11} & 0 & 0 & \mu_{14} \\
0 & \mu_{22} & \mu_{23} & \mu_{24} \\
0 & -\mu_{23} & \mu_{22} & \mu_{34} \\
0 & 0 & 0 & 0
\end{array}
\right): \mu_{11}, \mu_{22}, \mu_{23}, \mu_{14}, \mu_{24}, \mu_{34} \in \mathbb{R}\right\},
$$
which give rise to the basis of $V_{\mathfrak{s}_5}$ of the form
\begin{align*}
&X_1 = x_1 \partial_{x_1} + x_2 \partial_{x_2} + 2x_4 \partial_{x_4} + x_5 \partial_{x_5} + x_6 \partial_{x_6}, &
&X_2 = -x_5 \partial_{x_1} - x_6 \partial_{x_2}, &
&X_3 = x_1 \partial_{x_1} + x_2 \partial_{x_2} + x_3 \partial_{x_3}, \\
&X_4 = x_2 \partial_{x_1} - x_1 \partial_{x_2} + x_6 \partial_{x_5} - x_5 \partial_{x_6}, &
&X_5 = x_3 \partial_{x_1} - x_6 \partial_{x_4}, &
&X_6 = x_3 \partial_{x_2} + x_5 \partial_{x_4}.
\end{align*}

For an element $r \in \Lambda^2 \mathfrak{s}_5$, we get
\begin{equation*}
\begin{split}
[r,r]&= 2[x_1 x_5 + (\alpha + \beta)x_1 x_6 - (\alpha + \beta)x_2 x_5 + x_2 x_6 + 2\beta x_3 x_4]e_{123} + 2[(\beta - \alpha)x_3 x_5 + x_3 x_6]e_{124} \\
&+ 2[-x_3 x_5 + (\beta - \alpha)x_3 x_6]e_{134} - 2(x_5^2+x_6^2)e_{234}.
\end{split}
\end{equation*}

If $\alpha + 2\beta \neq 0$, then $(\Lambda^3 \mathfrak{s}_5)^{\mathfrak{s}_5} = 0$. Thus, the mCYBE and the CYBE are equal in this case and they read
$$
\beta x_3 x_4 = 0, \quad x_5 = 0, \quad x_6 = 0.
$$

For $\alpha + 2\beta = 0$, we have $(\Lambda^3 \mathfrak{s}_5)^{\mathfrak{s}_5} = \langle e_{123} \rangle$ and the mCYBE reads $x_5 = 0, x_6 = 0$.

Since $x_3$ is the only brick for $\mathfrak{s}_5$, we start our Darboux tree with the cases $x_3 = 0$ and $x_3 \neq 0$. The full Darboux tree is presented below.

\begin{center}
{\small
\begin{tikzpicture}[
roundnode/.style={rounded rectangle, draw=green!40, fill=green!3, very thick, minimum size=2mm},
squarednode/.style={rectangle, draw=red!30, fill=red!2, thick, minimum size=4mm}
]
\node[squarednode] (brick) at (0,0) {$x_3=0$};

\node[squarednode] (u)  at (2,0) {$x_6=0$};
\node[roundnode] (d)  at (2,-3) {$\stackrel{{\rm No \, solutions}}{x_6 \neq 0}$};

\node[squarednode] (uu)  at (4,0) {$x_5=0$};
\node[roundnode] (ud)  at (4,-3) {$\stackrel{{\rm No\, solutions}}{x_5 \neq 0}$};

\node[squarednode] (uuu)  at (6,0) {$x_4=0$};
\node[squarednode] (uud)  at (6,-2) {$x_4 \neq 0$};

\node[squarednode] (uuuu)  at (8,0) {$x_1^2 + x_2^2=0$};
\node[squarednode] (uuud)  at (8,-1) {$x_1^2 + x_2^2 \neq 0$};
\node[squarednode] (uudu)  at (8,-2) {$x_1^2 + x_2^2=0$};
\node[squarednode] (uudd)  at (8,-3) {$x_1^2 + x_2^2 \neq 0$};

\node[squarednode] (0) at (12,0) {0};
\node[squarednode] (I) at (12,-1) {I};
\node[squarednode] (II) at (12,-2) {II};
\node[squarednode] (III) at (12,-3) {III};

\draw[->] (brick.east) -- (u.west);
\draw[->] (brick.east) -- (d.west);

\draw[->] (u.east) -- (uu.west);
\draw[->] (u.east) -- (ud.west);

\draw[->] (uu.east) -- (uuu.west);
\draw[->] (uu.east) -- (uud.west);

\draw[->] (uuu.east) -- (uuuu.west);
\draw[->] (uuu.east) -- (uuud.west);
\draw[->] (uud.east) -- (uudu.west);
\draw[->] (uud.east) -- (uudd.west);

\end{tikzpicture}

\vspace{1cm}

\begin{tikzpicture}[
roundnode/.style={rounded rectangle, draw=green!40, fill=green!3, very thick, minimum size=2mm},
squarednode/.style={rectangle, draw=red!30, fill=red!2, thick, minimum size=4mm}
]
\node[squarednode] (brick) at (0,0) {$x_3\neq 0$};

\node[squarednode] (u)  at (3,0) {$x_6=0$};
\node[roundnode] (d)  at (3,-1) {$\stackrel{{\rm No \, solutions}}{x_6 \neq 0}$};

\node[squarednode] (uu)  at (6,0) {$x_5=0$};
\node[roundnode] (ud)  at (6,-1) {$\stackrel{{\rm No \, solutions}}{x_5 \neq 0}$};

\node[squarednode] (uuu)  at (9,0) {$x_4=0$};
\node[squarednode] (uud)  at (9,-1) {$\stackrel{ \beta = 0 \lor \alpha + 2\beta = 0}{x_4 \neq 0}$};

\node[squarednode] (0) at (12,0) {IV};
\node[squarednode] (I) at (12,-1) {${\rm V}^{\alpha = -2\beta}_{\beta = 0}$};

\draw[->] (brick.east) -- (u.west);
\draw[->] (brick.east) -- (d.west);

\draw[->] (u.east) -- (uu.west);
\draw[->] (u.east) -- (ud.west);

\draw[->] (uu.east) -- (uuu.west);
\draw[->] (uu.east) -- (uud.west);

\end{tikzpicture}
}
\end{center}

Orbits of ${\rm Aut}_c(\mathfrak{s}_5)$, described in Table \ref{Tab:g_orb_1}, are given by the connected parts of the loci of the Darboux families detailed in the above Darboux tree.

The automorphism group of $\mathfrak{s}_5$ reads
$$
{\rm Aut}(\mathfrak{s}_5) = \left\{
\left(
\begin{array}{cccc}
T^1_1 & 0 & 0 & T^4_1 \\
0 & T^2_2 & T^3_2 & T^4_2 \\
0 & -T^3_2 & T^2_2 & T^4_3 \\
0 & 0 & 0 & 1
\end{array}
\right): T^1_1\in \mathbb{R} , (T^2_2)^2 + (T^3_2)^2>0, \, T_2^2,T_2^3,T^4_1, T^4_2, T^4_3 \in \rr
\right\}.
$$

Let us set $T^2_2=\mu \cos\phi$ and $T^3_2=\mu \sin \phi$ with $\phi\in [0,2\pi[$ and $\mu \in \mathbb{R}_+$. Then, one sees that there are two connected components of ${\rm Aut}(\mathfrak{s}_5)$. A representative of each connected part and its lift to $\Lambda^2\mathfrak{s}_5$ read
\begin{equation*}
T_\lambda:=\left(
\begin{array}{cccc}
\lambda & 0 & 0 & 0 \\
0 & 1 & 0 & 0 \\
0 & 0 & 1 & 0 \\
0 & 0 & 0 & 1
\end{array}
\right), \qquad \Lambda^2T_\lambda:=\left(
\begin{array}{cccccc}
\lambda  & 0 & 0 & 0 &0&0 \\
0 & \lambda  & 0 & 0 &0&0 \\
0 & 0 & \lambda  & 0 &0&0 \\
0 & 0 & 0 & 1 &0&0 \\
0 & 0 & 0 & 0& 1 &0 \\
0 & 0 & 0 & 0&0& 1 \\
\end{array}
\right), \quad \lambda  \in \{\pm 1\}.
\end{equation*}

Using techniques from previous sections, we can easily  verify whether the orbits of ${\rm Aut}_c(\mathfrak{s}_5)$ are additionally connected by a Lie algebra automorphism of $\mathfrak{s}_5$ via $\Lambda^2T_\lambda$. Our results are summarised in Table \ref{Tab:g_orb_1}. Moreover, for $\beta \neq 0$, the classes of coboundary cocommutators are induced by all classes of equivalent $r$-matrices listed in Table \ref{Tab:g_orb_2} .

Meanwhile, for $\beta=0$, the classes of equivalent coboundary cocommutators come from the following representative $r$-matrices:
$$
r_{0} = 0, \quad r_1 = e_{12}, \quad r_2 = e_{14}
$$

\subsection{Lie algebra $\mathfrak{s}_{6}$}\label{Se:s6}

Structure constants of Lie algebra $\mathfrak{s}_6$ are given in Table \ref{Tab:StruCons}. Using this information, one verifies directly that $(\Lambda^3\mathfrak{s}_6)^{\mathfrak{s}_6}=\langle e_{123}\rangle$ while $(\Lambda^2\mathfrak{s}_6)^{\mathfrak{s}_6}=\{0\}$.

By Remark \ref{Re:DerAlg}, one obtains that
$$
\mathfrak{der}(\mathfrak{s}_6):=\left\{\left(
\begin{array}{cccc}
\mu_{11} & \mu_{12} & \mu_{13} & \mu_{14} \\
0 & \mu_{22} & 0 & -\mu_{13} \\
0 & 0 & \mu_{11} - \mu_{22} & -\mu_{12} \\
0 & 0 & 0 & 0
\end{array}
\right): \mu_{11}, \mu_{12}, \mu_{13}, \mu_{14}, \mu_{22} \in \mathbb{R}\right\},
$$
which gives rise to the basis of $V_{\mathfrak{s}_6}$ of the form
\begin{align*}
&X_1 = x_1 \partial_{x_1} + 2x_2 \partial_{x_2} + x_3 \partial_{x_3} + x_4 \partial_{x_4} + x_6 \partial_{x_6}, &
&X_2 = (-x_3 + x_4) \partial_{x_2} + x_5 \partial_{x_3} - x_5 \partial_{x_4}, \\
&X_3 = (-x_3 - x_4) \partial_{x_1} + x_6 \partial_{x_3} + x_6 \partial_{x_4}, &
&X_4 = -x_5 \partial_{x_1} - x_6 \partial_{x_2}, \\
&X_5 = x_1 \partial_{x_1} - x_2 \partial_{x_2} + x_5 \partial_{x_5} - x_6 \partial_{x_6}.
\end{align*}

For an element $r \in \Lambda^2 \mathfrak{s}_6$, we get
$$
[r,r] = 2(x_1 x_6 + x_2 x_5 + x_4^2)e_{123} + 2(x_3 + x_4)x_5 e_{124} + 2(x_4-x_3)x_6 e_{134} - 4x_5 x_6 e_{234}.
$$
Since $(\Lambda^3 \mathfrak{s}_6)^{\mathfrak{s}_6} = \langle e_{123} \rangle$, the mCYBE reads
$$
(x_3 + x_4) x_5 = 0\quad (x_4 - x_3) x_6 = 0, \quad x_5 x_6 = 0,
$$
whereas the CYBE is
$$
x_1 x_6 + x_2 x_5 + x_4^2 = 0, \quad (x_3 + x_4) x_5 = 0, \quad (x_4 - x_3) x_6 = 0, \quad x_5 x_6 = 0.
$$
Note that the CYBE obtained above is exactly the result obtained in \cite[eq. (3.7)]{BH96} under the substitution $x_5=-\alpha_+$, $x_6=-\alpha_-$, $x_4=\xi$ and $x_3=\vartheta$. 

Since $x_5, x_6$ are the bricks of $\mathfrak{s}_6$, our Darboux tree starts with cases $x_i = 0$ and $x_i \neq 0$, for $i \in \{5,6\}$. The full Darboux tree is presented below.

\begin{center}
{\small
\begin{tikzpicture}[
roundnode/.style={rounded rectangle, draw=green!40, fill=green!3, very thick, minimum size=2mm},
squarednode/.style={rectangle, draw=red!30, fill=red!2, thick, minimum size=4mm}
]
\node[squarednode] (brick) at (0,0) {$x_6=0$};

\node[squarednode] (u)  at (2,0) {$x_5=0$};
\node[squarednode] (d)  at (2,-10) {$x_5 \neq 0$};

\node[squarednode] (uu)  at (4,0) {$x_4=0$};
\node[squarednode] (ud)  at (4,-5) {$x_4 \neq 0$};
\node[roundnode] (du)  at (4,-10) {$\stackrel{\tiny{\rm No \, solutions}}{x_3 + x_4 \neq 0}$};
\node[squarednode] (dd)  at (4,-11) {$x_3 + x_4 = 0$};

\node[squarednode] (uuu)  at (6,0) {$x_3=0$};
\node[squarednode] (uud)  at (6,-4) {$x_3 \neq 0$};
\node[squarednode] (udu)  at (6,-5) {$x_3 = 0$};
\node[squarednode] (udd)  at (6,-6) {$x_3 \neq 0$};

\node[squarednode] (uuuu)  at (8,0) {$x_2=0$};
\node[squarednode] (uuud)  at (8,-2) {$x_2 \neq 0$};
\node[squarednode] (uddu)  at (9,-6) {$x_3 - x_4 = 0$};
\node[squarednode] (udddi)  at (9,-8) {$x_3 + x_4 = 0$};
\node[squarednode] (udddiid)  at (7,-11) {$x_2x_5+x_4^2=0$};
\node[squarednode] (udddiii)  at (7,-12) {$x_2x_5+x_4^2\neq 0$};
\node[squarednode] (udddii)  at (9,-10) {$\stackrel{k \notin \{-1,0,1\}}{x_3 - kx_4 = 0}$};

\node[squarednode] (uuuuu)  at (10,0) {$x_1 = 0$};
\node[squarednode] (uuuud)  at (10,-1) {$x_1 \neq 0$};
\node[squarednode] (uuudu)  at (10,-2) {$x_1 = 0$};
\node[squarednode] (uuudd)  at (10,-3) {$x_1 \neq 0$};
\node[squarednode] (udduu)  at (11,-6) {$x_2 = 0$};
\node[squarednode] (uddud)  at (11,-7) {$x_2 \neq 0$};
\node[squarednode] (udddiu)  at (11,-8) {$x_1 = 0$};
\node[squarednode] (udddid)  at (11,-9) {$x_1 \neq 0$};

\node[squarednode] (0) at (13,0) {0};
\node[squarednode] (I) at (13,-1) {I};
\node[squarednode] (II) at (13,-2) {I$^{\rm ext}$};
\node[squarednode] (III) at (13,-3) {II};
\node[squarednode] (IV) at (13,-4) {III};
\node[squarednode] (V) at (13,-5) {IV};
\node[squarednode] (VI) at (13,-6) {V};
\node[squarednode] (VII) at (13,-7) {VI};
\node[squarednode] (VII) at (13,-8) {${\rm V}^{{\rm ext}}$};
\node[squarednode] (IX) at (13,-9) {${\rm VI}^{{\rm ext}}$};
\node[squarednode] (X) at (13,-10) {VII$_{|k| \notin \{0, 1\}}$};
\node[squarednode] (XI) at (13,-11) {VIII};
\node[squarednode] (XII) at (13,-12) {IX};

\draw[->] (brick.east) -- (u.west);
\draw[->] (brick.east) -- (d.west);

\draw[->] (u.east) -- (uu.west);
\draw[->] (u.east) -- (ud.west);
\draw[->] (d.east) -- (du.west);
\draw[->] (d.east) -- (dd.west);

\draw[->] (uu.east) -- (uuu.west);
\draw[->] (uu.east) -- (uud.west);
\draw[->] (ud.east) -- (udu.west);
\draw[->] (ud.east) -- (udd.west);

\draw[->] (uuu.east) -- (uuuu.west);
\draw[->] (uuu.east) -- (uuud.west);
\draw[->] (udd.east) -- (uddu.west);
\draw[->] (udd.east) -- (udddi.west);
\draw[->] (udd.east) -- (udddii.west);

\draw[->] (uuuu.east) -- (uuuuu.west);
\draw[->] (uuuu.east) -- (uuuud.west);
\draw[->] (uuud.east) -- (uuudu.west);
\draw[->] (uuud.east) -- (uuudd.west);
\draw[->] (uddu.east) -- (udduu.west);
\draw[->] (uddu.east) -- (uddud.west);
\draw[->] (udddi.east) -- (udddiu.west);
\draw[->] (udddi.east) -- (udddid.west);

\draw[->] (dd.east) -- (udddiid.west);
\draw[->] (dd.east) -- (udddiii.west);

\end{tikzpicture}

\vspace{1cm}

\begin{tikzpicture}[
roundnode/.style={rounded rectangle, draw=green!40, fill=green!3, very thick, minimum size=2mm},
squarednode/.style={rectangle, draw=red!30, fill=red!2, thick, minimum size=4mm}
]
\node[squarednode] (brick) at (0,0) {$x_6 \neq 0$};

\node[squarednode] (u)  at (3,0) {$x_5 = 0$};
\node[roundnode] (d)  at (3,-1) {$\stackrel{\tiny{\rm No \, solutions}}{x_5 \neq 0}$};

\node[roundnode] (uu)  at (6,0) {$\stackrel{\tiny{\rm No \, solutions}}{- x_3 + x_4 \neq 0}$};
\node[squarednode] (ud)  at (6,-1) {$- x_3 + x_4 = 0$};

\node[squarednode] (udu)  at (9,-1) {$ x_1 x_6 + x_4^2 = 0$};
\node[squarednode] (udd)  at (9,-2) {$ x_1 x_6 + x_4^2 \neq 0$};

\node[squarednode] (XII) at (11,-1) {VIII$^{\rm ext}$};
\node[squarednode] (XIII) at (11,-2) {${\rm IX}^{{\rm ext}}$};

\draw[->] (brick.east) -- (u.west);
\draw[->] (brick.east) -- (d.west);

\draw[->] (u.east) -- (uu.west);
\draw[->] (u.east) -- (ud.west);

\draw[->] (ud.east) -- (udu.west);
\draw[->] (ud.east) -- (udd.west);

\end{tikzpicture}
}
\end{center}
The connected parts of the subsets denoted in the above Darboux tree are the orbits of ${\rm Aut}_c(\mathfrak{s}_6)$. The results are summarised in Table \ref{Tab:g_orb_2}.

The Lie algebra automorphism group of $\mathfrak{s}_6$ reads
$$
{\rm Aut}(\mathfrak{s}_6) = \left\{
\left(
\begin{array}{cccc}
-T^3_2 T^2_3 & T^4_2 T_3^2 & T^4_3 T^3_2 & T^4_1 \\
0 & 0 & T^3_2 & T^4_2 \\
0 & T^2_3 & 0 & T^4_3 \\
0 & 0 & 0 & -1
\end{array}
\right),\left(
\begin{array}{cccc}
T^2_2 T^3_3 & -T^2_2 T^4_3 & -T^4_2 T^3_3 & T^4_1 \\
0 & T^2_2 & 0 & T^4_2 \\
0 & 0 & T^3_3 & T^4_3 \\
0 & 0 & 0 & 1
\end{array}
\right):
\begin{array}{c}
 T^2_3,T^2_2 \in \rz, \\
 T^3_2,T^3_3 \in \rz, \\
 T^4_1, T^4_2, T^4_3 \in \rr
\end{array}
\right\}.
$$

One representative element of each connected component of ${\rm Aut}(\mathfrak{s}_6)$ and their lifts to $\Lambda^2\mathfrak{s}_6$ are given by, respectively,
\begin{equation*}
T_{\lambda_1, \lambda_2} := \left(
\begin{array}{cccc}
\mp\lambda_1 \lambda_2 & 0 & 0 & 0 \\
0 & \theta(\mp 1)\lambda_1 & \theta(\pm 1)\lambda_2 & 0 \\
0 & \theta(\pm 1)\lambda_1 & \theta(\mp 1)\lambda_2 & 0 \\
0 & 0 & 0 & \mp 1
\end{array}
\right),
\end{equation*}
and
\begin{equation*}
\Lambda^2 T_{\lambda_1, \lambda_2} := {\small
\left(
\begin{array}{cccccc}
\theta(\mp 1)\lambda_2 & -\theta(\pm 1)\lambda_1  & 0 & 0 &0&0 \\
 -\theta(\pm 1) \lambda_2 & \theta(\mp 1)\lambda_1 & 0 & 0 &0&0 \\
0 & 0 & \lambda_1 \lambda_2 & 0 &0&0 \\
0 & 0 & 0 & \mp \lambda_1\lambda_2&0&0 \\
0 & 0 & 0 & 0&\theta(\mp 1)\lambda_1&-\theta(\pm 1)\lambda_2 \\
0 & 0 & 0 & 0&-\theta(\pm 1)\lambda_1 &\theta(\mp 1)\lambda_2 \\
\end{array}
\right)},
\end{equation*}
where $\theta(x)$ stands for the Heaviside function and $\lambda_1,\lambda_2\in \{1,-1\}$. 

Using the maps $\Lambda^2 T_{\lambda_1, \lambda_2}$, we can verify whether some of the orbits of the action of ${\rm Aut}_c(\mathfrak{s}_6)$ on $\mathcal{Y}_{\mathfrak{s}_6}$ can be identified, giving rise to the classes of equivalent $r$-matrices up to the action of the full Lie group ${\rm Aut} (\mathfrak{s}_6)$. Our results are presented in Table \ref{Tab:g_orb_2}.

Since $(\Lambda^2\mathfrak{s}_6)^{\mathfrak{s}_6}=0$, each class of equivalent $r$-matrices gives rise to a separate class of coboundary Lie bialgebras on $\mathfrak{s}_6$. 

\subsection{Lie algebra $\mathfrak{s}_{7}$}\label{Sec:s7}

The structure constants of Lie algebra $\mathfrak{s}_7$ are given in Table \ref{Tab:StruCons}. Using this information, one verifies that $(\Lambda^2\mathfrak{s}_7)^{\mathfrak{s}_7}=0$ and $(\Lambda^3\mathfrak{s}_7)^{\mathfrak{s}_7}=\langle e_{123}\rangle$.

By Remark \ref{Re:DerAlg}, one obtains that
$$
\mathfrak{der}(\mathfrak{s}_7):=\left\{\left(
\begin{array}{cccc}
\mu_{11} & \mu_{12} & \mu_{13} & \mu_{14} \\
0 & \frac{1}{2}\mu_{11} & \mu_{23} & \mu_{12} \\
0 & -\mu_{23} & \frac{1}{2}\mu_{11} & \mu_{13} \\
0 & 0 & 0 & 0
\end{array}
\right):  \mu_{11}, \mu_{12}, \mu_{13}, \mu_{14}, \mu_{23} \in \mathbb{R}\right\},
$$
which give rise to the basis of $V_{\mathfrak{s}_7}$ of the form
\begin{align*}
&X_1 = \frac{3}{2} x_1 \partial_{x_1} + \frac{3}{2} x_2 \partial_{x_2} + x_3 \partial_{x_3} + x_4 \partial_{x_4} + \frac{1}{2} x_5 \partial_{x_5} + \frac{1}{2} x_6 \partial_{x_6}, &
&X_2 = x_3 \partial_{x_1} + x_4 \partial_{x_2} + x_5 \partial_{x_3} - x_6 \partial_{x_4}, \\
&X_3 = - x_4 \partial_{x_1} + x_3 \partial_{x_2} + x_6 \partial_{x_3} + x_5 \partial_{x_4}, &
&X_4 = x_2 \partial_{x_1} - x_1 \partial_{x_2} + x_6 \partial_{x_5} - x_5 \partial_{x_6} \\
&X_5 = -x_5 \partial_{x_1} - x_6 \partial_{x_2}.
\end{align*}

Then,
$$
[r,r] = 2 (x_4^2 + x_1 x_5 + x_2 x_6) e_{123} + 2 (x_4 x_5 + x_3 x_6) e_{124} + 2 (x_4 x_6 - x_3 x_5) e_{134} - 2 (x_5^2 + x_6^2) e_{234}.
$$
Since $(\Lambda^3\mathfrak{s}_7)^{\mathfrak{s}_7}=\langle e_{123}\rangle$, then the mCYBE reads $x_5=0, x_6=0$. Meanwhile, the CYBE takes the form $x_4=0, x_5=0, x_6=0$.

The Darboux tree is presented below.

\begin{center}
{\small
\begin{tikzpicture}[
roundnode/.style={rounded rectangle, draw=green!40, fill=green!3, very thick, minimum size=2mm},
squarednode/.style={rectangle, draw=red!30, fill=red!2, thick, minimum size=4mm}
]
\node[squarednode] (brick) at (0,0) {$x_6=0$};

\node[squarednode] (u)  at (2,0) {$x_5=0$};
\node[roundnode] (d)  at (2,-4) {$\stackrel{\tiny{\rm No \, solutions}}{x_5 \neq 0}$};

\node[squarednode] (uu)  at (4,0) {$x_4=0$};
\node[squarednode] (ud)  at (4,-3) {$x_4 \neq 0$};

\node[squarednode] (uuu)  at (6,0) {$x_3=0$};
\node[squarednode] (uud)  at (6,-2) {$x_3 \neq 0$};
\node[squarednode] (udu)  at (6,-3) {$x_3 = 0$};
\node[squarednode] (udd)  at (6,-4) {$x_4 - k x_3 = 0$};

\node[squarednode] (uuuu)  at (8,0) {$x_1^2 + x_2^2 = 0$};
\node[squarednode] (uuud)  at (8,-1) {$x_1^2 + x_2^2 \neq 0$};

\node[squarednode] (0) at (12,0) {0};
\node[squarednode] (I) at (12,-1) {I};
\node[squarednode] (II) at (12,-2) {II};
\node[squarednode] (III) at (12,-3) {III};
\node[squarednode] (IV) at (12,-4) {${\rm IV}_{k\neq 0}$};

\draw[->] (brick.east) -- (u.west);
\draw[->] (brick.east) -- (d.west);

\draw[->] (u.east) -- (uu.west);
\draw[->] (u.east) -- (ud.west);

\draw[->] (uu.east) -- (uuu.west);
\draw[->] (uu.east) -- (uud.west);
\draw[->] (ud.east) -- (udu.west);
\draw[->] (ud.east) -- (udd.west);

\draw[->] (uuu.east) -- (uuuu.west);
\draw[->] (uuu.east) -- (uuud.west);

\end{tikzpicture}
}
\end{center}

The connected parts of the subsets denoted in the above Darboux tree are the orbits of ${\rm Aut}_c(\mathfrak{s}_{7})$ in $\mathcal{Y}_{\mathfrak{s}_7}$. 

The Lie algebra automorphism group of $\mathfrak{s}_7$ reads
\begin{equation*}
\begin{split}
{\rm Aut}(\mathfrak{s}_7) &= \left\{
\left(
\begingroup
\setlength\arraycolsep{4.5pt}
\begin{array}{cccc}
\pm [(T^2_2)^2 + (T^2_3)^2] & \pm T^4_2 T^2_2 - T^4_3 T^2_3 & T^4_3 T^2_2 \pm  T^4_2 T^2_3 & T^4_1 \\
0 & T^2_2 & T^2_3 & T^4_2 \\
0 &\mp T^2_3 & \pm T^2_2 & T^4_3 \\
0 & 0 & 0 & \pm 1
\end{array}
\endgroup
\right):
\begingroup
\setlength\arraycolsep{4.5pt}
\begin{array}{c}
(T^2_2)^2+(T_3^2)^2 \in \mathbb{R} \backslash \{ 0\} \\
T^2_2,T^2_3,T^4_1, T^4_2, T^4_3 \in \rr
\end{array}
\endgroup
\right\}. 
\end{split}
\end{equation*}
Since the two subsets of $GL(2,\mathbb{R})$ of matrices of the form
$$
\left(\begin{array}{cc}T^2_2&T^2_3\\\mp T^2_3&\pm T^2_2\end{array}\right)\in GL(2,\mathbb{R}),\qquad (T^2_2)^2+(T^2_3)^2>0,$$
can be parametrised via $\phi\in [0,2\pi[$ and $\mu:=[(T^2_2)^2+(T^2_3)^2]^{1/2}\in \mathbb{R}_+$ by setting $T^2_2=\mu \cos\phi$ and $T^2_3=\mu \sin \phi$, one gets that ${\rm Aut}(\mathfrak{s}_7)$ has two connected components. As usual, we only need one element for every connected component of ${\rm Aut}(\mathfrak{s}_7)$ and their lifts to $\Lambda^2\mathfrak{s}_7$, namely
$$
T_{\pm}:=\left(\begin{array}{cccc}
\pm 1& 0 & 0&0 \\
0 & 1 & 0 & 0 \\
0 &0 & \pm 1 &0 \\
0 & 0 & 0 & \pm 1
\end{array}\right), \qquad \Lambda^2T_{\pm}=\left(
\begin{array}{cccccc}
\pm 1 & 0 & 0 & 0 &0&0 \\
0 & 1 & 0 & 0 &0&0 \\
0 & 0 & 1 & 0 &0&0 \\
0 & 0 & 0 & \pm1&0&0 \\
0 & 0 & 0 & 0&\pm 1&0 \\
0 & 0 & 0 & 0&0&1 \\
\end{array}
\right).
$$
Using $\Lambda^2T_{\pm}$, we can verify whether some of these parts are additionally connected by a Lie algebra automorphism of $\mathfrak{s}_{7}$. The results are summarised in Table \ref{Tab:g_orb_2}. 

Since $(\Lambda^2\mathfrak{s}_7)^{\mathfrak{s}_7}=0$, we obtain that each class of equivalent coboundary Lie bialgebras is given by the separate class of equivalent $r$-matrices, listed in Table \ref{Tab:g_orb_2}.

\subsection{Lie algebra $\mathfrak{s}_{8}$}

The structure constants for the Lie algebra $\mathfrak{s}_{8}$ are given in Table \ref{Tab:StruCons}. Recall that $\alpha\in ]-1,1]\backslash\{0\}$. As in the previous cases, one can obtain by direct computation that $(\Lambda^2 \mathfrak{s}_8)^{\mathfrak{s}_8} = \langle e_{13}\rangle$ for $\alpha = -\frac{1}{2}$ and $(\Lambda^2 \mathfrak{s}_8)^{\mathfrak{s}_8} = 0$,  otherwise. Moreover, $(\Lambda^3\mathfrak{s}_8)^{\mathfrak{s}_8}=0$.

By Remark \ref{Re:DerAlg}, one obtains that the derivations of $\mathfrak{s}^\alpha_8$, for a fixed value $\alpha\in ]-1,1[\backslash\{0\}$ read
\begin{equation}\label{Eq:Der8a}
\mathfrak{der}(\mathfrak{s}^\alpha_8)=\left\{\left(
\begin{array}{cccc}
\mu_{11} & \mu_{12} & \mu_{13} & \mu_{14} \\
0 & \mu_{22} & 0 & -\mu_{13} \\
0 & 0 & \mu_{11} - \mu_{22} & \alpha \mu_{12} \\
0 & 0 & 0 & 0
\end{array}
\right): \mu_{11}, \mu_{12}, \mu_{13}, \mu_{14}, \mu_{22} \in \mathbb{R}\right\}.
\end{equation}
In case $\alpha=1$, one gets
\begin{equation}\label{Eq:Der8a1}
\mathfrak{der}(\mathfrak{s}^1_8)=\left\{\left(
\begin{array}{cccc}
\mu_{11} & \mu_{12} & \mu_{13} & \mu_{14} \\
0 & \mu_{22} & \mu_{23} & -\mu_{13} \\
0 & \mu_{32} & \mu_{11} - \mu_{22} & \mu_{12} \\
0 & 0 & 0 & 0
\end{array}
\right): \mu_{11}, \mu_{12}, \mu_{13}, \mu_{14}, \mu_{22},\mu_{23},\mu_{32} \in \mathbb{R}\right\}.
\end{equation}
For any value of $\alpha\in ]-1,1]\backslash\{0\}$ and an element $r \in \Lambda^2\mathfrak{s} _{8}$, one obtains
$$
[r,r]=2[(2+ \alpha)x_1x_6 - (1 + 2\alpha)x_2x_5 + (1+\alpha)x_3x_4 + x_4^2]e_{123} + 2(x_4-\alpha x_3)x_5e_{124} + 2x_6(x_4-x_3)e_{134} + 2(\alpha - 1)x_5x_6e_{234}.
$$ 
Since $ (\Lambda^3\mathfrak{s}_{8})^{\mathfrak{s}_{8}}=0$, the mCYBE reads
$$
(2+ \alpha)x_1x_6 - (1 + 2\alpha)x_2x_5 + (1+\alpha)x_3x_4 + x_4^2= 0, \quad (x_4-\alpha x_3)x_5= 0, \quad x_6(x_4-x_3)= 0, \quad (\alpha - 1)x_5x_6 = 0.
$$
Let us now consider two cases given by $\alpha\in ]-1,1[\backslash\{0\}$ and $\alpha=1$. In the first case, the space of derivations gives rise to the basis of $V_{\mathfrak{s}^\alpha_8}$ of the form
\begin{align}\label{Eq:Vs8a}
&X_1=	x_1\partial_{x_1}+2x_2\partial_{x_2}+x_3\partial_{x_3}+x_4\partial_{x_4}+x_6\partial_{x_6}, &
&X_2=(x_3\alpha+x_4)\partial_{x_2}+x_5\partial_{x_3}+\alpha x_5\partial_{x_4}, \nonumber \\ 
&X_3=-(x_3+x_4)\partial_{x_1}+x_6\partial_{x_3}+x_6\partial_{x_4}, &
&X_4=x_1\partial_{x_1}-x_2\partial_{x_2}+x_5\partial_{x_5}-x_6\partial_{x_6},\\
&X_5=-x_5\partial_{x_1}-x_6\partial_{x_2} \nonumber.
\end{align}

	The following diagram depicts the Darboux families that give us the orbits of ${\rm Aut}_c(\mathfrak{s}^\alpha_{8})$ on $\mathcal{Y}_{\mathfrak{s}^\alpha_{8}}$. Since $x_5, x_6$ are the bricks for every $\mathfrak{s}^\alpha_8$, our Darboux tree starts with the cases $x_i = 0$ and $x_i \neq 0$, for $i \in \{5,6\}$. The full Darboux tree is presented below.
\begin{center}
{\small
	\begin{tikzpicture}[
roundnode/.style={rounded rectangle, draw=green!40, fill=green!3, very thick, minimum size=2mm},
squarednode/.style={rectangle, draw=red!30, fill=red!2, thick, minimum size=4mm}
]
	\node[squarednode] (brick) at (0,0) {$x_6=0$};

	\node[squarednode] (65)  at (2,0) {$x_5=0$};
	\node[squarednode] (65u)  at (2,-7) {$x_5\neq 0$};

	\node[squarednode]   (654)  at (4,0) {$x_4=0$};
	\node[squarednode]   (654u) at (4,-5)  {$x_4\neq 0$};
	\node[squarednode]   (65ue) at (4,-7)  {$x_4-\alpha x_3=0$};
	\node[roundnode]   (65ueu) at (4,-8)  {$\stackrel{\rm No\,\,solutions}{x_4-\alpha x_3\neq 0}$};

	\node[squarednode]   (6543) at (7,0)  {$x_3=0$};
	\node[squarednode]   (6543u) at (7,-4)  {$x_3\neq 0$};
	\node[squarednode]   (654ufu) at (7,-5)  {$\stackrel{(x_3\neq 0)}{x_3+x_4\neq 0}$};
	\node[roundnode]   (654uf) at (7,-6)  {$\stackrel{{\rm No\,\, solutions}}{x_3+x_4=0}$};
	\node[squarednode]   (65uey) at (7,-7)  {$\alpha x_3^2-x_2x_5=0$};
	\node[squarednode]   (65ueyu) at (7,-8)  {$\stackrel{(\alpha=-1/2)}{\alpha x_3^2-x_2x_5\neq 0}$};

	\node[squarednode]   (65432) at (10,0)  {$x_2=0$};
	\node[squarednode]   (65432u) at (10,-2)  {$x_2\neq 0$};
	\node[squarednode]   (654ufug) at (10,-5)  {$(1+\alpha)x_3+x_4=0$};
	\node[roundnode]   (654ufugu) at (10,-6)  {$\stackrel{{\rm No\,\,solutions}}{(1+\alpha)x_3+x_4\neq 0}$};

	\node[squarednode]   (654321) at (13,0)  {$x_1=0$};
	\node[squarednode]   (654321u) at (13,-1)  {$x_1\neq 0$};
	\node[squarednode]   (65432u1) at (13,-2)  {$x_1=0$};
	\node[squarednode]   (65432u1u) at (13,-3)  {$x_1\neq 0$};

	\node[squarednode] (0) at (14,0) {0};
	\node[squarednode] (I) at (14,-1) {I};
	\node[squarednode] (II) at (14,-2) {II};
	\node[squarednode] (III) at (14,-3) {III};
	\node[squarednode] (IV) at (14,-4) {IV};
	\node[squarednode] (V) at (14,-5) {V};
	\node[squarednode] (VI) at (14,-7) {VI};
	\node[squarednode] (VIa) at (14,-8) {VII$^{\alpha=-1/2}$};
	
	\draw[->] (brick.east) -- (65.west);
	\draw[->] (brick.east) -- (65u.west);
	\draw[->] (65.east) -- (654.west);
	\draw[->] (65.east) -- (654u.west);
	\draw[->] (65u.east) -- (65ue.west);
	\draw[->] (65u.east) -- (65ueu.west);
	\draw[->] (65ue.east) -- (65uey.west);
	\draw[->] (65ue.east) -- (65ueyu.west);
	\draw[->] (6543.east)--(65432.west);
	\draw[->] (654.east) -- (6543.west);
\draw[->] (654.east) -- (6543u.west);
	\draw[->] (654u.east) -- (654uf.west);
\draw[->] (654u.east) -- (654ufu.west);
	\draw[->] (654ufu.east) -- (654ufug.west);
	\draw[->] (654ufu.east) -- (654ufugu.west);
	\draw[->] (6543.east) -- (65432u.west);
\draw[->] (65432.east) -- (654321.west);
\draw[->] (65432.east) -- (654321u.west);
\draw[->] (65432u.east) -- (65432u1.west);
\draw[->] (65432u.east) -- (65432u1u.west);
\end{tikzpicture}
	
\vspace{1cm}

		\begin{tikzpicture}[
roundnode/.style={rounded rectangle, draw=green!40, fill=green!3, very thick, minimum size=2mm},
squarednode/.style={rectangle, draw=red!30, fill=red!2, thick, minimum size=4mm}
]
	\node[squarednode] (6u) at (0,0) {$x_6\neq 0$};

	\node[squarednode] (6u5)  at (2,0)  {$x_5=0$};
	\node[squarednode] (6u5u)  at (2,-2) {$\stackrel{(\alpha=1)}{x_5\neq 0}$};

	\node[squarednode]   (6u5e)  at (4,0) {$x_3-x_4=0$};
	\node[roundnode]   (6u5eu) at (4,-1)  {$\stackrel{{\rm No\,\, solutions}}{x_3-x_4\neq 0}$};
	\node[squarednode]   (6u5ue) at (4,-2)  {$x_4-x_3=0$};
	\node[roundnode]   (6u5ueu) at (4,-3)  {$\stackrel{\rm No\,\,solutions}{x_4-x_3\neq 0}$};

	\node[roundnode]   (6u5efu) at (8,0)  {$\stackrel{{\rm No\,\, solutions}}{x_3^2+x_1x_6\neq 0}$};
	\node[squarednode]   (6u5ef) at (8,-1)  {$x_3^2+x_1x_6=0$};
	\node[squarednode]   (6u5ueg) at (8,-2)  {$x_1x_6-x_2x_5+x_3^2=0$};
	\node[roundnode]   (6u5uegu) at (8,-3)  {$\stackrel{{\rm No\,\,solutions}}{x_1x_6-x_2x_5+x_3^2\neq0}$};

	\node[squarednode] (VII) at (13,-1) {VIII};
	\node[squarednode] (VIII) at (13,-2) {IX$_{\alpha=1}$};

	\draw[->] (6u.east) -- (6u5.west);
	\draw[->] (6u.east) -- (6u5u.west);
	\draw[->] (6u5.east) -- (6u5e.west);
	\draw[->] (6u5.east) -- (6u5eu.west);
	\draw[->] (6u5e.east) -- (6u5ef.west);
	\draw[->] (6u5e.east) -- (6u5efu.west);
	\draw[->] (6u5u.east)--(6u5ue.west);
	\draw[->] (6u5u.east) -- (6u5ueu.west);
	\draw[->] (6u5ue.east) -- (6u5ueg.west);
	\draw[->] (6u5ue.east) -- (6u5uegu.west);

	\end{tikzpicture}
	}
\end{center}

The Lie algebra automorphism group of $\mathfrak{s}^\alpha_8$ for each $\alpha \in ]-1,1[\backslash\{0\}$, reads
$$
{\rm Aut}(\mathfrak{s}^\alpha_8) = \left\{
\left(
\begin{array}{cccc}
T^2_2 T^3_3 & T^2_1 & T^3_1 & T^4_1 \\
0 & T^2_2 & 0 & -\frac{T^3_1}{T^3_3} \\
0 & 0 & T^3_3 & \frac{\alpha T^2_1}{T^2_2} \\
0 & 0 & 0 & 1
\end{array}
\right):
\begin{array}{c}
 T^2_2, \, T^3_3 \in \rz \\
T^2_1, T^3_1, T^4_1 \in \rr
\end{array}
\right\}.
$$

In reality, we are only concerned with obtaining one element of ${\rm Aut}(\mathfrak{s}^\alpha_8)$ and its lift to $\Lambda^2\mathfrak{s}_8$ for each one of its connected components. For instance, we can choose
\begin{equation*}
T_{\lambda_1, \lambda_2} = \left(
\begin{array}{cccc}
\lambda_1 \lambda_2 & 0 & 0 & 0 \\
0 & \lambda_1 & 0 & 0 \\
0 & 0 & \lambda_2 & 0 \\
0 & 0 & 0 & 1
\end{array}
\right), \quad 
\Lambda^2T_{\lambda_1,\lambda_2}=\left(
\begin{array}{cccccc}
\lambda_2 & 0 & 0 & 0 & 0 & 0 \\
0 & \lambda_1 & 0 & 0 & 0 & 0 \\
0 & 0 & \lambda_1 \lambda_2 & 0 & 0 & 0 \\
0 & 0 & 0 & \lambda_1 \lambda_2 & 0 & 0 \\
0 & 0 & 0 & 0 & \lambda_1 & 0 \\
0 & 0 & 0 & 0 & 0 & \lambda_2 
\end{array}
\right),  \quad \lambda_1, \lambda_2 \in \{\pm 1\}.
\end{equation*}

As in previous sections, the maps $\Lambda^2T_{\lambda_1,\lambda_2}$ allow us to identify the orbits of the action of ${\rm Aut} (\mathfrak{s}^\alpha_8)$ on $\mathcal{Y}_{\mathfrak{s}^\alpha_8}$. Our results are presented in Table \ref{Tab:g_orb_2}. Moreover, for $\alpha \neq -\frac{1}{2}$, each family of equivalent $r$-matrices give rise to a separate class of equivalent coboundary Lie bialgebras. For $\alpha = -\frac{1}{2}$, we obtain eight families of equivalent coboundary  Lie bialgebras given by:
 $$
r_1 = e_{12}, \qquad r_2 = e_{14}, \qquad r_3 = e_{14} - \frac{1}{2} e_{23}, \qquad r_{4, \pm} = \pm e_{24}, \qquad r_5 = e_{34}.
 $$
Let us know study the case $\alpha=1$ using the previous results. Since (\ref{Eq:Der8a}) for $\alpha=1$ is a Lie subalgebra of the space of derivations for $\mathfrak{s}_8^\alpha$ given in (\ref{Eq:Der8a1}), the Lie algebra spanned by the vector fields of (\ref{Eq:Vs8a}) for $\alpha=1$ is a Lie subalgebra of the Lie algebra of fundamental vector fields of the action of ${\rm Aut}(\mathfrak{s}_8^1)$ on $\Lambda^2\mathfrak{s}_8$. Hence, the loci of the above Darboux families allow us to characterise the strata of the distribution spanned by the vector fields (\ref{Eq:Vs8a}) for $\alpha=1$, which in turn gives as the orbits of a Lie subgroup of ${\rm Aut}(\mathfrak{s}_8^1)$. To easily follow our discussion, we detail that 
$$
{\rm Aut}(\mathfrak{s}^\alpha_8) = \left\{
\left(
\begin{array}{cccc}
T^2_2T^3_3-T_2^3T^2_3 & -T_2^4T_3^2+T^2_2T_3^4 & -T_2^4T^3_3+T_2^3T_3^4 & T^4_1 \\
0 & T^2_2 & T_2^3 & T_2^4 \\
0 & T_3^2 & T^3_3 & T_3^4\\
0 & 0 & 0 & 1
\end{array}
\right):
\begin{array}{c}
 T^2_2T^3_3-T_3^2T^3_2\neq 0, \\ 
 T_1^4, T_2^4, T_3^4\in \rr
\end{array}
\right\}.
$$

To obtain the orbits of ${\rm Aut}(\mathfrak{s}_8^1)$ on $\mathcal{Y}_{\mathfrak{s}_8^1}$, it is enough to write ${\rm Aut}(\mathfrak{s}_8^1)$ as a composition of the previous subgroup with certain Lie algebra automorphisms of ${\rm Aut}(\mathfrak{s}_8^1)$, e.g. the Lie algebra automorphisms $T_A=(\det A){\rm Id}\otimes A\otimes {\rm Id}$ for every $A\in GL(2,\mathbb{R})$. 
Hence, $\Lambda^2T_A=(\det A)A\otimes (\det A){\rm Id}\otimes (\det A){\rm Id}$.

The action of the  $\Lambda^2T_A$ on the loci of the Darboux families of (\ref{Eq:Vs8a}) assuming  $\alpha=1$ allows us to obtain the orbits of the action of ${\rm Aut}(\mathfrak{s}_8^1)$ on $\mathcal{Y}_{\mathfrak{s}_8^1}$. In particular, one obtains the subsets given in Table \ref{Tab:g_orb_2}. 
Since $(\Lambda^2\mathfrak{s}^1_8)^{\mathfrak{s}^1_8}=0$, the classes of equivalent coboundary coproducts for $\mathfrak{s}_8^1$ are given separately by each representative element listed in Tabee \ref{Tab:g_orb_2}.

\subsection{Lie algebra $\mathfrak{s}_{9}$}

As in previous cases, we use the structure constants for $\mathfrak{s}_9$ in Table \ref{Tab:StruCons} to verify that $(\Lambda^2 \mathfrak{s}_9)^{\mathfrak{s}_9} = 0$ and $(\Lambda^3\mathfrak{s}_9)^{\mathfrak{s}_9}=0$ for every $\alpha>0$. 

By Remark \ref{Re:DerAlg}, one obtains that
$$
\mathfrak{der}(\mathfrak{s}^\alpha_9):=\left\{\left(
\begin{array}{cccc}
2 \mu_{22} & \mu_{12} & \mu_{13} & \mu_{14} \\
0 & \mu_{22} & \mu_{23} & \mu_{12} - \alpha \mu_{13} \\
0 & -\mu_{23} & \mu_{22} & \alpha \mu_{12} + \mu_{13} \\
0 & 0 & 0 & 0
\end{array}
\right): \mu_{12}, \mu_{13}, \mu_{14}, \mu_{22}, \mu_{23} \in \mathbb{R}\right\}.
$$
The obtained derivations give rise to the basis of $V_{\mathfrak{s}^\alpha_9}$ of the form
\begin{align*}
&X_1 = x_3 \partial_{x_1} + (\alpha x_3 + x_4) \partial_{x_2} + x_5 \partial_{x_3} + (\alpha x_5 - x_6) \partial_{x_4}, &
&X_2 = -x_5 \partial_{x_1} - x_6 \partial_{x_2}, \\
&X_3 = (- \alpha x_3 - x_4) \partial_{x_1} + x_3 \partial_{x_2} + x_6 \partial_{x_3} + (x_5 + \alpha x_6) \partial_{x_4}, &
&X_5 = x_2 \partial_{x_1} - x_1 \partial_{x_2} + x_6 \partial_{x_5} - x_5 \partial_{x_6}, \\
&X_4 = 3x_1 \partial_{x_1} + 3x_2 \partial_{x_2} + 2x_3 \partial_{x_3} + 2x_4 \partial_{x_4} + x_5 \partial_{x_5} + x_6 \partial_{x_6}.
\end{align*}

For an element $r \in \Lambda^2 \mathfrak{s}_9$, we get
\begin{equation*}
\begin{split}
[r,r] &= 2(x_1 x_5 + 3\alpha x_1 x_6 - 3\alpha x_2 x_5 + x_2 x_6 + 2\alpha x_3 x_4 + x_4^2)e_{123} + 2(-\alpha x_3 x_5 + x_3 x_6 + x_4 x_5)e_{124} \\
&+ 2(x_4 x_6-x_3 x_5 - \alpha x_3 x_6 )e_{134} - 2(x_5^2 + x_6^2)e_{234}.
\end{split}
\end{equation*}

Since $(\Lambda^3 \mathfrak{s}_9)^{\mathfrak{s}_9} = 0$ for every value of $\alpha$, the mCYBE and the CYBE are equal and read
$$
(2\alpha x_3 + x_4) x_4 = 0, \quad x_5 = 0, \quad x_6 = 0.
$$

The Darboux tree for the class $\mathfrak{s}_9$ is presented below.

\begin{center}
{\small
\begin{tikzpicture}[
roundnode/.style={rounded rectangle, draw=green!40, fill=green!3, very thick, minimum size=2mm},
squarednode/.style={rectangle, draw=red!30, fill=red!2, thick, minimum size=4mm}
]
	\node[squarednode] (brick) at (0,0) {$x_6 = 0$};

	\node[squarednode] (u)  at (2,0) {$x_5=0$};
	\node[roundnode] (d)  at (2,-3) {$\stackrel{{\rm No \, solutions}}{x_5 \neq 0}$};

	\node[squarednode] (uu)  at (4,0) {$x_3=0$};
	\node[squarednode] (ud)  at (4,-2) {$x_3 \neq 0$};

	\node[squarednode] (uuu)  at (6,0) {$x_4=0$};
	\node[roundnode] (uud)  at (6,-1) {$\stackrel{{\rm No \, solutions}}{x_4 \neq 0}$};
	\node[squarednode] (udu)  at (6,-2) {$x_4=0$};
	\node[squarednode] (udd)  at (6,-3) {$x_4 \neq 0$};

	\node[squarednode] (uuuu)  at (9,0) {$x_1^2 + x_2^2=0$};
	\node[squarednode] (uuud)  at (9,-1) {$x_1^2 + x_2^2 \neq 0$};
	\node[squarednode] (uddu)  at (9,-3) {$x_4 + 2\alpha x_3  = 0$};
    \node[roundnode] (uddd)  at (9,-4) {$\stackrel{{\rm No \, solutions}}{x_4 + 2\alpha x_3  \neq 0}$};

	\node[squarednode] (0) at (12,0) {0};
	\node[squarednode] (I) at (12,-1) {I};
	\node[squarednode] (II) at (12,-2) {II};
	\node[squarednode] (I) at (12,-3) {III};
	
	\draw[->] (brick.east) -- (u.west);
	\draw[->] (brick.east) -- (d.west);

	\draw[->] (u.east) -- (uu.west);
	\draw[->] (u.east) -- (ud.west);

	\draw[->] (uu.east) -- (uuu.west);
	\draw[->] (uu.east) -- (uud.west);
	\draw[->] (ud.east) -- (udu.west);
	\draw[->] (ud.east) -- (udd.west);

	\draw[->] (uuu.east) -- (uuuu.west);
	\draw[->] (uuu.east) -- (uuud.west);
	\draw[->] (udd.east) -- (uddu.west);
	\draw[->] (udd.east) -- (uddd.west);

\end{tikzpicture}
}
\end{center}

If we define $\Delta:=(T^2_2)^2 + (T^2_3)^2$, the Lie algebra automorphisms group of each $\mathfrak{s}^\alpha_9$ reads
$$
{\rm Aut}(\mathfrak{s}^\alpha_9) = \left\{
\left(
\begin{array}{cccc}
\Delta & \frac{T^4_2( T^2_2 + \alpha T^2_3) + T^4_3(\alpha T^2_2 -  T^2_3)}{1+\alpha^2} & \frac{(-\alpha T^2_2 + T^2_3) T^4_2 + T^4_3 (T^2_2 + \alpha  T^2_3)}{1+\alpha^2} & T^4_1 \\
0 & T^2_2 & T^2_3 & T^4_2 \\
0 & -T^2_3 & T^2_2 & T^4_3 \\
0 & 0 & 0 & 1
\end{array}
\right): 
\begin{array}{c}
\Delta\in \mathbb{R}_+, \\
T^2_2,T^2_3 \in \rr \\
T^4_1, T^4_2, T^4_3 \in \rr
\end{array}
\right\}.
$$

Using ideas from Section \ref{Sec:s5}, we obtain that each ${\rm Aut}(\mathfrak{s}^\alpha_9)$ has one connected component.  Consequently, the orbits of ${\rm Aut}(\mathfrak{s}^\alpha_9)$ on $\Lambda^2\mathfrak{s}^\alpha_{9}$ are the strata of $\mathscr{E}_{\mathfrak{s}^\alpha_9}$. Since $(\Lambda^2\mathfrak{s}_9)^{\mathfrak{s}_9}=0$ for every $\alpha>0$, the strata of  $\mathscr{E}_{\mathfrak{s}^\alpha_9}$ within $\mathcal{Y}_{\mathfrak{s}^\alpha_9}$ amount for the families of equivalent coboundary Lie bialgebras on $\mathfrak{s}^\alpha_9$. Our final results are summarised in Table \ref{Tab:g_orb_2}.

\subsection{Lie algebra $\mathfrak{s}_{10}$}

As in the previous cases, we use the structure constants for Lie algebra $\mathfrak{s}_{10}$ in Table \ref{Tab:StruCons} to show that $(\Lambda^3\mathfrak{s}_{10})^{\mathfrak{s}_{10}}=0$ and $(\Lambda^2\mathfrak{s}_{10})^{\mathfrak{s}_{10}}=0$.

By Remark \ref{Re:DerAlg}, the derivations of $\mathfrak{s}_{10}$ read
$$
\mathfrak{der}(\mathfrak{s}_{10})=\left\{\left(
\begin{array}{cccc}
\mu_{11} & \mu_{12} & \mu_{13} & \mu_{14} \\
0 & \frac{1}{2}\mu_{11} & \mu_{23} & \mu_{12} - \mu_{13} \\
0 & 0 & \frac{1}{2}\mu_{11} & \mu_{12} \\
0 & 0 & 0 & 0
\end{array}
\right): \mu_{11}, \mu_{12}, \mu_{13}, \mu_{14}, \mu_{23} \in \mathbb{R}\right\},
$$
which give rise to the basis of $V_{\mathfrak{s}_{10}}$ of the form
\begin{align*}
&X_1 = \frac{3}{2} x_1 \partial_{x_1} + \frac{3}{2} x_2 \partial_{x_2} + x_3 \partial_{x_3} + x_4 \partial_{x_4} + \frac{1}{2} x_5 \partial_{x_5} + \frac{1}{2} x_6 \partial_{x_6}, &
&X_2 = -x_5 \partial_{x_1} - x_6 \partial_{x_2}, \\
&X_3 = x_3 \partial_{x_1} + (x_3 + x_4) \partial_{x_2} + x_5 \partial_{x_3} + (x_5 - x_6) \partial_{x_4}, &
&X_4 = x_2 \partial_{x_1} + x_6 \partial_{x_5} \\
&X_5 = (-x_3-x_4) \partial_{x_1} + x_6 \partial_{x_3} + x_6 \partial_{x_4}.
\end{align*}

For an element $r \in \Lambda^2 \mathfrak{s}_{10}$, we get
$$
[r,r] = 2(3x_1 x_6 - 3x_2 x_5 + x_2 x_6 + 2x_3 x_4 + x_4^2)e_{123} + 2(-x_3 x_5 + x_3 x_6 + x_4 x_5)e_{124} + 2(x_4-x_3) x_6e_{134} -2x_6^2e_{234}.
$$

Since $(\Lambda^3 \mathfrak{s}_{10})^{\mathfrak{s}_{10}} = 0$, the mCYBE and the CYBE are equal and they read
$$
- 3x_2 x_5 + 2x_3 x_4 + x_4^2 = 0, \quad (x_4 - x_3) x_5 = 0, \quad x_6 = 0.
$$

Since $x_6$ is the only brick for $\mathfrak{s}_{10}$, our Darboux tree starts with the cases $x_6 = 0$ and $x_6 \neq 0$. The full Darboux tree is presented below.

\begin{center}
{\small
\begin{tikzpicture}[
roundnode/.style={rounded rectangle, draw=green!40, fill=green!3, very thick, minimum size=2mm},
squarednode/.style={rectangle, draw=red!30, fill=red!2, thick, minimum size=4mm}
]
	\node[squarednode] (brick) at (0,0) {$x_6=0$};

	\node[squarednode] (u)  at (2,0) {$x_5=0$};
	\node[squarednode] (d)  at (2,-5) {$x_5 \neq 0$};

	\node[squarednode] (uu)  at (4,0) {$x_4=0$};
	\node[squarednode] (ud)  at (4,-4) {$x_4 \neq 0$};
	\node[roundnode] (du)  at (4,-5) {$\stackrel{{\rm No \, solutions}}{x_4 - x_3 \neq 0}$};
	\node[squarednode] (dd)  at (4,-6) {$x_4 -x_3 = 0$};

	\node[squarednode] (uuu)  at (6,0) {$x_3=0$};
	\node[squarednode] (uud)  at (6,-3) {$x_3 \neq 0$};
	\node[squarednode] (udu)  at (7,-4) {$x_3 + \frac{1}{2} x_4 = 0$};
	\node[roundnode] (udd)  at (7,-5) {$\stackrel{{\rm No \, solutions}}{x_3 + \frac{1}{2} x_4 \neq 0}$};
	\node[roundnode] (ddu)  at (7,-6) {$\stackrel{{\rm No \, solutions}}{x_3^2 - x_2 x_5  \neq 0}$};
	\node[squarednode] (ddd)  at (7,-7) {$x_3^2 - x_2 x_5 = 0$};
	
	\node[squarednode] (uuuu)  at (8,0) {$x_2=0$};
	\node[squarednode] (uuud)  at (8,-2) {$x_2 \neq 0$};

	\node[squarednode] (uuuuu)  at (10,0) {$x_1 = 0$};
	\node[squarednode] (uuuud)  at (10,-1) {$x_1 \neq 0$};
	
	\node[squarednode] (0) at (12,0) {0};
	\node[squarednode] (I) at (12,-1) {I};
	\node[squarednode] (II) at (12,-2) {II};
	\node[squarednode] (III) at (12,-3) {III};
	\node[squarednode] (IV) at (12,-4) {IV};
	\node[squarednode] (V) at (12,-7) {V};
	
	\draw[->] (brick.east) -- (u.west);
	\draw[->] (brick.east) -- (d.west);

	\draw[->] (u.east) -- (uu.west);
	\draw[->] (u.east) -- (ud.west);
	\draw[->] (d.east) -- (du.west);
	\draw[->] (d.east) -- (dd.west);

	\draw[->] (uu.east) -- (uuu.west);
	\draw[->] (uu.east) -- (uud.west);
	\draw[->] (ud.east) -- (udu.west);
	\draw[->] (ud.east) -- (udd.west);
	\draw[->] (dd.east) -- (ddu.west);
	\draw[->] (dd.east) -- (ddd.west);

	\draw[->] (uuu.east) -- (uuuu.west);
	\draw[->] (uuu.east) -- (uuud.west);
	
	\draw[->] (uuuu.east) -- (uuuuu.west);
	\draw[->] (uuuu.east) -- (uuuud.west);

\end{tikzpicture}
}
\end{center}
The connected parts of the subspaces denoted in the Darboux tree are the orbits of ${\rm Aut}_c(\mathfrak{s}_{10})$ within $\mathcal{Y}_{\mathfrak{s}_{10}}$. We proceed as in previous sections to obtain the orbits of ${\rm Aut}(\mathfrak{s}_{10})$ within $\mathcal{Y}_{\mathfrak{s}_{10}}$. The automorphism group of $\mathfrak{s}_{10}$ takes the form
$$
{\rm Aut}(\mathfrak{s}_{10}) = \left\{
\left(
\begin{array}{cccc}
(T^2_2)^2 & T^4_3 T^2_2 & T^4_3 T^2_2 + T^4_3 T^3_2 - T^4_2 T^2_2 & T^4_1 \\
0 & T^2_2 & T^3_2 & T^4_2 \\
0 & 0 & T^2_2 & T^4_3 \\
0 & 0 & 0 & 1
\end{array}
\right):
\begin{array}{c}
 T^2_2 \in \rz, \\
T^3_2, T^4_1, T^4_2, T^4_3 \in \rr
\end{array}
\right\}.
$$
It can be verified that there are two connected components of ${\rm Aut}(\mathfrak{s}_{10})$. These representative automorphisms of ${\rm Aut}(\mathfrak{s}_{10})$ and their extensions to $\Lambda^2\mathfrak{s}_{10}$ read
\begin{equation*}
T_\lambda:=\left(
\begin{array}{cccc}
1 & 0 & 0 & 0 \\
0 & \lambda & 0 & 0 \\
0 & 0 & \lambda & 0 \\
0 & 0 & 0 & 1
\end{array}
\right), \qquad \Lambda^2T_\lambda:={\small
\left(
\begin{array}{cccccc}
\lambda  & 0 & 0 & 0 &0&0 \\
0 & \lambda  & 0 & 0 &0&0 \\
0 & 0 & 1 & 0 &0&0 \\
0 & 0 & 0 & 1&0&0 \\
0 & 0 & 0 & 0 &\lambda&0 \\
0 & 0 & 0 & 0&0&\lambda \\
\end{array}
\right)},\quad \lambda \in \{\pm 1\}.
\end{equation*}
By applying $\Lambda^2T_\lambda$ to the orbits of ${\rm Aut}_c(\mathfrak{s}_{10})$ in $\mathcal{Y}_{\mathfrak{s}_{10}}$, we obtain the orbits of Aut($\mathfrak{s}_{10})$ in $\mathcal{Y}_{\mathfrak{s}_{10}}$, shown in Table \ref{Tab:g_orb_2}. Since $(\Lambda^2\mathfrak{s}_{10})^{\mathfrak{s}_{10}}=0$, each such orbit amounts to the separate class of equivalent coboundary Lie bialgebras up to Lie algebra automorphisms of $\mathfrak{s}_{10}$.

\subsection{Lie algebra $\mathfrak{s}_{11}$}

As previously, commutation relations in Table \ref{Tab:StruCons} allow to verify that $(\Lambda^2\mathfrak{s}_{11})^{\mathfrak{s}_{11}}=0$ and $(\Lambda^3\mathfrak{s}_{11})^{\mathfrak{s}_{11}}=0$.

By Remark \ref{Re:DerAlg}, one obtains 
$$
\mathfrak{der}(\mathfrak{s}_{11})=\left\{\left(
\begin{array}{cccc}
\mu_{11} & \mu_{12} & \mu_{13} & \mu_{14} \\
0 & \mu_{22} & 0 & -\mu_{13} \\
0 & 0 & \mu_{11} - \mu_{22} & 0 \\
0 & 0 & 0 & 0
\end{array}
\right): \mu_{11}, \mu_{12}, \mu_{13}, \mu_{14}, \mu_{22} \in \mathbb{R}\right\}.
$$
Such derivations give rise to the basis of $V_{\mathfrak{s}_{11}}$ of the form
\begin{align*}
&X_1 = x_1 \partial_{x_1} + 2x_2 \partial_{x_2} + x_3 \partial_{x_3} + x_4 \partial_{x_4}+ x_6 \partial_{x_6}, &
&X_2 = x_4 \partial_{x_2} + x_5 \partial_{x_3}, \\
&X_3 = (-x_3-x_4) \partial_{x_1} + x_6 \partial_{x_3} + x_6 \partial_{x_4}, &
&X_4 = -x_5 \partial_{x_1} - x_6 \partial_{x_2}, \\
&X_5 = x_1 \partial_{x_1} - x_2 \partial_{x_2} + x_5 \partial_{x_5} - x_6 \partial_{x_6}.
\end{align*}

For an element $r \in \Lambda^2 \mathfrak{s}_{11}$, we get
$$
[r,r] = 2(2x_1 x_6 - x_2 x_5 + x_3 x_4 + x_4^2)e_{123} + 2x_4 x_5e_{124} + 2(x_4-x_3) x_6 e_{134} -2x_6x_5e_{234}.
$$
Since $(\Lambda^3 \mathfrak{s}_{11})^{\mathfrak{s}_{11}} = \{0\}$, the mCYBE and the CYBE are equal and read
$$
2x_1 x_6 - x_2 x_5 + x_3 x_4 + x_4^2 = 0, \quad x_4 x_5 = 0, \quad (x_4 - x_3) x_6 = 0, \quad x_5 x_6 = 0.
$$

Since $x_5, x_6$ are the bricks for $\mathfrak{s}_{11}$, our Darboux tree starts with the cases $x_i = 0$ and $x_i \neq 0$, for $i \in \{5,6\}$. The Darboux tree is presented below.

\begin{center}
{\small 
\begin{tikzpicture}[
roundnode/.style={rounded rectangle, draw=green!40, fill=green!3, very thick, minimum size=2mm},
squarednode/.style={rectangle, draw=red!30, fill=red!2, thick, minimum size=4mm}
]
	\node[squarednode] (brick) at (0,0) {$x_6 \neq 0$};

	\node[squarednode] (u)  at (2,0) {$x_5=0$};
	\node[roundnode] (d)  at (2,-1) {$\stackrel{{\rm No \, solutions}}{x_5 \neq 0}$};

	\node[squarednode] (uu)  at (5,0) {$x_4 - x_3 = 0$};
	\node[roundnode] (ud)  at (5,-1) {$\stackrel{{\rm No \, solutions}}{x_4 - x_3 \neq 0}$};

	\node[squarednode] (uuu)  at (8,0) {$x_3^2 + x_1 x_6 = 0$};
	\node[roundnode] (uud)  at (8,-1) {$\stackrel{{\rm No \, solutions}}{x_3^2 + x_1 x_6 \neq 0}$};
	
	\node[squarednode] (IX) at (14,0) {IX};
	
	\draw[->] (brick.east) -- (u.west);
	\draw[->] (brick.east) -- (d.west);

	\draw[->] (u.east) -- (uu.west);
	\draw[->] (u.east) -- (ud.west);

	\draw[->] (uu.east) -- (uuu.west);
	\draw[->] (uu.east) -- (uud.west);

\end{tikzpicture}}
\end{center}
\begin{center}
{\small
\begin{tikzpicture}[
roundnode/.style={rounded rectangle, draw=green!40, fill=green!3, very thick, minimum size=2mm},
squarednode/.style={rectangle, draw=red!30, fill=red!2, thick, minimum size=4mm}
]
	\node[squarednode] (brick) at (0,0) {$x_6=0$};

	\node[squarednode] (u)  at (2,0) {$x_5=0$};
	\node[squarednode] (d)  at (2,-7) {$x_5 \neq 0$};

	\node[squarednode] (uu)  at (4,0) {$x_4=0$};
	\node[squarednode] (ud)  at (4,-6) {$x_4 \neq 0$};
	\node[roundnode] (du)  at (4,-7) {$\stackrel{{\rm No \, solutions}}{x_4 \neq 0}$};
	\node[squarednode] (dd)  at (4,-8) {$x_4 = 0$};

	\node[squarednode] (uuu)  at (7,0) {$x_3=0$};
	\node[squarednode] (uud)  at (7,-4) {$x_3 \neq 0$};
	\node[squarednode] (udu)  at (7,-6) {$x_3 + x_4 = 0$};
	\node[roundnode] (udd)  at (7,-7) {$\stackrel{{\rm No \, solutions}}{x_3 + x_4 \neq 0}$};
	\node[squarednode] (ddu)  at (7,-8) {$x_2 = 0$};
	\node[roundnode] (ddd)  at (7,-9) {$\stackrel{{\rm No \, solutions}}{x_2 \neq 0}$};

	\node[squarednode] (uuuu)  at (10,0) {$x_2=0$};
	\node[squarednode] (uuud)  at (10,-2) {$x_2 \neq 0$};
	\node[squarednode] (uudu)  at (10,-4) {$x_2 = 0$};
	\node[squarednode] (uudd)  at (10,-5) {$x_2 \neq 0$};
	\node[squarednode] (uduu)  at (10,-6) {$x_1 = 0$};
	\node[squarednode] (udud)  at (10,-7) {$x_1 \neq 0$};

	\node[squarednode] (uuuuu)  at (12,0) {$x_1=0$};
	\node[squarednode] (uuuud)  at (12,-1) {$x_1 \neq 0$};
	\node[squarednode] (uuudu)  at (12,-2) {$x_1 = 0$};
	\node[squarednode] (uuudd)  at (12,-3) {$x_1 \neq 0$};
	
	\node[squarednode] (0) at (14,0) {0};
	\node[squarednode] (I) at (14,-1) {I};
	\node[squarednode] (II) at (14,-2) {II};
	\node[squarednode] (III) at (14,-3) {III};
	\node[squarednode] (IV) at (14,-4) {IV};
	\node[squarednode] (V) at (14,-5) {V};
	\node[squarednode] (VI) at (14,-6) {VI};
	\node[squarednode] (VII) at (14,-7) {VII};
	\node[squarednode] (VIII) at (14,-8) {VIII};
	
	\draw[->] (brick.east) -- (u.west);
	\draw[->] (brick.east) -- (d.west);

	\draw[->] (u.east) -- (uu.west);
	\draw[->] (u.east) -- (ud.west);
	\draw[->] (d.east) -- (du.west);
	\draw[->] (d.east) -- (dd.west);

	\draw[->] (uu.east) -- (uuu.west);
	\draw[->] (uu.east) -- (uud.west);
	\draw[->] (ud.east) -- (udu.west);
	\draw[->] (ud.east) -- (udd.west);
	\draw[->] (dd.east) -- (ddu.west);
	\draw[->] (dd.east) -- (ddd.west);

	\draw[->] (uuu.east) -- (uuuu.west);
	\draw[->] (uuu.east) -- (uuud.west);
	\draw[->] (uud.east) -- (uudu.west);
	\draw[->] (uud.east) -- (uudd.west);
	\draw[->] (udu.east) -- (uduu.west);
	\draw[->] (udu.east) -- (udud.west);

	\draw[->] (uuuu.east) -- (uuuuu.west);
	\draw[->] (uuuu.east) -- (uuuud.west);
	\draw[->] (uuud.east) -- (uuudu.west);
	\draw[->] (uuud.east) -- (uuudd.west);

\end{tikzpicture}

\vspace{0.5cm}

}
\end{center}

The above Darboux tree gives the orbits of ${\rm Aut}_c(\mathfrak{s}_{11})$ in $\mathcal{Y}_{\mathfrak{s}_{11}}$. The Lie group ${\rm Aut}(\mathfrak{s}_{11})$ reads
$$
{\rm Aut}(\mathfrak{s}_{11}) = \left\{
\left(
\begin{array}{cccc}
T^2_2 T^3_3 & T^2_1 & -T^3_3 T^4_2 & T^4_1 \\
0 & T^2_2 & 0 & T^4_2 \\
0 & 0 & T^3_3 & 0 \\
0 & 0 & 0 & 1
\end{array}
\right): 
\begin{array}{c}
T^2_2, T^3_3 \in \rz \\
T^2_1, T^4_1, T^4_2 \in \rr
\end{array}
\right\}.
$$
There are four connected components of ${\rm Aut}(\mathfrak{s}_{11})$. Representative elements of these components and their lifts to $\Lambda^2\mathfrak{s}_{11}$ are given by
\begin{equation*}
T_{\lambda_1,\lambda_2}:=\left(
\begin{array}{cccc}
\lambda_1 \lambda_2 & 0 & 0 & 0 \\
0 & \lambda_1 & 0 & 0 \\
0 & 0 & \lambda_2 & 0 \\
0 & 0 & 0 & 1
\end{array}
\right), \quad \Lambda^2T_{\lambda_1,\lambda_2}= {\small
\left(
\begin{array}{cccccc}
  \lambda_2 & 0 & 0 & 0 &0&0 \\
0 & \lambda_1   & 0 & 0 &0&0 \\
0 & 0 & \lambda_1 \lambda_2 & 0 &0&0 \\
0 & 0 & 0 & \lambda_1 \lambda_2&0&0 \\
0 & 0 & 0 & 0 &\lambda_1&0 \\
0 & 0 & 0 & 0&0&\lambda_2 \\
\end{array}
\right),\qquad \lambda_1, \lambda_2 \in \{ \pm 1\}}.
\end{equation*}

The $\Lambda^2T_{\lambda_1,\lambda_2}$ allow us to check whether the orbits of ${\rm Aut}_c(\mathfrak{s}_{11})$ are connected by an element of ${\rm Aut}(\mathfrak{s}_{11})$ acting on $\Lambda^2\mathfrak{s}_{11}$. The representative $r$-matrices of each equivalence class are given in Table \ref{Tab:g_orb_2}. Since $(\Lambda^2\mathfrak{s}_{11})^{\mathfrak{s}_{11}}=0$, each family of $r$-matrices gives rise to a separate class of equivalent coboundary Lie bialgebras on $\mathfrak{s}_{11}$.

\subsection{Lie algebra $\mathfrak{s}_{12}$}

The structure constants for Lie algebra $\mathfrak{s}_{12}$ are given in Table \ref{Tab:StruCons}. This allows to verify that  $(\Lambda^2\mathfrak{s}_{12})^{\mathfrak{s}_{12}}=0$ and $(\Lambda^3\mathfrak{s}_{12})^{\mathfrak{s}_{12}}=0$.

By Remark \ref{Re:DerAlg}, one obtains that
$$
\mathfrak{der}(\mathfrak{s}_{12})=\left\{\left(
\begin{array}{cccc}
\mu_{11} & \mu_{12} & \mu_{13} & \mu_{14} \\
-\mu_{12} & \mu_{11} & \mu_{14} & -\mu_{13} \\
0 & 0 & 0 & 0 \\
0 & 0 & 0 & 0
\end{array}
\right): \mu_{11}, \mu_{12}, \mu_{13}, \mu_{14} \in \mathbb{R}\right\}.
$$
The previous derivations give rise to the basis of $V_{\mathfrak{s}_{12}}$ given by
\begin{align*}
&X_1 = 2x_1 \partial_{x_1} + x_2 \partial_{x_2} + x_3 \partial_{x_3} + x_4 \partial_{x_4}+ x_5 \partial_{x_5}, &
&X_2 = x_4 \partial_{x_2} + x_5 \partial_{x_3} -  x_2 \partial_{x_4} -  x_3 \partial_{x_5}, \\
&X_3 = (-x_3-x_4) \partial_{x_1} + x_6 \partial_{x_3} + x_6 \partial_{x_4}, &
&X_4 = (x_2 - x_5) \partial_{x_1} - x_6 \partial_{x_2} + x_6 \partial_{x_5}.
\end{align*}

For an element $r \in \Lambda^2 \mathfrak{s}_{12}$, we get
$$
[r,r] = -2(x_2 x_3 + x_4 x_5)e_{123} + 2(-2x_1 x_6 + x_2 x_5 -x_3^2 - x_3 x_4 - x_5^2)e_{124} - 2(x_2 + x_5)x_6e_{134} + 2(x_3 - x_4)x_6e_{234}.
$$
Since $(\Lambda^3 \mathfrak{s}_{12})^{\mathfrak{s}_{12}} = 0$, the mCYBE and CYBE are equal and read
$$
x_2 x_3 + x_4 x_5 = 0, \quad 2x_1 x_6 - x_2 x_5 + x_3^2 + x_3 x_4 + x_5^2 = 0, \quad (x_2 + x_5) x_6 = 0, \quad (x_3 - x_4) x_6 = 0.
$$

Since $x_6$ is the only brick for $\mathfrak{s}_{12}$, our Darboux tree starts with the cases $x_6 = 0$ and $x_6 \neq 0$. The full Darboux tree is given next.

\begin{center}
{\small
\begin{tikzpicture}[
roundnode/.style={rounded rectangle, draw=green!40, fill=green!3, very thick, minimum size=2mm},
squarednode/.style={rectangle, draw=red!30, fill=red!2, thick, minimum size=4mm}
]
	\node[squarednode] (brick) at (0,0) {$x_6=0$};

	\node[squarednode] (u)  at (3,0) {$\begin{array}{c}
 x_2-x_4=0\\
x_2+x_4=0 
\end{array} $};
	\node[squarednode] (d)  at (2,-2) {$x_2^2 + x_4^2 \neq 0$};

	\node[squarednode] (uu)  at (7,0) {$\begin{array}{c}
 x_3-x_5=0\\
x_3+x_5=0 
\end{array} $};
	\node[roundnode] (ud)  at (7,-1) {$\stackrel{{\rm No \, solutions}}{x_3^2 + x_5^2 \neq 0}$};
	\node[squarednode] (du)  at (5,-2) {$x_3^2 + x_5^2 = 0$};
	\node[squarednode] (dd)  at (5,-3) {$x_3^2 + x_5^2 \neq 0$};

	\node[squarednode] (uuu)  at (10,0) {$x_1=0$};
	\node[squarednode] (uud)  at (10,-1) {$x_1 \neq 0$};
	\node[squarednode] (ddu)  at (8,-3) {$x_2 x_3 + x_4 x_5 = 0$};
	\node[roundnode] (ddd)  at (8,-4) {$\stackrel{{\rm No \, solutions}}{x_2 x_3 + x_4 x_5 \neq 0}$};

	\node[roundnode] (dddu)  at (12,-3) {$\stackrel{{\rm No \, solutions}}{-x_2 x_5 + x_3^2 + x_3 x_4 + x_5^2 \neq 0}$};
	\node[squarednode] (dddd)  at (12,-4) {$-x_2 x_5 + x_3^2 + x_3 x_4 + x_5^2 = 0$};
	
	\node[squarednode] (0) at (15,0) {0};
	\node[squarednode] (I) at (15,-1) {I};
	\node[squarednode] (II) at (15,-2) {II};
	\node[squarednode] (III) at (15,-4) {III};
	
	\draw[->] (brick.east) -- (u.west);
	\draw[->] (brick.east) -- (d.west);

	\draw[->] (u.east) -- (uu.west);
	\draw[->] (u.east) -- (ud.west);
	\draw[->] (d.east) -- (du.west);
	\draw[->] (d.east) -- (dd.west);

	\draw[->] (uu.east) -- (uuu.west);
	\draw[->] (uu.east) -- (uud.west);
	\draw[->] (dd.east) -- (ddu.west);
	\draw[->] (dd.east) -- (ddd.west);

	\draw[->] (ddu.east) -- (dddu.west);
	\draw[->] (ddu.east) -- (dddd.west);

\end{tikzpicture}

\vspace{0.5cm}

\begin{tikzpicture}[
roundnode/.style={rounded rectangle, draw=green!40, fill=green!3, very thick, minimum size=2mm},
squarednode/.style={rectangle, draw=red!30, fill=red!2, thick, minimum size=4mm}
]
	\node[squarednode] (brick) at (0,0) {$x_6-k=0$, ($k\neq 0$)};

	\node[squarednode] (u)  at (4,0) {\parbox{1.7cm}{$x_2 + x_5=0 \\ x_3 - x_4 = 0$}};
	\node[roundnode] (d)  at (4,-1) {$\stackrel{{\rm No \, solutions}}{x_2 + x_5 \neq 0 \lor x_3 - x_4 \neq 0}$};

	\node[squarednode] (uu)  at (9.5,0) {$x_5^2 + x_3^2 - x_2 x_5 + x_3 x_4 + 2x_1 x_6 = 0$};
	\node[roundnode] (ud)  at (9.5,-1) {$\stackrel{{\rm No \, solutions}}{x_5^2 + x_3^2 - x_2 x_5 + x_3 x_4 + 2x_1 x_6 \neq 0}$};
	
	\node[squarednode] (IV) at (13.5,0) {IV$_{|k|>0}$};
	
	\draw[->] (brick.east) -- (u.west);
	\draw[->] (brick.east) -- (d.west);

	\draw[->] (u.east) -- (uu.west);
	\draw[->] (u.east) -- (ud.west);

\end{tikzpicture}
}
\end{center}
The connected parts of the loci of the Darboux families of the Darboux tree are the orbits of ${\rm Aut}_c(\mathfrak{s}_{12})$ within $\mathcal{Y}_{\mathfrak{s}_{12}}$. These orbits are given in Table \ref{Tab:g_orb_2}. 

Finally, let us obtain the orbits of ${\rm Aut}(\mathfrak{s}_{12})$ on $\mathcal{Y}_{\mathfrak{sl}_{12}}$. As usual, we first determine the Lie algebra automorphism group of $\mathfrak{s}_{12}$:
\begin{equation*}
{\rm Aut}(\mathfrak{s}_{12}) = \left\{
\left(
\begingroup
\setlength\arraycolsep{4pt}
\begin{array}{cccc}
T^1_1 & T^2_1 & T^3_1 & \pm T^3_2 \\
\mp T^2_1 & \pm T^1_1 & T^3_2 & \mp T^3_1 \\
0 & 0 & 1 & 0 \\
0 & 0 & 0 & \pm 1 
\end{array}
\endgroup
\right):
\begingroup
\setlength\arraycolsep{4pt}
\begin{array}{c}
 (T_1^1)^2 + (T^2_1)^2 \in \rz \\
T^1_1,T^2_1,T^3_1, T^3_2 \in \rr
\end{array}
\endgroup
\right\}.
\end{equation*}
There are two connected components of ${\rm Aut}(\mathfrak{s}_{12})$. One such element for each connected component of ${\rm Aut}(\mathfrak{s}_{12})$ and their extensions to $\Lambda^2\mathfrak{s}_{12}$ read
$$
T_\pm:=\left(
\begin{array}{cccc}
1 & 0 & 0 & 0 \\
0 & \pm 1 & 0 & 0 \\
0 & 0 & 1 & 0 \\
0 & 0 & 0 & \pm 1 
\end{array}
\right), \quad \Lambda^2T_\pm =\left(
\begin{array}{cccccc}
\pm 1 & 0 & 0 & 0&0&0 \\
0 & 1 & 0 & 0 &0&0\\
0 & 0 & \pm 1 & 0 &0&0\\
0 & 0 & 0 & \pm 1&0&0\\
0 & 0 & 0 & 0 &1 &0\\
0 & 0 & 0 & 0 &0&\pm 1\\
\end{array}
\right).
$$
 By using $\Lambda^2T_{\pm}$ to identify the orbits of ${\rm Aut}_c (\mathfrak{s}_{12})$, we obtain the orbits of the full Lie group ${\rm Aut}(\mathfrak{s}_{12})$. The representative $r$-matrices of each equivalence class are detailed in Table \ref{Tab:g_orb_2}.  Since $(\Lambda^2\mathfrak{s}_{12})^{\mathfrak{s}_{12}}=0$, each family of equivalent $r$-matrices gives rise to a separate class of equivalent coboundary Lie bialgebras.

\subsection{Lie algebra $\mathfrak{n}_1$}
Using the commutation relations for the Lie algebra $\mathfrak{n}_1$ given in Table \ref{Tab:StruCons}, one obtains that  $(\Lambda^2\mathfrak{n}_1)^{\mathfrak{n}_1}= \langle e_{12}\rangle$ and $(\Lambda^3\mathfrak{n}_1)^{\mathfrak{n}_1}=\langle e_{123},e_{124}\rangle$.

By Remark \ref{Re:DerAlg}, one has that the derivations of $\mathfrak{n}_1$ take the form
$$
\mathfrak{der}(\mathfrak{n}_1)=\left\{\left(
\begin{array}{cccc}
\mu_{11} & \mu_{12} & \mu_{13} & \mu_{14} \\
0 & \mu_{22} & \mu_{12} & \mu_{24} \\
0 & 0 & 2\mu_{22} - \mu_{11} & \mu_{34} \\
0 & 0 & 0 & \mu_{11} - \mu_{22}
\end{array}
\right): \mu_{11}, \mu_{12}, \mu_{13}, \mu_{14}, \mu_{22}, \mu_{24}, \mu_{34} \in \mathbb{R}\right\},
$$
which, by lifting them to $\Lambda^2\mathfrak{n}_1$, give rise to the basis of $V_{\mathfrak{n}_1}$ of the form
\begin{align*}
&X_1 = x_1 \partial_{x_1} + 2x_3 \partial_{x_3} - x_4 \partial_{x_4} + x_5 \partial_{x_5}, &
&X_2 = x_3 \partial_{x_2} + x_5 \partial_{x_4}, &
&X_3 = -x_4 \partial_{x_1} + x_6 \partial_{x_3}, \\
&X_4 = x_1 \partial_{x_1} + 2x_2 \partial_{x_2} - x_3 \partial_{x_3} + 3x_4 \partial_{x_4} + x_6 \partial_{x_6}, &
&X_5 = -x_5 \partial_{x_1} - x_6 \partial_{x_2}, &
&X_6 = x_3 \partial_{x_1} - x_6 \partial_{x_4}, \\
&X_7 = x_2 \partial_{x_1} + x_4 \partial_{x_2} + x_5 \partial_{x_3} + x_6 \partial_{x_5}.
\end{align*}

For an element $r \in \Lambda^2 \mathfrak{n}_1$, we get $[r, r] = 2(x_4 x_5 - x_2 x_6) e_{123} + 2(x_5^2 - x_3 x_6) e_{124} + 2 x_5 x_6 e_{134} + 2 x_6^2 e_{234}$.

Thus, CYBE reads $x_5=0, x_6 = 0$. Since $(\Lambda^3 \mathfrak{n}_1)^{\mathfrak{n}_1} = \langle e_{123}, e_{124}\rangle$, the mCYBE differs from the CYBE and reads $x_6 = 0$.

The Darboux tree for $\mathfrak{n}_1$ is given below.
\begin{center}
{\small
\begin{tikzpicture}[
roundnode/.style={rounded rectangle, draw=green!40, fill=green!3, very thick, minimum size=2mm},
squarednode/.style={rectangle, draw=red!30, fill=red!2, thick, minimum size=4mm}
]
\node[squarednode] (brick) at (0,0) {$x_6=0$};

\node[squarednode] (u)  at (2,0) {$x_5=0$};
\node[squarednode] (d)  at (2,-6) {$x_5 \neq 0$};

\node[squarednode] (uu)  at (4,0) {$x_4 = 0$};
\node[squarednode] (ud)  at (4,-4) {$x_4 \neq 0$};
\node[squarednode] (du)  at (5,-6) {$x_2 x_5 - x_3 x_4 = 0$};
\node[squarednode] (dd)  at (5,-7) {$x_2 x_5 - x_3 x_4 \neq 0$};

\node[squarednode] (uuu)  at (6,0) {$x_3=0$};
\node[squarednode] (uud)  at (6,-3) {$x_3 \neq 0$};
\node[squarednode] (udu)  at (6,-4) {$x_3 = 0$};
\node[squarednode] (udd)  at (6,-5) {$x_3 \neq 0$};

\node[squarednode] (uuuu)  at (8,0) {$x_2=0$};
\node[squarednode] (uuud)  at (8,-2) {$x_2 \neq 0$};

\node[squarednode] (uuuuu)  at (10,0) {$x_1 = 0$};
\node[squarednode] (uuuud)  at (10,-1) {$x_1 \neq 0$};

\node[squarednode] (0) at (12,0) {0};
\node[squarednode] (I) at (12,-1) {I};
\node[squarednode] (II) at (12,-2) {II};
\node[squarednode] (III) at (12,-3) {III};
\node[squarednode] (IV) at (12,-4) {IV};
\node[squarednode] (V) at (12,-5) {V};
\node[squarednode] (VI) at (12,-6) {VI};
\node[squarednode] (VII) at (12,-7) {VII};

\draw[->] (brick.east) -- (u.west);
\draw[->] (brick.east) -- (d.west);

\draw[->] (u.east) -- (uu.west);
\draw[->] (u.east) -- (ud.west);
\draw[->] (d.east) -- (du.west);
\draw[->] (d.east) -- (dd.west);

\draw[->] (uu.east) -- (uuu.west);
\draw[->] (uu.east) -- (uud.west);
\draw[->] (ud.east) -- (udu.west);
\draw[->] (ud.east) -- (udd.west);

\draw[->] (uuu.east) -- (uuuu.west);
\draw[->] (uuu.east) -- (uuud.west);

\draw[->] (uuuu.east) -- (uuuuu.west);
\draw[->] (uuuu.east) -- (uuuud.west);

\end{tikzpicture}
}
\end{center}
As in all previous sections, the orbits of {\rm Aut}$_c(\mathfrak{n}_1$) in $\mathcal{Y}_{\mathfrak{n}_1}$ are given by the connected components of the loci of the Darboux families of the above Darboux tree. Results can be found in Table \ref{Tab:g_orb_2}.

Let us obtain the orbits of {\rm Aut}$(\mathfrak{n}_1)$ in $\mathcal{Y}_{\mathfrak{n}_1}$. The automorphism group of $\mathfrak{n}_1$ takes the form
$$
{\rm Aut}(\mathfrak{n}_1) = \left\{\left(
\begin{array}{cccc}
T^3_3 (T^4_4)^2 & T^3_2 T^4_4 & T^3_1 & T^4_1 \\
0 & T^3_3 T^4_4 & T^3_2 & T^4_2 \\
0 & 0 & T^3_3 & T^4_3 \\
0 & 0 & 0 & T^4_4
\end{array}
\right): \quad T^3_3 , T^4_4 \in \rz, T^3_1, T^3_2, T^4_1, T^4_2, T^4_3 \in \rr\right\}.
$$
It has four connected components represented by the automorphisms $T_{\lambda_1, \lambda_2}$. These maps and their lifts to $\Lambda^2\mathfrak{n}_1$ read
$$
T_{\lambda_1,\lambda_2}:=\left(
\begin{array}{cccc}
\lambda_1 & 0 & 0 & 0 \\
0 & \lambda_1\lambda_2 & 0 & 0 \\
0 & 0 & \lambda_1 & 0 \\
0 & 0 & 0 & \lambda_2
\end{array}
\right), \quad \Lambda^2T_{\lambda_1,\lambda_2}=\left(
\begin{array}{cccccc}
\lambda_2 & 0 & 0 & 0&0&0 \\
0 & 1 & 0 & 0&0&0 \\
0 & 0 & \lambda_1\lambda_2 & 0&0&0 \\
0 & 0 & 0 & \lambda_2&0&0\\
0 & 0 & 0 & 0&\lambda_1&0\\
0 & 0 & 0 & 0&0&\lambda_1\lambda_2\\
\end{array}
\right),\quad \lambda_1,\lambda_2\in \{\pm 1\}.
$$
The action of the $\Lambda^2T_{\lambda_1,\lambda_2}$ on the loci of the above Darboux tree gives the families of equivalent $r$-matrices given in Table \ref{Tab:g_orb_2}. Since $(\Lambda^2\mathfrak{n}_1)^{\mathfrak{n}_1}=\langle e_{12}\rangle$ and using the information in Table \ref{Tab:g_orb_2}, we see that the families of equivalent coboundary Lie bialgebras on $\mathfrak{n}_1$ are induced by the following nine representative $r$-matrices:
$$
r_0 = 0, \quad r^{\pm}_1 = \pm e_{13}, \quad r_3 = e_{14}, \quad r_4 = e_{23}, \quad r_5 = e_{14} + e_{23}, \quad r_6 = e_{24}, \quad r^{\pm}_7 = e_{24} \pm e_{13}.
$$

{ 
	\begin{table}[h]
\centering
\resizebox{\columnwidth}{!}{
	\begin{tabular}{|c|c|c|c|c|c|c|c|c|c|}
\hline
$\mathfrak{g}$ & \textrm{Orbit} & {\rm Dim} & $x_1$ & $x_2$ & $x_3$ & $x_4$ & $x_5$ & $x_6$ & {\rm Repr. element } \\
\hhline{|=|=|=|=|=|=|=|=|=|=|}
\multirow{8}{*}{$\mathfrak{s}_1$} & ${\rm I}_\pm$ & 1  & $\rpm$ & 0 & 0 & 0 & 0 & 0 & $\pm e_{12}$ \\\hhline{|~|-|-|-|-|-|-|-|-|-|}
& {\rm II} & 1  & 0 & $\rz$ & 0 & 0 & 0 & 0 & $e_{13}$ \\\hhline{|~|-|-|-|-|-|-|-|-|-|}
& ${\rm III}_\pm$ & 2  & $\rpm$ & $\rz$ & 0 & 0 & 0 & 0 & $\pm e_{12} + e_{13}$ \\\hhline{|~|-|-|-|-|-|-|-|-|-|}
& {\rm IV} & 2  & 0 & $\rr$ & 0 & $\rz$ & 0 & 0 & $e_{23}$ \\\hhline{|~|-|-|-|-|-|-|-|-|-|}
& ${\rm V}_\pm$ & 3  & $\rpm$ & $\rr$ & 0 & $\rz$ & 0 & 0 & $\pm e_{12} + e_{23}$ \\\hhline{|~|-|-|-|-|-|-|-|-|-|}
& {\rm VI} & 3  & $\rr$ & $\rr$ & $\rz$ & 0 & 0 & 0 & $e_{14}$ \\\hhline{|~|-|-|-|-|-|-|-|-|-|}
& {\rm VII} & 3  & 0 & $\rr$ & 0 & $\rr$ & 0 & $\rz$ & $e_{34}$ \\\hhline{|~|-|-|-|-|-|-|-|-|-|}
& ${\rm VIII}_\pm$ & 4 & $\rpm$ & $\rr$ & 0 & $\rr$ & 0 & $\rz$ & $\pm e_{12} + e_{34}$  \\
\hline
\multirow{4}{*}{$\mathfrak{s}_2$} & ${\rm I}_\pm$ &1 & $\rpm$ & 0 & 0 & 0 & 0 & 0 & $\pm e_{12}$\\\hhline{|~|-|-|-|-|-|-|-|-|-|}
& ${\rm II}_\pm$ &2  & $\rr$ & $\rpm$ & 0 & 0 & 0 & 0 & $\pm e_{13}$\\\hhline{|~|-|-|-|-|-|-|-|-|-|}
& ${\rm III}$ &3 & $\rr$ & $\rr$ & $\rz$ & 0 & 0 & 0 & $ e_{14}$\\\hhline{|~|-|-|-|-|-|-|-|-|-|}
& ${\rm IV}_\pm$ &3 & $\rr$ & $\rr$ & 0 & $\rpm$ & 0 & 0 & $\pm e_{23}$\\
\hline
\multirow{13}{*}{\begin{tabular}{c} $\mathfrak{s}^{\alpha,\beta}_3$ \\ {\footnotesize ${\alpha\neq \beta}$} \\ {\footnotesize ${\alpha\neq 1}$} \\ {\footnotesize ${\beta\neq 1}$} \end{tabular}} & {\rm I} & 1  & $\rz$ & 0 & 0 & 0 & 0 & 0 & $e_{12}$ \\\hhline{|~|-|-|-|-|-|-|-|-|-|}
& {\rm II} & 1  & 0 & $\rz$ & 0 & 0 & 0 & 0 & $e_{13}$ \\\hhline{|~|-|-|-|-|-|-|-|-|-|}
& {\rm III} & 2  & $\rz$ & $\rz$ & 0 & 0 & 0 & 0 & $e_{12} + e_{13}$ \\\hhline{|~|-|-|-|-|-|-|-|-|-|}
& {\rm IV} & 1 & 0 & 0 & 0 & $\rz$ & 0 & 0 & $e_{23}$ \\\hhline{|~|-|-|-|-|-|-|-|-|-|}
& {\rm V} & 2 & $\rz$ & 0 & 0 & $\rz$ & 0 & 0 & $e_{12} + e_{23}$ \\\hhline{|~|-|-|-|-|-|-|-|-|-|}
& {\rm VI} & 2  & 0 & $\rz$ & 0 & $\rz$ & 0 & 0 & $e_{13} + e_{23}$ \\\hhline{|~|-|-|-|-|-|-|-|-|-|}
&  {${\rm VII}_{\pm}$} &  {3}  &  {$\rz$} &  {$\rz$} & {0} & $\pm x_1x_2x_4>0$  &  {0} &  {0} &  {$\pm(e_{12} + e_{13} + e_{23})$} \\
\hhline{|~|-|-|-|-|-|-|-|-|-|}
& {\rm VIII} & 3  & $\rr$ & $\rr$ & $\rz$ & 0 & 0 & 0 & $e_{14}$ \\\hhline{|~|-|-|-|-|-|-|-|-|-|}
& ${\rm IX}_{\alpha + \beta \in \{0, -1*\}}$ & 4  & $\rr$ & $\rr$ & $\rz$ & $\rz$ & 0 & 0 & $e_{14} + e_{23}$ \\\hhline{|~|-|-|-|-|-|-|-|-|-|}
& {\rm X} & 3 & $\rr$ & 0 & 0 & $\rr$ & $\rz$ & 0 & $e_{24}$ \\\hhline{|~|-|-|-|-|-|-|-|-|-|}
& ${\rm XI}_{\alpha+\beta=-1}$ & 4  & $\rr$ & $\rz$ & 0 & $\rr$ & $\rz$ & 0 & $e_{13} + e_{24}$ \\
\hhline{|~|-|-|-|-|-|-|-|-|-|}
& {\rm XIV} & 3  & 0 & $\rr$ & 0 & $\rr$ & 0 & $\rz$ & $e_{34}$ \\
\hhline{|~|-|-|-|-|-|-|-|-|-|}
& ${\rm XV}^{\alpha=-1}_{\alpha+\beta=-1,*}$ & 4  & $\rz$ & $\rr$ & 0 & $\rr$ & 0 & $\rz$ & $e_{12} + e_{34}$ \\
\hline
\hhline{|~|-|-|-|-|-|-|-|-|-|}
\multirow{7}{*}{\begin{tabular}{c} $\mathfrak{s}^{\alpha,\alpha}_3$ \\ {\footnotesize ${\alpha\neq 1}$} \end{tabular}} & ${\rm I}$& 2  & $\rr$ & $x_1^2+x_2^2\neq 0$ & $0$ & $0$ & $0$ & 0 & $e_{12}$ \\
\hhline{|~|-|-|-|-|-|-|-|-|-|}
& ${\rm II}$&1  & $0$ & $0$ & $0$ & $\rz$ & $0$ & 0 & $e_{23}$ \\
\hhline{|~|-|-|-|-|-|-|-|-|-|}
& ${\rm III}$& 3  & $\rr$ & $x_1^2+x_2^2\neq 0$ & $0$ & $\rz$ & $0$ & 0 & $e_{12}+e_{23}$ \\
\hhline{|~|-|-|-|-|-|-|-|-|-|}
& ${\rm IV}$& 3  & $\rr$ & $\rr$ & $\rz$ & $0$ & $0$ & 0 & $e_{14}$ \\
\hhline{|~|-|-|-|-|-|-|-|-|-|}
& ${\rm V}^*_{\alpha=-1/2}$& 4  & $\rr$ & $\rr$ & $\rz$ & $\rz$ & $0$ & 0 & $e_{14}+e_{23}$ \\
\hhline{|~|-|-|-|-|-|-|-|-|-|}
& ${\rm VI}^*_{\alpha=-1/2}$& 5  & \multicolumn{2}{c|}{$x_2x_5-x_1x_6\neq 0$}& $0$ & $\rr$ & $\rr$& $x_5^2+x_6^2\neq 0$ & $e_{12}+e_{34}$ \\
\hhline{|~|-|-|-|-|-|-|-|-|-|}
& ${\rm VII}$& 4  & \multicolumn{2}{c|}{$x_2x_5-x_1x_6=0$}& $0$ & $\rr$ & $\rr$ & $x_5^2+x_6^2\neq 0$ & $e_{34}$ \\
\hline 
\hhline{|~|-|-|-|-|-|-|-|-|-|}
\multirow{6}{*}{\begin{tabular}{c} $\mathfrak{s}^{1,\beta}_3$ \\ {\footnotesize ${\beta\neq 1}$} \end{tabular}} & ${\rm I}$ & 1  & $\rz$ & $0$ & $0$ & $0$ & $0$ & $0$ & $e_{12}$ \\
\hhline{|~|-|-|-|-|-|-|-|-|-|}
& ${\rm II}$ & 2  & $0$ & $\rr$ & $0$ & $x_2^2+x_4^2\neq 0$ & $0$ & $0$ & $e_{13}$ \\
\hhline{|~|-|-|-|-|-|-|-|-|-|}
&  ${\rm III}$ & 3  & $\rz$ & $\rr$ & $0$ & $x_2^2+x_4^2\neq 0$ & $0$ & $0$ & $e_{12}+e_{13}$ \\
\hhline{|~|-|-|-|-|-|-|-|-|-|}
& ${\rm IV}$ & 4  & $\rr$ & $\rr$ & $x_3^2+x_5^2\neq 0$ & $\rr$ & $x_2x_5-x_3x_4=0$ & $0$ & $e_{13}$ \\
\hhline{|~|-|-|-|-|-|-|-|-|-|}
& ${\rm V}_{\beta=-1}$ & 5  & $\rr$ & $\rr$ & $x_3^2+x_5^2\neq 0$ & $\rr$ & $x_2x_5-x_3x_4\neq 0$ & $0$ & $e_{14}+e_{23}$ \\
\hhline{|~|-|-|-|-|-|-|-|-|-|}
& ${\rm VI}$ & 4  & $0$ & $\rr$ & $0$ & $\rr$ & $0$ & $\rz$ & $e_{34}$ \\
\hline
\hhline{|~|-|-|-|-|-|-|-|-|-|}
\multirow{2}{*}{$\mathfrak{s}^{1,1}_3$} & ${\rm I}$ & 3  & $\rr$ & $\rr$ & $0$ & $x_1^2+x_2^2+x_4^2\neq 0$ & $0$& $0$& $e_{12}$ \\
 \hhline{|~|-|-|-|-|-|-|-|-|-|}
& ${\rm II}$ & 5  & \multicolumn{2}{c|}{$x_3x_4+\!x_1x_6=\!x_2x_5$}&   $\rr$ & $x_3x_4+x_1x_6=x_2x_5$ & $\rr$& $x_3^2+x_5^2+x_6^2\neq 0$ &$e_{13}+e_{34}$ \\
\hline
\multirow{10}{*}{\begin{tabular}{c} $\mathfrak{s}^\alpha_4$ \\ {\footnotesize $\alpha\notin \{0,1\}$} \end{tabular}} & {\rm I}$_\pm$ & 1  & $\rpm$ & 0 & 0 & 0 & 0 & 0 & $\pm e_{12}$ \\\hhline{|~|-|-|-|-|-|-|-|-|-|}
& {\rm II} & 1  & 0 & $\rz$ & 0 & 0 & 0 & 0 & $e_{13}$ \\\hhline{|~|-|-|-|-|-|-|-|-|-|}
& {\rm III}$_\pm$ & 2  &$\rpm$ & $\rz$ & 0 & 0 & 0 & 0 & $\pm e_{12} + e_{13}$ \\\hhline{|~|-|-|-|-|-|-|-|-|-|}
& {\rm IV} & 3  & $\rr$ & $\rr$ & $\rz$ & 0 & 0 & 0 & $e_{14}$ \\\hhline{|~|-|-|-|-|-|-|-|-|-|}
& {\rm V} & 2  & 0 & $\rr$ & 0 & $\rz$ & 0 & 0 & $e_{23}$ \\\hhline{|~|-|-|-|-|-|-|-|-|-|}
& {\rm VI}$_\pm$ & 3  & $\rpm$ & $\rr$ & 0 & $\rz$ & 0 & 0 & $\pm e_{12} + e_{23}$ \\\hhline{|~|-|-|-|-|-|-|-|-|-|}
& ${\rm VII}_{\alpha \in \{-2^*,-1\}}$ & 4  & $\rr$ & $\rr$ & $\rz$ & $\rz$ & 0 & 0 & $e_{14} + e_{23}$ \\\hhline{|~|-|-|-|-|-|-|-|-|-|}
& {\rm VIII} & 3  & 0 & $\rr$ & 0 & $\rr$ & 0 & $\rz$ & $e_{34}$ \\\hhline{|~|-|-|-|-|-|-|-|-|-|}
& ${\rm IX}^{\alpha = -2,*}_\pm$ & 4  & $\rpm$ & $\rr$ & 0 & $\rr$ & 0 & $\rz$ & $\pm e_{12} + e_{34}$ \\\hhline{|~|-|-|-|-|-|-|-|-|-|}
& ${\rm X}_{\alpha = 1}$ & 4  & $\rr$ & $\rr$ & $\rz$ & $-\frac{x_1x_6}{x_3}$ & 0 & $\rz$ & $e_{14} + e_{34}$ \\
\hline
\multirow{5}{*}{$\mathfrak{s}^1_4$} & ${\rm I}_\pm$ & 2  & $\rpm$ & $
\rr$ & 0 & 0 & 0 & 0 & $\pm e_{12}$ \\\hhline{|~|-|-|-|-|-|-|-|-|-|}
& ${\rm II}$& 1  & 0 & $\rz$ & 0 & 0 & 0 & 0 & $e_{13}$ \\\hhline{|~|-|-|-|-|-|-|-|-|-|}
& ${\rm III}$ & 3  & $\rr$ & $\rr$ & $\rz$ & 0 & 0 & 0 & $e_{14}$ \\\hhline{|~|-|-|-|-|-|-|-|-|-|}
& ${\rm IV}$ & 3  & $\rr$ & $\rr$ & 0 & $\rz$ & 0 & 0 & $  e_{23}$ \\
\hhline{|~|-|-|-|-|-|-|-|-|-|}
& ${\rm V}$ & 4  & $\rr$ & $\rr$ & \multicolumn{2}{c|}{$x_4x_3-x_1x_6=0$} & 0 & $\rz$ & $e_{14} + e_{34}$ \\
\hline
\multirow{5}{*}{$\mathfrak{s}_5$} & {\rm I} & 1  & $\rr$ & $x_1^2 + x_2^2 \neq 0$ & 0 & 0 & 0 & 0 & $e_{12}$ \\\hhline{|~|-|-|-|-|-|-|-|-|-|}
& ${\rm II}_{\pm}$ & 1  & 0 & 0 & 0 & $\rpm$ & 0 & 0 & $\pm e_{23}$ \\\hhline{|~|-|-|-|-|-|-|-|-|-|}
& ${\rm III}_{\pm}$ & 2 & $\rr$ & $x_1^2 + x_2^2 \neq 0$ & 0 & $\rpm$ & 0 & 0 & $e_{12} \pm e_{23}$ \\\hhline{|~|-|-|-|-|-|-|-|-|-|}
& {\rm IV} & 3  & $\rr$ & $\rr$ & $\rz$ & 0 & 0 & 0 & $e_{14}$ \\\hhline{|~|-|-|-|-|-|-|-|-|-|}
& $({\rm V}_\pm)^{\beta = 0}_{{\alpha=-2\beta}*}$ & 4  & $\rr$ & $\rr$ & $\rz$ & $\rpm$ & 0 & 0 & $e_{14} \pm e_{23}$ \\
\hline
 \multirow{7}{*}{$\mathfrak{s}_6$}& {\rm I} & 1  & $\rz$ & 0 & 0 & 0 & 0 & 0 & $e_{12}$ \\\hhline{|~|-|-|-|-|-|-|-|-|-|}
& {\rm I}$^{\rm ext}$ & 1  & 0 & $\rz$ & 0 & 0 & 0 & 0 & \\\hhline{|~|-|-|-|-|-|-|-|-|-|}
& {\rm II} & 2  & $\rz$ & $\rz$ & 0 & 0 & 0 & 0 & $e_{12} + e_{13}$ \\\hhline{|~|-|-|-|-|-|-|-|-|-|}
& {\rm III} & 3  & $\rr$ & $\rr$ & $\rz$ & 0 & 0 & 0 & $e_{14}$ \\\hhline{|~|-|-|-|-|-|-|-|-|-|}
& {\rm IV}$^*$ & 3  & $\rr$ & $\rr$ & 0 & $\rz$ & 0 & 0 & $e_{23}$ \\  \hhline{|~|-|-|-|-|-|-|-|-|-|}
& {\rm V}$^*$ & 2  & $\rr$ & 0 & $x_4$ & $\rz$ & 0 & 0 & $e_{14} + e_{23}$ \\
\hhline{|~|-|-|-|-|-|-|-|-|-|}
& ${\rm V}^{{\rm ext}*}$ & 2  & 0 & $\rr$ & $-x_4$ & $\rz$ & 0 & 0 &  \\

\hline
	\end{tabular}
	}
	\end{table}
}

{\scriptsize
\begin{table}[h]
\centering
\resizebox{\columnwidth}{!}{
\begin{tabular}{|c|c|c|c|c|c|c|c|c|c|}
\hline
$\mathfrak{g}$ & \textrm{Orbit} & {\rm Dim.}  & $x_1$ & $x_2$ & $x_3$ & $x_4$ & $x_5$ & $x_6$ & {\rm Repr. element} \\\hline
\hhline{|~|-|-|-|-|-|-|-|-|-|}
$\multirow{8}{*}{$\mathfrak{s}_6$}$& {\rm VI}$^*$ & 3  & $\rr$ & $\rz$ & $x_4$ & $\rz$ & 0 & 0 & $e_{13} + e_{14} + e_{23}$ \\\hhline{|~|-|-|-|-|-|-|-|-|-|}
& ${\rm VI}^{{\rm ext}*}$ & 3  & $\rz$ & $\rr$ & $-x_4$ & $\rz$ & 0 & 0 &  \\\hhline{|~|-|-|-|-|-|-|-|-|-|}	
& ${\rm VII}^*_{|k| \notin \{0,1\}}$ & 3 & $\rr$ & $\rr$ & $kx_4$ & $\rz$ & 0 & 0 & $k e_{14} + e_{23}$ \\
\hhline{|~|-|-|-|-|-|-|-|-|-|}
& {\rm VIII} & 4 & $\rr$ & $-x_4^2/x_5$ & $\rr$ & $-x_3$  & $\rz$ & 0 & $e_{24}$ \\
\hhline{|~|-|-|-|-|-|-|-|-|-|}
& {\rm VIII}$^{\rm ext}$ & 3 & $\frac{-x_4^2}{x_6}$ & $\rr$ & $\rr$ & $x_3$  & 0 & $\rz$ &   \\\hhline{|~|-|-|-|-|-|-|-|-|-|}
& {\rm IX}$^*$ & 4 & $\rr$ & $x_2+x_4^2/x_5\neq 0$ & $\rr$ & $-x_3$  & $\rz$ & 0 & $e_{24}+e_{13}$ \\
\hhline{|~|-|-|-|-|-|-|-|-|-|}
& ${\rm IX}^{{\rm ext}*}$ & 4 & $x _1 + \frac{x_4^2}{x_6} \neq 0$ & $\rr$ & $\rr$ & $x_3$  & 0 & $\rz$ &  \\\hhline{|-|-|-|-|-|-|-|-|-|-|}
\multirow{4}{*}{$\mathfrak{s}_7$} & {\rm I} & 2 & $\rr$ & $x_1^2 + x_2^2 \neq 0$ & 0 & 0 & 0 & 0 & $e_{12}$ \\\hhline{|~|-|-|-|-|-|-|-|-|-|}
& ${\rm II}_{\pm}$ & 3  & $\rr$ & $\rr$ & $\rpm$ & 0 & 0 & 0 & $\pm e_{14}$ \\\hhline{|~|-|-|-|-|-|-|-|-|-|}
& {\rm III}$^*$ & 3  & $\rr$ & $\rr$ & 0 & $\rz$ & 0 & 0 & $e_{23}$ \\\hhline{|~|-|-|-|-|-|-|-|-|-|}
& $({\rm IV}_{\pm})^*_{|k|\in \mathbb{R}_+}$ & 3  & $\rr$ & $\rr$ & $\rpm$ & $k x_3$ & 0 & 0& $\pm e_{14} \pm k e_{23}$ 
\\
\hline
\multirow{11}{*}{$\mathfrak{s}_8$} & {\rm I} & 1  & $\rz$ & 0 & 0 & 0 & 0 & 0 & $e_{12}$ \\\hhline{|~|-|-|-|-|-|-|-|-|-|}
& {\rm II} & 1  & 0 & $\rz$ & 0 & 0 & 0 & 0 & $e_{13}$ \\\hhline{|~|-|-|-|-|-|-|-|-|-|}
& {\rm III} & 2  & $\rz$ & $\rz$ & 0 & 0 & 0 & 0 & $e_{12} + e_{13}$ \\\hhline{|~|-|-|-|-|-|-|-|-|-|}
& {\rm IV} & 3  & $\rr$ & $\rr$ & $\rz$ & 0 & 0 & 0 & $e_{14}$ \\\hhline{|~|-|-|-|-|-|-|-|-|-|}
& {\rm V} & 3  & $\rr$ & $\rr$ & $\rz$ & $-(1+\alpha)x_3$ & 0 & 0 & $e_{14} - (1 + \alpha) e_{23}$ \\\hhline{|~|-|-|-|-|-|-|-|-|-|}
& {\rm VI} & 3  & $\rr$ &  $\frac{\alpha x_3^2}{x_5}$ & $\rr$ & $\alpha x_3$ & $\rz$ & 0 & $e_{24}$ \\\hhline{|~|-|-|-|-|-|-|-|-|-|}
&  {${\rm VII}^{\alpha=-1/2}_\pm$} &  {4}  &  {$\rr$} &  {$x_2x_5-\alpha x_3^2\in \rpm $} &  {$\rr$} &  {$\alpha x_3$} &  {$\rz$} &  {0} & $e_{13} \pm  e_{24}$ 
\\\hhline{|~|-|-|-|-|-|-|-|-|-|}
& {\rm VIII} & 3  & $-\frac{x_3^2}{x_6}$ & $\rr$ & $\rr$ & $x_3$ & 0 & $\rz$ & $e_{34}$ \\
\hline
\multirow{4}{*}{$\mathfrak{s}^1_8$} & ${\rm I}$ & 2  & $\rr$ & $x_1^2+x_2^2\neq 0$ & 0 & 0 & 0 & 0 & $e_{12}$ \\\hhline{|~|-|-|-|-|-|-|-|-|-|}
& ${\rm II}$ & 3  & $\rr$ & $\rr$ & $\rz$ & 0 & 0 & 0 & $e_{14}$ \\\hhline{|~|-|-|-|-|-|-|-|-|-|}
& ${\rm III}$  & 3  & $\rr$ & $\rr$ & $\rz$ & $-2x_3$ & 0 & 0 & $e_{14} -2 e_{23}$ \\\hhline{|~|-|-|-|-|-|-|-|-|-|}
&${\rm IV}$ & 4  & \multicolumn{3}{c|}{$x_1x_6-x_2x_5+x_3^2=0$}  & $ x_3$ & $\rr$ & $x_5^2+x_6^2\neq 0$ & $e_{24}$ 
\\
\hline
\multirow{3}{*}{$\mathfrak{s}_9$} & {\rm I} & 2 & $\rr$ & $x_1^2 + x_2^2 \neq 0$ & 0 & 0 & 0 & 0 & $e_{12}$ \\\hhline{|~|-|-|-|-|-|-|-|-|-|}
& {\rm II}$_\pm$ & 2  & $\rr$ & $\rr$ & $\rpm$ & 0 & 0 & 0 & $\pm e_{14}$ \\\hhline{|~|-|-|-|-|-|-|-|-|-|}
& {\rm III}$_\pm$ & 3  & $\rr$ & $\rr$ & $\rpm$ & $-2\alpha x_3$ & 0 & 0 & $\pm 2\alpha e_{23} \mp e_{14}$ \\
\hline
\multirow{5}{*}{$\mathfrak{s}_{10}$} & {\rm I} & 1  & $\rz$ & 0 & 0 & 0 & 0 & 0 & $e_{12}$ \\\hhline{|~|-|-|-|-|-|-|-|-|-|}
& {\rm II} & 2  & $\rr$ & $\rz$ & 0 & 0 & 0 & 0 & $e_{13}$ \\\hhline{|~|-|-|-|-|-|-|-|-|-|}
& {\rm III}$_\pm$ & 3  & $\rr$ & $\rr$ & $\rpm$ & 0 & 0 & 0 & $\pm e_{14}$ \\\hhline{|~|-|-|-|-|-|-|-|-|-|}
& {\rm IV}$_\pm$ & 3  & $\rr$ & $\rr$ & $\rpm$ & $-2x_3$ & 0 & 0 & $\pm e_{14} \mp 2 e_{23}$ \\\hhline{|~|-|-|-|-|-|-|-|-|-|}
& {\rm V} & 3  & $\rr$ & $\frac{x_3^2}{x_5}$ & $\rr$ & $x_3$ & $\rz$ & 0 & $e_{24}$ \\
\hline
\multirow{9}{*}{$\mathfrak{s}_{11}$} & {\rm I} & 1  & $\rz$ & 0 & 0 & 0 & 0 & 0 & $e_{12}$ \\\hhline{|~|-|-|-|-|-|-|-|-|-|}
& {\rm II} & 1  & 0 & $\rz$ & 0 & 0 & 0 & 0 & $e_{13}$ \\\hhline{|~|-|-|-|-|-|-|-|-|-|}
& {\rm III} & 2  & $\rz$ & $\rz$ & 0 & 0 & 0 & 0 & $e_{12} + e_{13}$ \\\hhline{|~|-|-|-|-|-|-|-|-|-|}
& {\rm IV} & 2  & $\rr$ & 0 & $\rz$ & 0 & 0 & 0 & $e_{14}$ \\\hhline{|~|-|-|-|-|-|-|-|-|-|}
& {\rm V} & 3  & $\rr$ & $\rz$ & $\rz$ & 0 & 0 & 0 & $e_{13} + e_{14}$ \\\hhline{|~|-|-|-|-|-|-|-|-|-|}
& {\rm VI} & 2  & 0 & $\rr$ & $\rz$ & $-x_3$ & 0 & 0 & $e_{14} - e_{23}$ \\\hhline{|~|-|-|-|-|-|-|-|-|-|}
& {\rm VII} & 3  & $\rz$ & $\rr$ & $\rz$ & $-x_3$ & 0 & 0 & $e_{12} + e_{14} - e_{23}$ \\\hhline{|~|-|-|-|-|-|-|-|-|-|}
& {\rm VIII} & 3 & $\rr$ & 0 & $\rr$ & 0 & $\rz$ & 0 & $e_{24}$ \\\hhline{|~|-|-|-|-|-|-|-|-|-|}
& {\rm IX} & 3  & $-\frac{x_3^2}{x_6}$ & $\rr$ & $\rr$ & $x_3$ & 0 & $\rz$ & $e_{34}$ \\
\hline
\multirow{4}{*}{$\mathfrak{s}_{12}$} & {\rm I} & 1  & $\rz$ & 0 & 0 & 0 & 0 & 0 & $e_{12}$ \\\hhline{|~|-|-|-|-|-|-|-|-|-|}
& {\rm II} & 3 & $\rr$ & $\rr$ & 0 & $x_2^2 + x_4^2 \neq 0$ & 0 & 0 & $e_{13}$ \\\hhline{|~|-|-|-|-|-|-|-|-|-|}
& {\rm III} & 3  & $\rr$ & $x_5$ & $\rr$ & $-x_3$ & $x_3^2 + x_5^2 \neq 0$ & 0 & $e_{13} + e_{24}$ \\\hhline{|~|-|-|-|-|-|-|-|-|-|}
& {\rm IV}$_{|k|>0}$ & 2  & $-\frac{x_2^2 + x_3^2}{x_6}$ & $\rr$ & $\rr$ & $x_3$ & $-x_2$ & $k$ & $ke_{34}$ \\
\hline
\multirow{7}{*}{$\mathfrak{n}_{1}$} & {\rm I} & 1 & $\rz$ & 0 & 0 & 0 & 0 & 0 & $e_{12}$ \\\hhline{|~|-|-|-|-|-|-|-|-|-|}
& {\rm II}$_\pm$ & 2 & $\rr$ & $\rpm$ & 0 & 0 & 0 & 0 & $\pm e_{13}$ \\\hhline{|~|-|-|-|-|-|-|-|-|-|}
& {\rm III} & 3 & $\rr$ & $\rr$ & $\rz$ & 0 & 0 & 0 & $e_{14}$ \\\hhline{|~|-|-|-|-|-|-|-|-|-|}
& {\rm IV} & 3 & $\rr$ & $\rr$ & 0 & $\rz$ & 0 & 0 & $e_{23}$ \\\hhline{|~|-|-|-|-|-|-|-|-|-|}
& {\rm V} & 4 & $\rr$ & $\rr$ & $\rz$ & $\rz$ & 0 & 0 & $e_{14} + e_{23}$ \\\hhline{|~|-|-|-|-|-|-|-|-|-|}
& {\rm VI}$^*$ & 4 & $\rr$ & $ \frac{x_3 x_4}{x_5}$ & $\rr$ & $\rr$ & $\rz$ & 0 & $e_{24}$ \\\hhline{|~|-|-|-|-|-|-|-|-|-|}
& {\rm VII}$_{\pm}^*$ & 5 & $\rr$& $x_2- \frac{x_3 x_4}{x_5}\in \rpm$&$\rr$ &$\rr$ & $\rz$ & 0 & $e_{24}\pm e_{13}$\\
\hline
\end{tabular}
}
\caption{Orbits of the action of ${\rm Aut} (\mathfrak{g})$ on $\mathcal{Y}_{\mathfrak{g}}$, their dimensions, elements, and representatives, for real four-dimensional indecomposable Lie algebras $\mathfrak{g}$. Each Lie algebra is divided into several subsets enumerated by roman numbers, which classify the orbits of ${\rm Aut}(\mathfrak{g})$ in $\mathcal{Y}_{\mathfrak{g}}$. The orbits of ${\rm Aut}_c(\mathfrak{g})$ are given by the (topologically) connected components of each orbit of ${\rm Aut}(\mathfrak{g})$. The trivial orbits given by $0\in\Lambda^2\mathfrak{g}$ are not considered. Symbol $A_\pm$ stands for two orbits, one with $+$ and other with $-$. In these cases, $\mathbb{R}_{\pm}$ stands for $\mathbb{R}_+$ for the first orbit and $\mathbb{R}_-$ for the second. If an orbit of Aut$(\mathfrak{g})$ in $\Lambda^2\mathfrak{g}$ belongs to $\mathcal{Y}_{\mathfrak{g}}$ only for a certain set of parameters, each family of possible values of the parameters is indicated in a subindex, first, or as a superindex,  if a second family of parameters is available. A star $(*)$ is used to denote $r$-matrices that are not solutions to the CYBE.}\label{Tab:g_orb_1}
\label{Tab:g_orb_2}
\end{table}
}

\section{Decomposable case: Lie algebra $\mathfrak{gl}_2$}\label{Ch:Darb_Sec:4Ddec}

As a final example illustrating the effectiveness of the classification method via Darboux families, we consider the decomposable Lie algebra $\mathfrak{gl}_2 \simeq \mathfrak{sl}_{2}\oplus \mathbb{R}$. The structure constants for the Lie algebra $\mathfrak{sl}_2$ are given in Table \ref{tabela3w}. On can immediately see that $(\Lambda^2\mathfrak{gl}_2)^{\mathfrak{gl}_2}=0$ and $(\Lambda^3\mathfrak{gl}_2)^{\mathfrak{gl}_2} = \langle e_{123} \rangle$.

By Remark \ref{Re:DerAlg}, the derivations of $\mathfrak{gl}_2$ take the form 
$$
\mathfrak{der}(\mathfrak{gl}_2)=\left\{\left(
\begin{array}{cccc}
 0 & \mu_{12} & \mu_{13} & 0 \\
 -\mu_{13} & \mu_{22} & 0 & 0 \\
 -\mu_{12} & 0 & -\mu_{22} & 0 \\
 0 & 0 & 0 & \mu_{44} \\
\end{array}
\right), \quad \mu_{12}, \mu_{13}, \mu_{22}, \mu_{44} \in \mathbb{R}\right\}.
$$
These derivations lead to the following basis of $V_{\mathfrak{gl}_2}$, fundamental vector fields of the action of ${\rm Aut}(\mathfrak{gl}_2)$ on $\Lambda^2 \mathfrak{gl}_2$:
\begin{align*}
&X_1:=x_4 \partial_{x_2} + x_5 \partial_{x_3} + x_1 \partial_{x_4} -x_3 \partial_{x_6}, & &X_2=-x_4 \partial_{x_1}+x_6 \partial_{x_3}-x_2 \partial_{x_4}-x_3 \partial_{x_5},\\
&X_3=x_1 \partial_{x_1}-x_2 \partial_{x_2}+x_5 \partial_{x_5}-x_6 \partial_{x_6}, & &X_4=x_3 \partial_{x_3}+x_5 \partial_{x_5}+x_6 \partial_{x_6}.
\end{align*}

For an element $r \in \Lambda^2 \mathfrak{sl}_2$, we get
$$
[r,r] = -2(2x_1x_2-x_4^2)e_{123}-2(x_1x_3-x_4x_5)e_{124}+2(x_2x_3+x_4x_6)e_{123}+2(x_2x_5+x_1x_6)e_{234}.
$$
Since $(\Lambda^3\mathfrak{gl}_2)^{\mathfrak{gl}_2}=\Lambda^3\mathfrak{sl}_2\simeq \langle e_{123}\rangle$, the mCYBE for $\mathfrak{gl}_2$ reads
$$
x_1x_3-x_4x_5=0,\quad  x_2x_3+x_4x_6=0 \quad x_2x_5+x_1x_6=0.
$$
Meanwhile, the CYBE reads
$$
2x_1x_2-x_4^2=0,\qquad x_1x_3-x_4x_5=0,\quad  x_2x_3+x_4x_6=0 \quad x_2x_5+x_1x_6=0.
$$

Let us construct a Darboux tree for this Lie algebra.

\begin{center}
{\small
\begin{tikzpicture}[
roundnode/.style={rounded rectangle, draw=green!40, fill=green!3, very thick, minimum size=2mm},
squarednode/.style={rectangle, draw=red!30, fill=red!2, thick, minimum size=4mm}
]
\node[squarednode] (brick) at (0,0) {$x_3=x_5=x_6=0$};

\node[squarednode] (u)  at (4,0) {$x_4=x_2=x_1=0$};
\node[squarednode] (d)  at (4,-1) {$x_1^2+x_2^2+x_4^2\neq 0$};

\node[squarednode] (du)  at (7,-1) {$2x_1x_2 - x_4^2=0$};
\node[squarednode] (dd)  at (7,-2) {$2x_1x_2 - x_4^2 = k, k\neq 0$};

\node[squarednode] (0) at (12,0) {0};
\node[squarednode] (I) at (12,-1) {I};
\node[squarednode] (II) at (12,-2) {II};

\draw[->] (brick.east) -- (u.west);
\draw[->] (brick.east) -- (d.west);

\draw[->] (d.east) -- (du.west);
\draw[->] (d.east) -- (dd.west);

\end{tikzpicture}
}
{\small
\begin{tikzpicture}[
roundnode/.style={rounded rectangle, draw=green!40, fill=green!3, very thick, minimum size=2mm},
squarednode/.style={rectangle, draw=red!30, fill=red!2, thick, minimum size=4mm},
rnf/.style={rounded rectangle, draw=blue!40, fill=blue!3, very thick, minimum size=2mm}
]
\node[squarednode] (brick)  at (0,0) {$x_3^2+x_5^2+x_6^2\neq 0$};

\node[squarednode] (u)  at (4,0) {$x_1=x_2=x_4=0$};
\node[squarednode] (d)  at (4,-2) {$x_1^2+x_2^2+x_4^2\neq 0$};

\node[squarednode] (uu)  at (7,0) {$2x_5x_6 + x_3^2 = 0$};
\node[squarednode] (ud)  at (8,-1) {$2x_5x_6 + x_3^2 = k, \, k \neq 0$};

\node[rnf] (du)  at (7,-2) {(...)};

\node[squarednode] (0) at (12,0) {III};
\node[squarednode] (I) at (12,-1) {IV};

\draw[->] (brick.east) -- (u.west);
\draw[->] (brick.east) -- (d.west);

\draw[->] (u.east) -- (uu.west);
\draw[->] (u.east) -- (ud.west);

\draw[->] (d.east) -- (du.west);

\end{tikzpicture}
}
\end{center}

Let us finalise the analysis of the orbits of ${\rm Aut}_c(\mathfrak{gl}_2)$ by considering the last region of $\mathbb{R}^6$ denoted in the Darboux tree by the blue box. It can be easily verified that in this case, the functions $f := 2x_1 x_2 - x_4^2$ and $g := 2x_5 x_6 + x_3$ give a Darboux family. Moreover, the level sets $f = k_1$ for all values $k_1 \in \mathbb{R}$ form a foliation of $\mathbb{R}^3$ spanned by $x_1, x_2, x_4$. Similar situation happens for $g = k_2$ with any $k_2 \in \mathbb{R}$. Thus, we can restrict the analysis of the orbits to any pair of level sets $f = k_1, g = k_2$ for arbitrary $k_1, k_2 \in \mathbb{R}$. In consequence, we get nine qualitatively different cases depending on the sign of $k_1$ and $k_2$. 

In order to facilitate the analysis, we will introduce suitable coordinates on each level set. First, let us focus on the function $f = 2x_1 x_2 - x_4^2$. Set $y_1 := \frac{1}{\sqrt{2}}(x_1 + x_2)$ and $y_2 := \frac{1}{\sqrt{2}}(x_1 - x_2)$. Then, $f = y_1^2 - y_2^2 - x_4^2$. Similarly, we introduce $y_5 := \frac{1}{\sqrt{2}}(x_5 + x_6)$ and $y_6 := \frac{1}{\sqrt{2}}(x_5 - x_6)$, which gives $g = y_5^2 - y_6^2 + x_3^2$.

In the next step, we define the appropriate coordinates on each level set $f = k_1$ and $g_{k_2}$ depending on the value of $k_i \in \mathbb{R}$, $i \in \{1,2\}$. 
For $k_1 < 0$ and $k_2 > 0$, we introduce the coordinates $(\theta_i, \phi_i)$, $\theta_i \in \mathbb{R}$, $\phi_i \in [0, 2\pi[$, by
\begin{align*}
&y_1 = \sqrt{|k_1|} \sinh{\theta_1}, & &y_2 = \sqrt{|k_1|} \cosh{\theta_1} \cos{\phi_1}, & &x_4 = \sqrt{|k_1|} \cosh{\theta_1} \sin{\phi_1} \\
&x_3 = \sqrt{k_2} \cosh{\theta_2} \cos{\phi_2}, & &y_5 = \sqrt{k_2} \cosh{\theta_2} \sin{\phi_2}, & &y_6 = \sqrt{k_2} \sinh{\theta_2}
\end{align*}
For $k_1 > 0$ and $k_2 < 0$, we define $(\theta_i, \delta_i)$, $\theta_i, \delta_i \in \mathbb{R}$, by
\begin{align*}
&y_1 = a \sqrt{k_1} \cosh{\theta_1} \cosh{\delta_1}, & &y_2 = a \sqrt{k_1} \cosh{\theta_1} \sinh{\delta_1}, & &x_4 = \sqrt{k_1} \sinh{\theta_1} \\
&x_3 = \sqrt{|k_2|} \sinh{\theta_2}, & &y_5 = b \sqrt{|k_2|} \cosh{\theta_2} \sinh{\delta_2}, & &y_6 = b \sqrt{|k_2|} \cosh{\theta_2} \cosh{\delta_2}
\end{align*}
where $a, b \in \{\pm 1\}$. 
For $k_1 = k_2 = 0$, the level set $f_{k_1} = 0$ consists of two cones without a vertex. This surface will be parametrised by the coordinates $(\phi_1, z_1)$, $\phi_1 \in [0,2\pi[$, $z_1 \in \rz$, given by
\begin{align*}
&y_1 = z_1, & &y_2 = z_1 \cos{\phi_1}, & &x_4 = z_1 \sin{\phi_1} \\
&x_3 = z_2 \sin{\phi_2}, & &y_5 = z_2 \cos{\phi_2}, & &y_6 = z_2.
\end{align*}

Let us focus on the components of the modified Yang--Baxter equation, namely
$$
x_1x_3-x_4x_5=0,\quad  x_2x_3+x_4x_6=0 \quad x_2x_5+x_1x_6=0.
$$
We will consider an equivalent system of equations, obtained by adding and substracting the second equation from the first one and expressed in the coordinate system $\{y_1, y_2, x_3, x_4, y_5, y_6\}$:
$$
y_1 x_3 - y_6 x_4 = 0, \quad y_2 x_3 - y_5 x_4 = 0, \quad y_1 y_5 - y_2 y_6 = 0.
$$
By tedious computations, we solve the system above in the coordinates chosen for each nine cases of pairs $k_1, k_2$. For several cases the system is inconsistent and it yields no solution. The only nontrivial solutions, expressed in the initial coordinate system $\{x_i\}_{i=1}^{6}$, are the following:
\begin{itemize}
\item for $k_1 < 0$ and $k_2 > 0$, we get: 
\begin{itemize}
\item[1)] $\phi_2 = \pi/2 - \phi_1$, $\theta_1 = \theta_2$, which yields
\begin{align*}
&x_1 = \sqrt{\frac{|k_1|}{2}} (\sinh{\theta_1} + \cosh{\theta_1} \cos{\phi_1}), & 
&x_2 = \sqrt{\frac{|k_1|}{2}} (\sinh{\theta_1} - \cosh{\theta_1} \cos{\phi_1}), \\ 
&x_3 = \sqrt{k_2} \cosh{\theta_1} \sin{\phi_1}, & 
&x_4 = \sqrt{|k_1|} \cosh{\theta_1} \sin{\phi_1}, \\ 
&x_5 = \sqrt{\frac{k_2}{2}} (\cosh{\theta_1} \cos{\phi_1} + \sinh{\theta_1}), & 
&x_6 = \sqrt{\frac{k_2}{2}} (\cosh{\theta_1}\cos{\phi_1} - \sinh{\theta_1}), 
\end{align*}
\item[2)] $\phi_2 = 3\pi/2 - \phi_1$, $\theta_1 = -\theta_2$, which yields
\begin{align*}
&x_1 = \sqrt{\frac{|k_1|}{2}} (\sinh{\theta_1} + \cosh{\theta_1} \cos{\phi_1}), & 
&x_2 = \sqrt{\frac{|k_1|}{2}} (\sinh{\theta_1} - \cosh{\theta_1} \cos{\phi_1}), \\ 
&x_3 = - \sqrt{k_2} \cosh{\theta_1} \sin{\phi_1}, & 
&x_4 = \sqrt{|k_1|} \cosh{\theta_1} \sin{\phi_1}, \\ 
&x_5 = -\sqrt{\frac{k_2}{2}} (\cosh{\theta_1} \cos{\phi_1} + \sinh{\theta_1}), & 
&x_6 = -\sqrt{\frac{k_2}{2}} (\cosh{\theta_1}\cos{\phi_1} - \sinh{\theta_1}),
\end{align*}
\end{itemize}
\item for $k_1 > 0$ and $k_2 < 0$, we get: 
\begin{itemize}
\item[3)] $\theta_2 = \theta_1$ and $\delta_2 = \delta_1$ (for $a =+1$, $b =+1$), which yields
\begin{align*}
&x_1 = \sqrt{\frac{k_1}{2}} \cosh{\theta_1} (\cosh{\delta_1} + \sinh{\delta_1}), & 
&x_2 = \sqrt{\frac{k_1}{2}} \cosh{\theta_1} (\cosh{\delta_1} - \sinh{\delta_1}), \\ 
&x_3 = \sqrt{|k_2|} \sinh{\theta_1}, & 
&x_4 = \sqrt{k_1} \sinh{\theta_1}, \\ 
&x_5 = \sqrt{\frac{|k_2|}{2}} \cosh{\theta_1} (\sinh{\delta_1} + \cosh{\delta_1}), & 
&x_6 = \sqrt{\frac{|k_2|}{2}} \cosh{\theta_1} (\sinh{\delta_1} - \cosh{\delta_1}), 
\end{align*}
\item[4)] $\theta_2 = -\theta_1$ and $\delta_2 = \delta_1$ (for $a =+1$, $b =-1$), which yields
\begin{align*}
&x_1 = \sqrt{\frac{k_1}{2}} \cosh{\theta_1} (\cosh{\delta_1} + \sinh{\delta_1}), & 
&x_2 = \sqrt{\frac{k_1}{2}} \cosh{\theta_1} (\cosh{\delta_1} - \sinh{\delta_1}), \\ 
&x_3 = -\sqrt{|k_2|} \sinh{\theta_1}, & 
&x_4 = \sqrt{k_1} \sinh{\theta_1}, \\ 
&x_5 = -\sqrt{\frac{|k_2|}{2}} \cosh{\theta_1} (\sinh{\delta_1} + \cosh{\delta_1}), & 
&x_6 = -\sqrt{\frac{|k_2|}{2}} \cosh{\theta_1} (\sinh{\delta_1} - \cosh{\delta_1}), 
\end{align*}
\item[5)] $\theta_2 = -\theta_1$ and $\delta_2 = \delta_1$ (for $a =-1$, $b =+1$), which yields
\begin{align*}
&x_1 = -\sqrt{\frac{k_1}{2}} \cosh{\theta_1} (\cosh{\delta_1} + \sinh{\delta_1}), & 
&x_2 = -\sqrt{\frac{k_1}{2}} \cosh{\theta_1} (\cosh{\delta_1} - \sinh{\delta_1}), \\ 
&x_3 = -\sqrt{|k_2|} \sinh{\theta_1}, & 
&x_4 = \sqrt{k_1} \sinh{\theta_1}, \\ 
&x_5 = \sqrt{\frac{|k_2|}{2}} \cosh{\theta_1} (\sinh{\delta_1} + \cosh{\delta_1}), & 
&x_6 = \sqrt{\frac{|k_2|}{2}} \cosh{\theta_1} (\sinh{\delta_1} - \cosh{\delta_1}), 
\end{align*}
\item[6)] $\theta_2 = \theta_1$ and $\delta_2 = \delta_1$ (for $a =-1$, $b =-1$), which yields
\begin{align*}
&x_1 = -\sqrt{\frac{k_1}{2}} \cosh{\theta_1} (\cosh{\delta_1} + \sinh{\delta_1}), & 
&x_2 = -\sqrt{\frac{k_1}{2}} \cosh{\theta_1} (\cosh{\delta_1} - \sinh{\delta_1}), \\ 
&x_3 = \sqrt{|k_2|} \sinh{\theta_1}, & 
&x_4 = \sqrt{k_1} \sinh{\theta_1}, \\ 
&x_5 = -\sqrt{\frac{|k_2|}{2}} \cosh{\theta_1} (\sinh{\delta_1} + \cosh{\delta_1}), & 
&x_6 = -\sqrt{\frac{|k_2|}{2}} \cosh{\theta_1} (\sinh{\delta_1} - \cosh{\delta_1}), 
\end{align*}
\end{itemize}
\item for $k_1 = 0, k_2 = 0$, we get 
\begin{itemize}
\item[7)] $\phi_2 = \phi_1$, which yields
\begin{align*}
&x_1 = \frac{z_1}{\sqrt{2}} (1 + \cos{\phi_1}), & 
&x_2 = \frac{z_1}{\sqrt{2}} (1 - \cos{\phi_1}), & 
&x_3 = z_2 \sin{\phi_1},\\ 
&x_4 = z_1 \sin{\phi_1},& 
&x_5 = \frac{z_2}{\sqrt{2}} (\cos{\phi_1} + 1), & &x_6 = \frac{z_2}{\sqrt{2}} (\cos{\phi_1} - 1). 
\end{align*}
\end{itemize}
\end{itemize}

In each case, we compute the rank of the distribution $\mathcal{D}_{\mathfrak{gl}_2}$ spanned by the vector fields $V_{\mathfrak{gl}_2}$, which informs us about the dimension of the orbits of the action of ${\rm Aut}_c (\mathfrak{gl}_2)$ on the analysed subset of $\Lambda^2 \mathfrak{gl}_2$. For solutions 1)-6) we obtain ${\rm rk} \mathcal{D}_{\mathfrak{gl}_2} = 2$, whereas we get ${\rm rk} \mathcal{D}_{\mathfrak{gl}_2} = 3$ in the case 7). Thus, we conclude that each presented solution is a single orbit of ${\rm Aut_c} (\mathfrak{gl}_2)$.

For convenience, we choose the following representative element of each family of orbits:
\begin{itemize}
\item[1)] $r_1 = \sqrt{k_2} e_{14} + \sqrt{|k_1|} e_{23}$ (for $\phi_1 = \pi/2$, $\theta_1 = 0$ and $k_1 < 0$, $k_2 > 0$)
\item[2)] $r_2 = \sqrt{k_2} e_{14} - \sqrt{|k_1|} e_{23}$ (for $\phi_1 = 3\pi/2$, $\theta_1 = 0$ and $k_1 < 0$, $k_2 > 0$)
\item[3)] $r_3 = \sqrt{\frac{k_1}{2}} e_{12} + \sqrt{\frac{k_1}{2}} e_{13} + \sqrt{\frac{|k_2|}{2}} e_{24} - \sqrt{\frac{|k_2|}{2}} e_{34}$ (for $\theta_1 = 0$, $\delta_1 = 0$ and $k_1 > 0$, $k_2 < 0$)
\item[4)] $r_4 = \sqrt{\frac{k_1}{2}} e_{12} + \sqrt{\frac{k_1}{2}} e_{13} - \sqrt{\frac{|k_2|}{2}} e_{24} + \sqrt{\frac{|k_2|}{2}} e_{34}$ (for $\theta_1 = 0$, $\delta_1 = 0$ and $k_1 > 0$, $k_2 < 0$)
\item[5)] $r_5 = -\sqrt{\frac{k_1}{2}} e_{12} - \sqrt{\frac{k_1}{2}} e_{13} + \sqrt{\frac{|k_2|}{2}} e_{24} - \sqrt{\frac{|k_2|}{2}} e_{34}$ (for $\theta_1 = 0$, $\delta_1 = 0$ and $k_1 > 0$, $k_2 < 0$)
\item[6)] $r_6 = -\sqrt{\frac{k_1}{2}} e_{12} - \sqrt{\frac{k_1}{2}} e_{13} - \sqrt{\frac{|k_2|}{2}} e_{24} + \sqrt{\frac{|k_2|}{2}} e_{34}$ (for $\theta_1 = 0$, $\delta_1 = 0$ and $k_1 > 0$, $k_2 < 0$)
\item[7)] $r_7 = z_1 \sqrt{2} e_{12} + z_2 \sqrt{2} e_{24}$ (for $\phi_1 = 0$).
\end{itemize}

Finally, let us study the orbits of the ${\rm Aut}(\mathfrak{gl}_2)$. The group of automorphisms of $\mathfrak{gl}_2$ reads
\begin{equation*}
\begin{split}
&{\rm Aut}(\mathfrak{gl}_2) = \left\{
\left( 
\begin{array}{cccc}
T_1^1 & T^2_1 & \frac{1 - (T_1^1)^2}{2 T_1^2} & 0 \\
\frac{1 - (T_1^1)^2}{2 T_3^1} & -\frac{(1 - T_1^1) T_1^2}{2 T_3^1} & \frac{(T_1^1 - 1)^2 (1 + T_1^1)}{4 T_1^2 T_3^1} & 0 \\
T_3^1 & \frac{T_1^2 T_3^1}{T_1^1 + 1} & -\frac{(1 + t_1^1) T_3^1}{2 T_1^2} & 0 \\
0 & 0 & 0 & T_4^4
\end{array}
\right): 
\begin{array}{c}
T_1^2, T^1_3, T_4^4 \in \rz \\
T_1^1 \in \mathbb{R} \backslash \{0, -1\}
\end{array}
\right\} \cup \\
&\left\{
\left( 
\begin{array}{cccc}
-1 & 0 & T_1^3 & 0 \\
0 & 0 & \frac{T_1^3}{T_3^1} & 0 \\
T_3^1 & \frac{T_3^1}{T_1^3} & \frac{- T_1^3 T_3^1}{2} & 0 \\
0 & 0 & 0 & T_4^4
\end{array}
\right): 
\begin{array}{c}
T_1^3 \in \rz \\ 
T_3^1 \in \rz \\ 
T^4_4 \in \rz
\end{array}
\right\} \cup
 \left\{
\left( 
\begin{array}{cccc}
-1 & T^2_1 & 0 & 0 \\
\frac{T_1^2}{T_3^2} & -\frac{(T_1^2)^2}{2 T_3^2} & \frac{1}{T_3^2} & 0 \\
0 & T_3^2 & 0 & 0 \\
0 & 0 & 0 & T_4^4
\end{array}
\right): 
\begin{array}{c}
T_1^2 \in \mathbb{R} \\
T_3^2 \in \rz \\
T_4^4 \in \rz
\end{array}
\right\} \cup \\
&\left\{
\left( 
\begin{array}{cccc}
0 & T^2_1 & \frac{1}{2 T_1^2} & 0 \\
\frac{1}{2 T_3^1} & -\frac{T_1^2}{2 T_3^1} & \frac{1}{4 T_1^2 T_3^1} & 0 \\
T_3^1 & T_1^2 T_3^1 & -\frac{T_3^1}{2 T_1^2} & 0 \\
0 & 0 & 0 & 1
\end{array}
\right): 
\begin{array}{c}
T_1^2 \in \rz \\ 
T_3^1 \in \rz \\
T^4_4 \in \rz
\end{array}
\right\} \cup
\left\{
\left( 
\begin{array}{cccc}
1 & 0 & T_1^3 & 0 \\
-\frac{T_1^3}{T_3^3} & \frac{1}{T_3^3} & -\frac{(T_1^3)^2}{2 T_3^3} & 0 \\
0 & 0 & T_3^3 & 0 \\
0 & 0 & 0 & T_4^4
\end{array}
\right): 
\begin{array}{c}
T_1^3 \in \mathbb{R} \\
T_3^3 \in \rz \\ 
T_4^4 \in \rz
\end{array}
\right\}. \cup \\
&\left\{
\left( 
\begin{array}{cccc}
1 & T^2_1 & 0 & 0 \\
0 & -\frac{T^2_1}{T_3^1} & 0 & 0 \\
T_3^1 & \frac{T_3^1 T_1^2}{2} & -\frac{T_3^1}{T_1^2} & 0 \\
0 & 0 & 0 & T_4^4
\end{array}
\right): 
\begin{array}{c}
T_1^2 \in \rz \\ 
T_3^1 \in \rz \\
T_4^4 \in \rz
\end{array}
\right\} \cup
\left\{
\left( 
\begin{array}{cccc}
-1 & 0 & 0 & 0 \\
0 & 0 & \frac{1}{T_3^2} & 0 \\
0 & T_3^2 & 0 & 0 \\
0 & 0 & 0 & T_4^4
\end{array}
\right): T_4^2,  T_4^4 \in \rz
\right\} \cup \\
&\left\{
\left( 
\begin{array}{cccc}
1 & 0 & 0 & 0 \\
0 & \frac{1}{T_3^3} & 0 & 0 \\
0 & 0 & T_3^3 & 0 \\
0 & 0 & 0 & T_4^4
\end{array}
\right): T_3^3,  T_4^4 \in \rz
\right\}.
\end{split}
\end{equation*}
By applying the elements of ${\rm Aut}(\mathfrak{gl}_2)$, we can identify some of the representative elements listed before. Moreover, it can be verified by direct computations that the level set $f = k_1$ is invariant under the action of ${\rm Aut}(\mathfrak{gl}_2)$, whereas the level sets $g = k_2$ with the same sign of $k_2$ can be connected by an element of ${\rm Aut}(\mathfrak{gl}_2)$. Thus, we obtain the following families of inequivalent r-matrices:
\begin{equation*}
r_{\alpha < 0} = e_{14} + |\alpha| e_{23}, \qquad
r_{\alpha > 0} = \alpha e_{12} + \alpha e_{13} +  e_{24} - e_{34}, \qquad
r_0 = e_{12} + e_{24}.
\end{equation*}
The remaining representative nonzero elements of classes of $r$-matrices relative to ${\rm Aut}(\mathfrak{gl}_2)$ associated with other Darboux families specified in the Darboux tree for $\mathfrak{gl}_2$ read
\begin{align*}
&r_{I, \pm} = e_{12} + e_{13} \pm \sqrt{2} e_{23}, &
&r_{II, k > 0, \pm} = \pm \sqrt{\frac{k}{2}} e_{12} \pm \sqrt{\frac{k}{2}} e_{13}, &
&r_{II, k < 0} = \sqrt{|k|} e_{23}, \\
&r_{III, \pm} = \pm \sqrt{2} e_{14} + e_{24} - e_{34}, &
&r_{IV, k = 1} = e_{14}, &
&r_{IV, k = -1, \pm} = \pm e_{24} \mp e_{34}.
\end{align*}
Thus, we obtain the following classes of $r$-matrices relative to ${\rm Aut}(\mathfrak{gl}_2)$:
\begin{align*}
&r_{1, \alpha > 0, \pm} = \pm \alpha e_{12} \pm \alpha e_{13}, &
&r_{2, \alpha > 0} = \alpha e_{23}, &
&r_{3, \pm} = e_{12} + e_{13} \pm \sqrt{2} e_{23} \\
&r_{4, \pm} = \pm e_{24} \mp e_{34}, &
&r_{5} = e_{14}, &
&r_{6, \pm} = e_{24} - e_{34} \pm \sqrt{2} e_{14}, \\
&r_{7, \alpha > 0} = \alpha e_{13} + e_{24}, &
&r_{8, \alpha > 0} = \alpha e_{12} + \alpha e_{13} + e_{24} - e_{34}, &
&r_{9} = e_{12} + e_{24}.
\end{align*}
Let us note that these classes agree with the results obtained in Section \ref{Ch:alg_Sec:4D}.

\chapter{Foliated Lie systems}

Recall that a {\it Lie system} is a nonautonomous  system of first-order ordinary differential equations (ODEs) in normal form
	whose general solution can be written as a function, a so-called
	{\it superposition rule}, of a family of particular solutions and
	some constants related to initial conditions \cite{CGM00,CGM07,Dissertationes,LS,PW}. The Lie--Scheffers theorem \cite{CGM07,LS,PW,LS20} states that a Lie system
	amounts to a $t$-dependent vector field taking values in a finite-dimensional Lie algebra
	of vector fields, called a {\it Vessiot--Guldberg Lie algebra}  of the Lie system. Although most differential equations are not Lie systems \cite{CGL08,Dissertationes}, geometric properties of such systems were extensively studied and relevant applications were found both in mathematics and physics \cite{CGM00,Dissertationes,LS20,LV15,PW}, which motivated their further analysis. The theory of Lie systems has been then extended in different manners to analyse much more general families of systems of differential equations. {\it PDE Lie systems} \cite{CGM07,OG00,Ra11} were applied to the study of conditional symmetries and B\"acklund transformations  \cite{CGL19,LG18}. {\it Quasi-Lie schemes} and {\it quasi-Lie systems}
	were developed to investigate integrability conditions for systems of ODES and PDEs, e.g. dissipative Milne-Piney equations and nonlinear oscillators \cite{CGL08,CGL19,CLL09}. Superposition rules for discrete differential equations were considered by Winternitz and his collaborators \cite{PW04,RW84,RSW97}. Super-superposition rules, which are aimed at the analysis of general solutions to superdifferential equations, were  analysed in \cite{BGHW87,BGHW90}. A detailed survey on the previous and other generalisations of the theory of Lie systems can be found in \cite{LS20}.

In this Chapter, based on \cite{CLW22}, we focus on the least studied generalisation of Lie systems, the {\it foliated Lie systems} \cite{CGM00}. We prove that such systems admit an analogue of a superposition rule, the so-called {\it foliated superposition rule}, and we discuss how to obtain such functions. We also discuss that the relation between standard Lie systems and automorphic Lie systems (see \cite{CGM00,Dissertationes,LS20,Ve93,Ve99} for details) extends to the realm of foliated Lie systems. This in turn allows us to define a special class of such systems, the so-called {\it foliated automorphic Lie systems}, which is studied in detail. Moreover, we show that foliated Lie systems provide a framework to describe a new generalisation of Ermakov systems. Finally, we show that a particular class of $\mathcal{F}$-foliated Lie systems on Lie algebras, related to standard Lie--Hamilton systems, can be studied Poisson structures induced by $r$-matrices.

	\section{On the definition of foliated Lie systems}\label{Sec:FLS_def}
	
	Let us introduce our stratified Lie system notion and
	illustrate its usefulness with several examples of physical and mathematical interest. 
	
	\begin{definition} 
 A $t$-dependent vector field $X$ on $N$ is called {\it foliated Lie system} on $N$, if it can be written as
		\begin{equation}\label{DFLS}
			X(t,x)=\sum_{\alpha=1}^rg_\alpha(t,x) X_\alpha(x),\qquad \forall t\in \mathbb{R},\quad \forall x\in N,
		\end{equation}
		where $X_1,\ldots,X_r$ span an $r$-dimensional  real  Lie algebra $V_X$ of vector fields such that the distribution $\mathcal{D}^{V_X}$ associated with $V_X$ is regular and $g_1,\ldots,g_r$ are common $t$-dependent constants of motion of the elements of $V$, i.e. $X_\alpha g_\beta =0$ on $N$ for every $\alpha,\beta=1,\ldots,r$. The expression (\ref{DFLS}) is referred to as a {\it decomposition} of $X$ and $V$ is a {\it Vessiot--Guldberg Lie algebra} of the foliated Lie system $X$. 
	\end{definition}
	
	In view of the results by Stefan and Sussmann \cite{La18,St80,Su73}, the regular distribution $\mathcal{D}^V$ associated with a Vessiot--Guldberg Lie algebra $V$ of a stratified Lie system is integrable. Thus, it gives rise to a regular stratification $\mathcal{F}$ of $N$ (also called a {\it foliation}, see Section \ref{Sec:StSus}) such that the tangent spaces to its strata  coincide  with $\mathcal{D}^V$. For that reason, a foliated Lie system is sometimes called {\it $\mathcal{F}$-foliated Lie system}, if the form of the stratification $\mathcal{F}$ is relevant in our considerations. 

 As the vector fields of $V$ are tangent to the
strata of the stratification $\mathcal{F}$, the system $X$ can be restricted to the strata of $\mathcal{F}$. Since $X_\beta g_\alpha =0$ for every $\alpha,\beta=1,\ldots,r$, the 
	restrictions of $g_1,\ldots,g_r$ to a stratum $\mathcal{F}_\lambda$ of $\mathcal{F}$ give rise to $r$ functions depending only on $t$ and not on any value $x \in \mathcal{F}_{\lambda}$. Indeed, consider a smooth curve $\gamma:u\in [0,1]\subset\mathbb{R}\mapsto \gamma(u)\in\mathcal{F}_\lambda$ connecting two points of a stratum $\mathcal{F}_\lambda$. Then, the tangent vector to each point $\gamma(u)$ of the curve $\gamma$, let us say $\dot \gamma(u)$, can be written as a linear combination $\dot \gamma(u)=\sum_{\alpha=1}^rf_\alpha(u)X_{\alpha}(\gamma(u))$ of the values of the tangent vectors $X_1(\gamma(u)),\ldots,X_r(\gamma(u))$ spanning   $T_{\gamma(u)}\mathcal{F}_\lambda$ for certain $u$-dependent functions $f_1,\ldots,f_r:[0,1]\rightarrow \mathbb{R}$. Moreover, for any function $g_i$, $i \in \{1,\ ldots, r\}$ we obtain
	$$
	g_i(t,\gamma(1))-g_i(t,\gamma(0))=\int_0^1\frac{\partial }{\partial u }[g_i(t,\gamma(u))]du=\int_0^1\sum_{\alpha=1}^rf_\alpha(u) (X_\alpha g_i)(t, \gamma(u)) du=0.
	$$
	Thus, $g(t,x)=g(t,x')$ for arbitrary points $x,x'\in \mathcal{F}_\lambda$ and any $t\in \mathbb{R}$. Therefore, the restriction of (\ref{DFLS}) to any stratum $\mathcal{F}_{\lambda} \in \mathcal{F}$ becomes a Lie system. More specifically, an $\mathcal{F}$-foliated Lie system $X$ with a Vessiot--Guldberg Lie algebra $V_X$ gives rise to a Lie system $X_{\mathcal{F}_{\lambda}}$ on each stratum $\mathcal{F}_{\lambda} \in \mathcal{F}$ with a Vessiot--Guldberg Lie algebra $V_{X_{\mathcal{F}_{\lambda}}}$ defined by the restriction to $\mathcal{F}_{\lambda}$ of the vector fields $V \in V_{X}$. Consequently, the dimensions of all the induced  Vessiot--Guldberg Lie algebras on the strata are bounded.
 
However, given a foliation $\mathcal{F}$ of a manifold $N$, it is not true in general that a family of Lie systems $X_{\mathcal{F}_{\lambda}}$ defined on each stratum $\mathcal{F}_{\lambda} \in \mathcal{F}$ gives rise to a foliated Lie system on $N$. 
Consider, for instance, the $t$-dependent vector field on $\mathbb{R}^3$ of the form 
	$$
	X(t,x,y,z)= \partial_x + \sum_{n=0}^{\infty}e^{t(n+1)}f_n(z)x^n \partial_y,
	$$
	where it is assumed that $f_0(z)$ does never vanish and each $f_n(z)$, for $n\in \mathbb{N}$, is a smooth function such that $f_n(z) = 0$ for $z\in [-n,n]$ and $f_n(z) \neq 0$ for $z \in \mathbb{R} \backslash [-n,n]$. It can be easily verified that $X$ gives rise to a regular integrable generalised distribution $\mathcal{D}^{V_X}$ spanned by the vector fields
	$$
	\mathcal{D}^{V_X}_{(x,y,z)}=\left\langle \partial_x, \partial_y \right\rangle, 
	$$
	with leaves $\mathcal{G}_z$ of the form $\mathcal{G}_z=\{(x,y,z)\in \mathbb{R}^3:x,y\in \mathbb{R}\}$ for any $z\in \mathbb{R}$.
At points with $k<|z_k|\leq k+1$, $k \in \mathbb{N}$, the vector field $X$ takes the form
	$$
	X(t,x,y,z_k)=\partial_x + \sum_{n=0}^{k}e^{t(n+1)}f_n( {z_0})x^n \partial_y = \sum_{n=0}^{k+1} g_k(t,z_0) X_n,
	$$
where $g_0(t,z) = 1$, $X_0 = \partial_x$ and $g_n(t,z) = e^{tn}f_{n-1}(z)$, $X_n = x^{n-1} \partial_y$ for $n \geq 1$. Thus, the restriction of $X$ to $\mathcal{G}_{z_k}$, denoted by $X_{z_k}$, gives a Lie system with a Vessiot--Guldberg Lie algebra $V_{z_k}=\langle \partial_x,\partial_y, x \partial_y\ldots, x^k\partial_y \rangle$. Obviously, $\dim V_{z_k} = k+1$. In consequence, we obtain the family of Lie systems $X_{z}$ on the foliation of $\mathbb{R}^3$ given by the leaves $\mathcal{G}_z$. However, we have $V_{z_k} \subset V_{z_p}$ for $p > k$, which imply that the dimensions of the Vessiot--Guldberg Lie algebras $V_z$ associated with the Lie systems in this family cannot be bounded, contrary to the earlier result. Therefore, $X$ is not a foliated Lie system.

Let us provide several examples of foliated Lie systems with relevant physical applications. 

 \begin{example}
Consider a Lie system defined by a $t$-dependent vector field on a manifold $M$ of the form
	\begin{equation}
		X(t,x)=\sum_{\alpha=1}^r a_\alpha(t)X_\alpha(x),\qquad \forall x\in M, \qquad \forall t\in \mathbb{R},
		\label{Liesyst}
	\end{equation}
	so that there exist $r^3$ real numbers $c_{\alpha\beta}\,^\gamma$, with $\alpha,\beta,\gamma=1,\ldots,r$, such that $[X_\alpha, X_\beta]={\displaystyle\sum_{\gamma=1}^r}c_{\alpha\beta}\,^\gamma X_\gamma$. 
Take any $f\in C^\infty(M)$. In general, $ f\, X$ is not a Lie system anymore, since the vector fields 
	$\bar X_\alpha=f\,X_\alpha$, with $\alpha=1,\ldots,r$, do not need to close on a finite-dimensional Lie algebra. However, if 
	$f \in C^\infty(M)$ is such that $X_\alpha f=0$ for $  \alpha=1,\ldots, r$, then the new
	$t$-dependent vector field $f\, X$ is a foliated Lie system. 
 \end{example}

 \begin{example}\label{Ex:FLS_cot}
	Consider a $t$-dependent completely integrable Hamiltonian system $(h,\omega,T^*\mathbb{R}^n)$, where $h\in C^\infty(\mathbb{R}\times T^*\mathbb{R}^n)$ and $T^*\mathbb{R}^n$ is equipped with its canonical symplectic form $\omega_0$. Let two coordinate systems on $T^*\mathbb{R}^n$ be denoted by $\{q^i,p_i\}$ and $\{Q^i,P_i\}$, $i \in \{1, \ldots, n\}$. Assume that there exists a $t$-dependent canonical transformation mapping $h(t,q,p)$ onto a new $t$-dependent Hamiltonian $H(t,P)$ that depends only on the momentum coordinates $P_i$ on $T^*\mathbb{R}^n$ \cite{AM87,Go80}.  In this case, the Hamilton equations for $H(t,P)$ read
	\begin{equation}\label{FLS1}
		\frac{dQ_i}{dt}=\frac{\partial H}{\partial P_i}(t,P),\qquad \frac{dP_i}{dt}=0,\qquad i=1,\ldots,n.
	\end{equation}
	This system is associated with the
	$t$-dependent vector field on $T^*\mathbb{R}^n$ given by
\begin{equation}\label{FLS1b}
	X^{HJ}(t,Q,P)=\sum_{i=1}^n\frac{\partial H}{\partial P_i}(t,P) \partial_{Q_i}.
	\end{equation}
	The family of vector fields $Y_t := X_{HJ}(t, Q,P)$, parametrised by $t \in \mathbb{R}$, span an abelian Lie algebra $V^{X_{HJ}}$. For certain choices of the Hamiltonian, $V^{X_{HJ}}$ is finite-dimensional. Then, $X_{HJ}$ gives a Lie system. However, this does not need to be the case. For instance, if $H(t,P)={\displaystyle \sum_{i=1}^n}\cos(tP_i)$, then
	\begin{equation}\label{FLS}
		X^{HJ}(t,Q,P)=-\sum_{i=1}^n\sin(tP_i)t\frac{\partial}{\partial Q_i}
	\end{equation}
	and $V^{X^{HJ}}$ is the infinite-dimensional abelian Lie algebra  spanned by the  vector fields 
	$
	Z_\lambda={\displaystyle\sum_{i=1}^n}\sin(\lambda P_i)\frac{\partial}{\partial Q_i},$ with $\lambda\in \mathbb{R}_+.$ Therefore, $X^{HJ}$ is not a Lie system in this particular  case.
	
	Independently of the specific form of $H(t,P)$, the manifold $T^*\mathbb{R}^n$ always admits a foliation $\mathcal{F}_{HJ}$ by leaves
	\begin{equation}\label{leavesTsQ}
		\mathcal{F}^{HJ}_k =\{(Q,P)\in T^*\mathbb{R}^n:P_1=k_1,\ldots,P_n=k_n\},
	\end{equation}
	parametrised by an $n$-tuple $k=(k_1,\ldots,k_n)\in \mathbb{R}^n$. 
	System (\ref{FLS1}) reduces on each $\mathcal{F}^{HJ}_k$ to 
	\begin{equation}\label{FLSk}
		\frac{dQ_i}{dt}=\frac{\partial H}{\partial P_i}(t,k),\qquad i=1,\ldots,n,
	\end{equation}
	which describes the integral curves of the restriction of (\ref{FLS1b}) to each $\mathcal{F}^{HJ}_k$, namely
	$$
	X^{HJ}_k=\sum_{i=1}^n\frac{\partial H}{\partial P_i}(t,k)\frac{\partial }{\partial Q_i}.
	$$
Given any fixed $k\in \mathbb{R}^n$, the family of vector fields $Y^{(k)}_t := X^{HJ}_{k}(t, Q,P)$ on $\mathcal{F}_{HJ}^k$, parametrised by $t \in \mathbb{R}$, span a Lie algebra, $V_k$, which is contained in the finite-dimensional Lie algebra spanned by the restrictions to $\mathcal{F}_k^{HJ}$ of the vector fields on $T^*\mathbb{R}^n$ given by 
	\begin{equation}\label{VI}
		V^{HJ}=\left\langle \frac{\partial}{\partial Q_1},\ldots,\frac{\partial}{\partial Q_n} \right\rangle.
	\end{equation}
	Hence, $V_k$ is finite-dimensional and $X^{HJ}_k$ is a Lie system for every  $k\in \mathbb{R}^n$. Moreover, the vector fields in $V^{HJ}$ span  an $n$-dimensional  integrable regular distribution on $T^*\mathbb{R}^n$, whose leaves are given by (\ref{leavesTsQ}).
	Therefore, (\ref{FLS1}) becomes a foliated Lie system with an associated Vessiot--Guldberg Lie algebra $V^{HJ}$. 
 \end{example}

 \begin{example}\label{Ex:FLS_lax}
 Let $\{e_i\}_{i = \{1, \ldots, n\}}$ and $\{h_j\}_{j = \{1, \ldots, n\}}$ be two bases of $\mathbb{R}$. Consider the semi-direct sum $\mathfrak{g}^{lp} := \mathbb{R}^n\rtimes \mathbb{R}^n$ with a basis $\{v_k\}_{i=1,\ldots,2n}$, where $v_i = e_i$ and $v_{i+n} = h_i$ for $i \in \{1, \ldots, n\}$ and the Lie algebra structure on $\mathfrak{g}^{lp}$ is given by
\begin{equation}\label{AlgRel}
[v_i, v_j] = 0, \quad  [v_{i+n}, v_{j+n}] = 0, \quad [v_{i+n},v_j]=2\delta^i_{j}v_j,\qquad i,j=1,\ldots,n,
\end{equation}
	where $\delta^i_j$ is the Kronecker delta function.
	Let $\{v^1,\ldots,v^{2n}\}$ be the basis of $\mathfrak{g}^{lp*}$   dual to the basis $\{v_1, \ldots v_{2n}\}$ of $\mathfrak{g}^{lp}$. 
	Hence, $\{v^1,\ldots,v^{2n}\}$ becomes a global coordinate system on $\mathfrak{g}^{lp}$. Define a family of $t$-dependent functions  $f_\alpha:  \mathbb{R}\times \mathfrak{g}^{lp} \to \mathbb{R}$  of the form  $f_\alpha(t,v)=f_\alpha(t,v^{n+1},\ldots,v^{2n})$ for $\alpha=1,\ldots,n$. 
	We set
	\begin{equation}\label{eq}
		\frac{dv}{dt}=-\sum_{\alpha=1}^{n}f_\alpha(t,v){\rm ad}_{v_\alpha}(v)=:X^{lp}(t,v),\qquad \forall v \in \mathfrak{g}^{lp},\forall t\in \mathbb{R},
	\end{equation}
	where ${\rm ad}_{v_\alpha}(v)=[v_\alpha,v]$.
	
	If $\mathfrak{g}^{lp}$ is a matrix Lie algebra, the Lie bracket of $\mathfrak{g}^{lp}$ becomes the matrix commutator. Then,  (\ref{eq}) can be rewritten in the more common manner as a Lax pair
	\begin{equation}\label{FLS2}
		\frac{dv}{dt}=[v,m],\qquad m(t,v)=\sum_{\alpha=1}^{n}f_\alpha(t,v) v_\alpha.
	\end{equation}
	If $\mathfrak{g}^{lp}$ is not a matrix Lie algebra, one can extend by $C^\infty(\mathbb{R}\times \mathfrak{g}^{lp})$-linearity the Lie bracket in $\mathfrak{g}^{lp}$ to the space $C^\infty(\mathbb{R}\times \mathfrak{g}^{lp})\otimes \mathfrak{g}^{lp}$ of $\mathfrak{g}^{lp}$-valued $t$-dependent functions on $\mathfrak{g}^{lp}$. This allows us to use the expression (\ref{FLS2}) to describe every system (\ref{eq}).
	
	Consider the unique connected and simply connected Lie group $G^{LP}$ with Lie algebra $\mathfrak{g}^{lp}$. Then, $G^{LP}$ acts on $\mathfrak{g}^{lp}$ via the adjoint action ${\rm Ad}:(g,v)\in G^{LP}\times \mathfrak{g}^{lp}\mapsto {\rm Ad}_gv\in \mathfrak{g}^{lp}$. The fundamental vector fields of the adjoint action read
	\begin{equation}\label{FVFg}
		X^{\rm ad}_w(v)=\frac{d}{dt}\bigg|_{t=0}{\rm Ad}_{\exp(-tw)}v=-{\rm ad}_w(v),
		\qquad \forall v,w\in \mathfrak{g}^{lp}.
	\end{equation}
	This enables us to bring (\ref{eq}) into the form 
	\begin{equation}\label{SF}
		\frac{dv}{dt}=\sum_{\alpha=1}^{n}f_\alpha (t,v)X^{\rm ad}_{v_\alpha}(v).
	\end{equation}
	In our chosen coordinate system and in view of (\ref{AlgRel}), the fundamental vector fields of the adjoint action for the Lie algebra $\mathfrak{g}^{lp}$ take the form 
	$$
	X^{\rm ad}_{e_\alpha}(v)=2v^{\alpha+n} \partial_{v^{\alpha}},\qquad 
	X^{\rm ad}_{h_\alpha}(v)=-2v^{\alpha} \partial_{v^{\alpha}},\qquad \alpha=1,\ldots,n.
	$$
	Hence, $X^{\rm ad}_{v_\alpha} f_\beta=0,$ for $\alpha=1,\ldots,2n$, $\beta=1,\ldots,n$. In particular, (\ref{eq}) takes the form
	$$
	\frac{dv^\alpha}{dt}=2f_\alpha(t,v)v^{\alpha+n},\qquad \frac{dv^{\alpha+n}}{dt}=0, \qquad \alpha=1,\ldots,n.
	$$
	Consider the Lie algebra $V^{\mathfrak{g}^{lp}}=\langle X^{\mathfrak{g}^{lp}}_i\rangle$, where $X^{\mathfrak{g}^{lp}}_i := 2 \partial_{v^{i}}$ for $i \in \{1, \ldots, n\}$. The previous system can be rewritten as
	\begin{equation}\label{FLS3}
		\frac{dv}{dt}=\sum_{\alpha=1}^ng_\alpha(t,v)X^{\mathfrak{g}^{lp}}_\alpha(v),\qquad g_\alpha(t,v)=f_\alpha(t,v)v^{\alpha+n},\qquad \alpha=1,\ldots,n.
	\end{equation}
	The elements of $V^{\mathfrak{g}^{lp}}$ span a distribution $\mathcal{D}^{V^{\mathfrak{g}^{lp}}}$ of rank $n$ on $\mathfrak{g}^{lp}$. The leaves of the foliation, $\mathcal{F}^{\mathfrak{g}^{lp}}$, associated with $\mathcal{D}^{V^{\mathfrak{g}^{lp}}}$ on $\mathfrak{g}^{lp}$ take the form
	\begin{equation}\label{Leafg}
		\mathcal{F}^{\mathfrak{g}^{lp}}_k=\{ (v^1,\ldots,v^{2n})\in \mathfrak{g}^{lp}\mid v^{n+1}=k_{1},\ldots v^{2n}=k_{n}\},\qquad \forall k=(k_{1},\ldots,k_{n})\in \mathbb{R}^n.
	\end{equation}
	Moreover, the functions $g_
	\alpha$, with $\alpha=1,\ldots,n$, are constants of motion of the vector fields belonging to $V^{\mathfrak{g}^{lp}}$. Therefore, (\ref{FLS3}) is a foliated Lie system associated with a Vessiot--Guldberg Lie algebra $V^{\mathfrak{g}^{lp}}$. In particular, the integration of the distribution $\mathcal{D}^{V^{\mathfrak{g}^{lp}}}$ gives rise to the foliation $\mathcal{F}^{\mathfrak{g}^{lp}}\!\!$. We can say then that $X$ is an $\mathcal{F}^{\mathfrak{g}^{lp}}\!\!$-foliated Lie system.
	
	The $t$-independent Lax pair studied in \cite{BV90} can be considered as a $t$-independent foliated Lie system of the form (\ref{FLS3}) with $f_\alpha=\partial h/\partial v^{\alpha+n}$, with $h=h(v^{n+1},\ldots,v^{2n})$ and $\alpha=1,\ldots,n$.
	
Finally, let us note that the system (\ref{FLS3}) does not need to be a Lie system. For instance, let the $t$-dependent functions $g_\alpha$ take the form 
	$$
	g_\alpha(t,v)= \sin(t v^{\alpha+n})v^{\alpha+n},\qquad \alpha=1,\ldots,n.
	$$
Then,
	$$
	X^{lp}(t,v)= 2\sum_{\alpha=1}^n\sin(t v^{\alpha+n})v^{\alpha+n}\frac{\partial}{\partial v^\alpha}
	$$
	and $V^{X^{lp}}$ is infinite-dimensional.
 \end{example}
	
	\section{Foliated superposition rules}
Recall that a Lie system is usually defined as a system admitting a superposition rule. However, our definition of a foliated Lie system does not refer to any such analogous notion and instead, it focuses on the properties of the decomposition of the system. It is then natural to ask whether one can define an analogue of a superposition rule that is suitable for foliated Lie systems. This section is devoted to address that question.
	
Let us consider again the foliated Lie system (\ref{FLS1}) on $T^*\mathbb{R}^n$. This system was associated with a foliation $\mathcal{F}^{HJ}$ whose leaves $\mathcal{F}^{HJ}_k$, with $k\in \mathbb{R}^n$, were given in (\ref{leavesTsQ}). Particular solutions to (\ref{FLS1}) have the form
	$$
	(Q^{(1)}(t),P_1=P_1^0,\ldots,P_n=P_n^0)
	$$
	for a point $P^0=(P^0_1,\ldots,P^0_n)\in \mathbb{R}^n$ and a particular solution, 
	$Q^{(1)}(t)$, to
	$$
	\frac{dQ_i}{dt}=\frac{\partial H}{\partial P_i}(t,P^0),\qquad i=1,\ldots,n.
	$$
	Moreover,  $(Q^{(1)}(t)+\hat{Q}, P^0_1,\ldots,P^0_n)$, for any $\hat {Q}\in \mathbb{R}^n$, is  another particular solution of (\ref{FLS1}) within $\mathcal{F}^{HJ}_k$. In fact, every solution to (\ref{FLS1}) within  
	$\mathcal{F}^{HJ}_k$ can be written as
	\begin{equation}
 		(Q(t),P^0_1,\ldots,P^0_n)=(Q^{(1)}(t)+\hat{Q},P^0_1,\ldots,P^0_n),
	\end{equation}
	for a unique $\hat Q\in \mathbb{R}^n$. Define a map $\Psi^{HJ}:T^*\mathbb{R}^n\times T^*\mathbb{R}^n\rightarrow T^*\mathbb{R}^n$ as
	\begin{equation}\label{SR_cot}
	\Psi^{HJ}(Q^{(1)},P^{(1)};\hat{Q},\hat{P})=(Q^{(1)}+\hat{Q},\hat {P}).
	\end{equation}
Then, every solution, $\xi(t)=(Q(t),P(t))$, to (\ref{FLS1}) with an initial condition in any $\mathcal{F}^{HJ}_k$ can be brought into the form
	$$
	\xi(t)=\Psi^{HJ}(\xi^{(1)}(t),\lambda),
	$$
	in terms of a particular solution $\xi^{(1)}(t)$ of (\ref{FLS1}) with an initial condition in $\mathcal{F}^{HJ}_k$ and a parameter $\lambda\in \mathcal{F}^{HJ}_k$. Moreover, there exists a one-to-one relation between the initial conditions of the solutions $\xi(t)$ of (\ref{FLS1}) in $\mathcal{F}^{HJ}_k$ and the values of $\lambda\in \mathcal{F}^{HJ}_k$. Finally, one has that $\Psi^{HJ}$ is a standard superposition rule involving one particular solution for any  Lie system on $T^*\mathbb{R}^n$ of the form
	$$
	\frac{dQ_i}{dt}=b_i(t),\qquad \frac{dP_i}{dt}=0,\qquad i=1,\ldots,n,
	$$
	for arbitrary $t$-dependent functions $b_1(t),\ldots,b_n(t)$. 
	
	One can see that the foliated Lie system related to the Lax pair (\ref{FLS2}) shares similar properties. This motivates to introduce 
	the following definition.
	
	\begin{definition} \label{Def:FLS}
		An {\it $\mathcal{F}$-foliated superposition rule} depending on $m$ particular solutions for a system $X$ on a manifold $N$ relative to a foliation $\mathcal{F}$ on $N$  is a superposition rule $\Psi: N^m \times N \rightarrow N$ for a certain Lie system on $N$ such that
		\begin{enumerate} 
			\item $\Psi(\mathcal{F}_k^{m+1})\subset \mathcal{F}_k$ for every leaf $\mathcal{F}_k$ of $\mathcal{F}$, 
			\item every particular  solution, $x(t)$, of $X$ containing a point of a leaf $\mathcal{F}_k$ of $\mathcal{F}$, namely there exists $t_0\in \mathbb{R}$ such that $x(t_0)\in \mathcal{F}_k$, takes the form
			$$
			x(t)=\Psi(x_{(1)}(t),\ldots,x_{(m)}(t),\lambda)
			$$
			in terms of $m$ generic particular solutions $x_{(1)}(t),\ldots, x_{(m)}(t)$ of $X$ contained in $\mathcal{F}_k$ and a unique $\lambda\in \mathcal{F}_k$. 
		\end{enumerate}
	\end{definition}

 We will refer to an $\mathcal{F}$-foliated superposition rule simply as a foliated superposition rule if $\mathcal{F}$ is known from the  context or its specific form is irrelevant to our considerations. 
	
It follows from Definition \ref{Def:FLS} that if $X$ admits an $\mathcal{F}$-foliated superposition rule, then the particular solutions of $X$ are always contained in a leaf of $\mathcal{F}$. Moreover, particular solutions of $X$ taking values in each leaf of $\mathcal{F}$ are always described  in terms of the same number $m$ of particular solutions within the same leaf.  

Recall that the system $X$ can be restricted to any leaf $\mathcal{F}_k \subset \mathcal{F}$ as every vector field $X_t := X(t, \cdot)$, for $t\in \mathbb{R}$, is tangent to the leaves of $\mathcal{F}$. Then, all the solutions of $X$ within $\mathcal{F}_k$ can be written as a function of a generic family of particular solutions of $X$ within $\mathcal{F}_k$. Additionally, the foliated superposition rule becomes a standard superposition rule for the restriction of $X$ to $\mathcal{F}_k$, which becomes a Lie system.

	In view of Definition \ref{Def:FLS}, the foliated Lie system (\ref{FLS1}) admits an $\mathcal{F}^{HJ}$-foliated superposition rule $\Psi^{HJ}:T^*\mathbb{R}^n\times T^*\mathbb{R}^n\rightarrow T^*\mathbb{R}^n$ given by (\ref{SR_cot}), whereas the leaves of $\mathcal{F}^{HJ}$ take the form (\ref{leavesTsQ}).
	
	\section{Foliated Lie--Scheffers theorem}\label{HOFSR}
	
In this section, we study the properties of first-order systems of ODEs admitting a foliated superposition rule. Our results can be
	considered as a generalisation of the standard geometric Lie--Scheffers theorem \cite{CGM07,LS,PW}. As a byproduct, our characterisation gives us an algorithm to obtain foliated superposition rules.
	
Let us begin with a well-known fact in Lie systems theory (see \cite{CGM07,Dissertationes,LS20}).
	
	\begin{definition}  If $X$ is a vector field on $N$, let us say $X={\displaystyle\sum_{i=1}^n}X^i(x)\partial/\partial x^i$, its {\it diagonal prolongation} to $N^m$ is the vector field  on $N^m$ given by
		$$
		X^{[m]}=\sum_{a=1}^m\sum_{i=1}^{n}X^i(x_{(a)})\frac{\partial}{\partial x_{(a)}^i}.
		$$
	\end{definition}
	\begin{lemma}\label{Lem:LS}
		Let $X_1,\ldots, X_r$ be vector fields on $N$ whose diagonal
		prolongations to $N^m$ are linearly independent at a generic point.
		If $b_1,\ldots,b_r\in C^\infty(N^{m+1})$, then ${\displaystyle\sum_{\alpha=1}^r}b_\alpha X^{[m+1]}_\alpha=Y^{[m+1]}$  for a vector field $Y$ on $N$ if and only if $b_1,\ldots,b_r$ are constant.
	\end{lemma}

As noted in the previous section, the definition of a foliated Lie system does not refer to the notion of a foliated superposition rule. However, these two concepts are closely related as shown by the following theorems.
 
	\begin{theorem}\label{FLST}
	  If $X$ is an $\mathcal{F}$-foliated Lie system, then it admits an $\mathcal{F}$-foliated superposition rule $\Psi:N^m\times N\rightarrow N$ such that  $
		ms\geq \dim V,
		$
		where $s$ is the dimension of the leaves of $\mathcal{F}$. 
	\end{theorem}
	\begin{proof}
Assume that $X$ is an $\mathcal{F}$-foliated Lie system with an associated Vessiot--Guldberg Lie algebra $V$ and admitting a decomposition (\ref{DFLS}). Let us construct an $\mathcal{F}$-foliated superposition rule for $X$.
		
As discussed in Section \ref{Sec:FLS_def}, any $\mathcal{F}$-foliated Lie system $X$ gives rise to a Lie system on each leaf $\mathcal{F}_k$ of $\mathcal{F}$ by restricting $X$ and its Vessiot--Guldberg Lie algebra $v$ to $\mathcal{F}_k$. Thus, we obtain a Lie system 
		$$
		X_{k}(t,x)={\displaystyle\sum_{\alpha=1}^r}g_\alpha(t,k)X_\alpha|_{\mathcal{F}_k}
		$$
with the associated Vessiot--Guldberg Lie algebra $V_{\mathcal{F}_k}$. On the other hand, $X_k$ can be considered as the restriction to $\mathcal{F}_k$ of a Lie system on $N$ of the form $Y_{(k)}={\displaystyle\sum_{\alpha=1}^r}g_\alpha(t,k)X_\alpha$. Remarkably, it possesses the same Vessiot--Guldberg Lie algebra $V$ as $X$. 

All Lie systems $Y_{(k)}$ admit a common superposition rule $\Psi:N^m\times N\rightarrow N$. However, it is not assured in general that $\Psi(\mathcal{F}_k^{m+1})$ belongs $\mathcal{F}_k$. Hence, the curve $\Psi(x_{(1)}(t),\ldots,x_{(m)}(t),\lambda)$, where the $x_{(1)}(t),\ldots,x_{(m)}(t)$ are particular solutions to $X$ and $\lambda \in \mathcal{F}_k$, might not be a solution of $X_k$ and, in consequence, a particular solution to $X$. Let us provide a method to obtain a superposition rule for all the $Y_{(k)}$ such that $\Psi(\mathcal{F}_k^{m+1})\subset \mathcal{F}_k$. This will give an $\mathcal{F}$-foliated superposition rule for $X$ because $\Psi(x_{(1)}(t),\ldots,x_{(m)}(t),\lambda)$ will be a particular solution to $Y_{(k)}$ within $\mathcal{F}_k$ and, therefore, a particular solution to $X$.

		On a neighbourhood of a generic point of $N$, there exists a local coordinate system adapted to the foliation $\mathcal{F}$ with $s$-dimensional leaves. In other words, we can construct a local coordinate system $\{\theta_1,\ldots,\theta_s,I_{s+1},\ldots, I_{n}\}$ whose first $s$ coordinates form a coordinate system on each leaf $\mathcal{F}_k$ and the last $n-s$ coordinates are constant on each leaf of $\mathcal{F}$. Then, one can write in coordinates
		$$
		X_\alpha=\sum_{i=1}^sX^i_\alpha(\theta,I)\frac{\partial}{\partial \theta^i},\qquad \alpha=1,\ldots,r,\qquad \theta=(\theta^1,\ldots,\theta^s),\,\,I=(I^{s+1},\ldots,I^n),
		$$
		for certain functions $X^i_\alpha(\theta,I)$ for $i=1,\ldots,s$ and $\alpha=1,\ldots,r$. Let us restrict ourselves to a generic leaf 
		of $\mathcal{F}$, e.g. $$
		\mathcal{F}_k=\{(\theta_1,\ldots,\theta_s,I_{s+1},\ldots,I_{n})\in N:I_{s+1}=k_{s+1},\ldots,I_{n}=k_{n}\},\qquad k=(k_{s+1},\ldots,k_{n}).
		$$
		The vector fields $X_1,\ldots,X_r$ are tangent to $\mathcal{F}_k$. Hence, their restrictions to $\mathcal{F}_k$ can be considered as vector fields on the leaf. For $m$ large enough, the diagonal prolongations $X_1^{[m]},\ldots,X_r^{[m]}$ on $N^m$ become linearly independent at a generic point  (see \cite{Dissertationes} for a proof). As the vector fields $X_1^{[m]},\ldots,X_r^{[m]}$ are tangent to $\mathcal{F}_k^{m}$, one obtains
		\begin{equation}\label{FLC}
			ms\geq \dim V,
		\end{equation}
		where $s$ is by assumption the dimension of a leaf of $\mathcal{F}$. It is worth stressing that the above expression is more accurate than  Lie's condition, which only shows that $m \dim N
		\geq \dim V$.
		
Consider the diagonal prolongations $X_1^{[m+1]},\ldots,X_{r}^{[m+1]}$ to $N^{m+1}$. Let us define a local coordinate system on $N^{m+1}$ given by $\{\theta^{(a)}_1,\ldots,\theta^{(a)}_s,I^{(a)}_{s+1},\ldots, I^{(a)}_{n}\}$ with $a=0,\ldots, m$. The vector fields $X_1^{[m+1]},\ldots,X_{r}^{[m+1]}$ admit the common first-integrals $\Psi_{s+1}=I^{(0)}_{s+1},\ldots,\Psi_n=I^{(0)}_{n}$. Since $X^{[m+1]}_1,\ldots,X^{[m+1]}_r$ span a distribution of rank $\dim V\leq m\dim N$, one can find, at least, $s$ new functionally independent first-integrals, $\Psi_1,\ldots,\Psi_s$, common to $X_1^{[m+1]},\ldots,X^{[m+1]}_r$ such that 
		$$
		\frac{\partial(\Psi_1,\ldots,\Psi_n)}{\partial(\theta_1^{(0)},\ldots,\theta_s ^{(0)},I^{(0)}_{s+1},\ldots,I^{(0)}_{n})}\neq 0
		\Longrightarrow \frac{\partial(\Psi_1,\ldots,\Psi_s)}{\partial(\theta_1^{(0)},\ldots,\theta_s^{(0)})}\neq 0.
		$$ 
		This gives rise to a mapping $\widetilde{\Psi}:N^{m+1} N\rightarrow N$ of the form
		$
		\widetilde{\Psi}(I^{(0)},\theta^{(0)},\ldots,\theta^{(m)},I^{(m)})=(\Psi_1,\ldots,\Psi_n).
		$
		Then, one can use the Implicit Function theorem to find the unique mapping $\Phi:N^m\times N\rightarrow N$ such that
		$$
		x_{(0)}=\Phi(x_{(1)},\ldots,x_{(m)},\sigma)\Leftrightarrow \widetilde{\Psi}(x_{(0)},\ldots,x_{(m)})=\sigma.
		$$
		In particular, 
		$
		\Phi(x_{(1)},\ldots,x_{(m)},\lambda)\in \mathcal{F}_{\bar k}
		$ for $\lambda=(\lambda_1,\ldots,\lambda_s,\bar{k})$, where $\lambda_1,\ldots,\lambda_s$ are arbitrary and $\bar k=(I_{s+1}^{(0)},\ldots, I_{n}^{(0)})$. In this way, given $m$ particular solutions $x_{(1)}(t),\ldots,x_{(m)}(t)$ to $X_{\bar{k}}$ belonging to the same leaf $\mathcal{F}_{\bar{k}}$, the mapping 
		$$
		\Phi(x_{(1)}(t),\ldots,x_{(m)}(t),k_1,\ldots,k_s,\bar k)=x(t)
		$$
		gives us every solution $x(t)$ to $X_{\bar{k}}$ within the leaf $\mathcal{F}_{\bar k}$ for every $(k_1,\ldots,k_s,\bar k)\in \mathcal{F}_{\bar k}$. Then, $x(t)$ is a solution of $X$ and the mapping $\Phi$ allows us to obtain all solutions to $X$ within $\mathcal{F}_{\bar{k}}$ out of particular solutions to $X$ in the same leaf and a parameter in $\mathcal{F}_{\bar k}$. 
	\end{proof}
	
	\begin{theorem}
		If a system $X$ on $N$
		admits an $\mathcal{F}$-foliated superposition rule $\Psi:N^m\times N\rightarrow N$, then there exists vector fields $X_1,\ldots,X_r$ on $N$ tangent to the leaves of $\mathcal{F}$ and common $t$-dependent constants of the motion $f_1,\ldots,f_r\in C^\infty(\mathbb{R}\times N)$ for $X_1,\ldots,X_r$ so that
		$$
		X(t,x)=\sum_{\alpha=1}^rf_\alpha(t,x)X_\alpha(x)
		$$
		and
		\begin{equation}\label{Eq:rel}
			[X_\alpha,X_\beta]=\sum_{\gamma=1}^rh_{\alpha\beta}^\gamma X_{\gamma},\qquad \alpha,\beta=1,\ldots,r,
		\end{equation}
		where $h_{\alpha\beta}^\gamma\in C^\infty(N^{m+1})$, with $\alpha,\beta,\gamma=1,\ldots,r$, are first-integrals for $X^{[m+1]}_1,\ldots,X^{[m+1]}_r$.
	\end{theorem}
	\begin{proof}
		Consider that an  $\mathcal{F}$-foliated superposition rule for the system $X$ on $N$ and the leaves of the foliation $\mathcal{F}$ have dimension $s$. By the definition of a foliated superposition rule, the particular solutions to $X$ are contained in the leaves of $\mathcal{F}$ and the vector fields $X_t$, for every $t\in \mathbb{R}$, must be tangent to its leaves. 
		
		Let us fix a point $(x_{(1)},\ldots, x_{(m)})\in \mathcal{F}_k^m$ of a leaf $\mathcal{F}_k$ of $\mathcal{F}$. We can then use the Implicit Function theorem
		to obtain a new function $\widetilde{\Phi}:N^{m+1}\rightarrow N$ such that
		$$
		F(x_{(1)},\ldots, x_{(m)},\lambda)=x_{(0)}\,\,\Leftrightarrow \,\, \widetilde{\Phi}(x_{(0)},x_{(1)},\ldots, x_{(m)})=\lambda.
		$$
		The function $\widetilde{\Phi}$ take constant values on generic families of $m+1$ particular solutions,  $x_{(0)}(t),\ldots, x_{(m)}(t)$,  of $X$  belonging to the same leaf of the foliation $\mathcal{F}$. Therefore,
		$$       
		0=\frac{d}{dt}\widetilde{\Phi}(x_{(0)}(t),\ldots,x_{(m)}(t))=[X_t^{[m+1]}\widetilde{\Phi}](x_{(0)}(t),\ldots,x_{(m)}(t))=0,
		$$
		for every $t\in \mathbb{R}$.
		Let us fix $x_{(1)},\ldots,x_{(m)}\in \mathcal{F}_k$. Then, $\lambda\in \mathcal{F}_k$ determines univocally a point $x_{(0)}$ via $\Psi$ and $x_{(0)}$ comes from a unique $\lambda$. In other words, the foliated superposition rule induces on $\mathcal{F}^{m+1}_k$ a so-called {\it horizontal foliation} over $\mathcal{F}^m_k$ relative to the projection from $\mathfrak{F}^{m+1}_{k}$ onto the last $m$ copies of $\mathcal{F}_k$. The coordinates of the function $\widetilde{\Phi}$ take constant values exactly on the leaves of the horizontal foliation. Consequently, the vector fields $\{X^{[m+1]}_t\}_{t\in \mathbb{R}}$ span on $\mathcal{F}_k^{m+1}$ a distribution $\mathcal{D}_0$ contained in the tangent space to the leaves of the horizontal foliation on $\mathcal{F}^{m+1}_k$. We can extend such a distribution with the linear combinations of successive Lie brackets of $\{X^{[m+1]}_t\}_{t\in \mathbb{R}}$ to obtain a regular distribution $\mathcal{D}$, at a generic point of $\mathcal{F}_{\bar{k}}^{m+1}$, containing $\mathcal{D}_0$ and contained in the tangent space to the leaves of the foliation on $\mathcal{F}_{\bar k}^{m+1}$. 
		
		Consider a finite family of elements $_kX_1^{[m+1]},\ldots,_kX_u^{[m+1]}$ forming a local basis of the distribution $\mathcal{D}_0$ obtained by restricting to $\mathcal{F}^{m+1}_{k}$ some vector fields $X_1^{[m+1]},\ldots,X_u^{[m+1]}$. As the linear combinations of Lie brackets of diagonal prolongations are diagonal prolongations, the previous basis can be expanded to produce a family $_kX^{[m+1]}_1,\ldots,_kX^{[m+1]}_r$  of vector fields  spanning a regular distribution $\mathcal{D}$. Since the leaves of our horizontal foliation on $\mathcal{F}^{m+1}_{\bar k}$ project diffeomorphically onto $\mathcal{F}^{m}_{\bar k}$, via the projections onto the last $m$ factors, the vector fields $_kX^{[m]}_1,\ldots,_kX^{[m]}_r$ on $\mathcal{F}_k^{m}$ become linearly independent at a generic point. By Lemma \ref{Lem:LS}, the Lie brackets of $_kX^{[m+1]}_1,\ldots,_kX^{[m+1]}_r$ close on  a finite-dimensional Lie algebra of vector fields. Hence,
		$$
		[\,_kX_\alpha,\,_kX_\beta]=\sum_{\gamma=1}^rc_{\alpha\beta}\,^\gamma(k) \,_kX_\gamma
		$$
		on $\mathcal{F}_k$. 
		The previous procedure can be extended to other leaves $\mathcal{F}_{k'}$ for different values $k'$ close enough to $k$ so that the vector fields $X^{[m+1]}_1,\ldots,X^{[m+1]}_r$ for the initial $\mathcal{F}_k$ can also be  used locally. 
		Moreover, $X_t^{[m+1]}$ must be on each leaf in $\mathcal{F}_{k}^{m+1}$ a linear combination of the $_kX_1^{[m+1]},\ldots,_kX_r^{[m+1]}$ with coefficients being functions on $\mathcal{F}^{m+1}_{k}$ depending only on $t$. Hence, 
		$$
		X_t^{[m+1]}(\xi)=\sum_{\alpha=1}^{r}g_\alpha(t,\xi) +kX_\alpha^{[m+1]}(\xi),
		$$
		for certain functions $g_\alpha(t,\xi)$, with $\xi\in \mathcal{F}_{k}^{m+1}$.
		Let us restrict the above expression to an arbitrary leaf $\mathcal{F}^{m+1}_k$.  Using Lemma \ref{Lem:LS}, we obtain that 
		$$
		X_t(x)=\sum_{\alpha=1}^{r}g_\alpha(t,x) X_\alpha(x),
		$$
		on $\mathcal{F}_k$ and $g_\alpha(t,x)=g_\alpha(t,x')$ for points $x,x'$ in the same leaf of $\mathcal{F}$ and every $t\in \mathbb{R}$. 
	\end{proof}
		
Observe that Theorem \ref{FLST} gives a procedure to construct a foliated superposition rule.	For reference, let us detail the steps of this method below:
	\begin{itemize}
		\item Consider a Vessiot--Guldberg Lie algebra $V_{\mathcal{F}}$ of an  $\mathcal{F}$-foliated Lie system $X$ on an $d_T$-dimensional manifold $T$. 
		\item Find the smallest $m\in \mathbb{N}$ so that the diagonal prolongations of the vector fields of $V_{\mathcal{F}}$ span a distribution of rank $\dim V_{\mathcal{F}}$ at a generic point of $T^m$.  This states the number of particular solutions of the foliated superposition rule, namely $m$.
		\item Consider a coordinate system $\theta^1,\ldots,\theta^s,I^{s+1},\ldots, I^{d_T}$, adapted to the foliation $\mathcal{F}$ around a generic point $x\in T$, i.e. the first $s$ coordinates form a local coordinate system on each leaf $\mathcal{F}_k$ of $\mathcal{F}$, while the last ${d_T}-s$ coordinates are constant on the leaves of $\mathcal{F}$. Define the same coordinate system $\theta^1,\ldots,\theta^s,I^{s+1},\ldots, I^{d_T}$ on each copy $T$ within $T^{m+1}$. This gives rise to a coordinate system on $T^{m+1}$ of the form $\theta_{(a)}^1,\ldots,\theta_{(a)}^s,I_{(a)}^{s+1},\ldots, I_{(a)}^{d_T}$ for $a=0,\ldots,m$. 
		\item Define $\Psi^{s+1}=I_{(0)}^{s+1},\ldots,\Psi^{{d_T}}= I_{(0)}^{d_T}$  and obtain $s$ common first-integrals $\Psi^1,\ldots, \Psi^s$ for the diagonal prolongations of the elements of $V_{\mathcal{F}}$ to $T^{m+1}$ satisfying that
		$$
		\frac{\partial(\Psi^1,\ldots,\Psi^{d_T})}{\partial (\theta_{(0)}^1,\ldots,\theta_{(0)}^s,I_{(0)}^{s+1},\ldots, I_{(0)}^{d_T})}\neq 0\Leftrightarrow \frac{\partial(\Psi^1,\ldots,\Psi^s)}{\partial (\theta_{(0)}^1,\ldots,\theta_{(0)}^s)}\neq 0.
		$$
		\item Assume $\Psi^1=k_1,\ldots, \Psi^{d_T}=k_{d_T}$ and obtain $\theta^1_{(0)},\ldots,\theta^s_{(0)},I_{(0)}^{s+1},\ldots,I_{(0)}^{d_T}$ as a function of $k_1,\ldots,k_{d_T}$ and $\theta_{(a)}^1,\ldots,\theta_{(a)}^s,I_{(a)}^{s+1},\ldots, I_{(a)}^{d_T}$ for $a=1,\ldots,m$, i.e.
		$$
		\theta^i_{(0)}=F^i(\theta_{(1)},I_{(1)},\ldots,\theta_{(m)},I_{(m)},k_1,\ldots, k_{d_T}),\qquad
		I^j_{(0)}=k_j,
		$$
		for certain functions $F^i:T^m\times T\rightarrow \mathbb{R}$, with $i=1,\ldots,s$ and $j=s+1,\ldots,{d_T}$.  This gives rise to a superposition rule $\Phi:T^m\times T\rightarrow T$ for every Lie system with a Vessiot--Guldberg Lie algebra $V_{\mathcal{F}}$ of the form (see {\cite{CGM07,Dissertationes,LS20}} for details)
		$$
		\Phi(\theta_{(1)},I_{(1)},\ldots,\theta_{(m)},I_{(m)},k_1,\ldots,k_{d_T})=(F^1,\ldots,F^{d_T}),
		$$
		where $F^{k}=I^k_{(0)}$ for $k=s+1,\ldots,{d_T}$.
		The map $\Psi:T^m\times T\rightarrow R$ becomes an $\mathcal{F}$- foliated superposition rule for $X$ at a generic point of $T$.
	\end{itemize}

 \begin{example}
Consider the Lax pair (\ref{FLS3}) discussed in Example \ref{Ex:FLS_lax}. In this case, the manifold where the Lax pair is defined is $2n$-dimensional and $s=n$. We recall that the foliated Lie system given by the Lax pair  (\ref{FLS3}) was related to the Vessiot--Guldberg Lie algebra $V^{\mathfrak{g}^{lp}}$. The vector fields of a basis of $V^{\mathfrak{g}^{lp}}$ are linearly independent at a generic point. Hence, we can obtain a foliated superposition rule depending on one particular solution. 
	The coordinates $\theta^1=v^1,\ldots,\theta^n=v^n,I^1=v^{n+1},\ldots,I^n=v^{2n}$ are adapted to the foliation $\mathcal{F}^{\mathfrak{g}^{lp}}$ of the system under study.  Consider the coordinate system on $(\mathfrak{g}^{lp})^2$ of the form $\theta_{(a)}^1,\ldots,\theta_{(a)}^n,I_{(a)}^1,\ldots,I_{(a)}^n$ with $a=0,1$. Take the diagonal prolongations of the elements of $V^{\mathfrak{g}^{lp}}$ to $(\mathfrak{g}^{lp})^2$. To obtain $2n$ functionally independent constants of motion for such diagonal prolongations choose  $\Psi^{s+1}=I^{1}_{(0)},\ldots,\Psi^{2n}=I^n_{(0)}$ and  $\Psi^i=\theta^i_{(0)}-\theta^i_{(1)}$ for $i=1,\ldots,n$. Then,
	$$
	\frac{\partial (\Psi^1,\ldots,\Psi^n)}{\partial( \theta^1_{(0)},\ldots, \theta^n_{(0)},I^1_{(0)},\ldots I^{n}_{(0)})}\neq 0.
	$$
	By fixing $\Psi^1=k_1,\ldots,\Psi^{2n}=k_{2n}$, a superposition rule $\Phi: \mathfrak{g}^{lp}\times  \mathfrak{g}^{lp}\rightarrow \mathfrak{g}^{lp}$ for every Lie system with a Vessiot--Guldberg Lie algebra $V^{\mathfrak{g}^{lp}}$ reads
	$$
	\Phi(\theta_{(1)},I_{(1)},k)=(\theta^1_{(1)}+k_1,\ldots,\theta^n_{(1)}+k_n,k_{n+1},\ldots,k_{2n}).
	$$
	Restricting oneself to the case  $k_{n+1}=I^1_{(1)},\ldots,k_{2n}=I^n_{(1)}$, one gets an $\mathcal{F}^{\mathfrak{g}^{lp}}$-foliated superposition rule $\Psi^\mathfrak{g}:\mathfrak{g}^{lp}\times \mathfrak{g}^{lp}\rightarrow\mathfrak{g}^{lp}$ such that the particular solutions to $X^{lp}$ in the leaf $\mathcal{F}_k^{\mathfrak{g}^{lp}}$, with $k=I_{(1)}\in \mathbb{R}^n$, are of the form
	$$
	\Psi(\theta_{(1)}(t),I_{(1)},k_1,\ldots,k_{2n})=(\theta^1_{(1)}+k_1,\ldots,\theta^n_{(1)}+k_n,I^1_{(1)},\ldots,I^n_{(1)})
	$$
	for a particular solution $(\theta_{(1)}(t),I_{(1)})$ of $X^{lp}$ in $\mathcal{F}_k^{\mathfrak{g}^{lp}}$.
 \end{example}
	
	\section{Foliated automorphic Lie systems}
	We showed in Section \ref{Sec:LS_intro} that the study of a Lie system can be reduced to that of an automorphic Lie system. In this section, we prove that the evolution of a stratified Lie system can be reduced to the integration of a foliated Lie system on a principal bundle of a particular type.
	\begin{definition}
Let $\pi: G\times M\rightarrow M$ be a trivial principal bundle with structural $r$-dimensional Lie group $G$ acting on $G\times M$ by $\varphi_\pi:(g,(h,k))\in G\times (G\times M)\mapsto (gh,k)\in G\times M$. A {\it foliated automorphic Lie system} is a $t$-dependent vector field $X_\pi^R$ on $G\times M$ given by
		\begin{equation}
			\label{PAut}
			X_\pi^R =\sum_{\alpha=1}^{r}f_\alpha X^R_\alpha,
		\end{equation}
		where  $X^R_1,\ldots,X^R_r$ form a basis of fundamental vector fields of the action of $G$ on $G\times M$, and $f_1,\ldots,f_r\in C^\infty(\mathbb{R}\times G\times M)$ are $t$-dependent common constants of motion of $X^R_1,\ldots,X_r^R$.
	\end{definition}
	
Obviously, the system given by (\ref{PAut}) is a foliated Lie system. Note that every vector field on $G$ can be considered in a canonical way as a vector field on $G\times M$ via the vector bundle isomorphism $T(G\times M)=TG\times TM$. Then, $X^R_1,\ldots,X^R_r$ can be understood as right-invariant vector fields  on $G$. They also span a finite-dimensional Lie algebra of vector fields on $G\times M$ and a regular $r$-dimensional generalised distribution. In fact, at each point of $G\times P$ the values of $X^R_1,\ldots,X_r^R$  span the vertical space at such point of the  bundle $\pi:G\times M\rightarrow M$.
	Since $f_1,\ldots,f_r$ are common $t$-dependent constants of motion for vertical vector fields, they are constant on the fibres of $\pi:G\times M\rightarrow  M$ and they can be considered, in a unique manner, as $t$-dependent functions on $M$. 
	
	When $M$ is a point,  $f_1,\ldots,f_r$ become functions depending only on $t$ and (\ref{PAut}) turns into a standard automorphic Lie system \cite{Dissertationes}. More generally, every trivial principal bundle $\pi:P\rightarrow M$ with a structural group $G$ and fundamental vector fields spanned by $X_1,\ldots,X_r$ give rise to a $t$-dependent vector field $X$ on $P$ of the form
	$$
	X(t,x)=\sum_{\alpha=1}^rf_\alpha(t,x)X_\alpha(x),\qquad \forall x\in P,\quad \forall t\in \mathbb{R},
	$$
	where $f_1,\ldots,f_r\in C^\infty(\mathbb{R}\times P)$ are $t$-dependent constants of motion of $X_1,\ldots,X_r$. The trivialisation of the principal bundle $\pi:P\rightarrow M$ maps diffeomorphically $X$ onto  $X^R_\pi$. Consequently, $X$ is, up to a bundle diffeomorphism, a foliated automorphic Lie system.
	
	Let us prove that, in analogy with Lie systems, every $\mathcal{F}$-foliated Lie system gives rise to a foliated automorphic Lie system, whose solutions allow us to obtain the general solution of the foliated Lie system. Recall that every $t$-dependent constant of motion for the vertical  vector fields of  a fibre bundle can be considered as a $t$-dependent function on its base manifold in a canonical way.

	\begin{theorem}\label{FLSAFLS} Let $X={\displaystyle  \sum_{\alpha=1}^{r}}f_\alpha X_\alpha$ be a foliated Lie system on $N$ associated with a Vessiot--Guldberg Lie algebra $V=\langle X_1,\ldots,X_r\rangle$. Let $\phi:(g,x)\in G\times N\mapsto \phi_x(g)=\phi(g,x)\in N$ be the effective local Lie group action associated with the integration of the vector fields of the Lie algebra $V$. Assume that the space $M$ of leaves of $\mathcal{D}^V$ admits a manifold structure so that the  $\pi:N\rightarrow N/G=M$ is smooth. Let us  define the foliated automorphic Lie system on the total space of the principal bundle $\pi:G\times M\rightarrow M$  given by
		\begin{equation}\label{RAFLS}
			X_\pi^R(t,g,k)=-\sum_{\alpha=1}^rf_\alpha(t,k)X^R_\alpha(g),\qquad \forall t\in \mathbb{R},\quad \forall g\in G,\quad \forall k\in M,
		\end{equation}
		where  $-X^R_\alpha$ and $X_\alpha$ are the fundamental vector fields associated with the same element of $T_eG$ relative to $\varphi_\pi$ and $\phi$, correspondingly. Then, each particular solution $x(t)$ of $X$ contained in the leaf indexed by $k\in M$ can be written as
		\begin{equation}\label{Rel}
			x(t)=\phi(g(t),x(0)),
		\end{equation}
		where $\gamma:t\in \mathbb{R} \mapsto (g(t),k)\in G\times M$ is a particular solution to $X_\pi^R(t,g,k)$ with $g(0)=e$.
	\end{theorem}
	\begin{proof} Let us prove that $x(t)$ given by (\ref{Rel}) is a particular solution to $X$ for every $x(0)$. Using (\ref{Rel}), we have that
		\begin{equation}\label{Exp1}
			\frac{dx}{dt}(t_0)=\frac{d}{dt}\bigg|_{t=t_0}\phi(g(t)g^{-1}(t_0),\phi(g(t_0),x(0)))=T_e\phi_{\phi(g(t_0),x(0))}\left(\frac{d}{dt}\bigg|_{t=t_0}g(t)g^{-1}(t_0)\right).
		\end{equation}
		Note that $(g(t),x)$ is a particular solution to (\ref{RAFLS}) and $X^R(g)=R_{g*e}X^R(e)$ for every $g\in G$ and every right-invariant vector field $X^R$ on $G$. Moreover, each $X_\alpha^R$ can be considered as a right-invariant vector field on $G$ in the natural way. Using these facts, we have
		$$
		\frac{dg}{dt}(t_0)=-\sum_{\alpha=1}^rf_\alpha(t_0,k)R_{g(t_0)*}X_\alpha^R(e)\Rightarrow 	R_{g^{-1}(t_0)*g(t_0)}\frac{dg}{dt}(t_0)=-\sum_{\alpha=1}^rf_\alpha(t_0,k)X^R_\alpha(e).
		$$
		Since $-X_\alpha^R$ and $X_\alpha$ are fundamental vector fields related to the same element of $T_eG$ relative to the $\phi$ and $\varphi_\pi$ actions, one obtains that $-T_e\phi_{x}X_\alpha^R(e)=X_\alpha(x)$ for every $x\in N$ and $\alpha=1,\ldots,r$. Substituting this relation in (\ref{Exp1}) and since the functions $f_\alpha(t,x)$ are just $t$-dependent on the orbits of $\phi$, we obtain
		$$
		\frac{dx}{dt}(t_0)=-T_e\phi_{x(t_0)}\left[\sum_{\alpha=1}^rf_\alpha(t,k)X^R_\alpha(e)\right]=\sum_{\alpha=1}^rf_\alpha(t,x(t_0))X_\alpha(x(t_0)),
		$$
		where we have used that functions $f(t,k)$ can be understood as $t$-dependent functions on $N$ that are constant on the fibres of the projection $N\rightarrow N/G=M$.
Therefore, $\phi(g(t),x(0))$ is a solution of $X$ for every $x(0)\in N$ and the theorem follows.
	\end{proof}
\begin{example}
Consider the foliated Lie system (\ref{FLS1}) analysed in Example \ref{Ex:FLS_cot}. Its associated Vessiot--Guldberg Lie algebra $V^{HJ}$ of the form (\ref{VI}) is isomorphic to the Lie algebra $(\mathbb{R}^n,+)$. Denote by $\{\lambda^1,\ldots,\lambda^n\}$ the dual basis to the canonical basis $\{e_1,\ldots,e_n\}$ on $\mathbb{R}^n$.

	The Lie group action $\vartheta: \mathbb{R}^{n}\times T^*\mathbb{R}^n \to T^*\mathbb{R}^n$ obtained by integrating the vector fields of $V^{HJ}$ reads $\vartheta(\lambda,Q,P) = (Q-\lambda,P)$.
 
	Observe that the Lie group action has been chosen so that the fundamental vector fields of the elements of the basis $\{e_1,\ldots,e_n\}$ be $\{\partial_{Q^1},\ldots \partial_{Q^n}\}$, respectively. 
	The space of leaves of the distribution spanned by the elements of $V^{HJ}$ is diffeomorphic to the manifold $M=\mathbb{R}^n$. Indeed, the variables $P^1,\ldots,P^n$ on $T^*\mathbb{R}^n\simeq \mathbb{R}^{2n}$ can be considered as a global coordinate system on $M$, which parametrises the leaves of the foliation $\mathcal{F}^{HJ}$. 
	
	The foliated automorphic Lie system related to (\ref{FLS1})  is, in virtue of Theorem \ref{FLSAFLS}, defined on the $(\mathbb{R}^n,+)$-principal bundle  $\pi: \mathbb{R}^n\times M\to  M$, $\pi:(\lambda,P) \mapsto P$, and it reads
	$$
	X(t,P,\lambda)=-\sum_{\alpha=1}^n\frac{\partial H}{\partial P^\alpha}(t,P)\frac{\partial}{\partial \lambda^\alpha}.
	$$
 \end{example}

 \begin{example}
	Consider now the system given by (\ref{FLS2}) in Example \ref{Ex:FLS_lax} with 
	\begin{equation}\label{Casem}
		m(t,v)=\sum_{\alpha=1}^{n}\frac{\partial H}{\partial P^\alpha}e_\alpha.
	\end{equation}
	This particular form of $m(t,v)$ was studied in \cite{BV90} for a $t$-independent function $H$ and it was shown that such choice gives a Lax pair for the system (\ref{FLS1}) under a simple change of variables. 
	
	The Vessiot--Guldberg Lie algebra $V^{\mathfrak{g}^{gl}}$ for (\ref{FLS2}) admits a basis given by $\{2\partial_{v^1},\ldots,2\partial_{v^n}\}$. Such vector fields span an abelian Lie algebra isomorphic to $(\mathbb{R}^n,+)$. The distribution spanned by the vector fields of $V^{\rm ad}$ gives rise to a family of leaves of the form (\ref{Leafg}).
	Therefore, the variables $v^{n+1},\ldots,v^{2n}$ can be considered as a coordinate system on the space of leaves, $M^{\rm ad}$, which becomes a manifold diffeomorphic to $\mathbb{R}^n$. 
	
	The vector fields $X^{\rm ad}_1,\ldots,X^{\rm ad}_{2n}$ can be integrated to obtain a Lie group action
	$$
	\begin{array}{rccc}
		\varphi_{\mathfrak{g}^{gl}}:&\mathbb{R}^{n}\times \mathfrak{g}^{gl}&\longrightarrow& \mathfrak{g}^{gl},\\
		&(\lambda;v^1,\ldots,v^{2n})&\mapsto&(v^1-2\lambda_1,\ldots, v^n-2\lambda_n,v^{n+1},\ldots,v^{2n}),
	\end{array}
	$$
	such that  the fundamental vector field of each $e_i$ in the canonical basis $\{e_1,\ldots,e_n\}$ of $\mathbb{R}^n$  is $X_i^{\rm ad}$.  
	The foliated automorphic Lie system associated with this foliated Lie system reduces to the form of a $t$-dependent vector fields on $\mathbb{R}^{n}\times \mathbb{R}^n$ of the form
	$$
	X(t,\lambda,v)=-\sum_{\alpha=1}^n\frac{\partial H}{\partial v^{\alpha+n}}(t,v^{n+1},\ldots,v^{2n})\frac{\partial}{\partial v^\alpha}.
	$$
	Consequently, the solution to the Lax pair (\ref{FLS2}) for the particular value of $m(t,v)$ given in (\ref{Casem}) reduces to the same automorphic Lie system as the foliated Lie system (\ref{FLS1}). Moreover, there exists a diffeomorphism $\phi:T^*\mathbb{R}^n\rightarrow \mathfrak{g}^{gl}$ mapping (\ref{FLS1}) onto (\ref{FLS2}). It is easy to see that when two foliated Lie systems are diffeomorphic, they share the same foliated  automorphic Lie system. It is also immediate that foliated Lie systems related to the same foliated automorphic Lie system do not need to be diffeomorphic as they may be defined on manifolds of different dimension.
 \end{example}
	
	\section{Applications}
	Let us use our results of previous sections to study several physical problems. First, we focus on an extension of the generalised Ermakov system \cite{Le91}. Then, we analyse the existence of related Hamiltonian structures, which extends, in a geometric way, some of the results given in \cite{BV90}.

	\subsection{A new class of generalised Ermakov systems}
	
	There exists an extensive literature on the so-called Ermakov systems and their generalisations (see \cite{CLR08c,Go90,Le91,Ra80,RR79,RR80} and references therein). We propose here a class of generalised Ermakov systems that cannot be described through Lie systems but they admit a description in terms of foliated ones. 
	
	The so-called {\it generalised Ermakov systems} \cite{RR79,RR80} is a class of systems of differential equations on $\mathbb{R}_0^2=\{(x,y):xy\neq 0\}$ of the form
	\begin{equation}\label{GEE}
		\frac{d^2 x}{dt^2}=-\omega^2(t)x+\frac{g(y/x)}{yx^2},\qquad \frac{d^2 y}{dt^2}=-\omega^2(t)y+\frac{f(x/y)}{xy^2},\qquad t \in \mathbb{R},
	\end{equation}
	where $f,g:\mathbb{R}\rightarrow \mathbb{R}$ are real functions. The equation introduced by Milne \cite{M30} was an example of (\ref{GEE}) given by setting $g(u)=u$ in one of the equations and ignoring the second one \cite{P50}. Meanwhile, the so-called  Ermakov--Pinney  system  	corresponds to the case $g(u)=u$ and $f(u)=0$. The latter generalisation, along with each particular example of (\ref{GEE}), admit a constant of motion, the so-called {\it generalised Lewis invariant}, of the form
\begin{equation}\label{defI}
	I(x,y,\dot x,\dot y)=\frac 12\left(x\dot y-y\dot x\right)^2+\int^{x/y}[f(u)-u^{-2}g(1/u)]du.
	\end{equation}
	It was noted by Ray, Reid, and Goedert \cite{Go90,Ra80,RR80} that the term $\omega^2(t)$ can be replaced by much more general expressions (see  \cite{Le91} and references therein). For instance, there exist generalisations of (\ref{GEE}) where $\omega(t)$ depends on the time-derivatives of $x$ and $y$ \cite{RR80}.  
 
 As a new generalisation, we propose the second-order system of differential equations given by
	\begin{equation}\label{Leach}
		\frac{d^2 x}{dt^2}={-\omega^2(t,I)x}+\frac{g(I,y/x)}{x^2y },\quad \frac{d^2 y}{dt^2}={-\omega^2(t,I)y}+\frac{f(I,x/y)}{xy^2 },\quad (x,y)\in \mathbb{R}_0^2,\quad t\in\mathbb{R},
	\end{equation}
	where $f,g:\mathbb{R}\rightarrow \mathbb{R}$ are arbitrary non-vanishing functions and $I$ is given by (\ref{defI}). In the case where $f$ and $g$ are constants, one recovers the generalised Ermakov system studied in \cite{Le91}. 
	
	Consider the system of first-order differential equations on $T\mathbb{R}_0^2$ of the form
$$
	\dfrac{dx}{dt} = v_x,\qquad \dfrac{dy}{dt} = v_y,\qquad
	\dfrac{dv_x}{dt} = {-\omega^2(t,I)x}+\dfrac{g(I,y/x)}{x^2y },\qquad
  \dfrac{dv_y}{dt} = {-\omega^2(t,I)y}+\dfrac{f(I,x/y)}{y^2x },
	$$
	obtained by adding the variables $v_x=\dot x$ and $v_y=\dot y$ to (\ref{Leach}) and where $I$ is a function as in (\ref{defI}) but with $\dot x$ and $\dot y$ replaced by $v_x$ and $v_y$, respectively. This system is associated with the $t$-dependent vector field $X=\omega^2(t,I)X_3+X_1$ on $T\mathbb{R}_0^2$, where the vector fields 
$$
	\begin{array}{rcl}
	X_1&=&\dfrac{f(I,x/y)}{xy^2} \partial_{v_y}+v_y \partial_{y}+\dfrac{g(I,y/x)}{x^2y } \partial_{v_x}+v_x \partial_{x},\\
	X_2&=&\dfrac 12\left[y \partial_{y}-v_y \partial_{v_y}+x \partial_{x}-v_x \partial_{v_x}\right],\\
	X_3&=&-y \partial_{v_y}-x \partial_{v_x},\end{array}
	$$
	satisfy the commutation relations
	$$
	[X_1,X_2]=X_1,\qquad [X_1,X_3]=2X_2,\qquad [X_2,X_3]=X_3\rc{,}
	$$
	and therefore span a Lie algebra of vector fields $V^{gES}$ isomorphic to $\mathfrak{sl}(2,\mathbb{R})$. Since a straightforward calculation shows that $I$ is a common first-integral to $X_1,X_2,X_3$ and $X_1\wedge X_2\wedge X_3\neq 0$ on a  dense open subset  $\mathcal{O}\subset T\mathbb{R}_0^2$, then $X$ becomes a foliated Lie system on $\mathcal{O}\subset T\mathbb{R}^2_0$. 
	
	Let us show that (\ref{Leach}) admits a Lie algebra of Lie symmetries given by $\mathfrak{sl}(2,\mathbb{R}).$  Note that the restriction of $X$ to every leaf of the foliation determined by the integral leaves of the distribution  $\mathcal{D}^{V^{gES}}$ becomes a Lie system. Since $X_1\wedge X_2\wedge X_3\neq 0$ on these leaves, one obtains a so-called {\it locally automorphic} Lie system on each leaf (see \cite{GLMV19}). It was also proved in \cite{GLMV19} that such a Lie system admits a Lie algebra of Lie symmetries isomorphic to $\mathfrak{sl}(2,\mathbb{R})$. Gluing together these vector fields on each leaf, we obtain a Lie algebra of Lie symmetries of $X$ on $T\mathbb{R}_0^2$ isomorphic to $\mathfrak{sl}(2,\mathbb{R})$. This feature is common to many other generalisations of Ermakov systems \cite{Le91}.

	\subsection{Foliated Lie--Hamilton systems and $r$-matrices}
	Let us illustrate the use of Poisson structures to investigate foliated Lie systems through a couple of examples. This suggests how to generalise the theory of Lie-Hamilton systems in \cite{CLS13} to the realm of foliated Lie systems. As a byproduct, several results on the use of $r$-matrices to study foliated Lie systems will be provided, which generalises previous findings from \cite{BV90}.
	
Let $\{e_1,\ldots,e_{d}\}$ be a basis of a Lie algebra $\mathfrak{g}$  and denote the dual basis in $\mathfrak{g}^*$ by $\{v^1,\ldots,v^d\}$. Moreover, let $G$ be the Lie group associated with $\mathfrak{g}$. Consider the first-order system of differential equations on $\mathfrak{g}$ given by 
\begin{equation}\label{AdFLS}
\frac{dx}{dt}=\sum_{\alpha=1}^{d}f_\alpha (t,x)X^{\rm ad}_{e_\alpha}(x),\qquad \forall x\in \mathfrak{g},
\end{equation}
where the $X^{\rm ad}_{e_\alpha}$ are the fundamental vector fields of the adjoint action of $G$ on $\mathfrak{g}$ and $f_1,\ldots,f_d$ are common $t$-dependent constants of motion for all the vector fields  $X^{\rm ad}_{e_\alpha}$. 
	
Let $\mathcal{D}^{\rm ad} := \langle X^{\rm ad}_1,\ldots,X_d^{\rm ad} \rangle$ and take the point $x_0\in \mathfrak{g}$ where the rank of the distribution $\mathcal{D}^{\rm ad}$ reaches its maximum value  $a_{\rm max}$. Then, there exist $a_{\rm max}$ vector fields on $G$ taking values in $\mathcal{D}^{\rm ad}$ which are linearly independent at $x_0$.  Such vector fields will be also linearly independent at every point of a certain local neighbourhood $\mathcal{O} \subset \mathfrak{g}$ of $x_0$. In consequence, $ {\rm rk} \mathcal{D}^{\rm ad} = a_{\rm max}$ on $\mathcal{O}$ (cf. \cite{Va94}). Hence, we can restrict (\ref{AdFLS}) to an open submanifold $\mathcal{O}\subset \mathfrak{g}$, where $\mathcal{D}^{\rm ad}$ is a regular distribution of order $a_{max}$. Then, system (\ref{AdFLS}) becomes a foliated Lie system on $\mathcal{O}$. 
	
	Assume that $\mathfrak{g}$ admits an  ${\rm ad}$-invariant non-degenerate constant metric ${\bf g}= {{\displaystyle\sum_{\alpha,\beta=1}^d}}g_{\alpha\beta}v^\alpha\otimes v^\beta$, i.e. ${\bf g}([x,x'],x'')+{\bf g}(x',[x,x'']) =0$ for all $x,x',x''\in \mathfrak{g}$ (see \cite{Mi76} for details on ad-invariant metrics). This allows us to define a metric tensor $\mathfrak{G}= {{\displaystyle\sum_{\alpha,\mu=1}^{d}}}g_{\mu\nu}dv^\mu\otimes dv^\nu$ on $\mathfrak{g}$ and a vector bundle isomorphism $\mathfrak{G}^\flat:e_x\in T\mathfrak{g}\mapsto \mathfrak{G}_x(e_x,\cdot)\in T^*\mathfrak{g}$ between the tangent and cotangent bundles over $\mathfrak{g}$. Since $\mathfrak{G}$ is non-degenerate, $\mathfrak{G}^\flat$ has an inverse $\mathfrak{G}^\sharp:T^*\mathfrak{g}\rightarrow T\mathfrak{g}$. Let $E$ be the Euler vector field on $\mathfrak{g}$ generating dilations, namely $E= {{\displaystyle\sum_{\alpha=1}^d}}v^\alpha \partial_{v^\alpha}$. It is immediate that $E$ does not depend on the  linear system on $\mathfrak{g}$ used to define it, which turns $E$ into a geometric object. 
	
Previously introduced geometric objects allow us to define a Poisson bracket on $\mathfrak{g}$, the so-called {\it Kirillov bracket} (see \cite{BV90}), of the form
	\begin{equation}\label{KB}
		\{f,h\}_K=\mathfrak{G}([\mathfrak{G}^\sharp (df),\mathfrak{G}^\sharp (dh)]_{\mathfrak{X}},E),\qquad \forall f,h \in C^\infty(\mathfrak{g}),
	\end{equation}
	where $[\cdot,\cdot]_{\mathfrak{X}}$ is the extension of the Lie bracket on $\mathfrak{g}$ to every tangent space in $T_x\mathfrak{g}$ for $x\in \mathfrak{g}$ by translations. More specifically, if $c_{\alpha\beta}\,^\gamma$, with $\alpha,\beta,\gamma=1,\ldots,d$, are the structure constants of the Lie algebra $\mathfrak{g}$ in the basis $\{e_1,\ldots,e_d\}$, i.e. $[e_\alpha,e_\beta]= {{\displaystyle\sum_{\gamma=1}^d}}c_{\alpha\beta}\,^\gamma e_\gamma$  for $\alpha,\beta=1,\ldots,d$, then $[\partial_{v^\alpha}, \partial_{v^\beta}]_\mathfrak{X} := {{\displaystyle\sum_{\gamma=1}^d}} c_{\alpha\beta}\,^\gamma \partial_{v^\gamma}$ for $\alpha,\beta=1,\ldots,d$. 
 
 Let us prove that (\ref{KB}) recovers the expression of the Kirillov bracket given in \cite{BV90}. 
	It is immediate that expression (\ref{KB}) is antisymmetric and satisfies the Leibniz property. Let us prove that (\ref{KB}) fulfils the Jacobi identity. Since ${\bf g}$ is non-degenerate and $\{e_1,\ldots,e_d\}$ is a basis of $\mathfrak{g}$, the linear functions  $f_\alpha: x\in\mathfrak{g}\mapsto {\bf g}(e_\alpha,x)\in \mathbb{R}$, with $\alpha=1,\ldots,d$, form a coordinate system on $\mathfrak{g}$. In the coordinates $\{f_1,\ldots,f_d\}$, expression (\ref{KB}) becomes
	$
	\{f_\alpha,f_\beta\}_K(x)={\bf g}([e_\alpha,e_\beta],x)
	$  for $\alpha,\beta=1,\ldots,d$ and $x\in \mathfrak{g}$. Then, 
	$
	\{f_\alpha,\{f_\beta,f_\gamma\}\}(x)=\mathfrak{G} ([e_\alpha,[e_\beta,e_\gamma]],x)
	$ for $\alpha,\beta,\gamma=1,\ldots,d$ and $x\in \mathfrak{g}$. It follows that (\ref{KB}) satisfies the Jacobi identity for  any three functions chosen among $f_1,\ldots,f_d$. Due to this and the fact that (\ref{KB}) satisfies the Leibniz property, one gets that (\ref{KB}) satisfies the Jacobi identity for all functions on $\mathfrak{g}$. Hence,  (\ref{KB}) becomes a Poisson bracket, and its Poisson bivector  reads
	\begin{equation}\label{Kirillov}
		\Lambda_K=\frac 12\sum_{\alpha,\beta=1}^d\Lambda(df_\alpha,df_\beta)\frac{\partial}{\partial f_\alpha}\wedge \frac{\partial}{\partial f_\beta}=\frac 12\sum_{\alpha,\beta,\gamma=1}^dc_{\alpha\beta}\,^\gamma f_\gamma\frac{\partial}{\partial f_\alpha}\wedge \frac{\partial}{\partial f_\beta}.
	\end{equation}

	Recall that the vectors $\{e_1,\ldots, e_d\}$ can be considered as a coordinate system on the dual space $\mathfrak{g}^*$. The Kostan-Kirillov-Souriau (KKS) bracket on $\mathfrak{g}^*$ reads  (see \cite{Va94} for details)
 $$
	\Lambda=\sum_{\alpha,\beta, \gamma=1}^dc_{\alpha\beta}\,^\gamma e_\gamma \frac{\partial}{\partial e_\alpha}\wedge \frac{\partial }{\partial e_\beta}.
	$$
	The diffeomorphism $\phi:x\in \mathfrak{g}\rightarrow {\bf g}(x,\cdot)\in \mathfrak{g}^*$ yields that $\Lambda_K=\phi_*\Lambda$. Hence, (\ref{KB}) is induced by the KKS bracket on $\mathfrak{g}^*$.
	
	Let us use the fact that ${\bf g}$ is ${\rm ad}$-invariant to prove that the vector fields $X^{\rm ad}_1,\ldots,X^{\rm ad}_d$ on $\mathfrak{g}$ are Hamiltonian relative to $\Lambda_K$. The ${\rm ad}$-invariance of ${\bf g}$ gives that ${\displaystyle \sum_{\delta=1}^d}c_{\alpha\beta}\,^\delta g_{\delta \gamma}=-{\displaystyle \sum_{\delta=1}^d}c_{\gamma\beta}\,^\delta g_{\delta\alpha}$ for $\alpha,\beta,\gamma=1,\ldots,d$. If $g^{\alpha\beta}$ are the entries of the inverse matrix of the metric ${\bf g}$, one gets that 
	$$
	\sum_{\delta,\gamma,\alpha=1}^dg^{\theta\alpha}c_{\alpha\beta}\,^\delta g_{\delta\gamma}g^{\gamma\pi}=-\sum_{\gamma,\alpha,\delta=1}^dc_{\gamma\beta}\,^\delta g_{\delta\alpha}g^{\gamma\pi}g^{\theta\alpha}\Longrightarrow \sum_{\alpha=1}^dg^{\theta\alpha}c_{\alpha\beta}\,^\pi=-\sum_{\gamma=1}^dc_{\gamma\beta}\,^\theta g^{\gamma\pi}.
	$$
	Renaming indexes,
${\displaystyle \sum_{\beta=1}^d}g^{\gamma \beta}c_{\alpha\beta}\,^\delta=-{\displaystyle \sum_{\beta=1}^d}g^{\delta\beta}c_{\alpha\beta}\,^\gamma$ for every $\gamma,\alpha,\delta=1,\ldots,d$. Since $v^\mu={\displaystyle \sum_{\gamma=1}^d}g^{\mu\gamma}f_\gamma$, using (\ref{Kirillov}), and in view of previous results, we have that 
	$$
	\Lambda_K=-\frac 12\sum_{\alpha,\beta,\gamma,\mu,\nu,\sigma=1}^dc_{\alpha\beta}\,^\gamma g_{\gamma\nu} v^\nu g^{\mu\alpha}g^{\beta \sigma}\frac{\partial}{\partial v^\mu}\wedge\frac{\partial}{\partial v^\sigma}=-\frac 12\sum_{\alpha,\beta,\sigma,\mu,\nu=1}^dg^{\beta\sigma}v^\nu c_{\nu\beta}\,^\mu\frac{\partial}{\partial v^\mu}\wedge \frac{\partial}{\partial v^\sigma}.
	$$
	A short calculation shows that $\Lambda_K(dv^{\bar{\sigma}},\cdot) ={\displaystyle\sum_{\beta=1}^d}g^{\beta\bar {\sigma}}X_\beta^{\rm ad}$ for $\bar\sigma=1,\ldots,r$, and 
	then 
	$$X^{\rm ad}_\beta={ \Lambda_K\left(\displaystyle\sum_{\bar\sigma=1}^dg_{\beta\bar{\sigma}}v^{\bar{\sigma}},\cdot \right)}$$ 
	for every $\beta=1,\ldots,d$. 
	This proves that the Vessiot--Guldberg Lie algebra of (\ref{AdFLS}) consists of Hamiltonian vector fields on {$\mathfrak{g}$}  relative to the  Kirillov bracket on $\mathfrak{g}$. The same applies to the restriction of the Vessiot--Guldberg Lie algebra of (\ref{AdFLS}) and $\Lambda$ to $\mathcal{O}$

	Whether the Lax pair (\ref{AdFLS}) is  a Hamiltonian system or not relative to the Kirillov bracket on $\mathcal{O}$ is not much  relevant to us. It was proved in Section \ref{HOFSR} that our method to derive foliated superposition rules for (\ref{AdFLS}) requires to determine some first-integrals for $X^{\rm ad}_1,\ldots, X^{\rm ad}_r$ and their diagonal prolongations. 
	This can be achieved by using that these vector fields are Hamiltonian (see \cite{CLS13}). Moreover, (\ref{AdFLS}) can be restricted to the intersection of $\mathcal{O}$ with any orbit $\mathcal{S}$ of $\mathcal{D}^{\rm ad}$. This restriction can also be considered as the restriction to $\mathcal{S}\subset \mathcal{O}$ of the Lie--Hamilton system $X={\displaystyle\sum_{\alpha=1}^d}f_\alpha(t,x)X^{\rm ad}_\alpha$, where $x\in \mathcal{S}$.  Again, whether (\ref{AdFLS}) is a Hamiltonian system or not, {\it per se}, is not relevant.
	
	Let us consider now the foliated automorphic Lie system related to the foliated Lie system (\ref{SF}) relative to its Vessiot--Guldberg Lie algebra $V^{\rm ad}$. The analysis of this system will again support our idea about how one should endow foliated Lie systems with a compatible Poisson structure to study their properties. 
	
	Consider a general Lie algebra $\mathfrak{g}^{lp}$ and its simply and simply connected Lie group $G^{LP}$. We consider the elements of $\mathfrak{g}^{lp}$ as left-invariant vector fields on $G^{LP}$. Let $r$ be an antisymmetric triangular $r$-matrix of a certain Lie algebra $\mathfrak{g}^{lp}$, i.e. $[r,r]_{SN}=0$ (see \cite{CP95,KS04} for details), and define
	$$
	\Lambda_r=\sum_{\alpha,\beta=1}^{2n}r_{\alpha\beta}X^R_\alpha\wedge X^R_\beta,
	$$
	where the vector fields $X^R_1,\ldots,X^R_{2n}$ satisfy that their non-zero commutation relations read $[X^R_{\alpha+n},X^R_{\alpha}]=2X^R_i$ for $\alpha=1,\ldots,n$.
	Then, $\Lambda$ is a Poisson bracket on $G^{LP}$ due to the fact that $r$ is a triangular $r$-matrix.  In particular,  consider the $r$-matrix $r={\displaystyle\sum_{\alpha=1}^n}e_\alpha\wedge h_{\alpha+n}$ in $\mathfrak{g}^{lp}$. This gives rise to a Poisson structure $\Lambda_r^{LP}={\displaystyle\sum_{\alpha=1}^n}X^R_\alpha\wedge X^R_{\alpha+n}$ on $G^{LP}$.
	
	Let us prove that the right-invariant vector fields $X_1^R,\ldots,X_n^R$ are  Hamiltonian relative to $\Lambda^{LP}_r$. Take the basis of right-invariant differential one-forms  $\{\eta_1^R,\ldots,\eta_{2n}^R\}$ on $G^{GL}$ that are dual  $\{X^R_1,\ldots,X^R_{2n}\}$. Then,
	$$
	d\eta^R_{\alpha}(X^R_\beta,X^R_\gamma)=-\eta_{\alpha}([X^R_\beta,X^R_\gamma]),\qquad \alpha,\beta,\gamma=1,\ldots,2n.
	$$
	Since $[X^R_\beta,X^R_\gamma]$ is a linear combination of $X^R_1,\ldots,X^R_n$, one gets that $d\eta^R_{\alpha}=0$ for $\alpha=n+1,\ldots,2n$. Since $G^{LP}$ is simply-connected, $\eta_\alpha^R=df_\alpha$  for $\alpha=n+1,\ldots,2n$. Consequently, 
	$$
	X^R_\alpha=\Lambda^{LP}(df_\alpha,\cdot),\qquad \alpha=1,\ldots,n,
	$$
	are Hamiltonian vector fields relative to $\Lambda^{LP}_r$. 
	
	The $t$-dependent vector field  given by (\ref{SF}) is related via Theorem \ref{FLSAFLS}, when one considers that it admits a Vessiot--Guldberg Lie algebra $V^{GL}$, to the foliated automorphic Lie system 
	\begin{equation}\label{AutSpe}
		\frac{dg}{dt}=\sum_{\alpha=1}^nf_\alpha(t,v_{n+1},\ldots,v_{2n})X^R_\alpha(g), \qquad g\in G^{LP},
	\end{equation}
	on the principal bundle $\pi^{PL}:G^{LP}\times \mathbb{R}^n\rightarrow \mathbb{R}^n$. This foliated automorphic Lie system
	admits a Vessiot--Guldberg Lie algebra of $\langle X^R_1,\ldots, X^R_{n}\rangle $ of Hamiltonian vector fields relative to $\Lambda^{LP}_r$. Hence, a superposition rule relative to this Lie algebra can be obtained using the methods in \cite{CLS13}. Once again, one obtains that it is interesting to consider foliated Lie systems whose Vessiot--Guldberg Lie algebras are Hamiltonian relative to some Poisson bivector.
	
	Let us provide another example of a foliated Lie system related to a Vessiot--Guldberg Lie algebra of Hamiltonian vector fields relative to a Poisson structure induced by a general $r$-matrix. Consider the foliated Lie system related to  (\ref{FLS2}). Such a foliated Lie system admits a Vessiot--Guldberg Lie algebra of the form $V=\langle X^{\rm ad}_1,\ldots, X^{\rm ad}_{2n}\rangle$, where $[X^{\rm ad}_{\alpha+n},X^{\rm ad}_\alpha]=2X^{\rm ad}_\alpha$  for $\alpha=1,\ldots,n$. The foliated automorphic Lie system related to this system reads as (\ref{AutSpe}). Recall that if $r={\displaystyle\sum_{\alpha,\beta=1}^{2n}}r^{\alpha\beta}e_\alpha\wedge e_\beta$ is an $r$-matrix for $\mathfrak{g}^{lp}$, one has the {\it Sklyanin Lie bracket} on $G^{LP}$ (see \cite{CP95}) given by the Poisson bivector
	$$
	\Lambda^{S}=\sum_{\alpha,\beta=n}^rr^{\alpha\beta}(X^L_\alpha\wedge X^L_\beta-X^R_\alpha\wedge X^R_\beta).
	$$
	Let us consider the $r$-matrix {$r={\displaystyle\sum_{\alpha=1}^n}e_\alpha\wedge h_\alpha$}. The Sklyanin Poisson bracket reads
	$$
	\Lambda^S=\sum_{\alpha=1}^n(X^L_\alpha\wedge X^L_{\alpha+n}-X^R_\alpha\wedge X^R_{\alpha+n})
	$$
	As before, we obtain $\Lambda^S\eta_{\alpha+n}^R=X_\alpha^R$ for $\alpha=1,\ldots,n$. Consequently,  (\ref{AutSpe})  admits a Hamiltonian Lie algebra of vector fields spanned by $X^R_1,\ldots,X^R_n$. A superposition rule can then obtained by studying these vector fields. Their constants of motions can be obtained through known methods for Hamiltonian--Lie systems.

	Previous examples justify the following definition.

	\begin{definition} A {\it foliated Lie--Hamilton system} is a foliated Lie system $X$ admitting a Vessiot--Guldberg Lie algebra of Hamiltonian vector fields relative to a Poisson structure.
	\end{definition}

As a final remark, let us briefly detail a new method to obtain Poisson structures on a Lie group $G$ from a KKS bracket on the dual to its Lie algebra $\mathfrak{g}$. This method does not require the use of an {\rm ad}-equivariant inner product on $\mathfrak{g}$, as needed to turn (\ref{AdFLS}) into a (foliated) Hamiltonian Lie system via our previous techniques. 

 Consider the two-dimensional Lie algebra $\mathfrak{b}_2=\langle e_1,e_2\rangle$ such that $[e_1,e_2]=e_2$. Let $\{e^1,e^2\}$ be the dual basis to $\{e_1,e_2\}$. Note that $\mathfrak{b}_2$ is a Frobenius Lie algebra \cite{CP95, E82}. Recall that a Frobenius Lie algebra over a field $k$ is a Lie algebra $(\mathfrak{g},[\cdot, \cdot]_{\mathfrak{g}})$ equipped with a non-degenerate $k$-bilinear skew-symmetric map $\beta: \mathfrak{g} \times \mathfrak{g} \to k$ of the form $\beta(v, w) := f([v,w]_{\mathfrak{g}})$, the so-called {\it Kirillov form}, for a certain linear map $f \in \mathfrak{g}^*$ such that $\beta([u,v],w) + \beta([w,u],v) + \beta([v,w],u) = 0$ for any $u,v,w \in \mathfrak{g}$. In case of $\mathfrak{b}_2$, the Kirillov form is given by the element $f := e^1 + e^2$. Note that the dual of $r := e_1 \wedge e_2$ gives rise to the map $\beta$. Such an element $r$ is called a {\it Jordan $r$-matrix} \cite{GS96}.  Moreover, a short calculation shows that any ad-invariant symmetric form on $\mathfrak{b}_2$ is degenerate. Consequently, one cannot study the system (\ref{AdFLS}) for $\mathfrak{b}_2$ as a foliated Lie--Hamilton system using the method discussed in this section.
 
Let $B_2$ be the connected and simply connected Lie group associated with the Lie algebra $\mathfrak{b}_2$. Then, $B_2$ admits a Poisson bivector $X^L_1\wedge X^L_2$, where $X_1^L,X_2^L$ are left-invariant vector fields on $B_2$ such that their values at a neutral element of $B_2$ read $X_1^L(e)=e_1$ and $X_2^L(e)=e_2$. In view of \cite[Proposition 7]{GS96} and the fact that $\mathfrak{b}_2$ is a Frobenius Lie algebra, there exists a local diffeomorphism $\phi_f:g\in B_2\mapsto {\rm Ad}^*_gf\in \mathfrak{b}^*_2$ for a certain $f\in \mathfrak{b}_2^*$. Indeed, the coadjoint orbits of $B_2$ in $\mathfrak{b}_2^*$ are the integral submanifolds of the KKS Poisson bracket on $\mathfrak{b}_2^*$, which have dimension 0 or 2. Since the KKS bracket on $\mathfrak{b}_2^*$ is different from zero, there exists a two-dimensional coadjoint orbit and thus, such $f$ must exist. Then,  $\phi_f(B_2)$ is an open subset of $\mathfrak{b}_2^*$. By \cite[Proposition 7]{GS96}, the map $\phi_f$ becomes a Poisson map relative to the KKS bracket on $\mathfrak{b}_2^*$ and the Poisson bracket (with the negative sign) on $B_2$ related to $r$.

\chapter{Deformation of Lie--Hamilton systems via Jacobi structures}

	\section{Deformation of Lie--Hamilton systems}\label{Sec:DfrLH}
	The section recalls a procedure devised in \cite{BCFHL18} to deform Lie--Hamilton systems via Poisson coalgebras within the class of Hamiltonian systems.
	First, we will briefly survey the theory on such systems. Next, we will detail a simple example.
	
	\subsection{General deformation procedure}
	Let us start by a considering the $t$-dependent system of first-order differential equations on $M$ determined by a $t$-dependent vector field on $M$ given by
	\begin{equation}\label{systemLH}
	\frac{ds}{dt}=X(t,s)=\sum_{\alpha=1}^rb_\alpha(t)X_\alpha(s),\qquad \forall t\in \mathbb{R},\quad \forall s\in M,
	\end{equation}
	for certain functions $b_1(t),\ldots,b_r(t)$ and vector fields $X_1,\ldots,X_r$  spanning a Lie algebra $V$, a so-called {\it Vessiot--Guldberg Lie algebra}, of Hamiltonian vector fields relative to a Poisson structure $\Lambda$ on $M$. A system of the form (\ref{systemLH}) is called a {\it Lie--Hamilton system} (see \cite{LS20} and references therein).
	Hamiltonian vector fields relative to a Poisson structure span its characteristic distribution, whose integral submanifolds are symplectic manifolds relative to the restriction of the initial Poisson bracket to the functions on each leaf \cite{Va94}. Therefore, (\ref{systemLH}) can be restricted to the leaves of such a distribution. On each leaf, the vector fields $X_1,\ldots,X_r$ become Hamiltonian and (\ref{systemLH}) becomes a Lie--Hamiltonian system relative to a symplectic form. Hence, the study of Lie--Hamilton systems reduces to the case of Lie--Hamilton systems relative to symplectic forms. We hereafter assume that (\ref{systemLH}) is a Lie--Hamilton system relative to a symplectic form $\omega$ on $M$.
	
	Let $\{e_1,\ldots,e_r\}$ be a basis of $\mathfrak{g}$ satisfying the same commutation relations as $\{h_1,\ldots,h_r\}$. By simplicity, we assume $V\simeq\mathfrak{g}$. 
	By assumption, every $X_{\alpha}$ admits a Hamiltonian function $h_{\alpha}$, with $\alpha=1,\ldots,r$. Although $\mathfrak{M} := \langle h_1, \ldots, h_r\rangle$ need not be a Lie algebra relative to the Poisson bracket, it can always be extended to a finite-dimensional one by adding the successive Poisson brackets of the elements of $\mathfrak{M}$ (cf. \cite{CLS13}). Let us assume that $\mathfrak{M}$ has already been extended by new appropriate functions to close a finite-dimensional Lie algebra. Let $\mathfrak{g}$ be the abstract Lie algebra isomorphic to  $\mathfrak{M}$. 
	
	If $\{e^1,\ldots,e^r\}$ is the dual basis to $
	\{e_1,\ldots,e_r\}$,  
	we can attach the Lie--Hamilton system (\ref{systemLH}) to a mapping
	$$
	\mu:s\in M\mapsto \sum_{i=1}^rh_i(s)e^i\in \mathfrak{g}^*.
	$$
	If $\mathfrak{g}^*$ is endowed with the Kirillov-Kostant-Souriau (KKS) Poisson bracket, then $\mu$ becomes a Poisson morphism.

	Our goal is to get a $z$-parametrised family of Hamiltonian systems determined by the $z$-parametrised family of $t$-dependent vector fields, $\{X_z\}_{z\in\mathbb{R}}$, so that the system determined by $X_z$ for $z=0$, i.e. $X_0$, becomes the starting Lie--Hamilton system. We do not require every $X_z$ to be a  Lie--Hamilton system as in (\ref{systemLH}). Instead, we will obtain a family of new Hamiltonian systems determined by $\{X_z\}_{z\in \mathbb{R}}$ and whose properties will be related to the ones of $X_0$ in a manner to be explained next.  
	
	In that respect, we define a set of $z$-parametrised functions $h_{1,z},\ldots,h_{r,z}$ on $M$ satisfying that
	\begin{equation}\label{def_Ham}
	\{h_{i,z},h_{j,z}\}=F_{ij,z}(h_{1,z},\ldots,h_{r,z}), \qquad h_{i,0} = h_i,\qquad i=1,\ldots,r,\qquad z\in \mathbb{R},
	\end{equation}
 for certain functions $F_{ij,z}:(x_1,\ldots,x_r)\in \mathbb{R}^r\mapsto F_{ij,z}(x_1,\ldots,x_r)\in \mathbb{R}$ with $i,j=1,\ldots,r$ and  $z\in \mathbb{R}$. 
The deformed system, for each value of the parameter $z$, is the system  with $t$-dependent Hamiltonian function $h_{z}(t,s)=\sum_{i=1}^{r}b_\alpha(t)h_{i,z}(s)$ for every $(t,s)\in \mathbb{R}\times M$. Condition (\ref{def_Ham}) implies that the Hamiltonian vector fields $X_{i,z}$ related to the functions $h_{i,z}$ satisfy that
	$$
	[X_{i,z},X_{j,z}]=\sum_{k=1}^r\frac{\partial F_{ij,z}}{\partial x_k}X_{k,z},\qquad i,j=1,\ldots,r,\qquad z\in \mathbb{R}.
	$$
Hence, for each fixed $z\in \mathbb{R}$, the family of vector fields $X_{1,z},\ldots,X_{r,z}$ on $M$ span an involutive distribution. If the rank of the distribution is constant, then the distribution is integrable \cite{OR04}.
 	
	The functions $h_{1,z},\ldots,h_{r,z}$ can be chosen in many manners. We derive them using a Lie bialgebra structure on a Lie algebra $\mathfrak{g}$, since this allows us to define, for every $z\in \mathbb{R}$, a Poisson algebra $C^{\infty}_z(\mathfrak{g}^*)$  equipped with a $z$-parametrised Poisson bivector $\Lambda_z$. To explain the role of $C^{\infty}_z(\mathfrak{g}^*)$ in obtaining $h_{1,z},\ldots,h_{r,z}$, let us remind first that if $\{v_1,\ldots, v_r\}$ is a basis of $\mathfrak{g}$ closing the same commutation relations as $h_1,\ldots, h_r$, then
	$$
	D: f(v_1, \ldots, v_r) \in C^{\infty}(\mathfrak{g}^*) \mapsto f(h_1, \ldots, h_r) \in C^{\infty}(M)
	$$
	is a Poisson algebra morphism if we endow $C^\infty(\mathfrak{g}^*)$ with the KKS Poisson bracket and $C^\infty(M)$ with the Poisson bracket related to $\omega$. Note that $D$ is the pullback on functions on $\mathfrak{g}^*$ induced by $\mu$, i.e. $D=\mu^*$.
	
	Let us construct a similar procedure for every $z\in \mathbb{R}$. Recall that every Poisson manifold is the union of disjoint submanifolds (possibly of different dimension) given by the integral submanifolds of its characteristic distribution, the so-called {\it strata}. Moreover, the initial Poisson manifold reduces to a symplectic form on each such a stratum \cite{Va94,We83}. 
 Assume that $\theta\in \mathfrak{g}^*$ is such that the stratification around it can be considered as a foliation, i.e. all symplectic submanifolds intersecting an open neighbourhood of $\theta$ have the same dimension. Almost every point of $\mathfrak{g}^*$ admits this condition (cf. \cite{Va84}). We also assume that the Poisson bivectors $\{\Lambda_z\}_{z\in \mathbb{R}}$ will additionally induce a foliation of the same codimension on an open neighbourhood of $\theta$ for values of $z$ within an interval $I\ni 0$ in $\mathbb{R}$. Then, one can introduce a $z$-parametrised family of local coordinates $\{\xi_{1,z},\ldots,\xi_{2n,z},c_{1,z},\ldots,c_{s,z}\}$ on an open neighbourhood $U$ of $\theta\in \mathfrak{g}^*$ for $z\in I$  so that $\Lambda_z = \sum_{i=1}^n\partial_{\xi_{2i-1,z}} \wedge \partial_{\xi_{2i,z}}$ on $U$ for $z\in I$.  This means that there exists a chart $\Psi_z:U\rightarrow \mathbb{R}^{2n+s}$, with $\dim \mathfrak{g}=2n+s=r$, such that $\Psi_{z*}\Lambda_z=\Lambda_c$, where $\Lambda_c$ is the Poisson bivector on $\mathbb{R}^{2n+s}$ of the form $$
	\Lambda_c=\sum_{i=1}^n\frac{\partial}{\partial x^{2i-1}}\wedge\frac{\partial}{\partial x^{2i}}
	$$
	and $\Psi_{0*}\Lambda_0=\Lambda_c$. Hence, $\Psi^{-1}_{z*}\Psi_{0*}\Lambda_0=\Lambda_z$ and $\phi_z=\Psi^{-1}_{z}\circ \Psi_{0}$ gives a Poisson morphism mapping $\Lambda_0$ onto $\Lambda_z$. We write $F_z=\phi_z^*$ for the pull-back on functions on $\mathfrak{g}^*$ via $\phi_z$.
 
    Note that, in particular, if $\{v_i,v_j\}_z=F_{ij,z}(v_1,\ldots,v_r)$, where $\{v_1,\ldots,v_r\}$ is a basis of $\mathfrak{g}^*$, then
$$
\{\phi^*_zv_i,\phi^*_zv_j\}=F_{ij,z}(\phi^*_zv_1,\ldots,\phi^*_zv_r).
 $$   
 In other words, the functions $\phi^*v_1,\ldots,\phi^*v_r$ close on $\mathfrak{g}^*$ the same commutation relations relative to $\{\cdot,\cdot\}_0$ as $v_1,\ldots,v_r$ relative to $\{\cdot,\cdot\}_z$. Moreover,
define $D_z:C^\infty( U)\rightarrow C^\infty(M)$ in the form $D_z=D\circ F_z$ and set 
	$$h_{i,z}:=D\circ \phi^*_z (v_i)=v_i\circ \phi_z\circ \mu,\qquad  i=1,\ldots,r,\qquad z\in I.
	$$ The functions $h_{1,z},\ldots,h_{r,z}$ will close the same Poisson bracket relations  relative to the Poisson bracket on functions on $M$ as $v_1,\ldots,v_r$ on $\mathfrak{g}^*$ relative to $\{\cdot,\cdot\}_z$.

	\begin{figure}
	    \centering
	\includegraphics[scale=0.20]{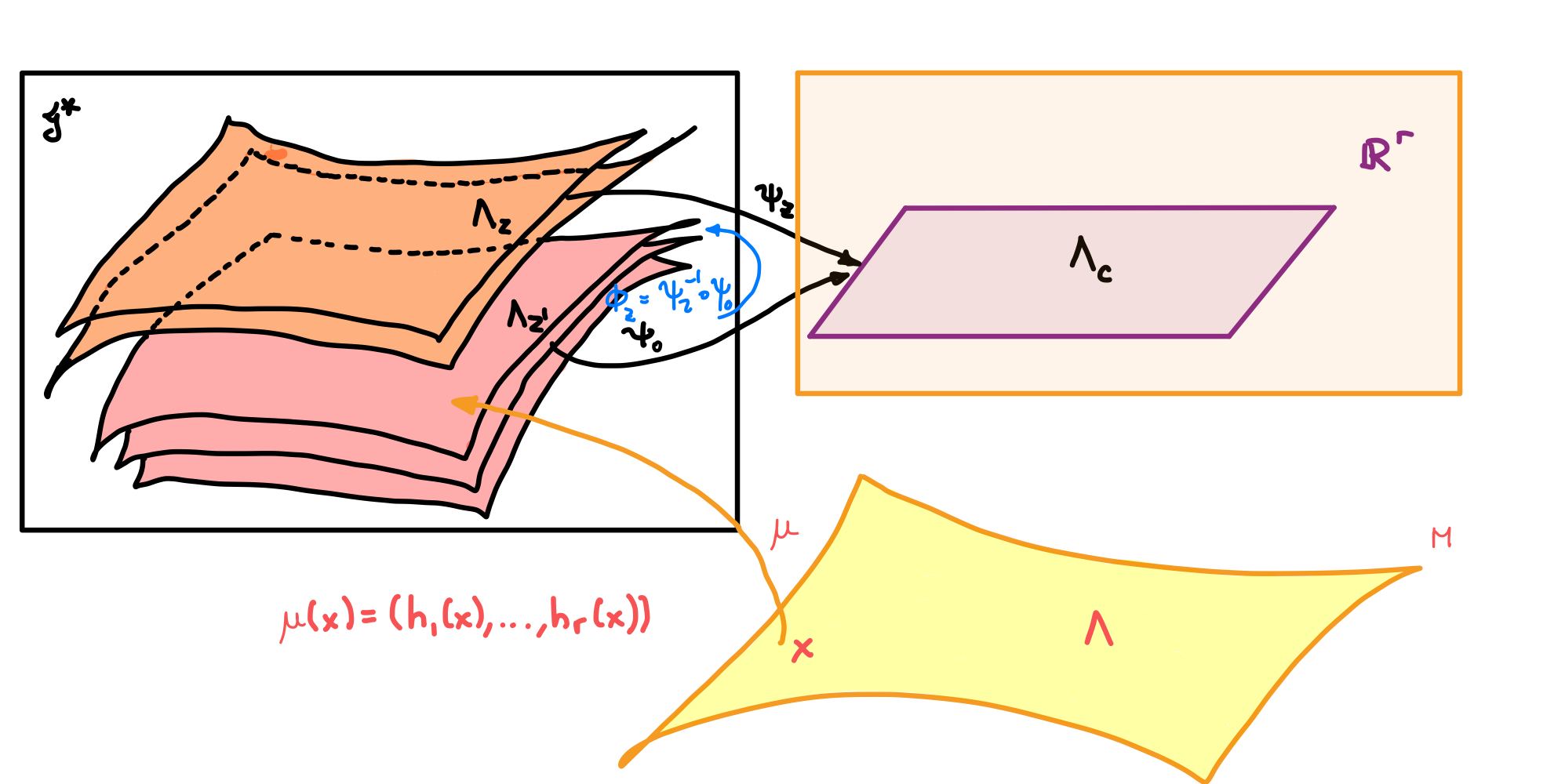}
	    \caption{Scheme of the rectification for the deformation procedure.}
	    \label{fig:my_label}
	\end{figure}
	
	As shown in next sections,  the above procedure uses the ideas implicitly given in  \cite{BCFHL18,BCFHL18b} to generalise previous techniques to a general manifold $M$ and a general Lie algebra $\mathfrak{g}$. Moreover, our approach in here is described in a slightly more detailed geometric picture without relying on Poisson brackets, which are a more of an algebraic entity. It is only worth stressing that $F_z:f\in C^\infty(\phi_z(U))\mapsto \phi_z^*f\in  C^\infty(U)$ gives rise to a Poisson algebra morphism for every $z\in \mathbb{R}$ between the Poisson algebra induced by $\Lambda_z$ on $C^\infty(\phi_z(U))$ and the one given by $\Lambda_0$ on $C^\infty(U)$.  
	
 \subsection{On the construction of $C^{\infty}_z(\mathfrak{g}^*)$ and its properties}\label{Sec:Cons}
	An important step in the presented procedure is to obtain a Poisson algebra $C^{\infty}_z(\mathfrak{g}^*)$ for every $z\in I$. Let us describe in detail how to get $C^{\infty}_z(\mathfrak{g}^*)$ using the Lie bialgebra structure on $\mathfrak{g}$.
	
	Recall that a {\it Poisson--Lie group} is a Lie group $G$ equipped with a Poisson structure whose group multiplication is a Poisson map. In other words, a Poisson--Lie group is determined by a Poisson bivector $\Lambda^G$ on $G$ such that \cite{CP95}
	\begin{equation}\label{gamma_cond}
	\Lambda^G_{gh} = (T_h L_g)(\Lambda^G_h) + (T_g R_h)(\Lambda^G_g), \quad \forall g,h \in G.
	\end{equation}
	In terms of the map $D: g\in G \mapsto  R_{g^{-1}*g}\Lambda^G_g\in\Lambda^2 \mathfrak{g}$, condition (\ref{gamma_cond}) reads
	$$
	D(gh) = {\rm Ad}_g (D(h)) + D(g), \quad \forall g,h \in G.
	$$
	Moreover, it can be proven \cite{Ma95} that
	\begin{equation}\label{PL_integration}
	\frac{d}{d\tau}\bigg|_{\tau=0}D(g\exp(\tau \xi))=(\mathcal{L}_{X^L_\xi}D)(g) = \textrm{Ad}_{g} (\delta(\xi)), \quad \forall g \in G,\qquad \forall \xi\in \mathfrak{g},
	\end{equation}
	where $\delta$ is a cocommutator on $\mathfrak{g}$ and $X^L_{\xi}$ is the left-invariant vector field satisfying $X^L_\xi(e)=\xi \in \mathfrak{g}$. 
	
	We will use (\ref{PL_integration}) to obtain a $z$-parametrised family of Poisson bivectors on $G$. Notice first that every Lie bialgebra $(\mathfrak{g}, [\cdot,\cdot]_{\mathfrak{g}}, \delta_{\mathfrak{g}})$ gives rise to a new one $(\mathfrak{g}, [\cdot,\cdot]_{\mathfrak{g}}, \delta_z := z \delta_{\mathfrak{g}})$ for every $z \in \mathbb{R}$. Then, by virtue of \cite[Theorem 13.2]{CP95}, one can associate a Poisson--Lie group $G^*_z$ with each dual Lie bialgebra $(\mathfrak{g}^*, [\cdot,\cdot]_{\mathfrak{g}^*} := \delta_z^*, \delta_{z,\mathfrak{g}^*} := [\cdot,\cdot]_{\mathfrak{g}}^*)$ for $z\in \mathbb{R}$. Solving (\ref{PL_integration}) for the chosen Lie bialgebra structure on the dual $\mathfrak{g}^*$, we get a $z$-parametrised family of Poisson bivectors on $G^*_z$. All $\delta^*_{z,\mathfrak{g}}$, for $z\neq 0$, determine isomorphic Lie algebras on $\mathfrak{g}^*$. Hence, the $G^*_z$ are isomorphic to $G^*$ for $z\neq 0$. We write $G^*$ for the corresponding Lie group. 
	
	The exponential map $\exp_{G^*}:\mathfrak{g}^*\rightarrow G^*$ is neither injective nor surjective in general \cite{Di57,Sa57}. Nevertheless, $\exp_{G^*}$ is a diffeomorphism for certain types of Lie algebras  \cite{Di57,Sa57}. In general, $\exp_{G^*}$ is always a local diffeomorphism, at least, from an open neighbourhood $\mathcal{O}\subset \mathfrak{g}^*$ containing $0$ and its image. We can then define a Poisson bivector on $\mathcal{O}$ via the condition
	$$
	\exp_{*p}\Lambda^{\mathfrak{g}^*}_p=\Lambda^{G^*}_{\exp(p)},\qquad \forall p\in \mathcal{O}.
	$$
	
	\begin{example}
		Let us derive a deformed Poisson structure on $\mathfrak{sl}_2=\langle J_3,J_+,J_-\rangle$ with the Lie bracket determined by the relations
		$$
	[J_3,J_+]=2J_+,\qquad [J_+,J_-]=J_3,\qquad [J_3,J_-]=-2J_-,\qquad 	
		$$
		via the coproduct satisfying the relations
		$$
		\delta_{\mathfrak{sl}_2}(J_3) = 2 J_3 \wedge J_{+}, \quad \delta_{\mathfrak{sl}_2}(J_{+}) = 0, \quad \delta_{\mathfrak{sl}_2}(J_{-}) = 2 J_{-} \wedge J_{+}.
		$$ 
		Its dual $\delta_{\mathfrak{sl}_2}^*:{\mathfrak{sl}_2}^*\wedge{\mathfrak{sl}_2}^*\rightarrow {\mathfrak{sl}_2}^*$  induces a Lie bracket on  $\mathfrak{sl}_2^* = \langle J_3^*, J_{+}^*, J_{-}^*\rangle$, where $\{J_3^*,J_+^*,J_-^*\}$ is the dual basis to 
		$\{J_3,J_+,J_-\}$ with the commutation relations
		$$
		[J_3^*, J_{+}^*]_{{\mathfrak{sl}^*_2}} = 2J_3^*, \quad [J_{-}^*, J_{+}^*]_{{\mathfrak{sl}_2}^*} = 2J_{-}^*, \quad [J_3^*, J_{-}^*]_{{\mathfrak{sl}^*_2}} = 0.
		$$
		The dual Lie algebra, namely $\mathfrak{sl}_2^*$, becomes then isomorphic to $
		\mathbb{R}
		\ltimes 
		\mathbb{R}^2$, where $\mathbb{R}^2\simeq \langle J_3^*,J^*_-\rangle$ and $\mathbb{R}\simeq \langle J^*_+\rangle$, i.e. $\mathfrak{sl}_2^*$ is isomorphic to the so-called {\it book Lie algebra} $\mathfrak{b}$ \cite{BBM12}. The book Lie algebra $\mathfrak{b}$ is solvable. As every Lie algebra, it admits a connected and simply-connected Lie group, let us denote it by  $B$, whose Lie algebra is isomorphic to $\mathfrak{b}\simeq \mathfrak{sl}_2^*$.  The book Lie algebra admits then a coproduct, $\delta_{{\mathfrak{b}}}$, of the form
		$$
		\delta_{\mathfrak{b}}(J_3^*) = J_{+}^* \wedge J_{-}^*, \quad \delta_{\mathfrak{b}}(J^*_{+}) = 2 J_3^* \wedge J_{+}^*, \quad \delta_{\mathfrak{b}}(J^*_{-}) = -2 J_3^* \wedge J_{-}^*,
		$$
		whose transpose retrieves the original Lie algebra structure on $\mathfrak{sl}_2$. 
		
		Since $B$ is a simply-connected, solvable, real Lie group and $\mathfrak{b}$ admits a flag of ideals, i.e. a series of ideals $\mathfrak{b}_1,\ldots,\mathfrak{b}_4$ such that $\mathfrak{b}_1=\mathfrak{b},\mathfrak{b}_4=0$ and $\dim \mathfrak{b}_i/\mathfrak{b}_{i+1}=1$ for $i=1,2,3$, then $\exp_{\mathfrak{b}}$ is a bijection (see \cite[p. 121]{Di57} or \cite{Sa57}). Then, the Poisson bracket on $B$ induced by $\delta_{\mathfrak{sl}_2}^*$ allows us to define a globally defined Poisson bracket on $\mathfrak{b}$. Let us calculate it.
		
		In canonical coordinates of the second kind \cite{Varadarajan}, every $g$ in an open neighbourhood $\mathcal{O}$ of the neutral element of $B$ can be expressed as
		\begin{equation}\label{sec_kind}
		g = \exp(aJ_3^*) \exp(bJ_{+}^*) \exp(cJ_{-}^*),
		\end{equation}
		where $a,b,c \in \mathbb{R}$ can be understood as local coordinates, at least, on an open subset of $B$ containing its neutral element. To determine the deformed brackets, we consider the Lie bialgebras $(\mathfrak{b}, \delta_{\mathfrak{sl}_2}^*, [\cdot,\cdot]_{\mathfrak{b}})$. By using the adjoint representation of $\mathfrak{b}$ and the decomposition (\ref{sec_kind}) of $g \in \mathcal{O}$, a basis of left-invariant vector fields on $B$ take the form
		$$
		X^L_1=e^{-2b}
\frac{\partial}{\partial a},\qquad X_2^L=\frac{\partial}{\partial b}+2c\frac{\partial}{\partial c},\qquad X_3^L=\frac{\partial}{\partial c}
		$$
		in canonical coordinates of the second kind. Note that
  $$
  [X_1^L,X_2^L]=2X_1^L,\qquad [X_1^L,X_3^L]=0,\qquad [X_2^L,X_3^L]=-2X_3^L.
  $$
  
  Then, condition (\ref{PL_integration}) can be written as
		\[
		\begin{array}{lll}
X_1^LD_1 = 0, & X_2^LD_1 = 2\exp(-2b), & X_3^LD_1 = 0,\\
		X_1^LD_2 = 2 a\exp(-2b), & X_2^LD_2 = 4 c\exp(-4b), & X_3^LD_2 = -2\exp(-4b),\\
		X_1^LD_3 = \exp(-2b), & X_2^LD_3 = 0, & X_3^LD_3 = 0,
		\end{array}
		\]
		where $D(g) = D_1(g) \partial_{a} \wedge \partial_{b} + D_2(g) \partial_{a} \wedge \partial_{c} + D_3(g) \partial_{b} \wedge \partial_{c}$. 
		The solutions of the previous PDEs give rise to a possible general form for $D$: 
		$$
		D(g(a,b,c)) = \left[c_1 - \exp(-2b)\right] \partial_{a} \wedge \partial_{b} + [c_2 + a^2  - 2c\exp(-4b)] \partial_{a} \wedge \partial_{c} + [c_3 + a] \partial_{b} \wedge \partial_{c},
		$$
		where $c_1, c_2, c_3$ are real constants. Since $D(e) = 0$, one has that $c_1 = 1, c_2 = 0, c_3 = 0$. Right-invariant vector fields on $B$ are spanned by 
		$$
		X_1^R=\frac{\partial}{\partial a},
\qquad X_2^R=-2a\frac{\partial}{\partial a}+\frac{\partial}{\partial b},\qquad X_3^R=e^{2b}\frac{\partial}{\partial c},
		$$
		and the Lie--Poisson bracket on $B$ reads, in canonical coordinates of the second kind, as
		$$
		\Lambda_g = [1 - e^{-2b}] \partial_{a} \wedge \partial_{b} - (2ce^{-2b}+a^2e^{2b}) \partial_{a} \wedge \partial_{c} + a e^{2b}\partial_{b} \wedge \partial_{c}.
		$$
		Consequently,
		\[
		\{a, b\} = 1 - e^{-2b}, \quad \{a, c\} =   - 2ce^{-2b}-a^2e^{2b}, \quad \{b, c\} = a e^{2b}.
		\]
		The change of coordinates $v_1 := -b/2, v_2 := a e^b/2, v_3 := c e^b$, transforms the above Poisson brackets into (\ref{DefPoi}). Moreover, the linear approximation $\{\cdot,\cdot\}_1$ of the Poisson bracket $\{\cdot,\cdot\}$ recovers the KKS Poisson bracket on $\mathfrak{sl}_2^*\simeq \mathfrak{b}$, namely
		\begin{equation}\label{Rel:Kir}
		\{a,b\}_1=2b,\qquad \{a,c\}_1=-2c,\qquad \{b,c\}_1=a.
		\end{equation}
		\demo
	\end{example}
	
    Let us study relation (\ref{Rel:Kir}) in a general context.
	The Lie--Poisson bracket $\{\cdot,\cdot\}_{ \mathfrak{g}^*}$ can be related to the KKS Poisson bracket $\{\cdot,\cdot\}_{K}$ on $\mathfrak{g}^*$. To study the relation, let us analyse the linear approximation of $\Lambda^{\mathfrak{g}^*}$ close to the zero in $\mathfrak{g}^*$. To do so, consider a basis $\{v_1,\ldots,v_r\}$ of $\mathfrak{g}$, which gives rise to a global coordinate system on $\mathfrak{g}^*$. If $\Lambda^{\mathfrak{g}^*}_{ij}:=\{v_i,v_j\}_{\mathfrak{g}^*}$ and $\{v_1^*,\ldots, v_r^*\}$ is the dual basis to $\{v_1,\ldots,v_r\}$, then, for $0=\xi\in \mathfrak{g}^*$, one has that
	\begin{equation*}
	\frac{\partial}{\partial v_k}\Bigg\vert_{\xi = 0}\!\!\!\!\!\!\Lambda^{\mathfrak{g}^*}_{ij} 
	= \frac{{\rm d}}{{\rm d}t}\Bigg\vert_{t = 0}\!\!\!\!\!\! \Lambda^{G^*}_{ij}(\exp(tv_k^*))\\
	= \frac{{\rm d}}{{\rm d}t}\Bigg\vert_{t = 0} \left[(R_{\exp(tv_k^*)})_* \left(\frac{\exp({\rm ad}(tv_k^*)) - 1}{{\rm ad}(tv_k^*)} \delta_{\mathfrak{g}^*}(tv_k^*)\right)\right]_{ij}\!\!\!\!
	= [\delta_{\mathfrak{g}^*}(v_k^*)]_{ij},
	\end{equation*}
	where  we have used the formula for the Lie--Poisson bivector presented in \cite[p. 57]{Ch00}. Since $\delta_{\mathfrak{g}^*} = [\cdot,\cdot]_{\mathfrak{g}}^{*}$, one has that $[\delta_{\mathfrak{g}^*}(v_k^*)]_{ij}$ is the coefficient of $v_k$ in the expression of $[v_i, v_j]_{\mathfrak{g}}$ in the chosen basis of $\mathfrak{g}$. Recalling that $\{\cdot,\cdot\}_{\mathfrak{g}^*}$ is a Lie--Poisson bracket and its zero-order term in the coordinates $\{v_1,\ldots,v_r\}$ vanishes, we obtain
	\begin{equation*}
	\{v_i, v_j\}_{\mathfrak{g}^*}(\xi) = \langle \xi, [v_i, v_j]_{\mathfrak{g}}\rangle + \ldots,
	\end{equation*}
	where the latest dots stand for the higher-order polynomial terms in the coordinates $\{v_1,\ldots,v_r\}$. In other words, the first-order term in the expansion of $\{\cdot,\cdot\}_{\mathfrak{g}^*}$ in the linear coordinates on $\mathfrak{g}^*$ retrieves the KKS Poisson bracket.  The following result is also of theoretical interest (see \cite[Theorem 1.3.2]{CP95}).

	\begin{proposition} All Poisson--Lie groups in a connected and simply connected $G$ related to isomorphic Lie bialgebras $(\mathfrak{g},\delta)$ are diffeomorphic, i.e. if $\Lambda_1,\Lambda_2$ are the Poisson bivectors of two Poisson--Lie groups on $G$ related to isomorphic Lie bialgebras on $\mathfrak{g}$, then there exists a Lie group diffeomorphism $\phi:G\rightarrow G$ such that $\phi_*\Lambda_1=\Lambda_2$.
	\end{proposition}

	To accomplish the deformation of the Lie--Hamilton system (\ref{systemLH}), we define a $z$-parametrized family of Poisson bivectors in the following theorem.
	
	\begin{proposition} Let $(\mathfrak{g},\delta)$ be a Lie bialgebra, let $\Lambda^K$ be the KKS Poisson bracket on $\mathfrak{g}^*$, and let $\Lambda^{\mathfrak{g}^*}$ be the Poisson bivector on $\mathfrak{g}^*$ induced by the cobracket $\delta$. If $\phi_z$ is the one-parametric group of diffeomorphisms of the dilatation vector field on $\mathfrak{g}^*$, then
		\begin{equation}\label{Def}
		\Lambda_z:=e^z\phi_{-z*}\Lambda^{\mathfrak{g}^*},\qquad \forall z\in \mathbb{R}\backslash\{0\},\qquad \Lambda_0:=\Lambda_K,\qquad z\in \mathbb{R},
		\end{equation}
		is a $z$-parametric family of Poisson bivectors.
	\end{proposition}
	Since $[\varphi_{*}\Lambda^{\mathfrak{g}^*},\varphi_{*}\Lambda^{\mathfrak{g}^*}]_{SN}=\varphi_{*}[\Lambda^{\mathfrak{g}^*},\Lambda^{\mathfrak{g}^*}]=0$, for an arbitrary diffeomorphism $\varphi:\mathfrak{g}^*\rightarrow \mathfrak{g}^*$, the Poisson bivectors on $\mathfrak{g}^*$ given by $\Lambda_z$ in (\ref{Def}) are Poisson bivectors. It is immediate that $\Lambda_0=\Lambda_K$.	 Note that the Poisson bivectors $\{\Lambda_z\}_{z\in \mathbb{R}}$ do not need to give rise to Poisson--Lie bivectors on the Lie group $G^*$. 
	
	\section{On the deformation procedure}\label{Sec:DfrProc}
	
	This section shows that the deformation procedure presented in the previous section do not lead to equivalent families of $z$-parametric dynamical systems, i.e. there is no $z$-parametric family of diffeomorphisms $\varphi_z:M\rightarrow M$ such that $X_z=\varphi_{z*}X_0$ for every $z\in \mathbb{R}$ and $X_z$ is the Hamiltonian vector field determined by the $t$-dependent Hamiltonian function $h_z$ and $\Lambda_z$. 
	In spite of that, the systems $\{X_z\}_{z\in \mathbb{R}}$ are related among themselves through a Poisson algebra structure in such a way that certain properties of their spaces of solutions can be studied simultaneously by means of our techniques, which will allow us to define a new class of Hamiltonian systems.
	
	\begin{definition} Two deformations $X_{z}$ and $X'_z$ on a manifold $M$ of a Lie--Hamilton system $X$ are called {\it equivalent} if there exists a $z$-parametric family of diffeomorphisms $\{\phi_z:U_z\subset M\rightarrow M\}_{z\in\mathbb{R}}$ such that $X'_z=\phi_{z*}X_z$ for every $z\in \mathbb{R}$ and a family of open, not-empty, subsets $U_z\subset M$.
	\end{definition}
	
	The process of obtaining a deformation of a Lie--Hamilton system described previously depends on the point of a symplectic leaf in $\mathfrak{g}^*$. 

	Recall that the $z$-parametric family of Poisson structures $\{\cdot,\cdot\}_z$ allows us to define a family of Hamiltonian systems on $M$. This is done by defining $X_{t,z}$ to be the Hamiltonian vector field related to $h_{t,z}$ for every $t,z\in \mathbb{R}$.  
	
	Consider a {\it Lie--Hamilton system} on $\mathbb{R}^2$ given by
	\begin{equation*}
	\frac{d\xi}{dt}=\sum_{\alpha=1}^{3}b_\alpha(t)X_\alpha(\xi)=X(t,\xi),\qquad \xi\in \mathbb{R}^2,
	\end{equation*}
	where $b_1(t),b_2(t),b_3(t)$ are arbitrary $t$-dependent functions and the vector fields $X_1,X_2,X_3$ on $\mathbb{R}$ span a Lie algebra $V$ of Hamiltonian vector fields relative to a symplectic form $\omega$ on $\mathbb{R}^2$. We also assume that $X_1,X_2,X_3$ satisfy the commutation relations (\ref{ConRel}), and thus $V$ is isomorphic to $\mathfrak{sl}_2$. Denote by $h_1,h_2,h_3$ some Hamiltonian functions for $X_1,X_2,X_3$ relative to $\omega$, respectively. Our goal is to construct a $z$-parametric family of Hamilton systems on $\mathbb{R}^2$ via a Poisson algebra  deformation of a Poisson algebra on $C^{\infty}(\mathfrak{sl}^*_2)$.
	
	The canonical isomorphism $\mathfrak{sl}_2^{**}\simeq \mathfrak{sl}_2$ allows us to understand $v_1 := J_+, v_2 := J_3/2, v_3 := -J_-$ as a system of global linear coordinates on $\mathfrak{sl}^*_2$. Then, $C^\infty(\mathfrak{sl}^*_2)$ is a Poisson algebra relative to the KKS Poisson bracket,  $\{\cdot,\cdot\}_K$, induced by the Lie algebra $\mathfrak{sl}_2$ \cite{Va94}, namely
	$$
	\{v_1, v_2\}_K = -v_1, \quad \{v_1, v_3\}_K = -2v_2, \quad \{v_2, v_3\}_K = -v_3.
	$$
	Additionally, $C^\infty(\mathfrak{sl}^*_2)$ admits a Poisson structure $\{\cdot,\cdot\}_z$ for every $z\in \mathbb{R}$, satisfying that \cite{BM13}
	\begin{equation}\label{DefPoi}
	\{v_1, v_2\}_z = -\textrm{shc}(2zv_1)v_1, \quad \{v_1, v_3\}_z = -2v_2, \quad \{v_2, v_3\}_z = -\textrm{ch}(2zv_1) v_3,
	\end{equation}
	where we recall that ${\rm shc}(2zv_1)={\rm sh}(2zv_1)/2zv_1$ and $\lim_{z\rightarrow 0}{\rm shc}(2zv_1)=1$. As standard, $C^\infty_z(\mathfrak{sl}_2^*)$ stands for the Poisson algebra induced by $\{\cdot,\cdot\}_z$ on $C^\infty(\mathfrak{sl}_2^*)$.
		
	Let us define the coordinate system $\{v_1, v_2, c_z\}$ on the open set $\mathcal{U}:= \{(v_1,v_2,v_3) \in \mathfrak{sl}_2^*: v_1\neq0\}$, where $c_z := \textrm{shc}(2zv_1)v_1 v_3 - v_2^2$ is the Casimir element relative to $\{\cdot,\cdot\}_z$. In these coordinates, the Poisson bivector, $\Lambda_z$, associated with the Poisson bracket $\{\cdot,\cdot\}_z$ reads $\Lambda_z = -\textrm{shc}(2zv_1)v_1 \partial_{v_1} \wedge \partial_{v_2}$ for each value of $z\in \mathbb{R}$. Therefore, $\Lambda_z$ is tangent to the submanifolds where $c_z$ takes a constant value and $\Lambda_z$ amounts to a symplectic form on it. 
 
 The new coordinates 
	$
	\xi_{1,z} := v_1,\, \xi_{2,z} := {v_2}/(- \textrm{shc}(2zv_1)v_1),\,c_z, 
	$ 
	 are such that $\{\xi_{1,z}, \xi_{2,z}\}_z = 1$, $\{\xi_{i,z},c_z\}_z=0$ for $i=1,2$, and then $\Lambda_z = \partial_{\xi_{1,z}} \wedge \partial_{\xi_{2,z}}$ for every $z\in \mathbb{R}$. That is, for each fixed value $z\in \mathbb{R}$, the coordinates $\{\xi_{1,z},\xi_{2,z},c_z\} $ enable us to write the Poisson bivector $\Lambda_z$ in a canonical form on each symplectic stratum of $\Lambda_z$ within $\mathcal{U}$. We can define so $\Psi_z:(v_1,v_2,v_3)\in \mathfrak{sl}_2^*\mapsto (\xi_{1,z},\xi_{2,z},c_z)\in \mathbb{R}^3$ and $\Psi_{z*}\Lambda_z=\Lambda_c$. Note that 
  $$\Psi_z^{-1}(\xi_{1,z},\xi_{2,z},c_z)=\left(\xi_{1,z},-\xi_{2,z}{\rm shc}(2z \xi_{1,z})\xi_{1,z},\frac{4 c_z z^2 \text{csch}(2 \xi_{1,z} z)+\xi_{2,z}^2 \sinh (2 \xi_{1,z} z)}{2 z}\right)\in \mathfrak{sl}_2^*.
  $$
  Moreover,
	$$
	\hat\xi_1 := \lim_{z\rightarrow 0}\xi_{1,z}=v_1, \quad \hat\xi_2 :=\lim_{z\rightarrow 0}\xi_{2,z}= \frac{v_2}{-v_1},\qquad  \{\hat\xi_1, \hat\xi_2\}_0 = 1.
	$$
	Since $\Lambda_0$ and $\Lambda_z$ have the same canonical form in the coordinates $\{\hat\xi_1,\hat\xi_2,c_0\}$ and $\{\xi_{1,z},\xi_{2,z},c_z\}$, it was shown that $\Psi_{z*}^{-1}\Psi_{0*}$ maps $\Lambda_0$ onto $\Lambda_z$. Since $\Psi_0:(v_1,v_2,v_3)\in \mathfrak{sl}_2^* \mapsto (\xi_{1,0},\xi_{2,0},c_0)\in \mathbb{R}^3$, one has that
 $$
 \Psi_z^{-1}\circ\Psi_0(v_1,v_2,v_3)=\left(v_1,v_2{\rm shc}(2 v_1 z),\frac{v_1v_3-v_2^2+v_2^2{\rm shc}(2zv_1)}{v_1{\rm shc}(2v_1 z)}\right).
  $$

Then, $\phi_z=\Psi_z^{-1}\circ\Psi_0$ reads
 $$
 \phi_z(v_1,v_2,v_3)=\left(v_1,{\rm shc}(2zv_1)v_2,\frac{v_1v_3-v_2^2+{\rm shc}(2zv_1)v_2^2}{{\rm shc}(2zv_1)v_1}\right).
 $$
 and the functions
 $$
 v_{1,z}=v_1\circ \phi_z=v_1,\qquad v_{2,z}=v_2\circ \phi_z={\rm shc}(2zv_1)v_2,\qquad v_{3,z}=v_3\circ \phi_z=\frac{v_1v_3-v_2^2+{\rm shc}^2(2zv_1)v_2^2}{{\rm shc}(2zv_1)v_1}
 $$
 close the same commutation relations relative to the KKS bracket than $v_1,v_2,v_3$ relative to the deformed one. 
 
	Let us determine functions $h_{i,z} \in C^{\infty}(\mathbb{R}^2)$ with the same commutation relations (\ref{DefPoi}), i.e.
	$$
	\{h_{1,z}, h_{2,z}\}_{\omega} = -\textrm{shc}(2zh_{1,z})h_{1,z}, \quad \{h_{1,z}, h_{3,z}\}_{\omega} = -2h_{2,z}, \quad \{h_{2,z}, h_{3,z}\}_{\omega} = -\textrm{ch}(2zh_{1,z})h_{3,z}.
	$$
	Recalls that there exists a Poisson algebra map $D: (C^{\infty}(\mathfrak{sl}_2^*),\{\cdot,\cdot\}_K) \to (C^{\infty}(\mathbb{R}^2), \{\cdot,\cdot\}_{\mathbb{R}^2})$, where $\{\cdot,\cdot\}_{\mathbb{R}^2})$  is the Poisson bracket that comes from the Poisson bivector field $\partial_x \wedge \partial_y$ on $\mathbb{R}^2$, given by
	$$
	D(v_1) = h_1,\quad D(v_2) = h_2, \quad D(v_3) = h_3, 
	$$
	and ensuring that $v_1v_3-v_2^2=c$. Otherwise, the map $\mu_M:z\in M\mapsto (h_1(z),h_2(z),h_3(z))\in \mathfrak{g}^*$ allows us to write $D(v_i)=v_i\circ \mu_M$, for $i=1,2,3$. If we define
	$$
	h_{1,z} := D(\phi_z^*(v_1)), \quad h_{2,z} := D(\phi_z^*(v_2)), \quad h_{3,z} := D(\phi_z^*(v_3)),
	$$
	then
	$$
	h_{1,z} = h_1, \quad h_{2,z} = \textrm{shc}(2zh_1)h_2, \quad h_{3,z} = \frac{\textrm{shc}^2(2zh_1)h_2^2 + h_1h_3-h_2^2}{\textrm{shc}(2zh_1)h_1}.
	$$
	Alternatively, it can be written that
	$$
	h_{1,z}=v_1\circ \phi_z\circ \mu=h_1,\quad 	h_{2,z}=v_2\circ \phi_z\circ \mu={\rm shc}(2zh_1)h_2,$$
 $$
	h_{3,z}=v_3\circ \phi_z\circ \mu=\frac{h_1h_3-h_2^2+{\rm shc}(2zh_1)h_2^2}{{\rm shc}(2zh_1)h_1},\qquad 
	$$
	This gives rise to a $z$-parametric family of deformed Hamiltonian systems given by
	$$
	\frac{ds}{dt}=\sum_{\alpha=1}^3b_\alpha(t)X_{\alpha,z}(s),\qquad \forall t\in \mathbb{R},\qquad \forall s\in \mathbb{R},
	$$
	where $X_{1,z},X_{2,z},X_{3,z}$ are the Hamiltonian vector fields for the functions $h_{1,z},h_{2,z},h_{3,z}$ relative to $\omega$. If one considers that particular case of $h_1=x^2/2, h_2=-px/2,h_3=c/(2x^2)+p^2/2$, then the functions $h_{1,z},h_{2,z},h_{3,z}$ retrieve the result computed in \cite{BCFHL18}.

	\section{Deformation procedure for Jacobi--Lie systems}\label{Se:DefJacob}
	
	The deformation procedure described in the previous sections can be accomplished when the initial non-deformed system is Hamiltonian relative to a Poisson structure. Nevertheless, there are $t$-dependent systems of first-order ordinary differential equations that cannot be considered as $t$-dependent Hamiltonian relative to any Poisson structure \cite{CGLS14}. 
	Among $t$-dependent vector fields that cannot be considered as Hamiltonian systems relative to a Poisson structure, we will here highlight certain subfamilies of the so-called {\it locally automorphic Lie systems} \cite{GLMV19}, whose definition is given next.
	
	\begin{definition} A {\it locally automorphic Lie system} $X$ on a manifold $M$ is a Lie system whose minimal Lie algebra $V^X$ satisfies that $\dim V^X=\dim M$ and whose associated distribution, $\mathcal{D}^{V^X}$, is equal to $TM$. 
	\end{definition}
	
Proposition \ref{Prop:NoGoLoc} to be given next shows a class of locally automorphic Lie systems that are not Hamiltonian relative to any Poisson structure. It is worth noting that Proposition \ref{Prop:NoGoLoc} can be considered as an application of Proposition 5.2 in \cite{CGLS14} to a relevant class of Lie systems.
	
	\begin{proposition}\label{Prop:NoGoLoc} Every locally automorphic Lie system $X$ on an odd-dimensional manifold $M$ is not Hamiltonian relative to any Poisson bivector $\Lambda$ on $M$.
	\end{proposition}
	\begin{proof} Let us proceed by reductio ad absurdum. The characteristic distribution of $\Lambda$  on $M$ is a generalised distribution whose rank is even, but not necessarily constant, at every point of $M$ \cite{We83}. Hence, all Hamiltonian vector fields on $M$ related to $\Lambda$ take values in a generalised distribution that has even rank at every point. Meanwhile, the vector fields of the minimal  Lie algebra, $V^X$, associated with $X$  span $TM$. If all the vector fields of $V^X$  are Hamiltonian, then the unique generalised distribution containing all of them is $TM$, which is odd-dimensional. But then, the vector fields of $V^X$ cannot be contained in a generalised distribution of even rank at every point of $M$. This is a contradiction and the locally automorphic Lie system $X$ is not Hamiltonian relative to any Poisson structure $\Lambda$ on $M$.
	\end{proof}

Let us  provide a Poisson algebra deformation of the Schwarz equation through a Jacobi structure. Since the so-called Schwarz equation \cite{Be07} cannot be considered as a Hamiltonian system relative to a Poisson structure \cite{CGLS14}, the standard method in \cite{BCFHL18} cannot therefore be applied and an extension of that method via Jacobi structures is fully justified. The procedure to be applied next will be clarified in the following sections, and its generalisation to a quite general class of Jacobi-Lie systems will be then provided. 

	The Schwarz equation \cite{Be07,BHLS15,CGL12,CGLS14} reads
	$$
	\frac{{\rm d}^3 x}{{\rm d}t^3} = \frac{3}{2}\left(\frac{{\rm d}x}{{\rm d}t} \right)^{-1} \left(\frac{{\rm d}^2 x}{{\rm d}t^2}\right)^2 + 2b(t) \frac{{\rm d}x}{{\rm d}t},\qquad x\in \mathbb{R},
 \qquad t\in \mathbb{R},
	$$
	where $b(t)$ is an arbitrary $t$-dependent function. The Schwarz equation is relevant due to its relation to Schwarz derivative and its relation to other relevant differential equations appearing in physics \cite{Be07,LG99,OT09}. It can also be considered as a particular kind of third-order Kummer-Schwarz equations. As such, its generalisations or deformations can be of use in the study of generalisations of third-order Kummer--Schwarz equations (see \cite{LS20} and references therein).
	
	The Schwarz equation can be rewritten by adding the new variables $a=dv/dt,v=dx/dt$ as a $t$-dependent system of first-order ordinary differential equations 
 $$\left\{
 \begin{aligned}
 \frac{{\rm d}x}{{\rm d}t}&=v,\\
 \frac{{\rm d}v}{{\rm d}t}&=a,\\
 \frac{{\rm d}a}{{\rm d}t}&=\frac 32 \frac{a^2}v+2b(t)v,\\
 \end{aligned}\right.
 $$
 associated with a $t$-dependent vector field $X$ on $\mathcal{O}=\{(x,v,a):v\neq 0\}$ of the form
	$
	X=N_3+b(t)N_1,
	$
	where
	\begin{equation}\label{schwarz_sl2}
	N_1 = 2v\partial_{a}, \quad N_2 = v \partial_{v} + 2a \partial_{a}, \quad N_3 = v\partial_{x} + a\partial_{v} + \frac{3a^2}{2v} \partial_{a}
	\end{equation}
	satisfy the commutation relations
	$$
	[N_1, N_2] = N_1, \qquad [N_1, N_3] = 2N_2, \qquad [N_2, N_3] = N_3,
	$$
	and therefore span a Vessiot--Guldberg Lie algebra, $V^S$, isomorphic to $\mathfrak{sl}_2$. Let us consider the Lie algebra of Lie symmetries of  (\ref{schwarz_sl2}), which is spanned by the vector fields (see \cite{LS_schwarz} and references therein)
	\begin{equation}\label{Eq:SymSch}
	S_1 = \partial_{x}, \quad S_2 = x\partial_{x} + v \partial_{v} + a \partial_{a}, \quad S_3 = x^2 \partial_{x} + 2vx \partial_{v} + 2(ax+v^2)\partial_{a},
	\end{equation}
	namely $[S_i, N_j] = 0$ for $i,j=1,2,3$. Moreover,
	\begin{equation*}
	\begin{split}
	&[S_1, S_2] = S_1, \quad [S_1, S_3] = 2S_2, \quad [S_2, S_3] = S_3.
	\end{split}
	\end{equation*}
	The  differential one-forms on $\mathcal{O}$ of the form
	$$
	\eta_1 = \textrm{d}x - \frac{2xv^2 + ax^2}{2v^3} \textrm{d}v+\frac{x^2}{2v^2}{\rm d}a, \quad \eta_2 = \frac{v^2 + ax}{v^3} \textrm{d}v - \frac{x}{v^2} \textrm{d}a, \quad \eta_3 = \frac{1}{2v^2} \textrm{d}a - \frac{a}{2v^3} \textrm{d}v
	$$
	are dual to $S_1,S_2,S_3$, namely $\iota_{S_i}\eta_j=\delta_i^j$ for $i,j=1,2,3$, where $\delta_i^j$ is the delta Kronecker function.
	Since $\textrm{d}\eta_2=-\frac 12\eta_1\wedge \eta_3$, it follows that  $\textrm{d}\eta_2 \wedge \eta_2 \neq 0$ and $\eta_2$ becomes  a contact form on $\mathcal{O}$, whereas $\eta_1$ and $\eta_3$ are not. Indeed,  a short calculation shows that $\eta_i\wedge d\eta_i=0$ for $i=1,3$. Moreover, $S_2$ is the Reeb vector field of the contact manifold $(\mathcal{O},\eta_2)$. 

Let us shows that the vector fields $N_1,N_2,N_3$ are Hamiltonian relative to the contact manifold $(\eta_2,\mathcal{O})$, which in turn will show that they are also so for the Jacobi structure $\Lambda_2$ on $\mathcal{O}$ induced by $\eta_\mathcal{O}$ as stated in Section \ref{La:JacCon}. To see this, it is enough to note that 
$$
(\mathcal{L}_{N_i}\eta_2)(S_j)=N_i\eta_2(S_j)-\eta_2([N_i,S_j])=0,\qquad i=1,2,3.
$$
Hence, $N_i$ is Hamiltonian and its Hamiltonian function, $\iota_{N_i}\eta_2$, is a first-integral of $R_\mathcal{O}$ for $i=1,2,3$ (see Section \ref{La:JacCon} for details). 

 The Hamiltonian functions for $N_1,N_2,N_3$ read
	$$
	-h_{1} = \iota_{N_1}\eta_2=-\frac{2x}{v}, \quad -h_{2} =\iota_{N_2}\eta_2= 1 - \frac{ax}{v^2}, \quad -h_{3} = \iota_{N_3}\eta_2=\frac{a}{v} - \frac{xa^2}{2v^3},
	$$
	respectively.
	
	Let us obtain the   Jacobi manifold $(\mathcal{O},\Lambda_\mathcal{O}, R_{\mathcal{O}})$ related to our contact manifold $(\mathcal{O},\eta_2)$, which will allow us to deform Schwarz equations. The Reeb vector field, $R_\mathcal{O}$, of the Jacobi structure $(\mathcal{O},\Lambda_\mathcal{O}, R_{\mathcal{O}})$ is chosen to be minus the Reeb vector field, $S_2$, of the contact manifold $(\mathcal{O},\eta_2)$.  Meanwhile, $\Lambda_{\mathcal{O}}$ is the solution to the system of algebraic equations
	$$
	N_i = \Lambda^\sharp_{\mathcal{O}}\textrm{d}h_{i} + h_{i} R_\mathcal{O}, \qquad i=1,2,3.
	$$
	A simple but long computation shows that
	$$
	\Lambda_\mathcal{O} = xv \partial_x \wedge \partial_v + (v^2 + ax) \partial_x \wedge \partial_a= \frac{1}{2} S_1 \wedge S_3.
	$$
	The Jacobi bivector field, namely $\Lambda_{\mathcal{O}}$, induces a Lie bracket on $C^{\infty}(\mathcal{O})$ given by \cite{HLS15}
	$$
	\{f, g\}_{\Lambda_\mathcal{O},R_\mathcal{O}} := \Lambda_\mathcal{O}(\textrm{d}f, \textrm{d}g) + fR_\mathcal{O}g - gR_\mathcal{O}f,
	$$
 which does not need to satisfy the Leibniz rule \cite{Vi17}.
	The Hamiltonian functions of the $N_i$ are first-integrals of $R_\mathcal{O}$, i.e.
	$$
	R_\mathcal{O}h_i =\mathcal{L}_{S_2} (\iota_{N_i}\eta_2) = \iota_{N_i}\mathcal{L}_{S_2} \eta_2 + \iota_{[S_2, N_i]}\eta_2 = \iota_{N_i}\iota_{S_2} \textrm{d}\eta_2 + \iota_{N_i}\textrm{d}\iota_{S_2} \eta_2 = 0,\qquad i=1,2,3.
	$$
	This makes $h_1,h_2,h_3$ to have a series of nice properties, which justifies calling them {\it good Hamiltonian functions} \cite{HLS15}. Further, 
	$
	\{h_{i}, h_{j}\}_{\Lambda_\mathcal{O},R_\mathcal{O}} = \Lambda(\textrm{d}h_{i}, \textrm{d}h_{j}),
	$
 for $i=1,2,3$, 
	and, more specifically,
	\begin{equation}\label{Ham_com}
	\{h_1, h_2\}_{\Lambda_\mathcal{O},R_\mathcal{O}} = h_1, \quad \{h_2, h_3\}_{\Lambda_\mathcal{O},R_\mathcal{O}} = h_3, \quad \{h_1, h_3\}_{\Lambda_\mathcal{O},R_\mathcal{O}} = 2h_2. 
	\end{equation}
	Let us find the coordinates $(\xi, u_1, u_2)$ in which the Reeb vector field $R_\mathcal{O} = -x\partial_x - v \partial_v - a \partial_a$ takes the form $R_\mathcal{O} = \partial_{\xi}$. The equations $R_\mathcal{O}u_1 = 0, R_\mathcal{O}u_2 = 0$ have a particular solution $u_1 = \frac{a}{v}, u_2 = \frac{x}{v}$ on $\mathcal{O}$, which allows us to write
	$$
	h_1 = 2u_2, \quad h_2 = -1 + u_1 u_2, \quad h_3 = -u_1 + \frac{u_2 u_1^2}{2}.
	$$
	Therefore, $h_i= h_i(u_1, u_2)$ for $i=1,2,3$. Since $\{u_1, u_2\}_{\Lambda_\mathcal{O}, R_\mathcal{O}} = -1$, the space $\mathcal{O}_{H}=\{f(u_1,u_2):f\in C^\infty(\mathbb{R}^2) \}$ becomes a Poisson algebra relative to the Poisson bracket $\{\cdot,\cdot\}_{\Lambda_\mathcal{O},R_\mathcal{O}}$ and $u_1,u_2$ are the corresponding canonical coordinates. Indeed, $R_\mathcal{O}\vert_{\mathcal{O}_{H}} = 0$ implies, by Lemma 4.4 in \cite{HLS15}, that $\mathcal{O}_H$ is a Poisson algebra relative to the restriction of $\Lambda_{\Lambda_{\mathcal{O}},R_{\mathcal{O}}}$ to $\mathcal{O}_H$. This restriction allows us to deform the Hamiltonian functions $h_1,h_2,h_3$ via the procedure described in the previous section.
	
	Recall that $C^{\infty}(\mathfrak{sl}^*_2)$ can be equipped with the  Poisson bracket \cite{BCFHL18}
	$$
	\{v_1, v_2\}_z = -\textrm{shc}(2zv_1)v_1, \quad \{v_1, v_3\}_z = -2v_2, \quad \{v_2, v_3\}_z = -\textrm{ch}(2zv_1)v_3
	$$
 relative to a certain $z\in \mathbb{R}$.
	
	However, to match the relations (\ref{Ham_com}) in the limit $z \to 0$, one has to change the coordinates to $v'_1 = v_1$, $v'_2 = -v_2$, $v'_3 = v_3$. Then, the Poisson bracket on $C_z^{\infty}(\mathfrak{sl}_2^*)$ reads
	$$
	\{v'_1, v'_2\}_z = \textrm{shc}(2zv'_1)v'_1, \quad \{v'_1, v'_3\}_z = 2v'_2, \quad \{v'_2, v'_3\}_z = \textrm{ch}(2zv'_1)v'_3.
	$$
	Henceforth, we will omit the prime sign and write $C_z^{\infty}(\mathfrak{sl}_2    ^*)$ for the Poisson algebra relative to $\{\cdot,\cdot\}_z$. Let $\Lambda_z$ be associated Poisson bivector on $\mathfrak{sl}_2^*$.
	
	If the symplectic leaves of the Poisson manifold $\Lambda_z$ give rise to a foliation of an open subset $U$ of $\mathfrak{sl}_2^*$, then $\Lambda_z$ admits Darboux coordinates mapping it to the canonical form on each symplectic leaf of $U$. For instance, if $\xi_{1,z}, \xi_{2,z}, c_z$ are local coordinates on $\mathcal{O}\subset \mathfrak{sl}_2^*$ of the form
	$$
	\xi_{1,z} = 2 \frac{v_2-1}{v_1 \textrm{shc}(2zv_1)}, \quad \xi_{2,z} = -\frac{1}{2}v_1, \quad c_z = \textrm{shc}(2zv_1)v_1 v_3 - v_2^2,
	$$
	then $\Lambda_z = \partial_{\xi_{1,z}} \wedge \partial_{\xi_{2,z}}$ and $c_z$ is a Casimir. 
 
 Note that we have the Poisson morphism $\mu:x\in M\mapsto (h_1,h_2,h_3)\in \mathfrak{sl}_2^*$
 relative to a KKS Poisson bracket on $\mathfrak{sl}_2^*$ and the Poisson bracket induced Lie bracket induced by $\Lambda_{\mathcal{O},R_\mathcal{O}}$ on the functions on $M$ of the form $f(h_1,h_2,h_3)$ for any $f\in \mathbb{R}^3\rightarrow \mathbb{R}$.
 
Proceeding like in previous section for the maps induced by previous rectifying variables
	$$
	h_{i,z} := v_i\circ \phi_z\circ \mu,\qquad i=1,2,3,
	$$
and hence
 $$
 h_{1,z} = 2u_2 \quad h_{2,z} =1+(-2+u_1 u_2){\rm shc}(4 u_2 z), $$
 $$h_{3,z} = (-2+u_1u_2)(1/u_2+(-1/u_2+u_1/2){\rm shc}(4u_2 z)).
 $$

	As expected,
	$$
	\{h_{1,z}, h_{2,z}\} = h_{1,z} \textrm{shc}(2zh_{1,z}),\quad \{h_{1,z}, h_{3,z}\} = 2h_{2,z}, \quad \{h_{2,z}, h_{3,z}\} = \textrm{ch}(2zh_{1,z})h_{3,z},
	$$
	Moreover, $\lim\limits_{z \to 0} h_{i,z}(\xi) = h_i(\xi)$ for each $\xi \in \mathcal{O}$ and the above commutation relations reduce to (\ref{Ham_com}).
	
	The deformed Hamiltonian vector fields $Y_{i,z}$ are defined as Hamiltonian ones relative to the previously introduced Jacobi structure, namely $Y_{i,z} = \Lambda^\sharp \textrm{d}h_{i,z} + h_{i,z}R$. In the coordinates
	$$
	u_1 = \frac{a}{v}, \quad u_2 = \frac{x}{v}, \quad \xi = \log v ,
	$$
	the used Jacobi manifold $(\mathcal{O}, \Lambda, R)$ reads 
	$$
	\Lambda = \partial_{u_2} \wedge \partial_{u_1} + u_2 \partial_{u_2} \wedge \partial_{\xi},
	\quad R = \partial_{\xi}.
	$$
	Therefore, in coordinates $\{u_1, u_2, \xi\}$, one obtains $Y_{1,z}= 2\partial_{u_1},$ while

\begin{multline*}
   Y_{2,z} = \frac{2 u_2 z (u_1 u_2-2) \cosh (4 u_2 z)+\sinh (4 u_2 z)}{2 u_2^2 z}\partial_{u_1} - \frac{\sinh (4 u_2 z)}{4 z} \partial_{u_2} \\-\left[u_2 \frac{2 u_2 z (u_1 u_2-2) \cosh (4 u_2 z)+\sinh (4 u_2 z)}{2 u_2^2 z}+\frac{(u_1 u_2-2) \sinh (4 u_2 z)}{4 u_2 z}+1\right] \partial_{\xi},\\
\end{multline*}
and
\begin{multline*}
		Y_{3,z} = \frac{(4 u_2 z + 
 u_2 (-2 + u_1 u_2)^2 z \cosh[4 u_2 z] + (-2 + u_1 u_2) {\rm sinh}[
   4 u_2 z])}{2 u_2^3 z}\partial_{u_1} \\- \left[1 + \frac{(-2 + u_1 u_2) \sinh[4 u_2 z]}{4 u_2 z}\right] \partial_{u_2} \\+ \left[u_2\frac{(4 u_2 z + 
 u_2 (-2 + u_1 u_2)^2 z \cosh[4 u_2 z] + (-2 + u_1 u_2) {\rm sinh}[
   4 u_2 z])}{2 u_2^3 z}\right.\\\left.+((-2 + u_1 u_2) (8 u_2 z + (-2 + u_1 u_2) \sinh[4 u_2 z]))/(8 u_2^2 z) \right] \partial_{\xi}.
\end{multline*}
 
	In the limit $z \to 0$,
	$$
	Y_{1,0} = 2\partial_{u_1}, \quad Y_{2,0} = u_1 \partial_{u_1} - u_2 \partial_{u_2} + \partial_{\xi}, \quad Y_{3,0} = \frac{1}{2} u_1^2 \partial_{u_1} + [1 - u_1 u_2] \partial_{u_2} + u_1 \partial_{\xi},
	$$
	which in coordinates $\{x,v,a\}$ read
	$$
	Y_{1,0} = 2v \partial_a, \quad Y_{2,0} = v\partial_v + 2a\partial_a, \quad Y_{3,0} = v\partial_x + a \partial_v + \frac{3a^2}{2v} \partial_a.
	$$
	In the limit $z \to 0$,
	\begin{equation*}
	\begin{split}
	&[Y_{1,0}, Y_{2,0}] = 2 \partial_{u_1} = Y_{1,0},\\
	&[Y_{1,0}, Y_{3,0}] = 2u_1 \partial_{u_1} - 2u_2 \partial_{u_2} + 2\partial_{\xi} = 2Y_{2,0},\\
	&[Y_{2,0}, Y_{3,0}] = \frac{1}{2} u_1^2 \partial_{u_1} + [1 - u_1 u_2] \partial_{u_2} + u_1 \partial_{\xi} = Y_{3,0}.
	\end{split}
	\end{equation*}
	The deformed dynamical system is governed by a $t$-dependent vector field $X_{z,t}$, defined similarly as in the undeformed case by $X_{z,t} =b(t) Y_{z,1}  + Y_{z,3}$.

\chapter{Conclusions and Outlook}

Both methods devised in this thesis to address the classification problem for Lie bialgebras have proven efficient in computations of inequivalent r-matrices for three-dimensional and four-dimensional indecomposable Lie algebras. In the former case, the results agreed with already known works on this problem \cite{FJ15, Go00}. From the practical standpoint, these methods allow for noticeably faster determination of equivalent classes of r-matrices, especially by using the first approach. The main advantage in this aspect comes from the fact that the analysis of invariant forms is a relatively simple linear problem. The second method, although not as straightforward, offers additional geometric perspectives on the classification issue. 

Despite these advantages, there several difficulties related to both approaches. As already hinted in Section \ref{Ch:alg_Sec:4D}, the analysis of the orbits of the action of ${\rm Inn}(\mathfrak{g})$ gets complicated when there are more than one invariant form admissible for the considered Lie algebra $\mathfrak{g}$. In this case, an appropriate coordinate system was necessary to proceed. Similarly, the study of the modified Yang--Baxter equation is practically impossible unless suitable coordinates are employed.  Additionally, the determination of Darboux families is a nontrivial task that requires careful examination of Lie algebra derivations. Another difficulty is the identification of the orbits (either of the action of ${\rm Inn}(\mathfrak{g})$ in the first method, or of ${\rm Aut}_c(\mathfrak{g})$ in the second one) relative to the whole automorphism group. The structure of such group is already complicated for ${\rm dim}(\mathfrak{g}) = 4$ and direct computations, even though supported by software, might be more intricate when ${\rm dim}(\mathfrak{g}) > 4$. Such problems might be a significant hindrance to an efficient study of higher-dimensional Lie algebras.

However, the considerations in Section \ref{Ch:alg_Sec:4D} suggest a few directions to possibly overcome some of these troubles. In this example, we dealt with the three-parameter family of invariant forms $b_{\Lambda^2 \mathfrak{gl}_2}$ on $\mathfrak{gl}_2$ or in other words, an invariant form that was a linear span of three separate ones. Its analysis was greatly simplified by using the coordinates in which $b_{\Lambda^2 \mathfrak{gl}_2}$ gets an especially manageable form, as the first two forms are diagonal. Thus, the problem splits into two subspaces of smaller dimension than the initial one, given accordingly by the coordinates of each of these forms, that are additionally interconnected by the quadratic function associated with the third form. In consequence, we could easily analyse the orbits on these subspaces and in the last step, we imposed the condition given by the remaining form. It would be interesting to verify whether an analogous procedure yields results efficiently in higher-dimensional cases.

Apart from the concept of a brick, we have not established any techniques for the determination of Darboux families that would significantly facilitate their search. Much work in this thesis related to this problem has been solely based on an inspection of Lie algebra derivations. Thus, an important improvement of the method would rely on finding such techniques. 

Let us also notice that, contrary to the first method, the structure of a Lie algebra seems to play no role in the second approach. Similarly, the structure of Lie algebra of its derivations is not relevant in subsequent computations. This issue opens another direction in the study of Darboux families and their geometry.

One technical question that remains to be tackled concerns the generalisation of both methods to Lie algebras over other fields than the reals. In particular, special attention should be paid to the complex case. It would be interesting to investigate whether there exists any relation between the solutions for a complex Lie algebra and the solutions for its real forms. 

Finally, one might attempt to generalise both methods to classify general cocommutators.

In the second part of this thesis, we focused on applications of Lie bialgebras.	In this regard, we have discussed the notion of foliated Lie systems, which significantly extend the mostly theoretical examples given in \cite{CGM00}. We have introduced foliated superposition rules and first-order systems of ODEs admitting such a superposition rule have been studied. Moreover, we have defined foliated automorphic Lie systems and their relations to foliated Lie systems have been investigated. 

The theory of foliated Lie systems has been applied to study to a generalisation of Ermakov systems, automorphic Lie systems related to Lax pairs and certain Hamiltonian systems. In particular, the theory of Lie--Hamilton systems has been extended to foliated Lie-Hamilton systems.
	
We expect that our results can be extended to the so-called PDE Lie systems \cite{CGM07,Dissertationes,OG00,Ra11}. This should be accomplished by using the same fundamental ideas as presented in Chapter 5, but the development is technically much more complicated due to the nature of PDEs.
	
Several foliated Lie systems studied in the applications of this work are concerned with Hamiltonian systems admitting a maximal number of functionally independent autonomous constants of motion in involution relative to the Poisson bracket induced by the associated symplectic structure. Such systems can be understood as a $t$-dependent analogue of completely integrable Hamiltonian systems (see \cite{GMS02,Sa98,Sa12} for related notions). 
We aim to apply our methods as well as their possible generalisations to the analysis of $t$-dependent integrable non-commutative systems (see \cite{FT98,LLV18,Ma12,SV00}). 

Finally, we have sketched the generalisation of the deformation procedure for Lie--Hamilton systems \cite{BCFHL18,BCFHL21} in terms of Jacobi structures. In particular, we have studied the Schwarz equation that cannot be considered as a Hamiltonian system relative to a Poisson structure. 

Following \cite{BCFHL21}, the more geometric approach to this generalised method needs to be introduced. It is also necessary to analyse further examples. As indicated in Section \ref{Se:DefJacob}, locally automorphic Lie systems might be of particular interest for that purpose.

\appendix
 
\chapter{Commands for Mathematica}\label{App:code}
	
	In this section, we detail several commands and procedures designed for Mathematica that allowed us to accomplish certain easy but slightly lengthy calculations in Chapter \ref{Ch:alg_met} and Chapter \ref{Ch:Darboux}.

Each command needs to be initialised by specifying the matrices ${\tt e[i]}$ that constitute a matrix representation of the basis of the Lie algebra $\mathfrak{g}$, the dimension {\tt Dim} of $\mathfrak{g}$ and the size {\tt nMatrix} of the representation space of $\mathfrak{g}$. For instance, initialisation parameters for the Lie algebra $\mathfrak{sl}_2$ used in this thesis are  ${\tt Dim = 3}$, ${\tt nMatrix = 3}$ and
 {\small
$$
{\tt
e[1] = \left( \begin{array}{ccc}
0 & 0 & 0 \\
0 & 1 & 0 \\
0 & 0 & -1 
\end{array} \right), \qquad
e[2] = \left( \begin{array}{ccc}
0 & 0 & 1 \\
-1 & 0 & 0 \\
0 & 0 & 0 
\end{array} \right), \qquad
e[3] = \left( \begin{array}{ccc}
0 & -1 & 0 \\
0 & 0 & 0 \\
1 & 0 & 0 
\end{array} \right).
}
$$
}
The basis of $\Lambda^2 \mathfrak{g}$ is given by the bivectors $e_{ij} := e_i \wedge e_j$, where $1 \leq i < j \leq {\tt Dim}$ and it is ordered lexicographically, namely it reads $(e_{12}, e_{13}, \ldots, e_{1 {\tt Dim}}, e_{23}, e_{24}, \ldots, e_{2, {\tt Dim}}, \ldots e_{({\tt Dim} - 1), {\tt Dim}})$. Similarly, if the output is an element of $\lambda^3 \mathfrak{g}$, it is represented in the basis $e_{ijk} := e_i \wedge e_j \wedge e_k$, where $1 \leq i < j < k \leq {\tt Dim}$ with the lexicographical order as well.

Whenever a command requires a bivector $r \in \Lambda^2 \mathfrak{g}$ as an input, it has to be represented as a tuple of length $\binom{Dim}{2}$. Each position corresponds to the coefficient in the described basis. For instance, let $\{e_1, e_2, e_3, e_4\}$ be a basis of the Lie algebra $\mathfrak{gl}_2$. Then, the induced basis for $\Lambda^2 \mathfrak{gl}_2$ reads $\{e_{12}, e_{13}, e_{14}, e_{23}, e_{24}, e_{34}\}$ and an element $r = e_{12} + 2 e_{13} - 3 e_{24} \in \Lambda^2 \mathfrak{gl}_2$ needs to be represented as ${\tt r = \{1,2,0,0,-3,0\}}$.

\paragraph{{\tt \bf Com[i, j]}}
This command computes the commutator $[e_i, e_j] = c_{ij}^k e_k$ of the basis elements ${\tt e[i]}$ and ${\tt e[j]}$ in form of a tuple of length ${\tt Dim}$, where the structure constant $c_{ij}^k$ is presented at the $k$-th position in the tuple.

\begin{verbatim}
Com[i_, j_] := {
A = Table[a[l], {l, 1, Dim}];
temp = {e[i] . e[j] - e[j] . e[i] == Sum[e[k] a[k], {k, 1, Dim}]};
sol = Solve[Flatten[temp], Flatten[A]];
Flatten[A /. sol]}[[1]]
\end{verbatim}

\paragraph{{\tt \bf Str}}
This command gives the structure constants for the commutators of basis elements of $\mathfrak{g}$. The result is presented as a matrix with ${\tt Dim}^2$ rows and ${\tt Dim}$ columns. The structure constant $c_{ij}^k$ is given in the $((i - 1) \cdot {\tt Dim} + j)$-th row and $k$-th column.

\begin{verbatim}
Str := Table[Com[i, j][[k]], {i, 1, Dim}, {j, 1, Dim}, {k, 1, Dim}]
\end{verbatim}

\paragraph{{\tt \bf Deriv}}
This command gives the derivations of a Lie algebra $\mathfrak{g}$, i.e. linear maps $D \in {\rm End}(\mathfrak{g})$ such that $D([v,w]) = [D(v), w] + [v, D(w)]$. The result is presented as a matrix with ${\tt Dim}$ rows and ${\tt Dim}$ columns, where $i$-th row corresponds to $D(e_i)$.

\begin{verbatim}
Deriv := {
Der = Table[der[i, j], {i, 1, Dim}, {j, 1, Dim}];
temp = Table[Sum[Der[[l, k]] Str[[i, j, k]] - Der[[k, i]] Str[[k, j, l]] 
  - Str[[i, k, l]] Der[[k, j]], {k, 1, Dim}] == 0, 
  {i, 1, Dim - 1}, {j, i + 1, Dim}, {l, 1, Dim}];
sol = Flatten[Solve[Flatten[temp], Flatten[Der]]];
Der /. sol}
\end{verbatim}

\paragraph{{\tt \bf Id}}
This command constructs the identity matrix of dimension given by the parameter ${\tt nMatrix}$.

\begin{verbatim}
Id := Table[KroneckerDelta[i, j], {i, 1, nMatrix}, {j, 1, nMatrix}]
\end{verbatim}

\paragraph{{\tt \bf OI2}}
This command constructs a table of length $\frac{1}{2} {\tt Dim}({\tt Dim} - 1)$, consisting of pairs of indices $(i,j)$, such that $1 \leq i < j \leq {\tt Dim}$, in lexicographic order.

\begin{verbatim}
OI2 := Flatten[Table[{i, j}, {i, 1, Dim - 1}, {j, i + 1, Dim}], 1]
\end{verbatim}

\paragraph{{\tt \bf OI3}}
This command constructs table of length $\frac{1}{6} {\tt Dim}({\tt Dim} - 1)({\tt Dim} - 2)$, consisting of pairs of indices $(i,j,k)$, such that $1 \leq i < j < k \leq {\tt Dim}$, in lexicographic order.

\begin{verbatim}
OI3 := Flatten[Table[{i, j, k}, {i, 1, Dim - 2}, {j, i + 1, Dim - 1}, {k, j + 1, Dim}], 2]
\end{verbatim}

\paragraph{{\tt \bf e[i, j]}}
This command creates the matrix representation for the bivector $e_{ij} := e_i \wedge e_j$, which is the basis element of $\Lambda^2 \mathfrak{g}$.

\begin{verbatim}
e[i_, j_] :=  KroneckerProduct[e[i], e[j]] - KroneckerProduct[e[j], e[i]]
\end{verbatim}

\paragraph{{\tt \bf e[i,j,k]}}
This command creates the matrix representation for the bivector $e_{ijk} := e_i \wedge e_j \wedge e_k$, which is the basis element of $\Lambda^3 \mathfrak{g}$.

\begin{verbatim}
e[i_, j_, k_] := KroneckerProduct[e[i], e[j], e[k]] - KroneckerProduct[e[j], e[i], e[k]] 
- KroneckerProduct[e[k], e[j], e[i]] - KroneckerProduct[e[i], e[k], e[j]] 
+ KroneckerProduct[e[k], e[i], e[j]] + KroneckerProduct[e[j], e[k], e[i]]
\end{verbatim}

\paragraph{{\tt \bf LiftMor2[T, i, j]}}
Given $T \in {\rm End}(\mathfrak{g})$, this command constructs a map $\Lambda^2 T \in {\rm End}(\Lambda^2 \mathfrak{g})$ defined as $\Lambda^2 T := T \otimes 1 + 1 \otimes T$. The endomorphism ${\tt T}$ must be a two-dimensional table with ${\tt Dim}$ rows and ${\tt Dim}$ columns. As the output, this command gives a tuple of length $\frac{1}{2}{\tt Dim}({\tt Dim} - 1)$ corresponding to $\Lambda^2 T(e_{ij}) = \sum_{p < q} a_{pq} e_{pq}$. The $((p - 1) \cdot {\tt Dim} + q)$-th position in this tuple is the coefficient $a_{pq}$.

\begin{verbatim}
LiftMor2[T_, i_, j_] := {
A = Flatten[Table[a[l, m], {l, 1, Dim - 1}, {m, l + 1, Dim}]];
temp = Sum[T[[i, l]] e[l, j] + T[[j, l]] e[i, l], {l, 1, Dim}] 
  == Sum[e[l, m] a[l, m], {l, 1, Dim - 1}, {m, l + 1, Dim}];
sol = Flatten[Solve[Flatten[temp], A]];
A /. sol}[[1]]
\end{verbatim}

\paragraph{{\tt \bf GLiftMor2[T, i, j]}}
Given $T \in {\rm End}(\mathfrak{g})$, this command constructs a map $\Lambda^2 T \in {\rm End}(\Lambda^2 \mathfrak{g})$ defined as $\Lambda^2 T := T \otimes T$. The endomorphism ${\tt T}$ must be a two-dimensional table with ${\tt Dim}$ rows and ${\tt Dim}$ columns. As the output, this command gives a tuple of length $\frac{1}{2}{\tt Dim}({\tt Dim} - 1)$ corresponding to $\Lambda^2 T(e_{ij}) = \sum_{p < q} a_{pq} e_{pq}$. The $((p - 1) \cdot {\tt Dim} + q)$-th position in this tuple is the coefficient $a_{pq}$.

\begin{verbatim}
GLiftMor2N[T_, i_, j_] := {
A = Flatten[Table[a[l, m], {l, 1, Dim}, {m, l + 1, Dim}]];
temp = Sum[T[[i, l]]*T[[j, m]] e[l, m], {l, 1, Dim}, {m, 1, Dim}] 
  == Sum[e[l, m] a[l, m], {l, 1, Dim}, {m, l + 1, Dim}];
sol = Flatten[Solve[Flatten[temp], A]];
A /. sol}[[1]]
\end{verbatim}

\paragraph{{\tt \bf Com2[k, i, j]}}
This command gives a Schouten bracket $[e_k, e_{ij}]_S$ in terms of basis elements of $\Lambda^2 \mathfrak{g}$. The result is presented as a tuple of length $\frac{1}{2} {\tt Dim} ({\tt Dim} - 1)$.  

\begin{verbatim}
Com2[k_, i_, j_] := {
A = Table[a[l, m], {l, 1, Dim}, {m, l + 1, Dim}];
temp = {KroneckerProduct[e[k], Id] . e[i, j] - e[i, j] . KroneckerProduct[e[k], Id] 
  + KroneckerProduct[Id, e[k]] . e[i, j] - e[i, j] . KroneckerProduct[Id, e[k]] 
  == Sum[e[l, m] a[l, m], {l, 1, Dim}, {m, l + 1, Dim}]};
sol = Solve[Flatten[temp], Flatten[A]];
Flatten[A /. sol]}[[1]]
\end{verbatim}

\paragraph{{\tt \bf Com3[p, i, j, k]}}
This command gives a Schouten bracket $[e_p, e_{i j k}]_S$ in terms of basis elements of $\Lambda^3 \mathfrak{g}$. The result is presented as a tuple of length $\binom{{\tt Dim}}{3}$.  

\begin{verbatim}
Com3[p_, i_, j_, k_] := {
fun=OI3;
A = Table[a[fun[[t, 1]], fun[[t, 2]], fun[[t, 3]]], {t, 1, Binomial[Dim, 3]}];
temp = {Sum[a[x, y, z]*e[x, y, z], {x, 1, Dim}, {y, x + 1, Dim}, {z, y + 1, Dim}] 
  == Sum[Com[p, i][[x]]*e[x, j, k], {x, 1, j - 1}] 
  + Sum[Com[p, j][[y]]*e[i, y, k], {y, i + 1, k - 1}] 
  +  Sum[Com[p, k][[z]]*e[i, j, z], {z, j + 1, Dim}]};
sol = Solve[Flatten[temp], Flatten[A]];
Flatten[A /. sol]}[[1]]
\end{verbatim}

\paragraph{{\tt \bf Ad2[k]}}
This command gives an action of basis elements of $\mathfrak{g}$ on the basis elements of $\Lambda^2 \mathfrak{g}$.

\begin{verbatim}
Ad2[k_] := {
Part1 = Flatten[Table[Com2[k, i, j], {i, 1, Dim}, {j, i + 1, Dim}]]; 
Table[Part1[[i + (j - 1)*Dim (Dim - 1)/2 ]], 
  {i, 1, Dim (Dim - 1)/2}, {j, 1, Dim (Dim - 1)/2}]}[[1]]
\end{verbatim}

\paragraph{{\tt \bf AllAd2Invar}}
This command computes all $\mathfrak{g}$-invariant forms on $\Lambda^2 \mathfrak{g}$.

\begin{verbatim}
AllAd2Invar := {
tmp = Table[If[i <= j, a[i, j], a[j, i]], {i, 1, Dim (Dim - 1)/2}, 
  {j, 1, Dim (Dim - 1)/2}];
zeroes = Table[0, {i, 1, Dim (Dim - 1)/2}, {j, 1, Dim (Dim - 1)/2}];
Eqn = Transpose[Sum[a[i] Ad2[i], {i, 1, Dim }]] . tmp 
  + tmp . Sum[a[i] Ad2[i], {i, 1, Dim }] == zeroes;
Eqs = Table[Equation[pil], {pil, 1, Dim}];
For[Index = 1, Index <= Dim, Index++, Eqs[[Index]] = D[Eqn, a[Index]]];
Solve[Flatten[Eqs], Flatten[tmp]];
tmp/.sol}[[1]]
\end{verbatim}

\paragraph{{\tt \bf Coboundary2D[$r_1$, $r_2$]}}
This command constructs the matrix form of the Schouten bracket $[r_1, r_2]_S$ for two bivectors $r_1, r_2 \in \Lambda^2 \mathfrak{g}$. It is an auxilliary function for the command ${\tt CobD2Sim}$.

\begin{verbatim}
CoboundaryD2[r1_, r2_] := {
fun = OI2; 
Sum[r1[[i]] r2[[k]] Sum[(-1)^(l + m)
Wedge[e[fun2[[i]][[l]]] . e[fun[[k]][[m]]] - e[fun[[k]][[m]]] . e[fun[[i]][[l]]], 
  e[fun[[i]][[3 - l]]], e[fun[[k]][[3 - m]] ]], {l, 1, 2}, {m, 1, 2}], 
  {i, 1, Dim (Dim - 1)/2}, {k, 1, Dim (Dim - 1)/2}]}[[1]]  
\end{verbatim}

\paragraph{{\tt \bf CobD2Sim[$r_1$, $r_2$]}}
This command computes the Schouten bracket $[r_1, r_2]_S$ of $r_1, r_2 \in \Lambda^2 \mathfrak{g}$. The result is presented as a tuple of length $\binom{Dim}{3}$, where $k$-th position corresponds to the $k$-th element of the ordered basis of $\Lambda^3 \mathfrak{g}$.

\begin{verbatim}
CobD2Sim[r1_, r2_] := {
fun = OI3;
temp = {CoboundaryD2[r1, r2] == Sum[a[z] Wedge[e[fun[[z]][[1]]], 
  e[fun[[z]][[2]]], e[fun[[z]][[3]]]], {z, 1, Binomial[Dim, 3]}]};
pars = Table[a[z], {z, 1, Dim!/(3! (Dim - 3)!)}];
Resul = Solve[Flatten[temp], Flatten[pars]];
pars /. Resul}[[1]][[1]]
\end{verbatim}

\paragraph{{\tt \bf Ad2InvarElements}}
This command gives all $\mathfrak{g}$-invariant elements in $\Lambda^2 \mathfrak{g}$ relative to the $\mathfrak{g}$-module $(\Lambda^2 \mathfrak{g}, \Lambda^2 {\rm ad})$.

\begin{verbatim}
Ad2InvarElements := {
varrob = Table[a[i], {i, 1, Binomial[Dim, 2]}];
fun = OI2;
eqn = Table[0, {i, 1, Dim}]; 
For[ind = 1, ind <= Dim, ind++, 
  eqn[[ind]] = Sum[varrob[[k]]*Com2N[ind, fun[[k, 1]], fun[[k, 2]]], 
  {k, 1, Binomial[Dim, 2]}]]; 
sol = Flatten[Solve[Table[eqn[[p]] == 0, {p, 1, Dim}], varrob]];
varrob /. sol}
\end{verbatim}

\paragraph{{\tt \bf Ad3InvarElements}}
This command computes the $\mathfrak{g}$-invariant elements in $\Lambda^3 \mathfrak{g}$ relative to the $\mathfrak{g}$-module $(\Lambda^3 \mathfrak{g}, \Lambda^3 {\rm ad})$. The result is presented as a possibly nonempty tuple, where each position corresponds to the element of the ordered basis of $\Lambda^3 \mathfrak{g}$. If it contains a combination of parameters ${\tt a[k]}$, we get the family of $\mathfrak{g}$-invariant elements in $\Lambda^3 \mathfrak{g}$ parametrised by these values. For example, if the output reads ${\tt \{a[1], 0, a[3], 0\}}$, we obtain the family of $\mathfrak{g}$-invariant elements of the form $a_1 e_{123} + a_3 e_{134}$, where $a_1, a_3 \in \mathbb{R}$.

\begin{verbatim}
Ad3InvarElements := {
varrob = Table[a[i], {i, 1, Binomial[Dim, 3]}];
eqn = Table[0, {i, 1, Dim}]; 
For[ind = 1, ind <= Dim, ind++, 
  eqn[[ind]] = Sum[varrob[[k]]*Com3[ind, fun[[k, 1]], fun[[k, 2]], fun[[k, 3]]], 
  {k, 1, Binomial[Dim, 3]}]]; 
Solve[Table[eqn[[p]] == 0, {p, 1, Dim}], varrob];
varrob/.sol}
\end{verbatim}

\paragraph{{\tt \bf AllAut}}
This command computes automorphisms of $\mathfrak{g}$. As an output, it gives the matrix form of the obtained automorphism together with its determinant and the matrix form of its extension to $\Lambda^2 \mathfrak{g}$.

\begin{verbatim}
AllAutN := {
autmatrix = Table[T[a, b], {a, 1, Dim}, {b, 1, Dim}];
temp = Table[Sum[T[p, i]*T[q, j]*Com[p, q][[k]], {p, 1, Dim}, {q, 1, Dim}] 
  == Sum[Com[i, j][[l]]*T[k, l], {l, 1, Dim}], 
  {i, 1, Dim - 1}, {j, i + 1, Dim}, {k, 1, Dim}];
solutions = Solve[Flatten[temp], Flatten[autmatrix]];
k = 1;
For[i = 1, i < Length[solutions] + 1, i++,
  autsol = autmatrix /. solutions[[i]];
  If[Det[autsol] === 0, Null,
    Print["Automorphism matrix no ", k];
    Print[Simplify[MatrixForm[autsol]]];
    Print["Determinant: "];
    Print[Simplify[Det[autsol]]];
    Print["Matrix of a automorphism no ", k, " lifted to bivector space:"]; 
    tempmatrix = {}; 
    For[p = 1, p <= Dim - 1, p++, 
      For[q = p + 1, q <= Dim, q++, 
      AppendTo[tempmatrix, GLiftMor2N[autsol, p, q]]]]; 
    Print[MatrixForm[Simplify[tempmatrix]]];
  k++]]}
\end{verbatim}

\chapter{Matrix representation for Lie algebras with nontrival center}	\label{Se:MatrRep}

If a Lie algebra $\mathfrak{g}$ has a nontrivial center, then the image of the map ${\rm ad}: \mathfrak{g} \ni v \mapsto [v, \cdot] \in \mathfrak{gl}(\mathfrak{g})$ is a matrix Lie algebra that is not isomorphic to $\mathfrak{g}$. Hence, the image of {\rm ad} does not give a matrix algebra representation of $\mathfrak{g}$, which may give rise to a problem as matrix Lie algebra representation are quite practical in computations \cite{GBGH19}. The aim of this section is to provide a method to obtain a matrix Lie algebra isomorphic to $\mathfrak{g}$ when its center, $\mathcal{Z}(\mathfrak{g})$, is not trivial, namely $\mathcal{Z}(\mathfrak{g})\neq 0$, and $\mathfrak{g}$ satisfy some additional conditions. Our method is to be employed in the rest of our work during calculations.

Let $\{e_1,\ldots, e_s, e_{s+1}, \ldots, e_n\}$ be a basis of $\mathfrak{g}$ such that $\{e_1,\ldots, e_s\}$ form a basis for  $\mathcal{Z}(\mathfrak{g})$. Let $c_{ij}^k$, with $i,j,k=1,\ldots,n$, be the structure constants of $\mathfrak{g}$ in the given basis, i.e. $[e_i,e_j]=\sum_{k=1}^nc_{ij}^ke_k$ for $i,j=1,\ldots,n$. Define then $\tilde{\mathfrak{g}} := \langle e_1, \ldots, e_n, e\rangle$. Let us determine the conditions ensuring that $\tilde{\mathfrak{g}}$ is a Lie algebra relative to a Lie bracket, $[\cdot,\cdot]_{\bar{\mathfrak{g}}}$, that is an extension of the Lie bracket on $\mathfrak{g}$, i.e. $[e',e'']_{\bar{\mathfrak{g}}}=[e',e'']_{\mathfrak{g}}$ for every $e',e''\in \mathfrak{g}$, and $[e, e_i]_{\bar{\mathfrak{g}}} = \alpha_i e_i$ for $i=1,\ldots,n$ so that 
\begin{equation}\label{cond1}
\alpha_1 \cdot \alpha_2 \cdot \ldots \cdot \alpha_s \neq 0.
\end{equation}
The meaning of the last condition will be clear in a while. If  $[\cdot,\cdot]_{\bar{\mathfrak{g}}}$ is to be a Lie bracket, its Jacobi identity leads to the following conditions 
\begin{equation}\label{cond2}
\begin{gathered}
[e, [e_i, e_j]_{\bar{\mathfrak{g}}}]_{\bar{\mathfrak{g}}} = [[e,e_i]_{\bar{\mathfrak{g}}},e_j]_{\bar{\mathfrak{g}}} + [e_i, [e,e_j]_{\bar{\mathfrak{g}}}]_{\bar{\mathfrak{g}}} \iff \sum_{k=1}^n\alpha_k c_{ij}^ke_k =\sum_{k=1}^n (\alpha_i + \alpha_j) c_{ij}^ke_k,\qquad \forall i,j=1,\ldots,n, \\
[e, [e, e_j]_{\bar{\mathfrak{g}}}]_{\bar{\mathfrak{g}}} = [[e,e]_{\bar{\mathfrak{g}}},e_j]_{\bar{\mathfrak{g}}} + [e, [e,e_j]_{\bar{\mathfrak{g}}}]_{\bar{\mathfrak{g}}}=[e, [e,e_j]_{\bar{\mathfrak{g}}}]_{\bar{\mathfrak{g}}},\qquad \forall j=1,\ldots,n.
\end{gathered}
\end{equation}
Thus, previous conditions can be reduced to requiring that, for those indices $i,j,k$ satisfying $c_{ij}^k \neq 0$, one gets 
$$
\alpha_i + \alpha_j = \alpha_k.
$$
If the latter condition is satisfied, a new Lie algebra $\tilde{\mathfrak{g}}$ arises. Note that $\mathfrak{g}$ is an ideal of $\tilde{\mathfrak{g}}$. This leads to a Lie algebra morphism $\mathscr{R}:v\in {\mathfrak{g}}\mapsto \mathscr{R}_v:=[v,\cdot]_{\tilde{\mathfrak{g}}}\in \mathfrak{gl}(\tilde{\mathfrak{g}})$. If $\mathscr{R}_v=0$, then $[v,e_i]_{\tilde{\mathfrak{g}}}=0$ for $i=1,\ldots,n$, which means that $v\in \mathcal{Z}(\mathfrak{g})$. Thus, $v = \sum_{i=1}^s \lambda_i e_i$ for certain constants $\lambda_1,\ldots,\lambda_s$ and $0 = [e,v]_{\tilde{\mathfrak{g}}} = \sum_{i=1}^s \lambda_i \alpha_i e_i$. Then, condition (\ref{cond1}) yields $\lambda_1 = \ldots =\lambda_s = 0$ and it turns out that $v=0$.
Hence, the elements $\mathscr{R}_v$, with $v\in \mathfrak{g}$, span a matrix Lie algebra of endomorphisms on $\tilde{\mathfrak{g}}$ isomorphic to $\mathfrak{g}$. 

It is clear that conditions (\ref{cond1}) and (\ref{cond2}) do not need to be satisfied for a general Lie algebra $\mathfrak{g}$ with a nontrivial center. 
Nevertheless, we used in this work that this method works for all indecomposable real four-dimensional Lie algebras with non-trivial center (see \cite{SW14} for a complete list of such Lie algebras). This will be enough for our purposes. 

Let us illustrate our method. Consider the Lie algebra $\mathfrak{s}_{ 1} = \langle e_1, e_2, e_3, e_4\rangle$ with non-vanishing commutation relations $[e_4, e_2] = e_1, [e_4, e_3] = e_3$. Let us construct a new Lie algebra $\tilde{\mathfrak{g}}=\langle e_1,e_2, e_3,e_4,e\rangle$ following (\ref{cond1}) and (\ref{cond2}). This gives rise to the following system of equations:
$$
\alpha_4 + \alpha_2 = \alpha_1, \quad \alpha_4 + \alpha_3 = \alpha_3, \quad \alpha_1 \neq 0.
$$
We have that $\mathcal{Z}(\mathfrak{s}_{1})=\langle e_1\rangle$. The previous system, under the corresponding restriction (\ref{cond1}),  has a solution $\alpha_4 =0, \alpha_3 \in \mathbb{R}, \alpha_2 = \alpha_1, \alpha_1 \in \mathbb{R}\backslash\{0\}$. In particular, set $\alpha_2=\alpha_1 = 1$ and $\alpha_3 = 0$. Thus, the endomorphisms $\mathscr{R}_{e_i}$ on $\tilde{\mathfrak{g}}$ read
{\footnotesize
$$
\mathscr{R}_{e_1} = \left(
\begingroup
\setlength\arraycolsep{4pt}
\begin{array}{ccccc}
0 & 0 & 0 & 0 & -1 \\
0 & 0 & 0 & 0 & 0 \\
0 & 0 & 0 & 0 & 0 \\
0 & 0 & 0 & 0 & 0 \\
0 & 0 & 0 & 0 & 0 
\end{array}
\endgroup
\right),
\mathscr{R}_{e_2} = \left(
\begingroup
\setlength\arraycolsep{4pt}
\begin{array}{ccccc}
0 & 0 & 0 & -1 & 0 \\
0 & 0 & 0 & 0 & -1 \\
0 & 0 & 0 & 0 & 0 \\
0 & 0 & 0 & 0 & 0 \\
0 & 0 & 0 & 0 & 0 
\end{array}
\endgroup
\right),
\mathscr{R}_{e_3} = \left(
\begingroup
\setlength\arraycolsep{4pt}
\begin{array}{ccccc}
0 & 0 & 0 & 0 & 0 \\
0 & 0 & 0 & 0 & 0 \\
0 & 0 & 0 & -1 & 0 \\
0 & 0 & 0 & 0 & 0 \\
0 & 0 & 0 & 0 & 0 
\end{array}
\endgroup
\right),
\mathscr{R}_{e_4} = \left(
\begingroup
\setlength\arraycolsep{4pt}
\begin{array}{ccccc}
0 & 1 & 0 & 0 & 0 \\
0 & 0 & 0 & 0 & 0 \\
0 & 0 & 1 & 0 & 0 \\
0 & 0 & 0 & 0 & 0 \\
0 & 0 & 0 & 0 & 0 
\end{array}
\endgroup
\right).
$$
}
As desired, the above matrices have the same non-vanishing commutation relations as $\mathfrak{s}_{ 1}$, i.e.
$$
[\mathscr{R}_{e_4},\mathscr{R}_{e_2}]=\mathscr{R}_{e_1},\qquad [\mathscr{R}_{e_4},\mathscr{R}_{e_3}]=\mathscr{R}_{e_3}.
$$

Let us prove that not every Lie algebra $\mathfrak{g}$ with non-trivial center admits an extended Lie algebra $\tilde{\mathfrak{g}}$ of the form required by our method. Consider the Lie algebra $\mathfrak{s}_{6,231} = \langle e_1, \ldots, e_6\rangle$ with nonzero commutation relations (see \cite{SW14} for details):
$$
[e_2, e_3] = e_1, \quad [e_5, e_1] = e_1, \quad [e_5, e_2] = e_2, \quad [e_6, e_1] = e_1, \quad [e_6, e_3] = e_3, \quad [e_6, e_5] = e_4.
$$
Thus, $\mathcal{Z}(\mathfrak{s}_{4,231}) = \langle e_4\rangle$. The corresponding system of equations (\ref{cond2}) reads
$$
\alpha_2 + \alpha_3 = \alpha_1, \quad \alpha_5 + \alpha_1 = \alpha_1, \quad \alpha_5 + \alpha_2 = \alpha_2, \quad \alpha_6 + \alpha_1 = \alpha_1, \quad \alpha_6 + \alpha_3 = \alpha_3, \quad \alpha_6 + \alpha_5 = \alpha_4.
$$
Then, $\alpha_5 = 0, \alpha_6 = 0$ and $\alpha_4 = \alpha_5=0$, which contradicts the assumption of our method (\ref{cond2}), namely $\alpha_4 \neq 0$. 

Consider the 7-dimensional Lie algebra $\mathfrak{g} := \langle e_1, \ldots, e_7\rangle$ with $\mathcal{Z}(\mathfrak{g}) = \langle e_7 \rangle$ and nonzero commutation relations (cf. \cite[p. 492, case $7_I$]{Se93})
\begin{equation*}
\begin{gathered}
[e_1, e_2] = e_3, \quad [e_1, e_3] = e_4, \quad [e_1, e_4] = e_5, \quad [e_1, e_6] = e_7, \quad [e_2, e_3] = e_6,\\
[e_2, e_4] = e_7, \quad [e_2, e_5] = e_7, \quad [e_2, e_6] = e_7, \quad [e_3, e_4] = -e_7.
\end{gathered}
\end{equation*}
The corresponding system (\ref{cond2}) reads
$$
\begin{array}{rlrlrl}
{\rm (i)} & \alpha_1 + \alpha_2 = \alpha_3, & {\rm (ii)} & \alpha_1 + \alpha_3 = \alpha_4, & {\rm (iii)} & \alpha_1 + \alpha_4 = \alpha_5, \\
{\rm (iv)} & \alpha_1 + \alpha_6 = \alpha_7, & {\rm (v)} & \alpha_2 + \alpha_3 = \alpha_6, & {\rm (vi)} & \alpha_2 + \alpha_4 = \alpha_7, \\
{\rm (vii)} & \alpha_2 + \alpha_5 = \alpha_7, & {\rm (viii)} & \alpha_2 + \alpha_6= \alpha_7, & {\rm (ix)} & \alpha_3 + \alpha_4 = \alpha_7.
\end{array}
$$
From (vii) and (viii), we  get $\alpha_5 = \alpha_6$; from (vi) and (ix), we get $\alpha_2 = \alpha_3$; and from (iv) and (viii), we get $\alpha_1 = \alpha_2$. Thus, $\alpha_1 = \alpha_2 = \alpha_3$. From (i), it follows that $\alpha_1 = 0$. Thus, $\alpha_2 =  \alpha_3 = 0$. From (ii), we get $\alpha_4 = 0$ and (iii) yields $\alpha_5 = 0$. From (iv), one gets $\alpha_6 = \alpha_7$. From (ix), we get $\alpha_7  =0$, and thus $\alpha_6 = 0$. Therefore, $\alpha_1 = \ldots = \alpha_7 = 0$ is the only possibility, but it does not satisfy (\ref{cond1}).

\end{document}